\documentclass[size=11pt, paper=a4, twoside=semi, BCOR=11mm]{scrreprt}

\makeatletter %
\newcommand*{\geburtsdatum}[1]{\def\@geburtsdatum{#1}}  %
\newcommand*{\geburtsort}[1]{\def\@geburtsort{#1}}      %
\newcommand*{\thesentyp}[1]{\def\@thesentyp{#1}}        %
\newcommand*{\gutachtera}[1]{\def\@gutachtera{#1}}      %
\newcommand*{\gutachterb}[1]{\def\@gutachterb{#1}}      %
\newcommand*{\abgabe}[1]{\def\@abgabe{#1}}
\newcommand*{\erscheinungsjahr}[1]{\def\@erscheinungsjahr{#1}}

\renewcommand\maketitle{
\KOMAoptions{twoside = false}   %
\begin{titlepage}               %
    \begin{center}
        \vspace{4ex}
        \parbox{14cm}{\centering\huge\bfseries \@title \par}\\[18ex]      %
        {\Large \@thesentyp}\\[4ex]
        zur\\
        Erlangung des Doktorgrades (Dr.\ rer.\ nat.)\\
        der\\
        Mathematisch-Naturwissenschaftlichen Fakult\"at\\
        der\\
        Rheinischen Friedrich-Wilhelms-Universit\"at Bonn\\[10ex]
        vorgelegt von\\[4ex]
        {\LARGE\@author}\\[4ex]
        aus\\[4ex] \@geburtsort\\[8ex]
        {Bonn, \@abgabe}
    \end{center}
    \vspace*{\fill}             %
    \clearpage                  %
    \newpage\null\thispagestyle{empty}\newpage %
    \begin{center}
      Angefertigt mit Genehmigung der Mathematisch-Naturwissenschaftlichen Fakult\"at\\
      der Rheinischen Friedrich-Wilhelms-Universit\"at Bonn
    \end{center}
    \vspace*{\fill}
    1. Gutachter: \@gutachtera\\[2ex]
    2. Gutachter: \@gutachterb\\[2ex]
    Tag der Promotion: 06. Juli 2018\\[2ex]
    Erscheinungsjahr: \@erscheinungsjahr
\end{titlepage}         %
\newpage\null\thispagestyle{empty}\newpage %
\KOMAoptions{twoside}   %
\global\let\maketitle\relax
\global\let\@author\@empty
\global\let\@geburtsdatum\@empty
\global\let\@geburtsort\@empty
\global\let\@thesentyp\@empty
\global\let\@gutachtera\@empty
\global\let\@gutachterb\@empty
\global\let\@date\@empty
\global\let\@title\@empty
\global\let\@erscheinungsjahr\@empty
\global\let\@abgabe\@empty
\global\let\title\relax
\global\let\author\relax
\global\let\geburtsdatum\relax
\global\let\geburtsort\relax
\global\let\thesentyp\relax
\global\let\gutachtera\relax
\global\let\gutachterb\relax
\global\let\date\relax
\global\let\erscheinungsjahr\relax
\global\let\abgabe\relax
\global\let\and\relax
}
\makeatother
\author{Johannes Holke}
\geburtsdatum{30.01.1989}
\geburtsort{K\"oln}
\date{\myformat\today}
\abgabe{\myformat{22. M\"arz 2018}}
\erscheinungsjahr{2018}
\gutachtera{Prof.\ Dr.\ Carsten Burstedde}
\gutachterb{Prof.\ Dr.\ Michael Griebel} %
\title{Scalable Algorithms for Parallel Tree-based Adaptive Mesh Refinement
with General Element Types}
\thesentyp{Dissertation}

\usepackage{amsmath,amssymb,stmaryrd,mathrsfs}
\usepackage{cite} %
\usepackage[table,usenames,dvipsnames]{xcolor}  %
\definecolor{linkcolor}{rgb}{0,0,0} %
\usepackage[USenglish]{babel}
\usepackage[utf8]{inputenc}
\usepackage{etex}
\usepackage[ruled,linesnumbered,algosection,vlined]{algorithm2e} %

\usepackage{amsthm}
\usepackage{enumerate}
\usepackage{makeidx}
\usepackage{array}
\usepackage{tabularx} %
\usepackage{hhline}   %
\usepackage{mathtools}    %
\usepackage{mathdots}     %
\usepackage[all,2cell]{xy} \UseAllTwocells \SelectTips{eu}{10} %
\usepackage{tikz}

\usepackage{caption} %
\usepackage{subcaption} %
\usepackage{mathrsfs} %
\usepackage[all,2cell]{xy} \UseAllTwocells \SelectTips{eu}{10} %
\usepackage{multirow}	  %
\usepackage{enumerate}
\usepackage[ruled,linesnumbered,algosection,vlined]{algorithm2e} 
\usepackage{sidecap}	%
\usepackage[hyphens]{url}
\usepackage[ngerman]{datetime}  %
\newdateformat{myformat}{\THEDAY. \monthnamengerman[\THEMONTH], \THEYEAR} 
\usepackage{listings}
\usepackage{longtable} %

\usepackage{calc}
\usepackage{doxygen}
\usepackage{graphicx}
\usepackage{textcomp}
\usepackage{pdfpages} %

\usepackage{verbatim} %

\newif\ifANMan
\ANMantrue %

\newif\ifDEBUG

\ifpdf
\newcommand{\set}[1]{\left\lbrace\, #1 \,\right\rbrace}
\newcommand{\setm}[2]{\left\lbrace\, #1 \,\middle|\, #2 \,\right\rbrace}
\newcommand{\IL}{{\mathbb{L}}}
\newcommand{\IR}{{\mathbb{R}}}

\newcommand{\IZ}{{\mathbb{Z}}}
\newcommand{\IN}{{\mathbb{N}}}
\newcommand{\abst}[1]{\, #1 \,} %
\newcommand{\mytabvspace}{\vphantom{${X^X}^X$}} %
\newcommand{\myhugetabvspace}{{\huge\vphantom{${X^X}^X$}}} %
\newcommand{\forest}[1]{{\mathscr{#1}}} %

\newcommand{\Comment}[1]{\tcc*[r]{#1}}   %
\newcommand{\IfComment}[1]{\tcc*[f]{#1}} %
\let\oldnl\nl%
\newcommand{\nonl}{\renewcommand{\nl}{\let\nl\oldnl}}%
\newcommand{\algomod}{\abst{\scalebox{0.9}{\%}}} %
\newcommand{\algoand}{\textbf{ and }} %
\newcommand{\algor}{\textbf{ or }} %
\newcommand{\algotrue}{\textbf{True }} %
\newcommand{\algofalse}{\textbf{False }} %
\newcommand{\algoresult}{\nonl\textbf{Result:\quad }} %
\newcommand{\algofor}[1]{\For{\upshape #1}} %
\newcommand{\algoforcom}[2]{\For(#1){\upshape #2}} %
\newcommand{\algoif}[1]{\If{\upshape #1}} %
\newcommand{\algoifcom}[2]{\If(#1){\upshape #2}} %
\newcommand{\algoeif}[1]{\eIf{\upshape #1}} %
\newcommand{\algoeifcom}[2]{\eIf(#1){\upshape #2}} %
\newcommand{\bitwand}{\;\&\;} %
\newcommand{\aorgl}{\abst{|\hspace{-0.5ex}=}} %
\newcommand{\cfbox}[2]{%
    \colorlet{currentcolor}{.}%
    {\color{#1}%
    \fbox{\color{currentcolor}#2}}%
  }
\newcommand{\e}{\text{e}} %
\newcommand{\ie}{i.e.\xspace} %
\newcommand{\adef}{\, {:=} \,}
\newcommand{\agl}{\abst{=}}
\newcommand{\indexset}[1]{\set{0,\dots,#1-1}} %
\newcommand{\ohne}{\backslash}

\newcommand{\pforest}{\texttt{p4est}\xspace} %
\newcommand{\tetcode}{\texttt{t8code}\xspace} %
\newcommand{\dtype}[1]{{#1}\xspace}
\newcommand{\Int}{\dtype{int}}
\newcommand{\Tet}{\texttt{Tet}\xspace}
\newcommand{\aTet}{{\Tet}} %
\newcommand{\ghostb}{\texttt{Ghost\_v1}\xspace} %

\ifANMan
\newcommand{\ANM}[1]{\footnote{ANMERKUNG: #1}}  %
\newcommand{\ANMtext}[1]{\footnotetext{ANMERKUNG: #1}}
\newcommand{\ANMmarkn}[1]{\footnotemark[#1]}           %
\newcommand{\ANMtextn}[2]{\footnotetext[#1]{ANMERKUNG: #2}}
\else
\newcommand{\ANM}[1]{}  %
\newcommand{\ANMtext}[1]{}
\newcommand{\ANMmarkn}[1]{}          %
\newcommand{\ANMtextn}[2]{}
\fi

\ifDEBUG
\newcommand{\figlabel}[1]{\caption*{\color{red} #1}\label{#1}}
\else
\newcommand{\figlabel}[1]{\label{#1}}
\fi

\definecolor{mygreen}{rgb}{0,0.7,0}
\newcommand{\red}[1]{{\color{red}#1}}     %
\newcommand{\green}[1]{{\color{mygreen}#1}} %

\DeclareMathOperator{\dist}{dist}
\DeclareMathOperator{\vol}{vol} %
\DeclareMathOperator{\inter}{\dot\perp} %
\DeclareMathOperator{\type}{type}				%
\DeclareMathOperator{\sign}{sign}				%
\DeclareMathOperator{\cid}{\textrm{cube-id}} %
\DeclareMathOperator{\sfcf}{\mathcal I} %

\newcommand{\ainter}{\abst{\inter}}

\newcommand{\lowlevel}[2]{
\vspace{3ex}
\begin{minipage}{\dimexpr\textwidth-3ex}
\texttt{#1}

#2
\end{minipage}
\vspace{3ex}
}

\newcolumntype{L}[1]{>{\raggedright\arraybackslash}p{#1}}
\newcolumntype{C}[1]{>{\centering\arraybackslash}p{#1}}
\newcolumntype{R}[1]{>{\raggedleft\arraybackslash}p{#1}}

\theoremstyle{plain}
\newtheorem{theorem}{Theorem}[chapter]
\newtheorem{proposition}[theorem]{Proposition}
\newtheorem{lemma}[theorem]{Lemma}
\newtheorem{pardgm}[theorem]{Paradigm}
\newtheorem{property}[theorem]{Property}
\newtheorem{corollary}[theorem]{Corollary}
\newtheorem{conjecture}[theorem]{Conjecture}
\theoremstyle{definition}
\newtheorem{example}[theorem]{Example}
\newtheorem{remark}[theorem]{Remark}
\newtheorem{definition}[theorem]{Definition}

\graphicspath{{./pictures_build/},{./pictures_build/ghost/},{./pictures_build/partition},{./pictures_build/tetsfc/},{./pictures_build/conn_comp/}, {./pictures_build/balance/},{./timings/ghost_balance/test1_hextetcube/one_file_per_run/},{./pictures_build/intro_sota/},{./pictures_build/application/},{./timings/advection/}}

\begin{document}
\maketitle

\chapter*{Summary}

In this thesis, we develop, discuss and implement algorithms for scalable
parallel tree-based adaptive mesh refinement (AMR) using space-filling curves
(SFCs). We create an AMR software that works independently of the used element
type, such as for example lines, triangles, tetrahedra, quadrilaterals,
hexahedra, and prisms. Along with a detailed mathematical discussion, this
requires the implementation as a numerical software and its validation, as
well as scalability tests on current supercomputers.

For triangular and tetrahedral elements (simplices) with red-refinement
(1:4 in 2D, 1:8 in 3D), we
develop a new SFC, the tetrahedral Morton space-filling curve (TM-SFC). Its
construction is similar to the Morton index for quadrilaterals/hexahedra, as it
is also based on bitwise interleaving the coordinates of a certain vertex of
the simplex, the anchor node.
Additionally, we interleave with a new piece of information, the so called
type. The type distinguishes different simplices with the same anchor node.
To store the necessary information of a $d$-dimensional simplex, we require
$10$ bytes per
triangle and $14$ bytes per tetrahedron, which is only one byte more than used
in the classical Morton index for quadrilaterals ($9$ bytes) and hexahedra
($13$ bytes).
For these simplices, we develop element local algorithms such as constructing the
parent, children, or face-neighbors of a simplex, and show that most of them
are constant-time operations independent of the refinement level.

With SFC based partitioning it is possible that the mesh elements that are
partitioned to one process do not form a face-connected domain. The amount of
parallel communication among processes with neighboring domains rises with the
number of face-connected components. 
We prove the following upper bounds for the number of face-connected
components of segments of the TM-SFC: With a maximum refinement
level of $L$, the number of face-connected components is bounded by $2(L-1)$ in
2D and $2L+1$ in 3D. Additionally, we perform a numerical investigation of the
distribution of lengths of SFC segments.

Furthermore, we develop a new approach to partition and repartition a coarse
(input) mesh among the processes. Compared to previous methods it optimizes for
fine mesh load-balance and reduces the parallel communication of coarse mesh
data. We discuss the coarse mesh repartitioning algorithm and demonstrate that
our method repartitions a coarse mesh of $371\e9$ trees on 917,504 processes
(405,000 trees per process) on the Juqueen supercomputer in $1.2$ seconds.

We develop an AMR concept that works independently of the element type;
achieving this independence by strictly distinguishing between functions that
operate on the whole mesh (high-level) and functions that locally operate on a
single element or a small set of elements (low-level).
We define an application programming interface (API) of low-level functions and
develop the high-level functions such that every element-local operation is
performed by a low-level function. Thus, by using different implementations of
the low-level API for different meshes, or different parts of the same mesh, we
are able to use different types of elements.

Many numerical applications, for example finite element and finite volume
solvers, require knowledge of a layer of ghost elements. Ghost elements of
a process are those elements that lie on a different process but are 
(face-)neighbors of a process local element.
We discuss a new approach to generate and manage these ghost elements that fits
into our element-type independent approach. We define and describe the
necessary low-level algorithms. Our main idea is the computation of
tree-to-tree face-neighbors of an element via the explicit construction of the
element's face as a lower dimensional element.
In order to optimize the runtime of this method we enhance the algorithm with a
top-down search method from Isaac, Burstedde, Wilcox, and Ghattas, and
demonstrate how it speeds up the computation by factors of 10 to 20
achieving runtimes comparable to state-of-the art implementations with
fixed element types.

With the ghost algorithm we build a straight-forward ripple version of the
2:1 balance algorithm. This is not an optimized version but 
it serves
as a
feasibility study for our element-type independent approach.

We implement all algorithms that we develop in this thesis in the new
AMR library \tetcode, using the TM-SFC for simplicial and tetrahedral elements.
Our modular approach allows us to reuse existing software, which we demonstrate
by using the library \pforest for quadrilateral and hexahedral elements.
In a concurrent Bachelor's thesis by David Knapp (INS, Bonn) the necessary
low-level algorithms for prisms were developed.
With \tetcode we demonstrate that we can create, adapt, (re-)partition,
and balance meshes, as well as create and manage a ghost layer.
In various tests we show excellent strong and weak scaling behavior of our
algorithms on up to 917,504 parallel processes on the Juqueen and 
Mira supercomputers using up to $858\e9$ mesh elements.

We conclude this thesis by demonstrating how an application can be coupled with
the AMR routines. We implement a finite volume based advection solver 
using \tetcode and show applications with triangular,
quadrilateral, tetrahedral, and hexahedral elements, 
as well as 2D and 3D hybrid meshes, the latter consisting of tetrahedra,
hexahedra, and prisms.

Overall, we develop and demonstrate a new simplicial SFC and create a fast and
scalable tree-based AMR software that offers a flexibility and generality that
was previously not available.

 \section*{Acknowledgements}
First of all, I want to thank my advisor Prof.\ Dr.\ Carsten Burstedde for his
constant support and feedback during the development of this thesis, various
helpful discussions and lessons in parallel C programming.
I would like to thank Prof.\ Dr.\ Michael Griebel for being the second reviewer
of this thesis. Thank you to Prof.\ Dr.\ Matthias Kreck for being a
great mentor. 
I appreciate the enlightening atmosphere at the INS and thank all my current and
previous colleagues. Special thanks to Jose Alberto Fonseca.

Additional thanks to Tobin Isaac for providing the
interface to the MPI-3 shared array functionality in the \texttt{sc} library. 

I appreciate the financial support by the Bonn International Graduate
School of Mathematics (BIGS) and travel support by the Bonn Hausdorff Center
for Mathematics (HCM).

I gratefully acknowledge the Gauss Centre for Supercomputing e.V.
(\url{www.gauss-centre.eu}) for funding this project by providing computing
time through the John von Neumann Institute for Computing (NIC) on the GCS
Supercomputer JUQUEEN at Jülich Supercomputing Centre (JSC).
Additionally, this research used resources of the Argonne Leadership Computing
Facility, which is a DOE Office of Science User Facility supported under
contract DE-AC02-06CH11357.

Danke an meine Familie f\"ur ihre Unterst\"utzung.
Vielen Dank an Felix Boes f\"ur viele produktive und lustige Nachmittage
und f\"ur au\ss ergew\"ohnliche Freundschaft. 

Mein gr\"o\ss ter Dank gilt
Clelia. Danke f\"ur alles, insbesondere f\"ur deine bedingungslose
Unterst\"utzung, viel Geduld und dass du immer die richtigen Worte zur
Motivation findest. \emph{Non ti ringrazier\`{o} mai abbastanza per tutto
il bene che mi hai fatto.}
 \tableofcontents
\clearpage
\chapter{Introduction}

Numerical simulations for the solution of partial differential equations (PDEs)
have a wide range of applications in scientific and industrial computing.
Almost all of the common methods, including the finite difference method (FD),
the finite element method (FEM), the finite volume method (FV), and the
discontinuous Galerkin method (DG) use meshes to discretize a given domain on
which to solve a PDE, e.g.~\cite{Braess97,SuliSchwabHouston00,LeVeque02}.

Most commonly, meshes consist of quadrilaterals or triangles in 2D, and
hexahedra or tetrahedra in 3D. Triangles and tetrahedra are in general better
suited to accurately approximate complex domains~\cite{Shewchuk96, Si06}.
However, with quadrilaterals and hexahedra we can use less elements to model the
same domain size and many numerical schemes can exploit a tensor product
structure and are thus easier to implement on these element types
\cite{Peano90, StewartEdwards04, TuOHallaronGhattas05, SundarSampathBiros08,
Zumbusch02}.
In some use cases, especially in industrial engineering applications, both
advantages are needed, which motivates the use of hybrid meshes with multiple
element types. In this case, it is necessary to use prisms or pyramids to
transition between hexahedra and
tetrahedra~\cite{Marcum95,VinchurkarLongest08,MartineauStokesMundayEtAl06,KirbyBrazellYangEtAl17}.
Less common but also in use are hybrid meshes with two element types, either
hexahedra and prisms, or tetrahedra and
prisms~\cite{De11,KhawajaKallinderis00}. 

One of the most successful tools to improve the performance of numerical
simulations is the concept of adaptive mesh refinement (AMR),
e.g.~\cite{BabuskaRheinboldt78,Dorfler96,SteinmanMilnerNorleyEtAl03}. 
For all numerical methods described above, a finer mesh resolution results in a
reduction of the computational error and AMR describes the concept of changing
this resolution locally; thus, it is possible to maintain fine resolutions in
those parts of the mesh where a smaller error is needed while keeping the mesh
resolution coarse elsewhere. For this decision process, we may utilize local
error
estimators~\cite{BraessVerfurth96,BabuskaRheinboldt78,ZienkiewiczKelleyGagoEtAl82,VendittiDarmofal00}.
Hence, with AMR we can reduce the number of mesh elements---and thus
the memory footprint and runtime of numerical solvers---significantly compared
to uniform meshes, which use the same mesh resolution in the whole 
mesh~\cite{Klein99,BursteddeStadlerAlisicEtAl13,MuellerBehrensGiraldoEtAl13}.
However, using AMR, especially in a parallel high-performance computing (HPC)
environment, introduces a substantial overhead in mesh management.
Particularly demanding tasks include: mesh refining and coarse\-ning,
(re-)\-par\-titioning in parallel, the creation of (off-process) ghost elements,
random access of mesh elements, and more.
Therefore, the implementation of complete AMR frameworks is often outsourced into
AMR-specialized software libraries, such as 
\texttt{BoxLib} and its
 successor \texttt{AMReX}~\cite{boxlib,amrex},
 \texttt{Chombo}~\cite{ColellaGravesKeenEtAl07},
 \texttt{Enzo}~\cite{BryanNormanOSheaEtAl14},
\texttt{libMesh}~\cite{KirkPetersonStognerEtAl06},
\texttt{\pforest}~\cite{Burstedde10a},
\texttt{ParFUM}~\cite{LawlorChakravortyWilmarthEtAl06},
\texttt{Peano}~\cite{WeinzierlMehl11},
\texttt{PUMI}~\cite{IbanezSeolSmithEtAl16},
\texttt{PYRAMID}~\cite{NortonLyzengaParkerEtAl04}.
 \texttt{SAMRAI}~\cite{WissinkHornungKohnEtAl01},
 \texttt{Uintah}~\cite{BerzinsLuitjensMengEtAl10}, and others.

A main challenge for AMR on HPC systems is storing and load-balancing the mesh
in parallel, in particular when the mesh changes frequently during the
computation.
Some examples are the simulation of earth mantle convection~\cite{BursteddeStadlerAlisicEtAl13,KronbichlerHeisterBangerth12}, blood flow~\cite{De11,MuellerSahniLiEtAl05},
two-phase flow and level-set methods~\cite{Albrecht16,MirzadehGuittetBursteddeEtAl16,KosterGriebelKevlahanEtAl98,Klitz14,DruiFiklKestenerEtAl16,DonathMeckeRabhaEtAl11},
and molecular dynamics~\cite{ArnoldLenzKesselheimEtAl12,PhillipsBraunWangEtAl05}.

One common method for storing the mesh is unstructured meshing, where the
element connectivity is modeled as a graph. Partitioning the mesh among
parallel processes is then delegated to graph based partitioning methods such
as ParMETIS or Scotch~\cite{KarypisKumar98, DevineBomanHeaphyEtAl02,
ChevalierPellegrini08}.
Well-known open source application codes that use unstructured meshes are
FEniCS \cite{LoggMardalWells12}, PLUM
\cite{OlikerBiswas98,OlikerBiswasGabow00}, OpenFOAM \cite{OpenCFD07}, or MOAB
from the SIGMA toolkit \cite{TautgesMeyersMerkleyEtAl04}.
While these methods allow for approximately optimal mesh partitions 
(i.e.\ in terms of surface-to-volume ratio) and provide maximum flexibility of
element-to-element connections, their complexity results in slow runtimes 
of the meshing routines and a high demand in memory.

A different approach to efficiently store and partition meshes are space-filling 
curves (SFCs)~\cite{Bader12, Zumbusch00, Sagan94, HaverkortWalderveen10}.
Instead of heuristically solving NP-hard graph partitioning problems, SFCs are
used to approximately solve the problem in linear runtime.
Starting with a single coarse element and a recursive refinement rule,
prescribing how we may replace an element with finer elements covering the same
area, we obtain a tree structure of all constructable elements. The tree's root
is the coarse element. A particular adaptive mesh resulting from the coarse
element corresponds to the leaves of a subtree. An SFC is a linear order of
these leaves. It allows us to efficiently store elements and element data in a
linear array according to these SFC indices. Partitioning to $P$ processes
reduces to splitting up the linear order into $P$ equally sized pieces. This
approach was used for quadrilaterals and hexahedra for example in the
\texttt{Octor}~\cite{TuOHallaronGhattas05} and
\texttt{Dendro}~\cite{SampathAdavaniSundarEtAl08} codes, which demonstrated
scaling to tens of thousands of processes.

The SFC approach has then been extended to multiple coarse elements, forming a
forest of trees~\cite{BangerthHartmannKanschat07,BursteddeWilcoxGhattas11,StewartEdwards04}.
With this technique it is possible to model any complex geometry as
computational domain~\cite{BangerthHartmannKanschat07, StewartEdwards04}.
The \pforest code~\cite{Burstedde10a} is a particular scalable
example of this approach.
\pforest uses quadrilateral and hexahedral elements with the Morton SFC
and was shown to scale to up to 3.1 million
processes~\cite{RudiMalossiIsaacEtAl15,MullerKoperaMarrasEtAl15,BursteddeGhattasGurnisEtAl10}.

In the $\textrm{\texttt{sam(oa)}}^2$ framework tree-based AMR with SFCs has been
implemented for triangles using the Sierpinski SFC~\cite{BaderRahnemaMeister11}
and was shown to scale to over a hundred thousand MPI
ranks~\cite{MeisterRahnemaBader16}. Extensions of the 2D Sierpinski curve to
3D tetrahedra exist~\cite{Bader12,Sagan94}, but we are not aware of a software
utilizing them for tree-based AMR.

\vspace{3ex}

We begin this thesis with developing a new SFC index for triangular and
tetrahedral elements (simplices) with red-refinement~\cite{Bey92}, the
tetrahedral Morton index (TM-index). Its construction is similar to the Morton
SFC index for
quadrilaterals/hexahedra~\cite{Morton66}, as it is also based on bitwise
interleaving the coordinates of a certain vertex, the anchor node.
Opposed to quadrilaterals/hexahedra, a simplex is not uniquely
determined by its anchor node and level. 
We identify a distinguishing property, the type of a simplex.
The TM-index is formed by bitwise interleaving the anchor node
coordinates and type of the simplex and all of its ancestors. 
To store the necessary information of a simplex, we require only the
coordinates of the anchor node plus the type of the simplex, resulting in $10$
bytes per triangle and $14$ bytes per tetrahedron. We demonstrate that
in terms of memory and runtime the TM-SFC is comparable to the Morton SFC
for quadrilaterals and hexahedra.
We prove upper bounds for the number of connected components of any segment
of the TM-SFC in 2D and 3D.
The simplicial meshes obtained via red-refinement contain so called hanging nodes,
faces and edged. These occur at neighboring elements of different refinement levels,
where for example a node of the finer element does not coincide with a node of
the coarser element, but lies on an edge/face of it.
So far, hanging nodes are relatively uncommon for simplicial meshes.  However,
since they are being successfully used with quadrilateral/hexahedral
AMR~\cite{BursteddeWilcoxGhattas11,SolinCervenyDolezel08,AinsworthPinchedez03},
we are certain that applications are able to incorporate them in their solvers
when using the TM-index. 

Continuing, we introduce a new approach for partitioning the coarse mesh among
the parallel processes. With coarse mesh partitioning we are able to scale the
size and complexity of the geometry representation with the number of parallel
processes. Typically, such coarse meshes are generated using external tools,
such as for example \texttt{Gmsh}~\cite{GeuzaineRemacle09} or
\texttt{TetGen}~\cite{Si06}.  Opposed to previously used methods, we allow
multiple owner processes for each tree.  We develop the necessary communication
pattern with a minimal number of senders, receivers, and messages.
This technique enables us to minimize the parallel communication during (re-)partitioning
of the coarse mesh while simultaneously maintaining an optimal load-balance of
the leaf elements.
We also discuss the identification and communication of ghost trees. These are
trees that are not local to a process, but are face-neighbor to a process's
local tree.

Another contribution of this thesis is the development of a new algorithm
to create the ghost elements of the mesh. The main advantage of our
method is that it works independently of the used type of elements and thus
readily applies, for example, to triangular, tetrahedral, quadrilateral, hexahedral,
and hybrid meshes. The core idea is to construct $d$-dimensional tree-to-tree
face-neighbor elements by explicitly building the $(d-1)$-dimensional face
and transforming its coordinates before building the neighbor element from it.
We optimize our algorithm by utilizing previously published recursive search
patterns~\cite{IsaacBursteddeWilcoxEtAl15} and obtain runtimes that are
comparable to state-of-the-art implementations with fixed
element-type~\cite{IsaacBursteddeWilcoxEtAl15}.

We implement the methods from this thesis in an element-type independent
scalable tree-based AMR library, which we call \tetcode\footnote{\url{https://github.com/holke/t8code}}. 
The core mesh handling algorithms that we discuss are \texttt{New} for creating
a uniform mesh on a given coarse mesh, \texttt{Adapt} for refining and/or
coarsening elements, \texttt{Partition} for redistributing the elements to
maintain a balanced load, \texttt{Ghost} for creating a layer of off-process
ghost elements, and 2:1 \texttt{Balance} to modify a mesh such that every
element's refinement level differs from the level of any face-neighbor by at
most $\pm 1$; see e.g.~\cite{BursteddeWilcoxGhattas11,
IsaacBursteddeWilcoxEtAl15}. We reformulate these algorithms in such a way
that they do not explicitly use
geometric or topological information of the individual elements, thus being
independent of the element-type.
A key technique is to strictly distinguish between functions that operate on the
whole mesh (high-level) and element-local functions (low-level).

We define a collection of low-level functions and develop the high-level functions
such that every element-local operation is performed by a low-level function.
Thus, by using different implementations of the low-level algorithms for different
meshes, or different parts of the same mesh, we are able to use different types
of elements or SFCs.
We also directly obtain a dimension independent formulation. For example, the
difference between 2D computations with triangles and 3D computations with
tetrahedra is merely a change of the low-level implementation.
The modular approach allows us to reuse existing software for the low-level
implementations, which we demonstrate by using the library \pforest for
quadrilateral and hexahedral elements.

We perform extensive strong and weak scaling tests of our methods and verify
optimal scalability on up to 917,504 parallel processes on the Juqueen and
Mira supercomputers using up to 858e9 mesh elements.

\vspace{3ex}
This thesis is organized as follows:

In Chapter~\ref{ch:AMR} we give an overview of currently existing
adaptive mesh refinement (AMR) techniques with emphasis on tree-based AMR.
We continue this overview in Chapter~\ref{ch:sfc}, where we discuss
SFCs and their connection with tree-based AMR.
We present a novel approach to the theory of discrete SFCs that is 
suited for the purpose of AMR.

In Chapter~\ref{ch:tetSFC} we introduce a new SFC, the tetrahedral Morton (TM-)
SFC for triangular and tetrahedral red-refinement. The TM-SFC is similar to
the Morton SFC for hypercubes in that it takes advantage of bitwise
interleaving techniques. This chapter was first published
in~\cite{BursteddeHolke16} in 2016 by the Society for Industrial and Applied
Mathematics (SIAM).

We continue the discussion of the TM-SFC in Chapter~\ref{ch:conncomp}, where we
prove bounds for the number of connected components of segments of the TM-SFC.
We show that the number of connected components in a refinement with elements
of maximum level $L$ is at most $2(L-1)$ in 2D and at most $2L+1$ in 3D
and conclude with a numerical investigation of the distribution of the 
different numbers of connected components occuring.

At the time of this writing the results from Chapter~\ref{ch:conncomp} have
been submitted for publication and are available as a
preprint~\cite{BursteddeHolkeIsaac17b}.  This publication also contains results
that were obtained by Burstedde and Isaac. However, Chapter~\ref{ch:conncomp}
contains those excerpts that are the author's own work.

Chapter~\ref{ch:cmesh} is devoted to the discussion of the coarse mesh of trees
with particular emphasis of a partition technique for it. We describe an
approach for the distribution of the trees to the processes that minimizes
communication and maintains an optimal load-balance of forest elements. This
chapter was first published in~\cite{BursteddeHolke17} in 2017 by the Society
for Industrial and Applied Mathematics (SIAM).

In Chapter~\ref{ch:ghost} we discuss a version of the \texttt{Ghost} algorithm
to create a layer of ghost elements that works with our framework for arbitrary
element types. We then discuss a first non-optimized version of
\texttt{Balance} to establish a 2:1 balance condition among the elements
in Chapter~\ref{ch:balance}.
We conclude these two chapters with numerical experiments of \texttt{Ghost} and
\texttt{Balance} on hexahedral and tetrahedral meshes. We demonstrate scaling
to the full system of JUQUEEN using 458k MPI ranks.
We are currently preparing Chapters~\ref{ch:ghost} and~\ref{ch:balance} for future
publication.

We conclude this thesis by presenting a numerical solver in
Chapter~\ref{ch:app} in order to demonstrate the usability of the presented AMR
algorithms for applications. We solve the advection equation with a finite
volume scheme on 2D and 3D non-hybrid and hybrid meshes and show excellent
strong scaling on up to 458k MPI ranks.

 \chapter{Adaptive Mesh Refinement}
\label{ch:AMR}
In this chapter we introduce uniform and adaptive mesh refinement (AMR).
We discuss advantages and disadvantages of both methods and motivate
the usage of AMR.
We carry on by giving a brief overview of unstructured and block-structured AMR,
and conclude this chapter with a more detailed elaboration of tree-based AMR.

\section{Uniform and adaptive mesh refinement}

Most numerical solvers use meshes to discretize an analytical domain on which
to solve a PDE. Common methods are finite differences, the finite
element method (FEM), the finite volume method (FV), and the discontinuous
Galerkin method (DG). See~\cite{Braess97,SuliSchwabHouston00,LeVeque02} for an
overview.

For all these methods a finer mesh resolution results in a smaller
computational error.
Thus, in order to increase the accuracy of the numerical solution, the
computational mesh is replaced with a finer mesh.

One method is to entirely remesh the whole domain with smaller elements with no
particular parent-child hierarchy between the old and new mesh. Another
refinement method is to replace mesh elements by smaller elements covering the
same area as the initial element. In this case, we have a refinement hierarchy
between the meshes.
The level of an element denotes the number of refinements that were performed
from a certain root mesh to that element.

We distinguish between uniform mesh refinement and adaptive mesh refinement (AMR).
The former means that each element is refined to the same level, resulting in 
an evenly structured mesh of same-level elements.
With AMR we refine or coarsen individual mesh elements to different levels
according to a given refinement criterion. In this way it is possible to
concentrate the mesh elements in areas where the computational error is large
and maintain a coarse resolution in those regions where the error is already
small.
In order to decide where to refine or coarsen the mesh, local error estimators
are usually taken into
account~\cite{BraessVerfurth96,BabuskaRheinboldt78,ZienkiewiczKelleyGagoEtAl82,VendittiDarmofal00}.
See also Figure~\ref{fig:uniada} for an example of uniform and adaptive meshes.

\begin{figure}
  \center
  \includegraphics[width=0.4\textwidth]{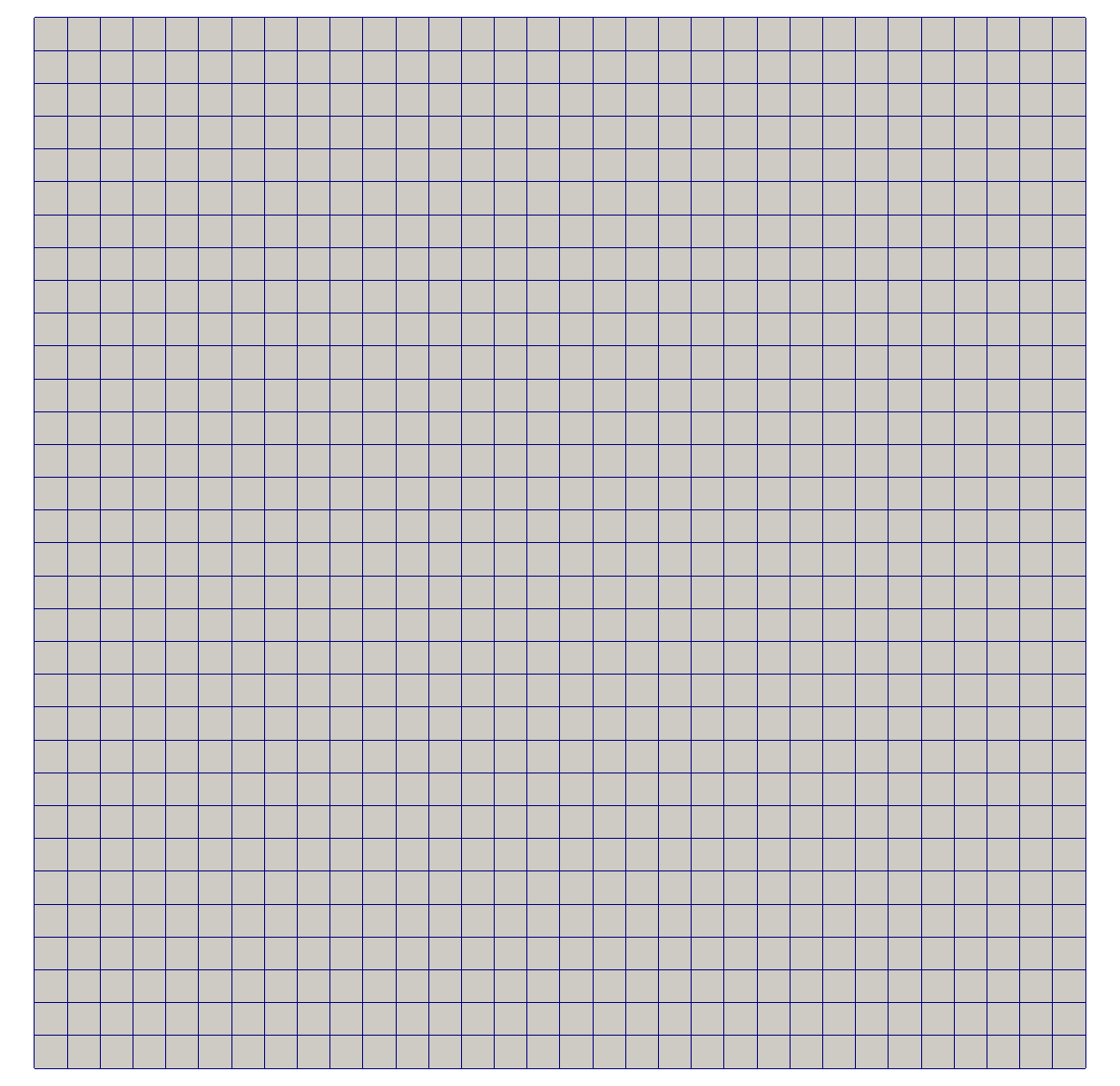}\hfill
  \includegraphics[width=0.4\textwidth]{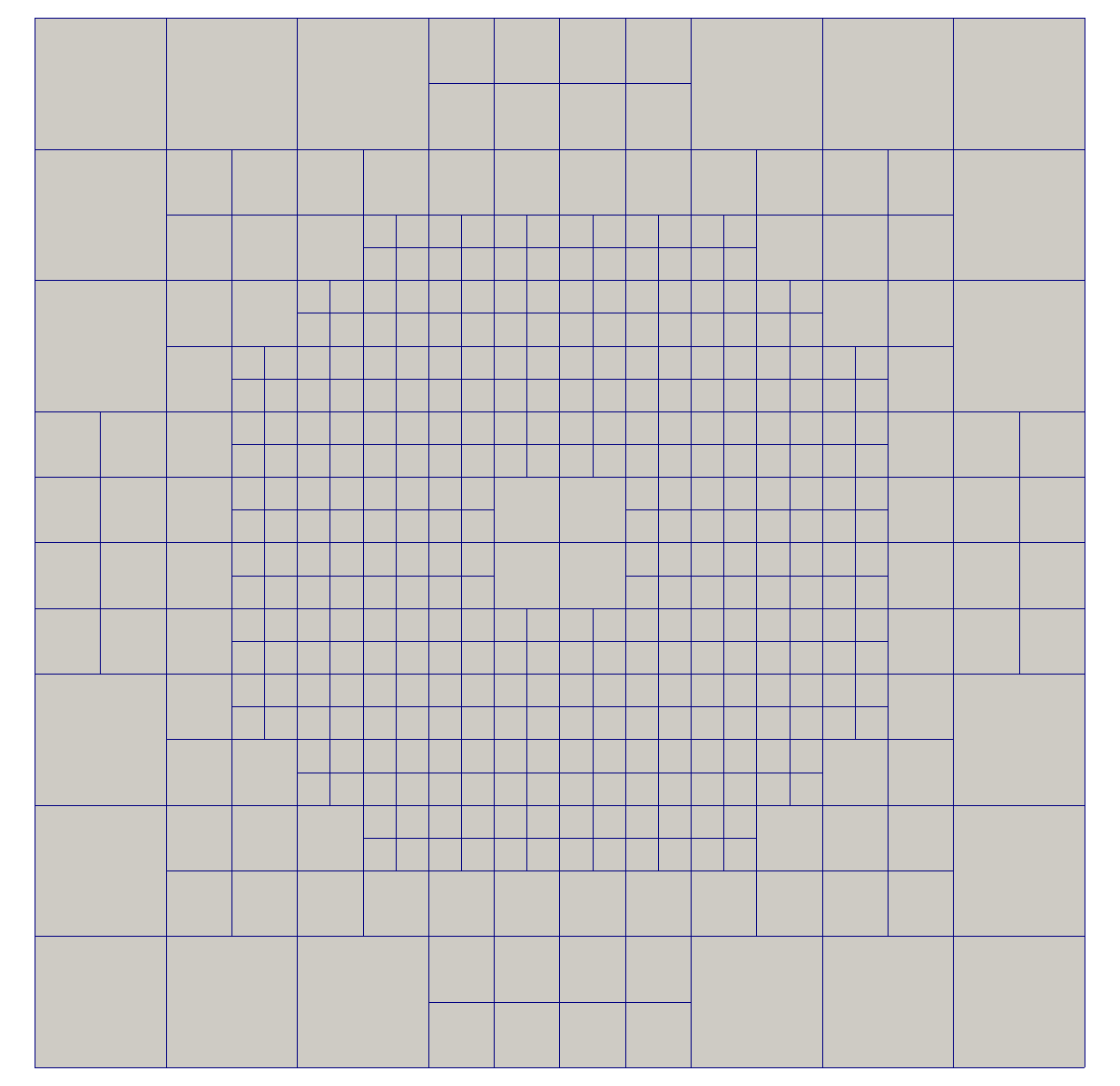}
  \caption[Uniform and adaptive mesh]{A uniform mesh (left) and an adaptive mesh (right).
  In a uniform mesh, all elements have the same size or refinement level, while in adaptive
  meshes the element's refinement levels can differ.
  }
  \figlabel{fig:uniada}
\end{figure}

When we consider time-dependent solvers, such as for example for the
Navier-Stokes equations \cite{Chorin68,KosterGriebelKevlahanEtAl98}, uniform
meshes do not change between time steps in general.
With AMR, on the contrary, the mesh often changes from one time step to the next.
In this case we also speak of dynamic AMR.
Dynamic AMR also occurs in static solvers when, for example, the problem
is first solved on a coarse mesh in order to estimate the local error according
to which the mesh is then refined and the solver started again.

First applications of AMR came up in the mid-80s~\cite{BergerOliger84}, and
using AMR in parallel started in the late 90s, e.g.~\cite{GriebelZumbusch98}.

\section{Motivation for AMR}

The main advantage of uniform meshes is that they are easy to implement.
The uniform structure makes iterating over the mesh and finding neighbor
elements, etc., near trivial tasks.
In parallel computations, load-balancing is straightforward and does not change 
in time-dependent simulations.

A disadvantage of uniform meshes is that resolving to higher accuracy needs
exponentially more elements. Such high element counts increase the memory
demand and runtime of an application. It may also happen that the number of
elements needed for a desired accuracy results in meshes that are too large to
fit into memory.

AMR overcomes these disadvantages of uniform meshes by design. With adaptive
meshes it is often possible to use less mesh elements and thus less memory to
reach the same accuracy as with a uniform
mesh~\cite{GermaschewskiBhattacharjeeGrauerEtAl03}. This permits
simulating problems with AMR to an accuracy that cannot be reached with uniform
meshes~\cite{Klein99,BursteddeStadlerAlisicEtAl13}.

Another gain of AMR is that computations are faster in comparison to uniform
meshes, or reach an increased accuracy while maintaining the same 
runtime~\cite{MuellerBehrensGiraldoEtAl13}. 

The main disadvantage of AMR is an increase in algorithmic complexity. For
example, for parallel uniform meshes assignment of elements to the processes
can be computed straightforwardly and usually does not change during a
computation.  Using AMR, the local count of elements on each process changes
with each adaptation step, which can be as frequent as every time step in a
time-dependent simulation.
These changes make it necessary to repartition the mesh in order to
load-balance the parallel computation. Likewise, computing ghost elements and
iterating through the mesh are straightforward operations for uniform meshes, but
pose a significant challenge to implement for adaptive meshes.

Furthermore, in variations of AMR so called hanging nodes, edges, and faces
may occur. These are nodes (or edges, or faces) of elements that do
not align with all nodes of the neighboring elements. Hanging nodes rather lie
on an edge or face of at least one neighbor element as in the middle and right
pictures in Figure~\ref{fig:unbltrAMR}. We call meshes with hanging nodes
non-conforming. 
When using AMR with hanging nodes, the numerical solvers have to be adapted 
to handle them properly.
See for example~\cite{RheinboldtMesztenyi80,
FischerKruseLoth02, KoprivaWoodruffHussaini02, AkcelikBielakBirosEtAl03,
LahnertBursteddeHolmEtAl16}.

By design AMR introduces a computational overhead, since an application cannot
spend 100\% of the computing time solving the actual problem, but has to invest
resources in managing the mesh. However, many experiments show that with
modern computers the overhead introduced by AMR is significantly smaller than
the gain in
runtime~\cite{DruiFiklKestenerEtAl16,RipleyLienYovanovich04,BursteddeGhattasGurnisEtAl10}.

In conclusion, using AMR certainly involves more elaborate algorithms, and a
significant amount of extra research has to be put into the meshing routines
compared to uniform meshes. This is a main reason why AMR routines are often
outsourced to external libraries and provides a core
motivation to implement the AMR library \tetcode in the context of this thesis.

\section{Unstructured AMR}
The three most commonly used methods for AMR are unstructured, block-structured,
and tree-based AMR, which we briefly compare now.

\begin{figure}
\center
\includegraphics[width=0.325\textwidth]{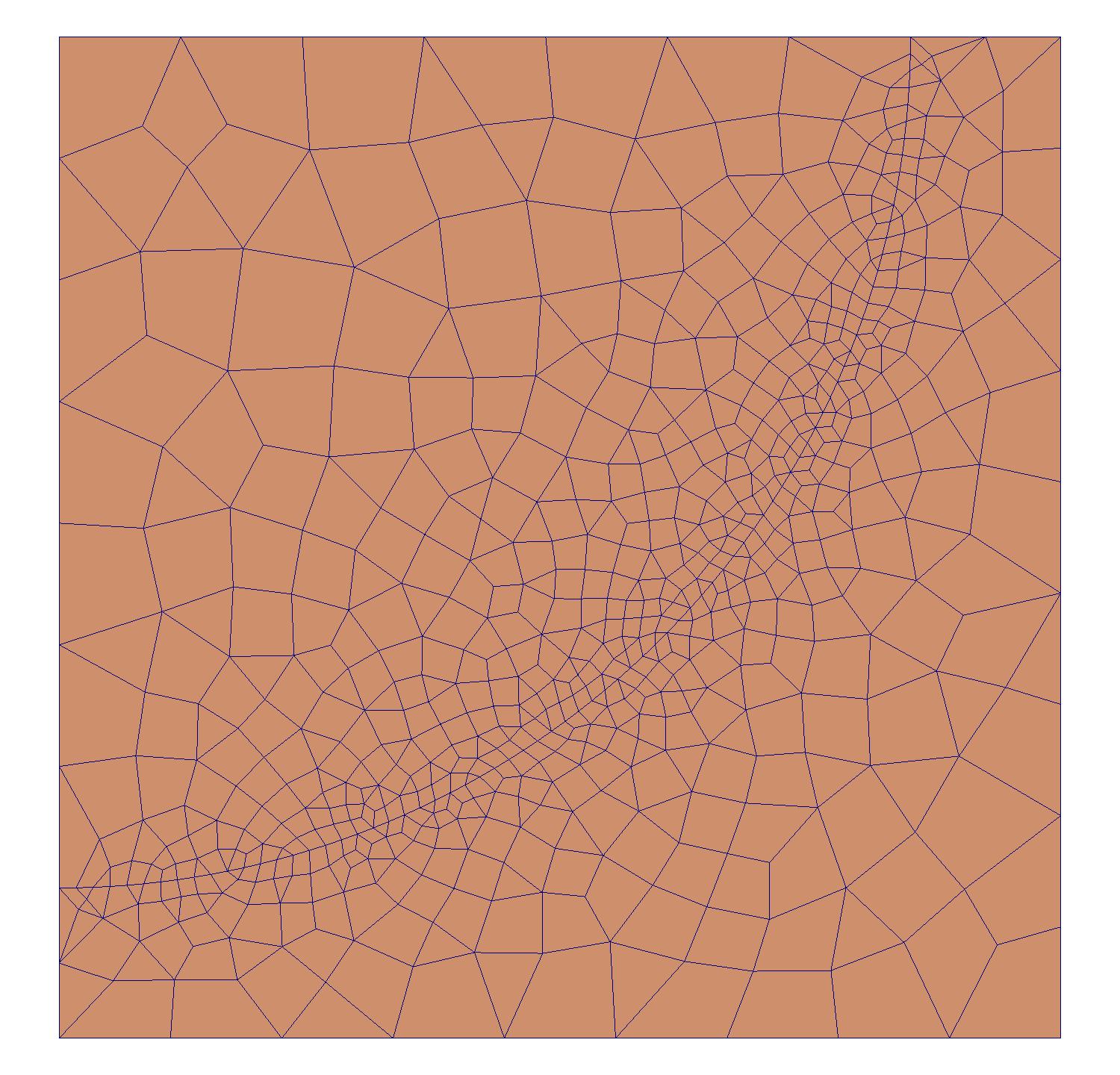}
\includegraphics[width=0.3215\textwidth]{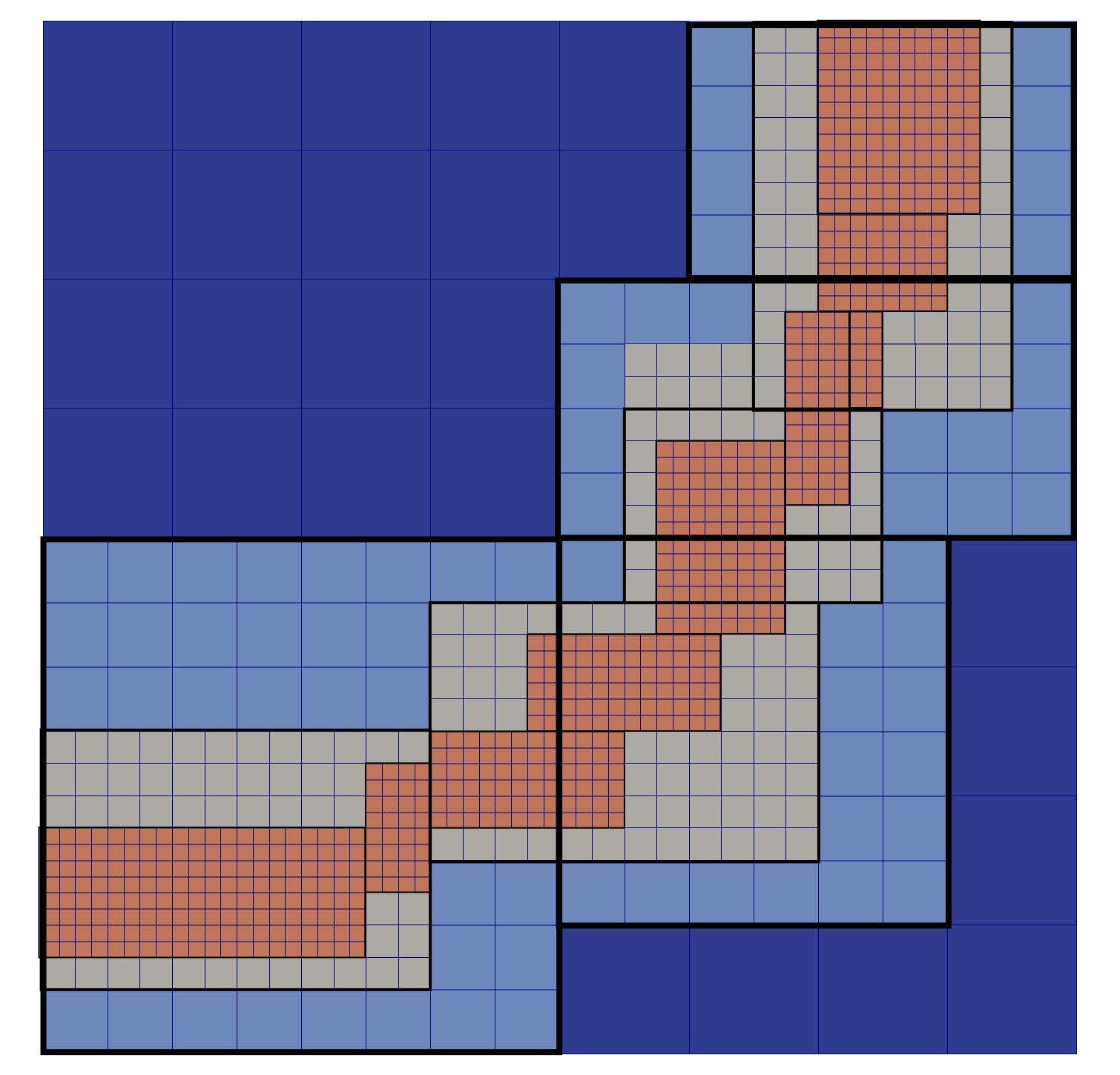}
\includegraphics[width=0.31\textwidth]{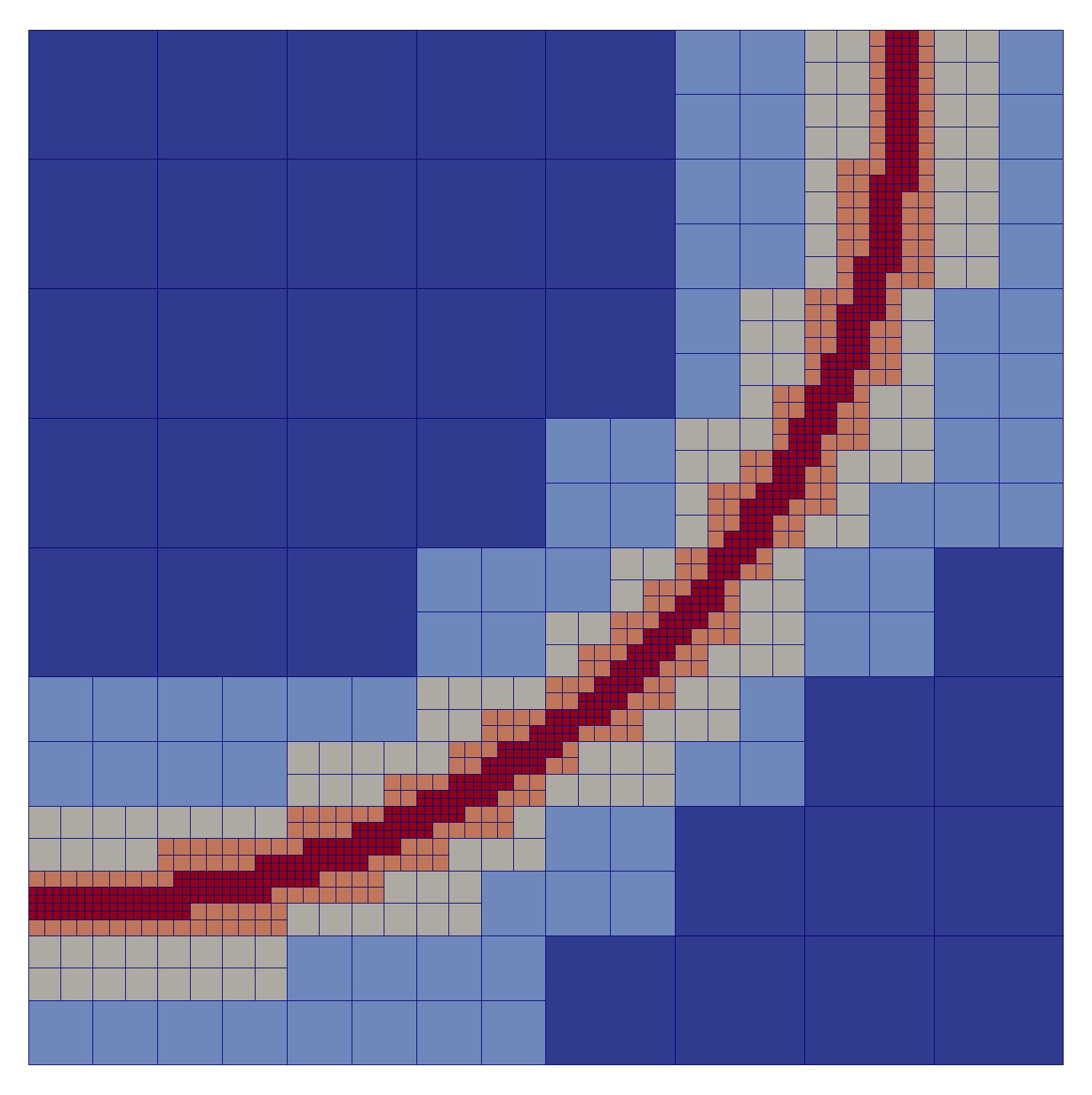}
  \caption[Three different AMR methods]
{Three different AMR methods for refining a quadrilateral mesh along a circular
arc. From left to right: unstructured AMR, block-structured AMR, and tree-based
AMR. For block-structured AMR we outline the boundaries of the rectangular
uniform patches in black. For block-structured and tree-based AMR, color refers
  to the refinement level of an element. In these two methods hanging nodes
  regularly occur when two elements of different refinement levels are
 neighbors.
}
\figlabel{fig:unbltrAMR}
\end{figure}

In unstructured AMR there is no regular pattern of connections between the
elements, meaning that the connectivity relations among elements are
arbitrary; see also Figure~\ref{fig:unbltrAMR} on the left.
These meshes are
usually conforming; that is, there are no hanging nodes/faces/edges.
Unstructured meshes can adapt arbitrarily well to domain
geometries. Therefore, they are often used when the domain is somewhat
complex~\cite{RasquinSmithChitaleEtAl14}.

However, the memory requirements for unstructured meshes are higher than for
other types of meshes. This is because 
for each element a list of its neighbors has to be
stored, as well as the coordinates of each vertex in the mesh.
Memory usage can be orders of magnitude higher compared to tree-based or block
structured AMR, where we store coordinates and connectivities only for the
elements of a coarse root mesh of trees that has significantly less elements
\cite{BaderBungartzWeinzierl13}.

Solvers using unstructured meshes often apply graph-based partitioning methods
to load-balance the mesh. Among the most common partitioner libraries are
ParMETIS~\cite{KarypisKumar98} and Scotch~\cite{ChevalierPellegrini08}.
Although these methods have been applied to billions of elements on millions
of processes, their runtime can be orders of magnitude slower compared to
partitioning structured
meshes~\cite{CatalyurekBomanDevineEtAl07,RasquinSmithChitaleEtAl14,SmithRasquinIbanezEtAl15}.

The most commonly used types of elements are triangles in 2D and tetrahedra in
3D, using the properties of Delaunay
triangulations~\cite{GoliasDutton97,Shewchuk96,HjelleDaehlen06}, but there
also exist unstructured mesh codes that use quadrilateral or hexahedral
elements~\cite{PeterKomatitschLuoEtAl11,KirkPetersonStognerEtAl06}.

Commonly used libraries for unstructured AMR are
\texttt{libMesh}~\cite{KirkPetersonStognerEtAl06},
\texttt{ParFUM}~\cite{LawlorChakravortyWilmarthEtAl06},
\texttt{PUMI}~\cite{IbanezSeolSmithEtAl16}, and
\texttt{PYRAMID}~\cite{NortonLyzengaParkerEtAl04}.
Some notable application codes are FEniCS~\cite{LoggMardalWells12},
PLUM~\cite{OlikerBiswas98,OlikerBiswasGabow00}, OpenFOAM~\cite{OpenCFD07}, and
MOAB from the SIGMA toolkit~\cite{TautgesMeyersMerkleyEtAl04}.

\section{Block-structured AMR}

In block-structured (or patched) AMR---see~\cite{BergerColella89,SkamarockOligerStreet89,RheinboldtMesztenyi80} for
example---the mesh is refined in rectangular patches as in the middle picture in
Figure~\ref{fig:unbltrAMR} resulting in a hierarchy of uniform rectangular
meshes. Refinement to level $l+1$ is only allowed inside a patch of level $l$
elements.
Thus, with block-structured AMR one can use some of the advantages of a uniform 
mesh within the patches. This comes at the cost of flexibility of the refinement
and possibly using more mesh elements than mathematically required.

Block-structured AMR is mainly used with quadrilateral and hexahedral meshes
and
some common software libraries and applications using block-structure AMR are
 \texttt{BoxLib} and its
successor \texttt{AMReX}~\cite{boxlib,amrex},
\texttt{Chombo}~\cite{ColellaGravesKeenEtAl07},
\texttt{Enzo}~\cite{BryanNormanOSheaEtAl14},
\texttt{SAMRAI}~\cite{WissinkHornungKohnEtAl01}, and
\texttt{Uintah}~\cite{BerzinsLuitjensMengEtAl10}.

\section{Tree-based AMR}

In this section we provide an overview over tree-based AMR. It is the
AMR technique that we use for the algorithms in this thesis.

We start with describing the main idea of tree-based AMR. We then present the
core mesh management algorithms and demonstrate in which order they may be
called by an application. At the end of this section, we introduce our notion
of high- and low-level algorithm and the API for general element types in
\tetcode.

\subsection{Overview}
The core idea of tree-based AMR is the following: Consider the unit square
as computational domain, meshed with a single quadrilateral coarse element.
This element represents the root of a refinement tree. We can refine it by
replacing it with four children and iterate this operation arbitrarily. Thus,
the resulting mesh can be represented as a refinement tree; see
Figure~\ref{fig:triangle-reftree}. 
Assigning each element a unique index, we encode the elements in the tree along
a so-called space-filling curve (SFC), 
We can thus store the elements linearly in an array in order of their SFC
index. 
Replacing an element with its children, or four children with its parent,
changes the order only locally. 
We discuss the theory of space-filling curves in more detail in
Chapter~\ref{ch:sfc}.

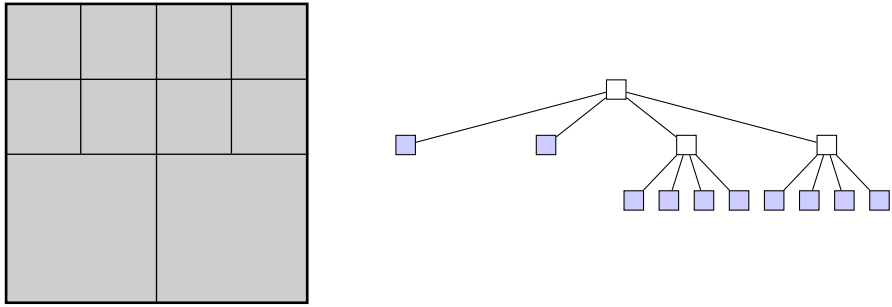
\begin{figure}
  \center
  \includegraphics[width=0.3\textwidth]{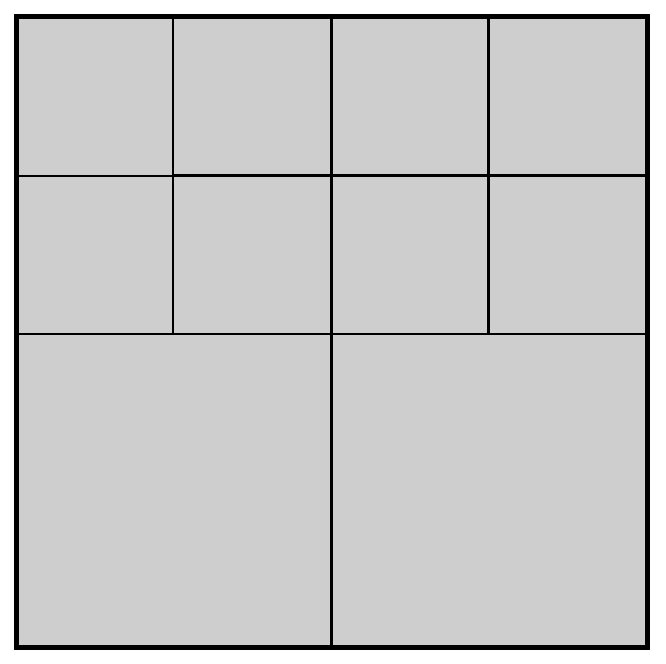}
\hspace{4ex}
  \raisebox{8ex}{
  \begin{tikzpicture}[yscale=0.4, xscale=0.4, %
                              level 1/.style={sibling distance=12em},%
                              level 2/.style={sibling distance=3em},%
                              level 3/.style={sibling distance=1.8em},%
                              every node/.style = {shape=rectangle,draw, %
                              minimum size=0.2em,anchor=north},%
                              leaf/.style = {draw,fill=blue!20}%
                              ]
   \node {}
   	child { node[leaf](j){}  }
   	child { node[leaf](j){}  }
    child { node{}  
    	child { node[leaf](a){}  }
    	child { node[leaf](b){}  }
    	child { node[leaf]{}}
    	child { node[leaf](e){}  }
    }
    child { node{}  
    	child { node[leaf](k){}  }
    	child { node[leaf]{}  }
    	child { node[leaf]{}  }
    	child { node[leaf](l){}  }
    };
   \end{tikzpicture}
  }
  \caption[A refined quadrilateral and the associated refinement tree.]
{
A refined quadrilateral and the associated refinement tree.
  Starting with the coarse quadrilateral (the root of the tree), we subdivide it
  into four children (the second row in the tree). We then further subdivide two
  of these level $1$ children into four level $2$ children each (the third row
  in the tree). The final mesh elements correspond to the leaves of the refinement
  tree.}
  \figlabel{fig:triangle-reftree}
\end{figure}

Suppose now that the computational domain is more complicated, for example the
wing of an airplane, or the mantle of the
earth~\cite{BursteddeGhattasGurnisEtAl10,RasquinSmithChitaleEtAl14}. There are
two approaches to model such domains using tree-based AMR.

In the first approach, we completely embed the domain inside a single
refinement tree and then refine the elements along the domain boundary up to a
desired accuracy. All elements that lie outside of the domain are then not
considered in the computation~\cite{LiIto06,GuntherMehlPoglEtAl06}. See also
Figure~\ref{fig:onetree-vs-cmesh} on the left. 

The second approach is to model the domain with an unstructured mesh of coarse
elements and then understand each coarse element as the root of one refinement
tree, giving rise to a forest of
elements~\cite{BangerthHartmannKanschat07,SteinmanMilnerNorleyEtAl03,StewartEdwards04}.

We call the unstructured mesh of tree roots the coarse mesh, which is created a
priori to map the topology and geometry of the domain with sufficient fidelity. 
Elements may then be refined and coarsened recursively, changing the
mesh below the root of each tree.
See Figure~\ref{fig:onetree-vs-cmesh} on the right.
In extreme cases, for example in industrial and medical applications, the
coarse meshes may consist of billions of trees
\cite{IbanezSeolSmithEtAl16,RasquinSmithChitaleEtAl14,FengTsolakisChernikovEtAl17}. 
The accuracy of the approximation of the domain may be further increased by
using curved tree edges with the help of higher order geometry functions; see
for example~\cite{HughesCottrellBazilevs05a, WilcoxStadlerBursteddeEtAl10,
ZhangGuChenEtAl09}.

When an application uses the second approach with a large coarse mesh,
this mesh has to be partitioned between the processes in order to decrease
the memory used by each process.  We will discuss coarse mesh partitioning in
Chapter~\ref{ch:cmesh}.

\begin{figure}
\center
  \raisebox{2.2ex}
  {\includegraphics[width=0.46\textwidth]{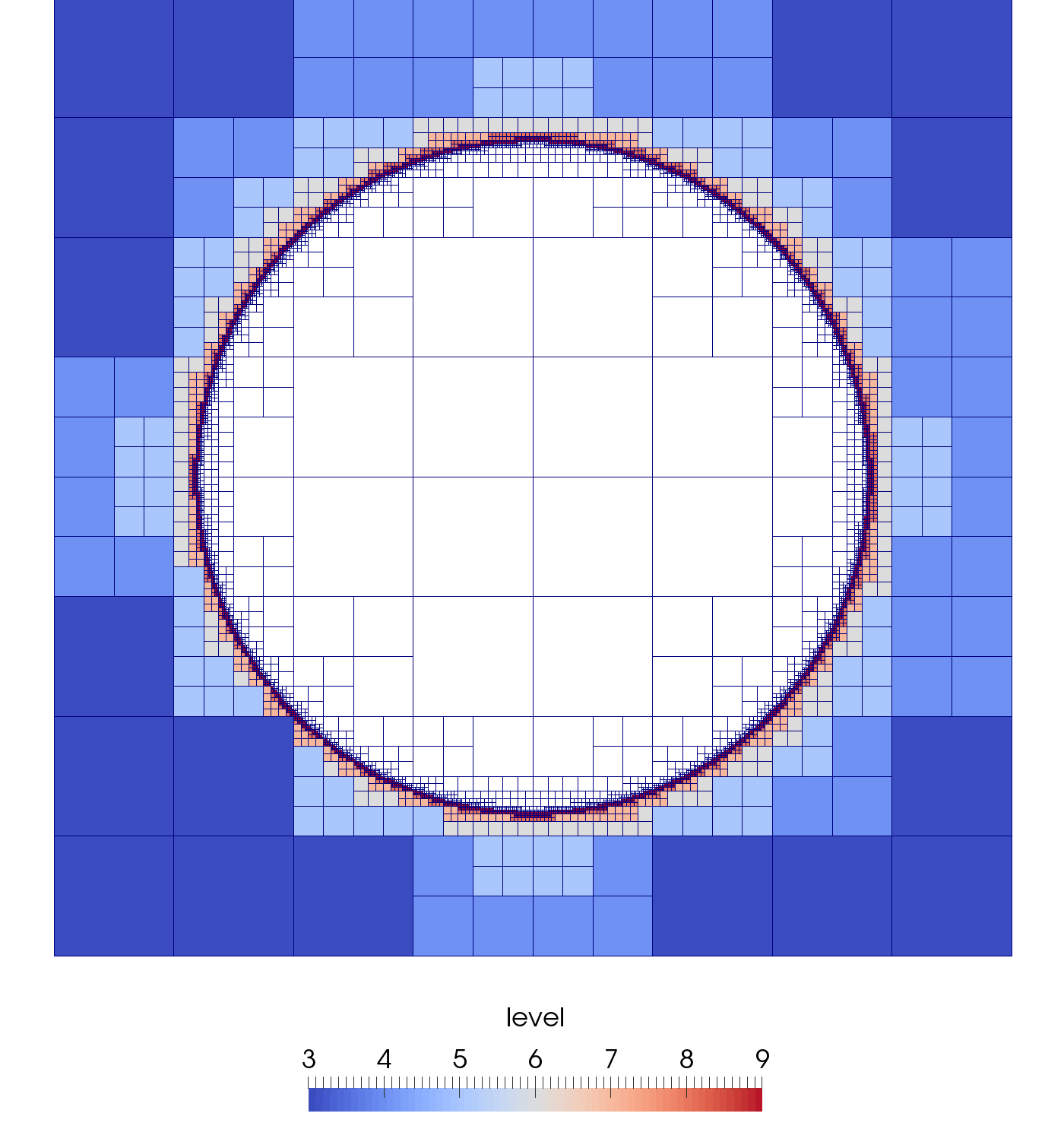}}\hfill
\includegraphics[width=0.446\textwidth]{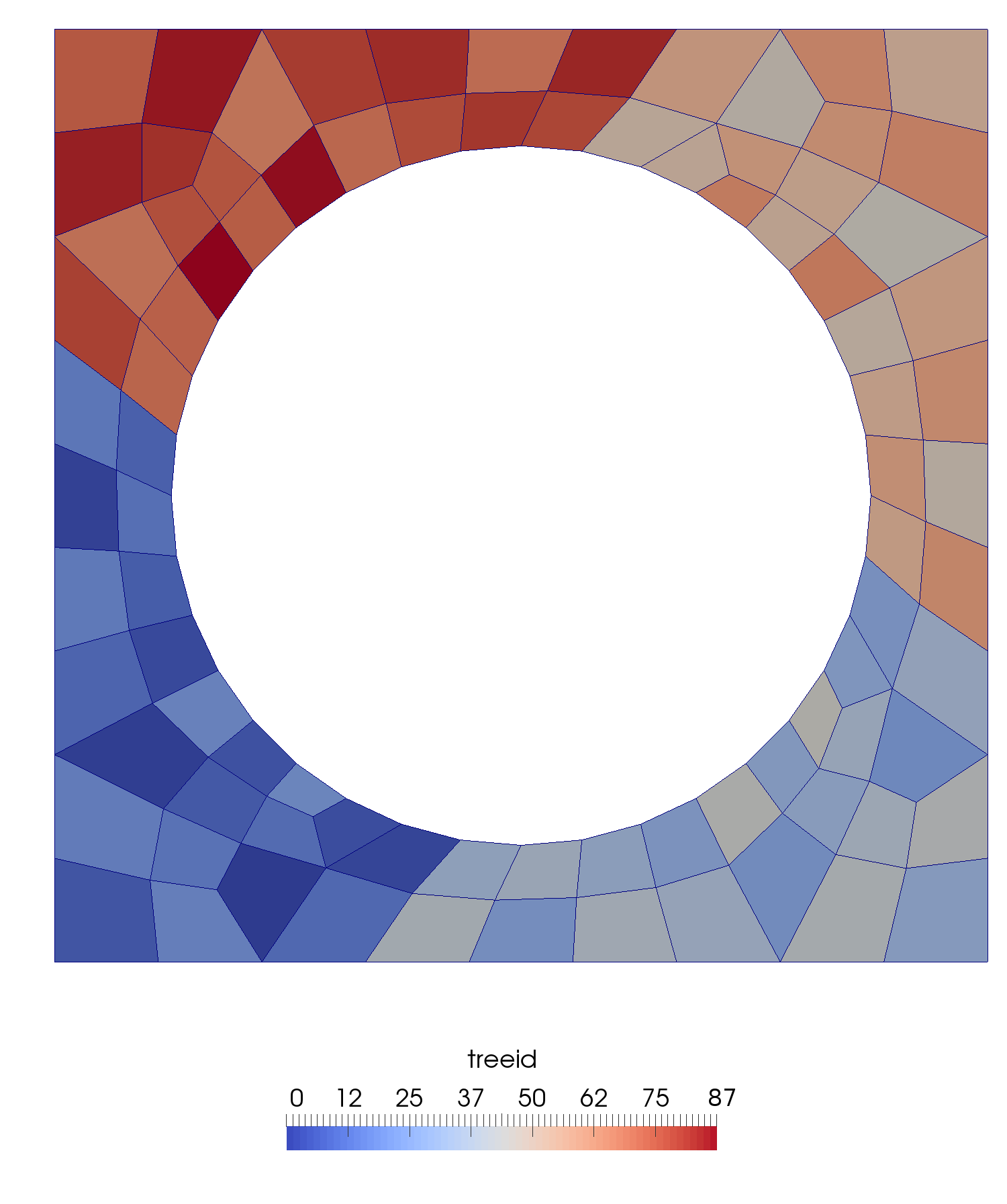}
\caption[Two ways to approximate a non-trivial domain.]
 {Two ways to approximate a non-trivial domain using tree-based AMR, in this
  case the region between
  a circular disk and a unit square. Left: The domain is embedded in a single
  refinement tree. This tree is then refined at the domain boundary. For a
  computation only the elements within the domain 
  and intersecting the domain boundary are considered. Color
  represents the refinement level. Right: The domain is modeled with an
  unstructured coarse mesh of trees. Each element in this coarse mesh
  represents one refinement tree. Color represents a consecutive enumeration
  of the trees. We generated this unstructured mesh with the \texttt{Gmsh}
  mesh generator
\cite{GeuzaineRemacle09}.}
\figlabel{fig:onetree-vs-cmesh}
\end{figure}

Thus, we have in fact two meshes:
The (unstructured) coarse mesh that stores the domain topology and the
tree-to-tree connections and the fine mesh that stores the actual computational
mesh of refined trees. Usually, the coarse mesh is obtained as the output of a
mesh generator, or it is constructed manually for small tree numbers. 

Since each coarse mesh element represents a refinement tree, we call the fine 
mesh forest mesh. Consequently, the tree-based AMR approach is also often
called forest-of-trees approach~\cite{BergerOliger84,BursteddeWilcoxGhattas11}.
The forest mesh stores an array of (process local) trees, and for each tree it stores 
the linear array of the process local fine elements.
We only store the finest elements of a forest and do not store any intermediate
elements between the coarse element and the finest elements. Since these
elements form the leaves of the refinement trees, we call them leaf elements.

Since each tree stores its own array of elements, the particular linear order
(SFC) used is local to the trees.
It is thus possible to combine trees of different kinds to form hybrid meshes,
for example quadrilaterals and triangles in 2D, or hexahedra, tetrahedra, prisms
(and pyramids) in 3D.
In theory it is also possible to use different types of linear orders for the
same kind of trees, for example Hilbert SFC~\cite{Hilbert91} in one quadrilateral tree and
Morton SFC~\cite{Morton66} in another. However, we do not discuss such applications in this
thesis.

\subsection{Core algorithms}
\label{sec:hlalgos}
We identify several core algorithms for tree-based AMR. 
\begin{itemize}
  \item \texttt{New} -- Generate a partitioned uniform mesh on a given geometry.
  \item \texttt{Adapt} -- Refine and coarsen mesh elements according to a
          given criterion.
  \item \texttt{Partition} -- Redistribute mesh elements among the processes in
          order to maintain a balanced load. This has a version for the coarse 
          mesh and one for the fine mesh.
  \item \texttt{Ghost} -- Construct and communicate a layer of ghost (halo) 
          elements for each process.
  \item \texttt{Balance} -- Establish a 2:1 balance in a mesh. That is, after
          \texttt{Balance}, each mesh element only has neighbors of the same 
          refinement level or at most one level higher or lower.
  \item \texttt{Iterate} -- Iterate through the mesh, executing a
          callback on each element and inter-element interface 
          (faces, edges, or vertices).
\end{itemize}

We discuss the typical pipeline of these algorithms with a solver application
as in Figure~\ref{fig:pipeline}.
In a preprocessing step, the geometry of the domain is meshed with a coarse mesh.
For small coarse meshes we can construct them in the AMR library itself
by explicitly giving the tree-to-tree connections.
This approach is not practical for larger coarse meshes, however, which is why we usually
use an external mesh generator for this task and feed its output into the AMR
library. Among the open source mesh generators we advise the reader to consider
\texttt{enGrid}~\cite{engrid}, 
\texttt{Gmsh}~\cite{GeuzaineRemacle09}, 
\texttt{NETGEN}~\cite{Schoeberl1997},
\texttt{TetGen}~\cite{Si06}, and
\texttt{Triangle}~\cite{Shewchuk96}.
In this step, we may already distribute large coarse meshes among several processes.

After preprocessing, the first step is to create a partitioned uniform forest
mesh on top of the coarse mesh with \texttt{New}. The initial uniform
refinement level depends on the application. Since this mesh is uniform, the
global number of elements and the number of elements per tree is known.  From
this, each process $p$ can calculate the index $E_p$ of its first element and then
locally create the elements with indices $E_p,\dots,E_{p+1}-1$. Thus,
\texttt{New} is completely distributed and does not involve communication.

The next step is usually an initial call to \texttt{Adapt} to create a first
adapted mesh according to an application's initial refinement criterion.
If a 2:1 balance condition is necessary, we call \texttt{Balance}.
We then often repartition the mesh with \texttt{Partition}. This is necessary
to maintain a balanced load, since the local element count of several (or all)
processes might change during \texttt{Adapt} and \texttt{Balance}. As a final
AMR step, we may call \texttt{Ghost} to create a layer of ghost elements.

At this point, the application carries out one or several solver steps, possibly
using \texttt{Iterate} to iterate through the mesh, for example to identify
degrees of freedom to assemble or apply matrices.
After solving, we may continue again with \texttt{Adapt} using an error
estimator and start a new cycle.

Throughout this thesis, we describe theoretical concepts and implementation
details of these algorithms. We describe our versions of \texttt{New} and
\texttt{Adapt} in Sections~\ref{sec:new} and \ref{sec:adapt}.
\texttt{Partition} is well-understood and we describe it in Section~\ref{sec:partitionsfc}.
\texttt{Ghost} and \texttt{Balance} are more complex than the other algorithms
and therefore we devote a whole chapter for each of them. We describe different
versions of \texttt{Ghost} in Chapter~\ref{ch:ghost} and discuss
\texttt{Balance} in Chapter~\ref{ch:balance}.
For \texttt{Iterate} see our remarks in Chapter~\ref{ch:app}.

\begin{remark}
  Typically, the application stores data for each mesh element, for example
  point values of an interpolated function.
  If the mesh changes because of adaptation, the data has to be projected
  onto the new mesh, which for example involves interpolating the function
  values to new finer mesh elements, or averaging function values of elements
  that are coarsened into the new value for the parent. 
  See Chapter~\ref{ch:app} for a more elaborate discussion.
\end{remark}
\begin{remark}
  The enumeration of the trees in the coarse mesh determines the order in which
  they are stored and thus their possible partitioning to the processes. The
  initial enumeration is given by the mesh generator and may be far from optimal
  for partitioning. For example, if the order is completely random, the
  partitions on the processes may contain arbitrarily many connected components.
  In order to reduce the number of connected components and obtain more optimal
  partitions, we may carry out one run of an unstructured mesh
  partitioner on the coarse mesh, for example ParMETIS~\cite{KarypisKumar95} or
  Scotch~\cite{ChevalierPellegrini08}, as an additional preprocessing step.
  We discuss the (re-)partitioning of the coarse mesh in detail in Chapter~\ref{ch:cmesh}.
\end{remark}

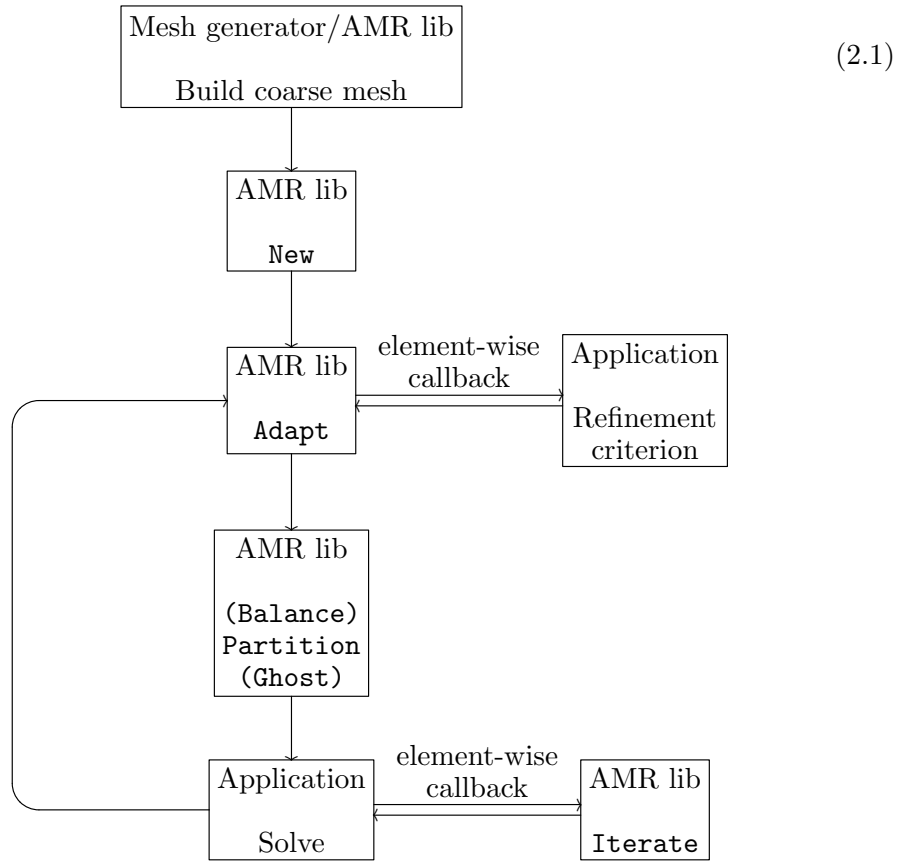
\begin{figure}
\begin{equation}
\xymatrix@C=8ex{
   &*+[F]{\txt{Mesh generator/AMR lib\\ \\Build coarse mesh}} \ar[d]
\\ &*+[F]\txt{AMR lib\\ \\\texttt{New}}\ar[d] 
\\ &*+[F]\txt{AMR lib\\ \\\texttt{Adapt}}\ar[d]\ar@<+2pt>[r]^-{\txt{element-wise\\callback}}
&*+[F]\txt{Application\\ \\Refinement\\criterion}\ar@<+2pt>[l]
\\ &*+[F]\txt{AMR lib\\ \\\texttt{(Balance)}\\\texttt{Partition}\\\texttt{(Ghost)}}\ar[d]
\\ &*+[F]\txt{Application\\ \\Solve}\ar `l[l]`[luu][uu]
\ar@<+2pt>[r]^-{\txt{element-wise\\callback}}  
& *+[F]\txt{AMR lib\\ \\\texttt{Iterate}}\ar@<+2pt>[l]
}
\end{equation}
  \caption[AMR algorithms pipeline]
  {This diagram represents a typical call pipeline for the high-level AMR
 algorithms of a solver application.
 At first a coarse mesh is constructed which models the domain geometry.
 On top of this coarse mesh we construct a distributed uniform forest mesh,
 which is then adapted according to a refinement criterion provided by
 the application. We can now balance and partition the mesh and create a layer
 of ghost elements. Before the application starts the solver.
 If the solver requires to readapt the mesh, the adapt/balance/partition/ghost 
 cycle starts again.}
\figlabel{fig:pipeline}
\end{figure}
\subsection{Independence of the element type}
\label{sec:highlow}

We distinguish between global (high-level) algorithms that operate on the
complete mesh, such as the core algorithms above, and local (low-level) algorithms
that operate on a single element. Examples for the latter type are computing
the children or parent of an element, computing the SFC index of an element, or
computing face-neighbors of an element.

A main strategy for the algorithms and the code developed in this thesis is 
the separation of the high-level algorithms from the low-level ones. Take for
example the algorithm \texttt{Adapt}. In \texttt{Adapt}, we iterate over all
leaf elements and call a user-provided function for each element and each family
of elements.
If the function returns positive, we refine the (first) element, thus constructing
its children. If the function returns negative on a family of elements, we
coarsen them, hence we construct the parent for one of the elements.
For an implementation of \texttt{Adapt}, however, it is irrelevant how exactly
the children or the parent are constructed. We only need to know when to call
the corresponding low-level algorithm and can use it as a blackbox.
In particular, the implementation of \texttt{Adapt} does not change for
different types of elements, for example 2D triangles and 3D hexahedra, as long
as we provide the appropriate low-level algorithms for the element type.

We carry out the separation of high- and low-level algorithms in the
implementation by defining a low-level API, a set of low-level algorithm
that define an element type. In \tetcode this low-level API is 
encoded in form of an abstract \texttt{C++} base class.
Each element type, such as quadrilaterals with Morton SFC index, or tetrahedra
with the TM-index, is an implementation of this base class.
Hence, introducing a new type of elements of SFC is achieved by simply
implementing their low-level algorithms. See for example the Bachelor's
thesis~\cite{Knapp17} in which prisms with a Morton type SFC are implemented.
 For a detailed
list of all low-level functions in \tetcode see Appendix~\ref{ch:appendix}.

Because of our separation of high- and low-level algorithms, the methods and
core algorithms can easily be applied to various kinds of elements and SFCs.
Furthermore, mixing elements of different types in the same mesh is possible, allowing
us to use hybrid meshes.

\chapter{Space-Filling Curves}
\label{ch:sfc}

In this chapter we introduce our notations for space-filling curves and
demonstrate their application to tree-based AMR. We discuss how SFCs offer us
a way to efficiently store and access the elements of a mesh in an array using
the element's SFC index. Furthermore, we show how to partition a mesh in
linear runtime to maintain an ideal load-balance using SFCs.

When talking about SFCs, we have to distinguish between SFCs in the analytical sense
and SFCs in the numerical (or discrete) sense.
In the analytical sense an SFC is a continuous mapping $f\colon \mathrm
I\rightarrow \IR^n$ from a compact set $\mathrm I \subset \IR$ into $\IR^n$
whose image $f(\mathrm I)$ has a positive $n$-dimensional volume
\cite{Bader12, Zumbusch00}. In most cases, $\mathrm I$ is the unit interval.

However, when SFCs are used in a numerical sense they are considered as maps
from or into a finite index set, which is a discrete version of the analytical
SFC. A typical approach is to define an analytical SFC as a limit of some
iteration rule and then define  an SFC in the numerical sense as a map resulting
from only finitely many iteration steps; see for example
\cite{Bader12,Sagan94,HaverkortWalderveen10} and the references therein.

In contrast to this indirect approach, we present a new formal definition of
(numerical) SFCs that is self-consistent and readily applies to adaptive mesh
refinement.

\section{Refinement spaces and refinements}

In order to define SFCs, we introduce refinement spaces and refinements.
As we have seen in the previous chapter, in tree-based AMR we associate the
elements resulting from a root element via refinement with the leaves of a
refinement tree.
In order to be independent of a particular geometric embedding of the mesh
elements, we define SFCs to be mappings on this refinement tree.
More precisely, we start with the set of all possible mesh elements, which 
naturally form the vertices of a tree.

\begin{definition}
  \label{def:refspace}
 A \textbf{refinement space} is a rooted tree with countably many vertices.
 Thus, it is a connected graph without circles and a distinguished vertex
 $\mathcal E$, the root.
 We call the vertices of the tree the \textbf{elements} of the refinement space.
\end{definition}
\begin{definition}
  \label{def:sfclevel}
 The \textbf{level} $\ell(E)$ of an element $E$ of a refinement space is its
 distance from the root vertex.
 By $\mathcal S^l$ we denote the set of all elements of level $l$. 

 The \textbf{parent} of an element $E$ is the unique element 
 $P$ with which $E$ shares an edge and for which  $\ell(E) = \ell(P)+1$.
 The root element is the only element that has no parent.
 
 Vice-versa, we say that $E$ is a \textbf{child} of $P$ if and only if $P$ is
 the parent of $E$.
\end{definition}
\begin{definition}
  \label{def:sfcrefmaps}
  Let $l\geq 0$.
 The $l$-th \textbf{refinement map} $R^l$ is a map
  $R^l\colon \mathcal S^l \rightarrow \mathcal P(\mathcal S^{l+1})$,
  mapping an element $E$ of level $l$ to the set of its children
  $R^l(E)$ (here $\mathcal P$ denotes the power set).
  We say that $E$ is refined into the elements $R^l(E)$.
\end{definition}

In fact, we can identify a refinement space solely from its elements, their
levels and the refinement maps.
\begin{proposition}
  \label{prop:refspace}
  A refinement space is equivalent to a triple $(\mathcal S, \ell, \mathcal
  R)$, where $\mathcal S$ is a set, 
  $\ell\colon\mathcal S\rightarrow \IN_0$ a map, and
  $\mathcal R=\setm{R^l}{l\in\IN_0}$ is a set of maps,
  $R^l\colon \mathcal S^l \rightarrow \mathcal P(\mathcal S^{l+1})$
  with $\mathcal S^l = \ell^{-1}(l)$,
  such that
  \begin{itemize}
    \item there exists exactly one element $\mathcal E\in\mathcal S$ with $\ell(\mathcal E) = 0$, and
    \item the image of $R^l$ is a partition of $\mathcal S^{l+1}$:
     \begin{align}
       R^l(E) \cap R^l(E') &= \emptyset \textrm{ for } E\neq E'\in\mathcal S^l,
        \label{eq:refinementa}\\
     \bigcup_{E\in\mathcal S^l} R^l(E) &= \mathcal S^{l+1}.\label{eq:refinementb}
     \end{align}
  \end{itemize}
\end{proposition}
\begin{proof}
  Given such a triple, we build a graph by connecting two elements of $\mathcal S$ if one is the
  refinement of the other.
  Properties~\eqref{eq:refinementa} and~\eqref{eq:refinementb} then directly
  imply that the graph has no circles and is connected. It is thus a tree and
  $\mathcal E$ is a possible choice for the root. The other direction of the
  equivalence follows from Definitions~\ref{def:sfclevel}
  and~\ref{def:sfcrefmaps}.
\end{proof}

By abuse of notation we also write $\mathcal S$ instead of $(\mathcal S,\ell,
\mathcal R)$.

\begin{remark}
  If a maximum level $\mathcal L = \max_{E\in\mathcal S}\ell(E)<\infty$ exists,
  then we obtain
  \begin{equation}
    \mathcal S^l = \emptyset \textrm{ for all } l > \mathcal L,
  \end{equation}
  and because of \eqref{eq:refinementb}, we get
  \begin{equation}
    \mathcal S^l \neq \emptyset \textrm{ for all } l\leq \mathcal L
  \end{equation} 
  as well as $R^\mathcal L(E) = \emptyset$ for each level
  $\mathcal L$ element $E\in\mathcal S$.
  In particular, such a maximum level always exists if there are finitely many
  elements.
  
  If a maximum level does not exist, then the set of elements must be infinite and
  we can also conclude
  \begin{equation}
    \mathcal S^l \neq \emptyset \textrm{ for all } l \in\IN_0.
  \end{equation} 
\end{remark}

The refinement maps define a hierarchy of the elements within a refinement space.
We introduce several notations for the relationships among elements.

\begin{definition}
  If $E \in R^l(E')$ in a refinement space, then we say that $E$ is a
\textbf{child} of $E'$ and $E'$ is the \textbf{parent} of $E$.
\end{definition}
Note that because of~\eqref{eq:refinementb} each element except the root element has a parent
and because of~\eqref{eq:refinementa} it is unique.
It is generally not true that each element of level $l<\mathcal L$ 
has a child, since we do not exclude the case $R^l(E)=\emptyset$.

\begin{definition}
  We say that $E$ is a \textbf{descendant} of $E'$ if $E$ results from $E'$ via
  successive refinement, thus $E=E'$, or $\ell(E)\geq \ell(E')$ and $E\in
  R^{\ell(E)-1}\circ\cdots\circ R^{\ell(E')}(E')$, 
  Furthermore, $E'$ is an \textbf{ancestor} of $E$ if and only if $E$ is a descendant of $E'$.
\end{definition}

\begin{definition}
 A \textbf{refinement} of a refinement space $\mathcal S$ is a subset
$\mathscr S\subset\mathcal S$ of elements that can be constructed from the 
root element via successively replacing a parent with its children. 
Thus, starting with the set $\mathcal S^0 = \set{\mathcal E}$,
a refinement is a set that can be constructed from it by applying the
following rule recursively:
\begin{itemize}
 \item If $\mathscr S$ is a refinement and $E\in\mathscr S$, then
  $\mathscr S\ohne \set{E} \cup R^{\ell(E)}(E)$ is a refinement.
\end{itemize}
\end{definition}

\begin{remark}
\label{rem:reftree}
Thus, in the language of trees, a refinement $\mathscr S$ consists of the leaves
of a subtree of the refinement space that contains the root element and for
each element either none or all of its children.  
\end{remark}
\begin{definition}
We call the tree from Remark~\ref{rem:reftree} the
\textbf{refinement tree} of $\mathscr S$, and the elements of $\mathscr S$ are
the \textbf{leaves} of the refinement.
\end{definition}

With this notion, we think of a refinement as an adaptive mesh, as it
could arise from any tree-based adaptive mesh application.
We display examples for 1:4 quadrilateral refinement in
Figures~\ref{fig:refspace-ex1} and~\ref{fig:refspace-ex2}.
\begin{remark}
  Each set $\mathcal S^l$ of level $l$ elements is a refinement, since we can construct it 
  from $\mathcal S^{l-1}$ by applying the refinement rule to each element. This
  iteration starts with $\mathcal S^0=\set{\mathcal E}$.
  We call the refinement $\mathcal S^l$ the \textbf{uniform refinement of level
  $l$}.
\end{remark}

We illustrate this definition with a geometrical example.
\begin{example}
 \label{ex:quad14}
  As an example we discuss the 1:4 refinement of quadrilateral elements.
  We fix a maximum level $\mathcal L$ and consider the scaled unit square 
  $[0, 2^\mathcal L]^2$. Starting with the root element $\mathcal E = [0,
  2^\mathcal L]^2$, all other elements in our refinement space are constructed
  by dividing an element into four subelements as in
  Figure~\ref{fig:refspace-ex1}, increasing their level by one. 
  The resulting refinement tree is a quadtree; each vertex
  of level $l<\mathcal L$ has exactly four children.
  To be formally correct, we describe the refinement space in terms of the
  level and the refinement maps.
  Let $\mathcal S^l$ be the set of all subsquares of $\mathcal E$ of side
  length $2^{\mathcal L -l}$ with coordinates an integer multiple of
  $2^{\mathcal L -l}$. We can identify such an element with its lower left
  corner $C=(m2^{\mathcal L -l},n2^{\mathcal L-l})$
  and its level $l$. Here, $m,n\in\IN_0,\, 0\leq n,m < 2^l$.
  The set $\mathcal S$  is defined as the union of all $\mathcal S^l$, $0\leq
  l\leq\mathcal L$.
  We describe the refinement maps $R^l$ in terms of the corner coordinates by
  \begin{equation}
   \begin{split}
    R^l((m2^{\mathcal L -l},n2^{\mathcal L -l})) =
    {\left\lbrace \vphantom{2^\mathcal L}\right.}%
    &(2m2^{\mathcal L -l-1},2n2^{\mathcal L -l-1}),\\
    &(2m2^{\mathcal L -l-1},(2n+1)2^{\mathcal L -l-1}),\\
    &((2m+1)2^{\mathcal L -l-1},2n2^{\mathcal L -l-1}),\\
    &((2m+1)2^{\mathcal L -l-1},(2n+1)2^{\mathcal L -l-1})
    {\left. \vphantom{2^\mathcal L}\right\rbrace}.%
   \end{split}
  \end{equation}
\end{example}

\begin{figure}
 \center
\def\svgwidth{0.6\textwidth}
\begingroup%
  \makeatletter%
  \providecommand\color[2][]{%
    \errmessage{(Inkscape) Color is used for the text in Inkscape, but the package 'color.sty' is not loaded}%
    \renewcommand\color[2][]{}%
  }%
  \providecommand\transparent[1]{%
    \errmessage{(Inkscape) Transparency is used (non-zero) for the text in Inkscape, but the package 'transparent.sty' is not loaded}%
    \renewcommand\transparent[1]{}%
  }%
  \providecommand\rotatebox[2]{#2}%
  \ifx\svgwidth\undefined%
    \setlength{\unitlength}{540.42788086bp}%
    \ifx\svgscale\undefined%
      \relax%
    \else%
      \setlength{\unitlength}{\unitlength * \real{\svgscale}}%
    \fi%
  \else%
    \setlength{\unitlength}{\svgwidth}%
  \fi%
  \global\let\svgwidth\undefined%
  \global\let\svgscale\undefined%
  \makeatother%
  \begin{picture}(1,0.43632568)%
    \put(0,0){\includegraphics[width=\unitlength,page=1]{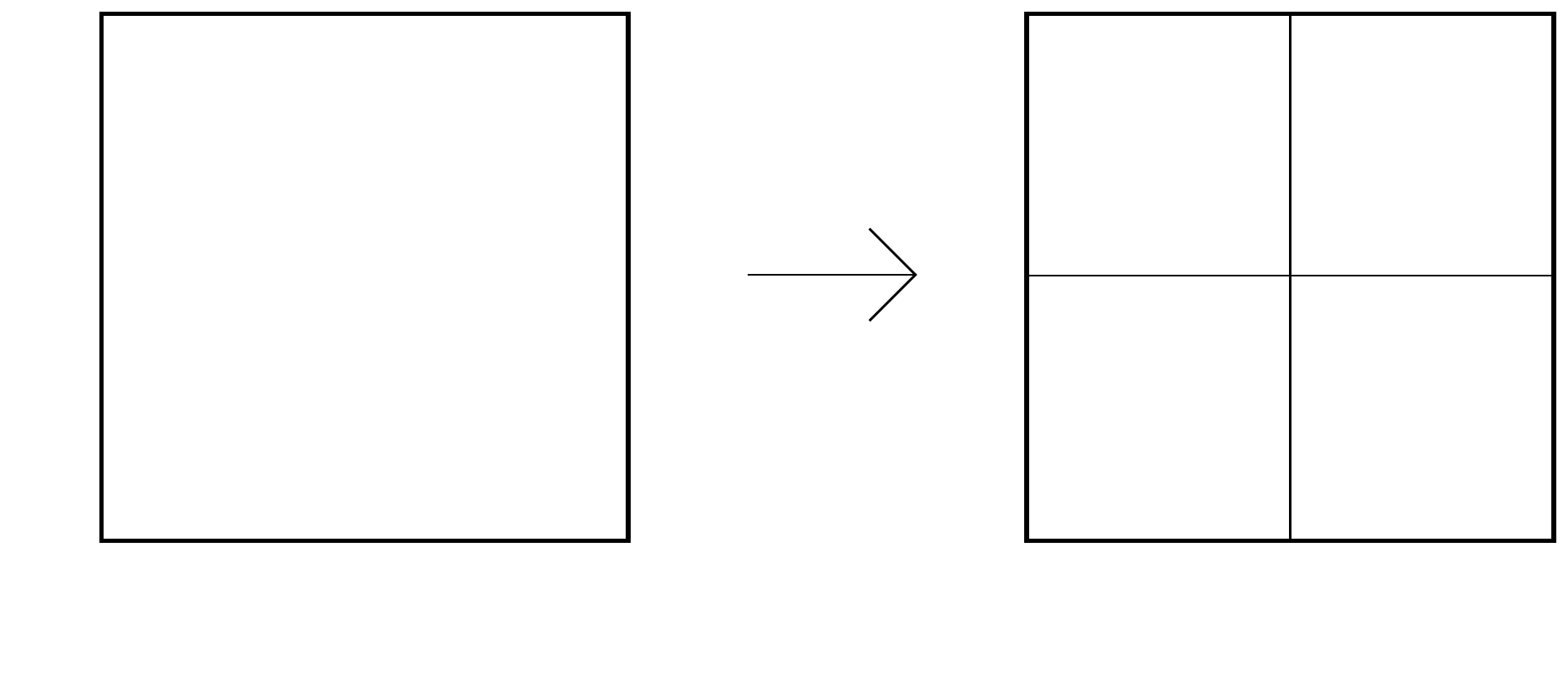}}%
    \put(0.50572059,0.29923624){\color[rgb]{0,0,0}\makebox(0,0)[lb]{\smash{$R^l$}}}%
    \put(0,0){\includegraphics[width=\unitlength,page=2]{refinesquare_tex.pdf}}%
    \put(0.00641853,0.06254479){\color[rgb]{0,0,0}\makebox(0,0)[lb]{\smash{$C$}}}%
    \put(0,0){\includegraphics[width=\unitlength,page=3]{refinesquare_tex.pdf}}%
    \put(0.16686816,0.00193135){\color[rgb]{0,0,0}\makebox(0,0)[lb]{\smash{$2^{\mathcal L - l}$}}}%
    \put(0,0){\includegraphics[width=\unitlength,page=4]{refinesquare_tex.pdf}}%
    \put(0.67581298,0.00193135){\color[rgb]{0,0,0}\makebox(0,0)[lb]{\smash{$2^{\mathcal L - l-1}$}}}%
    \put(0,0){\includegraphics[width=\unitlength,page=5]{refinesquare_tex.pdf}}%
    \put(0.596902,0.06290403){\color[rgb]{0,0,0}\makebox(0,0)[lb]{\smash{$C$}}}%
  \end{picture}%
\endgroup%
\hfill
\def\svgwidth{0.27\textwidth}
\begingroup%
  \makeatletter%
  \providecommand\color[2][]{%
    \errmessage{(Inkscape) Color is used for the text in Inkscape, but the package 'color.sty' is not loaded}%
    \renewcommand\color[2][]{}%
  }%
  \providecommand\transparent[1]{%
    \errmessage{(Inkscape) Transparency is used (non-zero) for the text in Inkscape, but the package 'transparent.sty' is not loaded}%
    \renewcommand\transparent[1]{}%
  }%
  \providecommand\rotatebox[2]{#2}%
  \ifx\svgwidth\undefined%
    \setlength{\unitlength}{240.24689941bp}%
    \ifx\svgscale\undefined%
      \relax%
    \else%
      \setlength{\unitlength}{\unitlength * \real{\svgscale}}%
    \fi%
  \else%
    \setlength{\unitlength}{\svgwidth}%
  \fi%
  \global\let\svgwidth\undefined%
  \global\let\svgscale\undefined%
  \makeatother%
  \begin{picture}(1,0.98134721)%
    \put(0,0){\includegraphics[width=\unitlength,page=1]{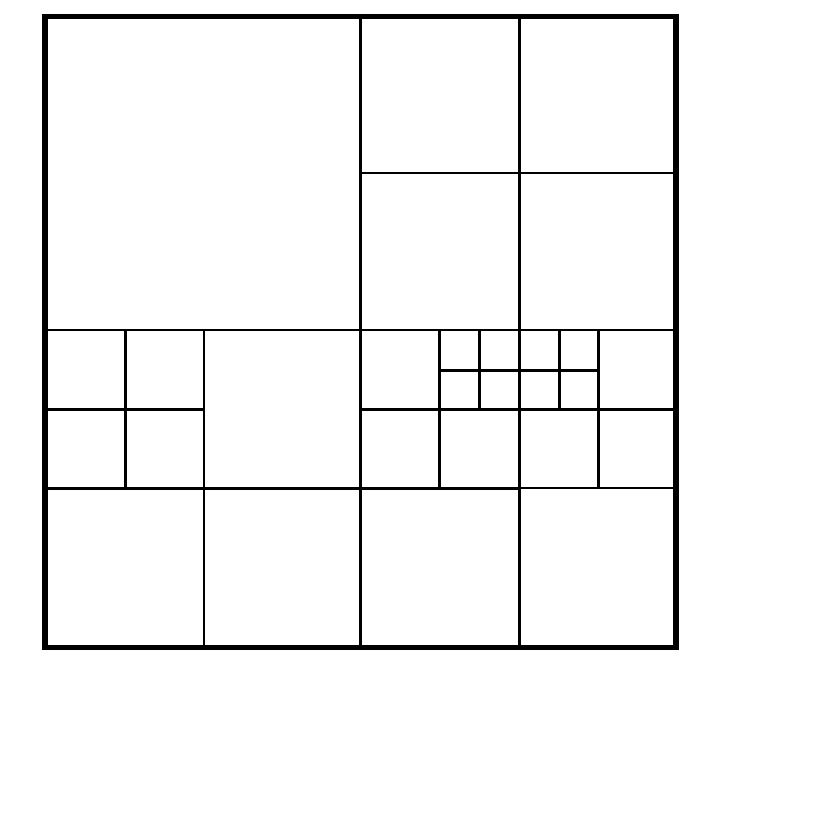}}%
    \put(0.01443828,0.09094114){\color[rgb]{0,0,0}\makebox(0,0)[lb]{\smash{$0$}}}%
    \put(0.77199233,0.09049284){\color[rgb]{0,0,0}\makebox(0,0)[lb]{\smash{$2^\mathcal L$}}}%
  \end{picture}%
\endgroup%
   \caption[1:4 refinement of a quadrilateral element.]
  {1:4 refinement of quadrilateral elements.
  Left: The refinement rule for an element of level $l<\mathcal L$ with
  anchor node $C$. The $x$ and $y$ coordinates of $C$ are integer multiples
  of $2^{\mathcal L - l}$ and lie within $[0,2^\mathcal L]^2$.
  Right: An example for a refinement in the refinement space.}
  \figlabel{fig:refspace-ex1}
\end{figure}

\section{Space-filling curves}

In short,  an SFC index is a map from a refinement space into the natural
numbers that fulfills certain properties.

\begin{definition}
\label{def:sfcindex}
  A \textbf{space-filling curve index} on a refinement space $\mathcal S$ is a map
  \begin{equation}
    \mathcal I\colon \mathcal S \rightarrow \IN_0
  \end{equation}
  that fulfills the following properties for any $E, E',\hat E\in \mathcal S$:
  \begin{enumerate}[(i)]
    \item The map $\mathcal I \times \ell\colon\mathcal
     S\rightarrow\IN_0\times\IN_0$ is injective. Thus, the index and the level
     uniquely determine an element of
      $\mathcal S$.
    \item If $E$ is an ancestor  of $E'$ then $\mathcal I(E)\leq \mathcal I(E')$.
         Hence, refining an element cannot decrease its SFC index.
    \item If $\mathcal I(E) < \mathcal I(\hat E)$ and $\hat E$ is not a descendant of $E$,
      then $\mathcal I(E) \leq \mathcal I(E') < \mathcal I(\hat E)$ for all descendants $E'$ of $E$.
      Therefore, refining an element is a 'local' operation in terms of the index.
  \end{enumerate}
\end{definition}

Restricted to the leaves of a refinement,  an SFC index becomes injective:

\begin{proposition}
  \label{prop:sfcinj}
  Let $\mathscr S$ be a refinement in a refinement space $\mathcal S$ with SFC index
  $\mathcal I$. Then any two leaves in $\mathscr S$ have different SFC indices, thus
  \begin{equation}
  E\neq E'\in\mathscr S \Rightarrow \mathcal I(E) \neq \mathcal I(E').
  \end{equation}
  \begin{proof}
    Let $E\neq E'\in\mathscr S$.
    Let $P$ be the nearest common ancestor of $E$ and $E'$, that is, the element
    in $\mathcal S$ of greatest level that is both an ancestor of $E$ and $E'$.
    Since the root element is ancestor of all elements, $P$ must exist and because
    each element has a unique parent, $P$ is also unique.
    Now $P\neq E$ and $P\neq E'$ since otherwise $E'$ would be a descendant of $E$ or
    vice-versa, which is a contradiction with $\mathscr S$ being a refinement.
    Furthermore, since $P$ is the nearest common ancestor of $E$ and $E'$, $P$
    must have two children $P_E\neq P_{E'}$ that are ancestors of
    $E$ and $E'$. Since $\ell(P_E) = \ell (P_{E'}) = \ell(P) + 1$, we know 
    from the injectivity of $\mathcal I\times\ell$ that the indices 
    of $P_E$ and $P_{E'}$ are not equal. We assume $\mathcal I(P_E)<\mathcal
    I(P_{E'})$ without loss of generality. 
    From property (ii) of Definition~\ref{def:sfcindex} we
    conclude that $\mathcal I(P_{E'})\leq \mathcal I(E')$ and with property
    (iii) follows
    \begin{equation}
      \mathcal I(P_E)\leq\mathcal I(E)<\mathcal I(P_{E'}),
    \end{equation}
    and thus $\mathcal I(E)<\mathcal I(P_{E'}) \leq\mathcal I(E')$,
    proving the claim.
  \end{proof}
\end{proposition}

Proposition~\ref{prop:sfcinj} gives us the theoretical justification for the desired
one-dimensional storage scheme of a refinement.  Since for any refinement the
SFC index is injective on the leaves, we can uniquely arrange the leaves in an
array such that the order induced by the SFC index is preserved.
\begin{corollary}
  \label{cor:consind}
  Let $\mathscr S$ be a refinement of $\mathcal S$ with SFC index $\mathcal I$
and let $N=|\mathscr S|$ be the number of leaves.  Then there exists a unique
bijective map
  \begin{equation}
    \label{eq:consindex}
    \mathcal I_{\mathscr S}\colon \mathscr S \rightarrow \set{0,\dots,N-1},
  \end{equation}
  that is monotonous under $\mathcal I$, thus
  \begin{equation}
    \mathcal I_{\mathscr S} (E) < 
    \mathcal I_{\mathscr S} (E') 
    \Leftrightarrow
    \mathcal I (E) < 
    \mathcal I (E').
  \end{equation}
  The statement remains true if $N=\infty$, in which case 
  the right-hand side of equation~\eqref{eq:consindex} is $\IN_0$.
\end{corollary}

\begin{definition}
  \label{def:consindex}
  We call this $\mathcal I_\mathscr S$ the \textbf{consecutive Index} of
  $\mathscr S$ with regard to $\mathcal I$.
\end{definition}

\section{The Morton space-filling curve} 
Most SFCs in use arise from a recursive pattern. The root element is 
a polytope; most common are lines, triangles, quadrilaterals, tetrahedra,
hexahedra, prisms, and pyramids. This root element is then subdivided
into children of the same polytope type~\footnote{With the exception of pyramids,
which may be subdivided into pyramids and tetrahedra.},
which are then provided with a local order, i.e.\ per parent.
 This refinement pattern is then applied recursively to the children, with the
possibility of rotations and/or reflections. See also
Figure~\ref{fig:mortonchildids}. 

For a detailed overview of SFCs we refer the reader
to~\cite{HaverkortWalderveen10,Haverkort17} and the references therein.
Here, we discuss the Morton SFC for quadrilateral/hexahedral elements as an
example.

\label{sec:mortonsfc}
The Morton curve (or $Z$-curve) for quadrilateral (2D) and hexahedral (3D)
elements was first described by Lebesgue~\cite{Lebesgue04} and its applications
to data storage were discussed by G.\ M.\ Morton~\cite{Morton66}. The
refinement space for the Morton index results from 1:4 refinement in 2D and
1:8 refinement in 3D.  In general, the Morton index can be defined in any
space dimension $n$ on the $n$-dimensional hypercube with 1:$2^n$ refinement.
Let us now consider dimension $n = 2$. We already discussed in
Example~\ref{ex:quad14} that each element $Q$ of the refinement space $\mathcal
S$ is identified with its level and its lower left corner coordinates $(x,y)$.
We call this corner the anchor node of $Q$.
These coordinates are integers in $[0,2^\mathcal L)\cap \IN_0$ and as such
they posses a binary representation of length $\mathcal L$:
\begin{equation}
x = \sum_{j = 0}^{\mathcal L - 1} x_j2^j 
= ( x_{\mathcal L -1}  x_{\mathcal L -2}  \dots x_0)_2,
\quad y = \sum_{j=0}^{\mathcal L - 1}y_j2^j 
= ( y_{\mathcal L -1} y_{\mathcal L -2} \dots y_0)_2,
\end{equation}
with $x_j,y_j\in\set{0,1}$.
The Morton index $m\colon \mathcal S \rightarrow \IN_0$ is defined by mapping
these coordinates to their bitwise interleaving,
\begin{equation}
\label{eq:mortoninter}
m(Q) := (
y_{\mathcal L -1}x_{\mathcal L-1}
y_{\mathcal L -2}x_{\mathcal L-2} \dots
 y_0x_0
)_2\in[0,2^{2\mathcal L}).
\end{equation}
See also Figure~\ref{fig:excomputemorton}. This scheme extends to higher dimensions $n>2$
via bitwise interleaving all $n$ coordinates of the anchor node. The Morton
index is  an SFC index in the sense of
Definition~\ref{def:sfcindex}~\cite{SundarSampathBiros08}.

\begin{figure}
  \center
  \includegraphics[width=0.3\textwidth]{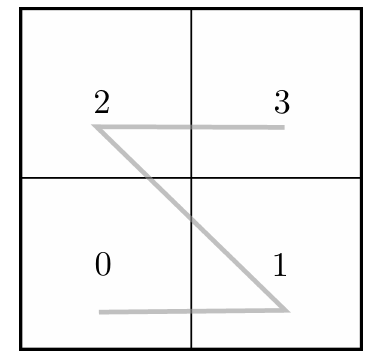}%
  \hspace{5ex}
  \includegraphics[width=0.3\textwidth]{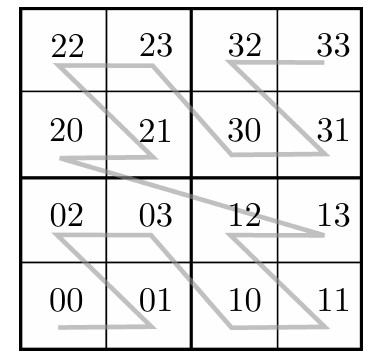}%
  \caption[Child-ids for the cubical Morton index]
  {The child-ids for the Morton index. 
   We label the four children of a quadrilateral with $0, 1, 2, 3$ in $Z$-order
   (left). We recursively apply this scheme to all descendants, appending the
   child-id on each level (right) to compute the Morton index. This is an alternate method
   to the bitwise interleaving~\eqref{eq:mortoninter}.
   See also
   equation~\eqref{eq:mortonviacid} and Figure~\ref{fig:excomputemorton}.
 }
  \figlabel{fig:mortonchildids}
\end{figure}

\begin{figure}
 \center
 \begin{minipage}{0.4\textwidth}
  \def\svgwidth{\textwidth}
\begingroup%
  \makeatletter%
  \providecommand\color[2][]{%
    \errmessage{(Inkscape) Color is used for the text in Inkscape, but the package 'color.sty' is not loaded}%
    \renewcommand\color[2][]{}%
  }%
  \providecommand\transparent[1]{%
    \errmessage{(Inkscape) Transparency is used (non-zero) for the text in Inkscape, but the package 'transparent.sty' is not loaded}%
    \renewcommand\transparent[1]{}%
  }%
  \providecommand\rotatebox[2]{#2}%
  \ifx\svgwidth\undefined%
    \setlength{\unitlength}{227.34702148bp}%
    \ifx\svgscale\undefined%
      \relax%
    \else%
      \setlength{\unitlength}{\unitlength * \real{\svgscale}}%
    \fi%
  \else%
    \setlength{\unitlength}{\svgwidth}%
  \fi%
  \global\let\svgwidth\undefined%
  \global\let\svgscale\undefined%
  \makeatother%
  \begin{picture}(1,0.92789623)%
    \put(0,0){\includegraphics[width=\unitlength,page=1]{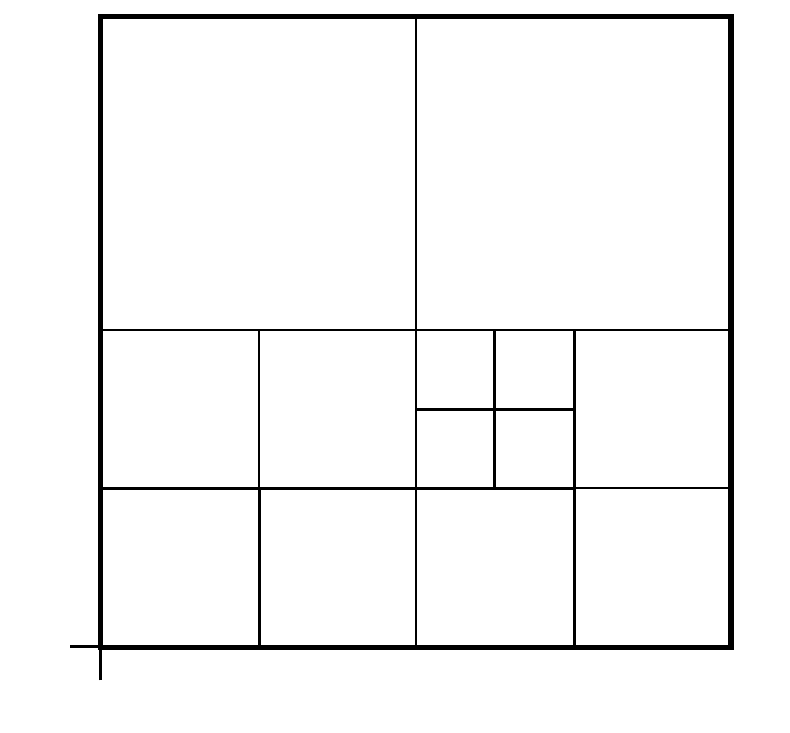}}%
    \put(0.01496542,0.09672958){\color[rgb]{0,0,0}\makebox(0,0)[lb]{\smash{$0$}}}%
    \put(0.0904747,0.02357696){\color[rgb]{0,0,0}\makebox(0,0)[lb]{\smash{$0$}}}%
    \put(0,0){\includegraphics[width=\unitlength,page=2]{mortoncode_tex.pdf}}%
    \put(0.80425054,0.02427781){\color[rgb]{0,0,0}\makebox(0,0)[lb]{\smash{$2^\mathcal{L} = 2^4$}}}%
    \put(0.6473233,0.33661513){\color[rgb]{0,0,0}\makebox(0,0)[lb]{\smash{$Q$}}}%
    \put(0,0){\includegraphics[width=\unitlength,page=3]{mortoncode_tex.pdf}}%
    \put(0.5902079,0.02224879){\color[rgb]{0,0,0}\makebox(0,0)[lb]{\smash{$10$}}}%
    \put(0,0){\includegraphics[width=\unitlength,page=4]{mortoncode_tex.pdf}}%
    \put(0.03995419,0.2994739){\color[rgb]{0,0,0}\makebox(0,0)[lb]{\smash{$4$}}}%
    \put(0,0){\includegraphics[width=\unitlength,page=5]{mortoncode_tex.pdf}}%
    \put(0.15003019,0.02567488){\color[rgb]{0,0,0}\makebox(0,0)[lb]{\smash{$1$}}}%
    \put(0.21051875,0.02567488){\color[rgb]{0,0,0}\makebox(0,0)[lb]{\smash{$2$}}}%
    \put(0.26136558,0.02567488){\color[rgb]{0,0,0}\makebox(0,0)[lb]{\smash{$3$}}}%
    \put(0,0){\includegraphics[width=\unitlength,page=6]{mortoncode_tex.pdf}}%
  \end{picture}%
\endgroup%
  \end{minipage}
 \begin{minipage}{0.54\textwidth}
   \begin{equation}
    \label{eq:mortonviainter}
     \begin{split}
      (\green x,\red y) &= (\green{10},\red 4) = ((\green{1010})_2,(\red{0100})_2)\\
   \Rightarrow  m(Q) &= (\red 0\green1\red1\green0\red0\green1\red0\green0)_2 = (100)_{10}\\
     \end{split}
   \end{equation}
    \begin{equation}
      \label{eq:mortonviapath}
          m(Q) = (1210)_4 = (100)_{10}
    \end{equation}
 \end{minipage}
\caption[Computing the Morton index of a quad.]{Two ways of computing the
Morton index $m(Q)$ of a quadrant $Q$ in the quadrilateral refinement with
maximum level $\mathcal L = 4$. The quadrant's anchor node has coordinates $x =
10$ and $y = 4$.  We can compute the Morton index via bitwise interleaving as
in equation~\ref{eq:mortonviainter}.
A second way to compute the $m(Q)$ is via $Q$'s refinement path.
We can construct $Q$ from the root element by taking its first child, 
then the second child of that first child, and finally taking the first child
of this quadrant. This leads to the sequence $(121)$ of child-ids.
We append zeroes until we reach the length $\mathcal L$ and interpret it
as a quarternary number, leading to $m(Q) = (1210)_4$.}
\figlabel{fig:excomputemorton}
\end{figure}

We also discuss a recursive way to describe the Morton index.
 We label the four children of a single quadrant
with 0, 1, 2, and 3 in Z-order, the so-called child-id. Hence, the lower left
child has id 0, the lower right child 1, the upper left child 2, and the upper
right child 3; see Figure~\ref{fig:mortonchildids}.

Each element $E$ in the refinement space is constructed from the root element
via successive refinement. These refinements are unique and thus we obtain a
sequence of child-ids $(c_1,c_2,\dots c_{\ell(E)})$ describing this
refinement process. It can be read as: Start with the root element, take its
child $c_1$, from this child take the child $c_2$, and so forth. The element
itself has child-id $c_{\ell(E)}$ with regard to its parent.

We can now extend this sequence up to length $\mathcal L$ by filling up with zeroes
and since $c_i\in\set{0,1,2,3}$, we can interpret it as a quarternary number:
\begin{equation}
\label{eq:mortonviacid}
  (c_1c_2\dots c_{\ell(E)} 0 \dots 0)_4 \in [0,4^\mathcal L) = [0,2^{2\mathcal L}).
\end{equation}
It is straightforward to show that this is exactly the Morton index of $E$, \cite{TropfHerzog81,BursteddeWilcoxGhattas11}:
\begin{equation}
  m(E) = (c_1c_2\dots c_{\ell(E)} 0 \dots 0)_4.
\end{equation}
From this we obtain the correspondence
\begin{equation}
  (y_ix_i)_2 = c_{\mathcal L - i},
\end{equation}
meaning that we can read of the child-id of $E'$s ancestor at level $i+1$ from
the $(\mathcal L-i)$-th bits of the anchor node coordinates. See also
Figure~\ref{fig:excomputemorton}.
We refer to Figure~\ref{fig:refspace-ex2} for an example of the SFC arising
from the Morton index.

Due to the bitwise interleaving, the Morton index of an element can be computed
in constant time and is efficient to implement. Since the index is given
implicitly by the anchor node coordinates, most algorithms that operate with
the Morton code
do not have to carry out the actual bit interleaving but work with the
coordinates instead; see for example~\cite{BursteddeWilcoxGhattas11}.
Another advantage of the Morton index is that it is memory efficient, storing only
the anchor node coordinates and level of an element.
In addition, many low-level algorithms such as finding an element's parent, refining
an element, or constructing its face-neighbors, etc.\ can be computed in
constant time, independent of the element's level. Since the logic of
the Morton code is the same for 2D and 3D, one can implement both versions with
the same source code, carrying out all operations on the $z$-coordinate within an
\texttt{if}-clause or a preprocessor macro, such as it is implemented in the
\pforest library for example \cite{Burstedde10a}.
Due to these advantages, the Morton index is chosen as SFC index by various AMR
packages and solvers~\cite{SundarSampathAdavanietal07, AkcelikBielakBirosEtAl03, BursteddeWilcoxGhattas11}.

The greatest perceived disadvantage of the Morton index are the jumps in
the SFC. The Morton index does not possess the same locality
properties such as for example the Hilbert curve~\cite{Hilbert91}, which may
lead to an increased runtime when iterating over the mesh elements and can
result in cache misses when locating neighbor elements
\cite{CampbellDevineFlahertyEtAl03b}. However, practical experiments show
that this effect is not measurable~\cite{BursteddeBurtscherGhattasEtAl09}.
Also, the Morton SFC may produce disconnected partitions, though the number of 
connected components is shown to be at most two \cite{Bader12, BursteddeHolkeIsaac17b}.

It is notable that the Morton index is not 
only used for AMR applications.
Examples for other use cases include fast matrix multiplication with the
Strassen-Algorithm \cite{ValsalamSkjellum02}, image encryption
\cite{ChangLiu97}, and databases \cite{RamsakMarklFenkEtAl00}.

In Chapter~\ref{ch:tetSFC} we develop a SFC for triangles and tetrahedra whose
core idea is based on the Morton SFC.

\begin{figure}
 \center
\begin{minipage}{0.35\textwidth}
\def\svgwidth{\textwidth}
\begingroup%
  \makeatletter%
  \providecommand\color[2][]{%
    \errmessage{(Inkscape) Color is used for the text in Inkscape, but the package 'color.sty' is not loaded}%
    \renewcommand\color[2][]{}%
  }%
  \providecommand\transparent[1]{%
    \errmessage{(Inkscape) Transparency is used (non-zero) for the text in Inkscape, but the package 'transparent.sty' is not loaded}%
    \renewcommand\transparent[1]{}%
  }%
  \providecommand\rotatebox[2]{#2}%
  \ifx\svgwidth\undefined%
    \setlength{\unitlength}{191.56700439bp}%
    \ifx\svgscale\undefined%
      \relax%
    \else%
      \setlength{\unitlength}{\unitlength * \real{\svgscale}}%
    \fi%
  \else%
    \setlength{\unitlength}{\svgwidth}%
  \fi%
  \global\let\svgwidth\undefined%
  \global\let\svgscale\undefined%
  \makeatother%
  \begin{picture}(1,0.9982147)%
    \put(0,0){\includegraphics[width=\unitlength,page=1]{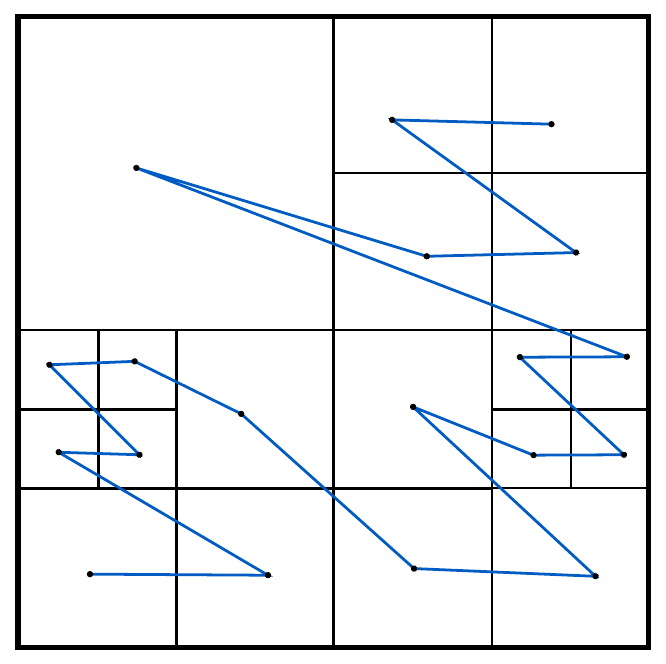}}%
  \end{picture}%
\endgroup%
 \end{minipage}\hfill
\begin{minipage}{0.6\textwidth}
\center
\begin{tikzpicture}[yscale=0.5, xscale=0.5, %
                            level 1/.style={sibling distance=12em},%
                            level 2/.style={sibling distance=3em},%
                            level 3/.style={sibling distance=1.8em},%
                            every node/.style = {shape=rectangle,draw, %
                            minimum size=0.2em,anchor=north},%
                            leaf/.style = {draw,fill=blue!20}%
                            ]
   \node {}
    child { node{}  
    	child { node[leaf](a){}  }
    	child { node[leaf](b){}  }
    	child { 
        node{}
        child { node[leaf](c){}  }
        child { node[leaf]{}  }
        child { node[leaf]{}  }
        child { node[leaf](d){}}  
      }
    	child { node[leaf](e){}  }
    }
    child { node{}  
    	child { node[leaf](f){}  }
    	child { node[leaf]{}  }
    	child { node[leaf](g){}  }
      child { node{}  
       	child { node[leaf](h){}  }
       	child { node[leaf]{}  }
      	child { node[leaf]{}  }
        child { node[leaf](i){}  }
      }
    }
   	child { node[leaf](j){}  }
    child { node{}  
    	child { node[leaf](k){}  }
    	child { node[leaf]{}  }
    	child { node[leaf]{}  }
    	child { node[leaf](l){}  }
    };
    \draw[arrows=->,very thick,blue] (a.center)--(b.center)--(c.center)--%
    (d.center)--(e.center)--(f.center)--(g.center)--(h.center)--(i.center)--%
    (j.center)--(k.center)--(l.center);
 \end{tikzpicture}\\[3ex]
\def\svgwidth{\textwidth}
\begingroup%
  \makeatletter%
  \providecommand\color[2][]{%
    \errmessage{(Inkscape) Color is used for the text in Inkscape, but the package 'color.sty' is not loaded}%
    \renewcommand\color[2][]{}%
  }%
  \providecommand\transparent[1]{%
    \errmessage{(Inkscape) Transparency is used (non-zero) for the text in Inkscape, but the package 'transparent.sty' is not loaded}%
    \renewcommand\transparent[1]{}%
  }%
  \providecommand\rotatebox[2]{#2}%
  \ifx\svgwidth\undefined%
    \setlength{\unitlength}{222.27009277bp}%
    \ifx\svgscale\undefined%
      \relax%
    \else%
      \setlength{\unitlength}{\unitlength * \real{\svgscale}}%
    \fi%
  \else%
    \setlength{\unitlength}{\svgwidth}%
  \fi%
  \global\let\svgwidth\undefined%
  \global\let\svgscale\undefined%
  \makeatother%
  \begin{picture}(1,0.09219034)%
    \put(0,0){\includegraphics[width=\unitlength,page=1]{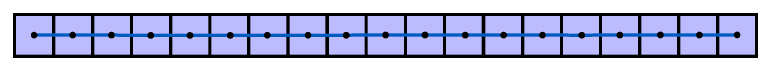}}%
  \end{picture}%
\endgroup%
 \end{minipage}
  \caption[The SFC curve arising from the Morton code for quadrilateral
1:4 refinement]
  {The SFC curve arising from the Morton code for quadrilateral
1:4 refinement. Left: A refinement of the root element
together with the SFC starting in the lower left corner and ending in 
the upper right corner.
Right: The associated refinement tree with the SFC passing the leaves (top).
Via the consecutive index, the SFC order induces a linear order of the elements
in the refinement, such that we can store these elements in an array (bottom).}
\figlabel{fig:refspace-ex2}
\end{figure}
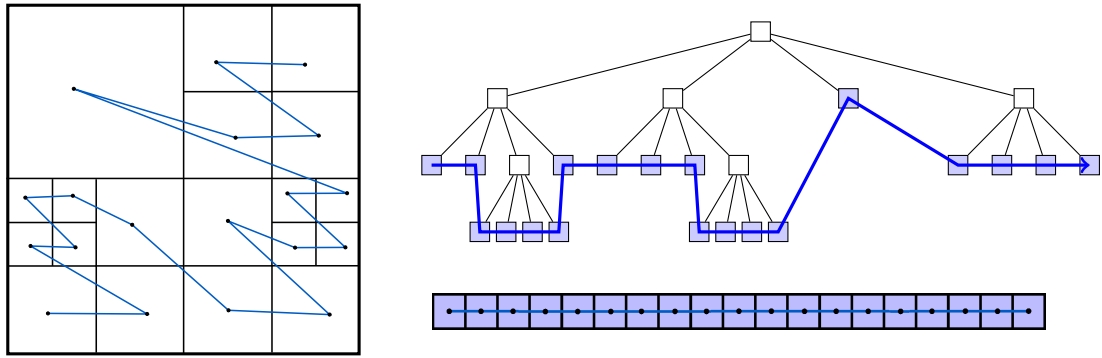

\section{Space-filling curves on forests of trees}
\label{sec:SFConforest}

We do not want to restrict our meshes to single refinement trees, but include
the possibility to patch trees together to form more complex coarse meshes.
Thus, from the SFC point of view, we have a collection of refinements with
individual SFCs on them, which we consider together as a so-called forest.

\begin{definition}
Let $\set{\mathcal K_0,\dots,\mathcal K_{K-1}}$ be refinements
of respective refinement spaces $\set{\mathcal S_0,\dots, \mathcal S_{K-1}}$.
Then the \textbf{forest} $\mathscr F$ with trees $\set{\mathcal K_k}_{k<K}$ is the
set of all leaves of the individual refinements paired with their tree number $k$:
 \begin{equation}
 \mathscr F : = \bigcup_{k=0}^{K-1}\set{k}\times\mathcal K_k.
 \end{equation}
The elements of $\mathscr F$ are the \textbf{leaves} of the forest.
\end{definition}

\begin{remark}
The refinement spaces $\mathcal S_k$ do not have to be different from each
other. In most cases, all $\mathcal S_k$ are the same or there are only
a few different types of refinement spaces.
\end{remark}
\begin{remark}
A forest with a single tree is isomorphic to a refinement.
\end{remark}

\begin{definition}
If for a forest $\mathscr F$ each refinement space $\mathcal S_k$ has  an SFC
index $\mathcal I_k$, then we extend these to an index $\mathcal I$
on the leaves of $\mathscr F$ by
\begin{align}
 \mathcal I\colon \mathscr F & \longrightarrow \set{0,\dots,K-1} \times \IN_0\\
  (k,E) & \longmapsto (k, \mathcal I_k(E))
\end{align}
with the order 
\begin{equation}
\label{eq:forestorder}
 (k, I) < (k',I') :\Leftrightarrow k < k' \textrm{ or } 
 (k=k' \textrm{ and } I < I')
\end{equation}
on $\set{0,\dots,K-1} \times \IN_0$, which extends the individual SFC orders across the
trees.
By extension of notation we call $\mathcal I$ an SFC index of $\mathscr F$.
\end{definition}
Similarly, the levels of the individual refinement spaces extend to a level
map for the forest by $\ell(k,E):=\ell(E)$.

With these definitions we form the analogon to Corollary~\ref{cor:consind}
for forests.
\begin{lemma}
\label{lem:forestconsindex}
 Let $\mathscr F$ be a forest with SFC index $\mathcal I$ and
 finite number $N$ of leaves, then there exists a unique bijective map 
 \begin{equation}
  \mathcal I_{\mathscr F}\colon \mathscr F\longrightarrow \set{0,\dots, N-1}
 \end{equation}
that is monotonous under $\mathcal I$, thus
 \begin{equation}
  (k, \mathcal I_k(E)) < (k',\mathcal I_k(E')) \Leftrightarrow 
  \mathcal I_{\mathscr F} (k,E) < \mathcal I_{\mathscr F} (k',E').
 \end{equation}
We call this map the \textbf{consecutive} index of $\mathscr F$.
\end{lemma}
\begin{proof}
 Let $O_k$ be the number of leaves in all trees of indices smaller
 than $k$:
 \begin{equation}
  O_k := \sum_{i=0}^{k-1}|\mathcal K_k|,
 \end{equation}
 and let $\mathcal I_{\mathcal K_k}$ be the consecutive index of the $k$-th tree.
 Then 
 \begin{equation}
 \mathcal I_{\mathscr F} (k,E) := O_k + \mathcal I_{\mathcal K_k} (E)
 \end{equation}
 satisfies the desired property.
\end{proof}
We illustrate this concept in Figure~\ref{fig:sota-twotreeforest}.

\begin{figure}
\center
\includegraphics[width=\textwidth]{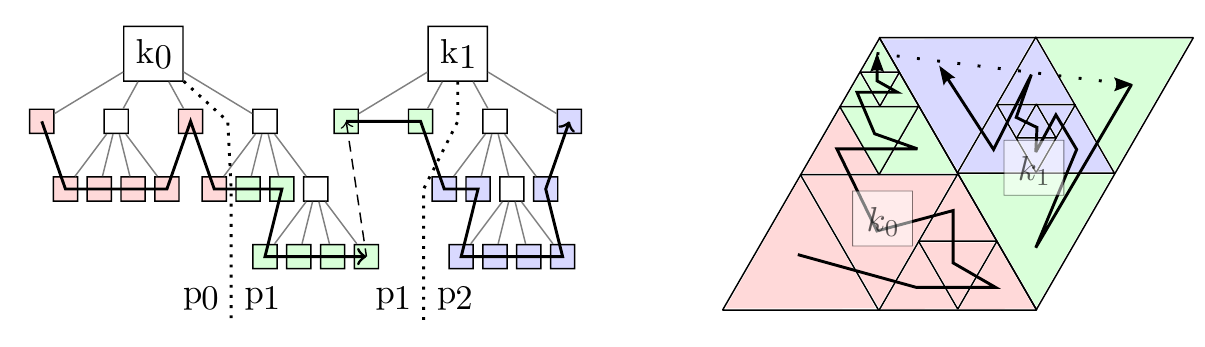}
\caption[Complex geometries with multiple trees]
{SFCs naturally extend to multiple trees to model complex geometries.
Here, we show two trees $k_0$ and $k_1$ with an adaptive refinement.
To enumerate the forest mesh, we establish an a priori order between the two
trees and use an SFC within each tree.
On the left-hand side of the figure the refinement trees and their linear storage
are shown. When we partition the forest mesh to $P$ processes (here, $P=3$),
we cut the SFC in $P$ equally sized parts and assign part $i$ to process $i$.
}
\figlabel{fig:sota-twotreeforest}
\end{figure}

\section{Partitioning with space-filling curves} 
\label{sec:partitionsfc}
In order to ensure the scalability of an application, it is necessary that 
each parallel process is assigned (approximately) the same amount of work.
This operation is called load-balancing.

While SFCs offer us a way to effectively store and access the leaves of a
forest, they also suggest a straightforward method to load-balance these
leaves across multiple processes.

Supposing a load-balanced mesh changes due to adaptation, then the workload
may not be balanced anymore. It is thus necessary to migrate elements from 
those processes with more elements to those with less. This procedure is 
called repartitioning. In order to fairly repartition the mesh, it is necessary
to know which elements have to be mapped to which process.

Let $\mathcal I_{\mathscr F}$ be the consecutive index of a forest $\mathscr F$
with $N$ leaves that are to be distributed among $P$ processes with ranks $0,
1,\dots, P-1$.
The partition should respect the SFC order, thus if $(k,E)$ is assigned to a
process $p$ and $\mathcal I_{\mathscr F}(k,E')
> \mathcal I_{\mathscr F}(k',E)$ then we demand that $(k',E')$ is assigned to a
 process $q\geq p$.

We use the scheme presented in e.g.\ \cite{BursteddeWilcoxGhattas11},
 assigning to process $p$ the set of leaves
\begin{equation}
 \mathscr F (p):=\setm{(k,E) \in\mathscr F}{\left\lfloor\frac{pN}{P}\right\rfloor\leq
    \mathcal I_{\mathscr F}(k,E) < \left\lfloor\frac{(p+1)N}{P}\right\rfloor}.
\end{equation}
This ensures that the count of elements on different processes differs by at most
one.
See also Figure~\ref{fig:sota-twotreeforest} for an illustration of the 
partitioning.

In some cases the computational load differs between elements.  We then may
assign a non-negative weight $w(k,E)$ to each element that is proportional to
the computational load and demand from the assignment of elements to processes
that the sums of all weights of the processes' elements are approximately equal,
rather than the count of elements. This is known as the weighted partitioning
problem to which several approximative solutions using SFCs exist 
\cite{BursteddeWilcoxGhattas11,DevineBomanHeaphyEtAl05}.

An example where weighted partitioning occurs is the $hp$-adaptive finite
element method. Here, different elements have different numbers of degrees of
freedom such that elements with smaller polynomial degree $p$ cause less
computational load than those with a larger degree \cite{LaszloffyLongPatra00}.

In contrast to load-balancing using graph-based
methods~\cite{KarypisKumar95,DevineBomanHeaphyEtAl02}, the partitions resulting
from SFCs can have more surface area and may be disconnected resulting in an
increase in parallel communication.
However, the partition quality was shown to be acceptable~\cite{BursteddeGhattasGurnisEtAl10}
and partitioning with SFC can be orders of magnitude faster, since it reduces
the NP-hard problem to an approximation using linear runtime. For this reason
SFC partitioning is a common choice when the mesh is repartitioned frequently,
e.g.\ in adaptive solvers for time-dependent PDEs; see for
example~\cite{BungartzMehlWeinzierl06,
GriebelZumbusch99,GuntherMehlPoglEtAl06,Bader12,Sagan94,Zumbusch99}.

 \chapter{The Tetrahedral Morton Index}
\label{ch:tetSFC}
This chapter is based on the paper~\cite{BursteddeHolke16}.
We edited it slightly in order to fit into the general notations of this
thesis, without changing its mathematical content.
Copyright \copyright\xspace by SIAM. Unauthorized reproduction of this chapter
is prohibited.
\vspace{2ex}

In this chapter, we develop a new SFC index for triangular and
tetrahedral mesh refinement that can be computed using bitwise interleaving
operations similar to the Morton index for cubical meshes; see
Section~\ref{sec:mortonsfc}. We demonstrate that this index has many of the
favorable properties known for hexahedra.
To store sufficient information for random access, we define a low-memory
encoding using 10 bytes per triangle and 14 bytes per tetrahedron.

Our starting point is to divide the simplices in a refined mesh into two (two
dimensions, 2D), respectively six (three dimensions, 3D), different types and
selecting for each type a specific reordering for Bey's red-refinement
\cite{Bey92}.
This type and the coordinates of one vertex serve as a unique identifier, the
Tet-id, of the simplex in question.
In particular, we do not require storing the type of all parent simplices to
the root, as one might naively imagine.
We then propose a Morton-like coordinate mapping that can be computed from the
Tet-id and gives rise to an SFC.
Based on this logic, we develop constant-time algorithms (independent of an
element's refinement level) to compute the Tet-id of any parent, child,
face-neighbor, and SFC-successor/predecessor and to decide whether for two
given elements one is an ancestor of the other. We conclude with scalability
tests of mesh creation and adaptation with over $850 \e9$
tetrahedral mesh elements on up to 131,072 cores of the JUQUEEN supercomputer
and 786,432 cores of MIRA.

\section{Mesh refinement on simplices}
\label{sec:BeyRef}

Our aim is to define and examine a new SFC for triangles and
tetrahedra by adding ordering prescriptions to the nonconforming
Bey-refinement (also called red-refinement) \cite{Bey92,Bey00,Zhang95}.
We briefly restate the red-refinement in this section and contrast it with the
well-known conforming (or red/green) refinement.

We refer to triangles and tetrahedra as $d$-simplices, where $d\in\set{2,3}$
specifies the dimension.
It is sometimes convenient to drop $d$ from this notation.
A $d$-simplex $T\subseteq \IR^d$ is uniquely determined by its $d+1$
affine-independent corner nodes $\vec x_0,\dots,\vec x_{d} \in \IR^d$.
Their order is significant, and therefore we write
\begin{subequations}
\begin{align}%
  T& = [\vec x_0,\vec x_1,\vec x_2] \hphantom{,x_3} \quad\textrm{in 2D},\\
  T& = [\vec x_0,\vec x_1,\vec x_2,\vec x_3] \quad\textrm{in 3D}.
\end{align}
\end{subequations}
\begin{definition}
 We define $\vec x_0$ as the \textbf{anchor node} of $T$.
 By $\vec x_{ij}$ we denote the midpoint between $\vec x_i$ and $\vec x_j$,
 thus $\vec x_{ij}=\frac{1}{2}(\vec x_i+\vec x_j)$.
\end{definition}

\begin{figure}
\center
\begin{minipage}{0.48\textwidth}
   \def\svgwidth{40ex}
   \includegraphics{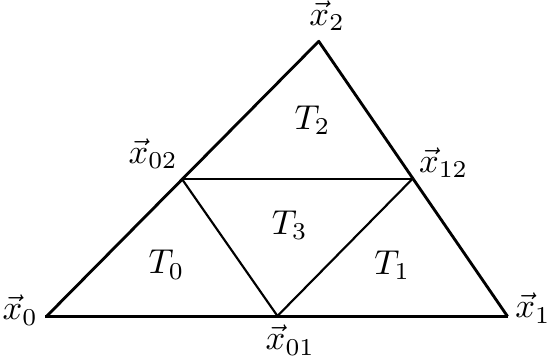}
\end{minipage}
\begin{minipage}{0.48\textwidth}
   \def\svgwidth{40ex}
   \includegraphics{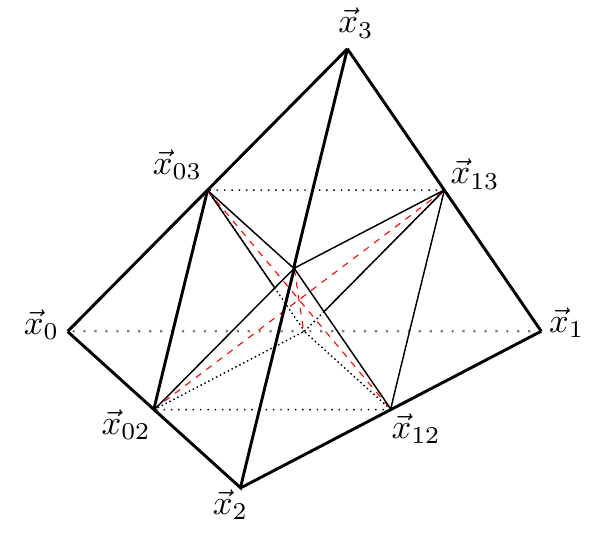}
\end{minipage}
   \caption[The Bey refinement schemes for triangles and tetrahedra]{%
   Left: The refinement scheme for triangles in two dimensions. A triangle
   $T=[\vec x_0,\vec x_1,\vec x_2]\subset\IR^2$ is refined by dividing each face
   at the midpoints $\vec x_{ij}$. We obtain four smaller triangles, all similar
   to $T$. Right: The situation in three dimensions. If we divide the edges of
   the tetrahedron $T=[\vec x_0,\vec x_1,\vec x_2,\vec x_3]\subset\IR^3$ in half,
   we get four smaller tetrahedra (similar to $T$) and an inner octahedron. By
   dividing the octahedron along any of its three diagonals (shown dashed) we
   finally end up with partitioning $T$ into eight smaller tetrahedra, all having
   the same volume. The refinement rule of Bey is obtained by always choosing the
   diagonal from $\vec x_{02}$ to $\vec x_{13}$ and numbering the corners of the
   children according to \eqref{eq:childnumbers}.}
   \figlabel{fig:cutoff}
\end{figure}%

\subsection{Bey's refinement rule}
\label{sec:tetsfc_beyrule}

Bey's rule is a prescription for subdividing a simplex.
It is one instance of the so-called red-refinement, where all faces of a
simplex are subdivided simultaneously.
\begin{definition}
  \label{def:SFCrefrule}
  Given a $d$-simplex $T=[\vec x_0,\dots,\vec x_d]\subset\IR^d$, the \textbf{refinement
  rule of Bey} consists of cutting off four subsimplices at the corners
 (as in Figure \ref{fig:cutoff}).
 In 3D the remaining octahedron is then divided along the diagonal from $\vec
 x_{02}$ to $\vec x_{13}$.
 Bey numbers the $2^d$ resulting subsimplices as follows.
 \begin{subequations}
\label{eq:childnumbers}
 \begin{equation}
 \label{eq:childnumbers_2d}
 2D:\quad
 \begin{array}{cccccl}
  T_0&:=&[\vec x_0,\vec x_{01},\vec x_{02}], &  T_1&:=&[\vec x_{01},\vec x_{1},\vec x_{12}],\\
  T_2&:=&[\vec x_{02},\vec x_{12},\vec x_{2}], & T_3&:=&[\vec x_{01},\vec x_{02},\vec x_{12}],
  \end{array}
 \end{equation}
 \begin{equation}
\label{eq:childnumbers3d}
 3D:\quad
 \begin{array}{cccccc}
 T_0 &:=& [\vec x_0,\vec x_{01},\vec x_{02},\vec x_{03}],   & T_4 &:=& [\vec x_{01},\vec x_{02},\vec x_{03},\vec x_{13}],\\
 T_1 &:=& [\vec x_{01},\vec x_{1},\vec x_{12},\vec x_{13}], & T_5 &:=& [\vec x_{01},\vec x_{02},\vec x_{12},\vec x_{13}],\\
 T_2 &:=& [\vec x_{02},\vec x_{12},\vec x_{2},\vec x_{23}], & T_6 &:=& [\vec x_{02},\vec x_{03},\vec x_{13},\vec x_{23}],\\
 T_3 &:=& [\vec x_{03},\vec x_{13},\vec x_{23},\vec x_{3}], & T_7 &:=& [\vec x_{02},\vec x_{12},\vec x_{13},\vec x_{23}].
\end{array}
 \end{equation}
 \end{subequations}
\end{definition}
\begin{remark}
  If we apply the refinement rule from Definition~\ref{def:SFCrefrule} 
  recursively to the descendants of a $d$-simplex $S$, we obtain a refinement space
  in the sense of Definition~\ref{def:refspace} with root element $S$.
  By a refinement of $S$ we mean a refinement of this refinement space.

  Note, that the refinement rule explicitly allows nonuniform meshes and
  thus hanging faces and edges.
\end{remark}

We recall some basic definitions of relations among mesh elements in the simplicial case.
\begin{definition}
\label{def:childrenancestor}
The $T_i$ from \eqref{eq:childnumbers} are called the \textbf{children} of $T$,
and $T$ is called the \textbf{parent} of the $T_i$, written $T=P(T_i)$.
Therefore, we also call the $T_i$ \textbf{siblings} of each other.
If a $d$-simplex $T$ belongs to a refinement of
another $d$-simplex $S$, then $T$ is a \textbf{descendant} of $S$, and $S$ is an \textbf{ancestor} of $T$.
The number $\ell$ of refining steps needed to obtain $T$ from $S$ is unique
and called the \textbf{level} of $T$ (with respect to $S$); we write $\ell=\ell(T)$.
Usually $S$ is clear from the context, and therefore we %
omit it in the notation.
By definition, $T$ is an ancestor and descendant of itself.
\end{definition}

Consider the six tetrahedra $S_0,\dots,S_5\abst{\subset} \IR^3$
displayed in Figure \ref{fig:sechstetra}.
\begin{figure}
  \begin{minipage}{0.60\textwidth}
   \def\svgwidth{0.95\textwidth}%
\begingroup%
  \makeatletter%
  \providecommand\color[2][]{%
    \errmessage{(Inkscape) Color is used for the text in Inkscape, but the package 'color.sty' is not loaded}%
    \renewcommand\color[2][]{}%
  }%
  \providecommand\transparent[1]{%
    \errmessage{(Inkscape) Transparency is used (non-zero) for the text in Inkscape, but the package 'transparent.sty' is not loaded}%
    \renewcommand\transparent[1]{}%
  }%
  \providecommand\rotatebox[2]{#2}%
  \ifx\svgwidth\undefined%
    \setlength{\unitlength}{380.95061035bp}%
    \ifx\svgscale\undefined%
      \relax%
    \else%
      \setlength{\unitlength}{\unitlength * \real{\svgscale}}%
    \fi%
  \else%
    \setlength{\unitlength}{\svgwidth}%
  \fi%
  \global\let\svgwidth\undefined%
  \global\let\svgscale\undefined%
  \makeatother%
  \begin{picture}(1,0.63546393)%
    \put(0,0){\includegraphics[width=\unitlength,page=1]{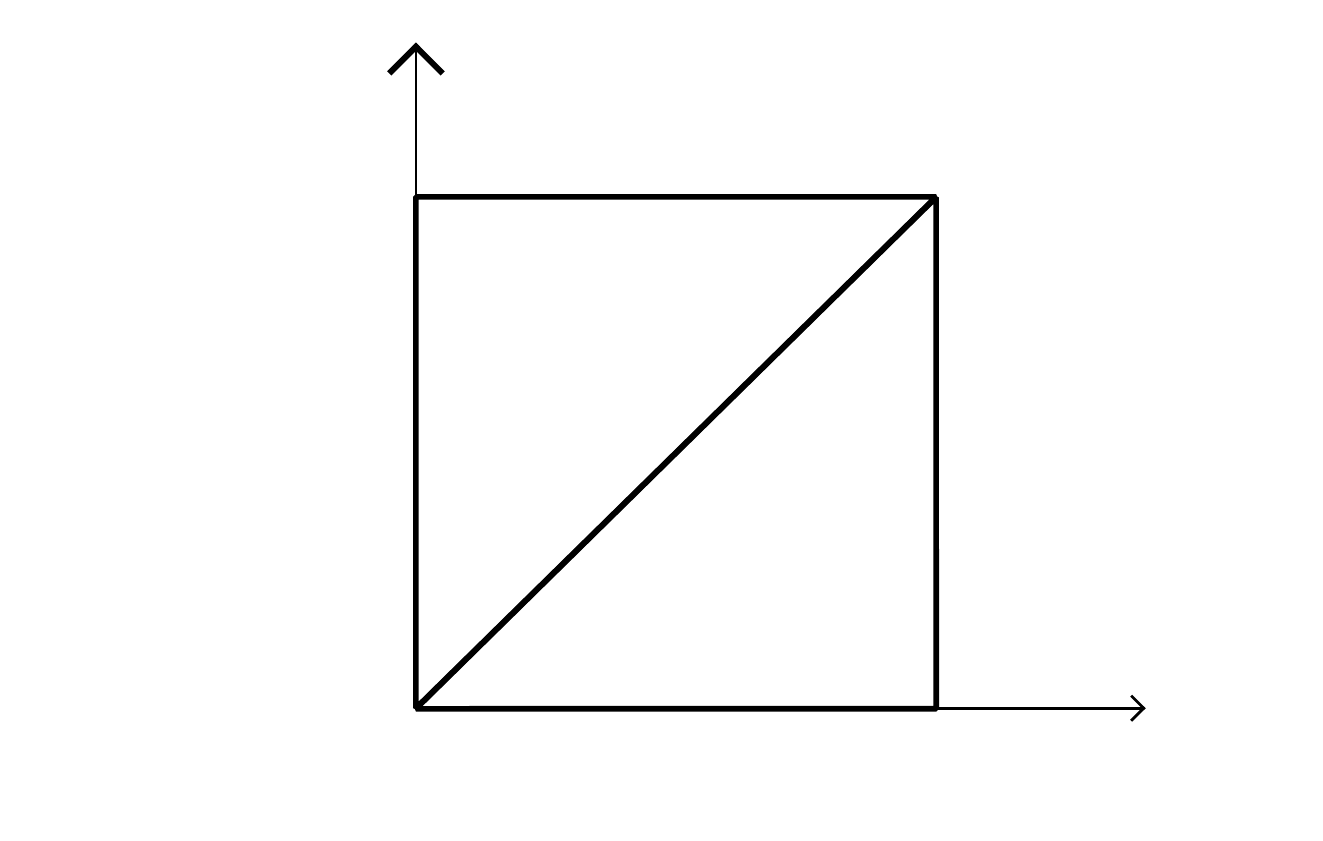}}%
    \put(0.83323174,0.13373889){\color[rgb]{0,0,0}\makebox(0,0)[lb]{\smash{$X$}}}%
    \put(0.33187745,0.57393166){\color[rgb]{0,0,0}\makebox(0,0)[lb]{\smash{$Y$}}}%
    \put(0.42372984,0.36473997){\color[rgb]{0,0,0}\makebox(0,0)[lb]{\smash{$S_1$}}}%
    \put(0.57283053,0.25028943){\color[rgb]{0,0,0}\makebox(0,0)[lb]{\smash{$S_0$}}}%
    \put(0.73033127,0.48024049){\color[rgb]{0,0,0}\makebox(0,0)[lb]{\smash{$c_3=\begin{pmatrix}\\ 1\\1\\ \end{pmatrix}$}}}%
    \put(0,0){\includegraphics[width=\unitlength,page=2]{zweidreiecks_tex.pdf}}%
    \put(0.22886495,0.12591086){\color[rgb]{0,0,0}\makebox(0,0)[lb]{\smash{}}}%
    \put(0.22886495,0.43750411){\color[rgb]{0,0,0}\makebox(0,0)[lb]{\smash{$c_2$}}}%
    \put(0.72709512,0.12591086){\color[rgb]{0,0,0}\makebox(0,0)[lb]{\smash{$c_1$}}}%
    \put(0.72709512,0.43750411){\color[rgb]{0,0,0}\makebox(0,0)[lb]{\smash{}}}%
    \put(0.06300029,0.08895305){\color[rgb]{0,0,0}\makebox(0,0)[lb]{\smash{$\begin{pmatrix}\\ 0\\0\\ \end{pmatrix}=c_0$}}}%
  \end{picture}%
\endgroup%

   \def\svgwidth{0.95\textwidth}%
\begingroup%
  \makeatletter%
  \providecommand\color[2][]{%
    \errmessage{(Inkscape) Color is used for the text in Inkscape, but the package 'color.sty' is not loaded}%
    \renewcommand\color[2][]{}%
  }%
  \providecommand\transparent[1]{%
    \errmessage{(Inkscape) Transparency is used (non-zero) for the text in Inkscape, but the package 'transparent.sty' is not loaded}%
    \renewcommand\transparent[1]{}%
  }%
  \providecommand\rotatebox[2]{#2}%
  \ifx\svgwidth\undefined%
    \setlength{\unitlength}{550.32397461bp}%
    \ifx\svgscale\undefined%
      \relax%
    \else%
      \setlength{\unitlength}{\unitlength * \real{\svgscale}}%
    \fi%
  \else%
    \setlength{\unitlength}{\svgwidth}%
  \fi%
  \global\let\svgwidth\undefined%
  \global\let\svgscale\undefined%
  \makeatother%
  \begin{picture}(1,0.71689905)%
    \put(0,0){\includegraphics[width=\unitlength,page=1]{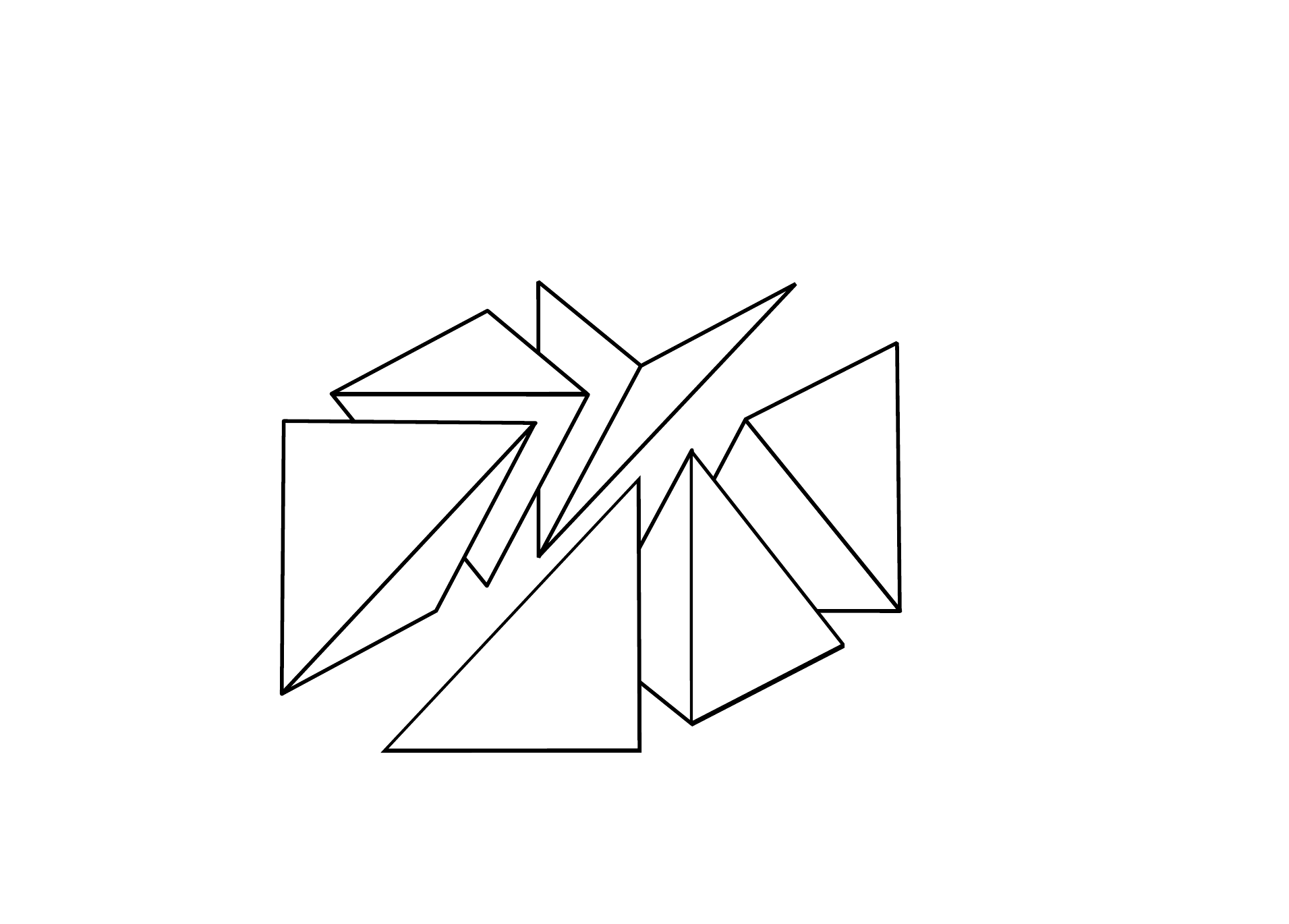}}%
    \put(0.25338018,0.30385227){\color[rgb]{0,0,0}\makebox(0,0)[lb]{\smash{$S_0$}}}%
    \put(0.41026491,0.18572463){\color[rgb]{0,0,0}\makebox(0,0)[lb]{\smash{$S_1$}}}%
    \put(0.56197011,0.23375293){\color[rgb]{0,0,0}\makebox(0,0)[lb]{\smash{$S_2$}}}%
    \put(0.63363954,0.35339003){\color[rgb]{0,0,0}\makebox(0,0)[lb]{\smash{$S_3$}}}%
    \put(0.352629,0.42566703){\color[rgb]{0,0,0}\makebox(0,0)[lb]{\smash{$S_5$}}}%
    \put(0,0){\includegraphics[width=\unitlength,page=2]{sechstetras_tex.pdf}}%
    \put(0.03878397,0.03761989){\color[rgb]{0,0,0}\makebox(0,0)[lb]{\smash{$X$}}}%
    \put(0.85164988,0.21446232){\color[rgb]{0,0,0}\makebox(0,0)[lb]{\smash{$Y$}}}%
    \put(0.37631668,0.64364674){\color[rgb]{0,0,0}\makebox(0,0)[lb]{\smash{$Z$}}}%
    \put(0,0){\includegraphics[width=\unitlength,page=3]{sechstetras_tex.pdf}}%
    \put(0.48367022,0.45951426){\color[rgb]{0,0,0}\makebox(0,0)[lb]{\smash{$S_4$}}}%
    \put(0,0){\includegraphics[width=\unitlength,page=4]{sechstetras_tex.pdf}}%
    \put(0.14314884,0.19751762){\color[rgb]{0,0,0}\makebox(0,0)[lb]{\smash{$c_1$}}}%
    \put(0.71592217,0.26052046){\color[rgb]{0,0,0}\makebox(0,0)[lb]{\smash{$c_2$}}}%
    \put(0.533815,0.10767503){\color[rgb]{0,0,0}\makebox(0,0)[lb]{\smash{$c_3$}}}%
    \put(0.42022717,0.52879485){\color[rgb]{0,0,0}\makebox(0,0)[lb]{\smash{$c_4$}}}%
    \put(0.16628655,0.42044046){\color[rgb]{0,0,0}\makebox(0,0)[lb]{\smash{$c_5$}}}%
    \put(0.6616011,0.50544892){\color[rgb]{0,0,0}\makebox(0,0)[lb]{\smash{$c_6$}}}%
    \put(0.82949204,0.44073145){\color[rgb]{0,0,0}\makebox(0,0)[lb]{\smash{$c_7$}}}%
    \put(0.2478235,0.03904336){\color[rgb]{0,0,0}\makebox(0,0)[lb]{\smash{$c_0$}}}%
    \put(0,0){\includegraphics[width=\unitlength,page=5]{sechstetras_tex.pdf}}%
  \end{picture}%
\endgroup%
  \end{minipage}
  \raisebox{7ex}{%
  \begin{minipage}{0.39\textwidth}
   \flushright
  \begin{subequations}
  \label{eq:coordsofSb}
    \begin{equation}
    \label{eq:coordsofSb_2d}
\begin{array}{r|cccc}
2\textrm D& \vec x_0 & \vec x_1 & \vec x_2 \\[0.5ex]\hline
 \vphantom{{X^X}^X}S_0& c_0 & c_1 & c_3 \\[0.2ex]
 S_1& c_0 & c_2 & c_3
\end{array}
\end{equation}\\[14ex]
  \begin{equation}
  \label{eq:coordsofSb_3d}
\begin{array}{r|cccc}
 3\textrm D& \vec x_0 & \vec x_1 & \vec x_2 & \vec x_3 \\[0.5ex]\hline
 \vphantom{{X^X}^X}S_0& c_0 & c_1 & c_5 & c_7 \\[0.2ex]
 S_1& c_0 & c_1 & c_3 & c_7 \\[0.2ex]
 S_2& c_0 & c_2 & c_3 & c_7 \\[0.2ex]
 S_3& c_0 & c_2 & c_6 & c_7 \\[0.2ex]
 S_4& c_0 & c_4 & c_6 & c_7 \\[0.2ex]
 S_5& c_0 & c_4 & c_5 & c_7
\end{array}
\end{equation}
\end{subequations}
 \end{minipage}}
 \caption[The triangle and tetrahedra types]
{The basic triangle (2D) and tetrahedra types (3D) obtained by dividing
$[0,1]^d$ into simplices of varying types, denoted by a subscript. Top left:
The unit square can be divided into two triangles sharing the
edge from $(0,0)^T$ to $(1,1)^T$. We denote these triangles by $S_0$ and $S_1$.
The four corners of the square are numbered $c_0,\ldots,c_3$ in $yx$-order.
Top right: The corner nodes of $S_0$ and $S_1$ in terms of the square corners.
Bottom left (exploded view): In three dimensions the unit cube can be divided
into six tetrahedra, all sharing the edge from the origin to $(1,1,1)^T$.  We
denote these tetrahedra by $S_0,\dots,S_5$.  The eight corners of the cube are
numbered $c_0,\ldots,c_7$ in $zyx$-order (redrawn and modified with permission
\cite{Bey92}).
Bottom right: The corner nodes of the six tetrahedra $S_0,\dots,S_5$ in
  terms of the cube corners.
} 
\figlabel{fig:sechstetra}
\end{figure}
These tetrahedra form a triangulation of the unit cube.
The results and algorithms in this chapter rely on the following property
\cite{Bey92}.
\begin{property}
\label{property:commref}
Refining the six tetrahedra from the triangulation of the unit cube simultaneously to level $\ell$ results in the same
mesh as first refining the unit cube to level $\ell$ and then
triangulating each smaller cube with the six tetrahedra $S_0,\dots,S_5$, scaled
by a factor of $2^{-\ell}$ (see Figure \ref{fig:refineddiagram}).
The same behavior can be observed in 2D when the unit square is
divided into two triangles.
\end{property}

\begin{figure}
\xymatrix@C=16ex@R=15ex{
*+<4pt,4pt>{
   \includegraphics[width=15ex]{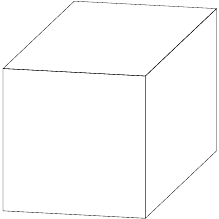}
   }
   \ar[r]^-{\textrm{refining the cube}}\ar[d]|{\textrm{triangulating the cube}}
   &
   *+<4pt,4pt>{
   \includegraphics[width=15ex]{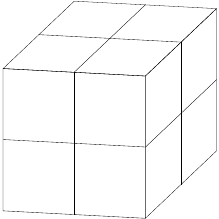}
   }
   \ar[d]|{\textrm{triangulating each cube}}
   \ar[r]^-{\textrm{refining each cube}}
   &
   *+<4pt,4pt>{
   \includegraphics[width=15ex]{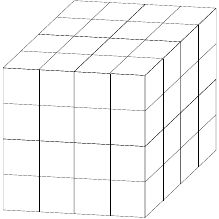}
   }
   \ar[d]|{\textrm{triangulating each cube}}
   \\
  *+<4pt,4pt>{
   \includegraphics[width=15ex]{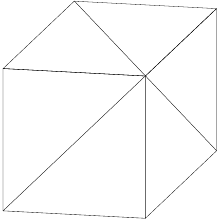}
   }
   \ar[r]^-{\textrm{refining}}_-{\textrm{each tetrahedron}} &
  *+<4pt,4pt>{
  \includegraphics[width=15ex]{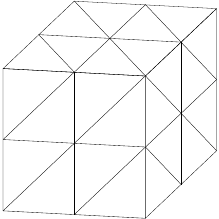}
   }
   \ar[r]^-{\textrm{refining}}_-{\textrm{each tetrahedron}} &
    *+<4pt,4pt>{
    \includegraphics[width=15ex]{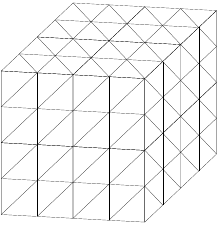}
    }
}
\caption[First refining and then triangulating a cube is the same as first
triangulating it and then refining the tetrahedra]
{Triangulating a cube according to Figure \ref{fig:sechstetra} and then
refining the tetrahedra via Bey's refinement rule results in the same mesh as
first refining the cube into eight subcubes and afterward triangulating each of
these cubes.  Each occurring tetrahedron is uniquely determined by the subcube
it lies in plus its type. The same situation can be observed in 2D if we
restrict our view to one side of the cube.}
\figlabel{fig:refineddiagram}
\end{figure}

\begin{remark}
  A key motivation to use the refinement scheme of Bey is that it produces
  numerically stable meshes. Thus, no matter how much we refine,
  the mesh elements do not degenerate.
  The degeneracy of a mesh element $T$ with volume $v$ and side lengths $\set{l_i}$
  may be measured by
  \begin{equation}
    \eta(T) = \frac{12(3v)^{\frac{2}{3}}}{\sum l_i^2} > 0.
  \end{equation}
  If $\eta(T) \ll 1$ the element's volume is small relative to the sum of its side lengths.
  Imagine a very flat tetrahedron, or a triangle with one angle being close to $\pi$.

  From Property~\ref{property:commref} we conclude that for each tetrahedron $T$ that is an
  ancestor of one of the $S_i$ we have
  \begin{equation}
    \eta(T) = \eta(S_j)
  \end{equation}
  for some $j\in\set{0,\dots,5}$.
  Hence, $\eta(T)$ is bounded and any mesh resulting from Bey's refinement rule is
  numerically stable.
  See~\cite{LiuJoe96} and the references therein for a more thorough
  discussion of this subject.
\end{remark}

\subsection{Removal of hanging nodes using red/green refinement}

It is worth noting that, although the methods and algorithms presented in this
chapter apply to red-refined meshes with hanging nodes, it is possible to augment
them to create meshes without hanging nodes.
For this we may use red/green or red/green/blue refinement methods
\cite{AndrewShermanWeiser83,Carstensen04a}.

After the red-refinement step we may add an additional and possibly nonlocal
refinement operation that ensures a maximum level difference of 1 between
neighboring simplices.
Such an operation is also called 2:1 balance
\cite{IsaacBursteddeGhattas12,SundarSampathBiros08,TuOHallaronGhattas05};
we describe it in detail in Chapter~\ref{ch:balance}.
Hanging nodes are then resolved by bisecting
those simplices with hanging nodes (green/blue refinement)
\cite[section 12.1.3]{Bader12}.
The 2D case is shown in Figure \ref{fig:greenblueref}. 
\begin{figure}
\centering
  \begin{minipage}{0.40\textwidth}
\centering
  \includegraphics[width=0.8\textwidth]{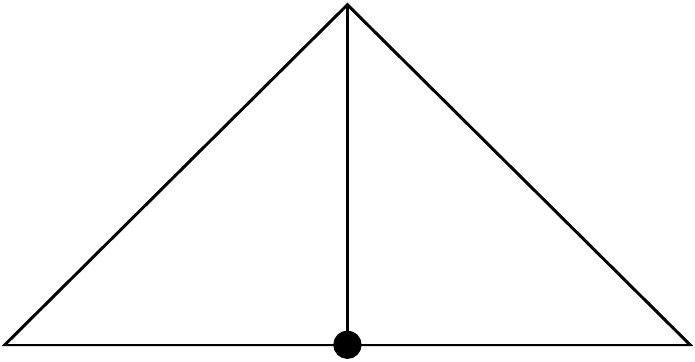}
\end{minipage}
 \begin{minipage}{0.40\textwidth}
  \centering
  \includegraphics[width=0.8\textwidth]{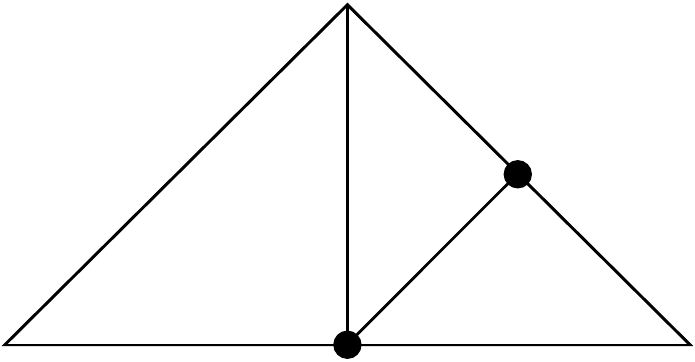}
\end{minipage}
\caption[Resolving hanging nodes]
   {To resolve hanging nodes we can execute an additional step of green- or
    blue-refinement after the last refinement step \cite{AndrewShermanWeiser83,
    Carstensen04a}.
    Here we show the 2D refinement rules.
    Left: green (1 hanging node).
    Right: blue (2 hanging nodes).
    If a triangle has 3 hanging nodes it is red-refined.}%
    \figlabel{fig:greenblueref}
\end{figure}

If one of the newly created simplices shall be further refined, the bisection
is reversed, the original simplex is red-refined, and the balancing and
green-refinement is repeated.
This may void the nesting property of certain discrete function spaces, yet
applications may still prefer this approach over the manual implementation of
hanging node constraints.

\section{The tetrahedral Morton index}
\label{sec:tetsfcdef}
As we repeat in Section~\ref{sec:mortonsfc},
the Morton index or Z-order for a cube in a hexahedral mesh is computed by bitwise interleaving the coordinates of the anchor node of the cube~\cite{Morton66}.
In this section we present an index for $d$-simplices that also uses the
bitwise interleaving approach, the tetrahedral Morton index (TM-index).
To define the TM-index we look at refinements of a reference simplex, which we
discuss in Section~\ref{sec:reference} below.
For each $d$-simplex in a refinement of the reference simplex we define a
unique identifier, the so-called Tet-id, which serves as the input to compute
the TM-index and for all algorithms related to it.
This Tet-id consists of the coordinates of the anchor node of the considered simplex
plus one additional number, the type of the simplex.
We define the Tet-id and type in Section \ref{sec:type_id}.
We then define the TM-index in Section \ref{sec:TMindex} and in the following
subsections.
We show that the TM-index defines a SFC index in the sense of
Definition~\ref{def:sfcindex} and discuss properties of the resulting SFC.

One novel aspect of this construction lies in logically including the types of
the simplex and all its parents in the interleaving, while only using the type
of the simplex itself in the algorithms.

\subsection{The reference simplex}
\label{sec:reference}

Throughout the rest of this chapter, let $\mathcal{L}$ be a given maximal refinement level.
Instead of the unit cube $[0,1]^d$, we consider the scaled cube $[0,2^\mathcal{L}]^d$, ensuring that all node coordinates
in a refinement up to level $\mathcal L$ are integers.
Suppose we are given some $d$-simplex $T\subset\IR^d$ together with a refinement $\mathscr S$ of $T$.
By mapping $T$ affine-linearly to $2^\mathcal{L} S_0$ the refinement $\mathscr S$ is mapped
to a refinement $\mathscr S'$ of $2^\mathcal{L} S_0$.
Therefore, to examine SFCs on refinements of $T$, it suffices to examine SFCs
on $2^\mathcal{L} S_0$.
Thus, we only consider refinements of the $d$-simplex $T^0_d:=2^\mathcal{L} S_0$.
Let $\mathcal{T}_d$  be the set of all possible descendants of this $d$-simplex with level smaller than or
equal to $\mathcal{L}$; thus
\begin{equation}
 \mathcal{T}_d\abst{=} \set{T \abst{|} T \textnormal{ is a descendant of } T^0_d
 \textnormal{ with } 0\leq \ell(T)\leq\mathcal{L}}.
\end{equation}
$\mathcal T_d$ together with the refinement from
Definition~\ref{def:SFCrefrule} and $T_d^0$ as root element is a refinement
space in the sense of Definition~\ref{def:refspace}.

Any refinement (up to level $\mathcal L$) of $T^0_d$ is a subset of $\mathcal T_d$, and for each $T\in\mathcal T_d$ there exists at least one
refinement $\mathscr T$ of $T^0_d$ with $T\in\mathscr T$.
In this context, we refer to $T^0_d$ as the \textbf{root simplex}.
Furthermore, let $\IL^d$ denote the set of all possible anchor node coordinates of $d$-simplices in $\mathcal{T}_d$, thus
\begin{equation}
 \begin{array}{ccl}
  \IL^2&=&\set{[0,2^\mathcal{L})^2\cap \IZ^2 \abst | y\leq x},\\
  \IL^3&=&\set{[0,2^\mathcal{L})^3\cap \IZ^3 \abst | y\leq z \leq x}.
 \end{array}
\end{equation}
Note that we could have chosen any other of the $S_i$ (scaled by $2^\mathcal{L}$) as the root
simplex and we do not see any advantage or disadvantage in doing so.
\subsection{The type and Tet-id of a $d$-simplex}
\label{sec:type_id}
Making use of Property \ref{property:commref}, we define the following.
\begin{definition}
\label{def:type}
Each $d$-simplex $T\in\mathcal{T}_d$ of level $\ell$ lies in a $d$-cube of the hexahedral mesh
that is part of a uniform level $\ell$ refinement of $[0,2^{\mathcal{L}}]^d$.
This specific cube is the \textbf{associated cube} of $T$ and denoted by $Q_T$.
The $d$-simplex $T$ is a scaled and shifted version of exactly one of the
six tetrahedra $S_i$ that constitute the unit cube, and
we define the \textbf{type} of $T$ as this number, $\type(T):=i$.
\end{definition}
The anchor node of a subcube of level $\ell$ is the particular corner of
that cube with the smallest $x$-, $y$- (and $z$-) coordinates.
This means that for each simplex $T$ in the refinement from Figure \ref{fig:refineddiagram}
the anchor node of $T$ and the anchor node of its associated cube coincide.
Any two $d$-simplices in $\mathcal{T}_d$ with the same associated cube are distinguishable by their type.

From Bey's observation from Figure \ref{fig:refineddiagram} it follows that any simplex in $\mathcal T_d$
can be obtained by specifying a level $\ell$, then choosing one level $\ell$ subcube of the root cube and
finally fixing a type.
This provides motivation for the following definition.
\begin{definition}[Tet-id]
 For $T=[\vec x_0,\dots,\vec x_d]\in\mathcal{T}_d$ we define the \textbf{Tet-id} of $T$ as the tuple of its anchor node
 and type; thus
 \begin{equation}
 \textrm{Tet-id}(T):=(\vec x_0,\type(T)).
 \end{equation}
 \end{definition}

 \begin{corollary}
 \label{cor:tetidunique}
 Let $T,T'\in\mathcal T_d$. Then $T=T'$ if and only if their Tet-ids and levels are the same.
 \end{corollary}

Note that in an arbitrary adaptive mesh there can be simplices with different levels and each 
simplex $T$ has an associated cube of level $\ell(T)$.
In particular, simplices with the same anchor node can have different associated cubes if their
levels are not equal.

 Since any simplex in $\mathcal T_d$ can be specified by the Tet-id and level, the Tet-id provides an im\-por\-tant tool
 for our work.
 The construction of the TM-index in the next section and the algorithms that
 we present in Section \ref{sec:algos} rely on the Tet-id as the fundamental
 data of a simplex.
 All information about a mesh can be extracted from the Tet-id and level of
 each element.

\begin{table}
\begin{center}
\raisebox{5.53ex}{
 \begin{tabular}{|rc|cccc|}
 \hline
 \multicolumn{2}{|c|}{{\mytabvspace Ct}}&\multicolumn{4}{c|}{Child}\\
 \multicolumn{2}{|c|}{2D}  &\mytabvspace $T_0$ & $T_1$ & $T_2$ & $T_3$\\[0.2ex] \hline
  \multirow{2}{*}{b}&0 & 0 & 0 & 0 & 1 \\
  &1 & 1 & 1 & 1 & 0 \\
  \hline
 \end{tabular}
 }
 \begin{tabular}{|rc|ccccc|}
 \hline
 \multicolumn{2}{|c|}{{\mytabvspace Ct}}&\multicolumn{5}{c|}{Child}\\
  \multicolumn{2}{|c|}{{\mytabvspace 3D}}
  & $T_0, \ldots, T_3$ & $T_4$ & $T_5$ & $T_6$ & $T_7$ \\[0.2ex] \hline
  \multirow{6}{*}{b}&0 & 0 & 4 & 5 & 2 & 1\\
  &1 & 1 & 3 & 2 & 5 & 0\\
  &2 & 2 & 0 & 1 & 4 & 3\\
  &3 & 3 & 5 & 4 & 1 & 2\\
  &4 & 4 & 2 & 3 & 0 & 5\\
  &5 & 5 & 1 & 0 & 3 & 4\\\hline
 \end{tabular}
 \end{center}
 \caption[The types of the children of a simplex]
  {For a $d$-simplex $T$ of type $b$ the table gives the types
  $\textrm{Ct}(T_0),\dots,\textrm{Ct}(T_{2^d-1})$ of $T's$ children.  The
  corner-children $T_0,T_1,T_2$ (and in 3D also $T_3$) always have the same type
  as $T$.  }
 \label{table:typesofchildren}
\end{table}

 Since the root simplex has type $0$, in a uniform refinement
 more simplices have type $0$ than any other type.
 However, a close examination of Table \ref{table:typesofchildren} together with a short inductive argument leads to the following proposition.
\begin{proposition}
 In the limit $\mathcal L,\ell\rightarrow\infty$ the different types of simplices
 in a uniform level $\ell$ refinement of $\mathcal{T}_d$
 occur in equal ratios.
\end{proposition}

\subsection{Encoding of the tetrahedral Morton index}
\label{sec:TMindex}
In addition to the anchor coordinates the TM-index also depends on the types of
all ancestors of a simplex. In order to define the TM-index we
start by giving a formal
definition of the interleaving operation and some additional information.
\begin{definition}
 We define the \textbf{interleaving} $a \ainter b$  of two $n$-tuples
 $a=$ $(a_{n-1},$ $\dots,a_0)$ and 
 $b=(b_{n-1},\dots,b_0)$  as the $2n$-tuple obtained
 by alternating the entries of $a$ and $b$:
 \begin{equation}
  a \ainter b \, {:=} \, (a_{n-1},b_{n-1},\dots,a_0,b_0).
 \end{equation}
 The interleaving of more than two $n$-tuples $a^1,\ldots,a^m$ is defined analogously as the $mn$-tuple
 \begin{equation}
a^1 \ainter \cdots \ainter a^m \, {:=} \, (a^1_{n-1},a^2_{n-1},\ldots,a^m_{n-1},a^1_{n-2},\ldots,a^{m-1}_0,a^m_0).
 \end{equation}
\end{definition}

\begin{remark}
The TM-index of a $d$-simplex $T\in\mathcal T_d$ that we are
going to define is constructed by interleaving $d+1$ $\mathcal L$-tuples, where
the first $d$ are the binary re\-pre\-sentations of the coordinates of $T$'s anchor
node and the last
is the tuple consisting of the types of the ancestors of $T$.
\end{remark}

\begin{definition}
\label{def:XYZ}
Let $T\in \mathcal{T}_3$ be a tetrahedron of refinement level $\ell$
 with anchor node $\vec x_0 = (x,y,z)^T\in\IL^3$.
 Since $x,y,z\abst{\in}\IN_{0}$ with $0\leq x,y,z < 2^\mathcal{L}$, we can express them as binary numbers with $\mathcal L$ digits, writing
 \begin{equation}
 x = \sum_{j=0}^{\mathcal{L}-1} x_j 2^j,\quad y = \sum_{j=0}^{\mathcal{L}-1} y_j 2^j,\quad z = \sum_{j=0}^{\mathcal{L}-1} z_j 2^j.
 \end{equation}
 We define the $\mathcal L$-tuples $X$, $Y$, and $Z$
 as the $\mathcal L$-tuples consisting of the binary digits of $x$, $y$, and $z$; thus,
 \begin{subequations}
 \begin{align}
 X = X(T)&:=(x_{\mathcal L -1},\dots,x_0),\\
 Y = Y(T)&:=(y_{\mathcal L -1},\dots,y_0),\\
 Z = Z(T)&:=(z_{\mathcal L -1},\dots,z_0).
 \end{align}
 \end{subequations}
 In 2D we get the same definitions with $X$ and $Y$, leaving out the $z$-coordinate.
\end{definition}

\begin{definition}\label{def:BvonT}
For a $T\in\mathcal T_d$ of level $\ell$ and each $0\leq j \leq \ell$ let $T^j$ be the (unique) ancestor of $T$ of level $j$. In particular, $T^\ell=T$.
We define $B(T)$ as the $\mathcal L$-tuple consisting of the types of $T$'s ancestors in the first $\ell$ entries, starting
with $T^1$. The last $\mathcal{L}-\ell$ entries of $B(T)$ are zero:
\begin{equation}
 B=B(T):=\left(\underbrace{\type(T^1),\type(T^2),\ldots,\type(T)}_{\ell\textrm{ entries}},0,\ldots,0\right),
\end{equation}
Thus, if we write $B$ as an $\mathcal{L}$-tuple with indexed entries $b_i$
\begin{equation}
B=B(T)=(b_{\mathcal L-1},\dots,b_0) \abst{\in} \set{0,\dots,d!-1}^\mathcal{L},
\end{equation}
then the $i$th entry $b_i$ is given as
\begin{equation}
 b_i= \left\lbrace\begin{matrix} \type(T^{\mathcal L -i}) & \mathcal L-1 \geq i \geq \mathcal L -\ell, \\[1ex]
                0      & \mathcal L -\ell > i \geq0.\hphantom{i-\ell}
                \end{matrix}\right.
\end{equation}
\end{definition}

\begin{definition}[tetrahedral Morton Index]\label{def:mortonindex}
We define the \textbf{tetrahedral Morton index} (\textbf{TM-index}) $m(T)$ of a $d$-simplex $T\in\mathcal{T}_d$ as the interleaving
of the $\mathcal L$-tuples $Z$ (for tetrahedra), $Y$, $X$ and $B$. Thus,
\begin{subequations}
 \label{eq:mortondef}
 \begin{equation}
   m(T)\abst{:=} Y\ainter X \ainter B
 \end{equation}
 for triangles and
 \begin{equation}
   m(T)\abst{:=} Z\ainter Y\ainter X \ainter B
 \end{equation}
 \end{subequations}
 for tetrahedra.
\end{definition}
This index resembles the well-known Morton index or Z-order for $d$-dimensional cubes, which we denote by $\widetilde m$ here.
For such a cube $Q$ the Morton index is usually defined as the bitwise interleaving of its coordinates.
Thus $\widetilde m(Q) = Z\inter Y\inter X$, respectively, $\widetilde m(Q) =
Y\inter X$; see \cite{Morton66,SundarSampathBiros08,BursteddeWilcoxGhattas11}
as well as Section~\ref{sec:mortonsfc}.

As we show in Section \ref{sec:algos}, the TM-index can be computed from the Tet-id of $T$ with no further information given.
Thus, in an implementation it is not necessary to store the $\mathcal{L}$-tuple $B$.

The TM-index of a $d$-simplex builds up from packs of $d$ bits $z_i$ (for
tetrahedra), $y_i$, and $x_i$ followed by a type $b_i\in\set{0,\ldots,d!-1}$.
Since $d! = 2 < 4$ for $d=2$, we can interpret the 2D TM-index as
a quarternary number with digits $(y_ix_i)_2$ and $b_i$:
\begin{subequations}
\label{eq:mortonoctal}
\begin{equation}
\begin{array}{ccc}
\label{eq:mortonquater}
 m(T)  &=& ((y_{\mathcal L -1}x_{\mathcal L -1})_2,b_{\mathcal L -1},\ldots,(y_0x_0)_2,b_0)_4\\[1ex]
 &=& \displaystyle\sum_{i=0}^{\mathcal L-1}\left( (2y_i+x_i)4^{2i+1} + b_i4^{2i}\right).
\end{array}
\end{equation}
Similarly we can interpret it as an octal number with digits $(z_iy_ix_i)_2$ and $b_i$ for $d=3$, since then $d! = 6 < 8$:
\begin{equation}
\begin{array}{ccc}
 m(T) & = &((z_{\mathcal L -1}y_{\mathcal L -1}x_{\mathcal L -1})_2,b_{\mathcal L -1},\ldots,(z_0y_0x_0)_2,b_0)_8\\[1ex]
 &=& \displaystyle\sum_{i=0}^{\mathcal L-1}\left( (4z_i+2y_i+x_i)8^{2i+1} + b_i8^{2i}\right).
 \end{array}
\end{equation}
\end{subequations}
The entries in these numbers are only nonzero up to the level $\ell$ of $T$, since
$x_{\mathcal{L}-i} = y_{\mathcal{L}-i} = (z_{\mathcal{L}-i} = ) b_{\mathcal{L}-i} = 0$ for all $i>\ell$.
The octal/quarternary representation \eqref{eq:mortonoctal} directly gives an
order on the TM-indices, and therefore it is possible to construct an
SFC from it, which we examine further in Section
\ref{sec:SFC}.
We use $m(T)$ to denote both the $(d+1)\mathcal L$-tuple from \eqref{eq:mortondef} and the number given by \eqref{eq:mortonoctal}.

Let us look at Figure \ref{fig:mortonandwrongmorton} for a short example to motivate this definition of the TM-index.
Since the anchor coordinates and the type together with the level uniquely determine a $d$-simplex in $\mathcal T_d$,
one could ask why we do not define the index to be $((Z\ainter )Y\ainter X,\type(T))$,
a pair consisting of the Morton index of the associated cube of $T$ and the type of $T$.
This index was introduced for triangles in a slightly modified version as semiquadcodes in \cite{OtooZhu93}
and would certainly require less information since the computation of the sequence $B$ would not be necessary.
However, it  results in an SFC that traverses the leaf cubes of a refinement in the usual Z-order and inside of each cube it
traverses the $d!$ different simplices in the order $S_0,\ldots,S_{d!-1}$.
As a result, there can be simplices $T$ whose children are not traversed as a group,
which means that there is a tetrahedron $T'$,
which is not an ancestor or descendant of $T$, such that some child $T_i$ of $T$ is traversed before $T'$ and $T'$ is traversed before another child $T_j$ of $T$.
Thus, this curve is not an SFC in the sense of Definition~\ref{def:sfcindex}.
In contrast to this, Theorem~\ref{thm:IndexProps} states that the TM-index is in fact an SFC-index.
Figure \ref{fig:mortonandwrongmorton} compares the two approaches for a uniform level 2 refinement of $T^0_2$.

\begin{figure}
\center
\begin{minipage}{0.48\textwidth}
   \includegraphics[width=0.9\textwidth]{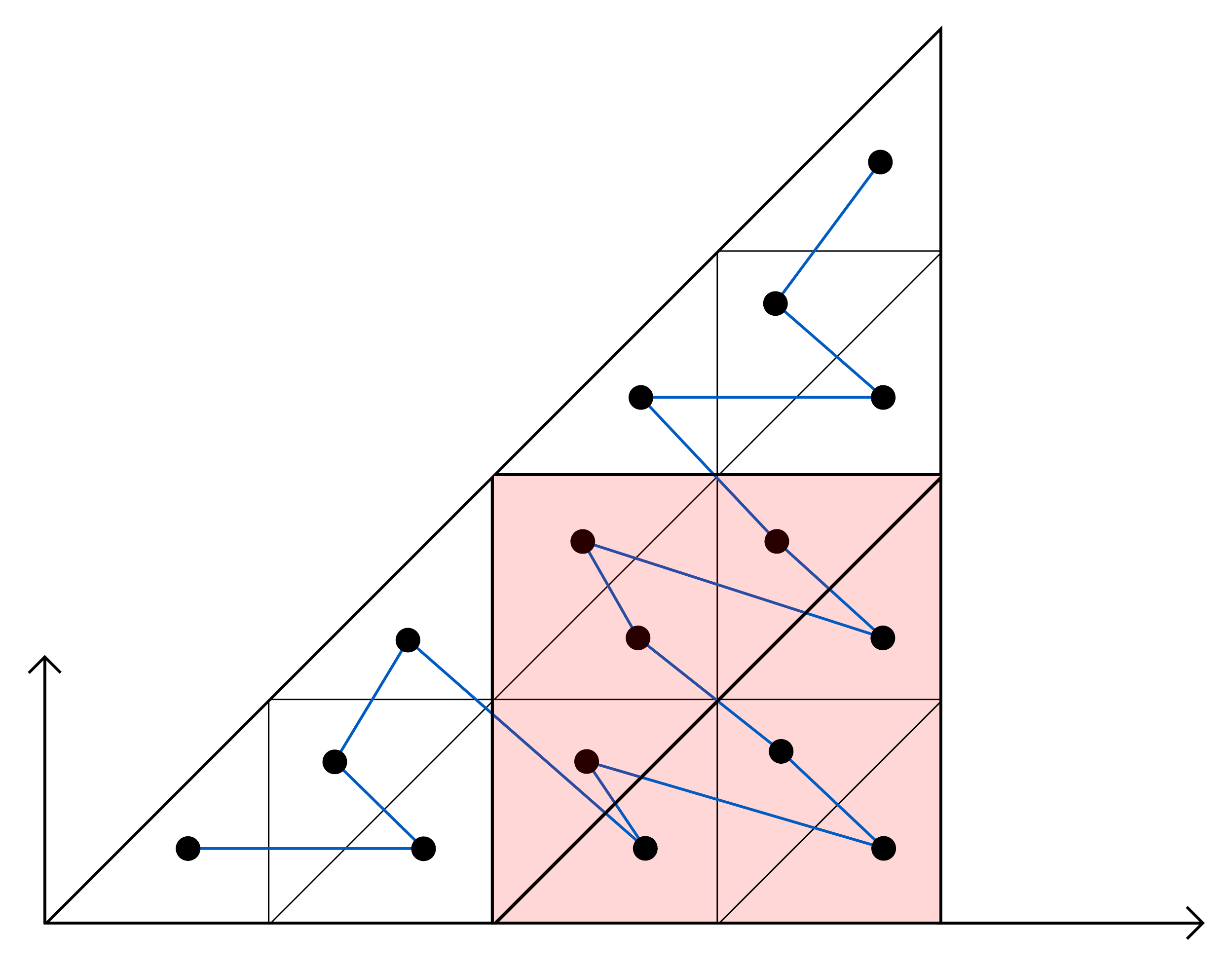}
\end{minipage}
\begin{minipage}{0.48\textwidth}
   \includegraphics[width=0.9\textwidth]{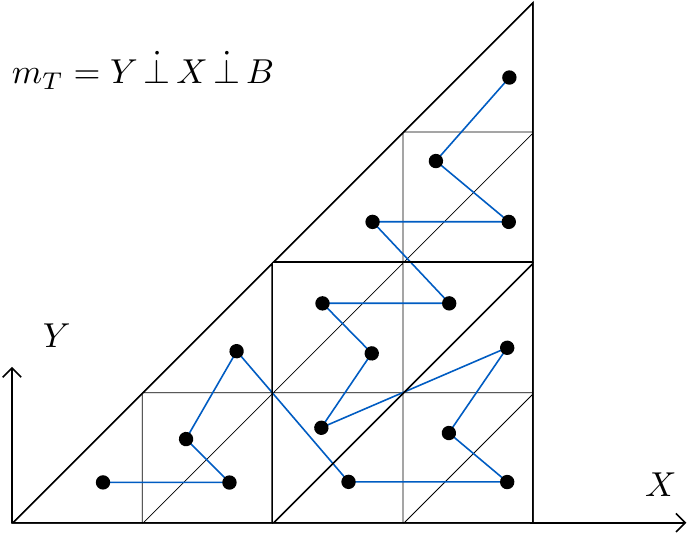}
\end{minipage}
   \caption[The TM curve preserves locality]
   {Comparing a straightforward definition of a Morton-type SFC
    with our approach.
    Left: The curve arising from taking the Morton order of the quadrants
    and only dividing into triangles on the last level.
    Thus the index is  $(Y\inter X,\type(T))$. As we see
    on the two coarse triangles that are shaded, the children of a level
    1 triangle are not necessarily traversed before any other triangle is
    traversed.
    Thus, it breaks the locality property that is part of the definition of 
    an SFC, and therefore this index
    is not suitable for our purposes. Right: The curve arising from the
    TM-index from our Definition \ref{def:mortonindex}. We see
    that for each level 1 triangle its four children are traversed as a
    group. Theorem \ref{thm:IndexProps} states that the curve is in fact 
    a proper SFC in the sense of Definition~\ref{def:sfcindex} and
    thus the locality property holds for any parent triangle/tetrahedron.
    The order in which the children are traversed depends (only) on the
    type of the parent and is different from Bey's order given by
    \eqref{eq:childnumbers}.}
   \figlabel{fig:mortonandwrongmorton}
\end{figure}

\subsection{A different approach to derive the TM-index}

There is another interpretation of the TM-index, which is particularly useful for the AMR algorithms presented in Section \ref{sec:algos}.
In order to define it we introduce the concept of the so-called cube-id.
According to Figure \ref{fig:sechstetra} we number the $2^d$ corners of a  $d$-dimensional cube by $c_0,\ldots,c_{2^{d-1}}$ in a $zyx$-order ($x$ varies fastest).
When refining a cube to $2^d$ children, each child has exactly one of the $c_i$ as a corner, and
it is therefore convenient to number the children by $c_0,\ldots,c_{2^{d-1}}$ as well.
For the child $c_i$ we call the number $i$ the cube-id of that child; see Figure \ref{fig:cubeid} for an illustration.
Since each cube $Q$ that is not the root has a unique parent, it also has a unique cube-id.
This cube-id can easily be computed by interleaving the last significant bits
of the $z$- (in 3D), $y$-, and $x$-coordinates of $Q$'s anchor node.
\begin{definition}
Because each $d$-simplex $T\abst{\in}\mathcal{T}_d$ of level $\ell$ has a unique associated cube we define
the \textbf{cube-id} of $T$ to be the cube-id of the associated cube of $T$, that is, the $d$-cube of level
$\ell$ that has the same anchor node as $T$.
\end{definition}
If $X$, $Y$ (and $Z$) are as is in Definition \ref{def:XYZ} then we can write the cube-id of $T$'s ancestors as
\begin{equation}
\begin{array}{lc}
  \cid(T^i) = (y_ix_i)_2  & \textrm{in 2D},\\[1ex]
\cid(T^i) = (z_iy_ix_i)_2 & \textrm{in 3D},
\end{array}
\end{equation}
and therefore using \eqref{eq:mortonoctal} we can rewrite the TM-index of $T$ as
\begin{equation}\label{eq:mortonentries}
 m(T) = (\cid(T^1),\type(T^1),\dots,\cid(T^\ell),\type(T^\ell),0,\dots,0)_{2^d}.
\end{equation}
This resembles the Morton index of the associated cube $Q_T$ of $T$, since we can write this as
\begin{equation}
 \widetilde m({Q_T}) = (\cid(Q^1),\dots,\cid(Q^\ell),0,\dots,0)_{2^d}.
\end{equation}

\begin{figure}
   \center
   \def\svgwidth{0.35\textwidth}
   \includegraphics{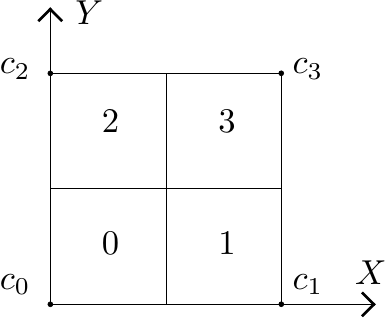}
   \hspace{10ex}
   \def\svgwidth{0.35\textwidth}
   \includegraphics{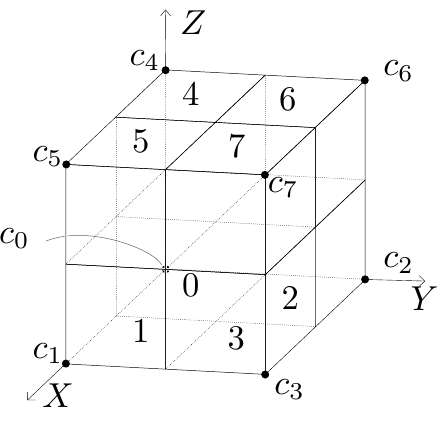}
    \caption[The cube-id]
  {Left: A square is refined to four children, each of which
  corresponds to a corner of the square. The number of the corner is the cube-id
  of that child. Right: In three dimensions a cube is refined to eight children.
  Their cube-ids and corner numbers are shown as well.}
    \figlabel{fig:cubeid}
\end{figure}

\subsection{Properties of the TM-index}\label{sec:Prop}
In this section we show that the TM-index is an SFC-index.
As a first result, we show that $m\times\ell$ is injective.
\begin{proposition}\label{prop:uniquetet}
Together with a refinement level $\ell$, the TM-index $m(T)$ 
u\-nique\-ly determines a $d$-simplex in $\mathcal T_d$.
\begin{proof}
If $\ell=0$, then there is only one simplex of level $\ell$ in $\mathcal T_d$, which is $T^0_d$.
So let $\ell>0$ and $m=m(T)$ be given as in \eqref{eq:mortondef}, and let $\ell$ be the level of $T$.
From $m$ we can compute the $x$- ,$y$- (and $z$-)
coordinates of the associated cube of $T$.
We can also compute the type of $T$ from the TM-index.
By Corollary \ref{cor:tetidunique} this information uniquely determines $T$.
\end{proof}
\end{proposition}
For the Morton index $\widetilde m$ for cubes
the following important properties are known\cite{SundarSampathBiros08}:
\begin{enumerate}[(i)]
 \item A Morton index of a cube $Q$ is the prefix of an index of a cube $P$ of higher level than $Q$ if and only if
       $P$ is a descendant of $Q$.
 \item The Morton indices of the descendants of a parent cube are larger than or equal to the index of the parent cube.
 \item Refining only changes the SFC locally. Thus, if $Q$ is a cube and $P$ is a cube with $\widetilde m(Q)<\widetilde m(P)$
       and $P$ is not a descendant of $Q$, then $\widetilde m({Q'})<\widetilde m(P)$ for each descendant $Q'$ of $Q$.
\end{enumerate}
Property (iii) defines a hierarchic invariant of the SFC
that is specific to our construction (see Figure~\ref{fig:mortonandwrongmorton}).
We show below that properties (i), (ii) and (iii) hold for $d$-simplices and
the TM-index described by \eqref{eq:mortondef}.
Proposition~\ref{prop:uniquetet} together with (ii) and (iii)
imply that the TM-index is an SFC-index in the sense of Definition~\ref{def:sfcindex}.
\begin{theorem}\label{thm:IndexProps}
 For arbitrary $d$-simplices $T\neq S\in\mathcal T_d$ the TM-index satisfies
the following:
\begin{enumerate}[(i)]
 \item  If $\ell(T) < \ell(S)$, then $m(T)$ is a prefix of $m(S)$ if and only if $S$ is a descendant of $T$.
 \item  If $T$ is an ancestor of $S$ then $m(T)\leq m({S})$.
 \item  If $m(T)<m(S)$ and $S$ is no descendant of $T$, then for each descendant $T'$ of $T$
we have
\begin{equation}
 m(T)\leq m({T'}) < m(S).
\end{equation}
\end{enumerate}
\end{theorem}
The proof of Theorem \ref{thm:IndexProps} requires some work and we need to show a technical result first.
Hereby, we consider only the 3D case, since for 2D the argument is completely analogous.
We define an embedding of the set of all TM-indices into the set of  Morton
indices for 6D cubes. Since the properties (i)--(iii) hold for
these cubes it follows that they hold
for tetrahedra as well.
To this end, for a given tetrahedron $T\abst{\in}\mathcal{T}_3$ we interpret each entry $b_j$ of $B(T)$ as a $3$-digit binary number
\begin{equation}
 b_j \abst{=} (b_j^2\, b_j^1\, b_j^0)_2,
\end{equation}
which is possible since $b_j\,{\in}\, \set{0,\dots,5}$.
We obtain three new $\mathcal{L}$-tuples $B^2, B^1, B^0$ satisfying
\begin{equation}
 B\abst{=} B^2\ainter B^1\ainter B^0 ,
\end{equation}
and thus we can rewrite the TM-index as
\begin{equation}
\label{eq:rewritemorton}
 m(T)\abst{=} Z\ainter Y\ainter X \ainter B^2\ainter B^1\ainter B^0.
\end{equation}
Note that we can interpret each $B^i$ as an $\mathcal L$-digit binary number
for which we have $0\leq B^i<2^\mathcal{L}$.
Now let $\mathcal{Q}$ denote the set of all 6D cubes that are a child of the cube $Q_0 \, {:=} \, [ 0, 2^\mathcal{L}]^6$:
\begin{equation}
 \mathcal{Q}\abst{=} \set{Q \abst{|} Q \textnormal{ is a descendant of } Q_0
 \textnormal{ of level } 0\leq \ell\leq\mathcal{L}}.
\end{equation}
Since a cube $Q\,{\in}\, \mathcal{Q}$ is uniquely determined by the six coordinates $(x_0,\dots,x_5)$ of its anchor node plus its level $\ell$,
we also write
$Q\abst{=} Q_{(x_0,\dots,x_5),\ell}$.
Note that the Morton index for a cube can be defined as the bitwise interleaving of its anchor node coordinates \cite{Morton66}:
\begin{equation}
 \widetilde{m}(Q)\abst{=} X^5\ainter X^4\ainter X^3\ainter X^2\ainter X^1\ainter X^0.
\end{equation}

\begin{proposition}
\label{prop:mapfromTtoQ}
The map
 \begin{equation}\label{eq:MapTQ}
  \begin{array}{rccl}
   \Phi\colon&\mathcal{T}_3 & \longrightarrow & \mathcal{Q}, \\
   &T & \longmapsto & Q_{(B^0(T),B^1(T),B^2(T),x(T),y(T),z(T)),\ell(T)} \\
  \end{array}
 \end{equation}
 is injective and satisfies
 \begin{equation}
   \widetilde m({\Phi(T)})\abst{=} m(T).
 \end{equation}
 Furthermore, it fulfills the property that $T'$ is a child of $T$ if and only if $\Phi(T')$ is a child of $\Phi(T)$.
 \begin{proof}
  The equation $m(T)=\widetilde m({\Phi(T))}$ follows directly from the definitions of the TM-indices on $\mathcal T_3$ and $\mathcal Q$.
  From Lemma \ref{prop:uniquetet} we conclude that $\Phi$ is injective.
  Now let $T',T\abst{\in}\mathcal{T}_3$, where $T'$ is a child of $T$.
  Furthermore, let $\ell=\ell(T)$.
  We know that $Q':=\Phi(T')$ is a child of $Q:=\Phi(T)$ if and only if for each $i\in\set{0,\dots,5}$
  it holds that
  \begin{equation}\label{eq:xiequation}
   x_i(Q')\in\set{x_i(Q),x_i(Q)+2^{\mathcal L -(\ell+1)}}.
  \end{equation}
  Because of the underlying cube structure (compare Figure~\ref{fig:refineddiagram}) we know that the $x$-coordinate
  of the anchor node of $T'$ satisfies
  \begin{equation}
   x(T')\in\set{x(T),x(T)+2^{\mathcal L -(\ell+1)}},
  \end{equation}
  and likewise for $Y(T')$ and $Z(T')$.
  Therefore, \eqref{eq:xiequation} holds for $i=3,4,5$.
  By definition $B^j(T')$ is the same as $B^j(T)$ except at position $\mathcal L-(\ell+1)$, where
   \begin{equation}
    B^j(T')_{\mathcal L-(\ell+1)} = b_{\mathcal L-(\ell+1)}^j(T')\in\set{0,1}
   \end{equation}
   and
   \begin{equation}
  B^j(T)_{\mathcal L-(\ell+1)} = 0.
   \end{equation}
   Hence, we conclude that \eqref{eq:xiequation} also holds for $i=0,1,2$.
   So $\Phi(T')$ is a child of $\Phi(T)$.

   To show the other implication, let us assume that
   $\Phi(T')$ is a child of $\Phi(T)$.
   Since $\ell(T')=\ell(\Phi(T'))>0$,  $T'$ has a parent and
   we denote it by $P$.
   In the argument above we have shown that $\Phi(P)$ is the parent of $\Phi(T')$
   and because each cube has a unique parent
   the identity $\Phi(P)=\Phi(T)$ must hold. Therefore, we get $P=T$ since $\Phi$ is injective; thus, $T'$ is the child of $T$.
\end{proof}
\end{proposition}
Inductively we conclude that $T'$ is a descendant of $T$ if and only if $\Phi(T')$ is a de\-scen\-dant of $\Phi(T)$.
Now Theorem \ref{thm:IndexProps} follows, because the desired properties (i)--(iii) hold for the Morton index of cubes~\cite{SundarSampathBiros08}.
\qed

\subsection{The space-filling curve associated to the TM-index}
\label{sec:SFC}
By interpreting the TM-indices as $2^d$-ary numbers as in \eqref{eq:mortonoctal}
we get a total order on the set of all possible TM-indices,
and therefore it gives rise to an SFC for any refinement $\mathscr{S}$ of $T^0_d$.
In this section we further examine the SFC derived from the TM-index.
We give here a recursive description of it,
similarly to how it is done for the Sierpinski curve and other
cubical SFC by Haverkort and van Walderveen \cite{HaverkortWalderveen10}.

Part (iii) of Theorem \ref{thm:IndexProps} tells us that the descendants of a simplex $T$ are traversed
before any other simplices with a higher TM-index than $T$ are traversed.
However, the order that the children of $T$ have relative to each other can be different to the
order of children of another simplex $T'$.
In particular the order of the simplices defined by the TM-index differs from the order
\eqref{eq:childnumbers} defined by Bey.
We observe this behavior in 2D in Figure~\ref{fig:mortonandwrongmorton} on the right-hand side:
For the level 1 triangles of type $0$ the children are traversed in the order
\begin{equation}
 T_0, T_1, T_3, T_2
\end{equation}
and the children of the level 1 triangle of type $1$ are traversed in the order
\begin{equation}
 T_0, T_3, T_1, T_2.
\end{equation}
In fact, the order of the children of a simplex $T$ depends only on the type of $T$, as
we show in the following Proposition.
\begin{proposition}
\label{proposition:locindexdeptype}
 If $T,T'\in \mathcal T_d$ are two $d$-simplices of given type $b = \type(T)=\type(T')$,
 then there exists a unique permutation $\sigma \equiv \sigma_b$ of
 $\set{0,\dots,2^d-1}$ such that
 \begin{equation}
 \begin{array}{c}
  m(T_{\sigma(0)})<m(T_{\sigma(1)})<\dots<m(T_{\sigma(2^d-1)}), \\[1ex]
  \textrm{and}\\[1ex]
  m(T'_{\sigma(0)})<m(T'_{\sigma(1)})<\dots<m(T'_{\sigma(2^d-1)}).
 \end{array}
 \end{equation}
 Thus, the children of $T$ and the children of $T'$ are in the same order with respect to their TM-index.
 \begin{proof}
 By ordering the children of $T$ and $T'$ with respect to their TM-indices,
 we obtain $\sigma$ and $\sigma'$ with
 \begin{equation}
 \begin{array}{c}
  m(T_{\sigma(0)})<m(T_{\sigma(1)})<\dots<m(T_{\sigma(2^d-1)}), \\[1ex]
  m(T'_{\sigma'(0)})<m(T'_{\sigma'(1)})<\dots<m(T'_{\sigma'(2^d-1)}).
 \end{array}
 \end{equation}
 These permutations are well-defined and unique with this property because different
 simplices of the same level never have the same TM-index; see Proposition
 \ref{prop:uniquetet}.
 It remains to show that $\sigma'=\sigma$.
Let $\ell=\ell(T)$ and $\ell'=\ell(T')$.
Since the TM-indices of the children of $T$ do all agree up to level $\ell$, we see,
using the notation from \eqref{eq:mortonoctal}, that their order $\sigma$ depends
only on the $d+1$ numbers ($z$ is omitted for $d=2$)
\begin{equation}
z_{\mathcal L-(\ell+1)}(T_i),\quad y_{\mathcal L-(\ell+1)}(T_i),\quad x_{\mathcal L-(\ell+1)}(T_i)\quad\textrm{and}\quad b_{\mathcal L-(\ell+1)}(T_i).
\end{equation}
The same argument applies to $\sigma'$ and $\ell'$.
From now on we carry out the computations for $d=3$.
Since $\type(T)=\type(T')$ we can write
\begin{equation}
 T = \lambda T' + \vec c,
\end{equation}
with
\begin{equation}
\lambda = 2^{\ell'-\ell},\quad\vec c = \begin{pmatrix}
                    x(T)-x(T')\\[0.6ex]
                    y(T)-y(T')\\[0.6ex]
                    z(T)-z(T')
                                  \end{pmatrix}.
\end{equation}
Since the refinement rules \eqref{eq:childnumbers} commute with scaling and translation we also obtain
\begin{equation}
 T_i = \lambda T'_i + \vec c
\end{equation}
for the children of $T$ and $T'$
and therefore
\begin{equation}
 b_{\mathcal L-(\ell+1)}(T_i)=\type(T_i)=\type(T'_i)=b_{\mathcal L-(\ell'+1)}(T'_i)
\end{equation}
for $0\leq i < 2^d$.
Furthermore, %
we have
\begin{equation}
x_{\mathcal L-(\ell+1)}(T_i) = (x(T_i)-x(T))2^{-\mathcal{L}+(\ell+1)}
\end{equation}
from which we derive
\begin{equation}
 \begin{array}{ccl}
  x_{\mathcal L-(\ell+1)}(T_i)&=& \lambda(x(T_i') - x(T')) 2^{-\mathcal{L}+(\ell+1)}\\[0.8ex]
  &=& 2^{\ell'-\ell} (x(T_i') - x(T'))
  2^{-\mathcal{L}+(\ell+1)} \\[0.8ex]
  &=& (x(T_i') - x(T')) 2^{-\mathcal{L}+(\ell'+1)} \\[0.8ex]
  &=& x_{\mathcal L-(\ell'+1)}(T'_i),
 \end{array}
\end{equation}
and analogously
\begin{equation}
\begin{array}{ccc}
 y_{\mathcal L-(\ell+1)}(T_i) &=&y_{\mathcal L-(\ell'+1)}(T'_i), \\[0.8ex]
 z_{\mathcal L-(\ell+1)}(T_i) &=&z_{\mathcal L-(\ell'+1)}(T'_i).
\end{array}
\end{equation}
This shows that the tetrahedral Morton order of the children of $T$ and $T'$ are the same and $\sigma'$ must equal $\sigma$.
\end{proof}
\end{proposition}

\begin{definition}
 \label{def:localindex}
 Let $T\in\mathcal T_d$ such that $T$'s parent $P$ has type $b$ and $T$ is the $i$th child of $P$
 according to Bey's order \eqref{eq:childnumbers},
  $0\leq i<2^d$. We call the number $\sigma_b(i)$ the \textbf{local index} of
 the $d$-simplex $T$ and use the notation
 \begin{equation}
   \label{eq:Ilocsigma}
I_\mathrm{loc}(T):=\sigma_b(i)
 \end{equation}
 to denote the child number in the TM-ordering, subsequently written TM-child.
 By definition, the local index of the root simplex is zero, $I_\mathrm{loc}(T^0_d):=0$.
 Table \ref{table:BeytoIndex} lists the local indices for each parent type.
\end{definition}
Thus, we know for each type $0\leq b<d!$ how the children of a tetrahedron of type $b$ are traversed.
This gives us an approach for describing the SFC arising from the TM-index in a recursive fashion \cite{HaverkortWalderveen10}.
By specifying for each possible type $b$ the order and types of the children of a type $b$ simplex, we can build up the SFC.
In Figure \ref{fig:haverkortSFC} %
we describe the SFC for triangles in this way.
In three dimensions it is not convenient to draw the six pictures for the different types, but
the SFC can be derived similarly from Tables \ref{table:typesofchildren} and \ref{table:BeytoIndex}.
\begin{table}
\centering
\raisebox{6ex}{
\begin{tabular}{|rc|cccl|}
\hline
\multicolumn{2}{|c|}{\mytabvspace$I_\mathrm{loc}$}&\multicolumn{4}{c|}{Child}\\
 \multicolumn{2}{|c|}{2D}  &\mytabvspace  $T_0$ & $T_1$ & $T_2$ & $T_3$ \\[0.2ex]\hline
 \multirow{2}{*}{b}&\mytabvspace0 & 0 & 1 & 3 & 2    \\[0.2ex]
 &1 & 0 & 2 & 3 & 1   \\ \hline
\end{tabular}
}
\begin{tabular}{|rc|cccccccl|}
\hline
\multicolumn{2}{|c|}{\mytabvspace$I_\mathrm{loc}$}&\multicolumn{8}{c|}{Child}\\
\multicolumn{2}{|c|}{3D}&
\mytabvspace   $T_0$ & $T_1$ & $T_2$ & $T_3$ & $T_4$ & $T_5$ & $T_6$ & $T_7$\\[0.2ex]\hline
 \multirow{6}{*}{b}&\mytabvspace0 & 0 & 1 & 4 & 7 & 2 & 3 & 6 & 5   \\[0.2ex]
 &1 & 0 & 1 & 5 & 7 & 2 & 3 & 6 & 4   \\[0.2ex]
 &2 & 0 & 3 & 4 & 7 & 1 & 2 & 6 & 5   \\[0.2ex]
 &3 & 0 & 1 & 6 & 7 & 2 & 3 & 4 & 5   \\[0.2ex]
 &4 & 0 & 3 & 5 & 7 & 1 & 2 & 4 & 6   \\[0.2ex]
 &5 & 0 & 3 & 6 & 7 & 2 & 1 & 4 & 5   \\ \hline
\end{tabular}
\caption[The local index of the children of a simplex]
  {The local index of the children of a $d$-simplex $T$ according to the
  TM-ordering.  For each type $b$, the $2^d$ children $T_0,\dots,T_{2^d-1}$ of a
  simplex of this type can be ordered according to their TM-indices.  The
  position of $T_i$ according to the TM-order is the local index
  $I_\mathrm{loc}(T_i) = \sigma_b (i)$.}
\label{table:BeytoIndex}
\end{table}

\begin{figure}
   \center
   \def\svgwidth{0.9\textwidth}
   \includegraphics{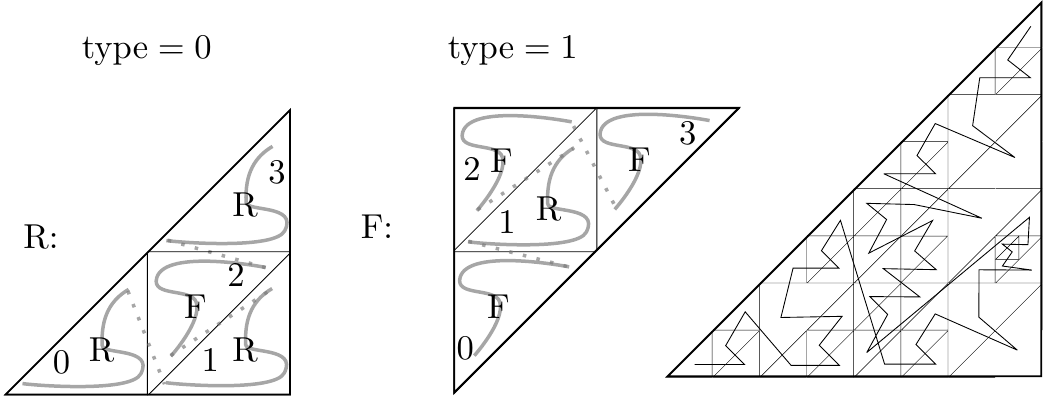}
    \caption[A recursive description of the TM curve]
  {Left: Using the notation from \cite{HaverkortWalderveen10} we
   recursively describe the SFC arising from the TM-index for triangles. The
   number inside each child triangle
   is its local index.  $R$ denotes the refinement scheme for type 0
   triangles and $F$ for type 1 triangles.  This pattern can be obtained from
   Tables \ref{table:typesofchildren} and \ref{table:BeytoIndex}.  Right: The SFC
   for an example adaptive refinement of the root triangle.}
   \figlabel{fig:haverkortSFC}
\end{figure}

\section{Low-level algorithms on simplices}
\label{sec:algos}
In this section we present fundamental algorithms that operate on $d$-simplices in $\mathcal{T}_d$.
These algorithms include computations of parent and child simplices, computation of face-neighbors and computations
involved with the TM-index.
To simplify the notation we carry out all algorithms for tetrahedra and then describe how to modify them for triangles.
We introduce the data type \Tet{}
and do not distinguish between the abstract concept
of a \Tet{} and the geometric object (tetrahedron or triangle) that it represents.
The data type \Tet{} $T$ has the following members:
\begin{itemize}
 \item $T.\ell$ --- the refinement level of $T$;
 \item $T.\vec{x}=(T.x,T.y,T.z)$ --- the $x$- ,$y$- and $z$-coordinates of $T$'s anchor node,
    also sometimes referred to as $T.x_0$, $T.x_1$, and $T.x_2$;
 \item $T.b$ --- the type of $T$.
 \end{itemize}
 In 2D computations the parameter $T.z$ is not present.
 To avoid confusion we use the notation $\vec{x}$ to denote vectors in $\IZ^d$ and $x$ (without arrow) for
 integers, thus numbers in $\IZ$.
From Corollary \ref{cor:tetidunique} we know that the values stored in a \Tet{}
suffice to uniquely identify a $d$-simplex $T\abst{\in}\mathcal{T}$.
\begin{remark}[Storage requirement]
The algorithms that we present in this section only need this data as input for
a simplex resulting in a fixed storage size per \Tet.
If, for example, the maximum level $\mathcal{L}$ is 32 or less, then the
coordinates can be stored in one $4$-byte integer per dimension, while the
level and type occupy one byte each, leading to a total storage of
\begin{equation}
\begin{array}{cccr}
2 \times 4+1+1 &=& 10 &\textrm{bytes per \Tet{} in 2D,}\\
3 \times 4+1+1 &=& 14 &\textrm{bytes per \Tet{} in 3D.}
\end{array}
\end{equation}
\end{remark}
\begin{remark}[Runtime]
Most of these algorithms run in constant time independent of the maximum level $\mathcal L$.
The only operations
using a loop over the level
$\mathcal L$ or $T.\ell$, thus having $\mathcal{O}(\mathcal{L})$ runtime, are
computing the consecutive index from a \Tet{} and initializing a \Tet{}
according to a given consecutive index.
Hence, we show how to replace repetitive calls of these relatively
involved algorithms by more efficient constant-time ones.
\end{remark}

\subsection{The coordinates of a $d$-simplex}
The coordinates of the $d+1$ nodes of a $d$-simplex $T$
can be obtained easily from its Tet-id, the relation \eqref{eq:coordsofSb}, and simple
arithmetic:
 If $T$ is a $d$-simplex of level $\ell$, type $b$ and anchor node $\vec{x}_0\in\IZ^d$, then
 \begin{equation}
 \label{eq:TcoordfromS}
  T = 2^{\mathcal{L}-\ell}S_b + \vec{x}_0.
 \end{equation}
Hence, in order to compute the coordinates of $T$ we can take the coordinates of $S_b$, as given in \eqref{eq:coordsofSb},
and then use relation \eqref{eq:TcoordfromS}.
A closer look at \eqref{eq:coordsofSb} reveals that it is not necessary to examine all coordinates of $S_b$
in order to compute the $x_i$, but that they can also be computed arithmetically.
This computation is carried out in Algorithm \ref{alg:coordsofT}.

\begin{algorithm}
\caption{\texttt{Coordinates}(\aTet{} $T$)}
\label{alg:coordsofT}
\DontPrintSemicolon
\algoresult{Array of coordinates of all of $T$`s vertices.}\;
$X\gets (T.\vec{x},0,0,0)$\;
$h\gets 2^{\mathcal{L}-\ell}$\;
$i\gets \lfloor\frac{T.b}{2}\rfloor$\Comment{Replace with $i\gets T.b$ for 2D}
\eIf{$T.b\algomod 2 = 0$}
{$j\gets (i+2) \algomod 3$}
{$j\gets (i+1) \algomod 3$}
$X[1]\gets X[0]+he_i$\;
$X[2]\gets X[1]+he_j$\Comment{Replace with $X[2]\gets X[0] + (h,h)^T$ for 2D}
$X[3]\gets X[0]+(h,h,h)^T$\Comment{Remove this line for 2D}
\Return $X$\;
\end{algorithm}

\subsection{Parent and child}
\label{sec:Parent}
In this section we describe how to compute the Tet-ids of the parent $P(T)$ and
of the $2^d$ children $T_i$, $0\leq i<2^d$, of a given $d$-simplex $T\in\mathcal
T_d$.  Computing the anchor node coordinates of the parent is easy, since their
first $T.\ell-1$ bits correspond to
the coordinates of $T$'s anchor node and the rest of their bits is zero.
For computing the type of $P(T)$, we need the function
\begin{equation}
\label{eq:Pt}
 \begin{array}{rcl}
\textrm{Pt}\colon \set{0,\ldots,2^d-1}\times\set{0,\ldots,d!-1} &\longrightarrow &\set{0,\ldots,d!-1} ,\\[.5ex]
(\textrm{cube-id}(T),T.b)&\longmapsto & P.b,
\end{array}
\end{equation}
giving the type of $T$'s parent in dependence of its cube-id and type.
In Figure \ref{fig:parenttype} we list all values of this function for $d\in\set{2, 3}$.
\begin{figure}
\flushleft
\raisebox{5.6ex}{
\begin{minipage}{0.21\textwidth}
\center
\scalebox{0.9}{
\begin{tabular}{|rc|cc|}
 \hline
 \multicolumn{2}{|c|}{\mytabvspace $\mathrm{Pt}(c,b)$} &\multicolumn{2}{c|}{b}\\
 \multicolumn{2}{|c|}{2D}
 &\mytabvspace  0 & 1 \\[0.5ex]\hline
 \multirow{4}{*}{$c$}&\mytabvspace0& 0 & 1   \\[0.2ex]
 &1& 0 & 0   \\[0.2ex]
 &2& 1 & 1   \\[0.2ex]
 &3& 0 & 1   \\[0.2ex]\hline
\end{tabular} %
}
\end{minipage}
}
\begin{minipage}{0.41\textwidth}
\center
\scalebox{0.9}{
\begin{tabular}{|rc|cccccc|}
 \hline
 \multicolumn{2}{|c|}{\mytabvspace $\mathrm{Pt}(c,b)$} &\multicolumn{6}{c|}{b}\\
 \multicolumn{2}{|c|}{3D}&
 \mytabvspace
   0 & 1 & 2 & 3 & 4 & 5 \\[0.5ex]\hline
 \multirow{8}{*}{$c$}&\mytabvspace0& 0 & 1 & 2 & 3 & 4 & 5  \\[0.2ex]
 &1& 0 & 1 & 1 & 1 & 0 & 0  \\[0.2ex]
 &2& 2 & 2 & 2 & 3 & 3 & 3   \\[0.2ex]
 &3& 1 & 1 & 2 &  \cfbox{Magenta}{\color{blue}2}& 2 & 1   \\[0.2ex]
 &4& 5 & 5 & 4 & 4 & 4 & 5   \\[0.2ex]
 &5& 0 & 0 & 0 & 5 & \cfbox{Dandelion}{\color{PineGreen}5} & 5   \\[0.2ex]
 &6& 4 & 3 & 3 & 3 & 4 & 4  \\[0.2ex]
 &7& 0 & 1 & 2 & 3 & 4 & 5 \\ \hline
\end{tabular} %
}

\end{minipage}
\begin{minipage}{0.34\textwidth}
   \def\svgwidth{0.9\textwidth}
\begingroup%
  \makeatletter%
  \providecommand\color[2][]{%
    \errmessage{(Inkscape) Color is used for the text in Inkscape, but the package 'color.sty' is not loaded}%
    \renewcommand\color[2][]{}%
  }%
  \providecommand\transparent[1]{%
    \errmessage{(Inkscape) Transparency is used (non-zero) for the text in Inkscape, but the package 'transparent.sty' is not loaded}%
    \renewcommand\transparent[1]{}%
  }%
  \providecommand\rotatebox[2]{#2}%
  \ifx\svgwidth\undefined%
    \setlength{\unitlength}{499.24575195bp}%
    \ifx\svgscale\undefined%
      \relax%
    \else%
      \setlength{\unitlength}{\unitlength * \real{\svgscale}}%
    \fi%
  \else%
    \setlength{\unitlength}{\svgwidth}%
  \fi%
  \global\let\svgwidth\undefined%
  \global\let\svgscale\undefined%
  \makeatother%
  \begin{picture}(1,0.88477599)%
    \put(0,0){\includegraphics[width=\unitlength,page=1]{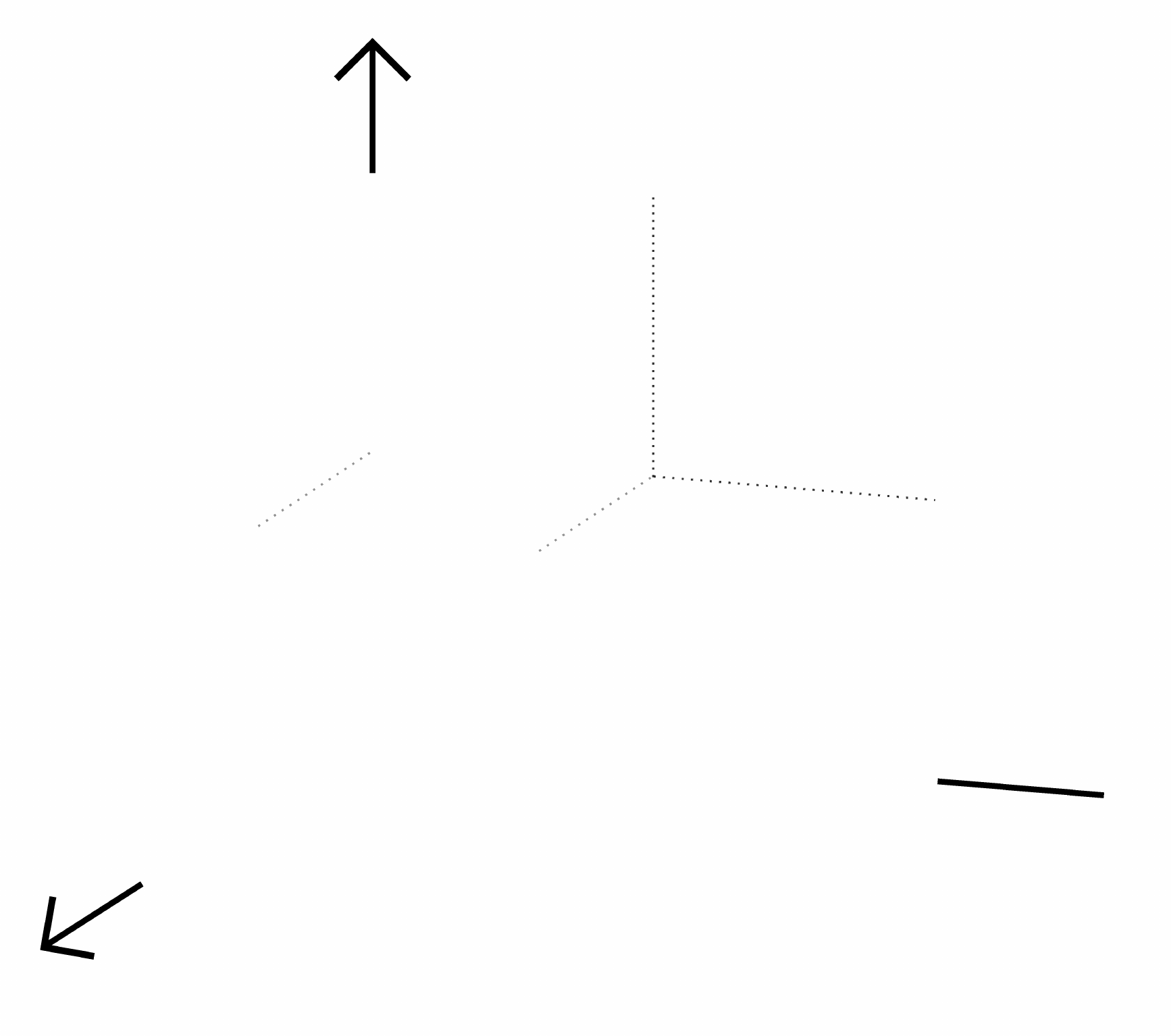}}%
    \put(0.38064754,0.80248984){\color[rgb]{0,0,0}\makebox(0,0)[lb]{\smash{$Z$}}}%
    \put(0.8682292,0.11017016){\color[rgb]{0,0,0}\makebox(0,0)[lb]{\smash{$Y$}}}%
    \put(0.08489371,0.02018408){\color[rgb]{0,0,0}\makebox(0,0)[lb]{\smash{$X$}}}%
    \put(0,0){\includegraphics[width=\unitlength,page=2]{fatheridexample_tex.pdf}}%
  \end{picture}%
\endgroup%
 \end{minipage}
   \caption[The type of the parent of a simplex]
    {The type of the parent of a $d$-simplex $T$ can be determined from $T$'s
    cube-id $c$ and type $b$.  Left: The values of $Pt$ from \eqref{eq:Pt} in the
    2D case.  Middle: The values of the function $Pt$ in the 3D case.  Right: Two
    examples showing the computation in 3D. (1) The small tetrahedron in the top
    left corner (orange) has cube-id $5$ and type $4$, and its parent (green) can
    be seen to have type $5$.  (2) The small tetrahedron at the bottom right (red)
    has cube-id $3$ and type $3$, and its parent (blue) has type $2$.}
   \figlabel{fig:parenttype}
\end{figure}

\begin{algorithm}
 \caption{\texttt{c-id}(\aTet{} $T$,\Int $\ell$)%
}
 \label{alg:cube-id}
 \DontPrintSemicolon
 \algoresult{The cube-id of $T$.}\;
 $i\gets0$, $h\gets 2^{\mathcal{L}-\ell}$\;
 $i\aorgl (T.x\bitwand h) \abst{?} 1 \colon 0$\;
 $i\aorgl (T.y\bitwand h) \abst{?} 2 \colon 0$\;
 $i\aorgl (T.z\bitwand h) \abst{?} 4 \colon 0$\Comment{Remove this line for 2D}
 \Return $i$
\end{algorithm}

\begin{algorithm}
 \caption{\texttt{Parent}(\aTet{} $T$)}
 \label{alg:parent}
 \DontPrintSemicolon
 \algoresult{The parent $P$ of $T$.}\;
 $h\gets 2^{\mathcal{L}-T.\ell}$\;
 $P.\ell\gets T.\ell-1$\;
 $P.x\gets T.x \bitwand \neg h$\;
 $P.y\gets T.y \bitwand \neg h$\;
 $P.z\gets T.z \bitwand \neg h$\Comment{Remove this line for 2D}
 $P.b\gets \mathrm{Pt}(\texttt{c-id}(T,T.\ell),T.b)$\Comment{See \eqref{eq:Pt} and Figure \ref{fig:parenttype} for Pt}
 \Return $P$\;
\end{algorithm}
The algorithm \texttt{Parent} to compute the parent of $T$ now puts these two ideas together, computing the coor\-di\-nates and type of $P(T)$.
Algorithm \ref{alg:parent} shows an implementation. It uses Algorithm \ref{alg:cube-id} to compute the cube-id of a $d$-simplex.

For the computation of one child $T_i$ of $T$ for a given $i\in\set{0,\dots,2^d-1}$
we look at Bey's definition of the subsimplices in \eqref{eq:childnumbers} and see that in order to compute the anchor node
of the child we need to know some of the node coordinates $\vec x_0,\dots,\vec x_d$ of the parent simplex $T$.
These can be obtained via Algorithm \ref{alg:coordsofT}.
However, it is more efficient to compute only those coordinates
of $T$ that are actually necessary.
To compute the Tet-id of $T_i$ we also need to know its type.
The type of $T_i$ depends only on the type of $T$, and in the algorithm
we use the function \texttt{Ct} (children type) to compute this type.
\texttt{Ct} is effectively an evaluation of Table \ref{table:typesofchildren}.
Algorithm \ref{alg:children} shows now how to compute the coordinates of the $i$th child of $T$ in Bey's order.

When we would like to compute the $i$th child of a $d$-simplex $T$ of type $b$ with respect to the tetrahedral Morton order
(thus the child $T_k$ of $T$ with $I_\mathrm{loc}(T_k)=i$) we just call Algorithm \ref{alg:children}
with $\sigma_b^{-1}(i)$ as input.
The permutations $\sigma_b^{-1}$ are available from Table \ref{table:BeytoIndex};
see \eqref{eq:Ilocsigma} and Algorithm \ref{alg:TMchild}.

\begin{algorithm}
 \caption{\texttt{Child}(\aTet{} $T$,\Int $i$)}
 \label{alg:children}
 \DontPrintSemicolon
 \algoresult{The $i$-th child (in Bey's order) $T_i$ of $T$.}\;
 $X\gets$ \texttt{Coordinates}($T$)\;
 \lIf{$i=0$}{$j\gets0$}

 \lElseIf{$i\in\set{1,4,5}$}{$j\gets1$}

 \lElseIf{$i\in\set{2,6,7}$}{$j\gets2$}

 \lElseIf(\IfComment{If $i={3}$\textbf{ then }$j\gets 1$ for 2D}){$i={3}$}{$j\gets3$}
 $T_i.\vec{x}\gets\frac{1}{2}(X[0]+X[j])$\;
 $T_i.b\gets$ \texttt{Ct}$(T.b,i)$\Comment{See Table \ref{table:typesofchildren}}
\end{algorithm}

\begin{algorithm}
  \caption{\texttt{TM-Child}(\aTet{} $T$,\Int $i$)}
  \label{alg:TMchild}
  \DontPrintSemicolon
  \algoresult{The $i$-th child of $T$ in TM-order.}\;
  \Return \texttt{Child} ($T$, $\sigma_{T.b}^{-1} (i)$)
  \Comment{See Table~\ref{table:BeytoIndex}}
\end{algorithm}

\subsection{Neighbor simplices}

Many applications---e.g., finite element me\-thods---need to gather information about the face-neighboring simplices of
a given simplex in a refinement.
In this section we describe a level-independent constant-runtime algorithm to compute the
Tet-id of the same level neighbor along a specific face $f$ of a given $d$-simplex $T$.
This algorithm is very lightweight since it only requires a few arithmetic
computations involving the Tet-id of $T$ and the number $f$.
In comparison to other approaches to finding neighbors in constant time
\cite{LeeDeSamet01,BrixMassjungVoss11},
our algorithm does not involve the computation of any of $T$'s ancestors.

The $d+1$ faces of a $d$-simplex $T=[\vec x_0,\dots,\vec x_d]$ are numbered  $f_0,\dots,f_d$ in such a way that face $f_i$ is the face
not containing the node $\vec x_i$.
To examine the situation where two $d$-simplices of the same level share a common face,
let $\mathcal{T}^\ell_d$ denote a uniform refinement of $T^0_d$ of a given level $0\leq \ell \leq \mathcal{L}$,
\begin{equation}
 \mathcal{T}_d^\ell\, {:=} \, \set{T \abst{|} T \textnormal{ is a descendant of } T^0_d
 \textnormal{ of level } \ell}\abst{\subset} \mathcal{T}_d.
\end{equation}
$\mathcal{T}_d$ can be seen as the disjoint union of all the $\mathcal T^\ell_d$:
\begin{equation}
 \mathcal T_d = \bigcup_{\ell=0}^{\mathcal L} \mathcal T^\ell_d.
\end{equation}
Given a $d$-simplex $T\abst{\in} \mathcal{T}^\ell_d$ and a face number
$i\abst{\in}\set{0,\dots,d}$, denote $T$'s neighbor in $\mathcal{T}^\ell_d$
across face $f=f_i$ by $\mathcal{N}_{f}(T)$, and denote the face number of the
neighbor simplex $\mathcal{N}_f(T)$ across which $T$ is its neighbor by
$\tilde{f}(T)$. Hence, the relation
\begin{equation}
\mathcal{N}_{\tilde f(T)}(\mathcal N_f(T))=T
\end{equation}
holds for each face $f$ of $T$.

Our aim is to compute the Tet-id of $\mathcal{N}_f(T)$ and $\tilde{f}(T)$
from the Tet-id of $T$.
Using the underlying cube structure this problem can be solved for each occuring type of $d$-simplex separately,
and the solution scheme
is independent of the coordinates of $T$ and of $\ell$.
In Figure \ref{fig:neighbor3} the situation for a tetrahedron of type $5$ is illustrated, and in Tables \ref{tab:faceneighbor2d}
and \ref{tab:faceneighbor}
we present the general solution for each type.

\begin{figure}
   \def\svgwidth{0.40\textwidth}
   \center
   \includegraphics{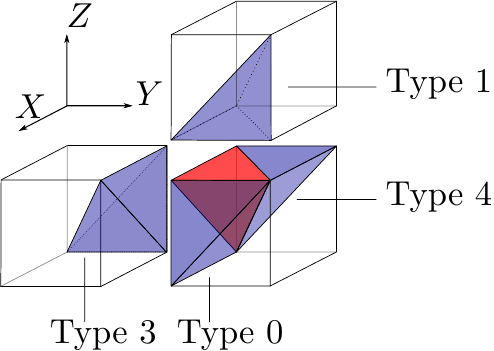}
    \caption[The face-neighbors of a tetrahedron]
    {A tetrahedron $T$ of type $5$ (in the middle, red) and its four
    face-neighbors (blue) of types $1,0,4$, and $3$, drawn with their associated
    cubes (exploded view). We see that the type of
    $T$'s neighbors depends only on its type, while their node coordinates
    relative to $T$'s depend additionally on $T$'s level.
}\figlabel{fig:neighbor3}
\end{figure}
Using these tables, a constant-time computation of the Tet-id of $\mathcal{N}_f(T)$ and of $\tilde{f}(T)$
from the Tet-id of $T$ is possible, and the 3D case
 is carried out in Algorithm~\ref{alg:face-neighbor}.
Note that this algorithm uses arithmetic expressions in $T.b$ to avoid the sixfold distinction of cases.

\begin{remark}
To find existing neighbors in a nonuniform refinement we use
Algorithm \ref{alg:face-neighbor} in combination with \texttt{Parent}
and \texttt{TM-Child} and comparison functions.

Of course it is possible that $\mathcal{N}_f(T)$ does not belong to $\mathcal T_d^\ell$ any longer.
If this is the case, then $f$ was part of the boundary of the root simplex $T^0_d$.
We describe in the next section how
we can decide in constant time whether $\mathcal{N}_f(T)$ is in $T^\ell_d$ or not.
\end{remark}

\begin{table}
\begin{center}
\begin{tabular}{|c|c|ccc|}
\hline
 \mytabvspace
 \mytabvspace
 2D & $f$              & 0 & 1 & 2   \\[0.2ex]\hhline{|=|=|===|}
\multirow{5}{*}{$T.b=0$} & \mytabvspace $\mathcal{N}_f.b$ & 1 & 1 & 1  \\[0.2ex]
& $\mathcal{N}_f.x$ & $T.x+h$ & $T.x$ & $T.x$\\[0.2ex]
& $\mathcal{N}_f.y$ & $T.y$ & $T.y$ & $T.y-h$\\[0.2ex]
& $\tilde{f}$      & 2 & 1 & 0       \\[0.2ex]\hline

\multirow{5}{*}{$T.b=1$} & \mytabvspace $\mathcal{N}_f.b$ & 0 & 0 & 0   \\[0.2ex]
 &$\mathcal{N}_f.x$ & $T.x$ & $T.x$ & $T.x-h$ \\[0.2ex]
 &$\mathcal{N}_f.y$ & $T.y+h$ & $T.y$ & $T.y$\\[0.2ex]
 &$\tilde{f}$      & 2 & 1 & 0       \\[0.2ex]\hline
\end{tabular}
\end{center}
\caption[Face neighbors for triangles]
 {Face neighbors in 2D. For each possible type $b\in\set{0,1}$ of a
  triangle $T$ and each of its faces $f=f_i$, $i\in\set{0,1,2}$, the type, anchor
  node coordinates, and corresponding face number $\tilde f$ of $T$'s neighbor
  across $f$ are shown. In the 2D case we can directly compute $\mathcal N.b =
  1-T.b$ and $\tilde f = 2-f$.  Here, $h=2^{\mathcal{L}-\ell}$ refers to the
  length of one square of level $\ell$.}
\label{tab:faceneighbor2d}
\end{table}

\begin{table}
\resizebox{0.5\textwidth}{!}{
\begin{tabular}{|c|c|cccl|}
\hline
 \mytabvspace
 \mytabvspace
 3D & $f$              & 0 & 1 & 2 & 3  \\[0.2ex]\hhline{|=|=|====|}
\multirow{5}{*}{$T.b=0$} & \mytabvspace $\mathcal{N}_f.b$ & 4 & 5 & 1 & 2  \\[0.2ex]
& $\mathcal{N}_f.x$ & $T.x+h$ & $T.x$ & $T.x$ & $T.x$\\[0.2ex]
& $\mathcal{N}_f.y$ & $T.y$ & $T.y$ & $T.y$ & $T.y-h$\\[0.2ex]
& $\mathcal{N}_f.z$ & $T.z$ & $T.z$ & $T.z$ & $T.z$\\[0.2ex]
& $\tilde{f}$      & 3 & 1 & 2 & 0       \\[0.2ex]\hline

\multirow{5}{*}{$T.b=1$} & \mytabvspace $\mathcal{N}_f.b$ & 3 & 2 & 0 & 5  \\[0.2ex]
 &$\mathcal{N}_f.x$ & $T.x+h$ & $T.x$ & $T.x$ & $T.x$\\[0.2ex]
 &$\mathcal{N}_f.y$ & $T.y$ & $T.y$ & $T.y$ & $T.y$\\[0.2ex]
 &$\mathcal{N}_f.z$ & $T.z$ & $T.z$ & $T.z$ & $T.z-h$\\[0.2ex]
 &$\tilde{f}$      & 3 & 1 & 2 & 0       \\[0.2ex]\hline

\multirow{5}{*}{$T.b=2$} & \mytabvspace$\mathcal{N}_f.b$ & 0 & 1 & 3 & 4  \\[0.2ex]
 &$\mathcal{N}_f.x$ & $T.x$ & $T.x$ & $T.x$ & $T.x$\\[0.2ex]
 &$\mathcal{N}_f.y$ & $T.y+h$ & $T.y$ & $T.y$ & $T.y$\\[0.2ex]
 &$\mathcal{N}_f.z$ & $T.z$ & $T.z$ & $T.z$ & $T.z-h$\\[0.2ex]
 &$\tilde{f}$      & 3 & 1 & 2 & 0       \\[0.2ex]\hline
 \end{tabular}
 }
 \hspace{-2ex}
\resizebox{0.5\textwidth}{!}{
 \begin{tabular}{|c|c|cccl|}
\hline
 \mytabvspace
 \mytabvspace
 3D & $f$              & 0 & 1 & 2 & 3  \\[0.2ex]\hhline{|=|=|====|}
\multirow{5}{*}{$T.b=3$} & \mytabvspace$\mathcal{N}_f.b$ & 5 & 4 & 2 & 1  \\[0.2ex]
 &$\mathcal{N}_f.x$ & $T.x$ & $T.x$ & $T.x$ & $T.x-h$\\[0.2ex]
 &$\mathcal{N}_f.y$ & $T.y+h$ & $T.y$ & $T.y$ & $T.y$\\[0.2ex]
 &$\mathcal{N}_f.z$ & $T.z$ & $T.z$ & $T.z$ & $T.z$\\[0.2ex]
 &$\tilde{f}$      & 3 & 1 & 2 & 0       \\[0.2ex]\hline

\multirow{5}{*}{$T.b=4$} & \mytabvspace$\mathcal{N}_f.b$ & 2 & 3 & 5 & 0  \\[0.2ex]
 &$\mathcal{N}_f.x$ & $T.x$ & $T.x$ & $T.x$ & $T.x-h$\\[0.2ex]
 &$\mathcal{N}_f.y$ & $T.y$ & $T.y$ & $T.y$ & $T.y$\\[0.2ex]
 &$\mathcal{N}_f.z$ & $T.z+h$ & $T.z$ & $T.z$ & $T.z$\\[0.2ex]
 &$\tilde{f}$      & 3 & 1 & 2 & 0       \\[0.2ex]\hline

\multirow{5}{*}{$T.b=5$} & \mytabvspace$\mathcal{N}_f.b$ & 1 & 0 & 4 & 3  \\[0.2ex]
 &$\mathcal{N}_f.x$ & $T.x$ & $T.x$ & $T.x$ & $T.x$\\[0.2ex]
 &$\mathcal{N}_f.y$ & $T.y$ & $T.y$ & $T.y$ & $T.y-h$\\[0.2ex]
 &$\mathcal{N}_f.z$ & $T.z+h$ & $T.z$ & $T.z$ & $T.z$\\[0.2ex]
 &$\tilde{f}$      & 3 & 1 & 2 & 0       \\[0.2ex]\hline
\end{tabular}
}
\caption[Face neighbors for tetrahedra]
  {Face neighbors in 3D. For each possible type $b\in\set{0,1,2,3,4,5}$ of a
  tetrahedron $T$ and each of its faces $f=f_i$, $i \in\set{0,1,2,3}$ the type
  $\mathcal{N}_f(T).b$ of $T$'s neighbor across $f$, its coordinates of the
  anchor node $\mathcal{N}_f(T).x$, $\mathcal{N}_f(T).y$, $\mathcal{N}_f(T).z$
  and the corresponding face number $\tilde{f}(T)$, across which $T$ is
  $\mathcal{N}_f(T)$'s neighbor, are shown.} 
\label{tab:faceneighbor}
\end{table}

\begin{algorithm}
 \caption{\texttt{Face-neighbor}(\aTet{} $T$,\Int $f$)}
 \label{alg:face-neighbor}
 \DontPrintSemicolon
 \algoresult{The face-neighbor $\mathcal N_f(T)$ of $T$ across $f$
and the corresponding face $\tilde f$.}\;
 $b\gets T.b$,
 $x_0\gets T.x_0$, $x_1\gets T.x_1$, $x_2\gets T.x_2$\;
 \eIf{$f=1$ or $f=2$}{
    $\tilde{f}\gets f$
    \eIf{$(b\algomod2=0 \textnormal{\algoand} f=2)\textnormal\algor(b\algomod2\neq0 \textnormal\algoand f=1)$}{
      $b\gets b+1$
    }
    {
      $b\gets b-1$
    }
 }
 {%
  $\tilde{f}\gets 3-f$

  $h\gets 2^{\mathcal{L}-T.\ell}$

  \eIf(\IfComment{$f=0$}){$f=0$}{%
    $i\gets b\textnormal{\textbf{ div }}2$\;
    $x_i\gets T.x_i+h$\;
    $b\gets b + (b\algomod2=1\abst{?}2:4)$\;
  }%
  (\IfComment{$f=3$}){%
    $i\gets (b+3)\algomod6 \textnormal{\textbf{ div }}2$\;
    $x_i \gets T.x_i - h$\;
    $b\gets b + (b\algomod2=0\abst{?}2:4)$\;
  }%
  }%
$N.\vec{x}\gets (x_0,x_1,x_2)$\;
$N.\ell\gets T.\ell$\;
$N.b\gets b\algomod 6$\;
\Return $(N, \tilde{f})$\;
\end{algorithm}

For completeness, we summarize the geometry and numbers of
$d$-simplices tou\-ching each other via a corner ($d=2$ or $d=3$) or edge (only
$d=3$).
In this chapter we do not list these neighboring tetrahedra in detail.

For $d=3$ each corner in the mesh $\mathcal T^\ell_3$ has 24 adjacent tetrahedra;
thus each tetrahedron has at each corner 23 other tetrahedra that share
this particular corner.
For $d=2$ the situation is similar, with six triangles meeting at each corner.
To examine the number of adjacent tetrahedra to an edge
we distinguish three types of edges in $\mathcal T^\ell_d$:
\begin{enumerate}
 \item edges that are also edges in the underlying hexahedral mesh;
 \item edges that are the diagonal of a side of a cube in the hexahedral mesh;
 \item edges that correspond to the inner diagonal of a cube in the hexahedral mesh.
\end{enumerate}
Edges of the first and third kind have six adjacent tetrahedra each,
and edges of the second kind do have four adjacent tetrahedra each.

\subsection{The exterior of the root simplex}
\label{sec:outside}
When computing neighboring $d$-simplices it is possible that the neighbor in question does not belong to the root simplex
$T^0_d$ but lies outside of it.
If we look at face-neighbors of a $d$-sim\-plex $T$, the fact that the considered neighbor lies outside
means that the respective face was on the boundary of $T^0_d$.
In order to check whether a computed $d$-simplex is outside the base simplex, we investigate
a more general problem:
Given anchor node coordinates $(x_0,y_0)^T\in\IZ^2$, respectively $(x_0,y_0,z_0)^T\in\IZ^3$, of type $b$ a level $\ell$,
decide whether the corresponding $d$-simplex $N$ lies inside or outside
of the root tetrahedron $T^0_d$: $N\in\mathcal{T}^\ell_d$ or $N\notin\mathcal{T}^\ell_d$.
At the end of this section we furthermore generalize this to the problem of deciding for any two $d$-simplices $N$ and $T$
whether or not $N$ lies outside of $T$.
We solve this problem in constant time and independent of the levels of $N$ and $T$.

We examine the 3D case.
Looking at $T^0_3$ we observe that two of its boundary faces correspond to faces of the root cube,
namely, the intersections of $T^0_3$ with the $y=0$ and the $x=2^\mathcal{L}$ planes.
The other two boundary faces of $T^0_3$ are the intersections with the $x=z$ and the $y=z$ planes.
Thus, the boundary of $T^0_3$ can be described as the intersection of $T^0_3$ with those planes.
We refer to the latter two planes as $E_1$ and $E_2$.

Let $N$ be a tetrahedron with anchor node $(x_0,y_0,z_0)^T\in\IZ^3$ of type $b$ and level $\ell$
and denote with $(x_i,y_i,z_i)^T$ the coordinates of node $i$ of $N$.
Since $(x_i,y_i,z_i)^T\geq(x_0,y_0,z_0)^T$
(componentwise), we directly conclude that if $x_0\geq 2^\mathcal L$ or $y_0<0$ then $N\notin\mathcal T_3$.
Because the outer normal vectors of $T^0_3$ on the two faces intersecting with $E_1$ and $E_2$ are
\begin{equation}
\vec{n}_1= \frac{1}{\sqrt 2}\begin{pmatrix}
  -1\\0\\1
 \end{pmatrix} \quad\textrm{and}\quad
\vec{n}_2= \frac{1}{\sqrt 2}\begin{pmatrix}
  0\\1\\-1
 \end{pmatrix},
\end{equation}
we also conclude that $N\notin\mathcal T_3$ if $z_0-x_0>0$ or $y_0-z_0>0$.
Now we have already covered all the cases except those where the anchor node of $N$ lies directly in $E_1$ or $E_2$.
In these cases we cannot solve the problem by looking at the coordinates of the anchor node alone, since there
exist tetrahedra $T'\in\mathcal{T}_3$ with anchor nodes lying in one of these planes
(see Figure \ref{fig:neighboroutside2d} for an illustration of the analogous case in 2D).
This depends on the type of the tetrahedron in question.
We observe that a tetrahedron $T'\in\mathcal{T}_3$ with anchor node lying in $E_1$ can have the types $0$ ,$1$, or $2$, and a tetrahedron
with anchor node lying in $E_2$ can have the types $0$, $4$, or $5$.
We conclude that to check whether $N$ is outside of the root tetrahedron we have to check
if any one of six conditions is fulfilled.
In fact these conditions fit into the general form below with $x_i = x$, $x_j=y$, $x_k=z$, and $T$ as the root tetrahedron;
thus $T.x = T.y = T.z = 0$ and $T.b=0$.

These generalized conditions solve the problem to check for any two given tetrahedra $N$ and $T$, whether
$N$ lies outside of $T$ or not.
\begin{proposition}
\label{prop:neighboroutside}
 Given two $d$-simplices $N, T$ with $N.\ell\geq T.\ell$, then $N$ is outside
of the simplex $T$---which is equivalent to saying that $N$ is no descendant of
$T$---if and only if at least one of the following conditions is fulfilled.

 For 2D,
\begin{subequations}
\begin{align}
&N.x_i - T.x_i \geq 2^{T.\ell},\hphantom{xxxxxxxxxxxxxxxxxxxxxxxxxxxxxxxxxxxxxxxxxx}\\
&N.x_j-T.x_j<0,\\
&(N.x_j-T.x_j)-(N.x_i-T.x_i)>0,\\
&N.x_i-T.x_i=N.x_j-T.x_j\textrm{ and }N.b=\left\lbrace\begin{array}{cc}
                                             1 &\textrm{if\, } T.b = 0, \\
                                             0 &\textrm{if\, } T.b = 1.
                                            \end{array}\right.
\end{align}
\end{subequations}
For 3D,
\begin{subequations}
\begin{align}
 &N.x_i -T.x_i\geq 2^{\mathcal L - T.\ell},\hphantom{xxxxxxxxxxxxxxxxxxxxxxxxxxxxxxxxxxxxxxxxxx}\\
 &N.x_j-T.x_j<0,\\
 &(N.x_k-T.x_k)-(N.x_i-T.x_i)>0,\\
 &(N.x_j-T.x_j)-(N.x_k-T.x_k)>0,\\[1ex]
 &\begin{aligned}
 &N.x_j-T.x_j=N.x_k-T.x_k \\ &\mathrm{ and }\quad N.b\in\left\lbrace\begin{array}{cc}
                                             \set{T.b+1,T.b+2,T.b+3}, & \textrm{if\, } T.b \textrm{ is even,} \\
                                             \set{T.b-1,T.b-2,T.b-3}, & \textrm{if\, }T.b \textrm{ is odd, }
                                            \end{array}\right.\\
 \end{aligned}\\
 &%
 \begin{aligned}
 &N.x_k-T.x_k=N.x_i-T.x_i\\ &\mathrm{ and }\quad N.b\in\left\lbrace\begin{array}{cc}
                                             \set{T.b-1,T.b-2,T.b-3}, &\textrm{if\, } T.b\textrm{ is even,} \\
                                             \set{T.b+1,T.b+2,T.b+3}, &\textrm{if\, } T.b\textrm{ is odd. }
                                            \end{array}\right.
 \end{aligned} \\
&%
\begin{aligned}
 &N.x_j - T.x_j = N.x_k - T.x_k \quad\mathrm { and }\quad N.x_i - T.x_i = N.x_k - T.x_k\quad
\mathrm{ and }\quad N.b \neq T.b
\end{aligned}
\end{align}
\end{subequations}
 The coordinates $x_i$, $x_j$, and $x_k$ are chosen in dependence of the type of $T$ according to Table \ref{table:xixjxk}.
\begin{table}
\centering
\raisebox{1.3ex}{
\begin{tabular}{|c|cc|}
 \hline
  & \multicolumn{2}{c|}{$T.b$} \\
  $2$D
        & 0 & 1 \\ \hline
  $x_i$ & $x$ & $y$  \\
  $x_j$ & $y$ & $x$  \\ \hline
 \end{tabular}}
 \begin{tabular}{|c|cccccc|}
 \hline
  & \multicolumn{6}{c|}{$T.b$} \\
   $3$D   & 0 & 1 & 2 & 3 & 4 & 5 \\ \hline
  $x_i$ & $x$ & $x$ & $y$ & $y$ & $z$ & $z$ \\
  $x_j$ & $y$ & $z$ & $z$ & $x$ & $x$ & $y$ \\
  $x_k$ & $z$ & $y$ & $x$ & $z$ & $y$ & $x$ \\
  \hline
 \end{tabular}
\caption[The coordinates $x_i, x_j$, and $x_k$ for the computation whether a 
  simplex is ancestor of another.]
  {Following the general scheme described in this section to compute whether a
  given $d$-simplex $N$ lies outside of another given $d$-simplex $T$, we give
  the coordinates $x_i$, $x_j$, and $x_k$ in dependence of the type of $T$.}
\label{table:xixjxk}
\end{table}
\begin{proof}
By shifting $N$ by $(-T.\vec x)$ we reduce the problem to checking whether the shifted $d$-simplex lies outside of a
$d$-simplex with anchor node $\vec 0$, level $T.\ell$ and type $T.b$.
For $d=3$ the proof is analogous to the above argument, where we considered the case $b=0$ and $\ell = 0$.
In two dimensions the situation is even simpler, since there exists only one face of the root
triangle that is not a coordinate axis (see Figure \ref{fig:neighboroutside2d}).
\end{proof}
\end{proposition}

\begin{figure}
\center
   \includegraphics[width=46ex]{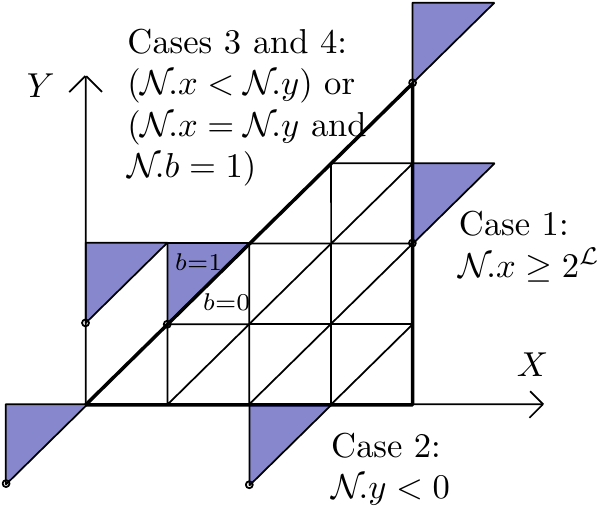}
   \caption[Determining when a triangle is outside of the root triangle]
  {A uniform level 2 refinement of the triangle $T^0$ in 2D and
  triangles lying outside of it with their anchor nodes marked. When deciding
  whether a triangle with given anchor node coordinates
  lies outside of $T^0$ there are four cases to consider, one for each face of
  $T^0$. For the two faces lying parallel to the $X$-axis, respectively,
  $Y$-axis, it suffices to check whether the $x$-coordinate is greater than or
  equal to $2^\mathcal{L}$, or whether the $y$-coordinate is smaller than $0$.
  Similarly one can conclude that the triangle lies outside of $T^0$ if its
  $x$-coordinate is smaller than its $y$-coordinate. If both coordinates agree
  (and none of the previous cases applies) then
  the given triangle is outside $T^0$ if and only if its type is $1$.}
   \figlabel{fig:neighboroutside2d}
\end{figure}

\subsection[A consecutive index for uniform refinements]{A consecutive index for uniform refinements}
\label{sec:consecindex}

In contrast to the Morton index for cubes, the TM-index for $d$-simplices does not
produce a consecutive range of numbers.
Therefore, two simplices $T$ and $T'$ of level $\ell$
that are direct successors/predecessors with respect to the tetrahedral Morton order do not
necessarily fulfill $m(T)=m(T')\pm 2^{d(\mathcal L-\ell)}$ or $m(T)=m(T')\pm 1$.
For $d=3$ this follows directly from the fact that each $b^j$ occupies three bits,
but there are only six values that each $b^j$ can assume, since there are only six different types.
In 2D this follows from the fact that not every combination of anchor
node coordinates and type can occur for triangles in $\mathcal T_2$,
the triangle with anchor node $(0,0)$ and type $1$ being one example.
This also means that the largest occurring TM-index is bigger than $2^{d\mathcal{L}}-1$.

Constructing a consecutive index 
as in Definition~\ref{def:consindex}
that respects the order given by the TM-index
is possible, as we show in this section.
Since in practice it is more convenient to work with this consecutive index instead
of the TM-index,
our aim is to construct for each uniform refinement level $\ell$ a consecutive index
$
\mathcal{I}_\ell(T)\in\set{0,\dots,2^{d\ell}-1},
$
such that
\begin{equation}
 \forall\ T,S\in\mathcal{T}^\ell_d\colon\quad \mathcal{I}_\ell(T)<\mathcal{I}_\ell(S) \abst\Leftrightarrow m(T)<m(S).
\end{equation}
This index can also be understood as a bijection
\begin{equation}
  \mathcal I_\ell\colon \set{m\in\IN_0 \abst | m = m(T) \textnormal{ for a } T\in\mathcal T^\ell_d}
  \xrightarrow{\cong}
  \set{0,\ldots,2^{d\ell}-1},
\end{equation}
mapping the TM-indices of level $\ell$ $d$-simplices to a consecutive range of numbers.
See also Definition~\ref{def:consindex}.
It is obvious that $\mathcal{I}_\ell(T^0_d)=0$.
The index $\mathcal{I}_\ell(T)$ can be easily
computed as the $\ell$-digit $2^d$-ary number consisting of the local indices as digits, thus
\begin{equation}
\label{eq:gloindex}
 \mathcal{I}_\ell(T) = (I_\mathrm{loc}(T^1),\dots,I_\mathrm{loc}(T^\ell))_{2^d}.
\end{equation}

Algorithm \ref{alg:ComputeIndexfromtet} shows an implementation of this computation.
It can be done directly from the Tet-id of $T$, and thus it is not
necessary to compute the TM-index of $T$ first.

\begin{algorithm}
 \caption{$\mathcal I$(\aTet{}$T$)}
 \label{alg:ComputeIndexfromtet}
 \DontPrintSemicolon
 \algoresult{The consecutive index $\mathcal I_{T.\ell} (T)$.}\;
 $I \gets 0$,
 $b \gets T.b$\;
 \For{$T.\ell \geq i \geq 1$}{
    $c \gets \texttt{c-id}(T,i)$\;
    $I\gets I + 8^i I_\mathrm{loc}^{b}(c)$
    \Comment{See Table \ref{table:candbtoiloc}; multiply with $4^i$ for 2D}
    $ b \gets \textrm{Pt}(c, b)$\;
 }
 \Return $I$\;
\end{algorithm}

The inverse operation of computing $T$ from $\mathcal{I}_\ell(T)$ and a given level $\ell$ can be carried out in a similar fashion;
see Algorithm \ref{alg:Init_id}.
For each $0\leq i\leq\ell$ we look up the type $b$ and the cube-id of $T^i$
from $I_\mathrm{loc}(T^{i})$ and the type of $\texttt{Parent}(T^i)=T^{i-1}$ (starting with $\type(T^0)=0$) via Tables \ref{table:IlocPbtocubeid} and \ref{table:IlocPbtotype}.
From the cube-ids we can build up the anchor node coordinates of $T$. The last computed
type is the type of $T$.
The runtime of this algorithm is $\mathcal O(T.\ell)$.

\begin{algorithm}
 \caption{\texttt{T}(consecutive index $\mathcal I$,\Int $\ell$)}
 \label{alg:Init_id}
 \DontPrintSemicolon
 \algoresult{The simplex $T$ with $\mathcal I_\ell (T)=\mathcal I$.}\;
 $T.x,T.y,T.z\gets 0$,
 $b \gets 0$\;
 \For{$1\leq i \leq \ell $}{
    Get $I_\mathrm{loc}(T^i)$ from $\mathcal I$\Comment{See \eqref{eq:gloindex}}
    $c\gets\texttt{c-id}(T^i),\, b\gets T^i.b$\Comment{See Tables \ref{table:IlocPbtocubeid} and \ref{table:IlocPbtotype}}
    \lIf{$c\bitwand 1$}{$T.x\gets T.x + 2^{\mathcal L -i}$}
    \lIf{$c\bitwand 2$}{$T.y\gets T.y + 2^{\mathcal L -i}$}
    \lIf(\IfComment{Remove this line for 2D}){$c\bitwand 4$}{$T.z\gets T.z + 2^{\mathcal L -i}$}
 }
 $T.b\gets b$\;
 \Return $T$\;
\end{algorithm}

Similar to Algorithm \ref{alg:ComputeIndexfromtet} is Algorithm
\texttt{Is\_valid}, which decides whether a given index $m\in[0,2^{6\mathcal{L}})
\cap \IZ$ is in fact a TM-index for a tetrahedron.
Thus, in the spirit of Section \ref{sec:Prop} we can decide whether a given
6D cube is in the image of map (\ref{eq:MapTQ}) that embeds
$\mathcal{T}_3$ into the set of 6D subcubes of $[0,2^{\mathcal{L}}]^6$.
The runtime of \texttt{Is\_valid} is $\mathcal{O}(\mathcal{L})$.

\begin{algorithm}
 \caption{\texttt{Is\_valid}($m\in [0,2^{6\mathcal{L}}) \cap \IZ, \ell$)}
 \label{alg:IsValid}
 \DontPrintSemicolon
 \algoresult{True if there exists a simplex $T$ with $m(T)=m$,
false otherwise.}\;
 $k \gets 6(\mathcal{L}-i)$\;
 \For{$\ell \geq i \geq 1$}{
    $b \gets (m_{k},m_{k + 1},m_{k + 2})_8$\;
    $c \gets (m_{k + 3},m_{k+4},m_{k+5})_8$\;
    $k \gets 6(\mathcal{L}-i + 1)$\;
    \If (\IfComment {Take $(0,0,0)_8$ if $i = 1$})
    {$(m_{k},m_{k + 1},m_{k + 2})_8 \neq \textrm{Pt}(c, b)$}
    {\Return \texttt{False}}}
 \Return \texttt{True}\;
\end{algorithm}

\begin{table}
\begin{center}
\raisebox{5.9ex}{
\begin{tabular}{|rc|cccl|}
\hline
\multicolumn{2}{|c|}{{$I^b_\mathrm{loc}(c)$}}&\multicolumn{4}{c|}{cube-id c}\\
 \multicolumn{2}{|c|}{2D} &\mytabvspace
 $0$ & $1$ & $2$ & $3$ \\[0.5ex]\hline
 \multirow{2}{*}{b}&\mytabvspace  0&  0 & 1 & 1 & 3 \\[0.2ex]
  &1&  0 & 2 & 2 & 3\\ \hline
\end{tabular}
}
\begin{tabular}{|rc|cccccccl|}
\hline
\multicolumn{2}{|c|}{{$I^b_\mathrm{loc}(c)$}}&\multicolumn{8}{c|}{cube-id c}\\
 \multicolumn{2}{|c|}{3D} &\mytabvspace
 $0$ & $1$ & $2$ & $3$ & $4$ & $5$ & $6$ & $7$\\[0.5ex]\hline
 \multirow{6}{*}{b}&\mytabvspace
  0&0& 1& 1& 4& 1& 4& 4& 7\\[0.2ex]
 &1&0& 1& 2& 5& 2& 5& 4& 7\\[0.2ex]
 &2&0& 2& 3& 4& 1& 6& 5& 7\\[0.2ex]
 &3&0& 3& 1& 5& 2& 4& 6& 7\\[0.2ex]
 &4&0& 2& 2& 6& 3& 5& 5& 7\\[0.2ex]
 &5&0& 3& 3& 6& 3& 6& 6& 7\\  \hline
\end{tabular}
\caption[The local index of a tetrahedron]
{The local index of a tetrahedron $T\abst{\in}\mathcal{T}$ in dependence of its
cube-id $c$ and type $b$.} 
\label{table:candbtoiloc}
\end{center}
\end{table}

\begin{table}
\begin{center}
\raisebox{5.9ex}{
\begin{tabular}{|rc|cccl|}
\hline
\multicolumn{2}{|c|}{\mytabvspace {$\cid(T)$}}&\multicolumn{4}{c|}{$I_\mathrm{loc}(T)$}\\
\multicolumn{2}{|c|}{2D} &
 \mytabvspace
 $0$ & $1$ & $2$ & $3$ \\[0.5ex]\hline
 \multirow{2}{*}{$P.b$}&\mytabvspace0& 0 & 1  & 1 &  3 \\[0.2ex]
 &1& 0 & 2 & 2 & 3 \\ \hline
\end{tabular}\hspace{0.5ex}
}%
\begin{tabular}{|rc|cccccccl|}
\hline
\multicolumn{2}{|c|}{{\mytabvspace $\cid(T)$}}&\multicolumn{8}{c|}{$I_\mathrm{loc}(T)$}\\
\multicolumn{2}{|c|}{3D}&
\mytabvspace $0$ & $1$ & $2$ & $3$ & $4$ & $5$ & $6$ & $7$\\[0.5ex]\hline
 \multirow{6}{*}{$P.b$}&\mytabvspace0& 0 & 1 & 1 & 1 & 5 & 5 & 5 & 7  \\[0.2ex]
 &1& 0 & 1 & 1 & 1 & 3 & 3 & 3 & 7     \\[0.2ex]
 &2& 0 & 2 & 2 & 2 & 3 & 3 & 3 & 7    \\[0.2ex]
 &3& 0 & 2 & 2 & 2 & 6 & 6 & 6 & 7    \\[0.2ex]
 &4& 0 & 4 & 4 & 4 & 6 & 6 & 6 & 7    \\[0.2ex]
 &5& 0 & 4 & 4 & 4 & 5 & 5 & 5 & 7    \\ \hline
\end{tabular}
\caption[The cube-id in dependence of the local index and the type of the parent]
 {For a tetrahedron $T\abst{\in}\mathcal{T}$ of local index $I_\mathrm{loc}$
  whose parent $P$ has type $P.b$ we give the cube-id of $T$.}
\label{table:IlocPbtocubeid}
\end{center}
\end{table}

\begin{table}
\begin{center}
\raisebox{5.9ex}{
\begin{tabular}{|rc|cccl|}
\hline
\multicolumn{2}{|c|}{\mytabvspace {$T.b$}}&\multicolumn{4}{c|}{$I_\mathrm{loc}(T)$}\\
\multicolumn{2}{|c|}{2D} &
 \mytabvspace
 $0$ & $1$ & $2$ & $3$ \\[0.5ex]\hline
 \multirow{2}{*}{$P.b$}&\mytabvspace0& 0 & 0 & 1 & 0 \\[0.2ex]
 &1& 1 & 0 & 1 & 1  \\ \hline
\end{tabular}\hspace{0.5ex}
}%
\begin{tabular}{|rc|cccccccl|}
\hline
\multicolumn{2}{|c|}{{\mytabvspace $T.b$}}&\multicolumn{8}{c|}{$I_\mathrm{loc}(T)$}\\
\multicolumn{2}{|c|}{3D}&
\mytabvspace $0$ & $1$ & $2$ & $3$ & $4$ & $5$ & $6$ & $7$\\[0.5ex]\hline
 \multirow{6}{*}{$P.b$}&\mytabvspace0& 0 & 0  & 4 & 5 & 0 & 1 & 2 & 0  \\[0.2ex]
 &1& 1 & 1 & 2 & 3 & 0 & 1 & 5 & 1     \\[0.2ex]
 &2& 2 & 0 & 1 & 2 & 2 & 3 & 4 & 2    \\[0.2ex]
 &3& 3 & 3 & 4 & 5 & 1 & 2 & 3 & 3    \\[0.2ex]
 &4& 4 & 2 & 3 & 4 & 0 & 4 & 5 & 4    \\[0.2ex]
 &5& 5 & 0 & 1 & 5 & 3 & 4 & 5 & 5    \\ \hline
\end{tabular}
\caption[The type in dependence of the local index and the type of the parent]
  {For a tetrahedron $T\abst{\in}\mathcal{T}$ of local index $I_\mathrm{loc}$
  whose parent $P$ has type $P.b$ we give the type of $T$.}
\label{table:IlocPbtotype}
\end{center}
\end{table}

The consecutive index simplifies the relation between the TM-index of a simplex
and its position in the SFC.
In the special case of a uniform mesh, the consecutive index and the position
are identical.

\subsection{Successor and predecessor}
Calculating the TM-index corresponding to a particular consecutive index is
occasionally needed in high-level algorithms.
This is relatively expensive,
since it involves a loop over all refinement levels, thus some 10 to 30 in extreme cases.
However often the task is to compute a whole range of $d$-simplices.
This occurs, for example, when creating an initial uniform refinement of a given mesh
(see Algorithm \texttt{New} in Section \ref{sec:new}).
That is, for a given consecutive index $\mathcal I$, a level $\ell$, and a
count $n$, find the $n$ level-$\ell$ simplices following the $d$-simplex corresponding to the consecutive index $\mathcal I$,
that is, the $d$-simplices corresponding to the $n$ consecutive indices
$\mathcal I,\mathcal I+1,\dots,\mathcal I+n-1$.
Ideally, this operation should run linearly in $n$, independent of $\ell$, but if we used Algorithm \ref{alg:Init_id} to create each of the $n+1$ simplices
we would have a runtime of $\mathcal{O}(n\mathcal{L})$.
In order to achieve the desired linear runtime
we introduce the operations \texttt{Successor} and \texttt {Predecessor} that run in average $\mathcal O(1)$ time.
These o\-pera\-tions compute from a given $d$-simplex $T$ of level $\ell$ with consecutive index $\mathcal I_\ell(T)$ the $d$-simplex $T'$ whose
consecutive index is $\mathcal I_\ell(T)+1$, respectively, $\mathcal I_\ell(T)-1$.
Thus, $T'$ is the next level $\ell$ simplex in the SFC after $T$ (resp.\ the previous one).
Algorithm \ref{alg:successor}, which we introduce to solve this problem does not require knowledge about the consecutive indices
$\mathcal I_\ell(T)$ and $\mathcal I_\ell(T)\pm 1$ and can be computed significantly faster than Algorithm \ref{alg:Init_id}; see Lemma \ref{lem:IncmortonRuntime}.

\begin{algorithm}
\DontPrintSemicolon
\SetKwFunction{proc}{proc}
\caption{\texttt{Successor}(\aTet{} $T$)}
\algoresult{The successor $T'$ of $T$.}\;
\label{alg:successor}
\Return \texttt{Successor\_recursion}$(T,T,T.\ell)$\\[2ex]

\setcounter{AlgoLine}{0}
\nonl \textbf{Function} \texttt{Successor\_recursion}(\Tet{} $T$,\Tet{} $T'$,\Int $\ell$)\;
$c\gets$ \texttt{c-id}$(T,\ell)$\;
From $c$ and $b$ look up $i:=I_\mathrm{loc}(T^\ell)$\Comment{See Table \ref{table:candbtoiloc}}
$i\gets (i+1)\abst \algomod 8$\label{line:changeforpred}\;
\eIf(\IfComment{Enter recursion (in rare cases)}){$i = 0$}{
  $T' \gets$ \texttt{Successor\_recursion} $(T,T',\ell-1)$\;
  $\hat b\gets T'.b$\Comment{$\hat b$ stores the type of $T'^{\ell-1}$}
}{
$\hat b\gets$ Pt$(c,b)$\;
}
From $\hat b$ and $I_\mathrm{loc} = i$ look up $(c',b')$\Comment{See Tables \ref{table:IlocPbtocubeid} and \ref{table:IlocPbtotype}}
Set the level $\ell$ entries of $T'.x, T'.y$ and $T'.z$ to $c'$\;
$T'.b\gets b'$\;
\Return $T'$\;
\end{algorithm}
To compute the predecessor of $T$ we only need to reverse the sign in Line \ref{line:changeforpred} in the
\texttt{Successor\_re\-cur\-sion} subroutine of Algorithm \ref{alg:successor}.

\begin{lemma}
\label{lem:IncmortonRuntime}
 Algorithm \ref{alg:successor} has constant average runtime (independent of $\mathcal{L}$).
 \begin{proof}
  Because each operation in the algorithm can be executed in constant time, the average runtime is $nc$, where $c$ is a
  constant (independent of $\mathcal{L}$) and $n-1$ is the number of average recursion steps.
  Since in consecutive calls to the algorithm the variable $i$ cycles through
$0$ to $2^d-1$ we conclude that the recursion is on average executed in every
$2^d$th step, allowing for a geometric series argument.
 \end{proof}
\end{lemma}

We see in Algorithm \ref{alg:successor} the usefulness of the consecutive index.
Because we are using this index instead of the TM-index, computing
the index of the successor/pre\-de\-cessor only requires adding/subtracting $1$ to the given index.
On the other hand, computing the TM-index of a successor/predecessor would involve more subtle computations.

\section{High-level AMR algorithms}

To develop the complete AMR functionality required by numerical applications,
we aim at a forest of quad-/octrees in the spirit of
\cite{BursteddeWilcoxGhattas11, IsaacBursteddeWilcoxEtAl15}.
Key high-level algorithms are (see also Section~\ref{sec:hlalgos}):
\begin{itemize}
 \item \texttt{New}.
   Given an input mesh of conforming simplices, each considered a root simplex,
   generate an initial partitioned uniform refinement.
 \item \texttt{Adapt}. Adapt (refine and coarsen) a mesh according to a given criterion.
 \item \texttt{Partition}. Repartition a mesh among all processes such that the load
is balanced, possibly according to weights.
 \item \texttt{Ghost}. For each process, assemble the layer of directly
     neighboring elements owned by other processes.
 \item \texttt{Balance}. Establish a 2:1 size condition between neighbors in a given refined mesh.
              Thus, the levels of any two neighboring simplices must differ by at most 1.
 \item \texttt{Iterate}. Iterate through the local mesh, executing a callback
     function on each element and on all inter-element interfaces.
\end{itemize}
Since partitioning via SFC only uses the SFC index as information, we refer to
already existing descriptions of \texttt{Partition} for hexahedral or
simplicial SFCs; see \cite{PilkingtonBaden94,BursteddeWilcoxGhattas11} and Section~\ref{sec:partitionsfc}.
\texttt{Ghost} and \texttt{Balance} 
are sophisticated parallel algorithms and require additional theoretical work.
We describe them in Chapters~\ref{ch:ghost} and~\ref{ch:balance}.
For \texttt{Iterate} see our remarks in Section~\ref{ch:app}.

Here, we briefly describe \texttt{New} and \texttt{Adapt}.
In the forest-of-trees approach we model an adaptive mesh
by a coarse mesh of level $0$ $d$-simplices, the \textbf{trees}.
Such a coarse mesh could be specified manually for simple geometries,
or obtained from executing a mesh generator.
Each level $0$ simplex is identified with the root simplex $T_d^0$
and then refined adaptively to produce the fine and potentially nonconforming
mesh of $d$-simplices.
These simplices are partitioned among all processes; thus each process holds
a range of trees, of which the first and last may be incomplete: Their leaves
are divided between multiple processes.

An entity $\mathcal F$ of the structure \texttt{forest} consists of  the following entries
\begin{itemize}
 \item $\mathcal{F}.C$ --- the coarse mesh;
 \item $\mathcal{F}.\mathcal{K}$ --- the process-local trees;
 \item $\mathcal{F}.\mathcal{E}_k$ --- for each local tree $k$ the list of process-local
                    simplices in tetrahedral Morton order.
\end{itemize}
We acknowledge that \texttt{New} and \texttt{Adapt} are essentially
communication-free, but still serve well to exercise some of the fundamental
algorithms described earlier.

\subsection{New}
\label{sec:new}

The \texttt{New} algorithm creates a partitioned uniform level $\ell$ refined forest from a given coarse mesh.
To achieve this, we first compute the first and last $d$-simplices belonging to the current process $p$.
From this range we can calculate which trees belong to $p$ and for each of these trees, the consecutive index of
the first and last $d$-simplices on this tree.
We then create the first simplex in a tree by a call to \texttt{T} (Algorithm \ref{alg:Init_id}).
In contrast to the \texttt{New} algorithm in \cite{BursteddeWilcoxGhattas11}
we create the remaining simplices by calls to \texttt{Successor} instead of \texttt{T} to
avoid the $\mathcal{O}(\ell)$ runtime of \texttt{T} in the case of simplices.
Our numerical tests, displayed in Figure \ref{fig:scalenew},
show that the runtime of \texttt{New} is in fact linear in the number of
elements and does not depend on the level $\ell$.
Within the algorithm, $K$ denotes the number of trees in the coarse mesh and 
$P$ the number of processes.
\begin{algorithm}
\caption{\texttt{New}(\texttt{Coarse Mesh} $C$, \Int{} $\ell$)}
\label{alg:New}
\DontPrintSemicolon
\algoresult{A partitioned uniform level $\ell$ forest with $C$ as coarse mesh.}\;
$n \gets 2^{d\ell}$, $N\gets nK$ \Comment{$d$-simplices per tree and global number of $d$-simplices}\vspace{0.8ex}
$g_\mathrm{first} \gets \left\lfloor {Np}/{P}\right\rfloor$,
$g_\mathrm{last} \gets \left\lfloor{N(p+1)}/{P}\right\rfloor-1$ \Comment{Global numbers of first and last..}\vspace{0.8ex}
$k_\mathrm{first} \gets \left\lfloor g_\mathrm{first}/n \right\rfloor,
\, k_\mathrm{last} \gets \left\lfloor g_\mathrm{last}/n \right\rfloor$\Comment{..local simplex and local tree range}\vspace{0.5ex}
\For{$t \in \set{k_\mathrm{first},\dots,k_\mathrm{last}}$}
{
  $e_\mathrm{first} \gets (t = k_\mathrm{first}) \abst ? g_\mathrm{first}-nt : 0$\;
  $e_\mathrm{last} \gets (t = k_\mathrm{last}) \abst ? g_\mathrm{last}-nt : n-1$\;
  $T\gets$\texttt{T} $(e_\mathrm{first},\ell)$\Comment{Call Algorithm \ref{alg:Init_id}}
  $\mathcal{E}_k \gets \set T$\;
  \For{$e \in \set{e_\mathrm{first},\dots,e_{last}-1}$}
  {
    $T\gets$ \texttt{Successor} $(T)$\;
    $\mathcal{E}_k \gets  \mathcal{E}_k \cup \set T$
  }
}
\end{algorithm}

After \texttt{New} returns, the process local number of elements is known, and
per-element data can be allocated linearly in an array of structures, or a
structure of arrays, depending on the specifics of the application.

We point out, that the operations that require specific knowledge of simplices are
outsourced to the low-level algorithms \texttt{T} and \texttt{Successor}.
If we replace, for example, the implementation of these by the appropriate versions for
quadrilaterals/hexahedra with the Morton index, we obtain a partitioned uniform mesh
with these elements. Even further, we can model hybrid meshes, when we store the 
low-level functions as part of the tree information, i.e.\ each tree has its own
element type.

Thus, this description of \texttt{New} fits into our general approach of
separating high- and low-level algorithms in order to handle different element
types, which we describe in Section~\ref{sec:highlow}

\subsection{Adapt}
\label{sec:adapt}

The \texttt{Adapt} algorithm modifies an existing forest by refining and
coar\-se\-ning the $d$-simplices of a given forest according to a callback
function.
It does this by traversing the $d$-simplices of each tree in tetrahedral Morton order
and passing them to the callback function. If the current $d$-simplex and its $2^d-1$ successors form a family (all having the same parent),
then the whole family is passed to the callback.
This callback function accepts either one or $2^d$ $d$-simplices as input
plus the index of the current tree.
In both cases, a return value greater than zero means that the first input $d$-simplex should be refined, and
thus its $2^d$ children are added in tetrahedral Morton order to the new
forest.  Additionally, if the input consists of $2^d$ simplices, they form a
family, and a return value
smaller than zero means that this family should be coarsened, thus replaced by their parent.
If the callback function returns zero, the first given $d$-simplex remains unchanged
and is added to the new forest, and \texttt{Adapt} continues
with the next $d$-simplex in the current tree.
The \texttt{Adapt} algorithm creates a new forest from the given one and can handle recursive refinement/coarsening.
For the recursive part we make use of the following reasonable assumptions:
\begin{itemize}
 \item A $d$-simplex that was created in a refine step will not be coarsened during the same adapt call.
 \item A $d$-simplex that was created in a coarsening step will not be refined during the same adapt call.
\end{itemize}
From these assumptions we conclude that for recursive refinement we only have to consider those $d$-simplices that
were created in a previous refinement step and that we only have to care about recursive coarsening directly after
we processed a $d$-simplex that was not refined and could be the last $d$-simplex in a family.
If refinement and coarsening are not done recursively, the runtime of \texttt{Adapt} is linear in the number
of $d$-simplices of the given forest.

An application will generally project or otherwise transform data from the
previous to the adapted mesh.
This can be done within the adaptation callback, which is known to proceed
linearly through the local elements, or after \texttt{Adapt} returns if a copy
of the old mesh has been retained.
In the latter case, one would allocate element data for the adapted mesh and then
iterate over the old and the new data simultaneously, performing the
projection in the order of the SFC.
Once this is done, the old data and the previous mesh are deallocated
\cite{BursteddeGhattasStadlerEtAl08}.

\section{Performance evaluation}

Given the design of the algorithms discussed in this chapter, we expect runtimes that are
precisely proportional to the number of elements and independent of the level
of refinement.
To verify this, we present scaling and runtime tests%
\footnote{Version \texttt{v0.1} is available at
\texttt{https://bitbucket.org/cburstedde/t8code.git}}
for \texttt{New} and \texttt{Adapt} on the
JUQUEEN supercomputer at the Forschungszentrum J\"ulich \cite{Juqueen}, an IBM
BlueGene/Q system with 28,672 nodes consisting of
$16$ IBM PowerPC A2 @ $1.6$ GHz and $16$ GB Ram per node.
We also present one runtime study on the full MIRA system at the Argonne Leadership Computing Facility,
which has the same architecture as JUQUEEN and 49,152 nodes.
The biggest occurring number of mesh elements is around $850\e 9$ 
tetrahedra with $13$ million elements per process.

 The first two tests are a strong scaling (up to 131k processes) and a runtime
study of \texttt{New} in 3D, shown in Figure \ref{fig:scalenew}.  For both
tests we use a coarse mesh of $512$ tetrahedra.
 We time the \texttt{New} algorithm with input level 8 (resp.\ level 10 for higher numbers of processes).
 We execute the runtime study to examine whether \texttt{New} has the proposed level-independent linear runtime in the number of generated tetrahedra,
 which can be read from the results presented in the Table in Figure \ref{fig:scalenew}.

 The last test is \texttt{Adapt} with a recursive nonuniform refinement pattern.
 The starting point for all runs is a mesh obtained by uniformly refining a
 coarse mesh of $512$ tetrahedra to a given initial level $k$.
 This mesh is then refined recursively
 using a single \texttt{Adapt} call, where only the tetrahedra of types $0$ and $3$ whose
 level does not exceed the fine level $k+5$ are refined recursively.
 The resulting mesh on each tetrahedron resembles a fractal pattern
 similar to the Sierpinski tetrahedron.
 We perform several strong and weak scaling runs on JUQUEEN starting with 128
 processes and scaling up to 131,072.
 The setting is 16 processes per compute node.
 We finally do another strong scaling run on the full system of the MIRA supercomputer at the
 Argonne Leadership Computing Facility with 786,432 processes and again $16$ processes per compute node.
 Figure \ref{fig:adaptType03} shows our runtime results.

\begin{figure}
\center
\begin{minipage}{0.48\textwidth}
   \resizebox{\textwidth}{!}{\includegraphics{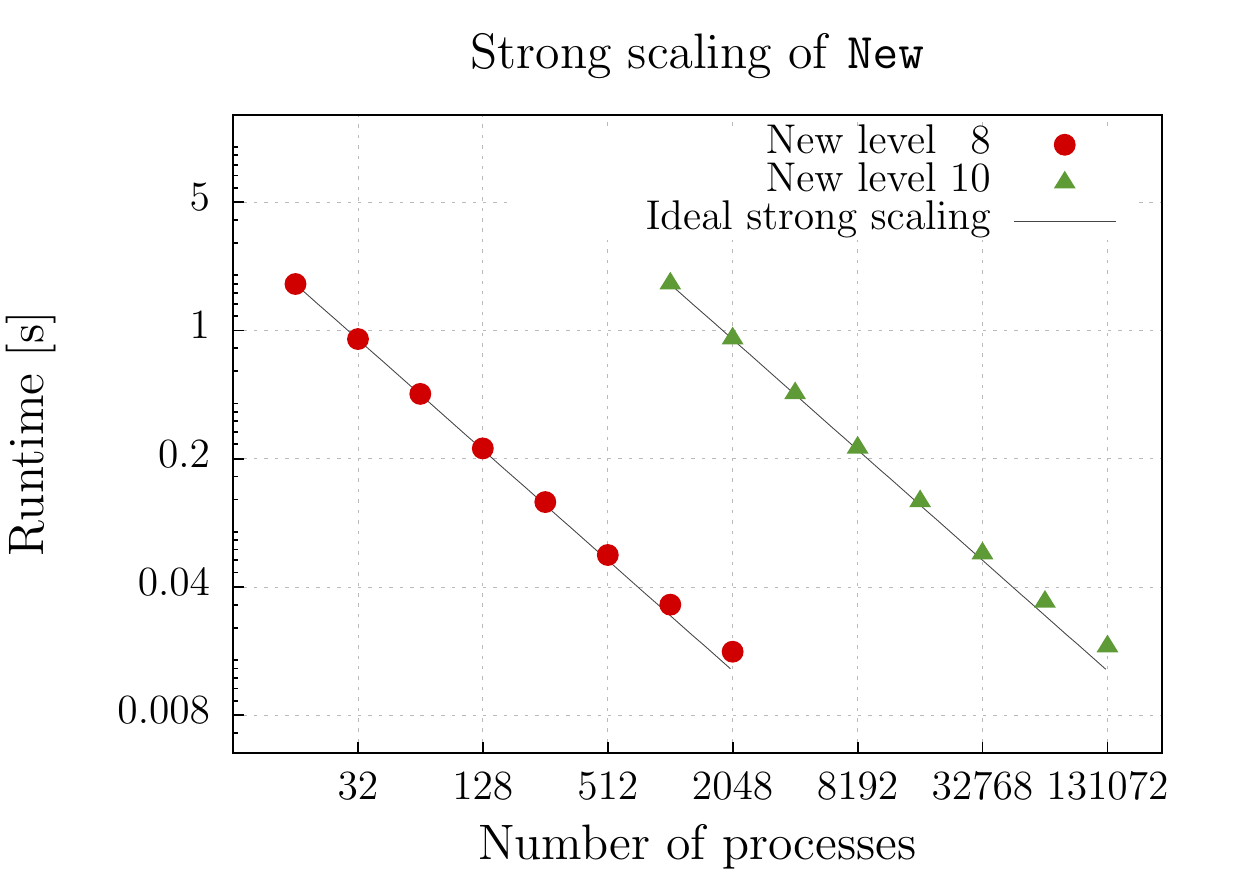}}
\end{minipage}
\begin{minipage}{0.5\textwidth}
\raisebox{18ex}{
\resizebox{0.96\textwidth}{!}{
\begin{tabular}{|lrrrr|}
 \multicolumn{5}{c}{Runtime tests for \texttt{New}}\\[2ex] \hline
 \mytabvspace 
 $\#$Procs&Level & $\#$Tetrahedra & Runtime [s] & Factor\\ \hline
 \mytabvspace
  64  & $7$ & $1.073\e9\hphantom{0}$  &  $0.059$ & --  \\
      & $8$ & $8.590\e9\hphantom{0}$  &  $0.451$ & $7.64$ \\
      & $9$ & $6.872\e10$ &  $3.58$\hphantom{1}  & $7.94$ \\
      &$10$ & $5.498\e11$ & $28.6$\hphantom{11}& $7.99$ \\ \hline
 \mytabvspace
  256 & $8$ & $8.590\e9\hphantom{0}$ & $0.116$ & --\\
 &$9$ & $6.872\e10$ & $0.898$ & $7.74$  \\
&$10$ & $5.498\e11$ & $7.15$\hphantom{1}& $7.96$ \\
 \hline
\end{tabular}
}%
} %
\end{minipage}
\caption[Runtime tests for \texttt{New} on JUQUEEN]
 {Runtime tests for \texttt{New} on JUQUEEN. Left: Two strong scaling
  studies. A new uniform level 8 (circles) and level 10 (triangles) refinement of
  a coarse mesh of 512 root tetrahedra, carried out with 16 up to 2,048 processes
  and 1,024 up to 131,072 processes with 16 processes per compute node. Right:
  The data shows that the runtime of \texttt{New} is linear in the number of
  generated elements and does not additionally depend on the level. The uniform
  refinement is created from a coarse mesh of 512 root tetrahedra. For the first
  computation on 64 processes we use 1 process per compute node and for the
  computation on 256 processes we use 2 processes per node.}
\figlabel{fig:scalenew}
\end{figure}

 \begin{figure}
\center
\begin{minipage}{0.48\textwidth}
 {\includegraphics{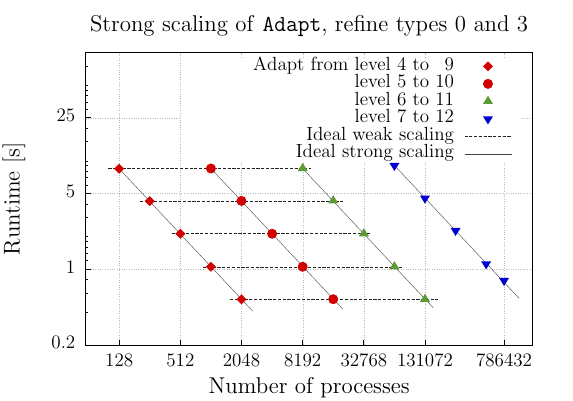}}
\end{minipage}
\hfill
\begin{minipage}{0.42\textwidth}
   \includegraphics[width=0.99\textwidth, trim=150 100 100 40, clip]{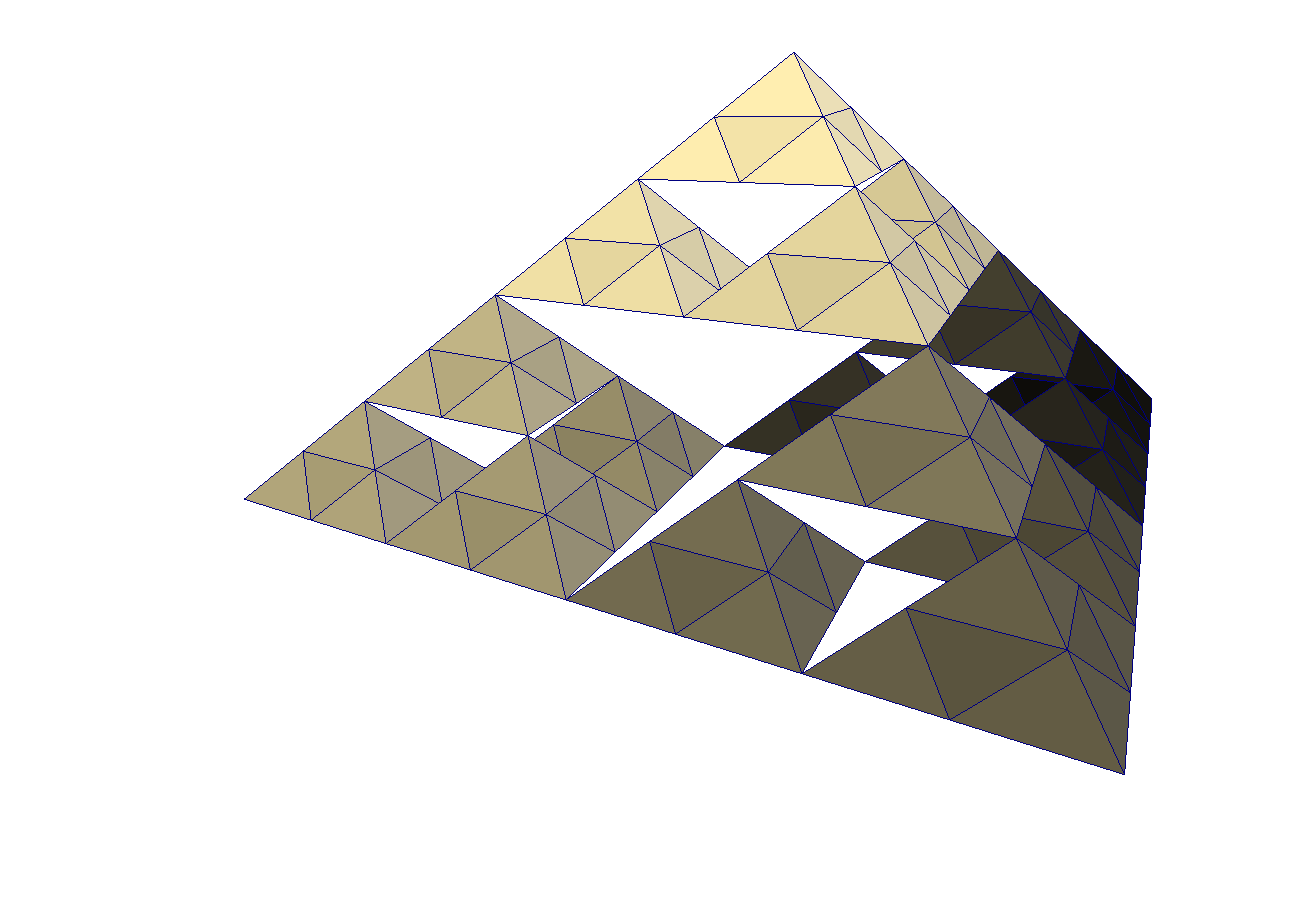}
\end{minipage}
\caption[Strong scaling for \texttt{Adapt} on JUQUEEN]
  {Strong scaling for \texttt{Adapt} with a fractal refinement pattern. Starting
  from an initial level $k$ on a coarse mesh of 512 tetrahedra we refine
  recursively to a maximal final level $k+5$.  The refinement callback is such
  that only subtetrahedra of types 0 and 3 are refined.  Left: Strong and weak
  scaling on JUQUEEN with up to 131,072 processes and strong scaling on MIRA with
  up to 786,432 processes.  On both systems we use 16 processes per compute node.
  The level 12 mesh consists of 858,588,635,136 tetrahedra.  Right: An initial
  level 0 and final level 3 refinement according to the fractal pattern.  The
  subtetrahedra of levels 1 and 2 are transparent.}
\figlabel{fig:adaptType03}
\end{figure}

\section{Conclusion}
We present a new encoding for adaptive nonconforming
triangular and tetrahedral mesh refinement based on Bey's red-refinement rule.
We identify six different types of tetrahedra (and two types of triangles)
and prescribe an ordering of the children for each of these types that differs
from Bey's original order.
By introducing an embedding of the mesh elements into a
Cartesian coordinate structure, we define a tetrahedral Morton index that can
be computed using bitwise interleaving similar to the Morton index for cubes.
This tetrahedral Morton index shares some properties with the well-known
cubical one and allows for a memory-efficient random access storage of the mesh
elements.

Exploiting the Cartesian coordinate structure, we develop several constant-time
algorithms on simplices.
These include computing the parent, the children, and the face-neighbors of a
given mesh element, as well as computing the next and previous elements according
to the SFC.

In view of providing a complete suite of parallel dynamic AMR capabilities, the
constructions and algorithms described in this chapter are just the beginning.
A repartitioning algorithm following our SFC, for example, is easy to imagine,
but challenging to implement if the tree connectivity is to be partitioned
dynamically, and if global shared metadata shall be reduced from being
proportional to the number of ranks to the number of compute nodes,
see Chapter~\ref{ch:cmesh}.
The present chapter provides atomic building blocks that can be used in
high-level algorithms for 2:1 balancing
\cite{IsaacBursteddeGhattas12}
and the computation of ghost elements and generalized topology iteration
\cite{IsaacBursteddeWilcoxEtAl15}.
We address these algorithms in Chapters~\ref{ch:ghost} and~\ref{ch:balance}.
The choices presented in this chapter
are sustainable for maintaining extreme scalability in the long term.

 \chapter{Connected Components of the TM-SFC}
\label{ch:conncomp}
This chapter is based on the preprint~\cite{BursteddeHolkeIsaac17b}.  Since the
preprint also contains contributions of Burstedde and Isaac, we only present
those parts of it that are work of the author of this thesis. In particular,
these are the results about the number of face-connected components for the TM-SFC.
\vspace{2ex}

When we store a mesh in parallel, its elements are distributed among the
processes along the SFC. This means that each process is assigned a range of
elements that is contiguous with regards to their SFC index. See also
Section~\ref{sec:partitionsfc}.

In many numerical applications, the processes communicate with each other
across the boundaries of their partitions. See for example the construction of
a ghost layer in Chapter~\ref{ch:ghost}.
The total volume of parallel communication is thus proportional to the number
of processes that share boundaries of a process' domain and proportional to the
number of elements at this boundary.
Hence, in order to minimize the communication, the surface-to-volume ratio of
these partitions should be small.
A good indicator for this is the number of face-connected components of a process'
domain, and thus the number of face-connected components of a segment of the SFC.

It is a known fact that the cubical Morton SFC can produce disconnected 
segments. However, the number of face-connected components was shown to be at most
two~\cite{BursteddeHolkeIsaac17b,Bader12}.
Similarly, the TM-SFC can produce disconnected domains 
(see for example Figure~\ref{fig:face_conn_ex}).
In this chapter, we prove the following bounds for the count of face-connected
components.

\begin{theorem}
  \label{thm:illthmalltets}
  A contiguous segment of a tetrahedral Morton curve through a uniform or
  adaptive tree of maximum refinement level $L\geq 2$ produces at most $2(L-1)$
  face-connected subdomains in 2D and at most $2L+1$ in 3D.
  For $L=1$ there are at most two face-connected subdomains.
\end{theorem}

We complete our study with numerical results to
illustrate the distribution of continuous vs.\ discontinuous segments.
This supports the conjecture that the tetrahedral Morton curve is no worse in
practice than the original cubical construction.

\section{Proof of Theorem~\ref{thm:illthmalltets}}
\label{sec:cc_illustrated}

We examine the number of face-connected components of a
segment of the tetrahedral Morton SFC.
As we show in Figure~\ref{fig:face_conn_ex}, there exist cases where the number
of face-connected components in a uniform 2D level $L$ refinement can be as
high as $2(L-1)$.
We show that this is in fact a sharp upper bound.
We also show that in three dimensions the number of face-connected components
does not exceed $2L+1$. There exists an example with $2L$ face-connected
components and we conjecture that $2L$ is in fact the sharp estimate.
The proof of these bounds is fairly analogous to the results for cubes from~\cite{BursteddeHolkeIsaac17b} and 
relies on a divide-and-conquer approach by splitting the segment into 
subsegments of which we know the number of face-connected components.

\begin{remark}
  In this chapter we sometimes refer to a TM-SFC on a type $1$ root triangle.
  By this we mean the SFC that would result from the construction in Section~\ref{sec:tetsfcdef}
  if we took the triangle $2^\mathcal L S_1$ as root triangle; see also Figure~\ref{fig:sechstetra}.
  We use it since the geometry of this curve is the same as if we consider the
original (type 0) TM-SFC and restrict it to the level $1$ subtriangle of type
$1$.
\end{remark}

\begin{remark}
  \label{rem:tetsym}
For the cubical Morton curve, the proof that there are at most two face-connected
components uses a symmetry property of the Morton index
(see~\cite[equation~(4)]{BursteddeHolkeIsaac17b}). 
  Given a $d$-dimensional cube
$q$ of level $L$ with Morton index $Q$, we can take the bitwise negation of $Q$
\begin{equation}
  R(Q) = 2^{dL}-1-Q
\end{equation}
to obtain the index of a new cube $q'$. This $q'$ is the cube that results from
traversing the SFC $Q$ steps in reverse.

We do not have such a strong symmetry property for simplices.
However, in $2D$ it holds that reversing the TM curve in a uniform refinement 
of a type $0$ root triangle results in the forward TM curve for a type $1$ root
triangle.
\end{remark}

Using this Remark, we show a weaker form of Corollary~8 from~\cite{BursteddeHolkeIsaac17b}.

\begin{figure}
   \begin{minipage}{0.6\textwidth}
   \center
   \def\svgwidth{0.7\textwidth}   
\begingroup%
  \makeatletter%
  \providecommand\color[2][]{%
    \errmessage{(Inkscape) Color is used for the text in Inkscape, but the package 'color.sty' is not loaded}%
    \renewcommand\color[2][]{}%
  }%
  \providecommand\transparent[1]{%
    \errmessage{(Inkscape) Transparency is used (non-zero) for the text in Inkscape, but the package 'transparent.sty' is not loaded}%
    \renewcommand\transparent[1]{}%
  }%
  \providecommand\rotatebox[2]{#2}%
  \ifx\svgwidth\undefined%
    \setlength{\unitlength}{1661.70019531bp}%
    \ifx\svgscale\undefined%
      \relax%
    \else%
      \setlength{\unitlength}{\unitlength * \real{\svgscale}}%
    \fi%
  \else%
    \setlength{\unitlength}{\svgwidth}%
  \fi%
  \global\let\svgwidth\undefined%
  \global\let\svgscale\undefined%
  \makeatother%
  \begin{picture}(1,0.72082152)%
    \put(-0.93903232,0.17130655){\color[rgb]{0,0,0}\makebox(0,0)[lb]{\smash{}}}%
    \put(-0.96733522,0.09574826){\color[rgb]{0,0,0}\makebox(0,0)[lt]{\begin{minipage}{0.04668695\unitlength}\raggedright \end{minipage}}}%
    \put(0,0){\includegraphics[width=\unitlength,page=1]{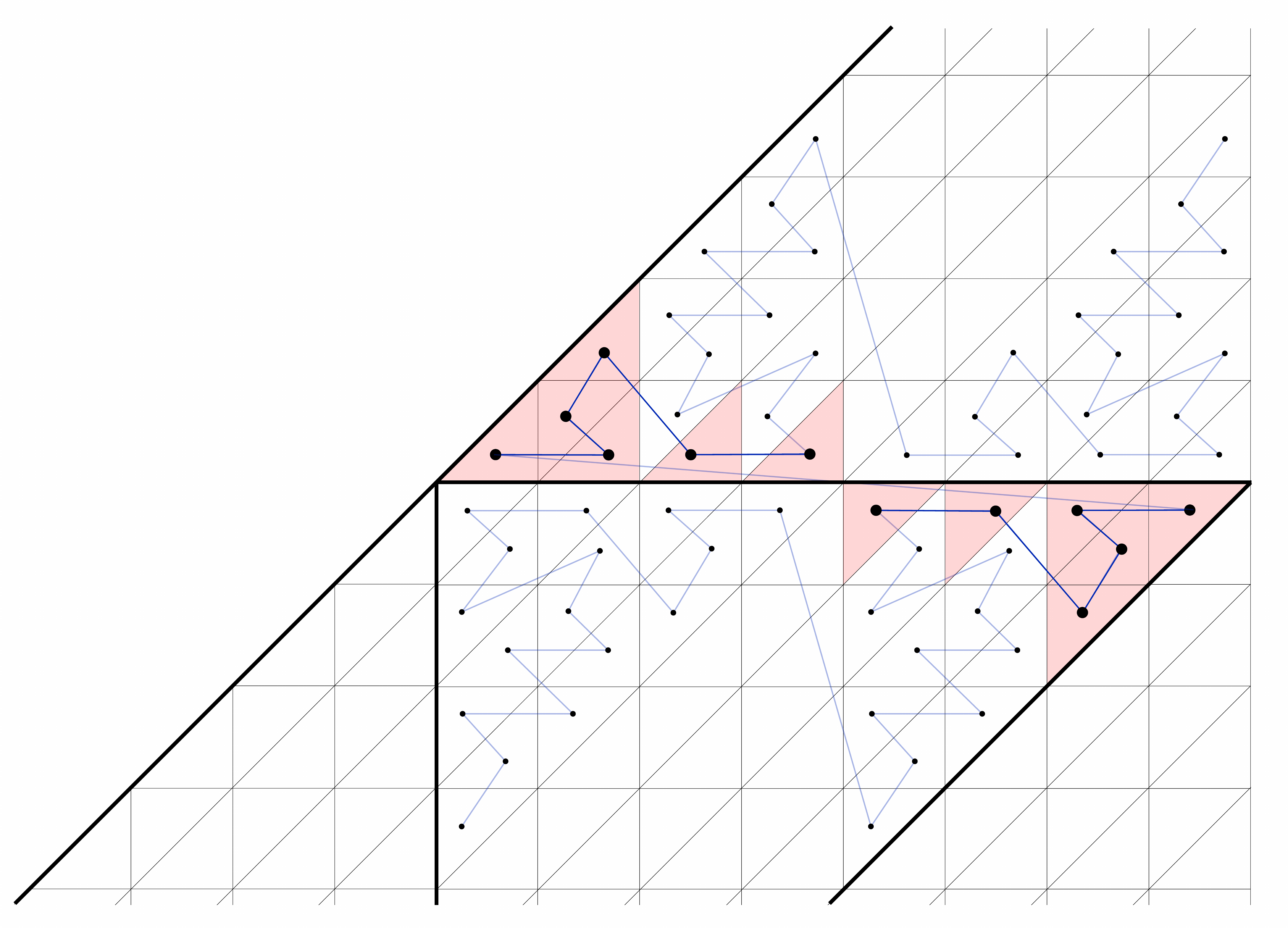}}%
  \end{picture}%
\endgroup%
    \end{minipage}
   \begin{minipage}{0.38\textwidth}    
   \center
   \def\svgwidth{0.8\textwidth}   
\begingroup%
  \makeatletter%
  \providecommand\color[2][]{%
    \errmessage{(Inkscape) Color is used for the text in Inkscape, but the package 'color.sty' is not loaded}%
    \renewcommand\color[2][]{}%
  }%
  \providecommand\transparent[1]{%
    \errmessage{(Inkscape) Transparency is used (non-zero) for the text in Inkscape, but the package 'transparent.sty' is not loaded}%
    \renewcommand\transparent[1]{}%
  }%
  \providecommand\rotatebox[2]{#2}%
  \ifx\svgwidth\undefined%
    \setlength{\unitlength}{372.91103516bp}%
    \ifx\svgscale\undefined%
      \relax%
    \else%
      \setlength{\unitlength}{\unitlength * \real{\svgscale}}%
    \fi%
  \else%
    \setlength{\unitlength}{\svgwidth}%
  \fi%
  \global\let\svgwidth\undefined%
  \global\let\svgscale\undefined%
  \makeatother%
  \begin{picture}(1,0.99810376)%
    \put(0,0){\includegraphics[width=\unitlength,page=1]{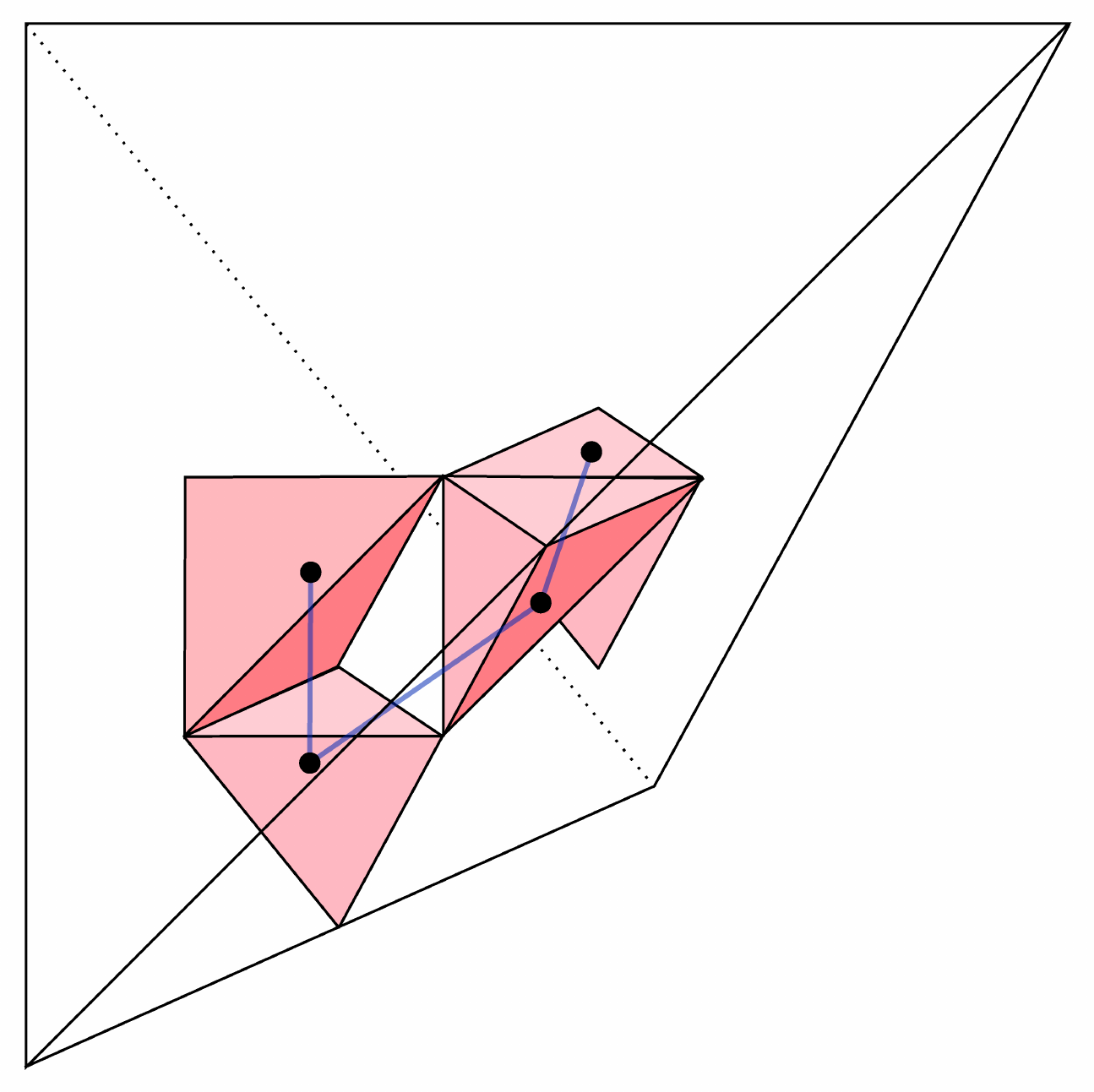}}%
  \end{picture}%
\endgroup%
    \end{minipage}
   \caption[Two example segments of the TM curve]
           {Left: A segment of the 2D SFC on a level $4$ refinement of $T^0_2$
            with six face-connected components (shaded pink).
   The number of face-connected components in 2D can be as high as $2(L-1)$;
   this estimate is sharp.
   Right: a 3D level 2 refinement of $T^0_3$ with four ($=2L$) face-connected
   components.
   We prove that an upper bound on the number of face-connected components is
   $2L+1$ and conjecture that $2L$ is sharp.
  }
  \label{fig:face_conn_ex}
\end{figure}%

\begin{lemma}
  \label{lem:faceconn2dhelp}

  The following two properties hold for the TM-index in 2D, where we consider a
  uniform level $L$ refinement of an initial type 0 triangle $T$.
  \begin{itemize}
    \item Each type 1 subsimplex is face-connected to a type 0
      subsimplex with a greater TM-index.
    \item Each type 0 subsimplex that is also a descendant of the level 1,
      type 1 subtriangle $T_3$
      is face-connected to a type 1 subsimplex with a greater TM-index.
  \end{itemize}
\end{lemma}
\begin{proof}
  The respective face-neighbor is the top face-neighbor for the type 1 subsimplex
  and the face-neighbor along the diagonal face for the type 0 subsimplex;
  see Figure~\ref{fig:faceconn2dhelp}.
  For type 0 we additionally require that the subsimplex is
  a descendant of $T_3$, since this ensures that the face-neighbor along the
  diagonal face is inside the root triangle.
  Despite this detail, the proofs for both items are identical, and we only
  present one for the first.

  Let $S$ denote an arbitrary type 1 subsimplex of level $L$ and let $S'$ be
  its neighbor across the top face.
  If $S$ and $S'$ share the same parent $P$ then there are two cases,
  which we also see in Figure~\ref{fig:haverkortSFC}:
  Either $\type(P) = 0$, then the local index of $S$ is 2 and that of $S'$ is
  3, or $\type(P) = 1$, in which case the local index of $S$ is 0 and that
  of $S'$ is 1.
  Thus, in both cases the TM-index of $S$ must be smaller than that of $S'$.
  We suppose now that $S$ and $S'$ have different parents, which implies
  $L\geq2$, and denote these different level $L-1$ subsimplices by $P$ and
  $P'$.
  The only possible combination is that $\type(P) = 1$ and $\type(P') = 0$, and
  that $P$ and $P'$
  are neighbors along $P$'s top face. Therefore, by an induction argument,
  $m(P) < m(P')$,
  and since the TM-index preserves the local order under refinement, each child
  of $P$ has a smaller TM-index than each child of $P'$; see Theorem~\ref{thm:IndexProps}.
  In particular we find $m(S)<m(S')$.
\end{proof}

\begin{figure}
  \center
  \def\svgwidth{0.8\textwidth}
\begingroup%
  \makeatletter%
  \providecommand\color[2][]{%
    \errmessage{(Inkscape) Color is used for the text in Inkscape, but the package 'color.sty' is not loaded}%
    \renewcommand\color[2][]{}%
  }%
  \providecommand\transparent[1]{%
    \errmessage{(Inkscape) Transparency is used (non-zero) for the text in Inkscape, but the package 'transparent.sty' is not loaded}%
    \renewcommand\transparent[1]{}%
  }%
  \providecommand\rotatebox[2]{#2}%
  \ifx\svgwidth\undefined%
    \setlength{\unitlength}{1272.22363281bp}%
    \ifx\svgscale\undefined%
      \relax%
    \else%
      \setlength{\unitlength}{\unitlength * \real{\svgscale}}%
    \fi%
  \else%
    \setlength{\unitlength}{\svgwidth}%
  \fi%
  \global\let\svgwidth\undefined%
  \global\let\svgscale\undefined%
  \makeatother%
  \begin{picture}(1,0.45727227)%
    \put(0.47865644,0.21027668){\color[rgb]{0,0,0}\makebox(0,0)[lb]{\smash{}}}%
    \put(0.44168894,0.11158709){\color[rgb]{0,0,0}\makebox(0,0)[lt]{\begin{minipage}{0.06097962\unitlength}\raggedright \end{minipage}}}%
    \put(0,0){\includegraphics[width=\unitlength,page=1]{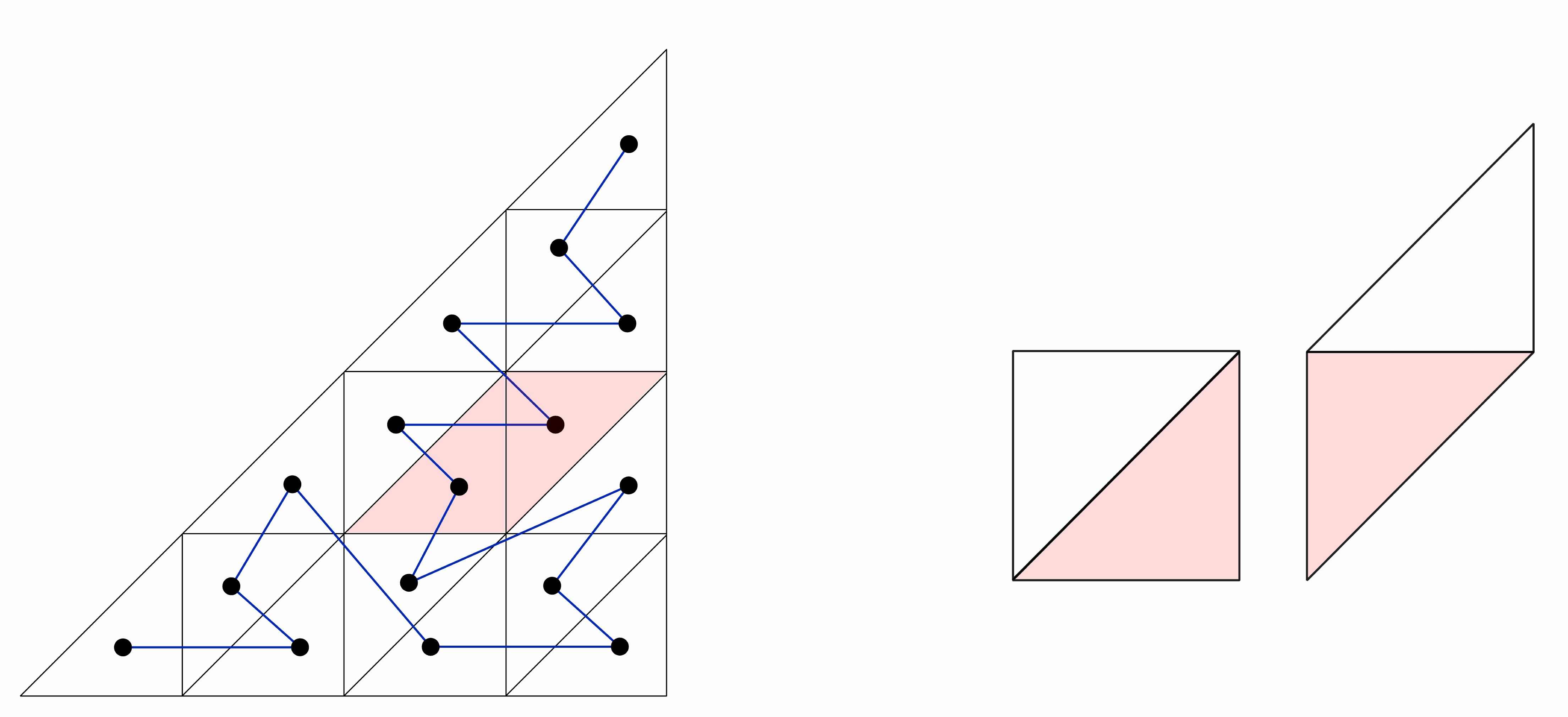}}%
    \put(0.86317745,0.18191688){\color[rgb]{0,0,0}\makebox(0,0)[lb]{\smash{$S_1$}}}%
    \put(0.90085919,0.26431289){\color[rgb]{0,0,0}\makebox(0,0)[lb]{\smash{$S_1$'}}}%
    \put(0.37803178,0.12930008){\color[rgb]{0,0,0}\makebox(0,0)[lb]{\smash{}}}%
    \put(0.25038127,0.13181536){\color[rgb]{0,0,0}\makebox(0,0)[lb]{\smash{$S_0$}}}%
    \put(0.32919351,0.15768572){\color[rgb]{0,0,0}\makebox(0,0)[lb]{\smash{$S_1$}}}%
    \put(0.66936407,0.17811985){\color[rgb]{0,0,0}\makebox(0,0)[lb]{\smash{$S_0$'}}}%
    \put(0.71225608,0.10700408){\color[rgb]{0,0,0}\makebox(0,0)[lb]{\smash{$S_0$}}}%
  \end{picture}%
\endgroup%
   \caption[Illustration of Lemma \ref{lem:faceconn2dhelp}]
  {Illustration of Lemma \ref{lem:faceconn2dhelp}.
   In 2D, choose any subsimplex $S_\ast$. If its neighbor along the top face
   $S_\ast'$ is inside the root triangle, then $m(S_\ast) < m(S_\ast')$.
   This condition is always fulfilled by any type 1 triangle and by type 0
   triangles that are descendants of the middle level 1 subtriangle.
  }
  \label{fig:faceconn2dhelp}
\end{figure}
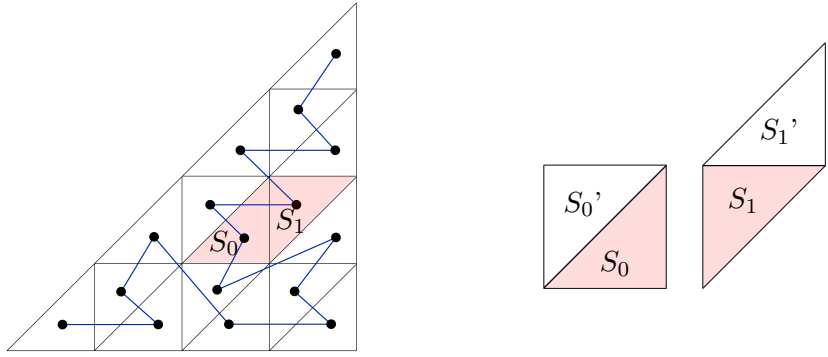

We can now prove a first connectivity result for the 2D case.

\begin{lemma}
  \label{lem:faceconn2d}
 Consider a triangle $T$ that is uniformly refined to level $L$.
 If $T$ has type 0, then a contiguous segment of the SFC ending in the last
 level $L$ subsimplex has just one face-connected component.
 If $T$ has type 1, then this holds for segments starting in the first level
 $L$ subsimplex.
\end{lemma}
\begin{proof}
  We present the proof for $\type (T)=0$, since we can then use
  the symmetry of the 2D curve (Remark \ref{rem:tetsym}) to obtain
  the result for the case $\type (T)=1$.
  We proceed by induction over $L$.

  For $L=0$ there is only one possible segment and it is face-connected.
  For $L=1$ we obtain the result by investigating all 10 cases.
  For $L>1$, let $j\in\{0,1,2,3\}$ be the local index of the level
  1 subtree $T'$ of $T$ in which the first level $L$ subsimplex of the segment lies.
  If $j\in\{0,1,3\}$, then the type of $T'$ is 0 and the statement follows
  by induction.
  Thus, let $j=2$, i.e., the segment starts in the type 1 subtree of $T$.
  The part of the segment that is not inside $T'$ is the full last subtree of
  $T$ (local index 3) and thus it is face-connected in itself.
  With Lemma \ref{lem:faceconn2dhelp} we conclude that each subsimplex in the
  subsegment in $T'$ is face-connected to a simplex with greater TM-index.
  Iterating this process, we conclude that each of these subsimplices is
  face-connected to a subsimplex of the full last subtree of $T$.
  Thus, the whole segment is face-connected.
\end{proof}

For all other segments beginning with the first or ending in the last level $L$
subsimplex, and notably for all of those segments in 3D, we obtain an upper
bound of $L+1$ face-connected components, which we show in the next two lemmas.

\begin{lemma}
\label{lem:faceconnlem0}
 Let a segment of the TM-SFC for a uniform level $L$ refined $d$-simplex
 consist of several full level 1 subsimplices plus
 one single
 level $L$ simplex either at the end or at the beginning, then this segment has
 at most two face-connected components.
\end{lemma}

\begin{proof}
  Similarly to the proof of Proposition~9 in~\cite{BursteddeHolkeIsaac17b},
  we can show this claim by enumerating all possible cases. There is no
  induction necessary.
\end{proof}
\begin{lemma} 
\label{lem:faceconnlem1}
 If a $d$-simplex is uniformly refined to level $L$, then any
 segment of the TM-SFC ending in the last subsimplex or starting in
 the first has at most $L+1$ face-connected components.
\end{lemma}
\begin{proof}
  Consider the case that the segment starts in the first simplex.
  For $L=0$ there is only one possible segment consisting of the unique level
  0 subsimplex and it is thus face-connected.
  Let now $L>0$.
  Since the segment begins at the very first level $L$ subsimplex, we can
  separate it into two parts.
  The first part at the beginning consists of 0 to $2^d-1$ full level 1 subtrees,
  and the second part is one possibly incomplete level 1 subtree.

  By the induction assumption, the second part has at most $L$ face-connected
  components.
  From Lemma \ref{lem:faceconnlem0} we obtain that the first part
  together with the first level $L$ subsimplex of the second part has at most
  two face-connected components.
  Since this first level $L$ subsimplex is contained in one of the components
  of the second part, we obtain
  \begin{equation}
    L + 2 - 1 = L + 1
  \end{equation}
  components in total.

  If the segments ends in the last simplex, the order of parts is reversed.
  The first part of the segment is the part in the level 1 subtree where
  the segment starts, and the second part consists of the remaining full
  level 1 subtrees.
  We obtain the bound on the number of face-connected components using the same
  inductive reasoning as above.
\end{proof}

We have so far argued the connectivity of specific kinds of SFC segments.
This suffices to proceed to arbitrary segments of the tetrahedral Morton SFC.
\begin{proposition}
\label{prop:faceconncomp}
 Any contiguous segment of the TM-SFC of a uniform level $L\geq2$
 refinement of a type 0 simplex has at most $2(L-1)$ face-connected components in 2D and
 $2L+1$ face-connected components in 3D.
 For $L=1$, there are at most two face-connected components, and one for $L=0$
 (this applies to both 2D and 3D).
\end{proposition}
\begin{proof}
 Again, the cases $L=0$ and $L=1$ follow by inspecting all cases.
 Thus, let $L\geq 2$.
 We first show that for $d \le 3$ the number of face-connected
 com\-po\-nents is bounded by $2L+1$:
 If a given segment is contained in a level 1 subtree, we are done by induction.
 Otherwise we can divide the segment into three (possibly empty) pieces:
 First,
 the segment in one incomplete level 1 subtree ending at its last level $L$
 subsimplex,
 then one contiguous segment of full level 1 subtrees
 and finally a segment in one (possibly incomplete) level 1 subtree that
 starts at its first level $L$ subsimplex.
 Lemma~\ref{lem:faceconnlem1}
 implies that the first and the last piece have at most $L$ face-connected
 components each.
 By Lemma~\ref{lem:faceconnlem0}, the second piece has one or two
 face-connected components, and if the number is two, then it is face-connected
 to the first or to the third piece.
 Thus, it adds only one face-connected component to the total number, and we
 obtain at most
 \begin{equation}
  L + 1 + L = 2L + 1
 \end{equation}
 face-connected components.

Let us now specialize to 2D.
We conclude from Lemma \ref{lem:faceconn2d} that
the first subsegment only adds more than one face-connected component if
it is contained in the only level 1 subtree of type 1 (local index 2).
Similarly, the third subsegment only adds more than one face-connected component
if it is contained in a level 1 subtree of type 0.
In particular, if both subsegments add more than one face-connected component, the
third subsegment is contained in the last level 1 subtree (local index 3).
Thus, the second subsegment is empty in this case.

If both of these subsegments have less than $L$ face-connected components, there
is nothing left to show since the overall number of components is then less
than or equal to $2(L-1)$.
So suppose that one of the subsegments has $L$ face-connected components and the other one
has at least $L-1$.
We depict this situation in Figure~\ref{fig:explaintrianglprop}.
We observe that the first and second level $L$ simplex in this first segment
are face-connected to the first and second level $L$ simplex in the second segment.
If, however, the second subsegment has $L$ face-connected components then its last two 
level $L$ simplices are face-connected to the last two level $L$ simplices of the first 
subsegment.

We thus can subtract two face-connected components from the total count, which leads
to at most
\begin{equation}
  L + L - 2 = 2(L-1)
\end{equation}
face-connected components in total.
\end{proof}

\begin{figure}
  \center
  \includegraphics[width=0.6\textwidth]{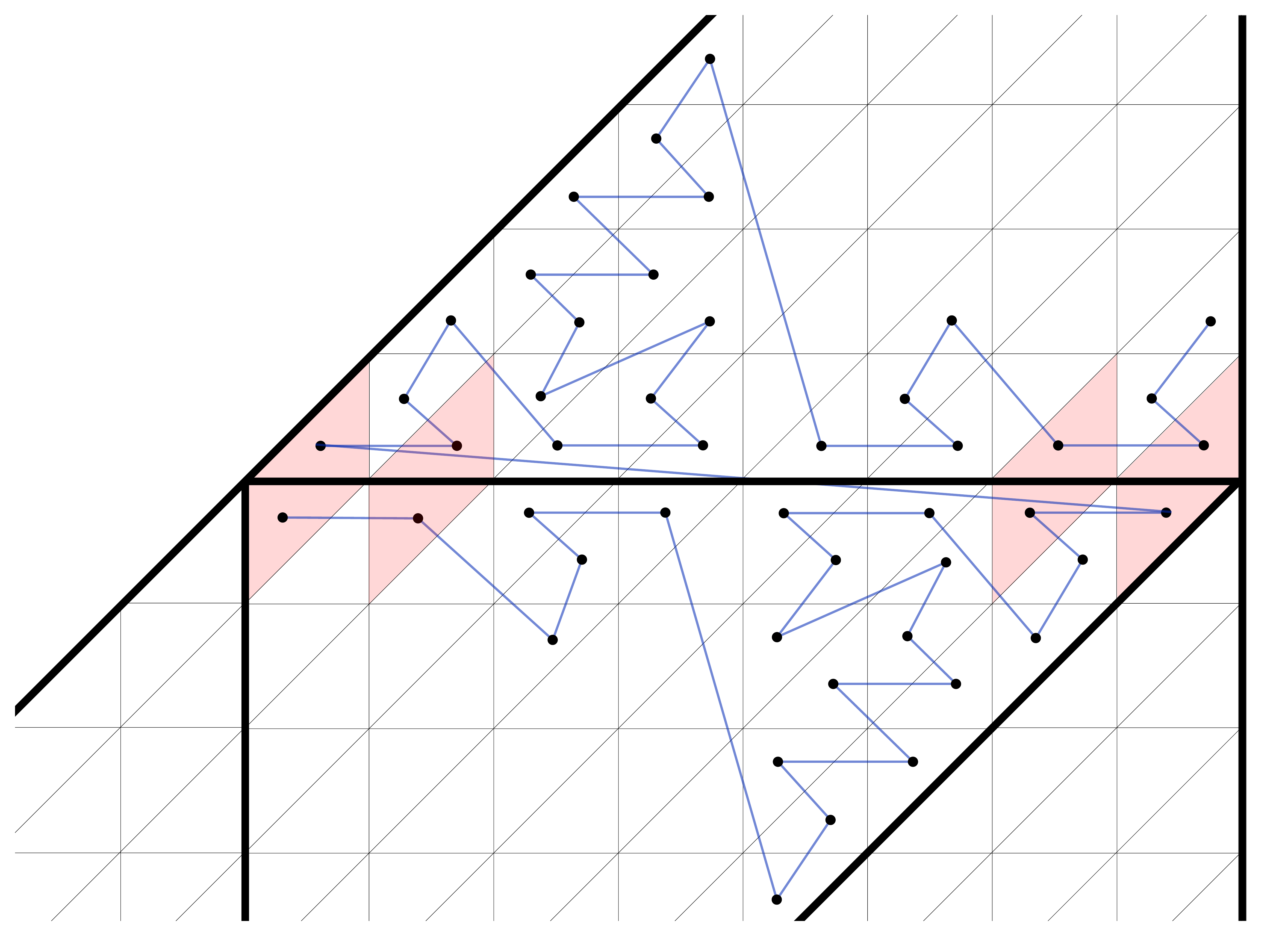}
  \caption[The 2D case in the proof of Proposition \ref{prop:faceconncomp}]
  {An illustration of the 2D case in the proof of Proposition
   \ref{prop:faceconncomp} for $L=4$. The bottom segment has the maximal number of
   $L$ face-connected components. 
   Since its first and second triangle (on the left, shaded in pink) are face-connected with the
  top segment, the possible number of face-connected components is reduced by two.
  If the second segment has $L$ face-connected components as well, then its
  last two triangles (on the right) are face-connected with the bottom segment.
  Thus, the number of face-connected components is less than or equal to $2L-2$.}
  \label{fig:explaintrianglprop}
\end{figure}

We briefly discuss whether we can sharpen these bounds.
In 2D, this is not possible by counterexample; see Figure~\ref{fig:face_conn_ex}.
In 3D, we construct a segment with $2L$ face-connected components using the
consecutive SFC-indices 22--25 of a uniform level 2 refinement of a type 0
tetrahedron.
We believe that the case that the first and the last piece described in
the proof of Proposition \ref{prop:faceconncomp} have $L$ face-connected
components each and that additionally the middle piece adds one component
does not occur.
\begin{conjecture}
\label{con:3TMconjecture}
 In 3D, the number of face-connected components is bounded by $2L$.
 This estimate is sharp.
\end{conjecture}

\subsection{From uniform to adaptive meshes}
\label{sec:cc_main}

We close this section with the extension of the proof
from uniform to adaptive meshes, which
is the remaining step to establish Theorem~\ref{thm:illthmalltets},

We have completed the necessary proofs for a uniform space division for
triangular and tetrahedral refinement (see
\ref{sec:cc_illustrated}).
As we state in this section, an adaptive space division does not require any
more effort (see also \cite{BursteddeHolkeIsaac17b} and \cite[page 176]{Bader12}).
\begin{proof}[Proof of Theorem~\ref{thm:illthmalltets}]%
  Any adaptive tree of simplices with level $l\leq L$ can be refined into level
  $L$ simplices exclusively. This operation does not change the connectivity
  between boundaries of the designated subdomain. In particular, the number of
  face-connected subdomains remains unchanged and the proof reduces to applying
  Proposition~\ref{prop:faceconncomp} above.
\end{proof}

\subsection{From one tree to a forest}
\label{sec:cc_forest}

If we consider a forest of octrees as in Section~\ref{sec:SFConforest} or in
\cite{StewartEdwards04,BangerthHartmannKanschat07, BursteddeWilcoxGhattas11}, a
contiguous segment of the TM curve may traverse more than one tree.
In this case, the segment necessarily contains the last subsimplex of any
predecessor tree, as well as the first subsimplex of any successor tree in the
segment.
For the simplicial case, we may use Lemmas~\ref{lem:faceconn2d} and
\ref{lem:faceconnlem1} to use
the bounds $L+1$ (2D) and $2L+1$ (3D) for the respective parts of the segment,
not counting the transitions between full trees.

\section{Enumeration of face-connected segments}
\label{sec:cc_enumeration}

We would like to examine not only how many pieces an SFC segment can have, but
also how frequently segments of different numbers of pieces occur.
To this end, we supply numerical studies for the TM-SFCs and compare
the results with the cubical Morton curve.

We enumerate all possible TM-SFC segments for a given uniform
refinement level and compute the number of their face-connected components.
We also compute the relative counts of face-connected and non-connected segments for
the cubical Morton curves.
We achieve this by performing a depth-first search on the connectivity
graph of the submesh generated by the segment%
\footnote{\url{https://github.com/holke/sfc_conncomp}}.

In an application, all possible lengths of SFC segments can occur. On the one
hand, we could have a forest consisting of a single tree. If the
number of participating processes is of the same magnitude than the number of
elements in that tree, then very small lengths of segments occur, possibly even
segments consisting only of a single element. 
On the other hand, consider a setting where we have many trees, possibly as
many or more trees than processes. In this case, the lengths of SFC segments
within a single tree can be arbitrarily large, reaching up to the maximum of
the full tree.
See for example our discussion in Section~\ref{sec:cmeshnumres}.

We now compute the fraction of face-connected segments of any length among all
possible segments.
More precisely, we compute for each possible count of face-connected components
the chance that any randomly choosen SFC segment (with a random length) has
exactly this number of face-connected components.

For a uniform level 5 refined tetrahedron we obtain that 61\% of all
SFC segments are face-connected and only 7\% have four or more face-connected
components.
For a uniform level 8 refined triangle, about 64\% of the segments are face-connected
with 2\% of the segments having four or more components.
For cubes and quadrilaterals, the respective ratios of face-connected segments are
60\% and 71\% (here we know that the disconnected segments have exactly two
components). Thus, comparing cubical and TM curve, we see that in 2D more
segments of the quadrilateral Morton curve are face-connected, and in 3D more
segments of the TM curve are face-connected.

We collect these results in Figure~\ref{fig:tritetnumcomp} and
Table~\ref{tab:componentpercent}.

\begin{figure}
        \begin{center}
\includegraphics[width=0.65\textwidth]{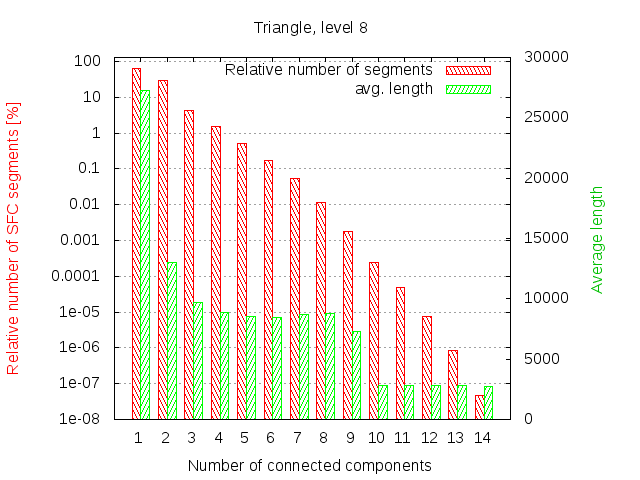}\\
\includegraphics[width=0.65\textwidth]{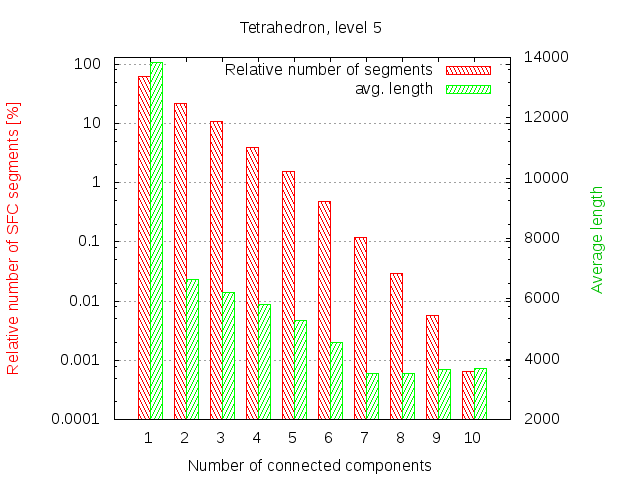}
        \end{center}
\caption[The relative count of TM-curve segments by number of face-connected
components] 
{The relative count of SFC segments (left $y$-axis) by number of face-connected
components and the average length (right $y$-axis) of these segments. 
Left: the distribution for a uniform level 8 refined
triangle. We observe that almost 98\% of all SFC segments have three face-connected
components or less.
63.7\% are face-connected, 29.7\% have two face-connected components and 4.4\% have
three face-connected components.
Right: the distribution for a uniform level 5 refined tetrahedron.
Here, more than 93\% of the segments have three face-connected components or less. 
61.0\% have exactly one face-connected component, 22.1\% two face-connected
components and 10.7\% three face-connected components.
The
highest number of segments occuring are $14=2(8-1)$ in 2D, and $10=2\cdot 5$ in 3D.
This is in agreement with Proposition~\ref{prop:faceconncomp} (2D) and
Conjecture~\ref{con:3TMconjecture} (3D).}
\label{fig:tritetnumcomp}
\end{figure}

\begin{table}
        \begin{center}
\begin{tabular}{|l|r|r|r|r||r|r|}
\hline
&\multicolumn{4}{c||}{Level 5} & \multicolumn{2}{c|}{Level 8}\\
\hline
& Quads & Cubes & Triangles &
\multicolumn{1}{c||}{Tets} & Quads & Triangles\\ \hline
Face-connected     &71.6\% & 60.0\% &63.9\% & 61.0\% &71.4\%&63.7\%\\
Non-connected &28.4\% & 40.0\% &36.1\% & 39.0\% &28.6\%&36.3\%\\
\hline
\end{tabular}
        \end{center}
\caption[The relative counts of face-connected and non-connected segments of 
         Morton SFCs]
  {The relative counts of face-connected and non-connected segments across all
  possible SFC segments of a uniform level 5 and level 8 (2D only) refinement.
  }
\label{tab:componentpercent}
\end{table}

\section{Conclusion}
\label{sec:cc_conclusion}

We show that the bound for the TM-SFC is of order $L$ and thus growing with the
level of refinement.
Yet, we can demonstrate numerically that the fraction of face-connected to
non-connected segments is close to the cubical case. In practice, we may
expect both approaches to behave similarly.

Our result would appear relevant to make informed choices about the
type of space-filling curve to use, for example in writing a new element-based
parallel code for the numerical solution of partial differential equations, or
any other code that benefits from a recursive subdivision of space.
Our theory and experiments support the existing numerical evidence that a
fragmentation of the parallel partition is not observed.

 \chapter{Coarse Mesh Partitioning}
\label{ch:cmesh}
This chapter is based on the paper~\cite{BursteddeHolke17}.
We edited it slightly in order to fit into the general notations of this
thesis, without changing its mathematical content. Copyright \copyright\xspace
by SIAM. Unauthorized reproduction of this chapter is prohibited.

\vspace{2ex}
As we discuss in Chapters~\ref{ch:AMR} and~\ref{ch:sfc}, a technique to model
complex domain shapes is to patch multiple trees together in an unstructured
coarse mesh, giving rise to a forest of elements.
The number of trees per process is limited by the available memory to roughly
$1\e5$ to $1\e6$~\cite{BursteddeWilcoxGhattas11}.  In industrial and medical
meshing however, numbers in the range of one billion or more trees are not
uncommon~\cite{IbanezSeolSmithEtAl16,RasquinSmithChitaleEtAl14,FengTsolakisChernikovEtAl17}.
In order to support such cases, we need to partition the coarse mesh among 
the parallel processes.

Two main approaches for partitioning a forest of elements have been discussed
\cite{Zumbusch03}, namely
(a) assigning each tree and thus all of its elements to one owner process
\cite{SelwoodBerzins99, BurriDednerKloefkornEtAl06, ItoShihErukalaEtAl07},
or
(b) allowing a tree to contain elements belonging to multiple processes
\cite{BangerthHartmannKanschat07, BursteddeWilcoxGhattas11}.
The first approach offers a simpler logic but may not provide acceptable
load balance when the number of elements differs vastly between trees.
The second allows for perfect partitioning of elements by number (the local
numbers of elements between processes differ by at most one) but presents the
issue of trees that are shared between multiple processes.

We choose paradigm (b) for speed and scalability, challenging ourselves to
solve an $n$-to-$m$ communication problem for every coarse mesh element.  Thus
the objective of this chapter is to develop how to do this without handshaking
(\ie, without having to determine separately
which process receives from which) and with a minimal number of senders,
receivers, and messages.
Our main contribution is to avoid identifying a single owner process for each tree
and instead treat all its sharer processes as algorithmically active, under the
premise that they produce a disjoint union of the information necessary to be
transferred.  In particular, each process shall store the relevant tree meta
data to be readily available, eliminating the need to transfer this data from a
single owner process.

In this chapter, we also integrate the parallel transfer of ghost trees.
The reason for this is that each process will eventually collect ghost elements, \ie,
remote elements adjacent to its own.
We discuss the ghost algorithm on elements in Chapter~\ref{ch:ghost}.
Ghost elements of any process may be part of trees that are not in its local
set. To disconnect the ghost element transfer from identifying and
transferring ghost trees, we perform the latter as part of the coarse mesh
partitioning, presently across tree faces.
We study in detail what information we must maintain to reference neighbor
trees of ghost trees (that may themselves be either local, ghost, or neither)
and propose an algorithm with minimal communication effort.

We have implemented the coarse mesh partitioning for triangles and tetrahedra
using the TM-SFC designed in Chapter~\ref{ch:tetSFC}, and for quadrilaterals and
hexahedra exploiting the logic from \cite{BursteddeWilcoxGhattas11}.
To demonstrate that our algorithms are safe to use, we verify that (a) small
numbers of trees require run times on the order of milliseconds and thus
present no noticeable overhead compared to a serial coarse mesh, and  (b)
the coarse mesh partitioning adds only a fraction of run time compared to the
partitioning of the forest elements, even for extraordinarily large numbers of
trees.
We show a practical example of 3D dynamic AMR on 8e3 processes using 383e6 trees
and up to 25e9 elements.
To investigate the ultimate limit of our algorithms, we partition coarse meshes
of up to 371e9 trees on JUQUEEN using 917e3 processes, obtaining a
total run time of about 1.2s and a rate of 340e3 trees per second per process.
On 131e3 processes with half as many ranks per compute node we obtain rates as
high as 750e3 trees per second per process.

We may summarize our results by saying that partitioning the trees can be made
even less costly than partitioning the elements and often executes so fast that
it does not make a difference at all.
This allows a forest code that partitions both trees and elements dynamically
to treat the whole continuum of forest mesh
scenarios, from one tree with nearly trillions of elements on the one extreme
to billions of trees that are not refined at all on the other, with comparable
efficiency.

\section{Tree-based AMR}

We repeat some of the concepts from Chapters~\ref{ch:AMR} and~\ref{ch:sfc}.

A forest $\mathscr F$ consist of a coarse mesh of trees and a fine mesh
of elements that resides from refining the trees.
We use SFCs to order the elements within each tree.
SFCs map the $d$-dimensional elements of a refinement tree to an interval by
assigning a unique integer index $\mathcal I_\mathscr F(E)$ to each element $E$,
see Lemma~\ref{lem:forestconsindex}.
Thus, we can order all elements of that refinement tree linearly in an array.
As in \cite{SundarSampathBiros08}, we do not store the internal (non-leaf)
nodes of the tree.

The choice of SFC affects the ordering of these elements of the forest mesh and
thus the parallel partition of elements.
Possibilities include, but are not limited to, the Hilbert, Peano, or Morton
curves for quadrilaterals and hexahedra \cite{Peano90, Hilbert91, Morton66,
WeinzierlMehl11}, as well as the Sierpi\'nski curve for triangles
\cite{Sierpinski12,BaderZenger06} and the tetrahedral Morton curve for
triangles and tetrahedra from Chapter~\ref{ch:tetSFC}.

As of equation~\eqref{eq:forestorder}, a global order of elements is established first by tree and then by their index
with respect to an SFC (see also \cite{Bader12}):
We enumerate the $K$ trees of the coarse mesh by $0,\ldots,{K-1}$ and
call the number $k$ of a tree its \emph{global index}.
With the global index we naturally extend the SFC order of the leaves:
Let a leaf element of the tree $k$ have SFC index $I$ (within that tree); then we
define the combined index $(k, I)$.
This index compares to a second index $(k', J)$ as
\begin{equation}
  (k, I) < (k', J) \quad :\Leftrightarrow \quad
  \text{$k < k'$ or ($k = k'$ and $I < J$)}
  .
\end{equation}
In practice we store the mesh elements local to a process in one contiguous
array per locally nonempty tree in precisely this order.

The algorithms and techniques discussed in this chapter assume an SFC induced
order among the elements, but they are not affected by the particular choice of
SFC. 
In the \tetcode
software used for the demonstrations in this thesis, we have so far implemented
Morton SFCs for quadrilaterals and hexahedra via the \pforest library
\cite{Burstedde10a} and the tetrahedral Morton SFC for tetrahedra and triangles.
These curves compute the index $m(E)$ of an element via bitwise interleaving
the coordinates of its lower left vertex in a suitable reference tree; see also
Figure~\ref{fig:twotreeforest} for an illustration of the curve on triangles.
Other SFC schemes may be added to the \tetcode in a modular fashion,
see Section~\ref{sec:highlow}.

\begin{figure}
\center
\includegraphics[width=0.45\textwidth]{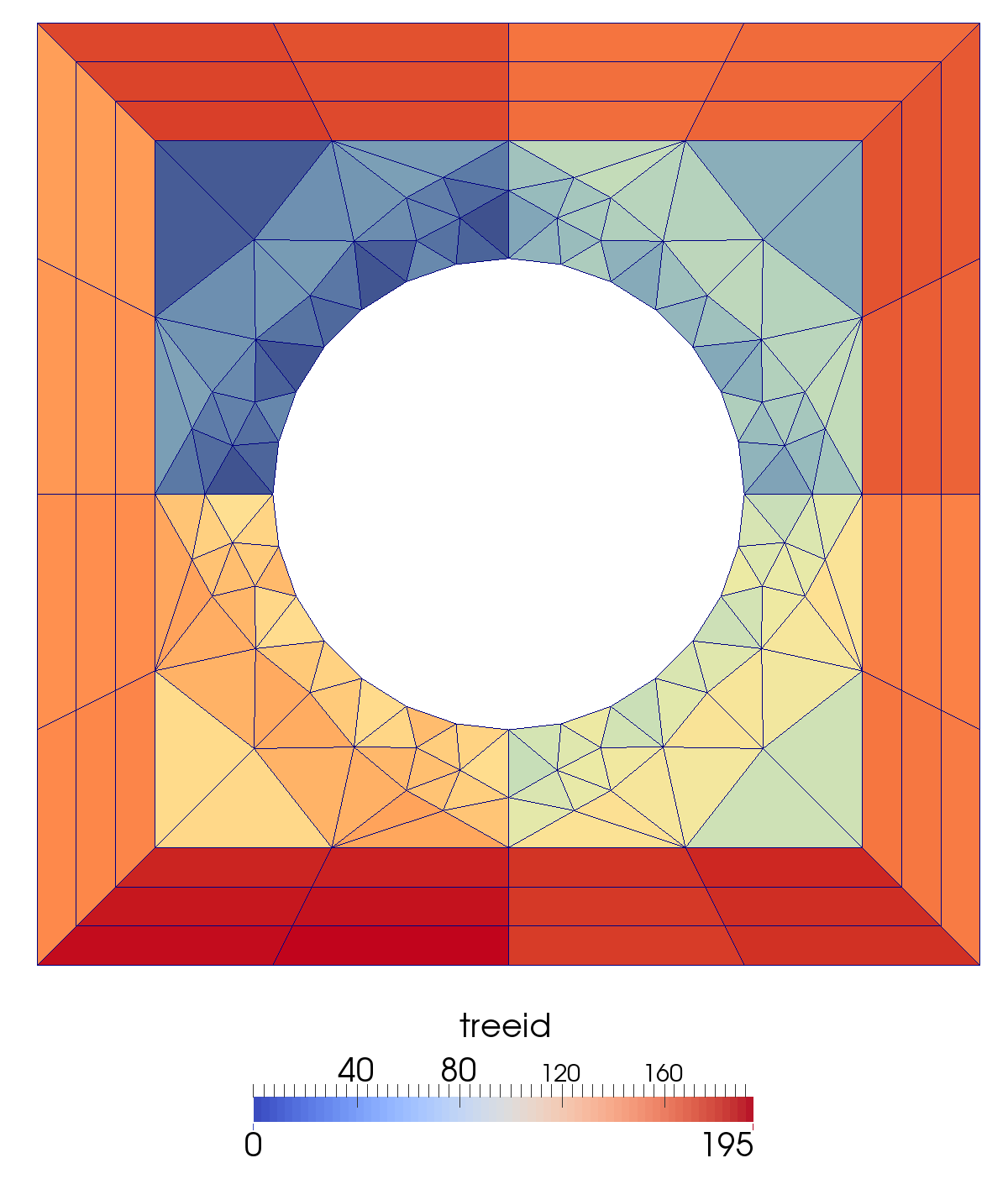}
\includegraphics[width=0.45\textwidth]{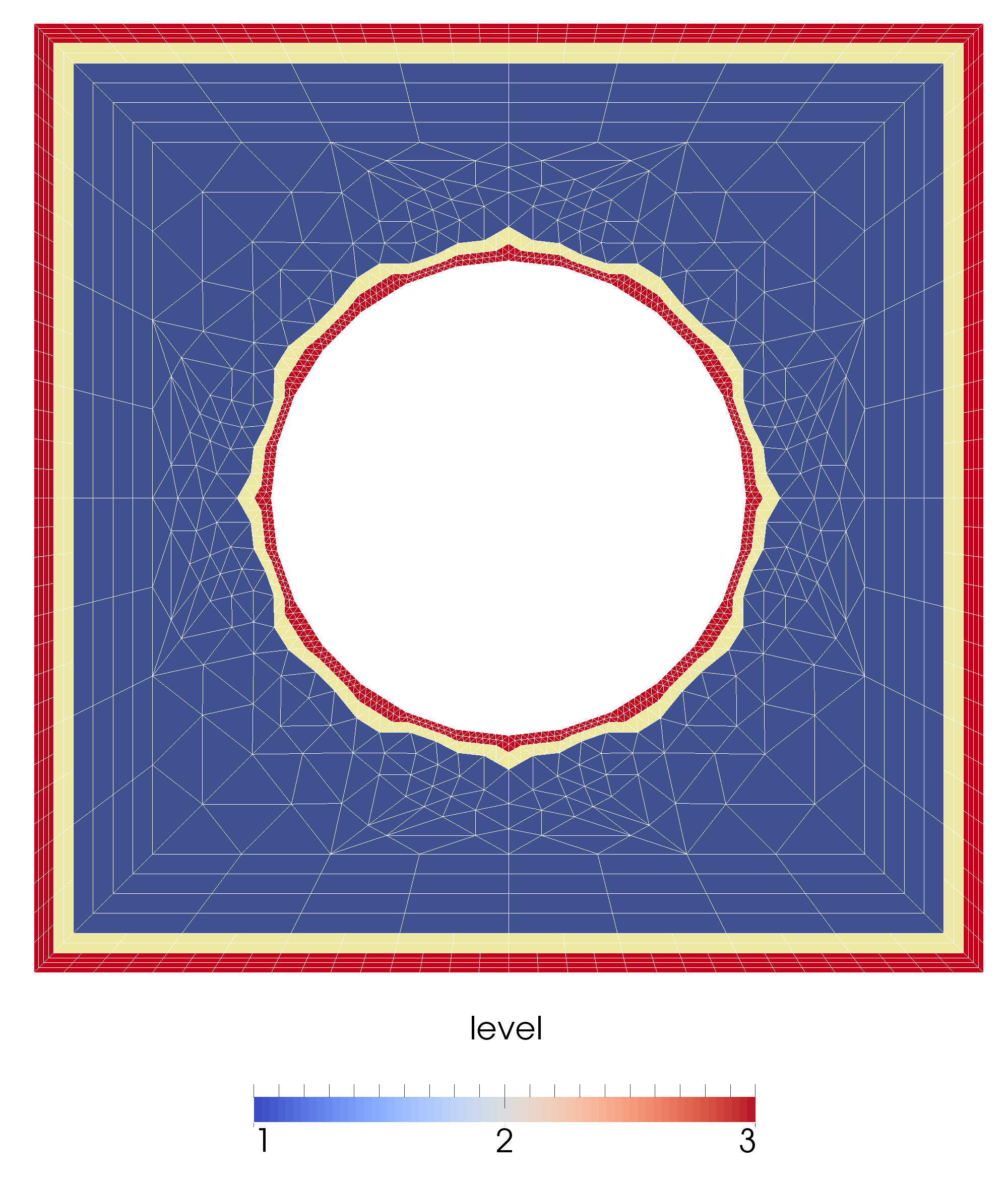}
\caption[Coarse and fine mesh]
{A mesh consists of two structures, the coarse mesh (left) that represents
the topology of the domain, and the forest mesh (right) that consists of the
leaf elements of a refinement and is used for computation.
In this example the domain is a unit square with a circular hole.
The color coding in the coarse mesh displays each tree's unique and consecutive
identifier, while the color coding in the forest mesh represents the refinement
level of each element.
In this example we choose an initial global level $1$ refinement and a
refinement of up to level $3$ along the domain boundary.}
\label{fig:hybridmesh}
\end{figure}

\begin{figure}[t]
\center
\includegraphics[width=\textwidth]{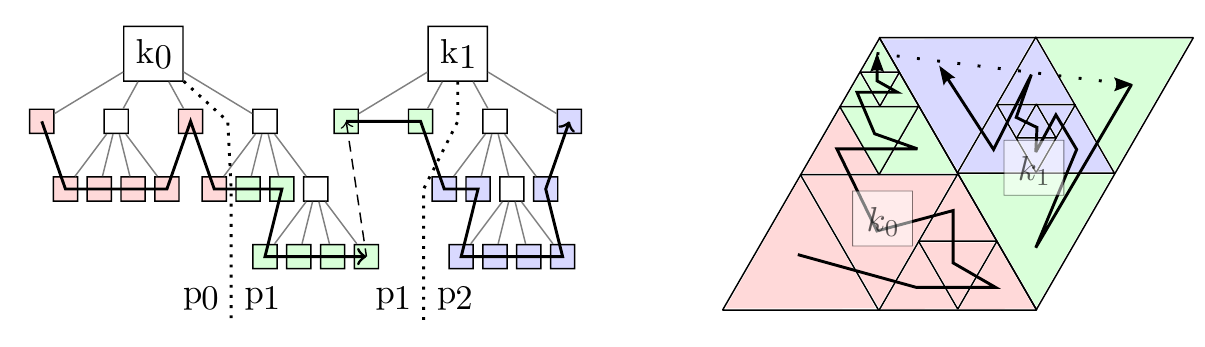}
\caption[Complex geometries with multiple trees]
{We connect multiple trees to model complex geometries.  Here, we
show two trees $k_0$ and $k_1$ with an adaptive refinement.
To enumerate the forest mesh, we establish an a priori order between the two
trees and use an SFC within each tree.
On the left-hand side of the figure the refinement tree and its linear storage
are shown. When we partition the forest mesh to $P$ processes (here, $P=3$),
we cut the SFC in $P$ equally sized parts and assign part $i$ to process $i$.}
\label{fig:twotreeforest}
\end{figure}

\begin{table}
\center
  \begin{tabular}{|r|l|}\hline
    Coarse mesh & Conforming mesh of tree roots\\
           Tree & An element of the coarse mesh \\
    Forest mesh & The adaptive mesh of elements (leaves of the trees) \\
   Element/leaf & Each element of the forest mesh is the leaf of a tree \\ \hline
  \end{tabular}
  \caption[Basic definitions]{The basic definitions for the coarse mesh and the forest mesh and
   their elements.  Throughout, we refer to the neighbor information of the
   trees as \emph{connectivity}.}\label{tab:basic}
\end{table}

\subsection{The tree shapes}

The trees of the coarse mesh can be of arbitrary shape as long as they 
are all of the same dimension and fit together along their faces.
In particular, we identify the following tree shapes:
\begin{itemize}
  \item Points in 0D.
  \item Lines in 1D.
  \item Quadrilaterals and triangles in 2D.
  \item Hexahedra and tetrahedra in 3D.
  \item Prisms and pyramids in 3D. %
\end{itemize}
Coarse meshes consisting solely of prisms or pyramids are quite uncommon; these\enlargethispage{1pc}
tree shapes are used primarily to transition between hexahedra and tetrahedra in
hybrid meshes.

\subsection{Encoding of face-neighbors}
\label{sec:faceneigh}

The connectivity information of a coarse mesh includes the neighbor
relation between adjacent trees.
Two trees are considered neighbors if they share at least one lower
dimensional face (vertex, face, or edge).
Since all of this connectivity information can be inferred from codimension-1
neighbors, we restrict ourselves to those, denoting them uniformly by
face-neighbors.
This choice does not lessen the generality of the partitioning algorithms to
follow and avoids a significant jump in complexity of the element-neighbor
code.

An application often requires a quick mechanism to access the face-neighbors of
a given forest mesh element.
If this neighbor element is a member of the same tree, the computation can be
carried out via the SFC logic, which involves only a few bitwise operations for
the hexahedral and tetrahedral Morton curves \cite{Morton66, SundarSampathBiros08,
BursteddeWilcoxGhattas11, BursteddeHolke16}.
If, however, the neighbor element belongs to a different tree, 
we need to identify this tree, given the parent tree of the original element
and the tree face at which we look for the neighbor element.
It is thus advantageous to store the face-neighbors of each tree in an array
that is ordered by the tree's faces.
To this end, we fix the enumeration of faces and vertices relative to each
other as depicted in Figure \ref{fig:vertexface}.
We discuss the computation of element face-neighbors in Chapter~\ref{ch:ghost}.

\begin{figure}[t]
   \center
    \includegraphics{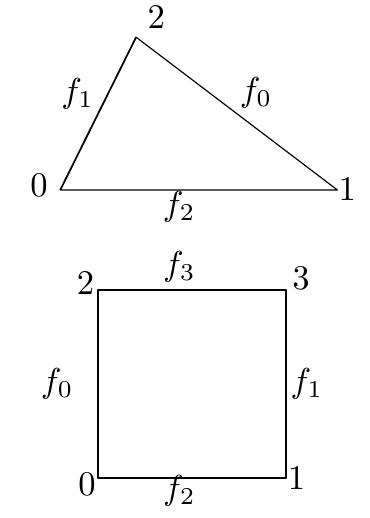}
    \includegraphics{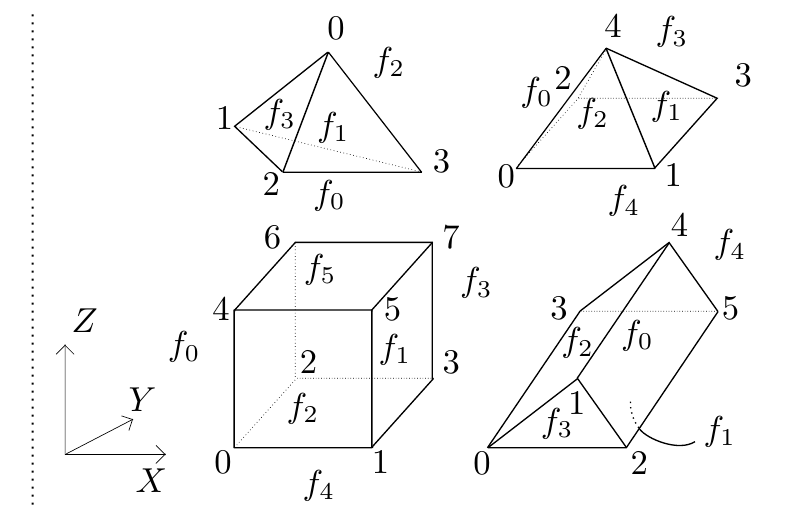}
   \caption[Vertex and face labels of the different tree types]
{The vertex and face labels of the $2$D (left) and $3$D (right) tree shapes.}%
   \label{fig:vertexface}%
\end{figure}%

\subsection{Orientation between face-neighbors}
\label{sec:orientation}

In addition to the global index of the neighbor tree across a face, we
describe how the faces of the tree and its neighbor are rotated relative to
each other.
We allow all connectivities that can be embedded in a compact 2- or 3-manifold
in such a way that each tree has positive volume.
This includes the Moebius strip and Klein's bottle and other quite exotic
meshes, e.g.,\ a hexahedron whose one face connects to another in some rotation.
We obtain two possible orientations of a line-to-line connection, three for a
tri\-angle-to-triangle connection, and four for a quadrilateral-to-quadrilateral connection.

We would like to encode the orientation of a face connection analogously to the
way it is handled in \pforest:
At first, given a face $f$, its vertices are a subset of the vertices of the
whole tree.
If we order them accordingly and renumerate them consecutively starting from
zero, we obtain a new number for each vertex that depends on the face $f$.
We call it the \emph{face corner number}.
If now two faces $f$ and $f'$ meet, the corner $0$ of the face with the
smaller face number is identified with a face corner $k$ in the other face. In
\pforest this $k$ is defined to be the orientation of the face connection.

In order for this operation to be well-defined, it must not depend on the
choice of the first face when the two face numbers are the same, which is
easily verified for a single tree shape.
When two trees of different shapes meet, we generalize to
determine which face is the first one.

\begin{definition}
 We impose a semiorder on the $3$-dimensional tree shapes as follows:
 \begin{equation}
 \begin{gathered}
 \xymatrix@C=4ex@R=1ex{ {Hexahedron} \ar@{}[rd]|-(0.6){\mathrel{\rotatebox{-20}{$<$}}} & \\
      & {Prism} \ar@{}[r]|-{<} & {Pyramid.} \\
      {Tetrahedron} \ar@{}[ru]|-{\mathrel{\rotatebox{20}{$<$}}} &
      }
 \end{gathered}
\end{equation}
\end{definition}

This comparison is sufficient since a hexahedron and a tetrahedron can never
share a common face.
We use it as follows.
\begin{definition}
\label{def:orientation}
 Let $t$ and $t'$ denote the tree shapes of two trees that meet at a common face
 with respective face numbers $f$ and $f'$.
 Furthermore, let $\xi$ be the face corner number of $f'$ matching corner $0$ of $f,$ and let
 $\xi'$ be the face corner number of $f$ matching corner $0$ of $f'$.
 We define the \textbf{orientation} of this face connection as
 \begin{equation}
  \texttt{or} \adef \left\lbrace
  \begin{array}{ll}
    \xi  &  \text{if $t < t'$ or ($t = t'$ and $f \leq f'$)}, \\
    \xi' &  \text{otherwise}.
  \end{array}
  \right.
 \end{equation}
\end{definition}

We now encode the face connection in the expression $\texttt{or} \cdot F + f'$
from the perspective of the first tree and $\texttt{or} \cdot F + f$ from the
second, where $F$ is the maximal number of faces over all tree shapes of this
dimension.

\section{Partitioning the coarse mesh}

As outlined above, tree-based AMR methods partition mesh elements with the help
of an SFC.
By cutting the SFC into as many equally sized parts as processes and assigning
part $i$ to process $i$, the repartitioning process is distributed and runs in
linear time.
Weighted partitions with a user-defined weight per leaf element are also
possible and practical \cite{PinarAykanat04, BursteddeWilcoxGhattas11}.

If the physical domain has a complex shape such that many trees
are required to optimally represent it, it becomes necessary to also
partition the coarse mesh in order to reduce the memory footprint.
This is even more important if the coarse mesh does not fit into the memory of
one process, since such problems are not even computable without coarse mesh
partitioning.

Suppose the forest mesh is partitioned among the processes.
Since the forest mesh frequently references connectivity information from the
coarse mesh, a process that owns a leaf $e$ of a tree $k$ also needs the
connectivity information of the tree $k$ to its neighbor trees.
Thus, we maintain information on these so-called ghost neighbor trees.

There are two traditional approaches to partition the coarse mesh.
In the first approach
\cite{BurriDednerKloefkornEtAl06, YilmazOezturanTosunEtAl10}, the coarse mesh
is partitioned, and the owner process of a tree will own all elements of that
tree.
In other words, it is not possible for two different processes to own elements
of the same tree, which can lead to highly imbalanced forest meshes.
Furthermore, if there are fewer trees than processes, there will be idle
processes assigned zero elements.
In particular, this approach prohibits the frequent special case of a single
tree.

The second approach \cite{Zumbusch03} is to first partition the forest mesh and
then deduce the coarse mesh partition from that of the forest.
If several
processes have leaf elements from the same tree, then the tree is assigned to
one of these processes, and whenever one of the other processes requests
information about this tree, communication is invoked.
This technique has the advantage that the forest mesh is load-balanced much
better, but it introduces additional synchronization points in the program and
can lead to critical bottlenecks if a lot of processes request information on
the same tree.

We propose another approach, which is a variation of the second, that overcomes
the communication issue.
If several processes have leaf elements from the same tree, we duplicate this
tree's connectivity data and store a local copy of it on each of the processes.
Thus, there is no further need for communication, and each process has exactly
the information it requires.
Since the purpose of the coarse mesh is not to store data that changes during
the simulation but to store connectivity data about the physical domain, the
data on each tree is persistent and does not change during the simulation.
Certainly, this concept poses an additional challenge in the (re)partitioning
process, because we need to manage multiple copies of trees without producing
redundant messages.

As an example, consider the situation in Figure \ref{fig:2trees3proc}.
Here the 2D coarse mesh consists of two triangles $0$ and $1$, and the forest
mesh is a uniform level 1 mesh consisting of 8 elements.  Elements $0,1,2,3$
belong to tree $0$ and elements $4,5,6,7$ to tree $1$.  If
we load-balance the forest mesh to three processes with ranks $0$, $1,$ and $2$,
then a possible forest mesh partition arising from an SFC could be\\
\begin{minipage}{0.35\textwidth}
\begin{equation}
  \begin{array}{cl}
    \mathrm{rank}&\mathrm{elements}\\\hline
    0 & 0,\, 1,\, 2\\
    1 & 3,\, 4,\, 5\\
    2 & 6,\, 7,
  \end{array}
\end{equation}
\end{minipage}
\begin{minipage}{0.3\textwidth}
leading to the coarse mesh partition
\end{minipage}
\hfill
\begin{minipage}{0.3\textwidth}
\begin{equation}
  \begin{array}{cl}
    \mathrm{rank}& \mathrm{trees} \\\hline
    0&0\\
    1&0,\,1\\
    2&1.
  \end{array}
\end{equation}
\end{minipage}\\[3ex]
Thus, each tree is stored on two processes.

\begin{figure}
   \center
\begin{minipage}{0.8\textwidth}
   \center
  \includegraphics{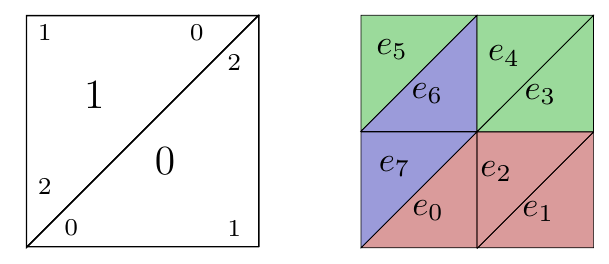}
\end{minipage}
   \caption[A partitioned coarse mesh with shared trees]
	   {A coarse mesh of two trees and a uniform level $1$ forest mesh.
            If the forest mesh is partitioned to three processes, each tree of the 
            coarse mesh is requested by two processes.
            The numbers in the tree corners denote the position and orientation of the tree vertices,
            and the global tree ids are the numbers in the center of each tree.
            As an example SFC we take the triangular Morton curve from
            {\cite{BursteddeHolke16}}, but the situation of multiple trees per
            process can occur for any SFC.}
   \label{fig:2trees3proc}
\end{figure}

\subsection{Valid partitions}

We allow arbitrary partitions for the forest, as long as they are induced by an
SFC.
This gives us some information on the type of coarse mesh partitions that we
can expect.

\begin{definition}
In general, a \textbf{partition} of a coarse mesh of $K$ trees $\indexset{K}$
to $P$ processes $\indexset{P}$ is a map $f$ that assigns each process a
certain subset of the trees,
\begin{equation}
  f\colon \indexset{P} \longrightarrow \mathcal{P}\indexset{K},
\end{equation}
and whose image covers the whole mesh:
\begin{equation}
\bigcup_{p=0}^{P-1} f(p) = \indexset{K}.
\end{equation}
Here, $\mathcal{P}$ denotes the set of all subsets (power set).
We call $f(p)$ the \emph{local trees} of process $p$ and explicitly allow that
$f(p)\cap f(q)$ may be nonempty. 
If so, the trees in this intersection are \emph{shared} between processes $p$ and $q$.
\end{definition}

The above definition includes all possibilities of attaching trees to processes.
In particular, we allow the same tree to reside on multiple processes.
For example, it is a partition in the above sense if all trees are partitioned
to all processes, expressed by $f(p) = \set{0,\ldots,K-1}$ for all $p$.
Any other arbitrary mapping is possible, with no restriction on the number of
trees per process, where empty processes with $f(p) = \emptyset$ are included.

For our method, we will not consider every possible partition of a coarse mesh.
Since we assume that a forest mesh partition comes from an SFC,
we
restrict ourselves to a subset of partitions.
In particular, the SFC order of a forest imposes the restriction that 
the order of trees is linear among the processes. Thus if tree $k$ is on
process $p,$ then every tree $l > k$ must be partitioned to a process $q\geq p$.
\begin{definition}
  \label{def:valid}
  Consider a partition $f$ of a coarse mesh with $K$ trees.
  We say that $f$ is a \textbf{valid partition} if there exist a forest
  mesh with $N$ leaves and a (possibly weighted) SFC partition of it that
  induces $f$.
  Thus, for each process $p$ and
  each tree $k,$ we have $k\in f(p)$ if and only if there exists a leaf $e$ of
  the tree $k$ in the forest mesh that is partitioned to process $p$.
  Processes without any trees are possible; in this case $f(p) =
  \emptyset$.

  We denote by ${k_p}$ the local tree on $p$ with the lowest global index and
  denote by ${K_p}$ the local tree with the highest global index.
\end{definition}

This, for example, excludes meshes where the processes are not mapped to the
trees in ascending order.
A simple example of a partition of two trees on two processes that is not valid
is given by $f(0) = \set{1}$ and $f(1) = \set{0}$.

The definition of valid partitions requires a forest mesh and a specific SFC-induced partition of it.
Since this is not convenient for theoretical investigations,
we deduce
three properties that characterize valid partitions independently of a forest
mesh and SFCs.

\begin{proposition}
  \label{prop:valid}
  A partition $f$ of a coarse mesh is valid if and only if it fulfills the following 
  properties.
\begin{enumerate}[{(i)}]
\item The tree indices of a process's local trees are consecutive; thus
  \begin{equation}
  f(p) = \set{{k_p},{k_p+1},\ldots,{K_p}}\,\text{or }\,f(p) =\emptyset.\label{eq:valid1}
 \end{equation}
\item A tree index of a process $p$ may be smaller than a
      tree index on a process $q$ only if $p\leq q$:
      \begin{equation}
        p \leq q  \Rightarrow  K_p \leq k_q \quad(\text{if}\,\,
        f(p)\neq\emptyset\neq f(q)).\label{eq:valid2} 
      \end{equation}
 \item The only trees that can be shared by process $p$ with other processes
   are $k_p$ and $K_p$:
   \begin{equation}
  f(p) \cap f(q) \subseteq \set{{k_p},{K_p}}\, \text{for}\,\, p\neq q.\label{eq:valid3}
   \end{equation}
\end{enumerate}
\begin{proof}
  We show the only-if direction first.
  Let an arbitrary forest mesh with SFC partition be given such that $f$ is
  induced by it. In the SFC order the leaves are sorted according to their SFC
  indices.
  If ${(i,I)}$ denotes the leaf corresponding to the $I$-th leaf in the $i$-th
  tree and tree $i$ has $N_i$ leaves,
  then the complete forest mesh consists of the leaves
  \begin{equation}
    \set{{(0, 0)},{(0, 1)},{(0, N_0-1)},{(1, 0)},\ldots,{(K-1, N_{K-1}-1)}}.
  \end{equation}
  The partition of the forest mesh is such that each process 
  $p$ gets a consecutive range
  \begin{equation}
    \set{{(k_p, i_p)},\ldots,{(K_p, i'_{p})}}
  \end{equation}
  of leaves,
  where $(k_{p+1}, i_{p+1})$ is the successor of $(K_p, i'_p)$
  and the $k_p$ and the $K_p$ form increasing sequences
  with $K_p\leq k_{p+1}$.
  The coarse mesh partition is then given by
  \begin{equation}
    f(p) = \set{k_p,k_p+1,\ldots,K_p}\,\, \text{for all } p,
  \end{equation}
  which shows properties (i) and (ii).
  To show (iii) we assume that $f(p)$ has at least three elements; thus
  $f(p) = \set{{k_p},{k_p+1},\ldots,{K_p}}$.
  However, this means that in the forest mesh partition each leaf
  of the trees $\set{k_p+1,\ldots,K_p-1}$ is partitioned to $p$.
  Since the forest mesh partitions are disjoint, no other process can hold leaf
  elements from these trees, and thus they cannot be shared.

  To show the if-direction, suppose the partition $f$ fulfills (i), (ii), and (iii).
  We construct a forest mesh with a weighted SFC partition as follows.
  Each tree that is local to a single process is not refined and thus contributes
  a single leaf element to the forest mesh.
  If a tree is shared by $m$ processes, then we refine it uniformly until we
  have more than $m$ elements.
  It is now straightforward to choose the weights of the elements
  such that the corresponding SFC partition induces $f$.
\end{proof}
\end{proposition}

We directly conclude the following.\begin{corollary}
\label{cor:valid1}
  In a valid partition, each pair of processes can share at most 
  one tree; thus
  \begin{equation}
      |f(p) \cap f(q)| \leq 1
  \end{equation}
  for each $p\neq q$.
\begin{proof} 
  Supposing the contrary, with \eqref{eq:valid3} we know that there would exist
  two processes $p$ and $q$ with $p<q$ such that $f(p) \cap f(q) \agl
  \set{{k_p},{K_p}}=\set{{k_q},{K_q}}$ and $k_p \neq K_p$.
  Thus $K_p > k_p = k_q$, which contradicts property~\eqref{eq:valid2}.
\end{proof}
\end{corollary}

\begin{corollary}
\label{cor:valid2}
  If in a valid partition $f$\;of a coarse mesh the tree $k$ is shared
  between processes $p$ and $q$, then for each $p < r < q,$
  \begin{equation}
    f(r) = \set{k} \quad\mbox{or}\quad f(r) = \emptyset.
  \end{equation}
\end{corollary}

\begin{proof}
We can directly deduce this from \eqref{eq:valid1}, \eqref{eq:valid2}, and
Corollary \ref{cor:valid1}.
\end{proof}

In order to properly deal with empty processes in our calculations, we
define start and end tree indices for these as well.
\begin{definition}
  \label{def:seforempty}
  Let $p$ be an empty process in a valid partition $f\!$; thus $f(p)=\emptyset$.
  Furthermore, let $q<p$ be maximal such that $f(q)\neq\emptyset$.
  Then we define the start and end indices of $p$ as
  \begin{subequations}
  \begin{align}
    k_p&:=K_q+1,\\
    K_p&:=K_q = k_p-1.
  \end{align}
  \end{subequations}
  If no such $q$ exists, then no rank lower than $p$ has local trees, and 
  we set $k_p=0$, $K_p=-1$.
  With these definitions,  \eqref{eq:valid1} and
  \eqref{eq:valid2} are valid if any of the processes are empty.
\end{definition}

From now on, all partitions in this chapter are assumed  valid even
if not stated explicitly.

\subsection{Encoding a valid partition}

A typical way to define a partition in a tree-based code is to store an array
\texttt{O} of tree offsets for each process, that is,
the global index of the first tree local to each process.
The range of local trees for process $p$ can then be computed as
$\set{{\texttt O[p]},\ldots,{\texttt O[p+1]-1}}$.
However, for valid partitions in the coarse mesh setting, this information would not
be sufficient because we would not know which trees are shared.
We thus modify the offset array by adding a negative sign when the first tree
of a process is shared.
\begin{definition}
  \label{def:offsetarray}
  Let $f$ be a valid partition of a coarse mesh, with $k_p$ being the 
  index of $p$'s first local tree.
  Then we store this partition in an array \texttt O of length $P+1$, where
  for $0 \leq p < P,$
  \begin{equation}
    \label{eq:offsetarray}
 \texttt O[p] := \left\lbrace \begin{array}{ll}
      k_p &
      \begin{array}{l}
        \textrm{if ${k_p}$ is not shared with the next smaller}\\
        \textrm{nonempty process or }f(p)=\emptyset,
      \end{array}
     \\[2ex]
               -k_p - 1 & \,\,\,\textrm{if it is.}
                   \end{array}\right.
  \end{equation}
  Furthermore, $\texttt O[P]$ shall store the total number of trees.
\end{definition}

Because of the definition of $k_p$, we know that $\texttt O[0] = 0$ for all
valid partitions.

\begin{lemma}
  \label{lem:rangefromoffset}
 Let $f$ be a valid partition, and let \texttt O be as in Definition {\ref{def:offsetarray}}.
 Then
  \begin{equation}
    \label{eq:getfirstfromoffset}
    k_p =  \begin{cases}
      \texttt O[p] &\textrm{if }\texttt O[p] \geq 0,\\[1ex]
   |\texttt O[p] + 1|\ & \textrm{if }\texttt O[p] < 0,
     \end{cases}
  \end{equation}
  and
  \begin{equation}
    \label{eq:getlastfromoffset}
    K_p = |\texttt O[p+1]| - 1.
  \end{equation}
\end{lemma}

  \begin{proof}
    The first statement follows since \eqref{eq:offsetarray} and \eqref{eq:getfirstfromoffset}
    are inverses of each other.
    For  \eqref{eq:getlastfromoffset} we distinguish two cases.
    First, let $f(p)$ be nonempty.
    If the last tree of $p$ is not shared with $p+1$, then
    it is ${k_{p+1}} - 1$ and $\texttt O[p+1] = k_{p+1}$, and thus we have 
    \begin{equation}
      K_p = k_{p+1}-1 = |\texttt O[p+1]| - 1.
    \end{equation}
    If the last tree of $p$ is shared with $p+1$, then it is
    $k_{p+1}$,
    the first local tree of $p+1,$ and thus $\texttt O[p+1] = -k_{p+1}-1$ and 
    \begin{equation}
      K_p= k_{p+1} = |-k_{p+1}| = |\texttt O[p+1]| - 1.  
    \end{equation}
    Now let $f(p)=\emptyset$.
    If $k_{p+1}$ is not shared, then $k_{p+1}=k_p=K_p+1$ by Definition \ref{def:seforempty},
    and $\texttt O[p+1]=k_{p+1}$ by \eqref{eq:offsetarray}.
    Thus,
    \begin{equation}
      K_p = k_p - 1 = k_{p+1}-1 = |\texttt O[p+1]| - 1.
    \end{equation}
    If $k_{p+1}$ is shared, then again by Definition \ref{def:seforempty},
    $k_{p+1}=k_p-1=K_p$ and $\texttt O[p+1] = -k_{p+1} - 1$ such that we 
    obtain
    \begin{equation}
      K_p = k_{p+1} = k_{p+1} + 1 - 1 = |\texttt O[p+1]| - 1.
    \end{equation}
  \end{proof}
\begin{corollary}
\label{col:numloctrees}
  In the setting of Lemma {\ref{lem:rangefromoffset}} the number $n_p$ of local
  trees of process $p$ fulfills
  \begin{equation}
    n_p = |\texttt O[p+1]| - k_p =
    \begin{cases}
      |\texttt O[p+1]| - \texttt O[p] & \textrm{if $\texttt O[p]\geq 0$,}\\
      |\texttt O[p+1]| - |\texttt O[p]+1| & \textrm{else.}
    \end{cases}
  \end{equation}
\end{corollary}

  \begin{proof}
    This follows from the identity $n_p = K_p - k_p + 1$.
  \end{proof}
  
Lemma \ref{lem:rangefromoffset} and Corollary \ref{col:numloctrees} show that
for valid partitions, the array \texttt O carries the same information as the
partition $f$.

\subsection{Ghost trees}
\label{sec:cmesh-ghosttrees}

A valid partition provides information on the local trees of a process.
These trees are all trees of which a forest has local elements.
In many applications it is common to collect a layer of ghost (or halo) elements
of the forest to support the exchange of data with neighboring processes.
Since these ghost elements may be descendants of nonlocal trees, we
store their trees as ghost trees.
We want to confine this logic to the coarse mesh to be independent of a forest
mesh, and thus we propose to store each nonlocal face-neighbor tree as a ghost
tree.
This means possibly storing more ghost trees than needed by a particular forest.
However, this only affects the first and the last local tree of a process,
which bounds the overhead.
Since we restrict the neighbor information to face-neighbors, we also restrict
ourselves to face-neighbor ghosts in this chapter. 
However, an extension to edge and vertex neighbor ghosts is planned for the
future. These will prompt a somewhat more elaborate discussion, since an
arbitrary number of trees can be neighbored across a vertex/edge.
There exist known algorithms
for quadrilaterals and hexahedra \cite{IsaacBursteddeWilcoxEtAl15}, which we
believe can be modified to extend to simplices, prisms, and pyramids.

\begin{definition}
  Let $f$ be a valid partition of a coarse mesh.
  A \textbf{ghost tree} of a process $p$ is any tree $k$ such that
  \begin{itemize}
    \item $k\notin f(p)$, and
    \item there exists a face-neighbor $k'$ of $k$ such that $k'\in f(p)$.
  \end{itemize}
\end{definition}

If a coarse mesh is partitioned according to $f$, then each process $p$ will
store its local trees and its ghost trees.

\subsection{Computing the communication pattern}

Suppose
a coarse mesh is partitioned among the
processes $\indexset P$ according to a partition $f$.
The input of the partition algorithm is this coarse mesh and a second partition
$f'$, and the output is a coarse mesh that is partitioned according to the
second partition.

Apart from these partitions being valid, no other restrictions are imposed on the partitions
$f$ and $f'$. Thus, we include the trivial case $f=f'$ as well as extreme
cases.  An example for such a case is an $f$ that concentrates all trees
on one process and an $f'$ that assigns them to another (or distributes them
evenly among all processes).
For dynamic forest repartitioning it is not unusual for almost all of the
elements to change their owner process
\cite{Mitchell07,BursteddeGhattasStadlerEtAl08}.
We expect similar behavior for the trees, especially when the number of trees is
on the order of the number of processes or higher.

We suppose that in addition to its local trees and ghost trees, each process knows the
complete partition tables $f$ and $f'$, for example, in the form of offset arrays.
The task is now for each process to identify the processes to which it needs to send 
local and ghost trees  and then to execute the sending.
A process also needs to identify the processes from which it receives local and ghost
trees  and to execute the receiving.
We discuss here how each process can compute this information from the
offset arrays without further communication.

It will become clear in section~\ref{sec:partghosknow} that ghost trees need
not yet be discussed at this point.
Thus, it is sufficient to concentrate on the local trees for the time being.

\subsubsection{Ownership during partition}

The fact that trees can be shared between multiple processes poses a challenge 
when repartitioning a coarse mesh.
Suppose we have a process $p$ and a tree $k$ with $k\in f'(p)$,
and $k$ is a local tree for more than one process in the partition $f$.
We do not want to send the tree multiple times, so how do we decide which
process sends $k$ to $p$?

A simple solution would be that the process with the smallest index to which $k$ is
a local tree  sends $k$.
This process is unique and can be determined without communication.
However, suppose that the two processes $p$ and $p-1$ share the tree $k$
in the old partition, and $p$ will also have this tree in the new partition.
Then $p-1$ would send the tree to $p$ even though this message would not be
needed.

We resolve this issue by only sending a local tree to a process $p$ if this
tree is not already local on $p$.
\begin{pardgm}
  \label{par:sendtrees}
  When repartitioning with $k \in f'(p)$, the process that sends $k$ to $p$ is
  \begin{itemize}
   \item $p$ if $k$ already is a local tree of $p$, or else
   \item $q$ with $q$ minimal such that $k\in f(q)$.
 \end{itemize}
\end{pardgm}

We acknowledge that sending from $p$ to $p$ in the first case is just a local
data movement not involving communication.

\begin{definition}
When repartitioning, given a process $p$ we define the sets $S_p$ and $R_p$ of
processes to and from which $p$ sends and receives local trees,  respectively, and thus
\begin{subequations}
\begin{align}
 S_p&:=\set{0\leq p' < P \abst\vert \textrm{$p$ sends local trees to $p'$}},\\
 R_p&:=\set{0\leq p' < P \abst\vert \textrm{$p$ receives local trees from $p'$}}.
\end{align}
\end{subequations}
Both sets may include the process $p$ itself.
Furthermore, we establish the notation for the smallest and largest ranks in
these sets, understanding that they depend on $p$:
\begin{subequations}
\begin{align}
  s_\mathrm{first}:= \min S_p, &\quad
  s_\mathrm{last}:= \max S_p,
  \\
  r_\mathrm{first}:= \min R_p, &\quad
  r_\mathrm{last}:= \max R_p.
\end{align}
\end{subequations}
\end{definition}

If $S_p$ is empty, we set $s_\mathrm{first} = -1$ and $s_\mathrm{last} = -2$,
and likewise for $R_p$.

$S_p$ and $R_p$ are uniquely determined by Paradigm \ref{par:sendtrees}. 

\subsubsection{An example}

We discuss a small example; see Figure~\ref{fig:partitionex}.  Here, we
repartition a partitioned coarse mesh of five trees among three processes.
The initial partition $f$ is given by
\begin{equation}
\label{eq:example_start}
\texttt{O}=\set{0, -2, 3, 5}
\end{equation}
and the new partition $f'$ by
\begin{equation}
\texttt{O'} = \set{0, -3, -4, 5}.
\end{equation}
Thus, initially tree $1$ is shared by processes $0$ and $1$, while in the new
partition, tree $2$ is shared by processes $0$ and $1,$ and tree $3$ is shared by processes
$2$ and $3$.
We arrange the local trees that each process will send to every other process
in a table, where the set in row $i,$ column $j$ is the set of local trees that
process $i$ sends to process $j$:
\begin{equation}
\begin{array}{c|ccc}
  & 0 & 1 & 2 \\ \hline
0 & \set{0,1} & \emptyset & \emptyset  \\
1 & \set{2} & \set{2} & \emptyset \\
2 & \emptyset & \set{3} & \set{3,4}
\end{array}
\end{equation}
This leads to the following sets $S_p$ and $R_p$:
\begin{subequations}
\label{eq:example_end}
\begin{align}
S_0 &= \set{0},   & R_0&=\set{0,1},\\
S_1 &= \set{0,1}, & R_1&=\set{1,2},\\
S_2 &= \set{1,2}, & R_2&=\set{2}.
\end{align}%
\end{subequations}%
We see that process $1$ keeps the tree $2$ that is also needed by process
$0$.  Thus, process $1$ sends tree $2$ to process $0$.  Process $0$ also needs
tree $1$, which is local on process $1$ in the old partition.
But, since it is also local to process $0$, process $1$ does not send it.
{%
\begin{figure}
\center
\includegraphics[width=0.6\textwidth]{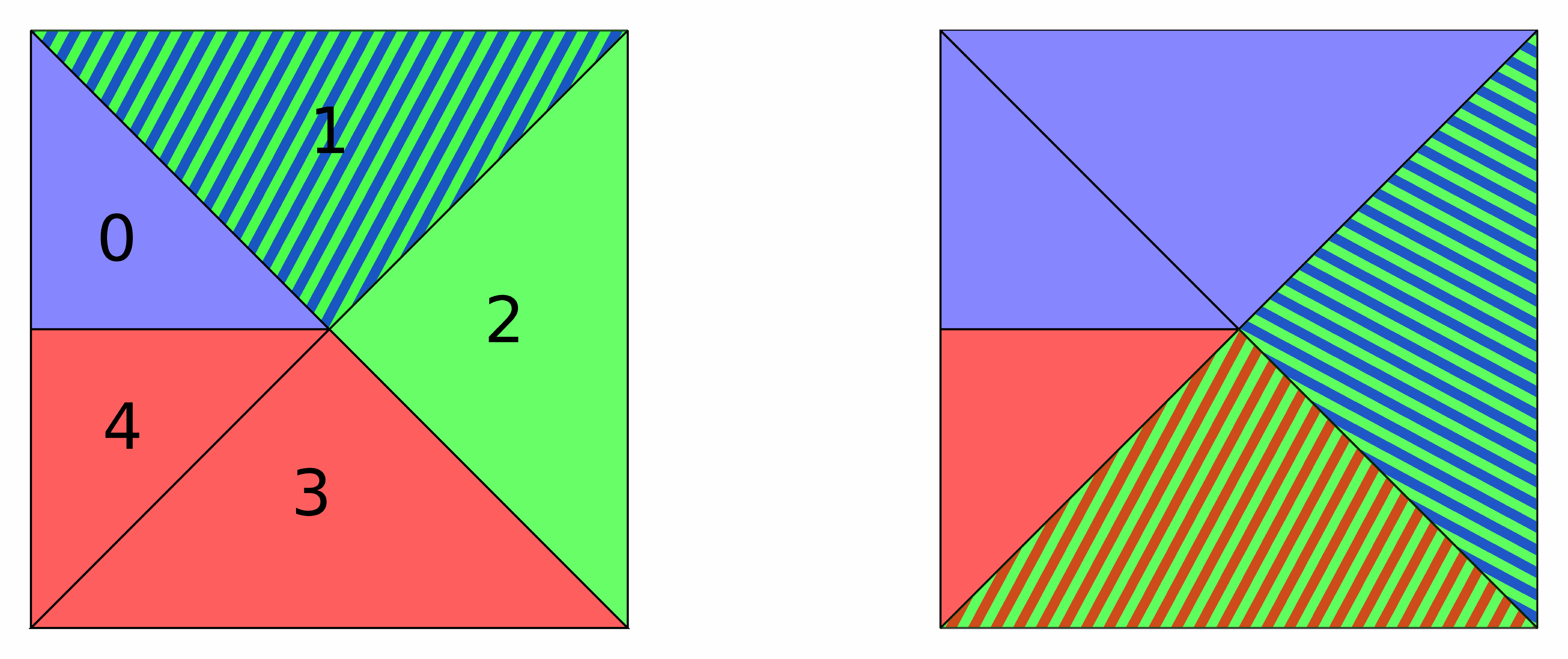}
\caption[An example for coarse mesh repartitioning]
{A small example.
The coarse mesh consists of five trees (numbered)
and is partitioned among three processes (color coded).
Left: the initial partition \texttt{O}.
Right: the new partition \texttt{O'}.
The colors of the trees encode the processes that have a tree as local tree.
Process $0$ is drawn in blue, process $1$ in green, and process $2$ in red.
Initially, tree $1$ is shared among processes $0$ and $1$, while in the new
partition, tree $2$ is shared among processes $0$ and $1,$ and tree $3$ is shared among
processes $1$ and $2$.
In
\eqref{eq:example_start}--\eqref{eq:example_end} we list the
sets \texttt{O} and \texttt{O'}, the trees that each process sends,
and the sets $S_p$ and $R_p$.}
\label{fig:partitionex}
\end{figure}
}%

\subsubsection{Determining \boldmath$S_p$ and $R_p$}
\label{sec:detSpRp}

In this section we show that each process can compute the sets $S_p$ and 
$R_p$ from the offset array without further communication.
\begin{proposition}
\label{prop:commpattern}
A process $p$ can calculate the sets $S_p$ and $R_p$ without further
communication from the offset arrays of the old and new partitions.
Once the first and last elements of each set are known, process $p$ can
determine in constant time whether any given rank is in any of those sets.
\end{proposition}

We split the proof into two parts.
First, we discuss how a process can compute $s_\mathrm{first}$,
$s_\mathrm{last}$, $r_\mathrm{first}$, and $r_\mathrm{last}$,
and then show how it can decide for two processes $\tilde p$ and $q$ whether $q\in
S_{\tilde p}$.
In particular, we will apply this decision to each of the process numbers
between the first and last elements of $S_p$ and $R_p$.

We begin by determining $s_\mathrm{first}$ and $s_\mathrm{last}$ for
$S_p \neq \emptyset$.
For $s_\mathrm{first}$ we consider two cases.
First, if the first local tree of $p$ is not shared with a smaller rank,
then $s_\mathrm{first}$ is the smallest process $q$ that has this tree in the new
partition and either is $p$ itself or did not have it in the old one.
We can find $q$ with a binary search in the offset array.

Second, if the first tree of $p$ is shared with a smaller rank, then $p$ only
sends it in the case that $p$ keeps this tree in the new partition.
Then $s_\mathrm{first} = p$.
Otherwise, we consider the second tree of $p$ and proceed with a binary search
as in the first case.

To compute $s_\mathrm{last}$, we notice that among all ranks that have
$p$'s old last tree in the new partition and did not already have it,
$s_\mathrm{last}$ is the largest
(except when $p$ itself is this largest rank, in which case it certainly had
the last tree).
We can determine this rank with a binary search as well.
If no such process exists, we proceed with the second-to-last tree of $p$, for which
we know that such a process must exist.

\begin{remark}
The special case $S_p=\emptyset$ occurs in the following situations:
\begin{enumerate}
\item $p$ does not have any local trees.
\item $p$ has one local tree that is shared with a smaller rank,
      and $p$ does not have this tree as a local tree in the new partition.
\item $p$ has two trees, the first of which case 2\ holds for.
      The second (last) tree is shared with a set $Q$ of bigger ranks, and
      there is no process $q\notin Q$ that has this tree as a local tree
      in the new partition.
\end{enumerate}
These conditions can be queried before computing $s_\mathrm{first}$ and
$s_\mathrm{last}$. To check condition 3, we need to perform one binary search,
while we evaluate conditions 1 and 2 in constant time.
\end{remark}

Similarly, to compute $R_p$, we first look at the smallest and largest elements
of this set.
These are the first and last processes from which  $p$ receives trees.
$r_\mathrm{first}$ is the smallest rank that had $p$'s new first tree as a
local tree in the old partition, or it is $p$ itself if this tree was also a
local tree on $p$.
Also $r_\mathrm{last}$ is the smallest rank greater than or equal to
$r_\mathrm{first}$ that had $p$'s new last local tree as a local tree in the
old partition, or $p$ itself.
We can find both of these with a binary search in the offset array of
the old partition.

\begin{remark}
$R_p$ is empty if and only if $p$ does not have any local trees in the new
partition.
\end{remark}

\begin{lemma}
\label{lem:Spconstant}
Given any two processes $\tilde p$ and $q$, the process $p$ can determine in
constant time whether $q\in S_{\tilde p}$.
Moreover, $p$ can determine for a given tree $k$ whether $\tilde p$ sends
$k$ to $q$.
In particular, this includes the cases $\tilde p = p$ and $q = p$.
\begin{proof}
 Let $\hat k_{\tilde p}$ be the first non-shared local tree of $\tilde p$ in the
 old partition.
 If such a tree does not exist, then $S_{\tilde p} = \emptyset$ or $S_{\tilde p}
=\set{\tilde p}$.
 Let $\hat K_{\tilde p}$ be the last local tree of $\tilde p$ in the old
 partition if it is not the first local tree of $q$ in the old partition,
 and let it be the second-to-last local tree otherwise.
 If such a second-to-last local tree does not exist, we conclude that $\tilde p$
 has only one tree in the old partition, and $q$  also has this tree in the
 old partition. Thus $q\notin S_{\tilde{p}}$.
 Furthermore, let $\hat k_q$ and $\hat K_q$ be the first and last local trees of
 $q$ in the new partition. %
 We add $1$ to $\hat k_q$ if $q$ sends its first local
 tree to itself, and this tree is also the new first local tree of $q$.
 We claim that $q\in S_{\tilde p}$ if and only if all of the four inequalities
 \begin{align}
  \hat k_{\tilde p} \leq \hat K_{\tilde p},\quad
  \hat k_{\tilde p} \leq \hat K_q,\quad
  \hat k_q \leq \hat K_{\tilde p},\quad\textrm{and}\quad
  \hat k_q\leq \hat K_q
 \end{align}
 hold.
 The only-if direction follows, since if $\hat k_{\tilde p}>\hat K_{\tilde p}$, then
$\tilde p$ does not have trees to send to $q$. If $\hat k_{\tilde p} > \hat K_q$, 
 then the last new tree on $q$ is smaller than the first old tree on $\tilde p$.
 If $\hat k_q > \hat K_{\tilde p}$, then the last tree that $\tilde p$ could send
 is smaller than the first new local tree of $p$. Also if $\hat k_q > \hat K_q$, then 
 $q$ does not receive any trees from other processes.
 Thus, $\tilde p$ cannot send trees to $q$ if any of the four conditions is not
 fulfilled.
 The if-direction follows, since if all four conditions are fulfilled, there
 exists at least one tree $k$ with
 \begin{align}
  \hat k_{\tilde p}\leq k \leq \hat K_{\tilde p} \quad\mathrm{and}\quad
  \hat k_q\leq k \leq K_q.
 \end{align}
 Any tree with this property is sent from $\tilde p$ to $q$.
 Process $p$ can compute the four values $\hat k_{\tilde p}, \hat K_{\tilde p},
 \hat k_q$, and $\hat K_q$
 from the partition offsets in constant time.
\end{proof}
\end{lemma}

\begin{remark}
Let $p$ be a process that is not empty in the new partition.
For symmetry reasons,
$R_p$ contains exactly those processes $\tilde p$ with $r_\mathrm{first}\leq
\tilde p \leq r_\mathrm{last}$ and $p\in S_{\tilde p}$.
\end{remark}

Thus, in order to compute $S_p$, we can compute $s_\mathrm{first}$ and
$s_\mathrm{last}$ and then check for each rank $q$ in between whether or not the conditions of
Lemma \ref{lem:Spconstant} are fulfilled with $\tilde p = p$.
For each process this check takes only constant run time.

Now, to compute $R_p$ we can compute $r_\mathrm{first}$ and $r_\mathrm{last}$ and then
check for each rank $q$ in between whether or not $p\in S_q$.

These considerations complete the proof of Proposition~\ref{prop:commpattern}.

\subsection{Face information for ghost trees}
\label{sec:partghosknow}

We identify the following five different types of possible face connections in a coarse
mesh:\\[1ex]
\begin{minipage}{0.4\textwidth}
\begin{enumerate}
  \item Local tree to local tree.
  \item Local tree to ghost tree.
  \item Ghost tree to local tree.
\end{enumerate}
\end{minipage}
\begin{minipage}{0.5\textwidth}
\begin{enumerate}
  \setcounter{enumi}{3}
  \item Ghost tree to ghost tree.
  \item Ghost tree to nonlocal and\\ non-ghost tree.
\end{enumerate}
\end{minipage}\\[1ex]
There are several possible approaches to which of these face connections of a local coarse
mesh we could actually store. As long as each face connection between any two
neighbor trees is stored at least once globally, the information of the coarse
mesh over all processes is complete, and a single process could reproduce all
five types of face connection at any time, possibly using communication.
Depending on which of these types we store, the pattern for sending and
receiving ghost trees during repartitioning changes.
Specifically, the tree that will become a ghost on the receiving process
may be either a local tree or a ghost on the sending process.

When we use the maximum possible information of all five types of connections, we
have the most data available and can minimize the communication required.
In particular, from the nonlocal neighbors of a ghost and
the partition table, a process can compute which other processes this ghost
is also a ghost of and of which it is a local tree.  With this information we can
ensure that a ghost is sent only once and only from a process that also sends
local trees to the receiving process.

The outline of the sending/receiving phase
then looks like the following:
\begin{enumerate}
 \item For each $q\in S_p$, send local trees that will be owned by $q$
   (following Paradigm \ref{par:sendtrees}).  
 \item Consider sending a neighbor of these trees to $q$ if
   it will be a ghost on $q$.
   Send one of these neighbors if either $p = q$ or both of the following
   conditions are fulfilled:
 \begin{itemize}
   \item $p$ is the smallest rank among those that consider sending this neighbor as a ghost, and
   \item $p\neq q$ and $q$ does not consider sending this neighbor as a ghost to itself.
 \end{itemize}
 \item For each $q\in R_p$, receive the new local trees and ghosts from $q$.
\end{enumerate}
In item 2\ a process needs to know, given a ghost that is considered for
sending to $q$, which other processes consider sending this ghost to $q$.
This can be calculated without further communication from the face-neighbor information
of the ghost.
Since we know for each ghost the global index of each of its
neighbors, we can check whether any of these neighbors is currently local on a
different process $\tilde p$ and will be sent to $q$ by $\tilde p$.  If so, we
know that $\tilde p$ considers sending this ghost to $q$.

Using this method, each local tree and ghost is sent only once to each receiver,
and only those processes send ghosts that send local trees anyway,
leading to minimal message numbers and message sizes.
Storing less information would increase either the number of communicating 
processes or the amount of data that is communicated.

Supposing we did not store the face connection type 5, for ghost trees
we would not have the information about to which nonlocal trees they connect.
With this face information we could use a communication pattern such that
each ghost is  received only once by a process $q$, by sending the new ghost
trees from a process that currently has it as a local tree (taking into account
Paradigm~\ref{par:sendtrees}).  However, that process might not
be an element of $R_q$, in which case additional processes would communicate.

If we stored only the local tree face information (types 1 and 2), then
we would have minimal control over the ghost face connections.
Nevertheless, we could define the partition algorithm by specifying that if a
process $p$ sends local trees to a process $q$, it will send all neighbors of
these local trees as potential ghosts to $q$.
The process $q$ is then responsible for deleting those trees that it received
more than once.
With this method the number of communicating processes would be the same but
the amount of data communicated would increase.

\begin{figure}
  \begin{center}
    \includegraphics{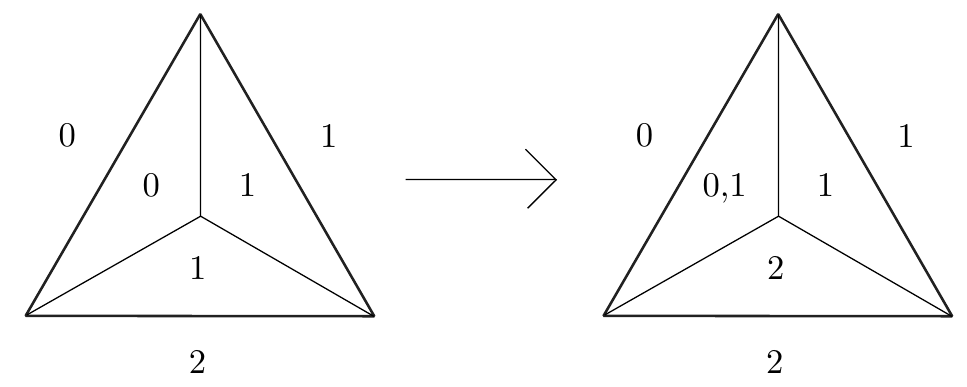}
    \\[4ex]
    \begin{tabular}{|c||c|c|c||c|c|c||c|c|c|}
      \hline
     &\multicolumn{3}{c||}{1, 2} &
      \multicolumn{3}{c||}{1, 2, 3, 4} &
      \multicolumn{3}{c|}{1, 2, 3, 4, 5} \\
      \hline
     p &0&1&2          & 0&1&2&            0&1&2\\ \hline 
     0&0(1,2)&0(2)&--- & 0&0&(0)          &0(1,2) & 0 & ---\\[0.8ex]
     1&---&1(2)&2(0,1) & (1,2) & 1(2) & 2(1)&--- & 1(2) & 2 (0,1)\\
      \hline
    \end{tabular}
  \end{center}
  \caption[The communication patterns for the different types of face information]
 {Repartitioning example of a coarse mesh showing the communication
  patterns controlled by the amount of face information available.
  Top: A coarse mesh of three trees is repartitioned. The numbers outside of the
  trees are their global indices. The numbers inside of each tree
  denote the processes that have this tree as a local tree.
  At first process $0$ has tree $0$, process $1$ has trees $1$ and $2$, and process $2$ has no
  local trees.  After repartitioning process $0$ has tree $0$, process $1$ has
  trees $0$ and $1$, and process $2$ has tree $2$ as local trees.
  Bottom: The table shows for each usage of face connection types which processes
  send which data.  The row of process $i$ shows in column $j$ which local
  trees $i$ sends to $j$, and---in parentheses---which ghosts it sends to $j$.
  Using face connection types $1$--$4$ we use more communication partners (process
  $0$ sends to process $2$ and process $1$ to process $0$) than with all five types.
  Using types $1$ and $2$ only, duplicate data is sent (process $0$ and
  process $1$ both send the ghost tree $2$ to process $1$).}
  \label{fig:comparelevels}
\end{figure}

In
Figure~\ref{fig:comparelevels}, we give an example comparing the three choices. 
To minimize the communication and overcome the need for postprocessing
steps, we propose to store all five types of face connection.

\section{Implementation}
\label{sec:cmesh-implementation}

Let us begin by outlining the data structures for trees, ghosts, and the
coarse mesh, and continue with a section on how to update the local tree and
ghost indices.
After this we present the partition algorithm to repartition a given coarse mesh
according to a precalculated partition array.
We emphasize that the coarse mesh data stores pure connectivity.
In particular, it does \emph{not} include the forest information, i.e.,\ leaf
elements and per-element payloads, which are managed by separate, existing
algorithms.

\subsection{The coarse mesh data structure}

Our data structure \texttt{cmesh} that describes a (partitioned) coarse mesh
has the following entries:
\begin{itemize}
 \item \texttt{O}: An array storing the current partition table; see
   Definition \ref{def:offsetarray}.
 \item $n_p$: The number of local trees on this process.
 \item $n_{\mathrm{ghosts}}$: The number of ghost trees on this process.
 \item \texttt{trees}: A structure storing the local trees in order of their global indices.
 \item \texttt{ghosts}: A structure storing the ghost trees in no particular order.
\end{itemize}

We use $64$-bit integers for global counts in \texttt O
and use $32$-bit signed integers for the local tree counts in \texttt{trees} and
\texttt{ghosts}.
This limits the number of trees per process to $2^{31}-1\cong 2\times 10^9$.
However, even with an overly optimistic memory usage of only $10$ bytes per tree,
storing that many trees would require about $18.6$ GB of memory per process.
Since on most distributed machines the memory per process is indeed much smaller,
choosing $32$-bit integers does not effectively limit the local number of 
trees.
In presently unimaginable cases, we could still switch to $64$-bit integers.

  We call the index of a local tree inside the \texttt{trees} array the \emph{local index}
  of this tree.
  Analogously,
  we call the index of a ghost in \texttt{ghosts} the \emph{local index} of that ghost.
On process $p$, we compute the global index $k$ of a tree in \texttt{trees}
from its local index $\ell$ and compute the global index $k_p$ of the first local tree
and vice versa, since
\begin{equation} 
  k = k_p + \ell
  .
\end{equation}
This allows us to address local trees with their local indices
using $32$-bit integers.

Each \texttt{tree} in the array \texttt{trees} stores the following data:
\begin{itemize}
  \item \texttt{eclass}:  The tree's shape as a small number (triangle,
    quadrilateral, etc.).
 \item \texttt{tree\_to\_tree}:  An array storing the local tree and ghost
   neighbors along this tree's faces. See section \ref{sec:faceneigh} and
  the text below.
 \item \texttt{tree\_to\_face}: An array encoding for each face the face-neighbor's
   face number and the orientation of the face connection.
   See section \ref{sec:orientation}.
 \item \texttt{tree\_data}: A pointer to additional data that we store with the
   tree, for example, geometry information or boundary conditions defined by an
   application.
\end{itemize}
The $i$-th entry of \texttt{tree\_to\_tree} encodes the tree number of the
face-neighbor at face $i$ using an integer $k$ with $0\leq k <
n_p+n_{\mathrm{ghosts}}$.  If $k<n_p$, the neighbor is the local tree with
local index $k$.  Otherwise, the neighbor is the ghost with local
index $k - n_p$.

We do not allow a face to be connected to itself. Instead, we use such a
connection in the face-neighbor array to indicate a domain boundary.
However, a tree can be connected to itself via two different faces.
This allows for one-tree periodicity, as say in a 2D torus consisting of a
single quadrilateral tree.

The \texttt{tree\_data} field can contain arbitrary data. An application can use these,
for example, to store higher-order geometry data per tree in order to account
for curved boundaries of the coarse mesh. Refined elements in the forest can then 
be snapped to the curved boundaries by evaluating the \texttt{tree\_data} 
field \cite{BursteddeGhattasGurnisEtAl10}.
\texttt{tree\_data} is partitioned to the processes together with the trees; thus possible duplicate
copies of it can exist.

Each \texttt{ghost} in the array \texttt{ghosts} stores the following data:
\begin{itemize}
 \item \texttt{Id}:  The ghost's global tree index.
 \item \texttt{eclass}:  The shape of the ghost tree.
 \item \texttt{tree\_to\_tree}:  An array giving for each face the global number
 of its face-neighbor.
 \item
   \texttt{tree\_to\_face}:
   As above.
\end{itemize}
Since a ghost stores the global number of all its face-neighbor trees,
we can locally compute all other processes that have this tree as a ghost by
combining the information from \texttt{O} and \texttt{tree\_to\_tree}.

\subsection{Updating local indices}

After partitioning, the local indices of the trees and ghosts change.
The new local indices of the local trees are determined by subtracting the global
index of the first local tree from the global index of each local tree.
The local indices of the ghosts are given by their positions in the 
data array. 

Since the local indices change after repartitioning, we update 
the
\texttt{tree\_to\_tree}
entries of the local trees to store those new values.
Because a neighbor of a tree can be either a local tree or a ghost on the
previous owning process $\tilde p$ and become either local or a ghost on the
new
owning process $p$, there are four cases that we shall consider.

We handle these four cases in two phases, the first phase being carried out on
process $\tilde p$ before the tree is sent to $p$. In this phase we change
all neighbor entries of the trees that become local.
The second phase executes on $p$ after the tree has been received from $\tilde p$.
At this point we change all neighbor entries belonging to trees that become
ghosts.

In the first phase,
$\tilde p$ has information about the first local tree on $\tilde p$ in the
old partition, its global number being $k_{\tilde p}$. Via
\texttt{O'} it also knows $k_p^\mathrm{new}$, the global index
of $p$'s first tree in the new partition.
Given a local tree on $\tilde p$ with local index $\tilde k$ in the old
partition, we compute its new local index $k$ on $p$ as
\begin{equation}
  \label{eq:idupdate1}
  k = k_{\tilde p} + \tilde k - k_p^\mathrm{new},
\end{equation}
which is its global index minus the global index of the new first local tree.
Given a ghost $g$ on $\tilde p$ that will be a local tree on $p$, we
compute its local tree number as
\begin{equation}
  \label{eq:idupdate2}
 k = g.\mathrm{Id} - k_p^\mathrm{new}.
\end{equation}

In the second phase, $p$ has received all its new trees and ghosts and
thus can give the new ghosts local indices to be stored in the \texttt{neighbors}
fields of the trees.
We do this by parsing its ghosts for each process $\tilde p \in R_p$ (in ascending order)
 and incrementing a counter.
For each ghost, we parse its neighbors for local trees, and for any of these
we set the appropriate value in its \texttt{neighbors} field.

Note that these four cases apply in the special case $\tilde p = p$ as well.

\subsection{\texttt{Partition\_cmesh}: Algorithm
\ref{alg:partgiven}}

The input is a partitioned coarse mesh $C$ and a new partition layout \texttt{O'},
and the output is a new coarse mesh $C'$ that carries the same information as
$C$ and is partitioned according to \texttt{O'}.

This algorithm follows the method described in section \ref{sec:partghosknow}
and is separated into two main phases, the \texttt{sending phase}
and the \texttt{receiving phase}.
In the former we iterate over each process $q\in S_p$ and decide which local trees
and ghosts we send to $q$.
Before sending, we carry out phase one of the update of the local tree numbers.
Subsequently, we receive all trees and ghosts from the processes in
$R_p$ and carry out phase two of the local index update.

In the \texttt{sending phase} we iterate over the trees that we send to $q$.
For each of these trees we check for each neighbor (local tree and ghost)
whether we send it to $q$ as a ghost tree. This is the second item in the list
of section \ref{sec:partghosknow}.
The function \texttt{Parse\_neighbors} decides for a given local tree or ghost
neighbor whether it is sent to $q$ as a ghost.

\begin{algorithm}
\SetVlineSkip{1pt} %
\DontPrintSemicolon
\caption{\texttt{Partition\_cmesh}(\texttt{cmesh} $C$, \texttt{partition}
\texttt{O'})}  
\label{alg:partgiven}
\algoresult{A \texttt{cmesh} $C'$ that consists of the same trees
as $C$ and is partitioned according to \texttt{O'}.}\;
  $p\gets$ this process 
  
  From \texttt{C.O} and \texttt{O'} determine $S_p$ and
  $R_p$. \Comment{see section \ref{sec:detSpRp}}

  \tcc*[l]{Sending phase}

 \algofor {each $q\in S_p$}
 {
 $G\gets\emptyset$  \Comment{trees $p$ sends as ghosts to $q$}
  
  $s \gets$ first local tree to send to $q$.
  
  $e \gets$ last local tree to send to $q$.
  
  $T \gets \set{C.\texttt{trees}[s],\ldots,C.\texttt{trees}[e]}$\Comment{local trees $p$ sends to $q$}

  \For {$k\in T$}
  {
  \texttt{Parse\_neighbors} ($C$, $k$, $q$, $G$, \texttt{O'}, $s$, $e$)
  }
  \texttt{update\_tree\_ids\_phase1} $(T\cup G)$\Comment{see equations \eqref{eq:idupdate1} and \eqref{eq:idupdate2}}

  Send $T\cup G$ to process $q$
  }

  \tcc*[l]{Receiving phase}
 \algofor {each $q\in R_p$}
 {
   Receive $T[q] \cup G[q]$ from process $q$
 }
 
 $C'.\texttt{trees}\gets \displaystyle\bigcup_{R_p} T[q]$
 \Comment{new array of local trees}
  
 $C'.\texttt{ghosts}\gets \displaystyle\bigcup_{R_p} G[q]$
 \Comment{new array of ghost trees}
  
 \texttt{update\_tree\_ids\_phase2} $(C'.\texttt{ghosts})$

 $C'.\texttt{O}\gets \texttt{O'}$

 \Return C'
   \vspace{1ex}
   \hrule
   \vspace{1ex}
 
 \tcc*[l]{decide which neighbors of $k$ to send as a ghost to $q$}
 \setcounter{AlgoLine}{0}
 \textbf{Function} \texttt{Parse\_neighbors}(\texttt{cmesh} $C$, \texttt{tree} $k$,
 \texttt{process} $q$, \texttt{ghosts} $G$, \texttt{partition} \texttt{O'},
 \texttt{tree\_indices} $s$, $e$)

   \algofor {$u\in k.\texttt{tree\_to\_tree}\ohne\set{s,\ldots,e}$
   }
   {
    \algoeif {$0\leq u<n_{p}
    \algoand \texttt{Send\_ghost}(C,\texttt{ghost}(u),q,\texttt{O'})$}
    {
    \If {$u+k_p\notin f'(q)$}
        {
          $G\gets G\cup \set{\texttt{ghost}(u)}$ \Comment{local tree $u$ becomes ghost of $q$}
        }
    }
    (\IfComment{$n_{p}\leq u$})
    {
     $g\gets C.\texttt{ghosts}[u-n_p]$
     
     \If {$g.Id\notin f(q) \algoand \texttt{Send\_ghost}(C,g,q,\texttt{O'})$}
     {
          $G\gets G\cup \set{g}$
          \Comment{$g$ is a ghost of $q$}
     }
    }
   }
   \vspace{1ex}
   \hrule
   \vspace{1ex}
  \tcc*[l]{Subroutine to decide whether to send a ghost or not}
 \setcounter{AlgoLine}{0}
 \textbf{Function} \texttt{Send\_ghost}(\texttt{cmesh} $C$, \texttt{ghost} $g$, \texttt{process} $q$,
 \texttt{partition} \texttt{O'})

 $S\gets\emptyset$
 
 \algofor {$u\in g.\texttt{tree\_to\_tree}$}
 {
   \algofor {$q'$ with $u$ is a local tree of $q'$}
   {
   \algoifcom{\IfComment{See Lemma~\ref{lem:Spconstant}}}{$q'$ sends $u$ to $q$} 
    {
      $S\gets S\cup \set{q'}$
    }
  }
 }
 
 \eIf {$q\notin S \algoand p = \min S$}
 {
  \Return true \Comment{$p$ is the smallest rank sending trees to $q$}
 }
 {
  \Return false
 }
\end{algorithm}

\section{Numerical results}
\label{sec:cmeshnumres}

The run time results that we present here have been obtained with version 0.2
of \tetcode\footnote{\url{https://github.com/cburstedde/t8code}}
using the JUQUEEN supercomputer at For\-schungs\-zentrum J{\"u}lich, Germany.
It is an IBM BlueGene/Q system with 28,672 nodes consisting of IBM PowerPC
A2 processors at 1.6~GHZ with 16~GB RAM per node \cite{Juqueen}.
Each compute node has 16 cores and is capable of running up to 64 MPI
processes using multithreading.

\subsection{How to obtain example meshes}
\label{sec:cmesh_example_forest}

To measure the performance and mem\-ory consumption of the algorithms presented
above, we would like to test the algorithms on coarse meshes
that are too big to fit into the memory of a single process,
which is 1 GB on JUQUEEN if we use 16 MPI ranks per node.
We consider the following three approaches to construct such meshes:
\begin{enumerate}
\item Use an external parallel mesh generator.
\item Use a serial mesh generator on a large-memory machine, transfer the
  coarse mesh to the parallel machine's file system, and read it using
  (parallel) file I/O.
\item
  Create a large coarse mesh by forming the disjoint union of smaller coarse
  meshes over the individual processes.
\end{enumerate}
Due to a lack of availability of parallel open source mesh generators,
we restrict ourselves to the second and third methods.
These have the advantage of being started with initial coarse meshes that fit
into a single process's memory such that we can work with serial mesh
generating software.
In particular, we use \texttt{gmsh}, \texttt{TetGen}, and \texttt{Triangle}
\cite{GeuzaineRemacle09, Shewchuk96, Si06}.

The third method is especially well suited for weak scaling studies, since
the individual small coarse meshes can be created programmatically and
communication-free on each process.
They may be of  the same size or different sizes among\ the processes.

We discuss two examples below. In the first we examine purely the coarse mesh
partitioning without regard for a forest and its elements (using hexahedral
meshes), and in the second
we drive the coarse mesh partitioning by a dynamically changing forest of
elements (using tetrahedral meshes).
The latter example produces shared trees and thus fully executes the
algorithmic ideas put forward above.

\subsection{Disjoint bricks}

In our first example we conduct both strong and weak scaling studies of coarse mesh
partitioning and test the maximal number of (hexahedral) trees that we can support
before running out of memory.
For our weak scaling results, we keep the same per-process
number of trees while increasing the total number of processes,
which we achieve by constructing
an $n_x \times n_y \times n_z$ brick of trees on each process
using three constant parameters $n_x, n_y$, and $n_z$.
We repartition this coarse mesh once, by the rule that each rank $p$ sends
43\% of its local trees to the rank $p+1$ (except the biggest rank $P-1$, which
keeps all its local trees).
We choose this odd percentage to create nontrivial boundaries between the
regions of trees to keep and trees to send.
See Figure \ref{fig:cmesh_partition1} for a depiction of the partitioned coarse
mesh on six processes.
The local bricks are created locally as \pforest connectivities with
\texttt{p4est\_connectivity\_new\_brick} and are then reinterpreted in parallel
as a distributed coarse mesh data structure.

\begin{figure}
\includegraphics[width=\textwidth]{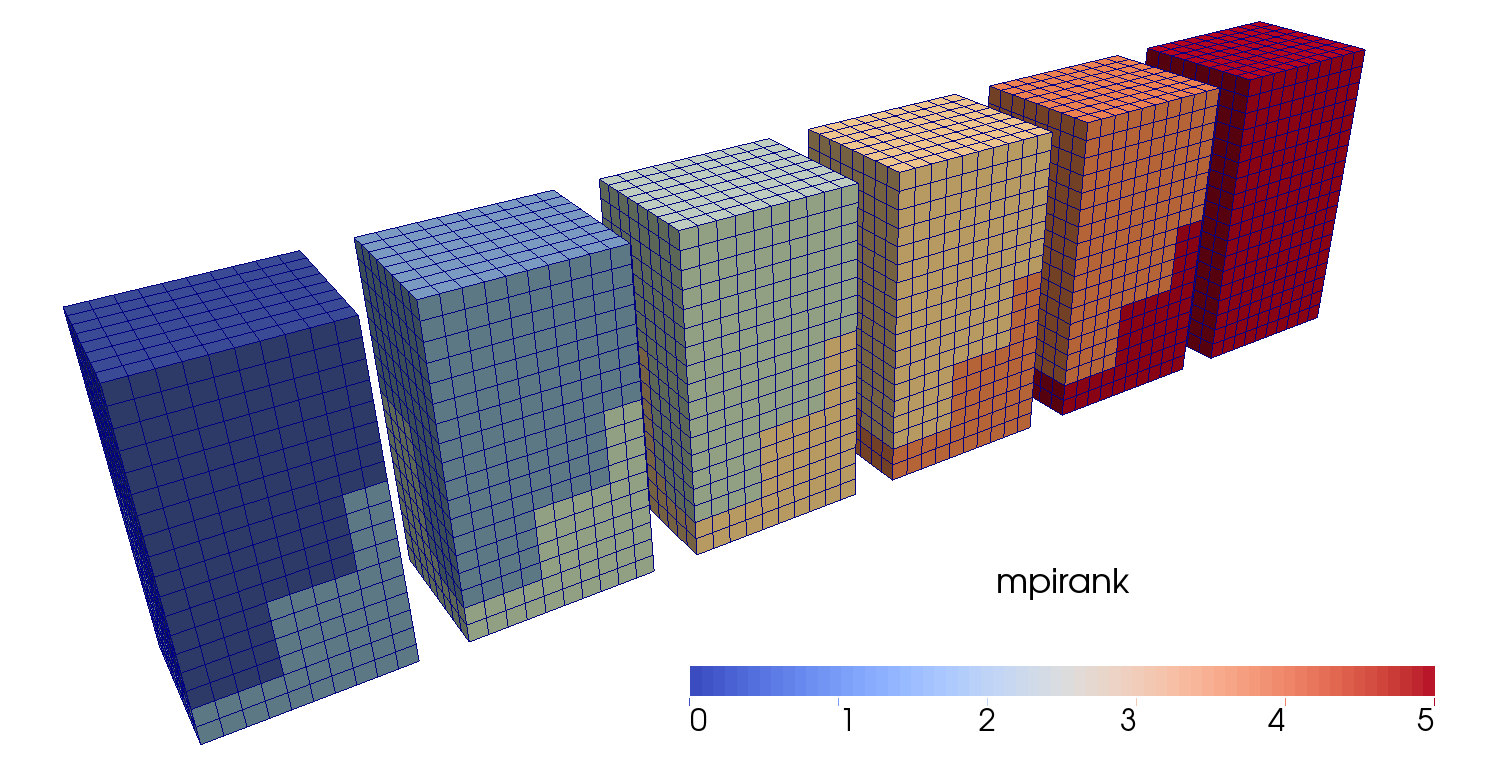}
\caption[The coarse mesh of disjoint bricks]
{The structure of the coarse mesh that we use to measure the maximum
possible mesh sizes and scalability for an example with six processes.
Before partitioning, the coarse mesh local to each process is created as one
$n_x\times n_y \times n_z$ block of hexahedral trees.
We repartition the mesh such that each process sends $43$\% of its local trees to
the next process.
The picture shows the resulting partitioned coarse mesh with parameters $n_x =
10, n_y = 18,$ and $n_z = 8$ and color coded MPI rank.}
\label{fig:cmesh_partition1}
\end{figure}

\begin{figure}
\center
\begin{minipage}{0.49\textwidth}
\includegraphics[width=\textwidth]{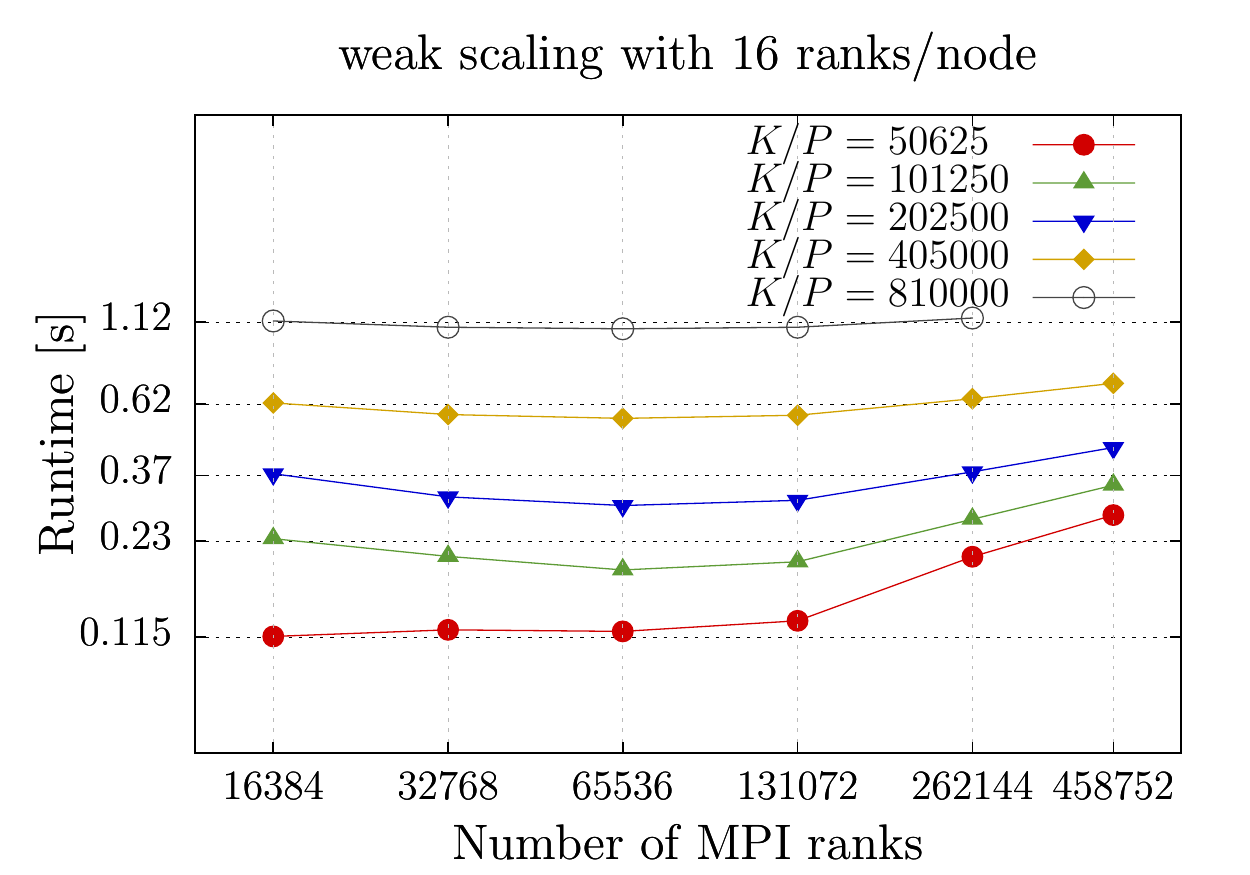}
\end{minipage}
\begin{minipage}{0.49\textwidth}
\includegraphics[width=\textwidth]{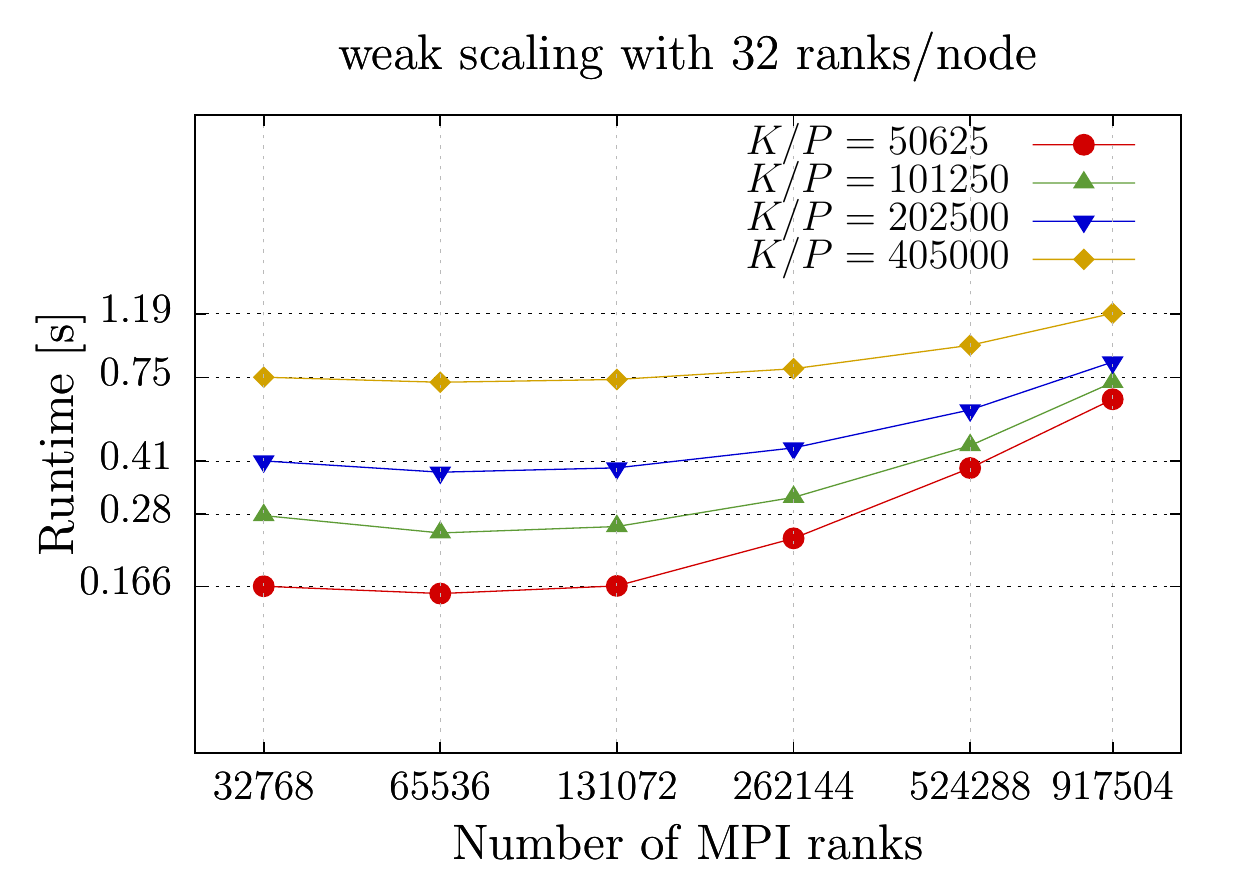}
\end{minipage}
\caption[Weak scaling of coarse mesh repartitioning]
{Weak scaling of \texttt{Partition\_cmesh} with disjoint bricks.
Left: $16$ ranks per node. Right: $32$ ranks per node.
We show the run times for the baseline on the y-axis and provide graphs for
different ratios between total coarse cells $K$ and MPI processes $P$.
On the left-hand side the time for the largest $458{,}752$ process run is
$0.72$ seconds; on the right-hand side the time
for the largest $917{,}504$ process run is $1.19$ seconds.
We obtain efficiencies of $0.62/0.72 = 86\%$ and $0.75/1.19 = 63\%$
compared to the baselines of $16{,}384/32{,}768$ MPI ranks, respectively (yellow lines).
The time for the $262,144$ process run with $810$e\/$3$ trees per process (black line)
increases from $1.12$ to $1.15$ seconds,
which translates into a weak scaling efficiency of $97.4$\%. (See online version for color.)}
\label{fig:cmesh_disjoint_scaling}
\end{figure}

\begin{figure}
\center
\begin{minipage}{0.49\textwidth}
\includegraphics[width=\textwidth]{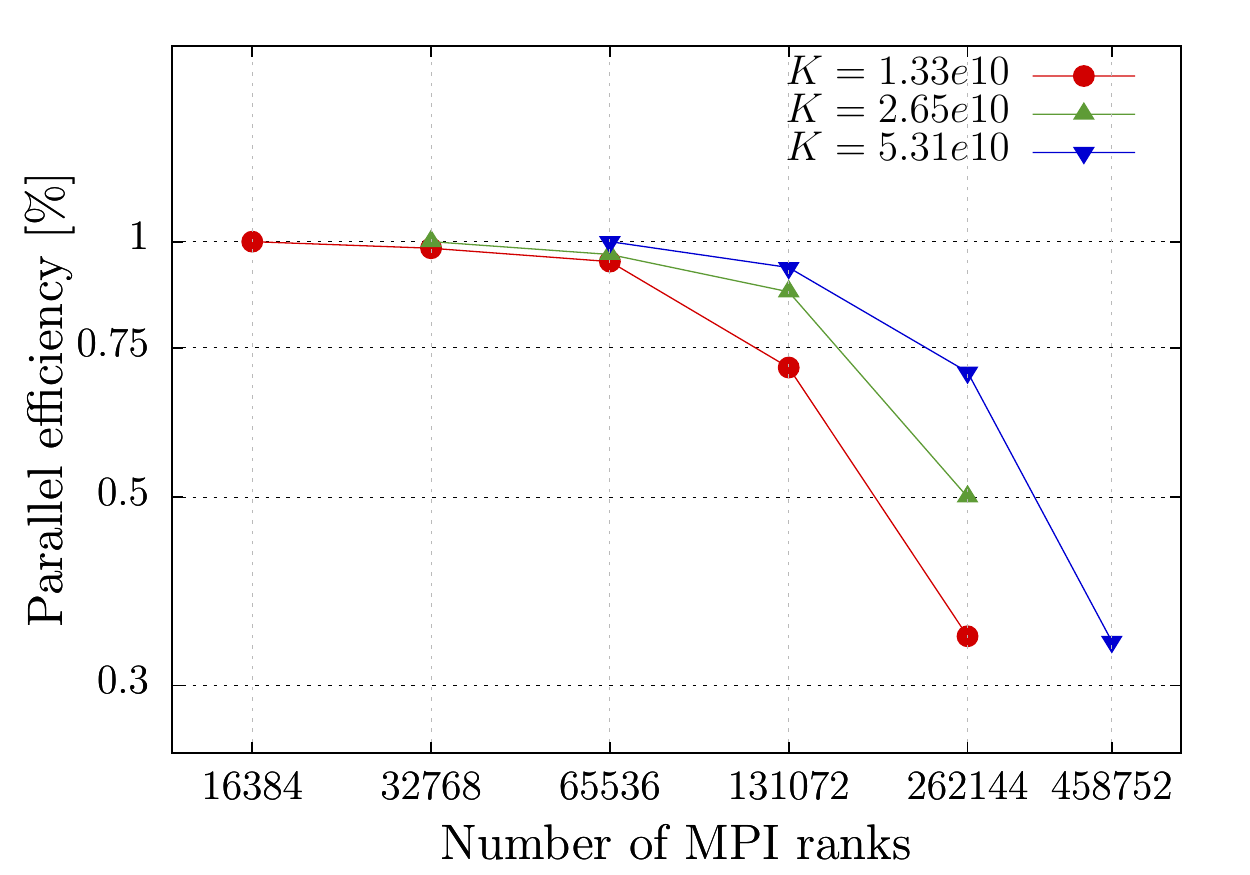}
\end{minipage}
\begin{minipage}{0.49\textwidth}
\includegraphics[width=\textwidth]{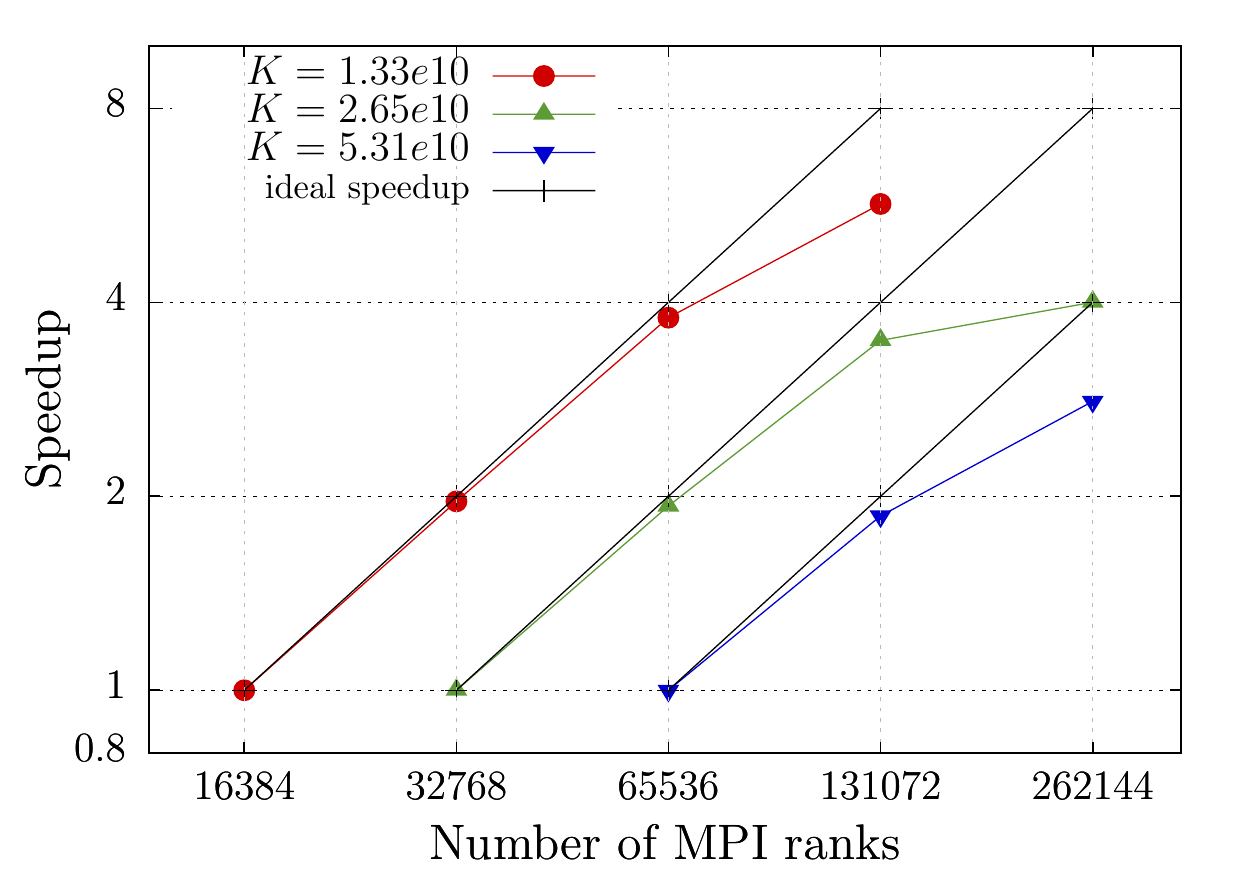}
\end{minipage}
\caption[Strong scaling of coarse mesh repartitioning]
{Strong scaling of \texttt{Partition\_cmesh} for the disjoint bricks
example on JUQUEEN with $16$ ranks per compute node,
for three runs with $13.3$e\/$9$, $26.5$e\/$9$, and $53.1$e\/$9$ trees.
We show the parallel efficiency on the left and the speedup on the right.
The absolute run times for $262{,}144$ processes are $0.21$, $0.27$, and $0.38$ seconds.}
\label{fig:cmesh_disjoint_scaling_strong}
\end{figure}

\begin{table}
\center
\begin{tabular}{|r|r|rr|r|r|}
 \multicolumn{6}{c}{Run time tests for \texttt{Partition\_cmesh}}\\[2ex] \hline
 \multicolumn{6}{|c|}{\mytabvspace 131,072 MPI ranks (16 ranks per node)}\\ \hline
  \mytabvspace Mesh size & Per rank & \multicolumn{2}{|c|}{Trees (ghosts) sent} & Time [s] & Factor\\ \hline
  \mytabvspace 6.635e9 &  50,625 & 21,767 &(3,414) &  0.13 & --\\
               13.27e9 & 101,250 & 43,536 &(5,504) &  0.20 & 1.53\\
               26.54e9 & 202,500 & 87,074 &(6,607) &  0.31  & 1.56\\
               53.08e9 & 405,000 & 174,149 & (11,381) & 0.57 & 1.85\\
               106.2e9 & 810,000 & 348,297 & (22,335) & 1.08  & 1.89\\
 \hline\hline
 \multicolumn{6}{|c|}{917,504 MPI ranks (32 ranks per node)}\\ \hline
  \mytabvspace Mesh size & Per rank & \multicolumn{2}{|c|}{Trees (ghosts) sent} & Time [s] & Factor \\ \hline
  \mytabvspace   46.45e9 & 50,625  & 21,768  & (3,413) & 0.64 & --\\
  \mytabvspace   92.90e9& 101,250  & 43,537  & (5,504) & 0.72 & 1.13 \\
  \mytabvspace  185.8e9 & 202,500  & 87,075  & (6,607) & 0.84 & 1.12 \\
  \mytabvspace  371.6e9 & 405,000  & 174,150 &(11,383) & 1.19 & 1.42 \\ \hline
\end{tabular}
\caption[Runtimes for \texttt{Partition\_cmesh} on JUQUEEN]
{The run times of \texttt{Partition\_cmesh} for $131{,}072$ processes with
$16$ processes per node (top) and $917{,}504$ processes with $32$ processes per node
(bottom). The largest coarse mesh that we created during the tests has $371$e\/$9$
trees.  In the middle column we list the average number of trees (ghosts)
that each process sends to another process.
The last column is the quotient of the current run time divided
by the previous run time. Since we double the mesh size in each step, we expect an increase in run time of a factor of $2$, which hints at parallel
overhead becoming negligible in the limit of many trees per process.%
}
\label{tab:131k917k}
\end{table}

We perform strong and weak scaling studies on up to 917,504 MPI ranks and display
our results in Figures~\ref{fig:cmesh_disjoint_scaling} and
\ref{fig:cmesh_disjoint_scaling_strong} and Table~\ref{tab:131k917k}.
We show the results of one study with 16 MPI ranks per compute node, thus 1 GB
available memory per process, and one with 32 MPI ranks per compute node,
leaving half of the memory per process.
In both cases we measure run times for different mesh sizes per process.
We observe that even for the biggest meshes of 405e3 and 810e3 trees
per process the absolute run times of partition are below 1.2 seconds.
Furthermore, we measure a weak scaling efficiency of 97.4\% for the 810e3
mesh on 262,144 processes and 86.2\% for the 405e3
mesh on 458,752 processes.
The biggest mesh that we created is partitioned between 917,504 processes and
uses 405e3 trees per process for a total of over 371e9 trees.

We notice a drop off in the scaling behavior when the number of trees per 
process is about 100e3 and smaller. 
At this stage the time for local computation is small relative to the time for
communication.
Since in many cases the number of trees per process will likely be even
smaller, we add Table~\ref{tab:smallmeshes}.
It documents running \texttt{Partition\_cmesh} with small coarse meshes, using a
number of trees on the order of the number of processes.
These tests show that
for such small meshes the run times are on the order of milliseconds.
Hence, for small meshes there is no disadvantage in using a partitioned coarse
mesh over a replicated one (i.e., each process holding a full copy).

\begin{table}
\center
\begin{tabular}{|r|r|r|}\hline
\#MPI ranks  & \#trees & Run time [s] \\\hline
1,024 &   4,096 & 0.00136\\ 
1,024 &   8,192 & 0.00149\\
1,024 &  16,384 & 0.00142\\
   64 &     105 & 0.00122\\ 
   32 &     105 & 0.00789 \\\hline
64 &  3,200 &  0.000293\\
64 & 19,200 &  0.000865\\\hline
\end{tabular}
\caption[Runtimes for \texttt{Partition\_cmesh} with small meshes]
{Run times for \texttt{Partition\_cmesh} for relatively small coarse
meshes. The bottom two rows of the table was not computed on JUQUEEN but on a
local institute cluster of $78$ nodes with $8$ Intel Xeon CPU E\/$5$-$2650$ v$2$ @ $2.60$GHz
each.}
\label{tab:smallmeshes}
\end{table}

\subsection{An example with a forest}

In this example we partition a tetrahedral coarse mesh according to a parallel
forest of fine elements.
While we pushed the maximum number of trees in the previous example,\enlargethispage{1pc}
we now consider mesh sizes that occur in more common usage scenarios.\newpage

When simulating shock waves or two-phase flows, there is often an interface along
which a finer mesh resolution is desired in order to minimize computational
errors.
Motivated by this example, we create the forest mesh as an initial uniform
refinement of the coarse mesh with a specified level $\ell$ and refine it in a
band along an interface defined by a plane in $\IR^3$ up to a maximum
refinement level $\ell+k$.
As the refinement rule we use 1:8 red refinement \cite{Bey92} together with the
tetrahedral Morton SFC \cite{BursteddeHolke16}. We move the 
interface through the domain with a constant velocity. Thus, in each time
step the mesh is refined and coarsened, and therefore we
repartition it to maintain an optimal load balance. We measure run times for
both coarse mesh and forest mesh partitioning for three time steps.

Our coarse mesh consists of tetrahedral trees modeling a brick with spherical
holes in it. To be more precise, the brick is built out of $n_x\times n_y
\times n_z$ tetrahedralized unit cubes, and each of those has one spherical hole
in it; see Figures \ref{fig:cmesh_partition2} and \ref{fig:forest_partition2}
for a small example mesh.

We create the mesh in serial using the generator \texttt{gmsh}
\cite{GeuzaineRemacle09}.
We read the whole file on a single process, and thus
use a local machine with 1 terabyte memory for preprocessing.
On this machine we partition the
coarse mesh to several hundred processes and write one file for each partition.
This data is then transferred to the supercomputer.
The actual computation consists of reading the coarse mesh from files, creating
the forest on it, and partitioning the forest and the coarse mesh
simultaneously.
To optimize memory while loading the coarse mesh, we open at most one
partition file per compute node.

The coarse mesh that we use in these tests has parameters $n_x = 26, n_y = 24,
n_z = 18$ and thus 11,232 unit cubes. Each cube is tetrahedralized with
about 34,150 tetrahedra, and the whole mesh consists of 383,559,464 trees.
In the first test, we create a forest of uniform level 1 and
maximal refinement level 2, and in the second, we create a forest of uniform level 2
and maximal refinement level 3. The forest mesh in the first test consists of
approximately 2.6e9 elements. In the second test, we also use a broader
band and obtain a forest mesh of 25e9 tetrahedra.

In Table \ref{tab:brickexample_cmesh} we show the run time results and further statistics for coarse mesh partitioning
and in Table \ref{tab:brickexample_forest} show results
for forest partitioning .
We observe that the run time for \texttt{Partition\_cmesh} is between $0.10$ and
$0.13$ seconds, and that about 88\% or 98\%, respectively, of all processes share
local trees with other processes.
The run times for forest partition are below $0.22$ seconds for the first example
and below $0.65$ seconds for the second example.

We also run a third test on 458,752 MPI ranks and refine the forest to a maximum
level of four. Here, the forest mesh has 167e9 tetrahedra, which we partition
in under 0.6 seconds. The coarse mesh partition routine runs in about 0.2
seconds. Approximately 60\% of the processes have a shared tree in the coarse
mesh. In Table 
\ref{tab:brickexample_forest_458k} we show the results from this test.

\begin{figure}
\center
\includegraphics[width=0.7\textwidth]{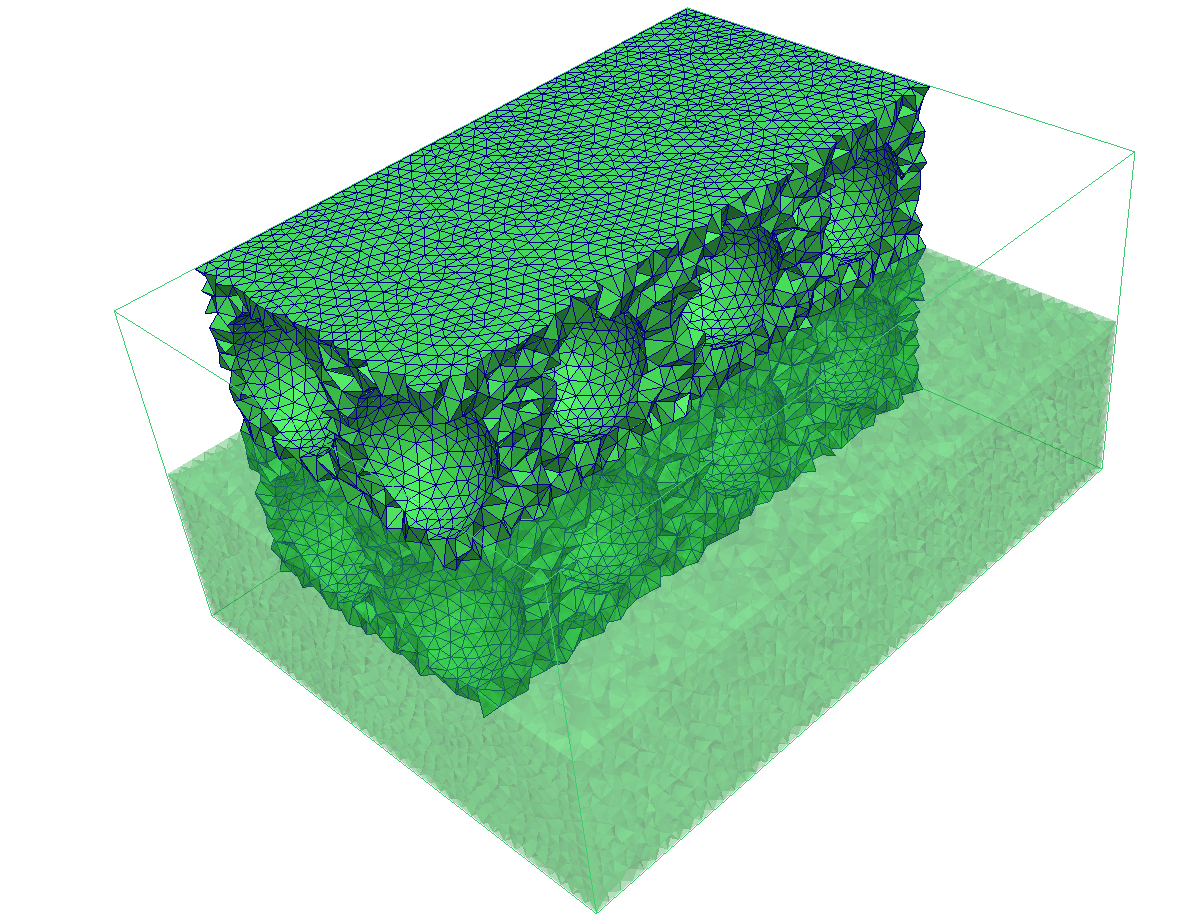}
\caption[A coarse mesh that models a brick with holes]
{The coarse mesh connectivity that we use for the
  partition tests motivated by an adapted forest.
It consists of $n_x\times n_y \times n_z$ cubes, with each cube having one spherical hole.
For this picture we use $n_x = 4,\, n_y = 3,\, n_z = 2$, and
each cube is triangulated with approximately $7{,}575$ tetrahedra.
For illustration purposes we show some parts of the mesh as opaque and other parts as
invisible.}
\label{fig:cmesh_partition2}
\end{figure}

\begin{figure}
\center
\includegraphics[width=0.3\textwidth]{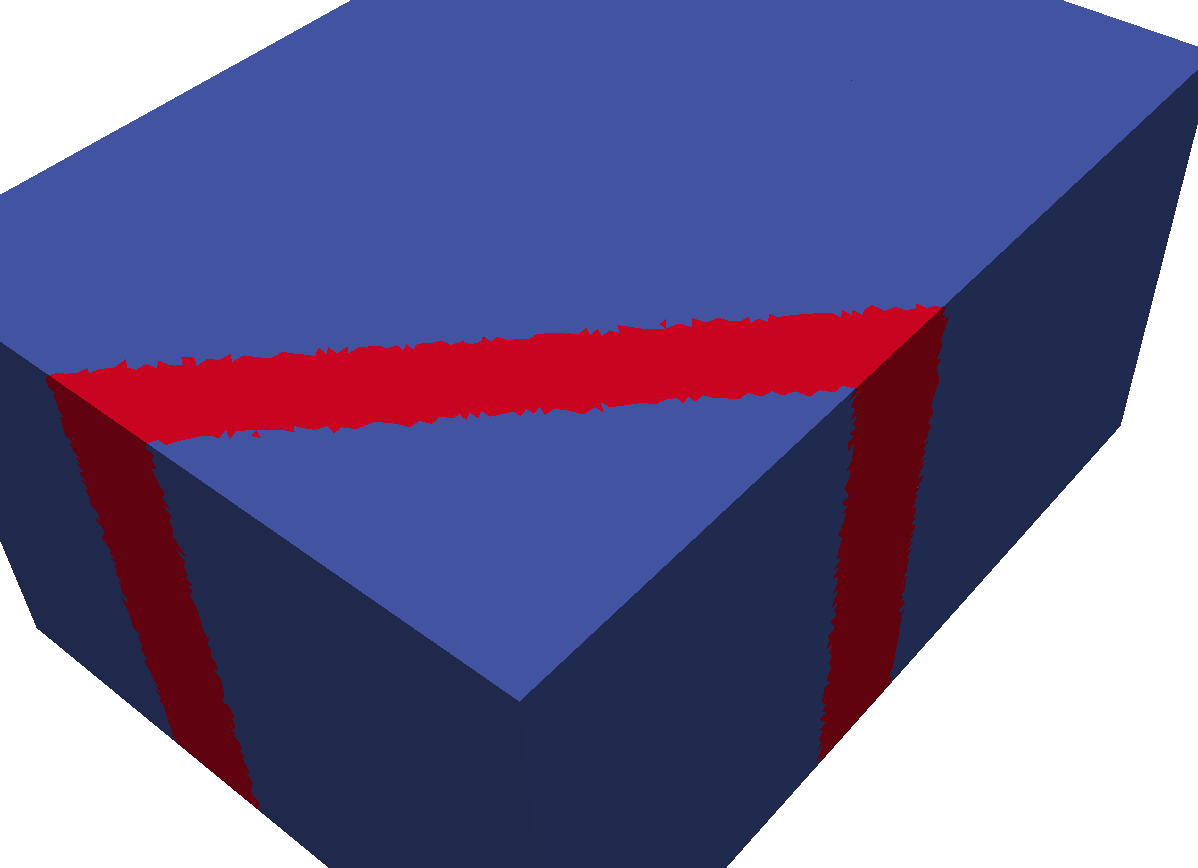}
\includegraphics[width=0.3\textwidth]{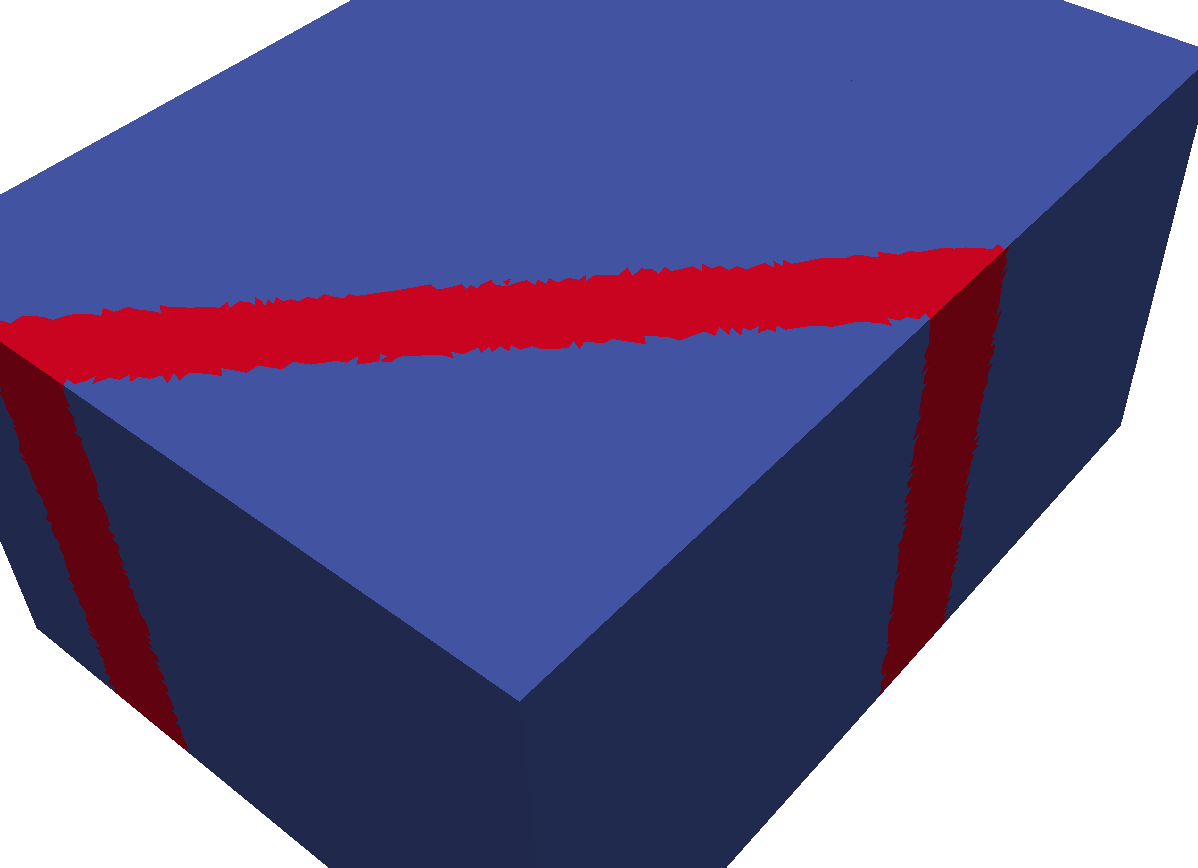}
\includegraphics[width=0.3\textwidth]{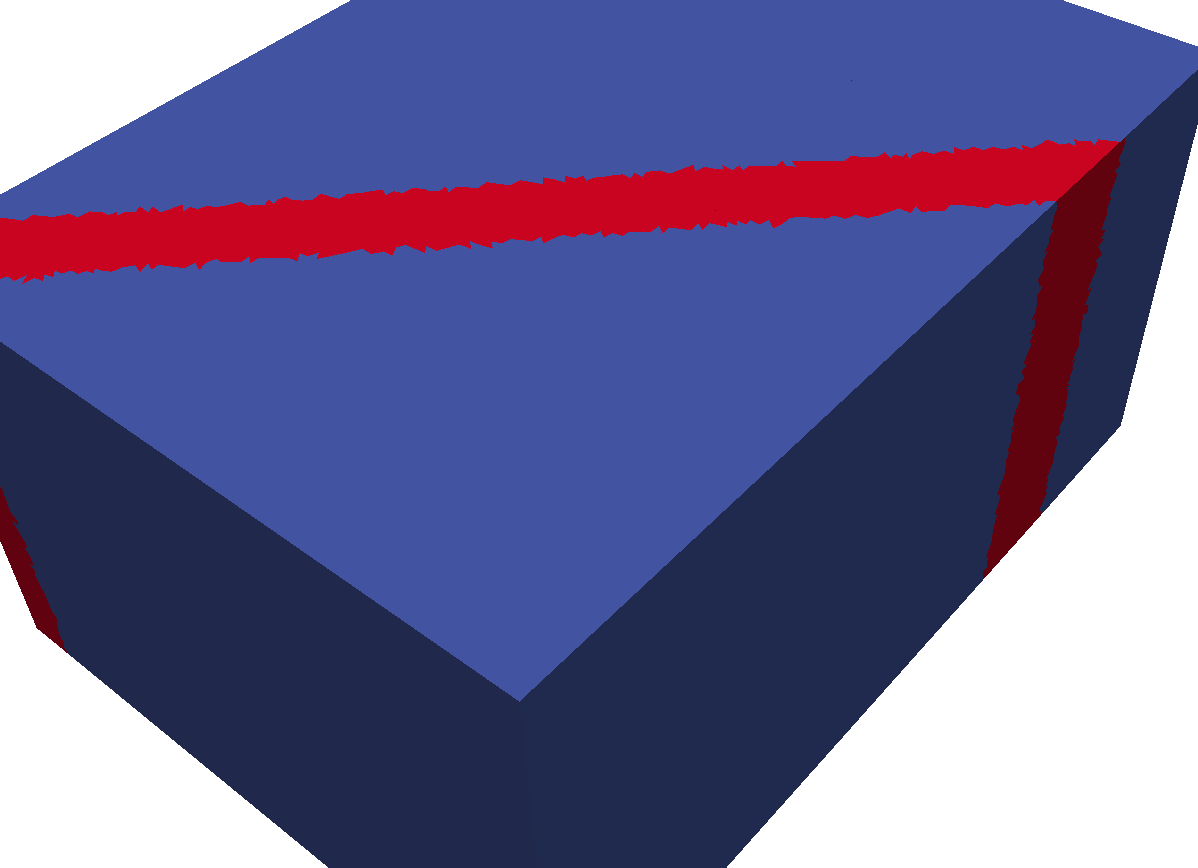}
\caption[The refined forest on the brick with holes mesh]
{An illustration of the band of finer forest mesh elements in the
example. The region of finer mesh elements moves through the mesh in each time
step. From left to right we see $t=1$, $t=2$, and $t=3$. In this illustration,
elements of refinement level $1$ are blue and elements of refinement level $2$ are
red.}
\label{fig:forest_partition2}
\end{figure}

\begin{table}
\center
\begin{tabular}{|c|r|r|r|r|r|}
\hline  
 t & Trees (ghosts) sent & Data sent [MiB] & $|S_p|$ & Shared trees & Run time [s]\\ \hline
1 & 25,117 (11,948) & 4.95 & 2.27 & 7,178 & 0.103 \\
2 & 34,860 (16,854) & 6.88 & 2.75 & 7,176 & 0.110 \\
3 & 36,386 (17,568) & 7.18 & 2.82 & 7,182 & 0.112 \\\hline
1 & 39,026 (18,334) & 7.67 & 2.97 & 8,096 & 0.128 \\
2 & 38,990 (18,268) & 7.66 & 2.95 & 8,085 & 0.129 \\
3 & 38,942 (18,074) & 7.64 & 2.93 & 8,085 & 0.128 \\
\hline
\end{tabular}
\caption[Runtimes of coarse mesh repartitioning with 8,192 MPI ranks]
{\texttt{Partition\_cmesh} with a coarse mesh of $324{,}766{,}336$
tetrahedral trees on $8{,}192$ MPI ranks.  We measure the duration of mesh
repartitioning for three time steps.
For each one, we show process average values of the
number of trees (ghosts) and the total number of bytes that each process sends to
other processes. The average number of other processes to which a process sends
($|S_p|$) is below three in each test.
We also provide the total number of shared
trees in the mesh, where $8{,}191$ is the maximum possible value.}
\label{tab:brickexample_cmesh}
\vspace{1pc}
\center
\begin{tabular}{|c|r|r|r|r|}
\hline  
t & Mesh size & Elements sent & Data sent [MiB]  & Run time [s] \\ \hline
1 &  2,622,283,453 & 203,858 & 3.49  & 0.215\\  %
2 &  2,623,842,241 & 281,254 & 4.82  & 0.215\\  %
3 &  2,626,216,984 & 293,387 & 5.03  & 0.214\\\hline  %
1 & 25,155,319,545 & 3,013,230 & 46.6 & 0.642\\ %
2 & 25,285,522,233 & 3,008,800 & 46.5 & 0.640\\ %
3 & 25,426,331,342 & 2,991,990 & 46.2 & 0.645\\ %
\hline
\end{tabular}
\caption[Runtimes of forest mesh repartitioning with 8,192 MPI ranks]
{Forest mesh partition on $8{,}192$ MPI ranks. For the same example as in Table {\ref{tab:brickexample_cmesh}} we display
statistics and run times for the forest mesh partition.  We show the total number
of tetrahedral elements and the average count of elements and bytes that each
process sends to other processes (their count is the same as in Table
{\ref{tab:brickexample_cmesh}}).}\label{tab:brickexample_forest}
\vspace{1pc}
\center
\begin{tabular}{|c|r|r|r|r|r|}\hline
t & Trees (ghosts) sent & Data sent [MiB] & $|S_p|$ & Shared trees & Run time [s]\\ \hline
1 & 704 (2,444) & 0.267 & 2.99 & 280,339 & 0.207 \\
2 & 707 (2,456) & 0.269 & 3.00 & 281,694 & 0.204 \\
3 & 708 (2,458) & 0.269 & 3.00 & 281,900 & 0.204 \\
\hline
\end{tabular}
\begin{tabular}{|c|r|r|r|r|}
\hline  
t & Mesh size        & Elements sent & Data sent [MiB]  & Run time [s] \\ \hline
1 & 167,625,595,829  & 362,863 & 5.55 & 0.522 \\  
2 & 167,709,936,554  & 364,778 & 5.58 & 0.578 \\  
3 & 167,841,392,949  & 365,322 & 5.59 & 0.567 \\\hline  
\end{tabular}
\caption[Runtimes of coarse and forest mesh repartitioning with 458,752 MPI ranks]
{Coarse and forest mesh partitions on $458{,}752$ MPI ranks. 
Run times for coarse mesh (top) and forest mesh (bottom) partition for the brick with holes
on $458{,}752$ MPI ranks. The setting and the coarse mesh are the same as in 
Table {\ref{tab:brickexample_cmesh}} except that for the forest we use an initial
uniform level three refinement with a maximum level of four.}
\label{tab:brickexample_forest_458k}
\end{table}

\section{Conclusion}

In this chapter we propose an algorithm that executes dynamic and in-core
coarse mesh partitioning.
In the context of tree-based adaptive mesh refinement (AMR), the coarse mesh
defines the connectivity of tree roots, which is used in all neighbor query
operations between elements.
This development is motivated by simulation problems on complex domains that
require large input meshes.
Without partitioning of the tree meta data, we will run out of memory around
one million trees, and with static or out-of-core partitioning, we might not
have the flexibility to transfer the tree meta data as required by the change
in process ownership of the trees' elements, which occurs in every AMR cycle.
With the approach presented here, this can be performed with run times that are
significantly smaller than those for partitioning the elements, even
considering that SFC methods for the latter are exceptionally fast in absolute
terms.
Thus, we add little to the run time of all AMR operations combined.

Our algorithm guarantees that
each process can provide the tree meta data for each of its fine
mesh elements that are themselves distributed using a SFC. 
We handle the communication without handshaking and develop a communication
pattern that minimizes data movement.
This pattern is calculated by each process individually, reusing information
that is already present.

Our implementation scales up to 917e3 MPI processes and up to 810e3 trees per
process, where the largest test case consists of 371e9 trees.
What remains to be done is extending the partitioning of ghost trees to edge
and corner neighbors, since only face-neighbor ghost trees are presently
handled.
It appears that the structure of the algorithm will allow this with little
modification.

 \chapter{Ghost}
\label{ch:ghost}
For many algorithms we need to know information about all neighbors of
a forest mesh element. 
A typical example is if an application needs to calculate the fluxes in a
finite volume solver or to compute integrals in a finite element setting;
see for example 
\cite{Braess97, WeinzierlMehl11, RasquinSmithChitaleEtAl14} and
our own discussion in Chapter~\ref{ch:app}.
Other examples include a refinement criterion that takes the size of
neighboring elements into account or the gradient of an approximated 
function.

If the forest mesh is partitioned among multiple processes, then a neighbor
element of a leaf element owned by process $p$ may be owned by a different process
$q\neq p$. We call this neighbor a \texttt{ghost} element of $p$.

In this section we describe how we create a layer of ghost elements
for a partitioned forest. Thus, each process obtains information about all
of its ghost elements.

As before, we restrict the algorithm to face-neighbors.

\begin{definition} %
 A \textbf{ghost element} (or \textbf{ghost} for short) of a process
$p$ in a forest $\forest F$ is a leaf element $G$ of a process
$q\neq p$, such that there exists a face-neighbor $E$ of $G$ that is a
local leaf element of $p$.
\end{definition}

\begin{definition} %
 We call a local element $E$ \textbf{boundary element} if it has at least
 one face-neighbor that is a ghost element.
 The \textbf{remote processes} of $E$ are all processes $q\neq p$ that own
 ghost elements of $E$.
 The union of all remote processes of all local elements of $p$ are the
 remote processes of $p$.
\end{definition}

\begin{definition} %
 By $R_p^q$ we denote the set of boundary elements of process $p$ that
 have process $q$ as a remote process.
\end{definition}

Throughout this section we assume that we can access the coarse mesh
information of each neighbor tree of a local tree. We can do this, since the
coarse mesh is either replicated or stores a layer of ghost trees; see
Section~\ref{sec:cmesh-ghosttrees}.

In the following, we need the definition of a $(2:1)$-balanced forest.
\begin{definition}
  \label{def:balance}
We call a forest $\forest F$ \textbf{balanced} if each pair $E, E'$ of
face-neighboring leaf elements of $\forest F$ satisfies
  \begin{equation}
    \ell(E) - 1 \leq \ell(E') \leq \ell(E) + 1.
  \end{equation}
Here, $E$ and $E'$ may belong to different processes. Thus, any two
face-neighbors differ by at most one in their refinement levels.
If the condition is not fulfilled, we say that $\forest F$ is unbalanced.
\end{definition}
\begin{remark}
  In some publications \emph{graded} is used instead of \emph{balanced}
  \cite{CohenKaberMuellerEtAl03, MuellerStiriba07}.
\end{remark}

We describe a basic version of a \texttt{Balance} algorithm that transforms an
unbalanced forest into a balanced one by successively refining elements in
Chapter~\ref{ch:balance}.

In this chapter we discuss three steps of implementing \texttt{Ghost}, which we
denote by \ghostb, \texttt{Ghost\_v2}, and \texttt{Ghost\_v3}.  \ghostb is a
relatively straight-forward version that only works with balanced forests.
\texttt{Ghost\_v2} is a more sophisticated version that works on arbitrary
forests, and we optimize its runtime in \texttt{Ghost\_v3}.
For the first two versions we orient ourselves to the \pforest implementations for
quadrilateral/hexahedral meshes
\cite{BursteddeWilcoxGhattas11,IsaacBursteddeWilcoxEtAl15}. However, we discuss
them in our new, element-type independent framework, which not only extends to
triangular/tetrahedral
meshes, but also to hybrid meshes consisting of different element types.
To this end, we outsource all operations that require specific knowledge of the
element type as low-level functions (c.f.\ Section~\ref{sec:highlow}).
Changing the element type is equivalent to changing the low-level
implementation. In this chapter, we discuss implementations for line,
quadrilateral, and hexahedral elements with the Morton index, as well as for
triangular and tetrahedral elements with the TM-index. This implicitly gives
us an implementation of prism elements, since we can model these as the cross
product of a line and a triangle; see~\cite{Knapp17}.

The basic idea of \ghostb and \texttt{Ghost\_v2} is to first identify all boundary
elements and their remote processes, thus building the sets $R_p^q$ and
identifying the non-empty ones. In a second step, each process $p$ sends all
elements in $R_p^q$ to $q$.

In the first step we iterate over all local leaves and for each over
all of its faces. We then have to decide for each face $F$ of a leaf $E$ which
processes own leaves that touch this face. The difference between \ghostb and
\texttt{Ghost\_v2} lies in this decision process.

A new improvement, \texttt{Ghost\_v3}, replaces the iteration over all leaves
with a top-down search. With this approach, we exclude portions of the
mesh from the iteration, if they lie entirely within a process's domain, which
improves the overall runtime by at least one order of magnitude.  In \pforest
the runtime is optimized by performing a so called $(3\times 3)$-neighborhood
check of an element~\cite{IsaacBursteddeWilcoxEtAl15}. For a local
hexahedral (or quadrilateral, in 2D) element, we check whether all same-level
face- (or edge-/vertex-)neighbors are also process local and if so, the element
is excluded from the iteration. Since this check makes explicit use of the
Morton code and its properties, it is difficult to generalize for general
element types. Therefore, we choose a different ansatz by exploiting the
top-down search, which directly leads to an element-type independent algorithm.

Note that for unbalanced forests the number of neighbors of an element
$E$ that are ghosts can be arbitrarily large and is only bounded by the 
number of elements at maximum refinement level that can touch the faces 
of $E$. Therefore, also the number of remote processes is unbounded as well.

\section{Element face-neighbors}

An important part of \texttt{Ghost} is to construct the same-level
face-neighbor of a given element $E$ across a face $f$.
As long as such a face-neighbor is inside the same tree as $E$, this problem is
solved by the corresponding low-level function
\texttt{t8\_element\_face\-\_neigh\-bor\_inside}.
We describe its version for the TM-index in 
Algorithm~\ref{alg:face-neighbor};
see \cite{BursteddeWilcoxGhattas11} for an implementation for the hexahedral
Morton index.

The challenging part is to find element face-neighbors across tree boundaries.
This is particularly demanding for hybrid meshes since multiple types of trees
exist in the same forest; 
this is for example the case if a hexahedron tree is neighbor of a
prism tree. We thus aim for a general type-independent
algorithm to compute face-neighbors across tree boundaries.

To detect whether an element's face $f$ is also a tree boundary, we assume the
existence of a low-level algorithm
\texttt{element\_neighbor\_inside\_root}\footnote{In \tetcode this function is
part of \texttt{t8\_element\_face\_neighbor\_inside}; see
Appendix~\ref{ch:appendix}} that returns \texttt{true} if and only if the face $f$
is an inner face and thus not on the tree boundary. For the TM-index we can
derive such an algorithm from the methods that we present in
Section~\ref{sec:outside}.

\begin{remark}
 In the following we will denote all functions that are part of the 
 \tetcode low-level API with a \texttt{t8\_element} prefix.
 See Appendix~\ref{ch:appendix} for a complete list.
\end{remark}

The core idea for the face-neighbor algorithm across tree boundaries is to
explicitly build the face as a lower dimensional element.

Our notation convention is to use capital letters ($T,E,F,G$) for entities,
such as trees, elements and faces, and to use lower case for indices 
($t,e,f,g$) of those entities.

We are concerned with the following issue:\\[1ex]
\emph{Given a $d$-dimensional element $E$ in a tree $K$ and a face $F$ of $E$ of
which we know that it is a subface of a face $G$ of $K$, construct the
same-level face-neighbor element $E'$ of $E$ across $F$.}\\[1ex] 
We display this situation for a quadrilateral-triangle tree connection in Figure~\ref{fig:facesitu}.
\begin{figure}
\center
\def\svgwidth{0.5\textwidth}
\begingroup%
  \makeatletter%
  \providecommand\color[2][]{%
    \errmessage{(Inkscape) Color is used for the text in Inkscape, but the package 'color.sty' is not loaded}%
    \renewcommand\color[2][]{}%
  }%
  \providecommand\transparent[1]{%
    \errmessage{(Inkscape) Transparency is used (non-zero) for the text in Inkscape, but the package 'transparent.sty' is not loaded}%
    \renewcommand\transparent[1]{}%
  }%
  \providecommand\rotatebox[2]{#2}%
  \ifx\svgwidth\undefined%
    \setlength{\unitlength}{458.05712891bp}%
    \ifx\svgscale\undefined%
      \relax%
    \else%
      \setlength{\unitlength}{\unitlength * \real{\svgscale}}%
    \fi%
  \else%
    \setlength{\unitlength}{\svgwidth}%
  \fi%
  \global\let\svgwidth\undefined%
  \global\let\svgscale\undefined%
  \makeatother%
  \begin{picture}(1,0.79044063)%
    \put(0,0){\includegraphics[width=\unitlength,page=1]{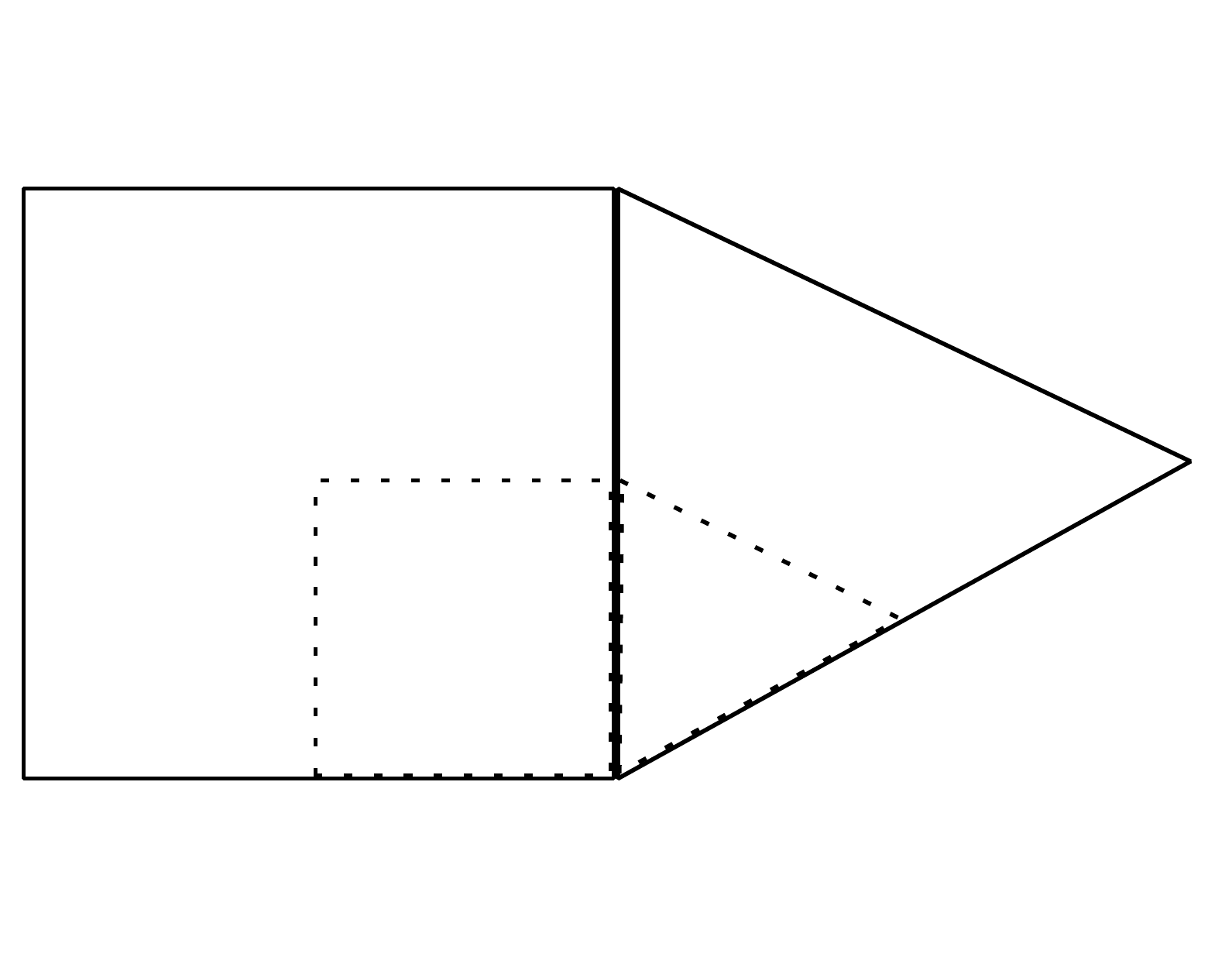}}%
    \put(0.10029941,0.50622876){\color[rgb]{0,0,0}\makebox(0,0)[lb]{\smash{K}}}%
    \put(0.71154646,0.41801682){\color[rgb]{0,0,0}\makebox(0,0)[lb]{\smash{K'}}}%
    \put(0.37475053,0.28167788){\color[rgb]{0,0,0}\makebox(0,0)[lb]{\smash{E}}}%
    \put(0.55863526,0.26403548){\color[rgb]{0,0,0}\makebox(0,0)[lb]{\smash{E'}}}%
    \put(0.37325352,0.72204715){\color[rgb]{0,0,0}\makebox(0,0)[lb]{\smash{G}}}%
    \put(0.57584834,0.72204715){\color[rgb]{0,0,0}\makebox(0,0)[lb]{\smash{G'}}}%
    \put(0.36726547,0.0184543){\color[rgb]{0,0,0}\makebox(0,0)[lb]{\smash{F}}}%
    \put(0.56561873,0.0184543){\color[rgb]{0,0,0}\makebox(0,0)[lb]{\smash{F'}}}%
    \put(0,0){\includegraphics[width=\unitlength,page=2]{cmesh_refine_facesituation_hybrid_tex.pdf}}%
  \end{picture}%
\endgroup%
 \caption[Element face-neighbor]{We show a tree $K$, an element $E$, and a face
$F$ of $E$ that is a subface of a tree face. The task is to find the
face-neighbor element $E'$.
 A subtask is to identify the tree
faces $G$ and $G'$, and the face $F'$.}
\figlabel{fig:facesitu}
\end{figure}

\begin{remark}
  We emphasize that neither the original element $E$ nor its face-neighbor $E'$
  need to be leaves in the forest.
  Even if $E$ is a leaf in the forest, $E'$ does not necessarily have to be a
  leaf, since $E'$ has the same refinement level as $E$.
\end{remark}

Since the trees can be rotated against each other, their coordinate systems may
not be aligned. We must properly transform the $(d-1)$-dimensional coordinates
of the faces
between the two systems.

In order to compute the face-neighbor $E'$ we consider the corresponding face
$G$ of the tree $K$ as a $(d-1)$-dimensional root tree. We then explicitly
construct the face $F$ of $E$ as a $(d-1)$-dimensional element descendant of
$G$.  The next step is to transform the coordinates of $F$ accordingly to build
the $(d-1)$-dimensional element $F'$ as a descendant of $G'$, the face's root
tree. In a final step, we construct the element $E'$ from the given face
element $F'$.

We thus identify four major substeps in the computation of face-neighbors across
tree boundaries:
\begin{enumerate}[(i)]
\item Identify the face number $g$ of the tree face.
\item Construct the  $(d-1)$-dimensional face element $F$.
\item Transform the coordinate system of $F$ to obtain the neighbor face
      element $F'$.
\item Extrude $F'$ to the $d$-dimensional neighbor element $E'$.
\end{enumerate}

We show these steps in Algorithm~\ref{alg:forestfaceneighbor} and describe
their details in the following sections.
See Figure~\ref{fig:face-neighbor-hex} for an illustration of steps (ii), (iii),
and (iv).

\begin{figure}
\begin{subfigure}[t]{0.48\textwidth}
 \def\svgwidth{\textwidth}
\begingroup%
  \makeatletter%
  \providecommand\color[2][]{%
    \errmessage{(Inkscape) Color is used for the text in Inkscape, but the package 'color.sty' is not loaded}%
    \renewcommand\color[2][]{}%
  }%
  \providecommand\transparent[1]{%
    \errmessage{(Inkscape) Transparency is used (non-zero) for the text in Inkscape, but the package 'transparent.sty' is not loaded}%
    \renewcommand\transparent[1]{}%
  }%
  \providecommand\rotatebox[2]{#2}%
  \ifx\svgwidth\undefined%
    \setlength{\unitlength}{799.10239258bp}%
    \ifx\svgscale\undefined%
      \relax%
    \else%
      \setlength{\unitlength}{\unitlength * \real{\svgscale}}%
    \fi%
  \else%
    \setlength{\unitlength}{\svgwidth}%
  \fi%
  \global\let\svgwidth\undefined%
  \global\let\svgscale\undefined%
  \makeatother%
  \begin{picture}(1,0.41401537)%
    \put(0,0){\includegraphics[width=\unitlength,page=1]{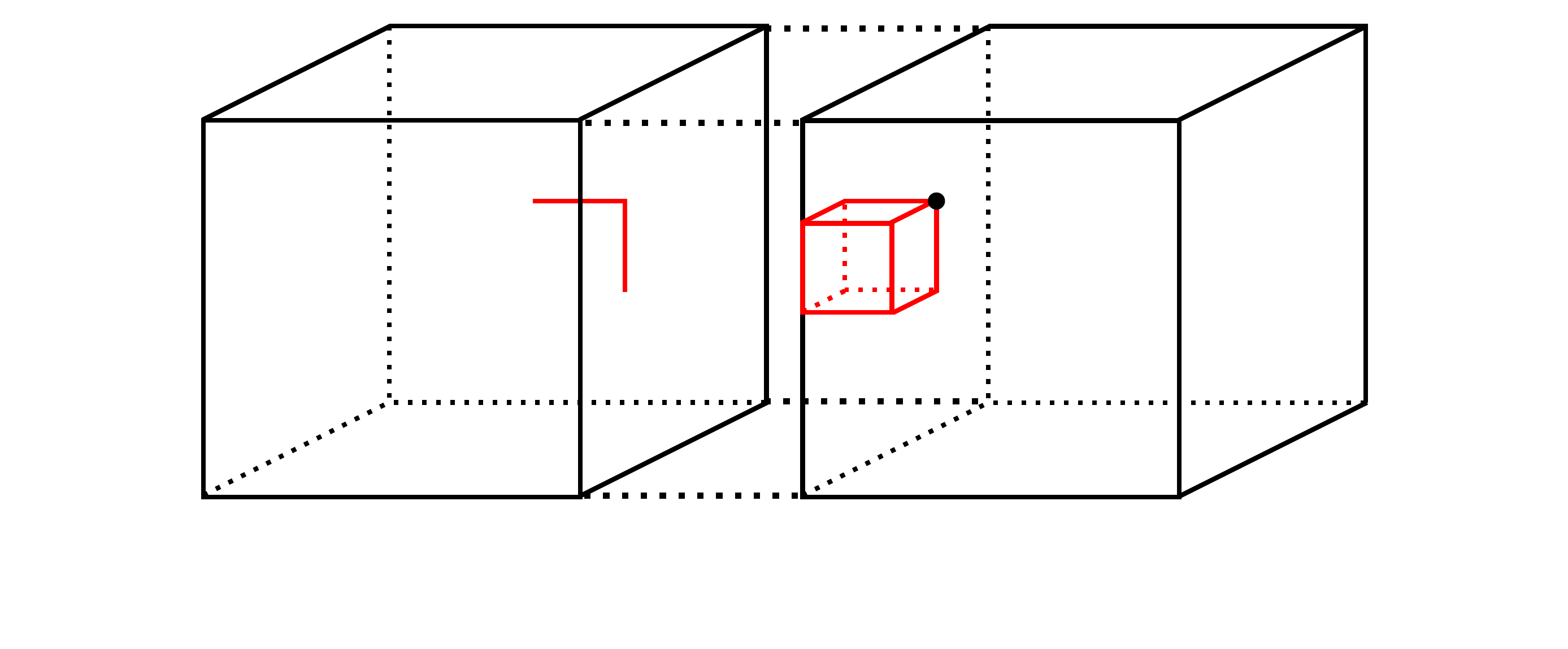}}%
    \put(0.21437256,0.03658345){\color[rgb]{0,0,0}\makebox(0,0)[lb]{\smash{$K$}}}%
    \put(0.60470357,0.03658345){\color[rgb]{0,0,0}\makebox(0,0)[lb]{\smash{$K'$}}}%
    \put(0,0){\includegraphics[width=\unitlength,page=2]{face_neighbor_hex1_tex.pdf}}%
    \put(0.33266427,0.03658345){\color[rgb]{0,0,0}\makebox(0,0)[lb]{\smash{$G$}}}%
    \put(0.46158235,0.03658345){\color[rgb]{0,0,0}\makebox(0,0)[lb]{\smash{$G'$}}}%
    \put(0,0){\includegraphics[width=\unitlength,page=3]{face_neighbor_hex1_tex.pdf}}%
    \put(0.25670404,0.23115977){\color[rgb]{0,0,0}\makebox(0,0)[lb]{\smash{$E$}}}%
    \put(0,0){\includegraphics[width=\unitlength,page=4]{face_neighbor_hex1_tex.pdf}}%
    \put(0.60728489,0.23115977){\color[rgb]{0,0,0}\makebox(0,0)[lb]{\smash{$E'$}}}%
    \put(0.05828918,0.0758545){\color[rgb]{0,0,0}\makebox(0,0)[lb]{\smash{$_x$}}}%
    \put(0.0555752,0.15506053){\color[rgb]{0,0,0}\makebox(0,0)[lb]{\smash{$_y$}}}%
    \put(0.0098054,0.20967906){\color[rgb]{0,0,0}\makebox(0,0)[lb]{\smash{$_z$}}}%
    \put(0.99306778,0.32776495){\color[rgb]{0,0,0}\makebox(0,0)[lb]{\smash{$_x$}}}%
    \put(0.9266613,0.34946698){\color[rgb]{0,0,0}\makebox(0,0)[lb]{\smash{$_y$}}}%
    \put(0.88941755,0.37601326){\color[rgb]{0,0,0}\makebox(0,0)[lb]{\smash{$_z$}}}%
  \end{picture}%
\endgroup%
\hfill
 \caption*{The starting point}
\end{subfigure}
\begin{subfigure}[t]{0.48\textwidth}
 \def\svgwidth{\textwidth}
\begingroup%
  \makeatletter%
  \providecommand\color[2][]{%
    \errmessage{(Inkscape) Color is used for the text in Inkscape, but the package 'color.sty' is not loaded}%
    \renewcommand\color[2][]{}%
  }%
  \providecommand\transparent[1]{%
    \errmessage{(Inkscape) Transparency is used (non-zero) for the text in Inkscape, but the package 'transparent.sty' is not loaded}%
    \renewcommand\transparent[1]{}%
  }%
  \providecommand\rotatebox[2]{#2}%
  \ifx\svgwidth\undefined%
    \setlength{\unitlength}{799.10400391bp}%
    \ifx\svgscale\undefined%
      \relax%
    \else%
      \setlength{\unitlength}{\unitlength * \real{\svgscale}}%
    \fi%
  \else%
    \setlength{\unitlength}{\svgwidth}%
  \fi%
  \global\let\svgwidth\undefined%
  \global\let\svgscale\undefined%
  \makeatother%
  \begin{picture}(1,0.41401368)%
    \put(0,0){\includegraphics[width=\unitlength,page=1]{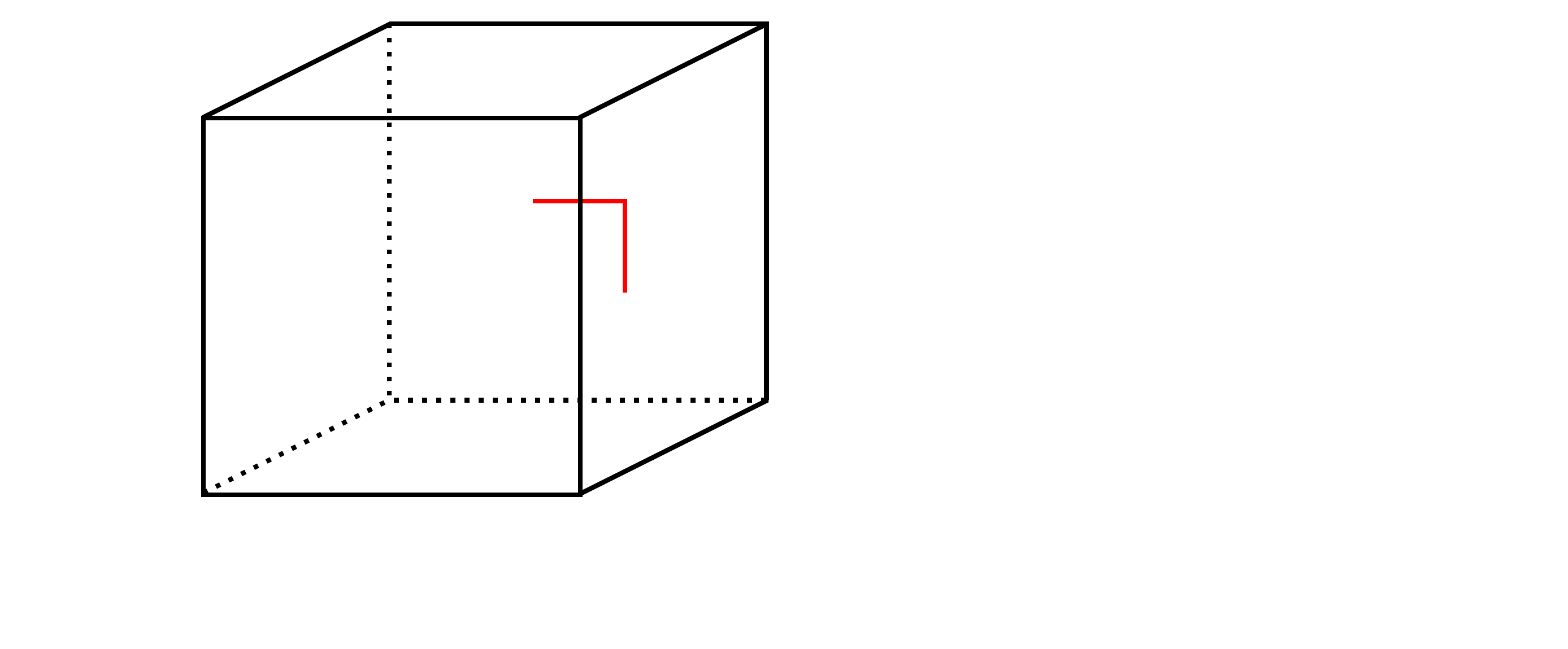}}%
    \put(0.21442341,0.03802052){\color[rgb]{0,0,0}\makebox(0,0)[lb]{\smash{$K$}}}%
    \put(0,0){\includegraphics[width=\unitlength,page=2]{face_neighbor_hex2_bdyface_tex.pdf}}%
    \put(0.78363383,0.03670257){\color[rgb]{0,0,0}\makebox(0,0)[lb]{\smash{$G$}}}%
    \put(0,0){\includegraphics[width=\unitlength,page=3]{face_neighbor_hex2_bdyface_tex.pdf}}%
    \put(0.78360257,0.22622664){\color[rgb]{0,0,0}\makebox(0,0)[lb]{\smash{$F$}}}%
    \put(0.06004851,0.07120157){\color[rgb]{0,0,0}\makebox(0,0)[lb]{\smash{$_x$}}}%
    \put(0.05733453,0.15040741){\color[rgb]{0,0,0}\makebox(0,0)[lb]{\smash{$_y$}}}%
    \put(0.01156483,0.20502583){\color[rgb]{0,0,0}\makebox(0,0)[lb]{\smash{$_z$}}}%
    \put(0.63894708,0.07523844){\color[rgb]{0,0,0}\makebox(0,0)[lb]{\smash{$_x$}}}%
    \put(0.5904634,0.2090627){\color[rgb]{0,0,0}\makebox(0,0)[lb]{\smash{$_y$}}}%
    \put(0.53282823,0.65583159){\color[rgb]{0,0,0}\makebox(0,0)[lt]{\begin{minipage}{0.17863403\unitlength}\raggedright \end{minipage}}}%
  \end{picture}%
\endgroup%
  \caption*{(ii)}
\end{subfigure}

\begin{subfigure}[t]{0.48\textwidth}
 \def\svgwidth{\textwidth}
\begingroup%
  \makeatletter%
  \providecommand\color[2][]{%
    \errmessage{(Inkscape) Color is used for the text in Inkscape, but the package 'color.sty' is not loaded}%
    \renewcommand\color[2][]{}%
  }%
  \providecommand\transparent[1]{%
    \errmessage{(Inkscape) Transparency is used (non-zero) for the text in Inkscape, but the package 'transparent.sty' is not loaded}%
    \renewcommand\transparent[1]{}%
  }%
  \providecommand\rotatebox[2]{#2}%
  \ifx\svgwidth\undefined%
    \setlength{\unitlength}{799.10400391bp}%
    \ifx\svgscale\undefined%
      \relax%
    \else%
      \setlength{\unitlength}{\unitlength * \real{\svgscale}}%
    \fi%
  \else%
    \setlength{\unitlength}{\svgwidth}%
  \fi%
  \global\let\svgwidth\undefined%
  \global\let\svgscale\undefined%
  \makeatother%
  \begin{picture}(1,0.41401368)%
    \put(0,0){\includegraphics[width=\unitlength,page=1]{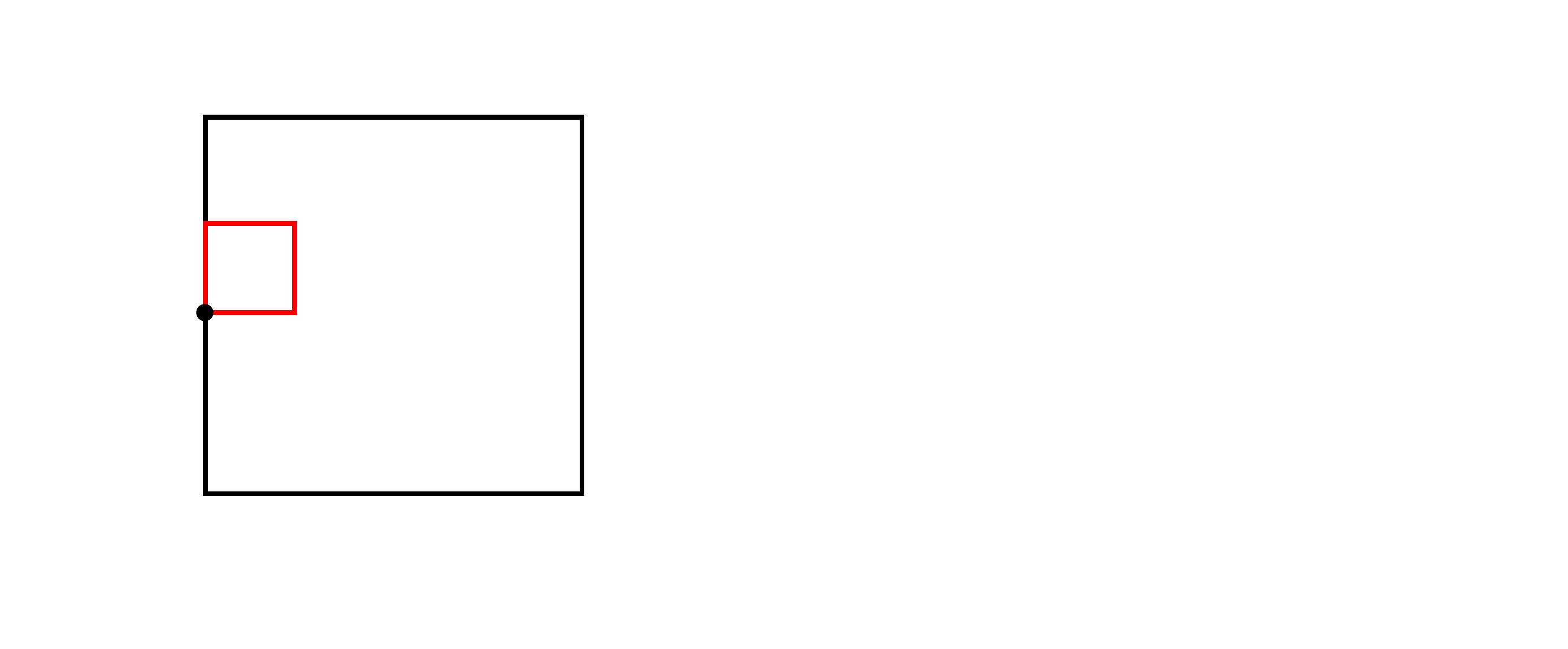}}%
    \put(0.20798911,0.03670257){\color[rgb]{0,0,0}\makebox(0,0)[lb]{\smash{$G$}}}%
    \put(0,0){\includegraphics[width=\unitlength,page=2]{face_neighbor_hex3_facetrans_tex.pdf}}%
    \put(0.58625273,0.03708453){\color[rgb]{0,0,0}\makebox(0,0)[lb]{\smash{$G'$}}}%
    \put(0,0){\includegraphics[width=\unitlength,page=3]{face_neighbor_hex3_facetrans_tex.pdf}}%
    \put(0.20761742,0.22645079){\color[rgb]{0,0,0}\makebox(0,0)[lb]{\smash{$F$}}}%
    \put(0.58574337,0.22660859){\color[rgb]{0,0,0}\makebox(0,0)[lb]{\smash{$F'$}}}%
    \put(0.0592312,0.07354792){\color[rgb]{0,0,0}\makebox(0,0)[lb]{\smash{$_x$}}}%
    \put(0.01074751,0.20737218){\color[rgb]{0,0,0}\makebox(0,0)[lb]{\smash{$_y$}}}%
    \put(0.86950839,0.27228219){\color[rgb]{0,0,0}\makebox(0,0)[lb]{\smash{$_x$}}}%
    \put(0.77787182,0.31652589){\color[rgb]{0,0,0}\makebox(0,0)[lb]{\smash{$_y$}}}%
  \end{picture}%
\endgroup%
\hfill
 \caption*{(iii)}
\end{subfigure}
\begin{subfigure}[t]{0.48\textwidth}
 \def\svgwidth{\textwidth}
\begingroup%
  \makeatletter%
  \providecommand\color[2][]{%
    \errmessage{(Inkscape) Color is used for the text in Inkscape, but the package 'color.sty' is not loaded}%
    \renewcommand\color[2][]{}%
  }%
  \providecommand\transparent[1]{%
    \errmessage{(Inkscape) Transparency is used (non-zero) for the text in Inkscape, but the package 'transparent.sty' is not loaded}%
    \renewcommand\transparent[1]{}%
  }%
  \providecommand\rotatebox[2]{#2}%
  \ifx\svgwidth\undefined%
    \setlength{\unitlength}{799.10400391bp}%
    \ifx\svgscale\undefined%
      \relax%
    \else%
      \setlength{\unitlength}{\unitlength * \real{\svgscale}}%
    \fi%
  \else%
    \setlength{\unitlength}{\svgwidth}%
  \fi%
  \global\let\svgwidth\undefined%
  \global\let\svgscale\undefined%
  \makeatother%
  \begin{picture}(1,0.41401368)%
    \put(0,0){\includegraphics[width=\unitlength,page=1]{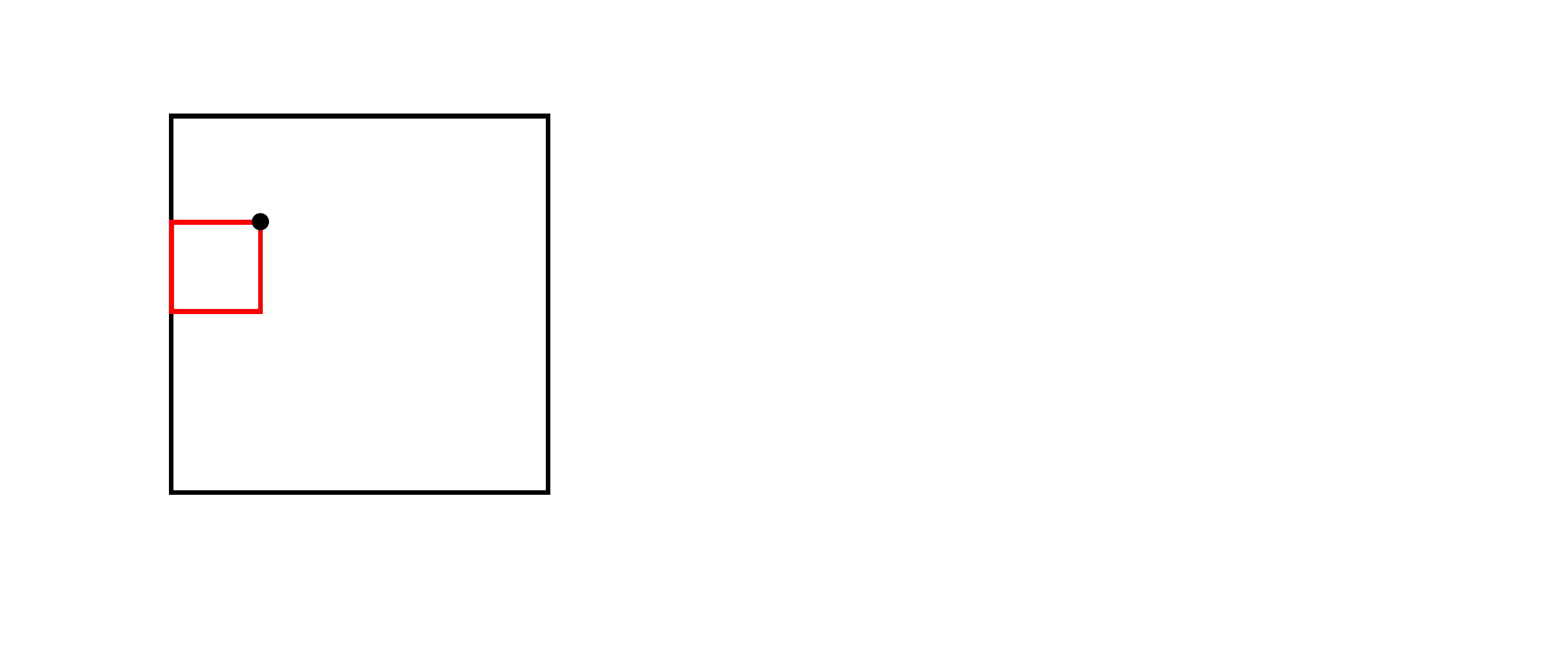}}%
    \put(0.18627798,0.0374569){\color[rgb]{0,0,0}\makebox(0,0)[lb]{\smash{$G'$}}}%
    \put(0,0){\includegraphics[width=\unitlength,page=2]{face_neighbor_hex4_faceext_tex.pdf}}%
    \put(0.60757599,0.03822802){\color[rgb]{0,0,0}\makebox(0,0)[lb]{\smash{$K'$}}}%
    \put(0,0){\includegraphics[width=\unitlength,page=3]{face_neighbor_hex4_faceext_tex.pdf}}%
    \put(0.18560335,0.22696489){\color[rgb]{0,0,0}\makebox(0,0)[lb]{\smash{$F'$}}}%
    \put(0.45584682,0.27293441){\color[rgb]{0,0,0}\makebox(0,0)[lb]{\smash{$_x$}}}%
    \put(0.36421025,0.31717811){\color[rgb]{0,0,0}\makebox(0,0)[lb]{\smash{$_y$}}}%
    \put(0.99543061,0.33055057){\color[rgb]{0,0,0}\makebox(0,0)[lb]{\smash{$_x$}}}%
    \put(0.92902433,0.35225255){\color[rgb]{0,0,0}\makebox(0,0)[lb]{\smash{$_y$}}}%
    \put(0.89178065,0.37879878){\color[rgb]{0,0,0}\makebox(0,0)[lb]{\smash{$_z$}}}%
  \end{picture}%
\endgroup%
  \caption*{(iv)}
\end{subfigure}
\caption[Two hexahedral elements that are face-neighbors across tree 
boundaries.]
{Two hexahedral elements that are face-neighbors across tree
boundaries (exploded view). Constructing the face-neighbor $E'$ of $E$ across
the face amounts to computing its anchor node (black) from the anchor node of
$E$ and the coarse mesh connectivity information about
the tree connection. Here, the coordinate systems of the two trees are rotated
against each other.
Top right: In step (ii) we construct the face element $F$ from the element $E$.
The coordinate system of the face root is inferred from that of the left tree.
Bottom left: In step (iii) we transform the face element $F$ to the neighbor
face element $F'$.  Bottom right: In the last step (iv) we extrude the
face-neighbor
$E'$ from the face element $F'$.
}
\figlabel{fig:face-neighbor-hex}
\end{figure}

\begin{remark}
  We deliberately choose this method of using lower dimensional entities over
  directly transforming the tree coordinates from one tree to the other---as it
  is done for example in \cite{BursteddeWilcoxGhattas11}---since
  our approach allows for maximum flexibility of the implementations of the different 
  element and SFC types. This holds since all intermediate operations are
  either local to one element type or only change the dimension (i.e.\
  hexahedra to quadrilaterals, tetrahedra to triangles, and back).
  Therefore, even if, for example, a hexahedron tree is neighbor to a
  prism tree, no function in the implementation of the hexahedral
  elements relies on knowledge about the implementation of the prism
  elements.
  Hence, it is possible to exchange different implementations of SFC of one
  element type without changing the others.
\end{remark}

\begin{algorithm}
\DontPrintSemicolon
\caption{\texttt{t8\_forest\_face\_neighbor (Forest $\forest F$, tree $T$, element $E$, face number $f$)}}
\label{alg:forestfaceneighbor}
\algoresult{The same-level face-neighbor $E'$ of $E$ across $f$.}\;
\algoif{\texttt{element\_neighbor\_inside\_root ($E, f$)}}
 { %
  $E'\gets $\,\texttt{t8\_element\_face\_neighbor\_inside ($E, f$)}\;
  \Return $E'$
 }
$g\gets$\,\texttt{t8\_element\_tree\_face ($E, f$)} \Comment{(i) The face number of $G$}
$F \gets $\,\texttt{t8\_element\_boundary\_face ($E, f$)}\Comment{(ii) Construct the face element}
$\mathrm{o}\gets $\,\texttt{face\_orientation ($\forest F, T, g$)}\Comment{The orientation of the tree face connection}
$F'\gets$\,\texttt{t8\_element\_transform\_face ($F$, o)}\Comment{(iii) Obtain the neighbor face element}
$g'\gets$\,\texttt{tree\_neighbor\_face ($\forest F, T, g$)}\Comment{The face number of $G'$}
$E' \gets$\,\texttt{t8\_element\_extrude\_face ($F', g'$)}\Comment{(iv) Build the neighbor from $F'$}
\Return $E'$
\end{algorithm}

\subsection{(i) Identifying the tree face}

The first subproblem is to identify the tree face $G$, respectively its
face index $g$, from $E$, $f$, and $K$. For this task we define a new
low-level function:

\begin{itemize}
\item
  \texttt{t8\_element\_tree\_face (element E, face\_index f)}\\
  If $f$ is a subface of a tree face, return the face index g of this root tree
  face.
\end{itemize}

The function \texttt{t8\_element\_tree\_face} returns the root tree face index
of the root face of an element's face $f$, provided this face is a boundary
face.  Its return value thus depends on the enumeration of the faces of an
element in relation to the faces of its root tree.  

For lines, quadrilaterals
and hexahedra with the Morton index, the root tree face indices are the same as
the element's face indices \cite{BursteddeWilcoxGhattas11} and thus
\texttt{t8\_element\_tree\_face} always returns $f$.

For simplices with the TM index, the enumeration of their faces depends on
their simplex type.  Face number $i$ refers to the unique face that
does not contain the vertex $\vec{x}_i$. We show the vertices of the different
types in Figure~\ref{fig:sechstetra}.

We observe that for triangles of type $0$, the face number is the same as the
face number of the root tree (since triangles of type $0$ are scaled copies of
the root tree).  Triangles of type $1$ cannot lie on the boundary of the root
tree and thus we never call \texttt{t8\_element\_tree\_face} with a type $1$
triangle.

For tetrahedra of type $0$ the same reasoning holds as for type
$0$ triangles, \linebreak \texttt{t8\_element\_tree\_face} returns $f$.
Tetrahedra of type $3$ cannot lie on the boundary of the root tree. For each 
of the remaining four types there is exactly one face that can lie
on the root tree boundary. Face $0$ of type $1$ tetrahedra is a descendant of the
root face $0$; face 2 of type $2$ tetrahedra is a descendant of the root face $1$;
face 1 of type $4$ tetrahedra is a descendant of the root face $2$.
Finally, face $3$ of type $5$ tetrahedra is a descendant of the root face $3$.
We list these indices in Table~\ref{tab:elementtreeface}.

Note that for face indices $f$ of faces that cannot lie on the root boundary,
the return value of \texttt{t8\_element\_tree\_face} is undefined.
This behavior is legal, since we ensure in Algorithm~\ref{alg:forestfaceneighbor}
that the function is only called if the face $f$ does lie on the root boundary.
We do so by calling \texttt{element\_neighbor\_inside\_root} beforehand
which queries whether the face-neighbor across $f$ is in the root tree or not.

\begin{table}
\center
\begin{tabular}{|c|c|c||c|c|c|}
\hline
\multicolumn{6}{|c|}{Tetrahedron}\\ \hline
  $\type(T)$ & $f$ & $g$ & $\type(T)$ & $f$ & $g$\\ \hline
0 & $i$ & $i$ & 3 &   -   &   -   \\\hline
1 & $0$ & $0$ & 4 & $1$ & $2$ \\\hline
2 & $2$ & $1$ & 5 & $3$ & $3$ \\\hline
\end{tabular}
\caption[Face number and type for tetrahedral subfaces]
  {\texttt{g = t8\_element\_tree\_face (T, f)} for a tetrahedron $T$ and a face $f$
  of $T$ that lies on a tree face.
  Depending on $T$'s type, all, exactly one, or none of its faces can be a
  subface of a face of the root tetrahedron tree. We show the tetrahedron's face
  number $f$ and the corresponding face number $g$ in the root tetrahedron.
  Type $3$ tetrahedra can never have a subface of the root tetrahedron as
  face. For type $0$ tetrahedra, each of their faces can be a subface of 
  the root tetrahedron's face with the same index.}
\figlabel{tab:elementtreeface}
\end{table}

\subsection{(ii) Constructing the face element}
As a next step, we build the face $F$ as a $(d-1)$-dimensional
element.
We do this via the low-level function

\lowlevel{t8\_element\_boundary\_face (element E, face\_index f)}{
   Return the $(d-1)$-dimensional face element $F$ 
   corresponding to the face index $f$.}

Thus, the lower dimensional face element $F$ has to be created from $E$.  For
the Morton index this is equivalent to computing the coordinates of its anchor
node and additionally its type for simplices. Hereby we interpret the tree face
$G$ as a $(d-1)$-dimensional root tree of which $F$ is a descendant element.
See also Figure~\ref{fig:construct-face}.

\begin{figure}
\center
\def\svgwidth{0.6\textwidth}
\begingroup%
  \makeatletter%
  \providecommand\color[2][]{%
    \errmessage{(Inkscape) Color is used for the text in Inkscape, but the package 'color.sty' is not loaded}%
    \renewcommand\color[2][]{}%
  }%
  \providecommand\transparent[1]{%
    \errmessage{(Inkscape) Transparency is used (non-zero) for the text in Inkscape, but the package 'transparent.sty' is not loaded}%
    \renewcommand\transparent[1]{}%
  }%
  \providecommand\rotatebox[2]{#2}%
  \ifx\svgwidth\undefined%
    \setlength{\unitlength}{658.38232422bp}%
    \ifx\svgscale\undefined%
      \relax%
    \else%
      \setlength{\unitlength}{\unitlength * \real{\svgscale}}%
    \fi%
  \else%
    \setlength{\unitlength}{\svgwidth}%
  \fi%
  \global\let\svgwidth\undefined%
  \global\let\svgscale\undefined%
  \makeatother%
  \begin{picture}(1,0.47149704)%
    \put(0,0){\includegraphics[width=\unitlength,page=1]{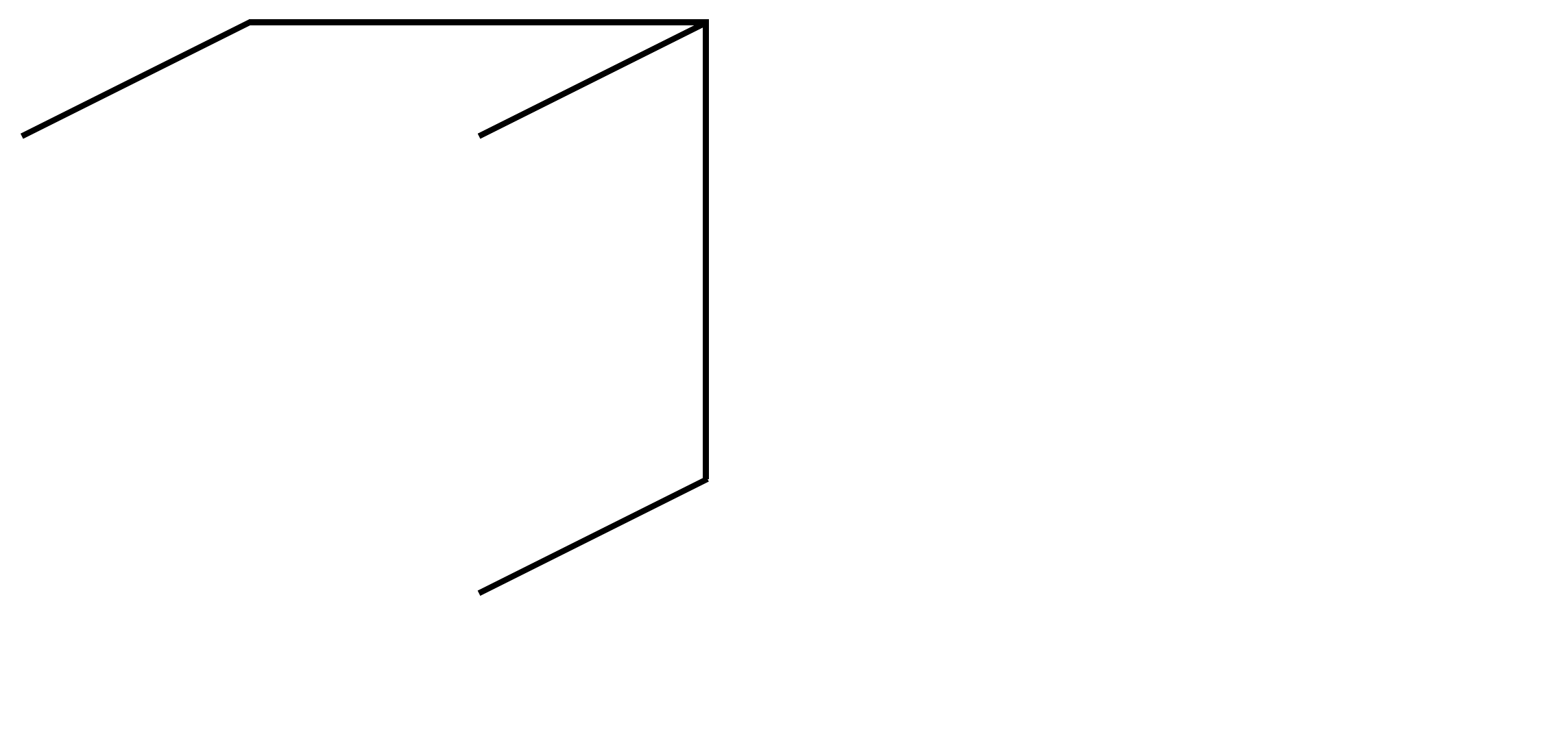}}%
    \put(0.13424745,0.03144546){\color[rgb]{0,0,0}\makebox(0,0)[lb]{\smash{$K$}}}%
    \put(0,0){\includegraphics[width=\unitlength,page=2]{face_neighbor_hex_facetree_tex.pdf}}%
    \put(0.79739285,0.03144546){\color[rgb]{0,0,0}\makebox(0,0)[lb]{\smash{$G$}}}%
    \put(0.18748731,0.25018946){\color[rgb]{0,0,0}\makebox(0,0)[lb]{\smash{$E$}}}%
    \put(0,0){\includegraphics[width=\unitlength,page=3]{face_neighbor_hex_facetree_tex.pdf}}%
    \put(0.67788548,0.28645259){\color[rgb]{0,0,0}\makebox(0,0)[lb]{\smash{$F$}}}%
  \end{picture}%
\endgroup%
 \caption[Constructing the face element $F$ to an element $E$ at a tree face
$G$.]
{Constructing the face element $F$ to an element $E$ at a tree face
$G$. We can interpret the face of the 3D tree $K$ as a 2D tree $G$.  The face
$F$ of $E$ is an element in this tree.}
\figlabel{fig:construct-face}
\end{figure}

\begin{remark}
  Since we construct a lower-dimensional element as the face of a
  higher-dimensional one, there are two conditions that need to be satisfied for
  the implementations of the two element types involved.
\begin{enumerate}
 \item The refinement pattern of a face of the higher dimensional elements must
       match the lower dimensional refinement pattern.
 \item The maximum possible refinement level of higher dimensional elements must
       not exceed the one of the lower dimensional elements.
\end{enumerate}
If one, or both, of these conditions are not fulfilled, then there exist 
faces of the higher dimensional elements for which an interpretation as a lower
dimensional element is not possible.
For Morton-type SFCs these conditions are naturally fulfilled.

\end{remark}
\begin{remark}
 For the simplicial and hexahedral Morton SFC with maximum refinement level
 $\mathscr{L}$, the anchor node coordinates of an element of level $\ell$ are
 integer multiples of $2^{\mathscr L - \ell}$.
 Suppose, the maximum level of hexahedral elements is $\mathscr{L}_1$ and the
 maximum level of a face boundary quadrilateral element is $\mathscr{L}_2 \geq
 \mathscr{L}_1$, then we will have to multiply a hexahedral coordinate with
 ${2^{\mathscr{L}_2-\mathscr{L}_1}}$ to transform it into a
 quadrilateral coordinate. We will assume in this section without loss of
generality that all element types have the same maximum possible refinement
level. Hence, we omit the scaling factor from our equations.
\end{remark}

\subsubsection{Quadrilaterals and hexahedra}
Let a quadrilateral $Q$ with anchor node $(Q.x, Q.y)$ and level $\ell$ be
given. Let furthermore a face index $f$ be given such that the face is a
subface of the root tree. The face element is a line $F$ with level $\ell$.
The Computation of its anchor node $(F.x)$ is a mere coordinate projection,
depending on $f$. If $f$ is $0$ or $1$, then $F.x = Q.y$, otherwise $f$
is $2$ or $3$ and we obtain $F.x = Q.x$.

In 3D we have similar projections, which we list in
Table~\ref{tab:quadhexfaces}.

\begin{table}
\center
\begin{tabular}[t]{|c|c|}\hline
\multicolumn{2}{|c|}{Quad}\\\hline
 $f$  & $F.x$ \\ \hline
 $0$ & $Q.y$ \\ \hline 
 $1$ & $Q.y$ \\ \hline
 $2$ & $Q.x$ \\ \hline
 $3$ & $Q.x$ \\ \hline
\end{tabular}\hspace{2ex}%
\begin{tabular}[t]{|c|c||c|c|}\hline
\multicolumn{4}{|c|}{Hexahedron}\\\hline
 $f$ & $(F.x, F.y)$  & $f$ & $(F.x, F.y)$\\ \hline
 $0$ & $(Q.y, Q.z)$  & $3$ & $(Q.x, Q.z)$ \\ \hline
 $1$ & $(Q.y, Q.z)$  & $4$ & $(Q.x, Q.y)$ \\ \hline 
 $2$ & $(Q.x, Q.z)$  & $5$ & $(Q.x, Q.y)$ \\ \hline
\end{tabular}
\caption
[\texttt{t8\_element\_boundary\_face} for quadrilaterals and hexahedra.]
{\texttt{t8\_element\_boundary\_face} for quadrilaterals and hexahedra.
  Left: For a quadrilateral $Q$ with anchor node $(Q.x, Q.y)$ and a face
$f$, the corresponding anchor node coordinate $F.x$ of the face line element.
Right: For a hexahedron $Q$ with anchor node $(Q.x, Q.y, Q.z)$ and a face
$f$, the corresponding anchor node coordinates $(F.x, F.y)$ of the face
quadrilateral element.
In either case, computing the coordinates is equivalent to a projection.}
\figlabel{tab:quadhexfaces}
\end{table}

\subsubsection{Triangles and tetrahedra}
Since we only construct these face elements for tree boundary faces,
we do not need to consider all combinations of element and face number,
but only those that occur on the tree boundary.
In particular, all possible faces are subfaces of the faces of the root simplex
$S_0$.

Triangles of type $1$ never lie on the root tree boundary, hence we only need
to consider type $0$ triangles.
Let $T$ be a type $0$ triangle with anchor node coordinates $(T.x, T.y)$
and level $\ell$. Furthermore, let a face index $f$ of a subface of the root
tree be given.
The corresponding face element $F$ is a line with level $\ell$ and anchor
node coordinate $F.x$.
We compute this $F.x$ from $T$'s anchor node as follows.
If $f = 0$, then $F.x = T.y$. If
$f = 1$, then $F.x = T.x$. Finally, if $f = 2$ then $F.x = T.x = T.y$. 
We show these cases in Table~\ref{tab:facecat}.

A tetrahedron that lies on the root tree boundary has type different
from $3$. In order to compute the boundary face, we further distinguish
two categories of such tetrahedra.
\begin{definition}
For a tetrahedron $T$ with $\type(T)\neq 3$, we identify two
\textbf{categories} of root face indices $g$ regarding how the anchor node
$(F.x, F.y)$ of the corresponding face triangle $F$ is computed from $T$'s
anchor node $(T.x, T.y, T.z)$:

Faces $g = 0$ and $g = 1$ of $S_0$ lie in the $(x = 0)$-plane or the $(x=z)$-plane of the
coordinate system and form category $1$. 
In this case $(F.x, F.y) = (T.z, T.y)$.

Category $2$ are the faces $g = 2$ and $g = 3$ of $S_0$. These lie in the
$(y = 0)$-plane or the $(y=z)$-plane and the anchor node of $F$ is given by $(F.x,
F.y) = (T.x,T.z)$.
\end{definition}

Depending on the type of a parent tetrahedron its faces fall in one of these
categories and have a distinguished triangle type $0$ or $1$ as subface of
$g\in\set{0, 1, 2, 3}$.

We display the various categories of faces depending on the type of a simplex
in Table \ref{tab:facecat}. We also list the anchor node coordinates
of the face elements in this table.

\begin{table}
\center
 \begin{tabular}[t]{|c|c|c|} \hline
 \multicolumn{3}{|c|}{Triangle}\\\hline
 \mytabvspace $\type{(T)}$ & $f$ & $F.x$ \\ \hline
        0  &  $0$ & $T.y$  \\
           &  $1$ & $T.x$  \\
           &  $2$ & $T.x$  \\\hline
 \end{tabular}\hspace{3ex}%
 \begin{tabular}[t]{|c|c|c|c|c|} \hline
 \multicolumn{5}{|c|}{Tetrahedron}\\\hline
 \mytabvspace $\type{(T)}$ & $f$ & Cat. & $\type{(F)}$ & $(F.x,F.y)$ \\ \hline
   0  &  $0$   & 1 & 0 & $(T.z,T.y)$ \\
      &  $1$   & 1 & 0 & $(T.z,T.y)$ \\
      &  $2$   & 2 & 0 & $(T.x,T.z)$ \\ 
      &  $3$   & 2 & 0 & $(T.x,T.z)$ \\\hline
   1  &  $0$   & 1 & 1 & $(T.z,T.y)$ \\\hline
   2  &  $2$   & 1 & 1 & $(T.z,T.y)$ \\\hline
   3  &    -     & - & - & $ -       $ \\\hline
   4  &  $1$   & 2 & 1 & $(T.x,T.z)$ \\\hline
   5  &  $3$   & 2 & 1 & $(T.x,T.z)$ \\\hline 
\end{tabular}
\caption[Coordinates and types for the faces of simplices]
  {\texttt{t8\_element\_boundary\_face (T, f)} for triangles and tetrahedra.
   Left: The $x$ coordinate of the anchor node of the boundary line $F$ at face
   $f$ of a triangle $T$ in terms of $T$'s coordinates. Right: The category and
   the type of the boundary triangle $F$ at a face $f$ of tetrahedron $T$ as
   well as the anchor node coordinates $(F.x, F.y)$.}
\figlabel{tab:facecat}
\end{table}

\subsection{(iii) Constructing $F'$ from $F$}
If we know the tree face index $g$, we can look up the corresponding face index
$g'$ of the face in $K'$ from the coarse mesh connectivity; see
Section~\ref{sec:cmesh-implementation}. 

In order to transform the coordinates of $F$ to obtain $F'$ we need 
to understand how the vertices of the face $g$ connect to the vertices of
the face $g'$. 
Each face's vertices form a subset of the vertices of the trees as in
Figure~\ref{fig:vertexface}. Let
$\set{v_0,\dots,v_{n-1}}$ and $\set{v'_0,\dots,v'_{n-1}}$ be these vertices for
$g$ and $g'$ in ascending order, thus $v_i<v_{i+1}$ and $v'_i < v'_{i+1}$.
The face-to-face connection of the two trees determines a permutation
$\sigma\in S_n$ such that vertex $v_i$ connects to vertex $v'_{\sigma(i)}$.
In theory, there are $n!$ possible permutations.
However, not all of them occur.
\begin{definition}
 Since we exclude the possibility of trees
with negative volume, there is exactly one combination in which the vertices
$v_0$ and $v'_0$ are connected (c.f.\ Section~\ref{sec:orientation}).
We call the corresponding permutation $\sigma_0$.
\end{definition}
All the other possible permutations result from rotating the face $g'$ in this
position. This rotation is encoded in the orientation information of the
coarse mesh; see Definition~\ref{def:orientation}.

Thus, in order to understand the permutation $\sigma$, it suffices to know the
 initial permutation $\sigma_0$ and the orientation. $\sigma_0$ is
determined by the types of $K$ and $K'$ and the face indices $g$ and $g'$.
In fact, since the orientation encodes the possible rotations, the only
data we need to know is the sign of $\sigma_0$.
\begin{definition}
\label{def:facesign}
  Let $K$ and $K'$ be two trees of types $t$ and $t'$, 
  and let $g$, $g'$ faces of $K$ and $K'$ of the same element type.
  We define the \textbf{sign} of $g$ and $g'$ as the sign of the
  permutation $\sigma_0$.
  \begin{equation}
    \sign_{t,t'}(g,g') := \sign(\sigma_0).
  \end{equation}

\end{definition}
\begin{remark}
Note that this definition does not depend on the order of the
faces $g$ and $g'$. Since if we switch their roles, the permutation changes to
its inverse $\sigma_0^{-1}$, thus
  \begin{equation}
    \sign_{t',t}(g',g) = \sign(\sigma_0^{-1}) = \sign(\sigma_0) = \sign_{t,t'}(g,g').
  \end{equation}
\end{remark}

\begin{remark}
  For hexahedra we can compute the sign of two faces via the 
  tables $\mathcal R, \mathcal Q, \mathcal P$ from~\cite[Table~3]{BursteddeWilcoxGhattas11} as
  \begin{equation}
   \label{eq:facetrafop4est}
    \sign_{\mathrm{hex},\mathrm{hex}}(g,g') = 
    \sign\left(i\mapsto \mathcal P\left(\mathcal Q( \mathcal R(g,g'),0), i\right)\right)
    = \neg\mathcal R(g,g').
  \end{equation}
  The permutation in the middle is exactly the permutation $\sigma_0$.
  The argument $0$ of $\mathcal Q$ is the orientation of a face-to-face
  connection, but the result is independent of it, and we could have chosen
  any other value.
\end{remark}
We display the sign for two example tree-to-tree connections, namely
tetrahedron to tetrahedron and hexahedron to prism, in Table~\ref{tab:facesigns}.

\begin{table}
  \center
  \begin{tabular}[t]{|cc|rrrr|}
    \hline
     \multicolumn{6}{|c|}{$K$ and $K'$ tetrahedra}\\\hline
    & & \multicolumn{4}{c|}{$g$}\\
    &  & 0 & 1 & 2 & 3\\\hline
      \multirow{4}{*}{$g'$}&
    0 &-1 & 1&-1 &1\\
   &1 &1 &-1 &1 &-1\\
   &2 &-1&1 &-1 &1\\
   &3 &1 &-1 &1 &-1\\\hline
  \end{tabular}
  \begin{tabular}[t]{|cc|rrrrrr|}
    \hline
     \multicolumn{8}{|c|}{$K$ hexahedron, $K'$ prism}\\\hline
    & & \multicolumn{6}{c|}{$g$}\\
    &  & 0 & 1 & 2 & 3 &4 & 5 \\\hline
      \multirow{4}{*}{$g'$}&
    0 & 1&-1&-1& 1& 1&-1\\
   &1 &-1& 1& 1&-1&-1& 1\\
   &2 & 1&-1&-1& 1& 1&-1\\ \hline
  \end{tabular}
\caption[$\sign_{t,t'}(g,g')$ from Definition~\ref{def:facesign} for two possible
 tree-to-tree connections.]
{$\sign_{t,t'}(g,g')$ from Definition~\ref{def:facesign} for two possible
 tree-to-tree connections. We obtain these values from
 Figure~\ref{fig:vertexface}. We do not show the remaining three cases of $K$
and $K'$ both being hexahedra or prisms, and $K$ being a tetrahedron and
$K'$ a prism, which can be obtained similarly.}
\figlabel{tab:facesigns}
\end{table}

Using the orientation, the sign, and the face index $g'$, we transform the
coordinates of $F$ to obtain the corresponding face $F'$ as a subface of the
face $G'$ of $K'$. For this task we introduce the low-level function 

\lowlevel{t8\_element\_transform\_face (face\_element F, orientation o, sign s)}{
   Return the transformed face element $F'$.}

\begin{remark}
\label{rem:facetrafo}
The transformation $o = i, s = -1$ is the same 
as first $o = 0, s = -1$ followed by $o = i, s = 1$.
We thus only need to compute the cases with $s = 1$ and the specific
case $o = 0, s = -1$.
\end{remark}

\begin{table}
\center
 \begin{tabular}[t]{|c|c|}\hline
  \multicolumn{2}{|c|}{Line}\\\hline
  \mytabvspace  $o$ & $\begin{pmatrix} F'.x\end{pmatrix}$\\[0.7ex] \hline
  \mytabvspace
   0  & $\begin{pmatrix}F.x\end{pmatrix}$ \\[0.5ex]
   1  & 
  $\begin{pmatrix}2^\mathcal{L} - F.x - h\end{pmatrix}$ \\ \hline
 \end{tabular}
 \begin{tabular}[t]{|c|c||c|c|}\hline
  \multicolumn{4}{|c|}{Quadrilateral}\\\hline
  \myhugetabvspace  $o$ & $\begin{pmatrix} F'.x\\F'.y\end{pmatrix}$ &
  $o$ & $\begin{pmatrix} F'.x\\F'.y\end{pmatrix}$\\[2.0ex] \hline
   \myhugetabvspace
   0  & $\begin{pmatrix}F.x\\ F.y\end{pmatrix}$ &
   2  &
  $\begin{pmatrix} F.y \\ 2^\mathcal{L}-F.x-h\end{pmatrix}$\\[2.8ex]
   1  & 
  $\begin{pmatrix}2^\mathcal{L} - F.y - h \\ F.x\end{pmatrix}$ &
   3  &
  $\begin{pmatrix} 2^\mathcal{L}-F.x-h\\2^\mathcal{L}-F.y-h \end{pmatrix}$\\ \hline
 \end{tabular}\\[2ex]
 \begin{tabular}{|c|c|c||c|c|c|}\hline
  \multicolumn{6}{|c|}{Triangle}\\\hline
  \myhugetabvspace $\type(F)$ & $o$ & $\begin{pmatrix} F'.x\\F'.y\end{pmatrix}$
 & $\type(F)$ & $o$ & $\begin{pmatrix} F'.x\\F'.y\end{pmatrix}$\\[2.0ex] \hline
   \myhugetabvspace
   0 &  0  & $\begin{pmatrix}F.x\\ F.y\end{pmatrix}$ & 
   1 &  0  & $\begin{pmatrix}F.x\\ F.y\end{pmatrix}$ \\[2.8ex]
     &  1  & 
  $\begin{pmatrix}2^\mathcal{L} - F.y - h \\ F.x - F.y\end{pmatrix}$ &
     &  1  & 
  $\begin{pmatrix}2^\mathcal{L} - F.y - h \\ F.x - F.y - h\end{pmatrix}$ \\[2.8ex]
     &  2  & 
  $\begin{pmatrix}2^\mathcal{L} - F.x + F.y - h\\ 2^\mathcal{L}-F.x-h\end{pmatrix}$ & 
     &  2  &
  $\begin{pmatrix}2^\mathcal{L} - F.x + F.y\\ 2^\mathcal{L}-F.x-h\end{pmatrix}$ 
     \\ \hline
 \end{tabular}
\caption[\texttt{t8\_transform\_face (F, o, s = 1)}.]
{\texttt{t8\_transform\_face (F, o, s = 1)} for lines (top left)
, quadrilaterals (top right) and triangles (bottom) with sign 1.
  For values with $s = -1$ see Table~\ref{tab:triangle-transform-b} and Remark~\ref{rem:facetrafo}.}
\figlabel{tab:triangle-transform}
\end{table}

Because of $\sigma_0(0) = 0$, the sign is always $1$ for vertex faces and line
faces. Hence, the sign argument is not relevant for the computation of 1D and
2D tree connections.

For the classical and tetrahedral Morton indices we need to compute the anchor
node of $F'$ from the anchor node of the input face $F$.
We show the computation $o = i, s = 1$ cases for lines, triangles and
quadrilaterals in Table~\ref{tab:triangle-transform}. Since we transform faces,
it is not necessary to discuss the
routine for 3 dimensional face element types.
We describe the formulas for $o = 0, s = -1$ for triangles and quadrilaterals
in Table~\ref{tab:triangle-transform-b}. As we mention in
Remark~\ref{rem:facetrafo}, we can compute all $o$ and $s$ combinations from
these two tables. Note that for quadrilaterals and hexahedra 
\texttt{t8\_element\_transform\_face} is equivalent to the internal coordinate
transformation in \texttt{p4est\_transform\_face} of the \pforest library
due to~\eqref{eq:facetrafop4est}.

\begin{table}
  \center
  \begin{tabular}[t]{|c|c|}\hline
  \multicolumn{2}{|c|}{Triangle}\\\hline
  \myhugetabvspace $\type(F)$ & $\begin{pmatrix} F'.x\\F'.y\end{pmatrix}$\\[2.0ex] \hline
   \myhugetabvspace
   0 &  $\begin{pmatrix}F.x\\ F.x - F.y\end{pmatrix}$  
     \\[2.8ex]
   1 &   
  $\begin{pmatrix} F.x \\ F.x - F.y - h\end{pmatrix}$ 
    \\ \hline
 \end{tabular}
  \begin{tabular}[t]{|c|}\hline
  \multicolumn{1}{|c|}{Quadrilateral}\\\hline
  \myhugetabvspace $\begin{pmatrix} F'.x\\F'.y\end{pmatrix}$\\[2.0ex] \hline
   \myhugetabvspace
    $\begin{pmatrix}F.y\\ F.x\end{pmatrix}$  
    \\ \hline
 \end{tabular}
  \caption[\texttt{t8\_transform\_face (F, o = 0, s = -1)}.]
{\texttt{t8\_transform\_face (F, o = 0, s = -1)} for triangles (left) and 
  quadrilaterals (right). We compute any arbitrary combination of values for
  $o$ with $s = -1$ by first applying \texttt{t8\_transform\_face (F, 0, -1)}
  and then \texttt{t8\_transform\_face (F, o, 1)} from Table~\ref{tab:triangle-transform}.}
  \label{tab:triangle-transform-b}
\end{table}

\subsection{(iv) Constructing $E'$ from $F'$}
We now have $E, F, F', K$ and $K'$ and can construct the neighbor element $E'$.
For this we use the function

\lowlevel{t8\_element\_extrude\_face (Face\_element F', Tree K', face index g')}
{
Return the element $E'$ that is a descendant of $K'$ and 
has the face $F'$ at the tree face $g'$.}

\texttt{t8\_element\_extrude\_face} has as input a face element and a root tree
face index and as output the element within the root tree that has as a boundary
face the given face element.
How to compute the element from this data depends on the element type and 
the root tree face.
For quadrilaterals, triangles, hexahedra, and tetrahedra with the (TM-)Morton index
we show the formulas to compute the anchor node coordinates of $E'$ in
Table~\ref{tab:face-extrude}.

\begin{table}
\center
 \begin{tabular}[t]{|c|c||c|c|}\hline
  \multicolumn{4}{|c|}{Quadrilateral from line $F'$}\\\hline
   \myhugetabvspace
  $g'$ & $\begin{pmatrix} E'.x\\ E'.y\end{pmatrix}$ &
  $g'$ & $\begin{pmatrix} E'.x\\E'.y\end{pmatrix}$\\[2.0ex] \hline
   \myhugetabvspace
   0  & $\begin{pmatrix}0\\ F'.x\end{pmatrix}$ &
   2  &
  $\begin{pmatrix} F'.x \\ 0\end{pmatrix}$\\[2.8ex]
   1  & 
  $\begin{pmatrix}2^\mathcal{L} - h \\ F'.x\end{pmatrix}$ &
   3  &
  $\begin{pmatrix} F'.x \\2^\mathcal{L} - h \end{pmatrix}$\\ \hline
 \end{tabular}
 \begin{tabular}[t]{|c|c||c|c|}\hline
  \multicolumn{4}{|c|}{Triangle from line $F'$}\\\hline
  \myhugetabvspace  $g'$ & $\begin{pmatrix} E'.x\\ E'.y\end{pmatrix}$ &
  $g'$ & $\begin{pmatrix} E'.x\\E'.y\end{pmatrix}$\\[2.0ex] \hline
   \myhugetabvspace
   0  & $\begin{pmatrix} 2^\mathcal L - h\\ F'.x\end{pmatrix}$ &
   2  &
  $\begin{pmatrix} F'.x \\ 0\end{pmatrix}$\\[2.8ex]
   1  & 
  $\begin{pmatrix}F'.x \\ F'.x\end{pmatrix}$ &
     & \\ \hline
 \end{tabular}
 \begin{tabular}[t]{|c|c||c|c||c|c|}\hline
  \multicolumn{6}{|c|}{Hexahedron from quadrilateral $F'$}\\\hline
   \myhugetabvspace
    $g'$ & $\begin{pmatrix} E'.x\\ E'.y \\ E'.z\end{pmatrix}$ &
    $g'$ & $\begin{pmatrix} E'.x\\ E'.y \\ E'.z\end{pmatrix}$ &
    $g'$ & $\begin{pmatrix} E'.x\\ E'.y \\ E'.z\end{pmatrix}$\\[2.0ex] \hline
   \myhugetabvspace
   0  & $\begin{pmatrix}0\\ F'.x \\ F'.y\end{pmatrix}$ &
   2  & 
  $\begin{pmatrix} F'.x \\ 0 \\ F'.y\end{pmatrix}$ &
   4  &
  $\begin{pmatrix} F'.x \\ F'.y \\ 0 \end{pmatrix}$\\[3.6ex]
   1  & 
  $\begin{pmatrix} 2^\mathcal L - h \\ F'.x \\ F'.y \end{pmatrix}$ &
   3  &
  $\begin{pmatrix} F'.x \\ 2^\mathcal L - h \\ F'.y \end{pmatrix}$ &
   5  &
  $\begin{pmatrix} F'.x \\ F'.y \\ 2^\mathcal{L} - h \end{pmatrix}$\\ \hline
 \end{tabular}\\
 \begin{tabular}[t]{|c|c||c|c|}\hline
  \multicolumn{4}{|c|}{Tetrahedron from triangle $F'$}\\
  \multicolumn{4}{|c|}{Coordinates}\\\hline
   \myhugetabvspace
    $g'$ & $\begin{pmatrix} E'.x\\ E'.y \\ E'.z\end{pmatrix}$ &
    $g'$ & $\begin{pmatrix} E'.x\\ E'.y \\ E'.z\end{pmatrix}$\\[2.0ex] \hline
   {\scalebox{1.2}{\myhugetabvspace}}
   0  & $\begin{pmatrix} 2^\mathcal L - h\\ F'.y \\ F'.x\end{pmatrix}$ &
   2  &
  $\begin{pmatrix} F'.x \\ 0 \\ F'.y\end{pmatrix}$\\[3.6ex]
   1  & 
  $\begin{pmatrix} F'.x \\ F'.y \\ F'.x\end{pmatrix}$ &
   3  &
  $\begin{pmatrix} F'.x \\ 0 \\ F'.y \end{pmatrix}$\\\hline
 \end{tabular}
 \begin{tabular}[t]{|c|c|c|}\hline
  \multicolumn{3}{|c|}{Tetrahedron from triangle $F'$}\\
  \multicolumn{3}{|c|}{Type}\\\hline
    $g'$ & $\type(F')$ & $\type(E')$ \\ \hline
   0 & 0 & 0 \\
     & 1 & 1 \\\hline
   1 & 0 & 0 \\
     & 1 & 2 \\\hline
   2 & 0 & 0 \\
     & 1 & 4 \\\hline
   3 & 0 & 0 \\
     & 1 & 5 \\\hline
 \end{tabular}
\caption
[The computation of \texttt{t8\_element\_extrude\_face (F', T', g')}]
{The computation of $E'=$ \texttt{t8\_element\_extrude\_face (F', T', g')}
for $T'$ a quadrilateral (top left), triangle (top right), hexahedron (middle),
or tetrahedron (bottom).  Depending on the anchor node coordinates of $F'$ and
the tree face index $g'$, we determine the anchor node of the extruded element
$E'$.  For tetrahedra, we additionally need to compute the type of $E'$ from
$g'$ and the type of the triangle $F'$ (bottom right). In the case of a triangle,
the type of $E'$ is always $0$, since type $1$ triangles cannot lie on a tree
boundary.
$h$ refers to the length of the element $E'$ (resp.\ $F'$) and is computed as
$2^{\mathcal L -\ell}$ where $\ell$ is the refinement level of $E'$ and $F'$.
}
\figlabel{tab:face-extrude}
\end{table}

\subsection{A note on vertex/edge-neighbors}
Despite the restriction to face-neighbors in this thesis, we are certain that
for tree-to-tree neighbors the same method of constructing the lower
dimensional element, transforming it into its neighbor and then extruding it to
the neighbor element can be applied to vertex and edge neighbors as well.
The challenge with these neighbors compared to face-neighbors is that at a single
vertex/edge an arbitrary number of trees can be connected.
Identifying the correct neighbor trees is, however, a task that entirely relies
on the coarse mesh connectivity. Once this is accomplished, the neighbor elements
can be constructed by using the techniques described in this section.

\section{Half-size face-neighbors}
\label{sec:ghost-halfneighbors}

In order to implement the ghost algorithm for balanced forests from
\cite{BursteddeWilcoxGhattas11}, we need to compute half face-neighbors of an
element. That is, given an element $E$ and a face $f$, construct the
neighbors\footnote{These do not need to be of half the size. If for example
refinement is 1:9 Peano refinement, the neighbors are one third the size.} of
$E$ across $f$ of refinement level $\ell(E) + 1$.

We construct the half face-neighbors in three steps:
\begin{enumerate}[(i)]
 \item Construct the children $C_f$ of $E$ that have a face in $f$.
 \item For each child $C_f[i]$ compute the face index  $f_i$ of the face
       that is a child of $f$ and a face of $C_f[i]$.
 \item For each child $C_f[i]$ compute its face-neighbor across $f_i$.
\end{enumerate}

(i) and (ii) are performed by low-level algorithms:

\lowlevel{t8\_element\_children\_at\_face (Element E, face\_index f)}{
Returns an array of children of $E$ that share a face with $f$.}

\lowlevel{t8\_element\_child\_face (Element E, child\_index i, face\_index f)}
{Given an element $E$, a child index $i$, and a face index $f$ of a face of
$E$, compute the index of the $i$-th child's face that is a subface of $f$.}

A typical implementation of \texttt{t8\_element\_children\_at\_face}
would look up the child indices of these children in a table and then 
construct the children with these indices.
The child indices can be obtained from the refinement pattern. For the
quadrilateral Morton index, for example, the child indices at face $f = 0$ are
$0$ and $2$. For a hexahedron the child indices at face $f = 3$ are $2, 3, 6$,
and $7$. For the TM index these indices additionally depend on the type of the
simplex. We list all cases in Table~\ref{tab:ghost-childrenatface}.

\begin{table}
\begin{tabular}[t]{|c|c|c|c|}
\hline
\multicolumn{4}{|c|}{Triangle}\\\hline
&\multicolumn{3}{c|}{$f$}\\\hline
$\type(T)$  & 0 & 1 & 2 \\ \hline
 0 & 1,3 & 0,3 & 0,1 \\ 
 1 & 2,3 & 0,3 & 0,2 \\\hline
\end{tabular}
\begin{tabular}[t]{|c|c|c|c|c|}
\hline
\multicolumn{5}{|c|}{Tetrahedron}\\\hline
&\multicolumn{4}{c|}{$f$}\\\hline
$\type(T)$  & 0 & 1 & 2 & 3 \\ \hline
 0 & 1, 4, 5, 7 & 0, 4, 6, 7 & 0, 1, 2, 7 & 0, 1, 3, 4 \\ 
 1 & 1, 4, 5, 7 & 0, 5, 6, 7 & 0, 1, 3, 7 & 0, 1, 2, 5 \\
 2 & 3, 4, 5, 7 & 0, 4, 6, 7 & 0, 1, 3, 7 & 0, 2, 3, 4 \\
 3 & 1, 5, 6, 7 & 0, 4, 6, 7 & 0, 1, 3, 7 & 0, 1, 2, 6 \\
 4 & 3, 5, 6, 7 & 0, 4, 5, 7 & 0, 1, 3, 7 & 0, 2, 3, 5 \\
 5 & 3, 5, 6, 7 & 0, 4, 6, 7 & 0, 2, 3, 7 & 0, 1, 3, 6 \\ \hline
\end{tabular}
\caption[The child indices of all children of an element touching a given face.]
{The child indices of all children of an element touching a given face.
These indices are needed for \texttt{t8\_element\_children\_at\_face}.
Left: The indices for a triangle $T$ in dependence on its type and the face index $f$.
Right: The same data for a tetrahedron $T$.}
\figlabel{tab:ghost-childrenatface}
\end{table}

The low-level algorithm \texttt{t8\_element\_child\_face} can also be described
via lookup tables. Its input is a parent element $E$, a face index $f$ and a
child index $i$, such that the child $E_i$ of $E$ has a subface of the face
$f$. In other words, $E_i$ is part of the output of
\texttt{t8\_element\_children\_at\_face}. The return value of
\texttt{t8\_element\_child\_face} is the face index $f_i$ of the face of
$E[i]$ that is the subface of $f$.

For the classical Morton index, the algorithm is the identity on $f$, since the faces
of child quadrilaterals/hexahedra are labeled in the same manner as those of the
parent element. 
For the TM index for triangles, the algorithm is also the identity, since only
triangle children of the same type as the parent can touch a face of the parent
and for same type triangles the faces are labeled in the same manner.

For tetrahedra, the algorithm is the identity on those children that have the
same type as the parent. However, for each face $f$ of a tetrahedron $T$, 
there exists a child of $T$ that has the middle face child of $f$
as a face. This child has not the same type as $T$.
For this child the corresponding face value is computed as $0$ if $f=0$, $2$ if
$f = 1$, $1$ if $f=2$, or $3$ if $f=3$.

\section{Finding owner processes of elements}
\label{sec:ghost-findowner}

For the ghost algorithm, after we have successfully constructed an element's
(half) face-neighbor, we need to identify the owner process of this neighbor.
\begin{definition}
Let $E$ be an element in a (partitioned) forest. A process $p$ is an
\textbf{owner} of $E$ if there exists a leaf $L$ in the forest such that
\begin{enumerate}
 \item $L$ is in the partition of $p$, and
 \item $L$ is an ancestor or a descendant of $E$.
\end{enumerate}
Note that the owner of an arbitrary element is not unique. Unique ownership is, however, guaranteed
for leaf elements and their descendants.
Also, each element has at least one owner.
\end{definition}

In his section, we describe how to find all owner processes of a given element
and how to find those processes that own leaf elements sharing a given face with an element.

\subsection{\texttt{t8\_forest\_owner}}

We begin with the algorithm \texttt{t8\_forest\_owner} that determines all
owner processes of a given forest element.

\begin{definition}
The \textbf{first/last descendant} of an element $E$ is the descendant
of $E$ of maximum refinement level with smallest/largest SFC index.
\end{definition}

Since first/last descendants cannot be refined further, they are either a leaf or descendants
of a leaf. Hence, they have a unique owner process.
See also Figure~\ref{fig:owner-ex1} for an illustration.
We denote these owners by $p_\textrm{first}(E)$ and $p_\textrm{last}(E)$.
Since a forest is always partitioned along the SFC in ascending order, it must hold for each 
owner process $p$ of $E$ that
\begin{equation}
  \label{eq:ghost-owner}
p_\textrm{first}(E) \leq p \leq p_\textrm{last}(E).
\end{equation}
On the other hand, if a process $p$ fulfills inequality~\ref{eq:ghost-owner} and
its partition is not empty, then it must be an owner of $E$.
Furthermore, we conclude that an element has a unique owner if and only if
$p_\textrm{first}(E) = p_\textrm{last}(E)$.

\begin{figure}
  \center
  \def\svgwidth{0.5\textwidth}
\begingroup%
  \makeatletter%
  \providecommand\color[2][]{%
    \errmessage{(Inkscape) Color is used for the text in Inkscape, but the package 'color.sty' is not loaded}%
    \renewcommand\color[2][]{}%
  }%
  \providecommand\transparent[1]{%
    \errmessage{(Inkscape) Transparency is used (non-zero) for the text in Inkscape, but the package 'transparent.sty' is not loaded}%
    \renewcommand\transparent[1]{}%
  }%
  \providecommand\rotatebox[2]{#2}%
  \ifx\svgwidth\undefined%
    \setlength{\unitlength}{93.19533691bp}%
    \ifx\svgscale\undefined%
      \relax%
    \else%
      \setlength{\unitlength}{\unitlength * \real{\svgscale}}%
    \fi%
  \else%
    \setlength{\unitlength}{\svgwidth}%
  \fi%
  \global\let\svgwidth\undefined%
  \global\let\svgscale\undefined%
  \makeatother%
  \begin{picture}(1,0.94514086)%
    \put(0,0){\includegraphics[width=\unitlength,page=1]{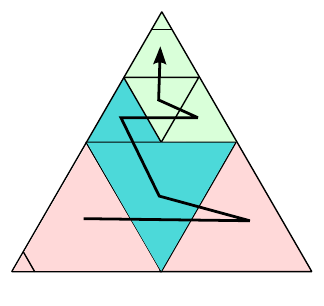}}%
    \put(0.1029275,0.02584887){\color[rgb]{0,0,0}\rotatebox{-0.16875751}{\makebox(0,0)[lb]{\smash{first descendant}}}}%
    \put(0.04890653,0.86824351){\color[rgb]{0,0,0}\makebox(0,0)[lb]{\smash{last descendant}}}%
    \put(0,0){\includegraphics[width=\unitlength,page=2]{owners_ex1_tex.pdf}}%
  \end{picture}%
\endgroup%
   \caption[An element $E$ and its leaf elements.]
  {An element $E$ and its leaf elements. We depict its first descendant 
  (bottom left) and last descendant (top). Their owners are unique and we
  denote them by $p_\mathrm{first}(E)$ (pink) and $p_\mathrm{last}(E)$ 
  (light green). For all other owners $p$---in this case, the process owning
  the blue leaves---we have $p_\textrm{first}(E) \leq p \leq p_\textrm{last}(E)$.
  }
  \figlabel{fig:owner-ex1}
\end{figure}

Each process can compute the SFC index of the first descendant of its first
local element. From these SFC indices we build an array of size $P$, which
is the same on each process. We can then
determine the owner process of a first or last descendant by performing a
binary search in this array if we combine it with the array of tree offsets.
This is the same approach as in \cite{BursteddeWilcoxGhattas11}.

Hence, we can compute all owner processes of an element by constructing its first and last descendant
and computing their owners.
When we know that an element has a unique owner---for example when it is a leaf
element---it suffices to construct its first descendant and compute its owner.

\subsection{Owners at a face}
\label{sec:ownersatface}
For the \texttt{Ghost\_v2} algorithm that works on an unbalanced forest---as
described in \cite{IsaacBursteddeWilcoxEtAl15} for cubical elements---we will
have to identify all owners of leaves at a face of a face-neighbor element of a
given element. In contrast to the algorithm \texttt{find\_range\_boundaries}
that the authors of \cite{IsaacBursteddeWilcoxEtAl15} use, we introduce the
algorithm \texttt{t8\_owners\_at\_face}. Given an element $E$ and a face $f$, 
\texttt{t8\_owners\_at\_face} determines the set $P_E$ of all processes that
have leaf elements that are descendants of $E$ and share a face with $f$.  It
is a recursive algorithm that we now describe in detail.

\begin{definition}
The first/last \textbf{face descendant} of an element $E$ at a face $f$ is the
descendant of $E$ of maximum refinement level that shares a subface with $f$
and has smallest/largest SFC index.
\end{definition}
We denote the owner processes of an element's first and last
face descendants by $p_\textrm{first}(E,f)$ and $p_\textrm{last}(E,f)$.
If these are equal to the same process $q$, 
we can return $q$ as the single owner at that face.

As opposed to the owners of an element, not all nonempty processes in the range
from $p_\textrm{first}(E,f)$ to $p_\textrm{last}(E,f)$ are necessarily owners
of leaves at the face of $E$; see for example face $f=0$ in
Figure~\ref{fig:ownersatface}. Here, $p_\textrm{first}(E,0) = 0$, $p_\textrm{last}(E,0) = 2$,
and the owners at the face are $\set{0, 2}$ despite process $1$ being nonempty.

It is thus not sufficient to determine all nonempty processes between
$p_\textrm{first}(E,f)$ and $p_\textrm{last}(E,f)$.
Hence, if $p_\textrm{first}(E,f) < p_\textrm{last}(E,f) - 1$, we enter a
recursion with all children of $E$ that lie on the face $f$. 
Thus, the recursion is guaranteed to terminate if the input element has
only descendants owned by a single process, which happens at the latest when
the input element is a leaf. However, it could terminate earlier for elements
whose descendants at the face $f$ are all owned by a single
process, or by two processes whose ranks differ by $1$.

We outline the algorithm in Algorithm~\ref{alg:ownersatface} and illustrate an
example in Figure~\ref{fig:ownersatface}.

\begin{algorithm}
  \DontPrintSemicolon
  \caption{\texttt{t8\_owners\_at\_face (Forest $\forest F$, element $E$, face\_index $f$)}}
  \label{alg:ownersatface}
  \algoresult{The set $P_E$ of all processes that own leaf elements that are descendants
  of $E$ and have a face that is a subface of $f$.}\;
  $P_E\gets\emptyset$\;
  $\mathrm{fd} \gets\,$\texttt{t8\_element\_first\_desc\_face ($E, f$)} 
    \Comment{First and last descendant of $E$ at $f$}
  $\mathrm{ld} \gets\,$\texttt{t8\_element\_last\_desc\_face ($E, f$)}\;  
  $p_\textrm{first} \gets\,$\texttt{t8\_forest\_owner ($\forest F$, fd)}
    \Comment{The owners of \texttt{fd} and \texttt{ld}}
  $p_\textrm{last} \gets\,$\texttt{t8\_forest\_owner ($\forest F$, ld)}\;
  \algoeifcom {\IfComment{Only $p_\textrm{first}$ and $p_\textrm{last}$ are owners 
  of leaves at $f$}} {$p_\textrm{first} \in \set{p_\textrm{last},p_\textrm{last} - 1} $ } 
  {
    \Return $\set{p_\textrm{first}, p_\textrm{last}}$ 
  }(\IfComment{There may be other owners. Enter the recursion.})
  { %
    $C_f[]\gets\,$\texttt{t8\_element\_children\_at\_face ($E, f$)}\;
    \algofor {$0\leq i < $\texttt{t8\_element\_num\_face\_children ($E, f$)}}
    {
      $j\gets $\texttt{child\_index} $(C_f[i])$\Comment{The child number relative to $E$.}
      $f' \gets\,$\texttt{t8\_element\_child\_face ($E, j, f$)} \Comment{The face number of the child}
      $P_E \gets P_E \abst\cup\,$\texttt{t8\_owners\_at\_face ($\forest F, C_f[i], f'$)}
    \Comment{Recursion}
    }
    \Return $P_E$
  }
\end{algorithm}

\begin{figure}
  \center
  \def\svgwidth{0.55\textwidth}
\begingroup%
  \makeatletter%
  \providecommand\color[2][]{%
    \errmessage{(Inkscape) Color is used for the text in Inkscape, but the package 'color.sty' is not loaded}%
    \renewcommand\color[2][]{}%
  }%
  \providecommand\transparent[1]{%
    \errmessage{(Inkscape) Transparency is used (non-zero) for the text in Inkscape, but the package 'transparent.sty' is not loaded}%
    \renewcommand\transparent[1]{}%
  }%
  \providecommand\rotatebox[2]{#2}%
  \ifx\svgwidth\undefined%
    \setlength{\unitlength}{93.19533691bp}%
    \ifx\svgscale\undefined%
      \relax%
    \else%
      \setlength{\unitlength}{\unitlength * \real{\svgscale}}%
    \fi%
  \else%
    \setlength{\unitlength}{\svgwidth}%
  \fi%
  \global\let\svgwidth\undefined%
  \global\let\svgscale\undefined%
  \makeatother%
  \begin{picture}(1,0.94514086)%
    \put(0,0){\includegraphics[width=\unitlength,page=1]{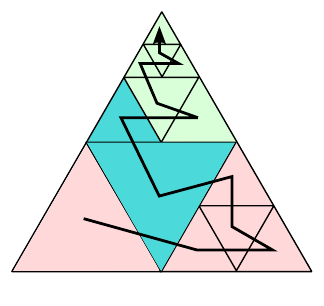}}%
    \put(0.17896196,0.17939129){\color[rgb]{0,0,0}\makebox(0,0)[lb]{\smash{$p=0$}}}%
    \put(0.41928747,0.27725086){\color[rgb]{0,0,0}\makebox(0,0)[lb]{\smash{$p=1$}}}%
    \put(0.53913259,0.52962814){\color[rgb]{0,0,0}\makebox(0,0)[lb]{\smash{$p=2$}}}%
    \put(0.77331875,0.51579107){\color[rgb]{0,0,0}\makebox(0,0)[lb]{\smash{$f = 0$}}}%
    \put(0.07587091,0.51258001){\color[rgb]{0,0,0}\makebox(0,0)[lb]{\smash{$f = 1$}}}%
    \put(0.42474783,0.04064462){\color[rgb]{0,0,0}\makebox(0,0)[lb]{\smash{$f = 2$}}}%
  \end{picture}%
\endgroup%
\hfill
  \def\svgwidth{0.28\textwidth}
\begingroup%
  \makeatletter%
  \providecommand\color[2][]{%
    \errmessage{(Inkscape) Color is used for the text in Inkscape, but the package 'color.sty' is not loaded}%
    \renewcommand\color[2][]{}%
  }%
  \providecommand\transparent[1]{%
    \errmessage{(Inkscape) Transparency is used (non-zero) for the text in Inkscape, but the package 'transparent.sty' is not loaded}%
    \renewcommand\transparent[1]{}%
  }%
  \providecommand\rotatebox[2]{#2}%
  \ifx\svgwidth\undefined%
    \setlength{\unitlength}{294.81169434bp}%
    \ifx\svgscale\undefined%
      \relax%
    \else%
      \setlength{\unitlength}{\unitlength * \real{\svgscale}}%
    \fi%
  \else%
    \setlength{\unitlength}{\svgwidth}%
  \fi%
  \global\let\svgwidth\undefined%
  \global\let\svgscale\undefined%
  \makeatother%
  \begin{picture}(1,1.85104584)%
    \put(0,0){\includegraphics[width=\unitlength,page=1]{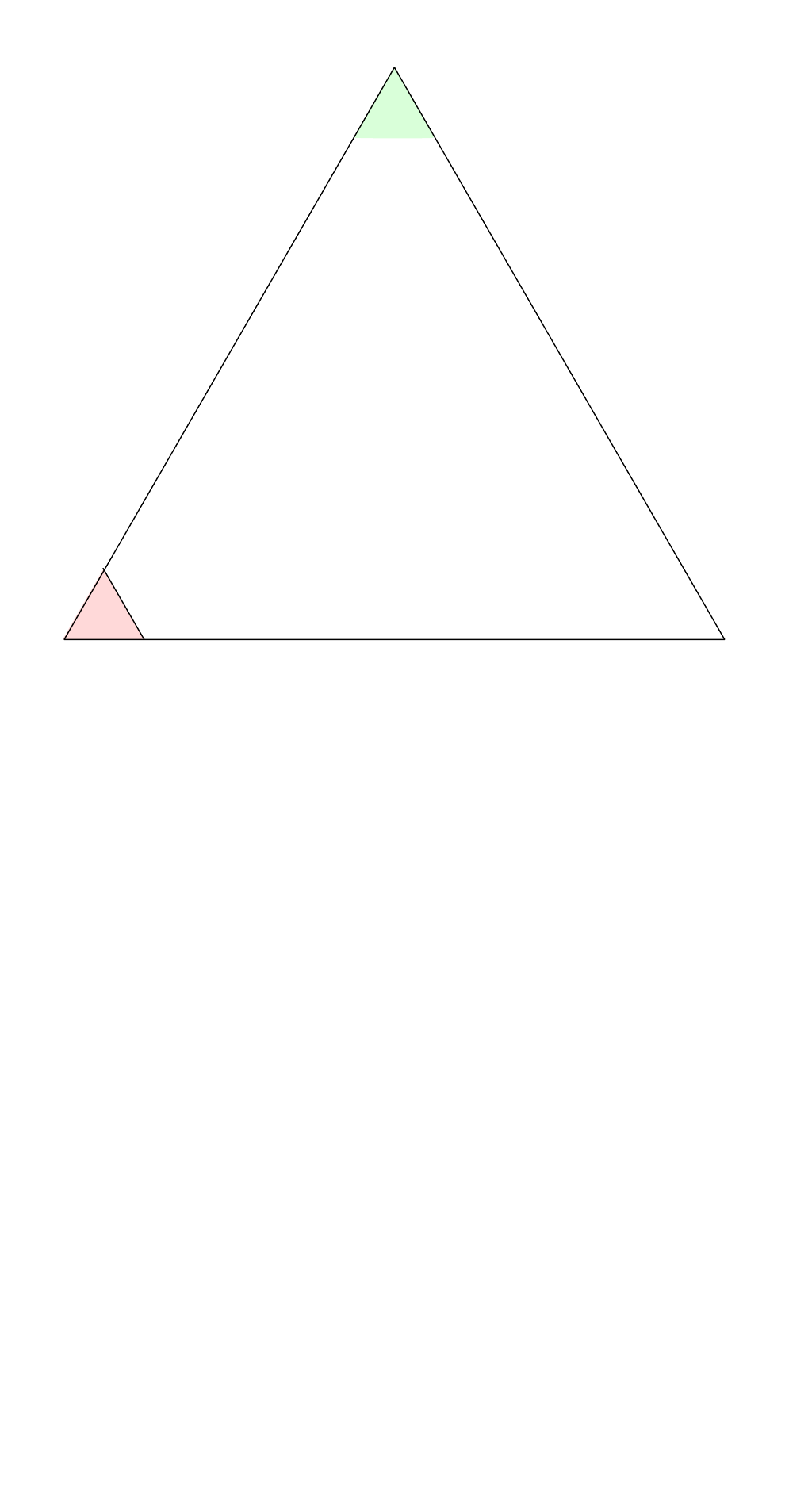}}%
    \put(0.1861849,1.11045593){\color[rgb]{0,0,0}\makebox(0,0)[lb]{\smash{0}}}%
    \put(0,0){\includegraphics[width=\unitlength,page=2]{owners_at_face_ex1_detail_b_tex.pdf}}%
    \put(0.46557029,1.59564253){\color[rgb]{0,0,0}\makebox(0,0)[lb]{\smash{2}}}%
    \put(0.05497321,1.49788517){\color[rgb]{0,0,0}\makebox(0,0)[lb]{\smash{$f = 1$}}}%
    \put(0,0){\includegraphics[width=\unitlength,page=3]{owners_at_face_ex1_detail_b_tex.pdf}}%
    \put(0.15098151,0.96949184){\color[rgb]{0,0,0}\makebox(0,0)[lb]{\smash{\small{first face descendant}}}}%
    \put(0,0){\includegraphics[width=\unitlength,page=4]{owners_at_face_ex1_detail_b_tex.pdf}}%
    \put(0.53683924,1.74521062){\color[rgb]{0,0,0}\makebox(0,0)[lb]{\smash{\small{last face}}}}%
    \put(0,0){\includegraphics[width=\unitlength,page=5]{owners_at_face_ex1_detail_b_tex.pdf}}%
    \put(0.19431916,0.18348732){\color[rgb]{0,0,0}\makebox(0,0)[lb]{\smash{0}}}%
    \put(0.27029986,0.3373873){\color[rgb]{0,0,0}\makebox(0,0)[lb]{\smash{0}}}%
    \put(0.38969799,0.56610424){\color[rgb]{0,0,0}\makebox(0,0)[lb]{\smash{1}}}%
    \put(0.46606641,0.6948064){\color[rgb]{0,0,0}\makebox(0,0)[lb]{\smash{2}}}%
    \put(0.56248599,1.6669476){\color[rgb]{0,0,0}\makebox(0,0)[lb]{\smash{\small{descendant}}}}%
  \end{picture}%
\endgroup%
   \caption[An example for \texttt{t8\_owners\_at\_face}]{An example for Algorithm~\ref{alg:ownersatface}, \texttt{t8\_forest\_owners\_at\_face}.
   Left: A triangle element $E$ with the TM-index as SFC whose descendants 
   are owned by three different processes: $0$ (red), $1$ (blue), and
   $2$ (green). The owners at the faces are $\set{0, 2}$ at face $0$,
   $\set{0, 1, 2}$ at face $1$, and $\set{0}$ at face $2$.
   Right: The iterations of \texttt{t8\_forest\_owners\_at\_face} at face $f =
   1$.  At first the first and last descendant of $E$ at $f$ are
   constructed. We compute their owner processes $0$ and $2$, and since their
   difference is greater one, we continue the recursion. In the second
   iteration the algorithm is called once for the lower left child and once for
   the upper child of $E$. We determine their first and last descendants at the
   respective subface of $f$.  For the lower left child, the recursion stops since
   both face descendants are owned by process $0$. For the upper child the owner
   processes are $1$ and $2$ and since there are no other possible owner processes
   in between, we stop the recursion as well.
  }
  \figlabel{fig:ownersatface}
\end{figure}

\subsubsection{Notes on the implementation}
In our implementation of \texttt{t8\_owners\_at\_face} we take into account that
the first and last owners $p_f$ and $p_l$ at the current recursion step form
lower and upper bounds for the first and last owners in any upcoming recursion
step. Thus, we restrict the binary searches in \texttt{t8\_forest\_owner}
to the interval $[p_f, p_l]$ instead of $[0, P-1]$.

We also exploit that the first descendant of an element $E$ at a face $f$ is also
the first face descendant of $E$'s first child at $f$. The same holds
for the last descendant and the last child at $f$.
Thus, we reuse the first/last face descendants and owners of $E$ when we enter
the recursion with the first/last child at $f$.

\section{The ghost algorithms}

For \ghostb we assume that the forest is balanced and hence we know that all
face-neighbor leaf elements of $E$ have a refinement level between $\ell(E)
-1$ and $\ell(E) + 1$. Therefore, all neighbor elements of $E$ with level
$\ell(E) + 1$ must have a unique owner process.
Thus, to identify the remote processes of $E$ at $F$ it suffices
to construct the face-neighbors of $E$ across $F$ of level $\ell(E) + 1$
and determine their owner processes.
We construct these face-neighbors via the function
\texttt{t8\_forest\_half\_face\_neighbors}; see Section~\ref{sec:ghost-halfneighbors}.
We present the complete \ghostb algorithm in Algorithm~\ref{alg:ghostb}.

For \texttt{Ghost\_v2} we drop the assumption of a balanced forest and 
thus there is no a priori knowledge about the face-neighbor leaves of $E$.
To compute $E$'s remote processes across $F$, we first construct the
corresponding face-neighbor $E'$ of the same-level as $E$.
For this element we know that it is either a descendant of a
forest leaf, which then has a unique owner, or an ancestor of
multiple forest leaves, which could all have different owners.
We need to compute only the owners of those descendant/ancestor forest
leaves of $E'$ that touch the face $F$. We achieve this with
Algorithm~\ref{alg:ownersatface} \texttt{t8\_forest\_owners\_at\_face} that we
describe in Section~\ref{sec:ownersatface}. 
We show the complete \texttt{Ghost\_v2} Algorithm in \ref{alg:ghost}.  It uses
the function \texttt{dual\_face} which, given an element $E$ and a face index
$f$, simply returns the face index $f'$ from the neighboring element.

\begin{algorithm}
\DontPrintSemicolon
\caption{\ghostb (\texttt{Forest $\forest F$}) (for balanced forests only)}
\label{alg:ghostb}
\algoresult{The ghost layer of $\forest F$ is constructed.}\;
Ghost\_init ()\;
\algofor{$K\in\, \forest F$\texttt{.trees}}
 { %
 \algofor{$E \in\,$\texttt{K.elements}}
  { %
   \algofor {$0\leq f < \,$\texttt{t8\_element\_num\_faces($E$)}}
    { %
      $E'[]\gets\,$\texttt{t8\_forest\_half\_face\_neighbors ($\forest F, E, f$)}\;
      \algofor {$0\leq i < $\texttt{t8\_element\_num\_face\_children($E, f$)}}
      {   %
      $q \gets $\texttt{t8\_forest\_owner ($\forest F, E'[i]$)}\;
        \algoif {$q \neq p$}
        {
          $R_p^q = R_p^q \cup \set {E}$\;
        }
      }
    }
  }
 }
 Ghost\_communicate ()\;
 \vspace{1ex}
   \tcc*[l]{We outsource the init and communication routine, for later reuse}
 \setcounter{AlgoLine}{0}
 \textbf{Function} \texttt{Ghost\_init}\;
\algofor{$0\leq q < P$} {$R_p^q =\emptyset$}
 \textbf{Function} \texttt{Ghost\_communicate}\;
 $\forest F$.\texttt{ghosts}$\gets\emptyset$\;
 \algofor{$\set{q \,|\, R_p^q \neq \emptyset}$}
  {
   Send $R_p^q$ to $q$\;
   Receive $R_q^p$ from $q$\;
   $\forest F$.\texttt{ghosts}$\gets\forest F$.\texttt{ghosts}$\, \cup R_q^p$\;
  }
\end{algorithm}

\begin{algorithm}
\DontPrintSemicolon
\caption{\texttt{Ghost\_v2} (\texttt{Forest $\forest F$})}
\label{alg:ghost}
\algoresult{The ghost layer of $\forest F$ is constructed.}\;
Ghost\_init ()\;
\algofor{$K\in\, \forest F$\texttt{.trees}}
 { %
 \algofor{$E \in\,$\texttt{T.elements}}
  { %
   \algofor {$0\leq f < \,$\texttt{t8\_element\_num\_faces ($E$)}}
    { %
      $E'\gets\,$\texttt{t8\_forest\_face\_neighbor ($\forest F, E, f$)}\;
      $f'\gets\,$\texttt{dual\_face($E, E', f$)}\;
      $P_{E'} \gets $\texttt{t8\_forest\_owners\_at\_face ($\forest F, E', f'$)}\;
      \algofor {$q\in P_{E'}$}
      {   %
        \algoif {$q \neq p$}
        {
          $R_p^q = R_p^q \cup \set {E}$\;
        }
      }
    }
  }
 }
 Ghost\_communicate ()\;
\end{algorithm}

\section{Optimizing the runtime of \texttt{Ghost}}

The \ghostb and \texttt{Ghost\_v2} algorithms that we present here both iterate 
over all local leaf elements to identify the boundary leaves on the process's
boundary. For each leaf we generate all (half) face-neighbors and compute their
owners.
However, for most meshes only a portion of the leaf elements actually are
boundary elements, depending on the surface-to-volume ratio of the process's
partition. Since the surface of a volume grows quadratically while the volume itself
grows cubically, the number of boundary leaves can become arbitrarily small in
comparison to the number of all leaves.

We thus aim to improve the runtime of the algorithms by excluding inner leaves
from the iteration. 
In \pforest the inner leaves are excluded from the iteration by checking
for each quadrilateral/hexahedron whether its $3\times 3$ neighborhood,
thus all same-level face-(edge-/vertex-)neighbors, are process local.
Since this approach particularly uses geometrical properties of the 
quadrilaterals/hexahedra and of the Morton SFC, it is not practical
for our element-type independent approach.

To exclude the inner leafs, we exchange the leaf iteration with a 
top-town search using the recursive approach from
\cite{IsaacBursteddeWilcoxEtAl15}. Starting with a tree's root element, we
check whether it may have boundary leaf descendants, and if so, we create the
children of the element and continue recursively. If we reach a leaf element,
we check whether it is a boundary element---and if so for which processes---in
the way described in the previous section.
 This approach allows us to terminate the recursion as soon as we 
reach an element that lies completely within the process's partition, thus
saving the iteration over all descendant leaves of that element.

We now discuss the details of the top-down search and how we use it to improve
the ghost algorithm.

\subsection{The recursive top-down search}

In \cite{IsaacBursteddeWilcoxEtAl15} the authors present the general recursive
\texttt{search} algorithm for octree AMR, which easily extends to arbitrary
tree-based AMR. The setting is that we search a leaf or a set of leaves in a
forest that satisfy given conditions. One numerical example for such a search
arises in semi-Lagrangian advection solvers
\cite{Albrecht16,MirzadehGuittetBursteddeEtAl16}. To interpolate the values of
an advected function $\phi_t$ at time
$t$, each grid point $x_i$ is tracked back in time to its previous position
$\hat x_i$ at $t-\Delta t$. This point $\hat x_i$ lies in a leaf element $E_i$
of the forest and an element-local Hermite interpolation with the values of
$\phi_{t-1}$ is used to determine the value $\phi_t(x_i)$.
Thus, in each time step, we have to search the forest for the leaf elements $\set{E_i}$ given the
points $\set{\hat x_i}$.

In our case, we apply \texttt{search} to the problem of identifying all leaf
elements at a process's boundary. The \texttt{search} algorithm has been shown
to be especially efficient when looking for multiple matching leaves at once
\cite{IsaacBursteddeWilcoxEtAl15}, which is the case in our setting.

As presented in \cite{IsaacBursteddeWilcoxEtAl15} the idea of search is to
perform a recursive top-down traversal for each tree by
starting with the root element of that tree and recursively creating its children
until we end up with a leaf element. On each intermediate element we call a
user-provided callback function which returns true only if the search should
continue with this element. If the callback returns false, the recursion for
this element stops and its children are excluded from the search.
If the search has reached a leaf element, the callback also performs the
desired operations if the leaf matches the search.

For our ghost algorithm the callback returns false for elements that lie
entirely within the process's domain, thus excluding possibly large areas from
the search and hence speeding up the computation.
Once a leaf element is reached, we check whether it is a boundary element or not.
Thus, we iterate over the leaf's faces and compute the owners at the respective
neighbor faces as in the inner for loop of Algorithm~\ref{alg:ghost}.

We show our version of \texttt{search} in Algorithm~\ref{alg:search}.
It is a simplified version of Algorithm~3.1 in \cite{IsaacBursteddeWilcoxEtAl15}
without queries, since we do not need these for \texttt{Ghost}.
We also use the function \texttt{split\_array} from
\cite{IsaacBursteddeWilcoxEtAl15}. This function takes as input an element $E$
and an array $L$ of (process local) leaf elements in $E$, sorted in SFC order.
\texttt{split\_array} returns a set of arrays $\set{M[i]}$, such that for the
i-th child $E_i$ of $E$ the array $M[i]$ contains exactly the leaves in $L$ that
are also leaves of $E_i$. Thus, $L = \dot\bigcup_i M[i]$.

For a search of the complete forest, we iterate over all trees and in each tree
we compute the finest element $E$ such that all tree leaves are still
descendants of $E$. We compute $E$ as the nearest common ancestor of the first
and last leaf element of the tree. With this $E$ and the leaf elements of the
tree, we call \texttt{element\_search}. See Algorithm~\ref{alg:forest_search}.

\begin{algorithm}
  \DontPrintSemicolon
  \caption{\texttt{element\_search} (Element $E$, Leaf elements $L$, Callback \texttt{Match})
  [See Algorithm 3.1 in \cite{IsaacBursteddeWilcoxEtAl15}]}
  \label{alg:search}
  \algoresult{\texttt{Match} is called with $E$ as input. If the result is
  true, we continue recursively with $E$'s children.}\;
  \algoif {$L = \emptyset$} {\Return}
  isLeaf $\gets L = \set{E}$ \Comment{Boolean to determine whether $E$ is a leaf element}
  \algoifcom{\IfComment{Decide whether to continue recursion}}{\texttt{Match($E$, isLeaf)}\algoand \textbf{not} \texttt{isLeaf}}
  {
    $M[] \gets $ \texttt{split\_array} ($L$, $E$)\;
    $C[] \gets $ \texttt{t8\_element\_children} ($E$)\;
    \algofor {$0\leq i < $\texttt{t8\_element\_num\_children} ($E$)}
    {
      \texttt{element\_search} ($C[i]$, $M[i]$, \texttt{Match})\;
    }
  }
\end{algorithm}

\begin{algorithm}
  \DontPrintSemicolon
  \caption{\texttt{t8\_forest\_search} (Forest $\forest F$, Callback \texttt{Match})}
  \label{alg:forest_search}
  \algoresult{\texttt{element\_search} is called on each tree.}\;
  \algofor {$K \in \forest F.\mathrm{trees}$}
  {
    $E_1 \gets $ \texttt{first\_tree\_element} ($\forest F, K$) 
    \Comment{First and last local leaf}
    $E_2 \gets $ \texttt{last\_tree\_element} ($\forest F, K$)
    \Comment{in the tree}
    $E \gets $ \texttt{t8\_element\_nearest\_common\_ancestor} ($E_1, E_2$)\;
    $L \gets $ \texttt{tree\_leaves} ($\forest F, K$) 
    \Comment{Array of tree leaves}
    \texttt{element\_search} ($E, L$, \texttt{Match})\;
  }
\end{algorithm}

\subsection{The optimized \texttt{Ghost} algorithm}
We use \texttt{forest\_search} for an optimized version of \texttt{Ghost}.
When iterating over all leaves of the forest and checking the neighbors for each
one, a lot of these elements are in the interior domain of the process.
By using search we can exclude a set of interior leaves as soon as the search
recursion enters an ancestor that is completely in the interior of the domain.

We show our callback algorithm \texttt{t8\_ghost\_match} in
Algorithm~\ref{alg:ghost_match}, which works as follows. If the element $E$
which is passed to \texttt{t8\_ghost\_match} is not a leaf element, we check
whether the element and all of its possible face-neighbors are owned by the
current process.
For the element's owners, we do not call the function \texttt{t8\_forest\_owner}, but instead
save runtime by computing the first and last process that own leaves of the
element and checking whether they are equal.
For these computations we construct $E$'s first and last descendant.
Analogously, for the owners at the neighbor faces we compute the first and last
owner processes.  If for $E$ the first and last process is $p$ and at each
face-neighbor the first and last owner at the corresponding face is also $p$,
$E$ is an inner element and cannot have any boundary leaves as descendants. Thus,
we return 0 and the search does not continue for the descendants of $E$.

If $E$ is a leaf element, then it may or may not be a boundary element.
We thus compute all owner processes for all face-neighbors using \texttt{t8\_forest\_owners\_at\_face}
and add $E$ as a boundary element to all of these that are not $p$.

\begin{algorithm}
  \caption{\texttt{t8\_ghost\_match} (Element $E$, Bool \texttt{isLeaf})}
  \label{alg:ghost_match}
  \DontPrintSemicolon
  \algoresult{If $E$ is a leaf, compute the owners of the face-neighbors
and add to the sets $R_p^q$. If not, query whether all descendants of $E$ 
and all face-neighbors are owned by $p$.}\;

  \algoeifcom {\IfComment{$E$ is a leaf. Compute the owners at}}{\texttt{isLeaf}}
  {
  \algoforcom{\IfComment{the face and add $E$ as boundary.}}{$0\leq f < \,$\texttt{t8\_element\_num\_faces ($E$)}}
    { %
      $E' \gets\, $ \texttt{t8\_forest\_face\_neighbor ($\forest F, E, f$)}\;
      $f'\gets\, $ \texttt{dual\_face($E, E', f$)}\;
      $P_{E'} \gets $ \texttt{t8\_forest\_owners\_at\_face ($\forest F, E', f'$)}\;
      \algofor {$q\in P_{E'}$}
      {   %
        \algoif {$q \neq p$}
        {
          $R_p^q = R_p^q \cup \set {E}$\;
        }
      }
    }
  }(\IfComment{$E$ is not a leaf.})
  {
   $p_\textrm{first}(E) \gets$ \texttt{t8\_element\_first\_owner} ($E$)\;
   $p_\textrm{last}(E) \gets$ \texttt{t8\_element\_last\_owner} ($E$)\label{algline:ghostmatchrec}\;
   \algoifcom{\IfComment{No leaf of $E$ is owned by $p$}}{$p_\textrm{first}(E)>p$\algor$p_\textrm{last}(E)<p$}
   {
    \Return 1\;
  }
    \algofor {$0\leq f < \,$\texttt{t8\_element\_num\_faces ($E$)}
    \label{algline:ghostmatchfor}}
    { %
      $E' \gets\, $ \texttt{t8\_forest\_face\_neighbor ($\forest F, E, f$)}\;
      $f'\gets\, $ \texttt{dual\_face($E, E', f$)}\;
      $p_\textrm{first}(E',f') \gets $ \texttt{t8\_first\_owner\_at\_face ($\forest F, E', f'$)}\;
      $p_\textrm{last}(E',f') \gets $ \texttt{t8\_last\_owner\_at\_face ($\forest F, E', f'$)}\;
      \algoif{$ p_\textrm{first}(E',f') \neq p$\algor $p_\textrm{last}(E',f') \neq p$}{
        \Return 1
\Comment{Not all face-neighbor leaves are owned by $p$}
      }
    }
    \algoif{$p_\textrm{first}(E) = p_\textrm{last}(E) = p$}
    {
      \Return 0\;
    }
  }
  \Return 1\;
\end{algorithm}

\begin{algorithm}
\caption{\texttt{Ghost\_v3} (Forest $\forest F$)}
\label{alg:ghost_v3}
\DontPrintSemicolon
\algoresult{The ghost layer of $\forest F$ is constructed.}\;
Ghost\_init ()\;
\texttt{t8\_forest\_search} ($\forest F$, \texttt{t8\_ghost\_match})\;
Ghost\_communicate ()\;
\end{algorithm}

\subsubsection{Implementation details}
For each child $C$ of an element $E$ the ranks $p_\mathrm{first}(E),
p_\mathrm{last}(E), p_\mathrm{first}(E,f)$, and $p_\mathrm{last}(E,f)$ serve as
lower and upper bounds for the corresponding ranks for $C$.
Thus, in our implementation of \texttt{ghost\_match} in \tetcode, we store these ranks 
for each recursion level reducing the search range for the binary owner search
for $C$ from $[0, P-1]$ to $[p_\mathrm{first}(E), p_\mathrm{last}(E)]$, and
to $[p_\mathrm{first}(E,f), p_\mathrm{last}(E,f)]$ for the faces.
To compute these bounds it is necessary to always enter the \texttt{for}-loop
in Line~\ref{algline:ghostmatchfor}, even though we do not exercise this
in Algorithm~\ref{alg:ghost_match}.

\section{Numerical comparison of the ghost versions}

To verify that the additional complexity of implementing the top-down search
is worth the effort, we perform runtime tests of the different ghost methods.

We perform tests with hexahedral and tetrahedral elements, each time on a unit
cube geometry. For hexahedra the unit cube is modeled with a single tree and
for tetrahedra with six trees with a common diagonal as in
Figure~\ref{fig:sechstetra}.
For each element type we run two types of tests, one with a uniform mesh and one
with an adaptive mesh, where we use a regular refinement pattern, refining
every third element (in SFC order) recursively from level $\ell$ up to a level
$\ell + k$; see Figure~\ref{fig:ghost-comparetest}.

Since we are interested in comparing the algorithms and not in 
their particular extreme scaling behavior, we run the tests on 1024 MPI ranks
on JUQUEEN~\cite{Juqueen}. We refer to Chapter~\ref{ch:balance} for more
elaborate scaling tests of \texttt{Ghost} on significantly more ranks (up
to 458k).
For the tests in this section we use 64 compute nodes with 16 cores and
16 GB memory each. We use 1 rank per core, thus 16 MPI ranks per node.
We display our results in Table~\ref{tab:ghost-comparetest} showing runtime
results for uniform levels $\ell$ equal to $9$, $8$ and $4$, and adaptive levels 
$\ell$ equal to $8$, $7$, and $3$ (tetrahedra), respectively $4$ (hexahedra),
with $k = 2$.

As expected, the iterative versions of \texttt{Ghost} scale linearly with the
number of elements. The improved version of \texttt{Ghost}, however, scales
with the number of ghost elements, which grows less quickly compared
to the number of elements. 
From this we conclude that we indeed skip most of the elements that do not lie
on the boundary of a process's domain.
The improved version shows overall a significantly better performance and is up
to a factor of 23.7 faster (adaptive tetrahedra, level $8$) than the iterative
version. For smaller or degraded meshes where the number of ghosts is on the
same order as the number of leaf elements, the improved version shows no
disadvantage compared to the iterative version. This shows that we do not loose
runtime to the \texttt{Search} overhead, even if each element is a boundary
element. For small meshes all algorithms show negligible runtime on the order
of milliseconds.

We conclude that the ghost version with top-down search is the ideal choice among
the three version that we discuss. From now on, we use this algorithm for all
tests.

\begin{figure}
\center
\includegraphics[width=0.49\textwidth]{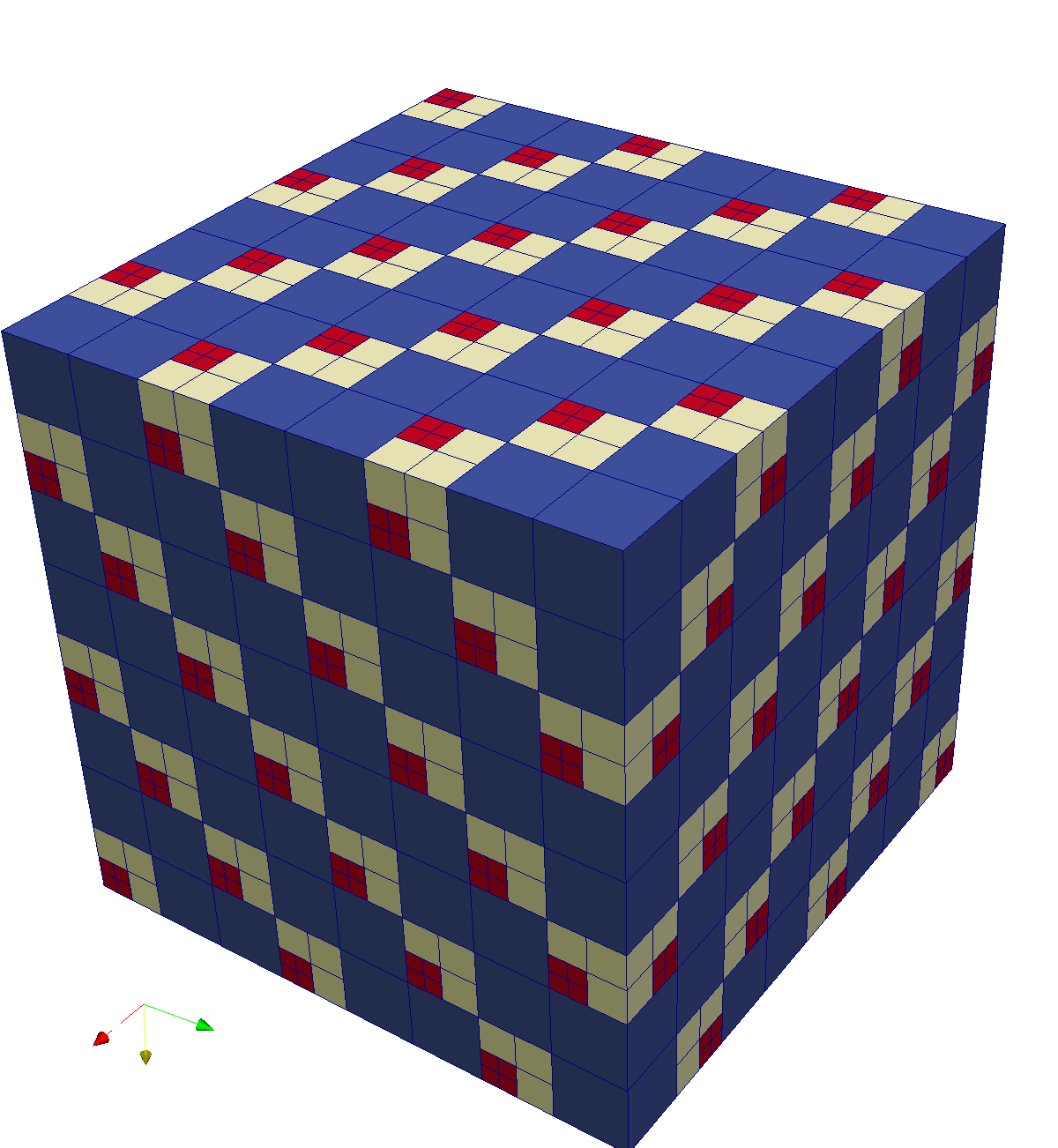}
\includegraphics[width=0.49\textwidth]{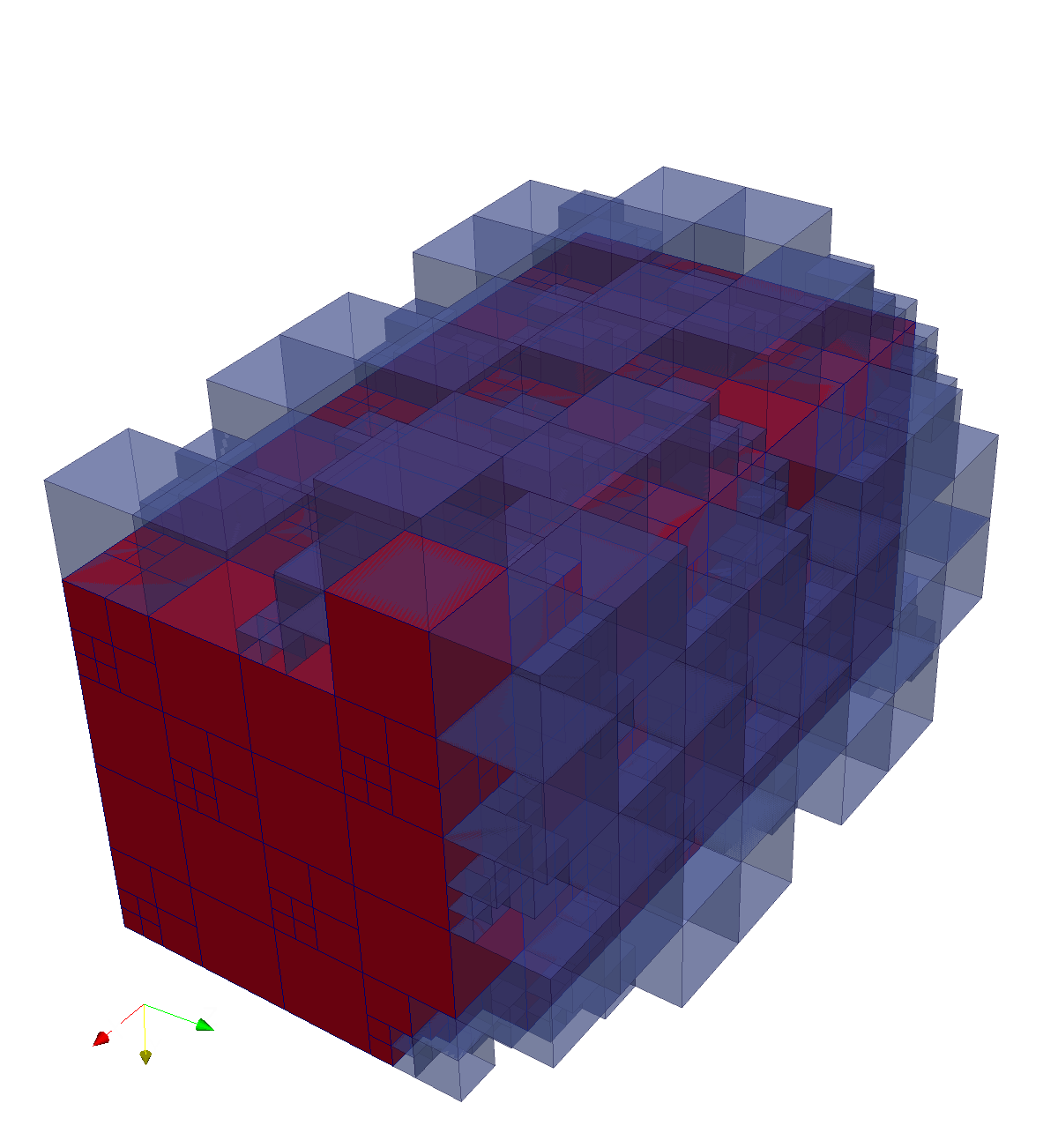}
\caption[Comparing the different implementations of \texttt{Ghost}.]
{We compare the different implementations of \texttt{Ghost} by
testing them on a unit cube geometry with 1024 MPI ranks.
Left: The adaptive mesh with minimum level $\ell = 3$ for hexahedra.
We refine every third element in SFC order and repeat the process once with
the refined elements. 
Right: For an adaptive computation with 4 MPI ranks, we show the local leaf
elements of the process with MPI rank 1 (red) and its ghost elements (blue, transparent). }
\figlabel{fig:ghost-comparetest}
\end{figure}

\begin{table}
\center
 \begin{tabular}{|c||r|r|r||r|r|r|}\hline
  \multicolumn{7}{|c|}{tetrahedra}\\ \hline
  & \multicolumn{3}{c||}{uniform}
  & \multicolumn{3}{c|}{adaptive}\\ \hhline{|=||===||===|}
  $\ell$ & 9 & 8 & 4 & 8--10 & 7--9 & 3--5 \\ \hline
  elements/proc & 786,432 & 98,304 & 24 &1,015,808 & 126,976 & 31 \\ \hline
  ghosts/proc   & 32,704  &  8,160 & 30 &   31,604 &   8,137 & 56 \\ \hhline{|=||===||===|}
  Ghost\_v1 [s]    & 172.3  &  21.64 & 7.99\e-3 & -&-&-\\ \hline
  Ghost\_v2 [s]    & 129.6  &  16.19 & 5.93\e-3 & 167.94 & 20.88 & 8.10\e-3 \\ \hline
  Ghost\_v3 [s]    &  7.41  &   1.75 & 5.01\e-3 &   7.08 &  1.69 & 8.12\e-3 \\ \hline
 \end{tabular}\\[1ex]

 \begin{tabular}{|c||r|r|r||r|r|r|}\hline
  \multicolumn{7}{|c|}{hexahedra}\\ \hline
  & \multicolumn{3}{c||}{uniform}
  & \multicolumn{3}{c|}{adaptive}\\ \hhline{|=||===||===|}
  $\ell$ & 9 & 8 & 4 & 8--10 & 7--9 & 4--6 \\ \hline
  elements/proc & 131,072 & 16,384 & 4 & 169,301 & 21,162 & 41\\ \hline
  ghosts/proc   & 8,192   &  2,048 & 8 &   7,681 &  1,913 & 30\\ \hhline{|=||===||===|}
  Ghost\_v1 [s]    & 29.51 & 3.742 & 2.87\e-3 & -&-&-\\ \hline
  Ghost\_v2 [s]    & 18.25 & 2.302 & 2.32\e-3 & 23.79 & 2.964 & 8.01\e-3 \\ \hline
  Ghost\_v3 [s]    &  3.14 & 0.711 & 2.90\e-3 &  2.81 & 0.649 & 8.12\e-3 \\ \hline
 \end{tabular}
\caption[Runtime tests for the three different \texttt{Ghost} algorithms.]
{Runtime tests for the three different \texttt{Ghost} algorithms that we describe
in this chapter.
We run the tests on JUQUEEN with 1024 MPI ranks and 16 MPI ranks per compute node.
The domain geometry is a unit cube modeled by one tree in the hexahedral case
and six trees in the tetrahedral case. With each element type we test a uniform level $\ell$
mesh and a mesh that adapts every third element of a uniform level $\ell$ mesh
up to level $\ell + 2$; see Figure~\ref{fig:ghost-comparetest}. 
The different \texttt{Ghost} methods are: \ghostb that works on balanced
forests only; \texttt{Ghost\_v2} that works also on unbalanced forests;
\texttt{Ghost\_v3} that utilizes \texttt{search} to improve the runtime.
Since the adaptive forests are not balanced, we do not test \ghostb in that
case. The table at the top shows the mesh sizes and runtimes for tetrahedra
while the bottom table shows the data for hexahedra. We observe that our new
\texttt{Ghost\_v3} is superior to the other versions by a factor of up to 23
and scales with the number of ghosts and not with the number of elements.
}
\figlabel{tab:ghost-comparetest}
\end{table}
 \chapter{2:1 Balance}
\label{ch:balance}

It is common for AMR applications, such as finite element and finite volume
solvers, to rely on a forest mesh that is 2:1 balanced, meaning that each element
only has neighbors whose levels differ by at most one ($\pm 1$) from the
element's level \cite{CohenKaberMuellerEtAl03, MuellerStiriba07,
LahnertBursteddeHolmEtAl16, BursteddeStadlerAlisicEtAl13}; see also
Definition~\ref{def:balance}.
This restricts the number and configurations of hanging nodes/edges/faces that 
can occur, simplifying the necessary interpolation schemes and reducing the
number of neighboring processes.

However, mesh refinement and coarsening in applications is usually driven
by some kind of error estimator and/or geometric constraints and such
adaptation rules may not produce balanced meshes on their own. It can become a
significant challenge to change the adaptation rule such that it produces
balanced meshes that still respect the desired constraints.

We thus aim to decouple the operation of balancing a forest from the adaptation
routine. An application can then call \texttt{Adapt} to modify the mesh and
optionally call the algorithm \texttt{Balance} afterwards to reestablish a
balanced mesh. In this section we discuss our implementation of
\texttt{Balance}.

Note that there are also applications that can handle arbitrarily hanging
nodes at elements with refinement level difference greater than one
\cite{SolinCervenyDolezel08, GuittetTheillardGibou15}. In general it is up to
an application whether to use \texttt{Balance} or not.

The algorithm \texttt{Balance} gets as input a forest that may be unbalanced
and modifies it by successively refining leaf elements such that it becomes
balanced.
\texttt{Balance} should not coarsen any elements in order to guarantee that an
application can keep its desired accuracy.

As with \texttt{Ghost}, we distinguish between corner-balance, edge-balance,
and face-balance, regarding the different possible neighbor connections.
As we mention before, we restrict ourselves to face-neighbors and thus we
consider face-balance here, sometimes also referred to as 1-balance
\cite{IsaacBursteddeGhattas12}.

In \texttt{Balance}, a leaf element with a large refinement level that is
surrounded by leaves of smaller refinement levels can trigger refinement of
leaves over large regions that may stretch across multiple process
boundaries; see Figure~\ref{fig:balance_ex1}. This is one of the reasons why
\texttt{Balance} was shown to be the most expensive high-level algorithm 
\cite{BursteddeWilcoxGhattas11}. The relatively high run time costs of
\texttt{Balance} have sparked efforts to optimize and speed up the algorithm
\cite{IsaacBursteddeGhattas12}.

In this thesis, we restrict ourselves to a straightforward implementation of
\texttt{Balance} via the existing algorithms \texttt{Adapt} and \texttt{Ghost}.
The idea is similar to the ripple algorithm
from~\cite{TuOHallaronGhattas05,TuOHallaron04}. We see the implementation that
we give here as a feasibility study of \texttt{Balance} for meshes with
arbitrary element types and do not claim to achieve an optimal runtime.
We thus also refer to our algorithm as \texttt{Ripple-balance}.
Implementing an optimized algorithm in the spirit
of~\cite{IsaacBursteddeGhattas12} remains a challenge for future work.

\begin{figure}
\center
\includegraphics[width=0.49\textwidth]{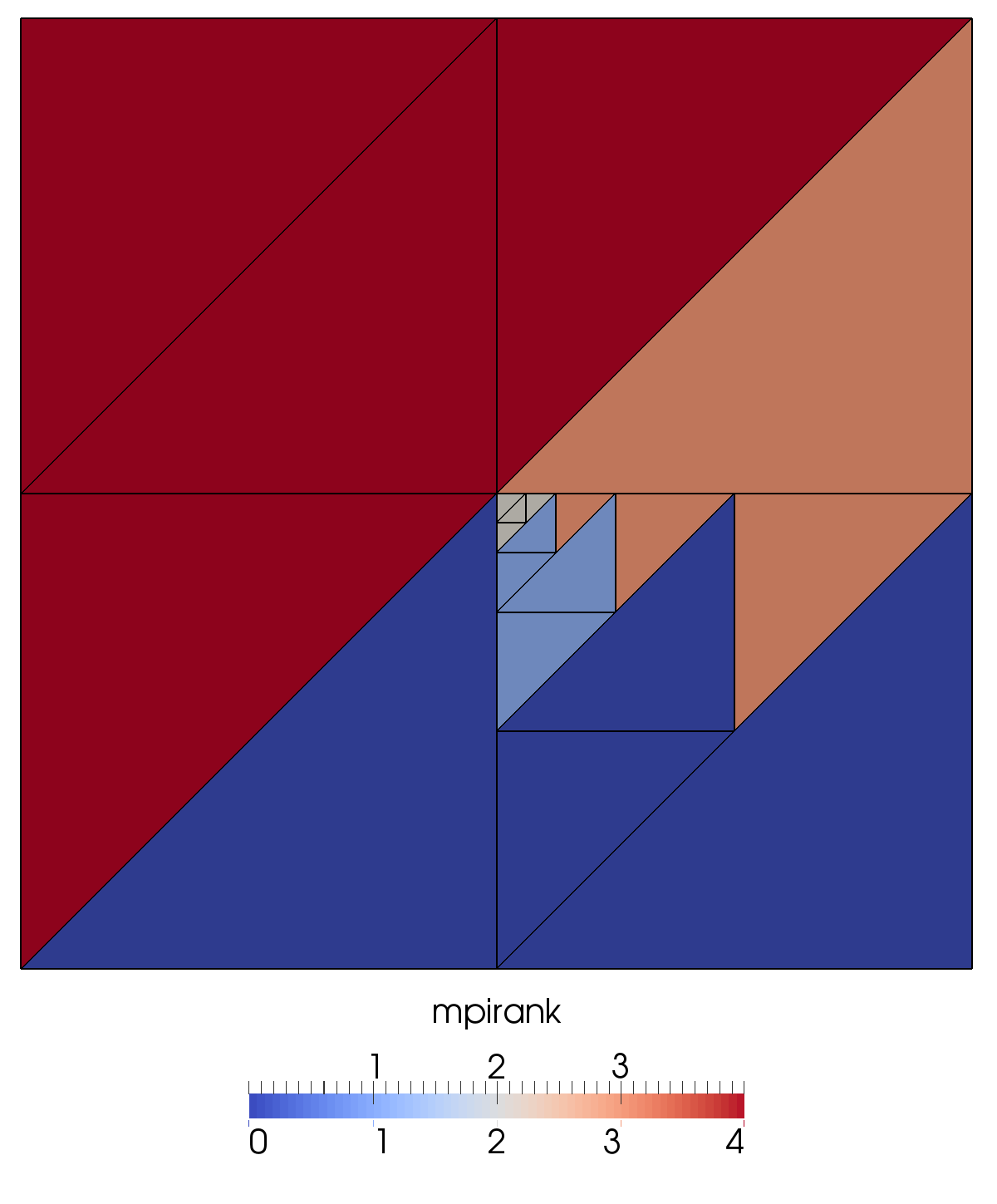}
\hfill
\includegraphics[width=0.49\textwidth]{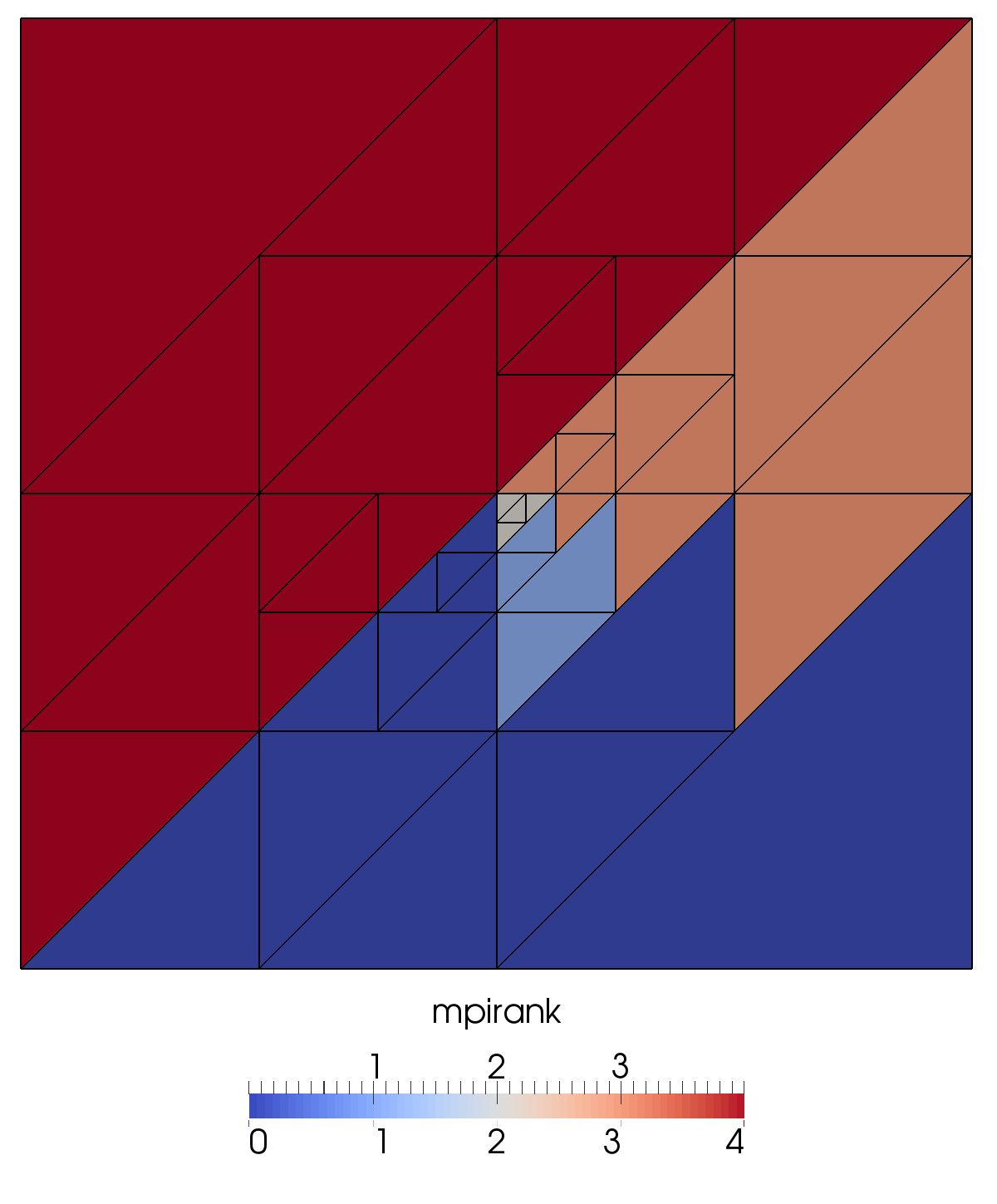}
\caption[An example for face-balance with a triangular forest and 5 MPI
ranks.]{An example for face-balance with a triangular forest and 5 MPI ranks.
Left: An unbalanced and fairly equally partitioned forest with two trees.
Right: After balancing the forest, each leaf is refined in such a way that no
  two face-neighboring leaf elements have a level difference of more than one.
  Note that the finest leaf elements reside on rank 2 and influence the
  refinement of leaves on ranks 0, 3, and 4. We also observe that the
  load-balance has been disturbed.
  To prevent this, we may add additional partition steps during and/or after
  \texttt{Balance}.}
\label{fig:balance_ex1}
\end{figure}

\section{Finding leaf descendants of an element}
\label{sec:leafdescex}

A key ingredient of our primitive---yet functional---version of \texttt{Balance}
is to identify those leaves that have neighbors of greater refinement level
than the leaf's level plus one.
We thus need to query for a leaf element $E$ in a forest $\forest F$ whether
there exists a face-neighbor leaf of $E$ in $\forest F$ with larger refinement level
than $\ell(E)+1$. To this end, we construct $E$'s face-neighbor elements of level
$\ell(E) + 1$ via the function \texttt{t8\_element\_half\_face\_neighbors} from
Section~\ref{sec:ghost-halfneighbors}.
For each of these half face-neighbors $E'$, we check whether there exists a
true descendant---i.e.\ a descendant that is not $E'$ itself---that is a local
leaf element or ghost element in $\forest F$.

In order to perform this check, we create the last descendant $D$ of $E'$ 
and search for an element $L\in \forest F.$\texttt{elements}$ \abst \cup \forest F.$\texttt{ghosts}
such that 
\begin{equation}
\label{eq:ELD}
\sfcf(E') < \sfcf(L) \leq \sfcf(D).
\end{equation}
Here, $\sfcf$ is the forest wide SFC index as in Section~\ref{sec:SFConforest}.
Such an $L$, if found, is a descendant of $E'$ 
because of the properties of $\sfcf$.
Also, $L$ is not $E'$ itself since then $\sfcf (E') = \sfcf (L) \nless
\sfcf(L)$. If no such $L$ exists, we know that $E'$ does not have a true
descendant in the leaves or ghosts.

We show the algorithm \texttt{t8\_forest\_leaf\_desc\_exists} in
Algorithm~\ref{alg:leafdesc}. To search for an $L$ that fulfills~\eqref{eq:ELD} 
in Lines~\ref{line:binsearch1} and \ref{line:binsearch2} we need to perform a
binary search for an element $D$ in a sorted (in SFC-order) array $A$
of $n$ elements, where we cannot guarantee that $D\in A$, but expect as
result the largest integer $i$ such that $\sfcf(A[i])\leq \sfcf(D)$.
If we assume that $\sfcf(A[0]) \leq \sfcf(D)$, this binary search is
possible: Choose bounds $l = 0, h = n - 1$ and a guess $g = (l + h + 1)/2$; if
$\sfcf(A[g]) > \sfcf(D)$ then set $h = g - 1$, else 
$\sfcf(A[g]) \leq \sfcf(D)$ and we set $l = g$. Start again with $g =
(l + h + 1)/2$ and iterate. We show this operation in the function
\texttt{binary\_search} in Algorithm~\ref{alg:leafdesc}.

If the assumption $\sfcf(A[0])\leq \sfcf(D)$ is not fulfilled, an $L$
that satisfies equation~\eqref{eq:ELD} does not exist. We can check this in
constant time and the search returns $L = A[0]$, for which the check
$\sfcf(E') < \sfcf(L) \leq \sfcf(D)$ from Line~\ref{line:compare1} or
Line~\ref{line:compare2} fails as expected.

\begin{algorithm}
 \caption{\texttt{t8\_forest\_leaf\_desc\_exists} (\texttt{Forest} $\forest F$, 
          \texttt{Element} $E'$)}
 \label{alg:leafdesc}
 \DontPrintSemicolon
 \algoresult{True if a leaf in $\forest F.$\texttt{elements} or $\forest F.$\texttt{ghosts}
                  exists that is a true descendant of $E'$. False otherwise.}\\[1ex]
 $D\gets$ \texttt{t8\_element\_last\_descendant} ($E'$)\;
 $L\gets$ \texttt{binary\_search} ($\forest F.$\texttt{elements}, $D$)
\Comment{Search in the local elements}\label{line:binsearch1}
  \If{$\sfcf (E') < \sfcf (L) \leq \sfcf (D)$\label{line:compare1}}
  {
    \Return \algotrue
  }
 $L\gets$ \texttt{binary\_search} ($\forest F.$\texttt{ghosts}, $D$)
\Comment{Search in the ghost elements}\label{line:binsearch2}
 \If{$\sfcf(E') < \sfcf(L) \leq \sfcf(D)$\label{line:compare2}}
  {
    \Return \algotrue
  }
 \Return \algofalse\\[1.5ex]
\setcounter{AlgoLine}{0}
\nonl \textbf{Function} \texttt{binary\_search} (\texttt{Array} $A$,
      Element $D$)\; 
\algoresult If $\sfcf(A[0]) \leq \sfcf(D)$, $A[i]$ for the
largest index $i$ such that $\sfcf(A[i]) \leq \sfcf(D)$, otherwise
$A[0]$.\\[1ex]

  \If{$\sfcf(A[0]) > \sfcf(D)$}{\Return $A[0]$\;}
$l \gets 0$\;
  $h \gets A.$\texttt{length} $-\ 1$\;
\While{$l < h$}
 {
  $g \gets \frac{l + h + 1}{2}$\;
  \leIf{$\sfcf(A[g]) \leq \sfcf(D)$}{$l = g$\;}
  {
    $h = g-1$
  }
 }
\Return $A[g]$\;
\end{algorithm}

\section{The \texttt{Ripple-balance} algorithm}

Our ripple version of \texttt{Balance} (Algorithm~\ref{alg:balance}) is an
iterative one. In each iteration, we construct a new forest $\forest F_{i+1}$ from the
current forest $\forest F_i$, starting with the original forest $\forest F_0$ that we want to
balance. In each iteration, we check for each leaf $E$ of $\forest F_i$ whether
there are face-neighbors of $E$ in the local leaves or ghosts of
$\forest F_i$ with a larger refinement level than $\ell(E) + 1$. If so, we
refine the element $E$ and add the children to the new forest $\forest
F_{i+1}$, otherwise, we add the element $E$ to $\forest F_{i+1}$.
For this check we use the function \texttt{t8\_element\_leaf\_desc\_exists}
that we describe in the previous Section. We repeat these refinement steps
until the forest mesh on each process does not change anymore.

We use the ghost layer of $\forest F_i$ to take into account that a local
leaf may need to be refined if a face-neighbor leaf on a neighboring process
has a larger refinement level. Thus, in each iteration, we call \texttt{Ghost}
for the newly constructed forest $\forest F_i$. Note that we need the unbalanced
version of \texttt{Ghost} here (\texttt{Ghost\_v2} or \texttt{Ghost\_v3}).

In Algorithm~\ref{alg:balance} we explicitly write down the element loop;
however, in the actual implementation we replace it by a call to \texttt{Adapt}
with the appropriate callback function.

\begin{proposition}
Algorithm~\ref{alg:balance} terminates and produces a balanced forest $\forest F^\ast$.
\end{proposition}
\begin{proof}
\newcommand\parent[1]{\hat #1}
The algorithm terminates, if on each process no leaf element is refined any
longer. We now show that we eventually reach this status.
Let $m_i$ be the maximum refinement level of all global leaf elements in $\forest F_i$
(across all processes).
We claim $m_i = m_0$. Each leaf $E$ in $\forest F_i$ is either a leaf element of
$\forest F_{i-1}$ or a child of a leaf element $\parent E$ in $\forest
F_{i-1}$. In the first case $\ell(E)\leq m_{i-1}$ by definition. Let us
consider the second case: The element $E$ of $\forest F_i$ is a child of $\parent
 E$ which is refined in iteration $i-1$. Hence, $\ell(E) = \ell(\parent E) + 1$.
Since $\parent E$ is refined in iteration $i-1$, there exists a face-neighbor $E'$ of
$\parent E$ in $\forest F_{i-1}$ with $\ell(E') > \ell(\parent E) + 1 = \ell(E)$.
Since $m_{i-1} \geq \ell(E')$ we obtain $m_{i-1} \geq m_i$ and since no leaf
is coarsened, we get the equality $m_{i-1} = m_i$, proving our claim.

Therefore, the maximum refinement level of $\forest F_i$ is bounded by $m_0$
and because no leaf is coarsened there must exist a final step $i^\ast$
in which no element is refined anymore. Hence, the algorithm
terminates after $i^\ast$ steps.

Let $\forest F^\ast = \forest F_{i^\ast}$ be the result of \texttt{Ripple-balance}.
Since in the last step no element is refined, $\forest F^\ast$ is an exact
copy of $\forest F_{i^\ast-1}$. Thus, for each leaf $E$ in $\forest F^\ast$
we know that there exists no face-neighbor $E'$ with level $\ell(E') > \ell(E) +
1$, since otherwise $E'$ would be refined in iteration $i^\ast$ and thus the
algorithm would continue with $i^\ast + 1$.  Suppose that a leaf $E$ exists
in $\forest F^\ast$ with a face-neighbor leaf $E'$ with level $\ell(E') <
\ell(E) - 1$, then $E'$ has a face-neighbor $E'' = E$ in $\forest F^\ast$ with
$\ell(E'') > \ell(E') + 1$, which is a contradiction.

Thus, $\forest F^\ast$ fulfills the balance condition.
\let\parent\undefined
\end{proof}

\begin{algorithm} %
\caption{\texttt{Ripple-balance} (\texttt{Forest} $\forest F_0$)}
\label{alg:balance}
\DontPrintSemicolon
\algoresult{A new forest $\forest F^\ast$ consisting of (possibly) refined elements of
$\forest F_0$, such that $\forest F^\ast$ fulfills the face-balance condition.}\;
\texttt{done} $\gets 0$\;
$i\gets 0$\;
  \While{\upshape\textbf{not} \texttt{done}}
{
  \texttt{done} $\gets 1$ 
  \Comment{We are done if no element has to be refined any more}
  $\forest F_i\gets$ \texttt{Partition} ($\forest F_i$)\Comment{Partition the forest (optional)}\label{line:balance_partition}
  \texttt{Ghost} ($\forest F_i$)\Comment{Create the ghost layer}
  $\forest F_{i+1} \gets \forest F_i$\;
    \algofor{$E\in \forest F_i.$\texttt{elements}}
    {
      \texttt{refine\_flag} $\gets 0$\;
      \algofor{$0\leq f < $\texttt{t8\_element\_num\_faces} ($E$)}
      {
        $E'[]\gets $\texttt{t8\_forest\_half\_face\_neighbors} ($\forest F_i, E, f$)\;
        \algofor{$0\leq i < $\texttt{t8\_element\_num\_face\_children} ($E, f$)}
        {
          \algoif {\texttt{t8\_forest\_leaf\_desc\_exists} ($\forest F_i, E'[i]$)}
          {
            \texttt{refine\_flag} $\gets 1$ 
                                           \Comment{Mark $E$ for refinement}
            \textbf{goto} \ref{line:refflag}\Comment{No need to check the remaining neighbors}
          }
        }
      }
      \algoifcom{\IfComment{Refine $E$ if necessary}}{\texttt{refine\_flag}
\label{line:refflag}
}
      {
        $\forest F_{i+1}$.\texttt{elements} $\gets \forest F_{i+1}$.\texttt{elements}$\ohne\set{E} \cup$
        \texttt{t8\_element\_children} ($E$)\;
        \texttt{done} $\gets 0$ \;
      }
  }
  $i\gets i+1$\;
  \texttt{MPI\_Allreduce} (\texttt{done}, MPI\_LAND)\Comment{Logical \lq
  and\rq\ of all values of \texttt{done}}\label{line:allreduce} 
  \Comment{on the different MPI ranks}
} 
\Return $\forest F_i$\;
\end{algorithm}

\begin{remark}
As we observe in Figure~\ref{fig:balance_ex1}, repartitioning of the forest may
be necessary after \texttt{Ripple-balance}. In order to prevent the algorithm to produce
largely imbalanced loads on the different ranks, we may also repartition each
intermediate forest before we start the next iteration as we do in
Line~\ref{line:balance_partition} of the algorithm.
The resulting forest $F^\ast$ is then partitioned as a consequence.
\end{remark}
\begin{remark}
Even though we might not refine local elements of a process $p$ in one iteration,
thus keeping the variable \texttt{done} set to \texttt{true}, changes of the forest
in neighboring processes may render it necessary that we need to refine elements
on $p$ in later iterations. For this reason we need to compute the logical 'and'
  of all \texttt{done} values on all processes, hence the \texttt{MPI\_Allreduce}
  call in Line~\ref{line:allreduce}.
  Figure~\ref{fig:balance_ex1} shows an example
for this situation. Here, the elements on process $p=4$ (in dark red) do not
change in the first two iterations of \texttt{Ripple-balance}, but they are refined
multiple times eventually.
\end{remark}

\section{Numerical results}

In this section, we present numerical results for the \texttt{Ghost} and
\texttt{Ripple-balance} routines. All results are obtained with the
\texttt{t8\_time\_forest\_partition} example of \tetcode version~0.3.
We perform the tests on the JUQUEEN supercomputer \cite{Juqueen} and use 
16 MPI ranks per compute node throughout.

\subsection{The test case}
\label{sec:testcase}

In the test we use a similar setting to the test in
Section~\ref{sec:cmesh_example_forest} for coarse mesh partitioning.
We start with a uniform forest of level $\ell$ and refine it in a band along
an interface defined by a plane to level $\ell + k$.
We then call \texttt{Ripple-balance} to establish a 2:1 balance among the elements 
and we create a layer of ghost elements with \texttt{Ghost} afterwards.
The interface moves through the domain in time in direction of the plane's normal vector.
In each time step we adapt the mesh, such that we coarsen elements outside
of the band to level $\ell$ and refine within the band to level $\ell + k$. We
then repeat the \texttt{Ripple-balance} and \texttt{Ghost} calls.
As opposed to the test in Section~\ref{sec:cmesh_example_forest}, we take the
unit cube as our coarse mesh geometry. We run the test once with a hexahedral
mesh consisting of one tree and once with a tetrahedral mesh of six trees
forming a unit cube as in Figure~\ref{fig:sechstetra} in
Section~\ref{sec:tetsfc_beyrule}.

We choose the normal vector 
$\frac{3}{2}\begin{pmatrix}1, &1,& \frac{1}{2}\end{pmatrix}^t$, and we choose $\frac{1}{4}$
as width of our refinement band. We move the refinement band with speed
$v$ and scale the time step $\Delta t$ with the refinement level as 
\begin{equation}
\label{eq:balance-deltat}
 \Delta t (\ell) = \frac{C}{2^\ell v},
\end{equation}
$C$ being the CFL-number. It is a measurement for the
width of the band of level $\ell$ elements that will be refined to level $k$ in
the next time step. We set $C=0.8$ and choose $v$ such that $\frac{1}{v} = 0.64$.
We start the band at position $x_0 (\ell) =0.56 - 2.5\Delta t(\ell)$
and measure up to 5 time steps.

Thus, for level $\ell$ and band width $k$, we use the program call
\texttt{t8\_time\_forest\_partition -c MESH -n1 -l$\ell$ -r$k$ -x $x_0(\ell)$
-X $x_0(\ell) + 0.25$ -C 0.8 -T $6\Delta t(\ell)$ -gbo}.
Here, MESH stands for a file storing the coarse mesh, i.e.\ either the tetrahedralized unit cube or
the unit cube of one hexahedron tree. These coarse meshes can also be generated with the 
\texttt{t8\_cmesh\_new\_hypercube} function of \tetcode.
The settings \texttt{-g} and \texttt{-b} tell the program to construct the
ghost layer, and to balance the forest mesh after adaptation.
The \texttt{-o} setting disables output of visualization files.
We refer to Figure~\ref{fig:baghoex-1} for an illustration of the setting.

\begin{figure}
\center
\includegraphics[width=0.45\textwidth]{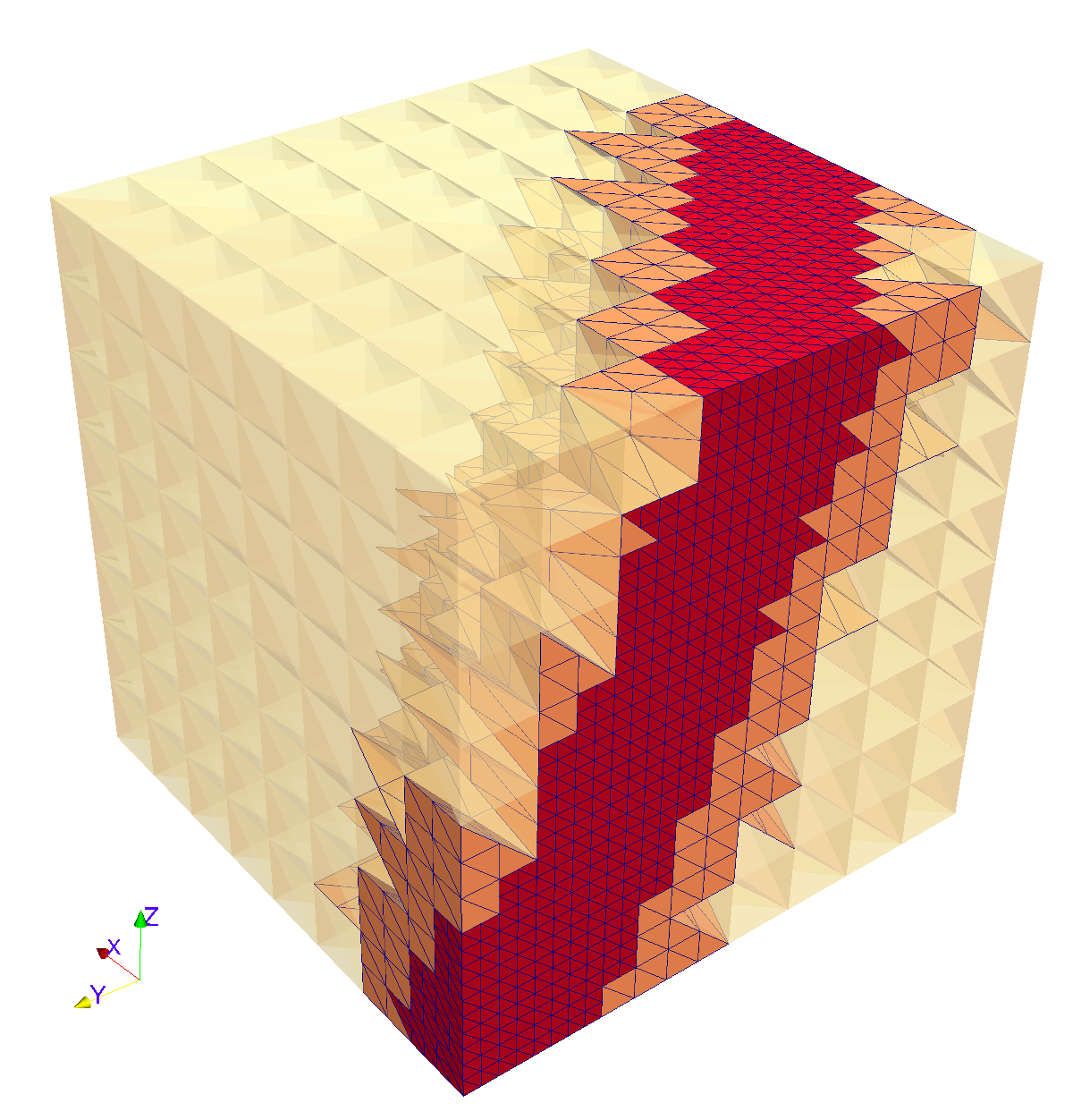}\hfill
\includegraphics[width=0.45\textwidth]{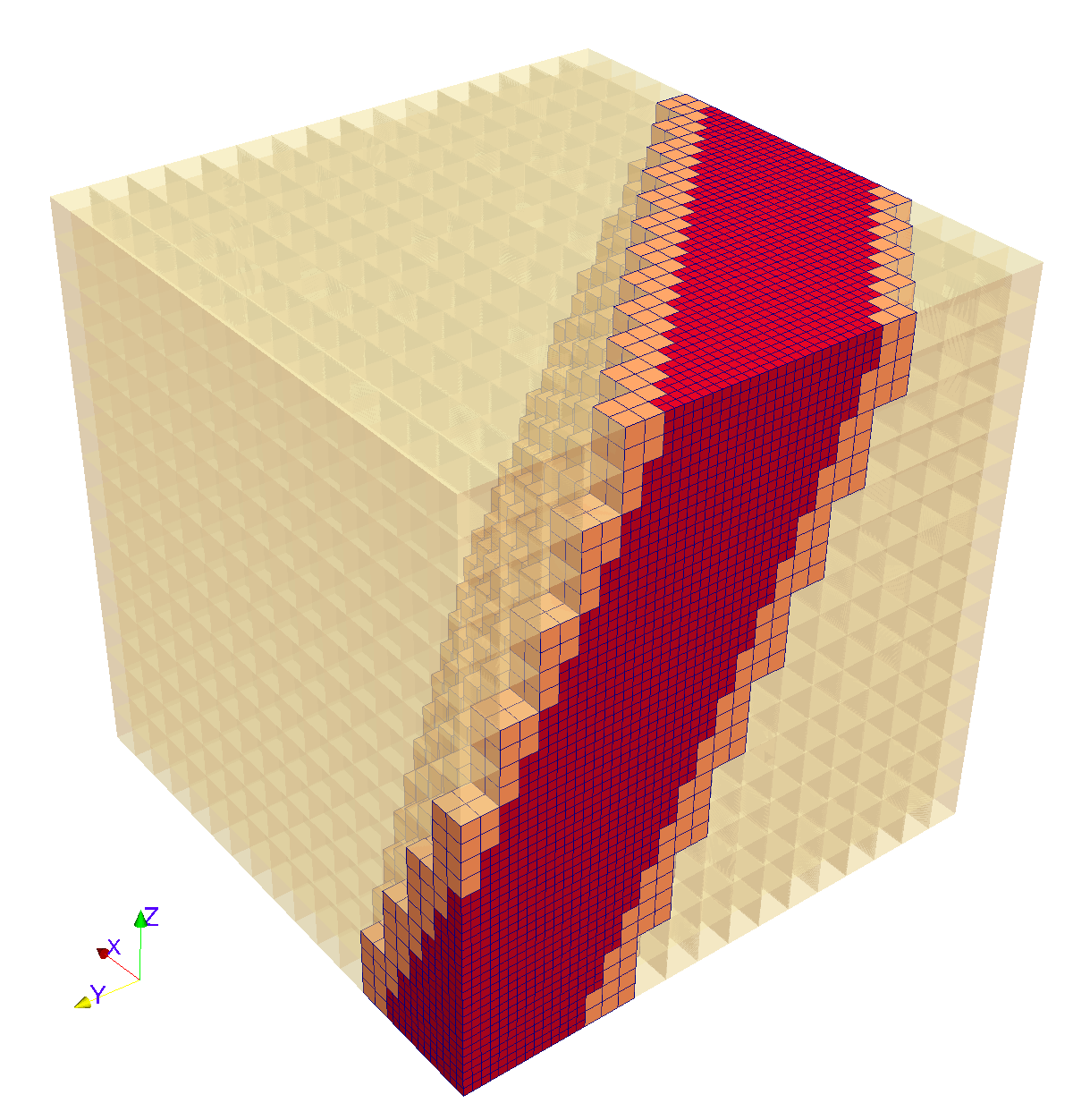}
\caption
  [Testing \texttt{Ghost} and \texttt{Ripple-balance} on a unit cube geometry.]
  {
  We test \texttt{Ghost} and \texttt{Ripple-balance} on a unit cube geometry
  consisting of six tetrahedral trees (left) or one hexahedral tree (right).
  Starting with a uniform level $\ell$, we refine the forest in a band around a
  plane to level $\ell + k$. We then balance the forest and create the ghost
  layer. In the next time step, the band moves in the direction of the plane's
  normal vector and we repeat the steps, coarsening previously fine forest
  elements if they now reside outside of the band. We show the forest after
  \texttt{Ripple-balance} at time step
  $t = 2\Delta t(\ell)$ for two different configurations.
  Left: Tetrahedral elements with $\ell = 3$, $k = 2$. In total
  we have 56,566 tetrahedral elements.  Right: Hexahedral elements with $\ell =
  4$, $k = 2$, summing up to 78,100 hexahedral elements in total. The color
  represents the
  refinement level. We draw level $\ell$ elements opaque.}
\figlabel{fig:baghoex-1}
\end{figure}

\subsection{Strong scaling}

We run a strong scaling test with tetrahedral elements and
refinement parameters $\ell = 8$, $k = 2$ on 8,192 up to 131,072 MPI ranks,
increasing the process count by a factor of 2 in each step. We list the
runtimes at time $t = 4\Delta t$ for \texttt{Ghost} and \texttt{Ripple-balance} in
Table~\ref{tab:baghoex-1-strong}, and plot them together with
\texttt{Partition} in Figure~\ref{fig:baghoex-1-strong}.

Since the runtime of \texttt{Ripple-balance} depends on the number of
process-local leaf elements, we expect it to drop by a factor of 2 if we double
the number of processes and thus divide the number of elements per process in
two. In particular, consider two different runs with process counts $P_1$ and
$P_2$, local element counts $E_1$ and $E_2$ and runtimes $T_1$ and $T_2$. We
compute the parallel efficiency $e$ of the run with $P_2$ processes in relation
to the $P_1$ run as the fraction
\begin{equation}
\label{eq:effbal}
 e_\textrm{Ripple-balance} = \frac{T_1 E_2}{T_2 E_1}.
\end{equation}

As wee see in Table~\ref{tab:ghost-comparetest} in the previous chapter,
the runtime of \texttt{Ghost} depends linearly on the number of ghost elements
per process. The number of ghosts is proportional to the surface area of a
process's partition and thus ideally scales with $\mathcal
O((N/P)^\frac{2}{3})$, with $N$ the global number of elements
\cite{IsaacBursteddeWilcoxEtAl15}. Consider two runs with $P_1$ and $P_2$
processes as above and let $G_1$ and $G_2$ denote the numbers of ghost elements
per process, then the parallel efficiency of the second run in relation to the
first run is
\begin{equation}
\label{eq:effgho}
 e_\mathrm{Ghost} = \frac{T_1 G_2}{T_2 G_1}.
\end{equation}

We achieve ideal strong scaling efficiency for \texttt{Ghost} and even more
than ideal efficiency for \texttt{Ripple-balance}, which hints at the runtime of
\texttt{Ripple-balance} being slightly worse than $\mathcal O (N)$. We also observe
that this basic variant of \texttt{Balance} is indeed by far the slowest
of the AMR algorithms as we already hinted above.

\begin{table}
\center
\begin{tabular}{|r|r|r|r|r|r|r|}\hline
\multicolumn{7}{|c|}{Tetrahedral case with $\ell = 8$, $k = 2$, $C=0.8$ at $t=4\Delta t$}\\ \hline
   &              &            & \multicolumn{2}{c|}{\texttt{Ripple-balance}} & 
\multicolumn{2}{c|}{\texttt{Ghost}} \\ 
 $P$ & $E/P$ & $G/P$ & Time [s] & Par.\ Eff.\ & Time [s] & Par.\ Eff.\ \\ \hline
 8,192  & 234,178 & 17,946 & 687.0 & 100.0\%  & 3.25  & 100.0\% \\   
16,384  & 117,089 & 11,311 & 336.2 & 102.1\%  & 2.12  &  96.6\% \\  
32,768  &  58,545 &  7,184 & 161.2 & 106.5\%  & 1.27  & 102.4\% \\  
65,536  &  29,272 &  4,560 &  78.3 & 109.6\%  & 0.79  & 104.5\% \\  
131,072 &  14,636 &  2,859 &  37.7 & 113.8\%  & 0.52  &  99.5\% \\  \hline
\end{tabular}
\caption[Strong scaling with tetrahedral elements.] {
The results for strong scaling of \texttt{Ripple-balance} and \texttt{Ghost}
with tetrahedral elements. The problem parameters are $\ell = 8$, $k=2$,
and $C=0.8$ with $\Delta t$ according to \eqref{eq:balance-deltat}. 
We show the runtimes of time step $t=4\Delta t$.
After \texttt{Ripple-balance} the mesh consists of approximately 1.91\e9 Tetrahedra.
In addition to the runtimes, we show the number of elements per process, $E/P$,
and ghosts per process, $G/P$. We also compute the parallel efficiency of
\texttt{Ripple-balance} and \texttt{Ghost} according to \eqref{eq:effbal} and
\eqref{eq:effgho} in reference to the run with 8,192 processes.
We observe a more than ideal scaling for \texttt{Ripple-balance} and a nearly ideal
  scaling for \texttt{Ghost}.  See also Figure~\ref{fig:baghoex-1-strong} for a
  plot of these runtimes.
} 
\figlabel{tab:baghoex-1-strong}
\end{table}

\begin{figure}
\center
\includegraphics[width=0.49\textwidth]{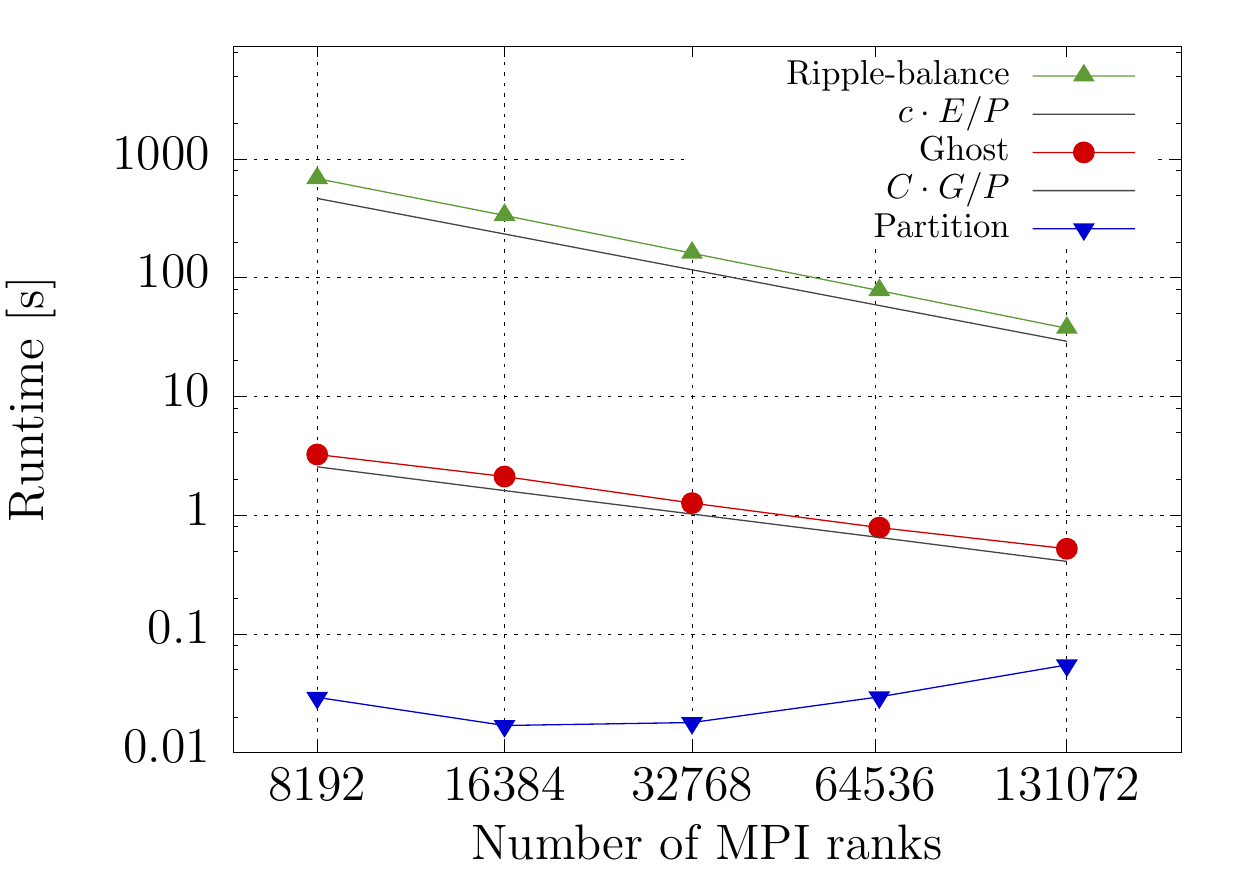}
\caption[Strong scaling with tetrahedral elements.]
{Strong scaling with tetrahedral elements. We show the 
runtimes of \texttt{Ripple-balance}, \texttt{Ghost}, and \texttt{Partition} for
  the test case from Section~\ref{sec:testcase} with $\ell = 8$, $k = 2$ at
  time step $t = 4\Delta t$.  The forest mesh consists of approximately 1.9\e9
  tetrahedra.
We use 8,192 up to 131,072 processes on JUQUEEN with 16 processes per compute
node. Ideally, \texttt{Ripple-balance} scales with the number of elements per process,
$N/P$, and \texttt{Ghost} with the number of ghost elements per process, $G/P$.
We show these measures scaled by a constant in black lines.
As we observe in the plot and in Table~\ref{tab:baghoex-1-strong}, we achieve
perfect scaling for \texttt{Ripple-balance} and nearly perfect scaling for
\texttt{Ghost}.
  The runtime of \texttt{Partition} is below 0.1 seconds even for
the largest process count.
  }
\figlabel{fig:baghoex-1-strong}
\end{figure}

\subsection{Weak scaling}
For weak scaling we increase the global number of elements while also increasing 
the process count, keeping the local number of elements nearly constant.
Since with each refinement level $\ell$ the number of global elements grows by a factor
of 8, we multiply the process count with 8 as well.
We test the following configurations:
\begin{itemize}
  \item Tetrahedral elements with 8,192 processes, 65,536 processes, and 458,752 processes,
  with refinement levels $\ell = 8$, $\ell = 9$, $\ell = 10$. This amounts to
  about 235k elements per process.
  Thus, the largest run has about $107.8\e9$ elements.
  \item Tetrahedral elements with 2,048 processes, 16,384 processes, and
131,072 processes, with refinement levels $\ell = 8$, $\ell = 9$, $\ell = 10$.
Here we have about 155k elements per process, summing up to
 $20.3\e9$ elements on 131,072 processes.
  \item Hexahedral elements with the same process counts and level increased by one ($162\e9$ elements in total).
\end{itemize}

Note that 458,752 is actually 7 times 65,536. We choose it since it is the
maximum possible process count on JUQUEEN with 16 processes per node, using all
28,672 compute nodes. The number of elements per process is thus about 14\%
greater than on the other process counts in the configuration. However,
\eqref{eq:effbal} and \eqref{eq:effgho} still apply for computing the parallel
efficiency.

The largest test case that we run is for hexahedra on 458,752 process with $\ell = 11$ amounting
to 162\e9 elements.

\begin{table}
\center
\begin{tabular}{|r|r|r|r|r|r|r|r|}\hline
\multicolumn{8}{|c|}{Tetrahedral case with $k = 2$, $C=0.8$ at $t=4\Delta t$}\\ \hline
  &   &              &            & \multicolumn{2}{c|}{\texttt{Ripple-balance}} & 
\multicolumn{2}{c|}{\texttt{Ghost}} \\ 
  $P$ & $\ell$ & $E/P$ & $G/P$ & Time [s] & Par.\ Eff.\ & Time [s] & Par.\ Eff.\ \\ \hline
 8,192  & 8 & 234,178 & 17,946 & 687.0 & 100.0\%  & 3.25  & 100.0\% \\   
65,536  & 9 & 233,512 & 18,282 & 732.5 &  93.5\%  & 3.76  &  88.2\% \\  
458,752 &10 & 266,494 & 20,252 & 913.5 &  85.5\%  & 3.79  &  96.8\% \\  \hhline{|========|}
  2,048 & 7 & 117,630 & 10,999 & 305.5 & 100.0\% & 1.99 & 100.0\% \\  
 16,384 & 8 & 117,089 & 11,311 & 336.2 &  90.4\% & 2.12 &  96.5\% \\
131,072 & 9 & 116,756 & 11,478 & 360.8 &  84.0\% & 2.18 &  95.2\% \\ \hline
\end{tabular}\\[2ex]
\begin{tabular}{|r|r|r|r|r|r|r|r|}\hline
\multicolumn{8}{|c|}{Hexahedral case with $k = 2$, $C=0.8$ at $t=2\Delta t$}\\ \hline
  &   &              &            & \multicolumn{2}{c|}{\texttt{Ripple-balance}} & 
\multicolumn{2}{c|}{\texttt{Ghost}} \\ 
  $P$ & $\ell$ & $E/P$ & $G/P$ & Time [s] & Par.\ Eff.\ & Time [s] & Par.\ Eff.\ \\ \hline
 8,192  & 9 & 309,877 & 34,600 & 268.5 & 100.0\% & 6.79 & 100.0\% \\   
65,536  &10 & 310,163 & 35,136 & 272.8 &  98.5\% & 6.85 & 100.7\% \\  
458,752 &11 & 354,746 & 38,833 & 313.0 &  98.2\% & 7.86 &  96.9\% \\  \hhline{|========|}
  2,048  & 8 & 156,178 & 21,536 & 131.4 & 100.0\% & 4.18 & 100.0\% \\  
 16,384  & 9 & 155,702 & 22,036 & 132.7 &  98.7\% & 4.25 & 100.6\% \\
131,072  &10 & 155,460 & 22,284 & 134.4 &  97.2\% & 4.36 &  98.9\% \\ \hline
\end{tabular}
  \caption[Weak scaling with tetrahedral and hexahedral elements.]
{Weak scaling for \texttt{Ripple-balance} and \texttt{Ghost} with tetrahedral (top) and
  hexahedral (bottom) elements. We increase the level by one and multiply the process count
  by eight to maintain the same number of local elements per process. 
 Notice that the highest process count of 458,752 is only seven times 65,536
resulting in $\sim14\%$ more local elements.
  For hexahedra \texttt{Ripple-balance} has an overall better performance while \texttt{Ghost} performs
  better on tetrahedra, which we expect due to the lower number of faces per
element. Both algorithms show good scaling with efficiencies greater than 85\%
(tetrahedra) and
  96.9\% (hexahedra). See also Figure~\ref{fig:baghoex-1-weak}.
  The maximum global number of elements is $107.8\e9$ with tetrahedra
  and $162\e9$ with hexahedra.}
\figlabel{tab:baghoex-1-weak}
\end{table}

\begin{figure}
  \center
\includegraphics[width=0.49\textwidth]{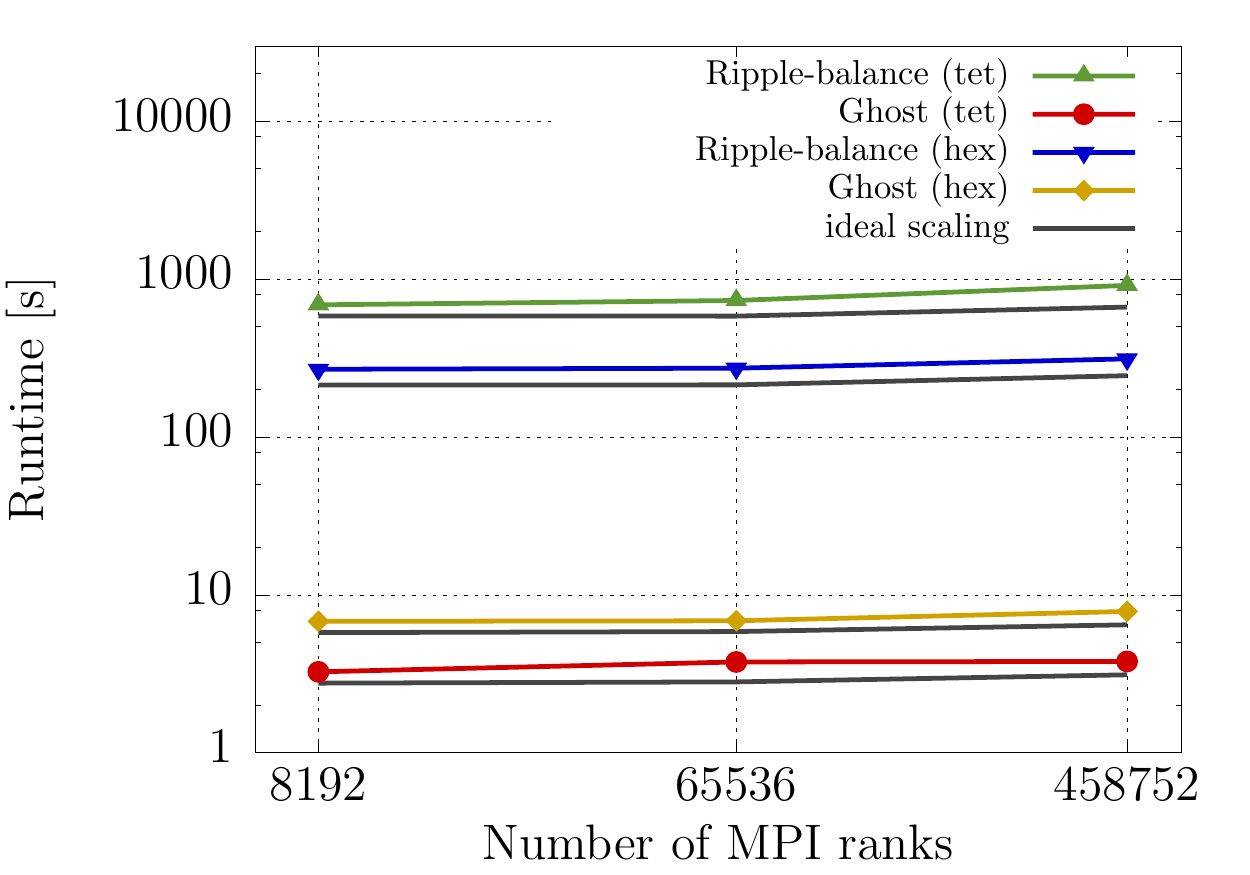}
  \includegraphics[width=0.49\textwidth]{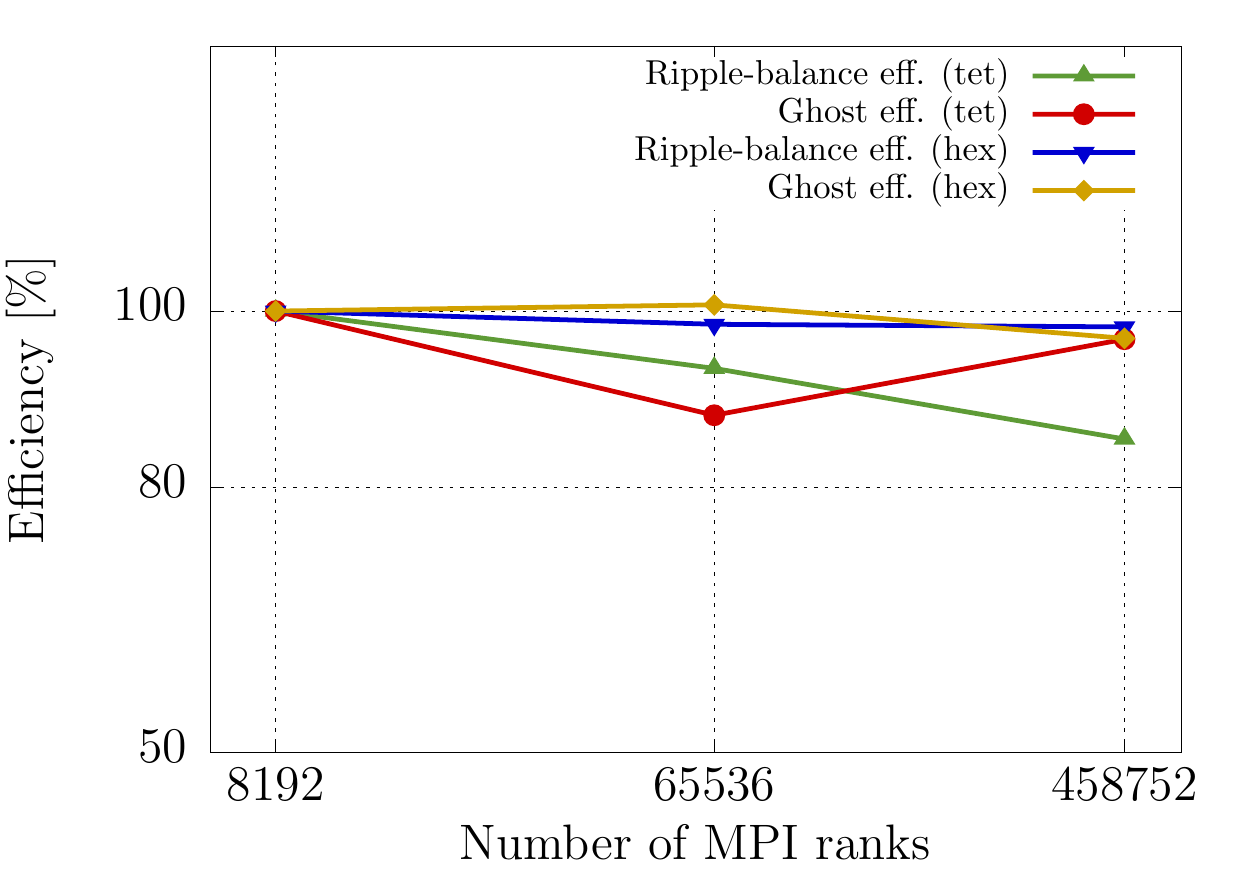}
  \caption[Weak scaling with tetrahedral and hexahedral elements.]
 {
  Weak scaling results for tetrahedra with refinement levels 8,
  9, and 10, and for hexahedra with refinement levels 9, 10, and 11. 
  This amounts to about 233k elements per process for tetrahedra
  and 310k elements per process for hexahedra.
  This number increases for the 458,752 process runs, since 458,752
  is only seven times 65,536, while we increase the number of mesh elements
  by the factor 8.
  On the left-hand side we plot the runtimes of \texttt{Ghost} and
  \texttt{Ripple-balance} with the ideal scaling in black. For
  \texttt{Ripple-balance} the ideal scaling line
  is based on the number of elements per process while for \texttt{Ghost} it
  is based on the number of ghost elements per process.
  On the right-hand side we plot the parallel efficiency in $\%$.
  We display all values in Table~\ref{tab:baghoex-1-weak}}
  \figlabel{fig:baghoex-1-weak}
\end{figure}

We show the results in Table~\ref{tab:baghoex-1-weak} and Figure~\ref{fig:baghoex-1-weak}.
We notice that \texttt{Ghost} for tetrahedra is faster than \texttt{Ghost} for Hexahedra, which
we explain by the smaller number of ghosts due to less faces per element.
\texttt{Ripple-balance} on the other hand is slower for tetrahedra than for hexahedra.
The main part of \texttt{Ripple-balance} is a loop over all local elements.
Thus, the overall slower performance of the tetrahedral Morton code compared to
the cubical Morton code is probably the reason for this difference in runtime.
We observe this behavior also in the weak parallel efficiency of
\texttt{Ripple-balance} for tetrahedra which drops to 85.5\% compared to the 98.2\%
efficiency for hexahedra. We even observe a strong scaling efficiency of more
than 100\% for tetrahedral \texttt{Ripple-balance}, which means that for tetrahedra
the algorithm's runtime is not perfectly linear in the number of elements per
process. This is also a hint that there is further potential to optimize 
\texttt{Ripple-balance}, which we expected, since there exist more 
efficient algorithms for 2:1 balancing on hexahedra~\cite{IsaacBursteddeGhattas12}.

For hexahedral \texttt{Ripple-balance}, and for hexahedral and tetrahedral
\texttt{Ghost}, however, we observe excellent strong and weak scaling with
efficiencies in the order of 95\%.
 \chapter{A Numerical Application}
\label{ch:app}

In this chapter we discuss how an application could use the AMR
routines described in this thesis. We implement a finite volume (FV)
solver for the advection equation and discuss important data
handling algorithms, such as communicating ghost data and interpolating
data after changing the mesh.

We show that the solver works with different element types,
in particular triangles and tetrahedra with the TM index described in
Chapter~\ref{ch:tetSFC} and quadrilaterals and hexahedra with the Morton index
using the \pforest implementation~\cite{Burstedde10a}.
We also show runs on hybrid meshes consisting of triangles and quadrilaterals
in 2D and of hexahedra, tetrahedra and prisms in 3D. For the low-level
implementation of the prism elements we use the work from~\cite{Knapp17}, which
models a prism as the cross product of a line with Morton index and a triangle
with TM-index.

\section{The advection equation}

We consider the $d$-dimensional advection equation. It describes how a quantity
$\phi$ is advected with a given flow $u$ over time.
For a compact domain $\Omega\subset \IR^d$ and a time-dependent flow function
$u\colon \Omega\times\IR_{\geq0}\rightarrow\IR^d$, we are interested in the
solution $\phi\colon\Omega\times\IR_{\geq 0}\rightarrow\IR$ of the PDE

\begin{equation}
  \label{eq:advectioneq}
 \frac{\partial{\phi}}{\partial t} + \nabla\cdot (\phi u) = 0
\end{equation}
with initial condition
\begin{equation}
  \phi(\cdot,0) = \phi_0
\end{equation}
and appropriate boundary conditions.
We assume that the flow $u$ is divergence free,
simplifying equation~\eqref{eq:advectioneq} 
to
\begin{equation}
\label{eq:advectioneqdivfree}
 \frac{\partial{\phi}}{\partial t} + u\cdot\nabla \phi = 0.
\end{equation}

\subsection{Level-set functions}

In this section we discuss level-set functions as a possible interpretation
of $\phi$.
There are various examples in which the advection equation is used to keep
track of the movement---under the advection of the flow $u$---of a
$(d-1)$-dimensional interface between two disjoint subsets
$\Omega^1$ and $\Omega^2$ with $\Omega^1 \cup \Omega^2 =
\Omega$~\cite{Albrecht16,MirzadehGuittetBursteddeEtAl16,Klitz14,OsherFedkiw02,OsherSethian88}.
A typical example is two-phase flow~\cite{SussmanSmerekaOsher94,Klitz14}, where
$\Omega^1$ marks the region occupied by the first phase of a fluid (i.e.\ liquid)
and $\Omega^2$ the
region occupied by the second (i.e.\ gaseous). We pick the initial condition
$\phi_0$ such that it is continuous and satisfies
\begin{subequations}
\begin{align}
  \phi_0(x) \geq 0,& \textrm{ for all } x \in\Omega^1,\\
  \phi_0(x) < 0,& \textrm{ for all } x \in\Omega^2.
\end{align}
\end{subequations}
A common choice for $\phi_0$ is the signed distance to the interface 
$\Gamma=\overline{\Omega^1}\cap\overline{\Omega^2}$, thus
\begin{equation}
 \phi_0(x) = \begin{cases}
    \phantom{-}\dist(x,\Gamma),\, x\in\Omega^1,\\
    -\dist(x,\Gamma),\,x\in\Omega^2.
\end{cases}
\end{equation}

Let us define for each time step $t$ the sets $\Omega^i_t$ as the collection
of all points in $\Omega^i$ after they have been transported with $u$,
\begin{subequations}
 \begin{align}
   \Omega_t^1 :=& 
    \setm{x\in\Omega}{\phi(x,t) \geq 0},\\
   \Omega_t^2 :=&
    \setm{x\in\Omega}{\phi(x,t) < 0}.
 \end{align}
\end{subequations}
Furthermore, we define $\Gamma_t$ as the interface between $\Omega^1_t$ and $\Omega^2_t$:
\begin{equation}
  \Gamma_t := \overline{\Omega_t^1}\cap\overline{\Omega_t^2} = \setm{x\in\Omega}{\phi(x,t) = 0}=\phi^{-1}(0,t).
\end{equation}
For a two-phase flow we interpret $\Omega^i_t$ as the region
of $\Omega$ that is occupied by fluid $i$ at time $t$.
We can use the sign of $\phi(x,t)$ to decide whether $x\in\Omega^1_t$ or
$x\in\Omega^2_t$.

\begin{remark}
 \label{rem:reinit}
  Over the simulation time the approximated solution $\phi$ may loose its
  signed distance property. A common observation is that the norm of the
  gradient at the zero level-set approaches zero, resulting in numerical irregularities when
  reconstructing the interface.
  A way to restore the signed distance property at a time step $t$ is to replace
  $\phi$ with a new level-set function $\phi^\ast$ which is obtained as the solution to
  the Hamilton-Jacobi equation
  \begin{equation}
    \frac{\partial \phi^\ast}{\partial\tau} = \sign(\phi_t)(1-\lvert \nabla \phi^\ast \rvert)
  \end{equation}
  with pseudo-time $\tau$.
  We omit this reinitialization process in our solver, since it does not affect the AMR 
  routines. We refer the reader to~\cite{Klitz14,Min10} for more details.
\end{remark}

\section{Numerically solving the advection equation}
The numerical method we choose to solve the advection equation is a
finite volume method with polynomial degree $0$. This method leads to a
rather simple solver that is easy to implement and provides first order
convergence rates.
We are well aware that far more accurate solvers exist. However, since we use the
application as a proof-of-concept to show the coupling of application and AMR
routines, we use degree $0$ for the sake of simplicity.
In fact, a cheap numerical solver is most challenging for the AMR routines and their
absolute runtimes.
For a detailed description of the method, see~\cite{SuliSchwabHouston00} and
the references therein.
Although in practical applications the flow $u$ could be the (discrete)
solution of a fluid solver, we assume for our computations that $u$ is
given analytically.

Integrating equation~\eqref{eq:advectioneqdivfree} over a volume $V\subset
\Omega$ and applying the Gauss divergence theorem we obtain
\begin{equation}
\label{eq:advectionvolint}
 \frac{\partial}{\partial t}
  \int_V \phi(x,t)\,dx + \int_{\partial V} \phi(s,t) u(s,t) \cdot \vec n(s)\,ds = 0,
\end{equation}
with $\vec n(s)$ being the outward pointing normal vector to $V$ at $s\in\partial V$.

With the finite volume ansatz, we discretize the domain with mesh
elements and consider equation~\eqref{eq:advectionvolint} on each mesh element
$E$.  We model the approximate solution of $\phi$ at time $t$ with a constant
value $\phi_{E,t}$ on each element $E$. 
We interpret this value as an approximation to $\phi$ at the midpoint 
$m_E$ of $E$, thus $\phi_{E,t}\approx \phi(m_E,t)$.
The initial condition at $t=0$ is then $\phi_{E,0}=\phi_0(m_E)$.

The boundary integral is an integral over the faces of $E$.
\begin{definition}
 Let 
 \begin{equation}
    \omega(E):=\setm{(E',F)}{E' \textrm{ is face-neighbor of } E \textrm{ across
face } F}
 \end{equation}
  be the set of all pairs $(E', F)$ of face-neighbors $E'$ of $E$ together with the respective face
$F$ of $E$.
\end{definition}

Using this notation we obtain from equation~\eqref{eq:advectionvolint}
with $V=E$:
\begin{equation}
\label{eq:advectionint}
 0 = \frac{\partial}{\partial t}
 \int_E \phi\,dx + \int_{\partial E} \phi u \cdot \vec n\,ds = 
 \frac{\partial}{\partial t}
 \int_E \phi\, dx + \sum_{(E',F)\in \omega(E)}\int_{F}\phi u \cdot \vec n\,ds.
\end{equation}

We now discretize the equation in time using a constant time step $\Delta t$
and additionally discretize the right-hand side integrals as fluxes
$\psi(E,E';F)$. For the latter approximation we use an upwind
method~\cite{CourantIsaacsonRees52,GentryMartinDaly66}.
To this end, let $A(F)$ be the area of the face $F$. Furthermore, let $\hat u$
be the flow $u$ evaluated at the midpoint of $F$.
We then define
\begin{equation}
  \label{eq:advection-flux}
 \psi(E,E';F) := 
  \begin{cases}
    \phi_{E,t}  (\vec n\cdot \hat u) A(F)\quad \textrm{if }\vec n\cdot\hat u \geq 0,\\
    \phi_{E',t} (\vec n\cdot \hat u) A(F)\quad \textrm{if }\vec n\cdot\hat u < 0;
  \end{cases}
\end{equation}
see also Figure~\ref{fig:advect-flux}, left.
Note, that for a non-hanging face-connection we have
\begin{equation}
  \psi(E,E',F) = -\psi(E',E,F).
\end{equation}

We use the fluxes $\psi$ and the approximation of $\phi$ in
equation~\eqref{eq:advectionint} and discretize the time derivative in order to
obtain
\begin{subequations}
\label{eq:advectiontimestep}
\begin{align}
 0 &= \vol(E)\frac{\phi_{E,t+\Delta t}-\phi_{E,t}}{\Delta t}
 + \sum_{(E',F) \in \omega(E)}\psi(E,E';F)\\
\Rightarrow \phi_{E, t+\Delta t} &= 
  \phi_{E,t}
  -\frac{\Delta t}{\vol (E)} \sum_{(E',F) \in \omega(E)} \psi(E,E';F).
\end{align}
\end{subequations}

Thus, in each time step $t+\Delta t$ we iterate over the (process local)
elements and for each element $E$ we iterate over all of its faces $F$ with
face-neighbors $E'$ computing $\psi(E,E';F)$. We then compute the value
$\phi_{E,t+\Delta t}$ according to \eqref{eq:advectiontimestep} and continue
with the next element.
We store the values $\psi(E,E';F)$ in order to reuse them for the computation
of $\phi_{E',t+\Delta t}$.
\begin{remark}
This iterating over the elements is a primitive version of the \texttt{Iterate}
functionality. It remains a future project to implement a recursive
version of \texttt{Iterate} as in~\cite{IsaacBursteddeWilcoxEtAl15}
 into \tetcode.
\end{remark}

\subsection{Hanging faces and face-neighbors}
Since we use adaptive forests with non-conforming refinement methods, 
it is possible that an element $E$ has more than one face-neighbor across
a face $F$. 
Let $E'$ be such a face-neighbor and let $F'$ be the corresponding face of
$E'$. Hence, $F'$ is a subface of $F$.
In this case, we must compute the face integral of \eqref{eq:advectionint} only
over the part of the face of $E$ that coincides with $F'$.
Thus we set
\begin{equation}
  \label{eq:advection-flux-hang}
  \psi (E,E';F) = -\psi (E',E;F') = \begin{cases}
    -\phi_{E',t}  (\vec n'\cdot \hat u') A(F')\quad \textrm{if }\vec n'\cdot\hat u' \geq 0,\\
    -\phi_{E,t} (\vec n'\cdot \hat u') A(F')\quad \textrm{if }\vec n'\cdot\hat u' < 0.
  \end{cases}
\end{equation}
See the right-hand side of Figure~\ref{fig:advect-flux} for an illustration.

Hence, given a leaf element $E$ in the forest and a face $F$ of it, we need to
compute all face-neighbor leaf elements $E'$ of $E$ across the face $F$. This
neighbor finding poses a challenge for adaptive forests when the face-neighbors
of $E$ are not of the same level $\ell$ as $E$. If we assume that the forest is
2:1 balanced the possible refinement levels for face-neighbors restrict to
$\ell-1$, $\ell$, and $\ell+1$. In particular, all face-neighbors across the
same face have the same refinement level. Using this, we can find the
face-neighbors in the forest with the
\texttt{t8\_forest\_half\_face\_neighbors} function from
Section~\ref{sec:ghost-halfneighbors} and appropriate binary searches in the
arrays of leaf elements and ghost elements.
\begin{remark}
 Since we assume 2:1 balance of the forest to compute the face-neighbor
 leaves we call \texttt{Ripple-balance} after we adapt the forest to ensure
  the balance property.
\end{remark}

\begin{figure}
 \center
 \def\svgwidth{0.45\textwidth}
\begingroup%
  \makeatletter%
  \providecommand\color[2][]{%
    \errmessage{(Inkscape) Color is used for the text in Inkscape, but the package 'color.sty' is not loaded}%
    \renewcommand\color[2][]{}%
  }%
  \providecommand\transparent[1]{%
    \errmessage{(Inkscape) Transparency is used (non-zero) for the text in Inkscape, but the package 'transparent.sty' is not loaded}%
    \renewcommand\transparent[1]{}%
  }%
  \providecommand\rotatebox[2]{#2}%
  \ifx\svgwidth\undefined%
    \setlength{\unitlength}{242.76008301bp}%
    \ifx\svgscale\undefined%
      \relax%
    \else%
      \setlength{\unitlength}{\unitlength * \real{\svgscale}}%
    \fi%
  \else%
    \setlength{\unitlength}{\svgwidth}%
  \fi%
  \global\let\svgwidth\undefined%
  \global\let\svgscale\undefined%
  \makeatother%
  \begin{picture}(1,0.71968339)%
    \put(0,0){\includegraphics[width=\unitlength,page=1]{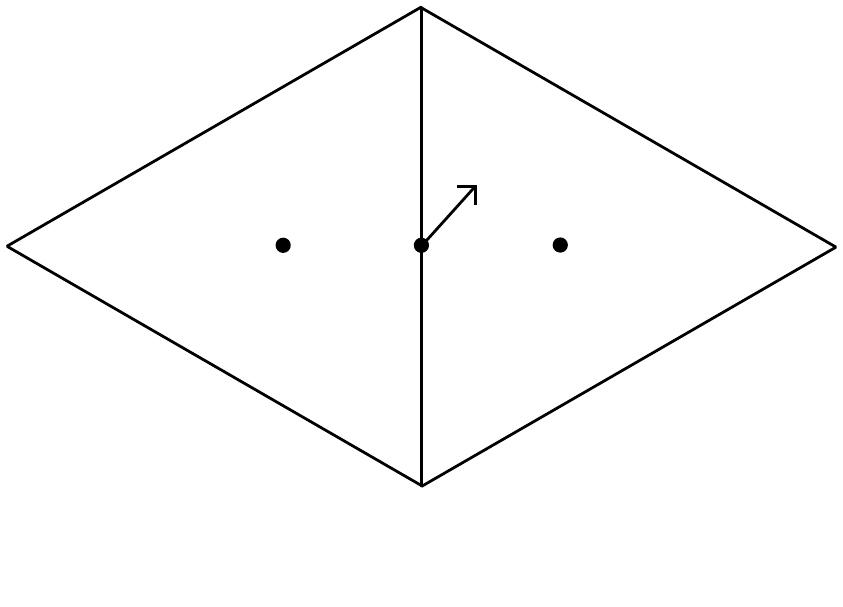}}%
    \put(0.15322796,0.15534294){\color[rgb]{0,0,0}\makebox(0,0)[lb]{\smash{$E$}}}%
    \put(0.7412732,0.14221038){\color[rgb]{0,0,0}\makebox(0,0)[lb]{\smash{$E'$}}}%
    \put(0.36629979,0.01627767){\color[rgb]{0,0,0}\makebox(0,0)[lb]{\smash{$F$}}}%
    \put(0,0){\includegraphics[width=\unitlength,page=2]{compute_flux_tex.pdf}}%
    \put(0.27908507,0.37239329){\color[rgb]{0,0,0}\makebox(0,0)[lb]{\smash{$m_E$}}}%
    \put(0.62074279,0.37480025){\color[rgb]{0,0,0}\makebox(0,0)[lb]{\smash{$m_E'$}}}%
    \put(0.50968062,0.51544591){\color[rgb]{0,0,0}\makebox(0,0)[lb]{\smash{$\hat u$}}}%
    \put(0,0){\includegraphics[width=\unitlength,page=3]{compute_flux_tex.pdf}}%
    \put(0.52326297,0.36626841){\color[rgb]{0,0,0}\makebox(0,0)[lb]{\smash{$\vec n$}}}%
  \end{picture}%
\endgroup%
\hfill
 \def\svgwidth{0.35\textwidth}
\begingroup%
  \makeatletter%
  \providecommand\color[2][]{%
    \errmessage{(Inkscape) Color is used for the text in Inkscape, but the package 'color.sty' is not loaded}%
    \renewcommand\color[2][]{}%
  }%
  \providecommand\transparent[1]{%
    \errmessage{(Inkscape) Transparency is used (non-zero) for the text in Inkscape, but the package 'transparent.sty' is not loaded}%
    \renewcommand\transparent[1]{}%
  }%
  \providecommand\rotatebox[2]{#2}%
  \ifx\svgwidth\undefined%
    \setlength{\unitlength}{372.97885742bp}%
    \ifx\svgscale\undefined%
      \relax%
    \else%
      \setlength{\unitlength}{\unitlength * \real{\svgscale}}%
    \fi%
  \else%
    \setlength{\unitlength}{\svgwidth}%
  \fi%
  \global\let\svgwidth\undefined%
  \global\let\svgscale\undefined%
  \makeatother%
  \begin{picture}(1,0.90176088)%
    \put(0,0){\includegraphics[width=\unitlength,page=1]{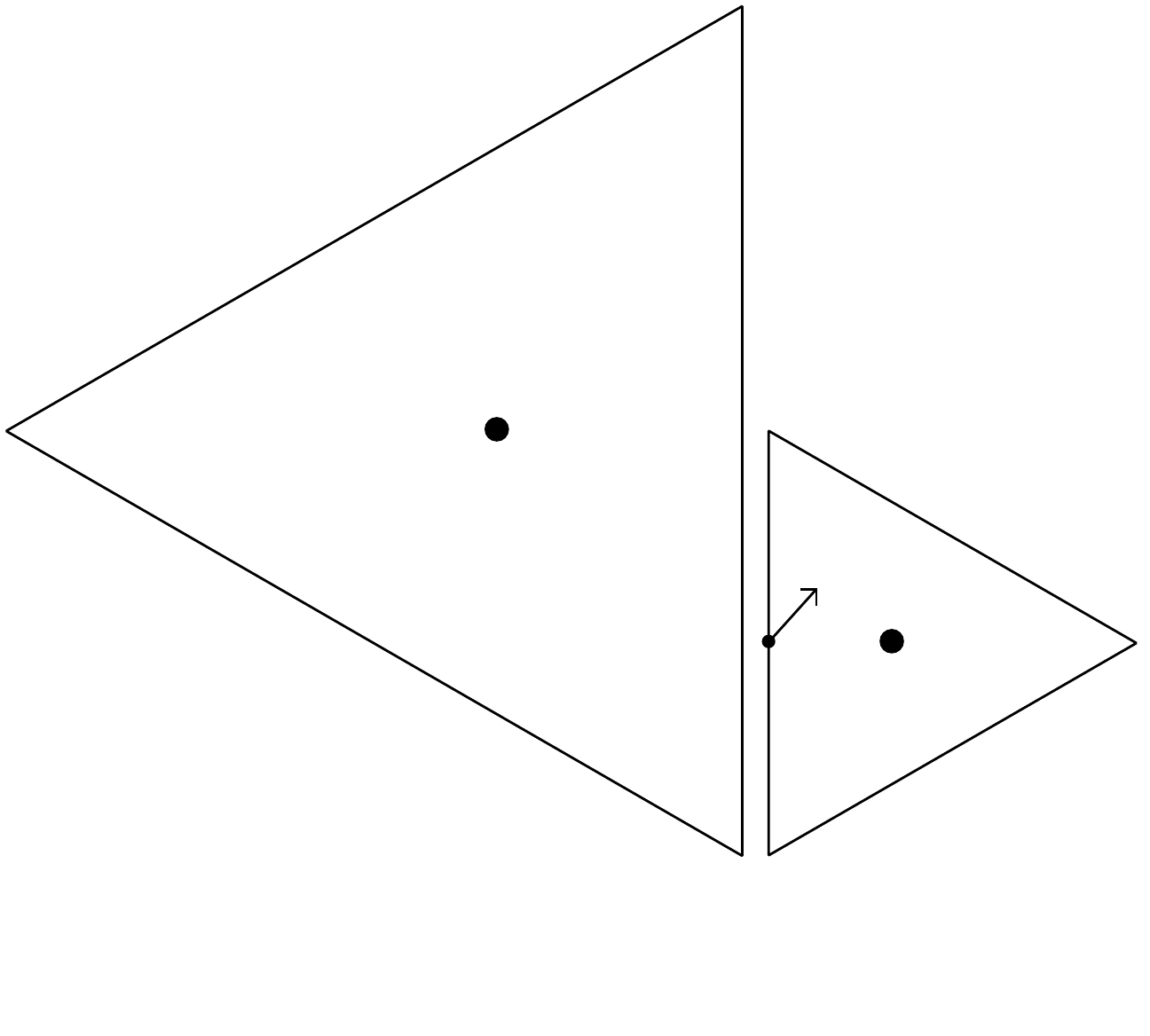}}%
    \put(0.28656924,0.1815413){\color[rgb]{0,0,0}\makebox(0,0)[lb]{\smash{$E$}}}%
    \put(0.83124867,0.17406618){\color[rgb]{0,0,0}\makebox(0,0)[lb]{\smash{$E'$}}}%
    \put(0.47351103,0.04920271){\color[rgb]{0,0,0}\makebox(0,0)[lb]{\smash{$F$}}}%
    \put(0,0){\includegraphics[width=\unitlength,page=2]{compute_flux_hanging_tex.pdf}}%
    \put(0.39958659,0.47295503){\color[rgb]{0,0,0}\makebox(0,0)[lb]{\smash{$m_E$}}}%
    \put(0.76030642,0.28255366){\color[rgb]{0,0,0}\makebox(0,0)[lb]{\smash{$\scriptstyle m_E'$}}}%
    \put(0.67515019,0.39983418){\color[rgb]{0,0,0}\makebox(0,0)[lb]{\smash{$\scriptstyle\hat u'$}}}%
    \put(0,0){\includegraphics[width=\unitlength,page=3]{compute_flux_hanging_tex.pdf}}%
    \put(0.83124867,0.54350026){\color[rgb]{0,0,0}\makebox(0,0)[lb]{\smash{$E''$}}}%
    \put(0.75548034,0.65520503){\color[rgb]{0,0,0}\makebox(0,0)[lb]{\smash{$\scriptstyle m_E''$}}}%
    \put(0.67742517,0.74538067){\color[rgb]{0,0,0}\makebox(0,0)[lb]{\smash{$\scriptstyle\hat u''$}}}%
    \put(0,0){\includegraphics[width=\unitlength,page=4]{compute_flux_hanging_tex.pdf}}%
    \put(0.75931121,0.05227552){\color[rgb]{0,0,0}\makebox(0,0)[lb]{\smash{$F'$}}}%
    \put(0,0){\includegraphics[width=\unitlength,page=5]{compute_flux_hanging_tex.pdf}}%
    \put(0.89632035,0.85920672){\color[rgb]{0,0,0}\makebox(0,0)[lb]{\smash{$F''$}}}%
  \end{picture}%
\endgroup%
 \caption[Illustration for flux computations.]
{Illustration of our notations for computing the fluxes as
in~\eqref{eq:advection-flux} and~\eqref{eq:advection-flux-hang}.
Left: A non-hanging face connection of a triangle element $E$ across face $F$.
Depending on the vector product $\hat u \cdot \vec n$ we either use the value
of $\phi$ on $E$ or on $E'$ for the numerical flux and multiply with the 
area $A(F)$ of the face $F$.
Right: A hanging face connection of a triangle element $E$ across face $F$
(exploded view). In order to compute the flux across $F$, we sum up the two
fluxes across the neighbor elements' faces $F'$ and $F''$.}
\figlabel{fig:advect-flux}
\end{figure}

\subsection{The CFL number}
In order for our explicit FV scheme to be numerically stable, we have to choose the
time step $\Delta t$ appropriately. If we choose it too big, the
method becomes unstable. However, decreasing the time step increases
the computational load, and thus increases the overall runtime. A well-known
way to control the time step is the CFL
number~\cite{CourantFriedrichsLewy28}.
\begin{definition}
 The \textbf{CFL-number} of an element $E$ at time $t$ is
 \begin{equation}
   C_{E,t}:=u(m_E,t)\frac{\Delta t}{\vol(E)^{\frac{1}{d}}},
  \end{equation}
and the \textbf{global CFL-number} is
 \begin{equation}
   C_t = \max_{E} C_{E,t}.
 \end{equation}
\end{definition}
In our solver, we choose an initial global CFL-number $C$ and then compute
$\Delta t$ such that $C_0 = C$. We use this time step throughout the
simulation, but it would be possible to change the time step during the
simulation~\cite{Johnson88}.

The FV theory states that the method is numerically stable for small
values of $C$, i.e.\ $C<1$~\cite{WarburtonHagstrom08}. 
In practice, we observe that we may need lower values in order to ensure
stability. One reason is that $\vol(E)$ changes over time when the
mesh is readapted.

\subsection{The refinement criterion}
\label{sec:refcrit}
In order to properly decide for a point $x$ whether $x\in\Omega^i_t$ it is 
important to resolve $\Gamma_t$ as accurately as possible.
This motivates refining the forest close to $\Gamma_t$, or
rather its approximation $\setm{E}{\lvert\phi_{E,t}\rvert < \epsilon}$, as in~\cite{Albrecht16}.
Consider a leaf element $E$ and let $h=\vol(E)^{\frac{1}{d}}$ be an 
approximation of its diameter and $b>0$ a given
parameter, then we refine $E$ at time step $t$ if
\begin{equation}
\label{eq:refcrit}
 \lvert\phi_{E,t}\rvert < bh.
\end{equation}
The parameter $b$ controls the width of the refinement band around $\Gamma_t$.
We coarsen a family of elements if the equation does not hold for each of its
members. Furthermore, we do not refine more than a specified maximum refinement
level and do not coarsen above a minimum refinement level. 
This does not lead to element volume inbalances at the finest or coarsest
level, since all trees have approximately the same volume.

\subsection{Error measurement}

In our examples we use level-set functions with closed zero level-sets for the
initial function $\phi_0$,
for example a circle, or a sphere.
Since the flow $u$ is divergence free, the volume of the sets
$\Omega^i_t$ remains constant over time, if we do not have inflow or
outflow.
It is a typical behavior of FV solvers that we numerically loose or gain volume in
$\Omega^i_t$ over the computation. 
Whereby volume loss occurs in the convex region, which is the case for
$\Omega^2_t$ with our choices of initial conditions.
Hence, we use the amount of lost volume in $\Omega^2_t$ as a
measurement for the computational error,
\begin{equation}
 \mathcal E_t^{\vol}:= 1 - \frac{\vol(\Omega^2_t)}{\vol(\Omega_0^2)}.
\end{equation}
We compute $\vol(\Omega^2_t)$ by summing up the volumes of those 
elements for which $\phi_{E,t}<0$.
\begin{remark}
  It is an inherent property of the method that, provided we have
  no in- or outflow, the integral of the
  level-set function
  \begin{equation}
    \int_\Omega \phi(t,x)\, dx
  \end{equation}
  does not change over time (up to computational errors).
  This is due to the fact that the numerical flux~\eqref{eq:advection-flux} 
  from element $E$ to $E'$ is exactly the negative as from element $E'$ to $E$.
\end{remark}

\section{Handling application data}
For each leaf element $E$ we store some application specific data. In
particular, this includes the value $\phi_{E,t}$ of the approximated solution
inside the element at the current time step. We store these data in an array
which has one entry for each local leaf element and one for each ghost element.
The data for the local elements is stored in order of their SFC index, while
the order of the ghost elements is implicitly given by the \texttt{Ghost}
algorithm.

In addition to the main mesh handling algorithms \texttt{New}, \texttt{Adapt},
\texttt{Partition}, \texttt{Ghost}, and \texttt{Balance} from the list in
Section~\ref{sec:hlalgos}, we need subroutines in order to manage the
application data. These perform the interpolation of data to a new forest after
the mesh changes, the redistribution of the data after \texttt{Partition}, and
the exchange of data of the ghost elements.
These routines are common in AMR, see for
example~\cite{BursteddeGhattasStadlerEtAl08,TautgesKraftcheckBertramEtAl12},
and we briefly summarize them in this section.

\subsection{Interpolation}
After we have computed the values $\phi_{E,t}$ for each local element, we may
modify the forest with \texttt{Adapt} to refine and coarsen the forest
and \texttt{Balance} to reestablish the 2:1 balance. Thus, the current
forest $\forest F$ is changed to a new forest $\forest F'$.
For all new elements in $\forest F'$ we need to calculate an interpolated 
value of $\phi$, which we do in the following way as 
described in~\cite{BursteddeGhattasStadlerEtAl08}.
We restrict to non-recursive refinement. Hence, all elements in
$\forest F'$ result from an element in $\forest F$ by either refining once,
coarsening once, or keeping the element as it is.
In the first case, an element $E$ in $\forest F$ is refined into $n$ children
$E_i$, $0\leq i<n$, and we set the value of the new elements to the one of the
parent: 
\begin{equation}
  \label{eq:advectinter}
  \phi_{E_i,t}=\phi_{E,t}
\end{equation}
 for all $0\leq i < n$. In the second case,
a family $\set{E_i}$ of $n$ elements is coarsened into their parent $E$. We 
assign the average of the values of the $E_i$, thus 
\begin{equation}
 \label{eq:advectaverage}
 \phi_{E,t} = \frac{1}{n}\sum_i \phi_{E_i,t}.
\end{equation}
If an element is unchanged, we also do not change the value of $\phi$.

The challenge is to identify the corresponding pairs of new and old elements.
In order to do so, we iterate simultaneously through the leaf elements of
$\forest F$ and $\forest F'$ with indices $i$ and $j$, starting with $i = j =
0$. Let $\ell_i$ and $\ell_j$ be the refinement levels of the $i$-th element in
$\forest F$ and the $j$-th element in $\forest F'$.
If $\ell_i=\ell_j$, they are the same element, hence we do not change the
$\phi$ value and increase both indices by $1$. If $\ell_i = \ell_j - 1$, we
know that the element was refined into $n$
children. We carry out the interpolation \eqref{eq:advectinter} and increment
$i$ by $1$ and $j$ by $n$. If $\ell_i = \ell_j + 1$, the element is the first
in a family of $n$ elements that was coarsened and we compute the
average~\eqref{eq:advectaverage}. We then increment $i$ by $n$ and $j$ by $1$.

We refer to this routine as \texttt{Interpolate}.

\subsection{Repartitioning of data}
After the forest is adapted and the new element values are interpolated, we may
repartition the forest via \texttt{Partition} in order to maintain a balanced
load. After \texttt{Partition} the element data needs to be partitioned as
well. Partitioning the element data follows the same logic as partitioning the
elements. We call the routine \texttt{Partition\_data}.

\subsection{Ghost exchange}
Before every new time step, we need to update the function values in the ghost
elements. Thus, each process has to send the values $\phi_{E,t}$ for its
inter-process boundary leaf elements to all processes that have face-neighbors
of this boundary element. 
This operation is called \texttt{Ghost\_exchange}~\cite{BursteddeHolke16b}.

Its input is the data array with valid entries $\phi_{E,t}$ for all process
local leaf elements. On output the values in the data array corresponding to
the ghost elements are filled with the entry of the respective owner process.

The implementation is straight-forward if we store the indices of boundary
elements and of which remote processes these are ghosts at the time of creation
of the ghost layer, hence during \texttt{Ghost}.

\section{Tests on a unit cube geometry}
We run tests in 2D and 3D with $\Omega = [0,1]^d$, the $d$-dimensional unit cube
with periodic boundary conditions for $\phi$.
\subsection{The 2D test case}
For the 2D tests, we use a flow $u$ that simulates a rotation around the midpoint
$(0.5,0.5)$ of $\Omega$,
\begin{equation}
\label{eq:2dflow}
  u(x,y) = 2\pi\begin{pmatrix}
    y-0.5\\
    -(x-0.5)
  \end{pmatrix}.
\end{equation}
This $u$ is divergence free and chosen such that for the analytical solution
at time $t = 1$ we obtain 
\begin{equation}
  \phi(\cdot,1) = \phi_0.
\end{equation}
As level-set function $\phi_0$ we choose the signed distance to a circle of radius
$0.25$ with midpoint $(0.6,0.6)$, hence
\begin{equation}
  \phi_0(x,y) = \sqrt{(x-0.6)^2 + (y-0.6)^2} - 0.25.
\end{equation}

We compare three different coarse meshes of $[0,1]^2$,
which we show in the top part of Figure~\ref{fig:2dquadtrihybrid}.
The first one 
consists of one quadrilateral tree; for the second one we divide the square along
a diagonal into two triangle trees; for the third one we use a hybrid coarse mesh
of four triangle trees and two quadrilateral trees.
For the refinement levels our solver accepts two arguments: $\ell$ and $r$.
$\ell$ describes the minimum refinement level and $r$ the maximum number of additional
refinement levels. Thus, we start with a uniform level $\ell$ forest and refine
it up to level $\ell+r$ before starting the computation. We depict adaptive 
forests with an initial adaptive refinement along the zero-level set of $\phi_0$
in Figure~\ref{fig:2dquadtrihybrid}.

We show a computation of the hybrid forest with $\ell = 4$ and $r=3$ in
Figure~\ref{fig:rotation2d}.  This figure also shows the flow field $u$.

\subsection{The 3D test case}

For computations on the 3D unit cube geometry, we construct a specific flow
function.

\begin{definition}
Let $f\in C^2([0,1])$ with $f(0) = f(1) = 0$.
We define the vector field $u^f:[0,1]^3\rightarrow \IR^3$ as
\begin{subequations}
 \begin{align}
 \label{eq:ufflow}
  u^f_1 (x,y,z) &= f(x)(f'(y) - f'(z))\\
  u^f_2 (x,y,z) &= -f'(x)f(y)\\
  u^f_3 (x,y,z) &= f'(x)f(z).
 \end{align}
\end{subequations}
\end{definition}
$u^f$ has the following properties:
\begin{lemma}
\label{lem:ufdivfree}
 $u^f$ is divergence free and has no outflow from the unit cube, thus
 $u^f\cdot \vec n = 0$ at each face of $[0,1]^3$.
\end{lemma}
\begin{proof}
  The first part follows from computing the partial derivative of $u^f$ and
  the second part follows since 
 \begin{subequations}
  \begin{align}
    u^f_1(0,y,z) &= u^f_2(x,0,z) = u^f_3(x,y,0) = 0, \textrm{\quad and}\\
    u^f_1(1,y,z) &= u^f_2(x,1,z) = u^f_3(x,y,1) = 0.
  \end{align}
 \end{subequations}
\end{proof}

For our test case we use $f(x) = \sin(\pi x)$ and consider the time-dependent flow
\begin{subequations}
\label{eq:cubeflow}
  \begin{align}
    u\colon[0,1]^3\times\IR_{\geq 0}&\rightarrow \IR^3,\\
    (x,y,z,t)&\mapsto\begin{cases} \phantom{-}u^f(x,y,z),\,t<0.5,\\-u^f(x,y,z),\,t\geq 0.5.\end{cases}
  \end{align}
\end{subequations}
In analogy to the 2D case, the analytical solution to the advection equation
satisfies $\phi(\cdot,1) = \phi_0$, which allows us to compute the numerical 
error at time $t=1$. We show the flow $u^f$ in Figure~\ref{fig:3Dflowincube}.

For the initial level-set function $\phi_0$ we use a signed distance to 
a sphere of radius $0.25$ and midpoint $(0.6,0.6,0.6)$, hence
\begin{equation}
  \phi_0(x,y,z) = \sqrt{(x-0.6)^2 + (y-0.6)^2 + (z-0.6)^2} - 0.25.
\end{equation}

We use three different coarse meshes of $[0,1]^3$. Firstly, one singular
hexahedral tree; secondly, six tetrahedral trees as in
Figure~\ref{fig:sechstetra}; and thirdly, a hybrid mesh of six tetrahedral trees,
six prism trees, and four hexahedral trees, which we depict in
Figure~\ref{fig:3Dflowincube}.
Figure~\ref{fig:hybrid-3d-cutout} additionally shows a detailed cut out view
of the hybrid mesh, displaying refined forest elements within the unit cube.

\begin{figure}
\center
 \includegraphics[width=0.307\textwidth]{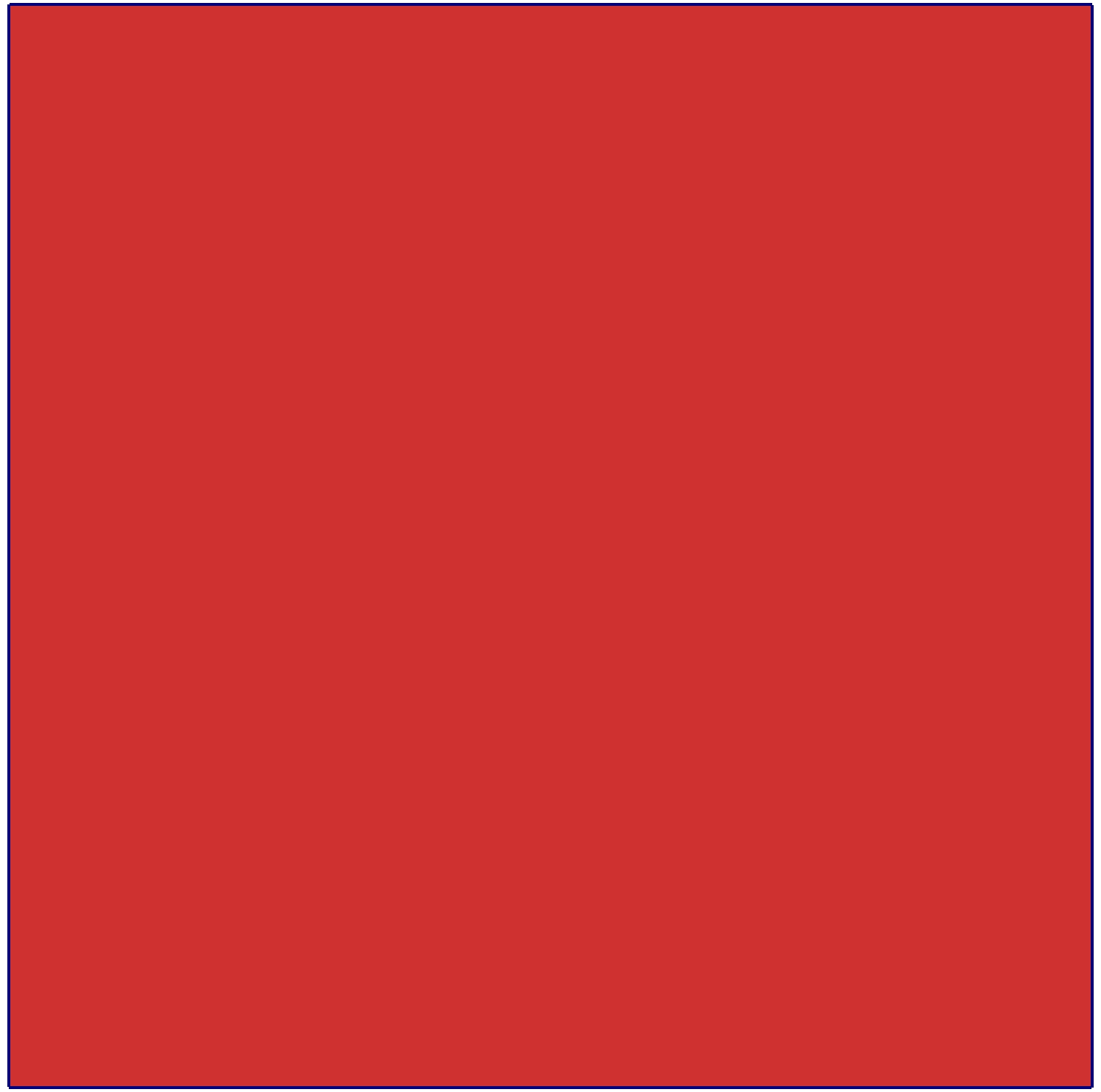}\hspace{1.1ex}
 \includegraphics[width=0.309\textwidth]{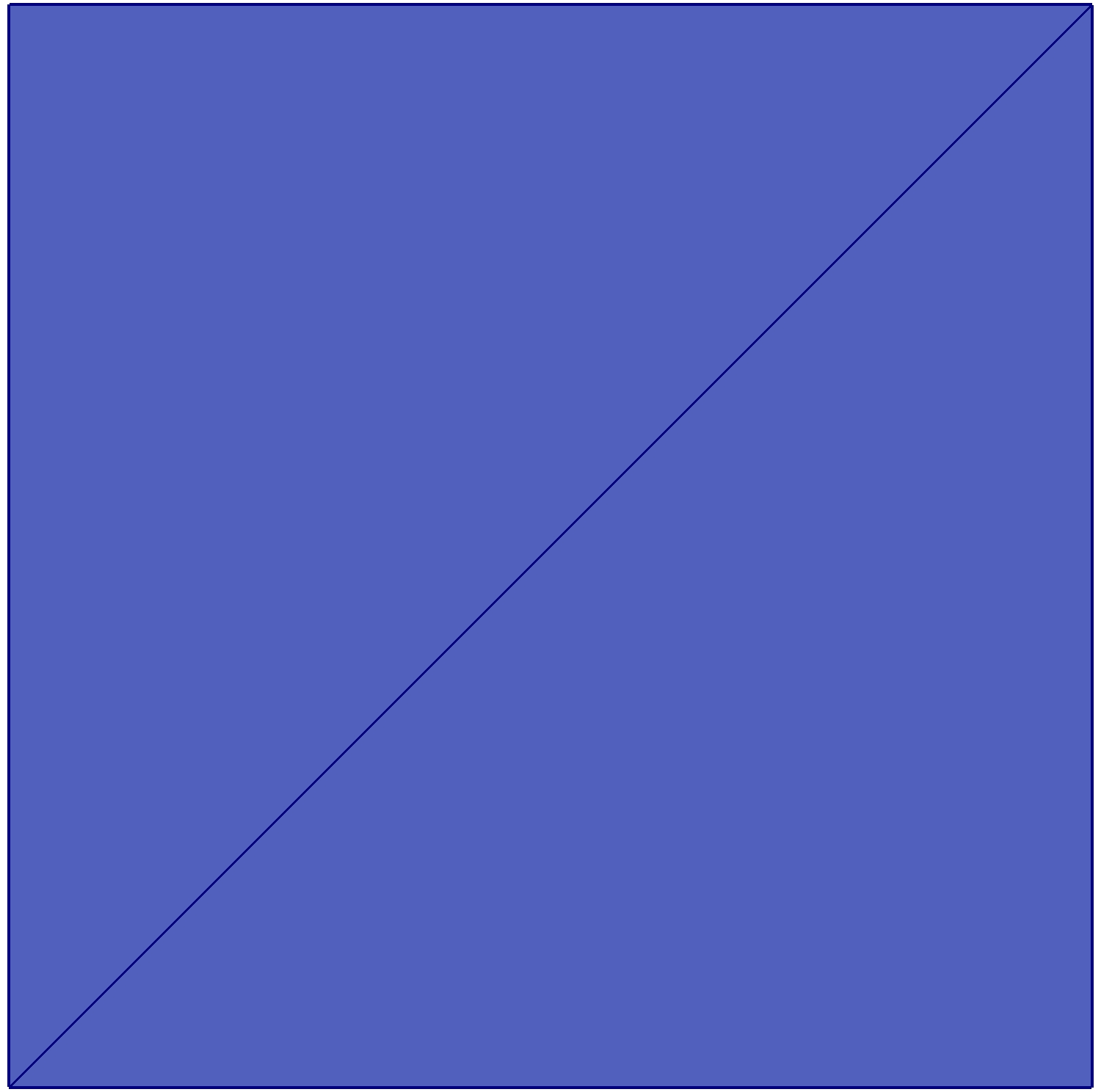}\hspace{1.1ex}
 \includegraphics[width=0.307\textwidth]{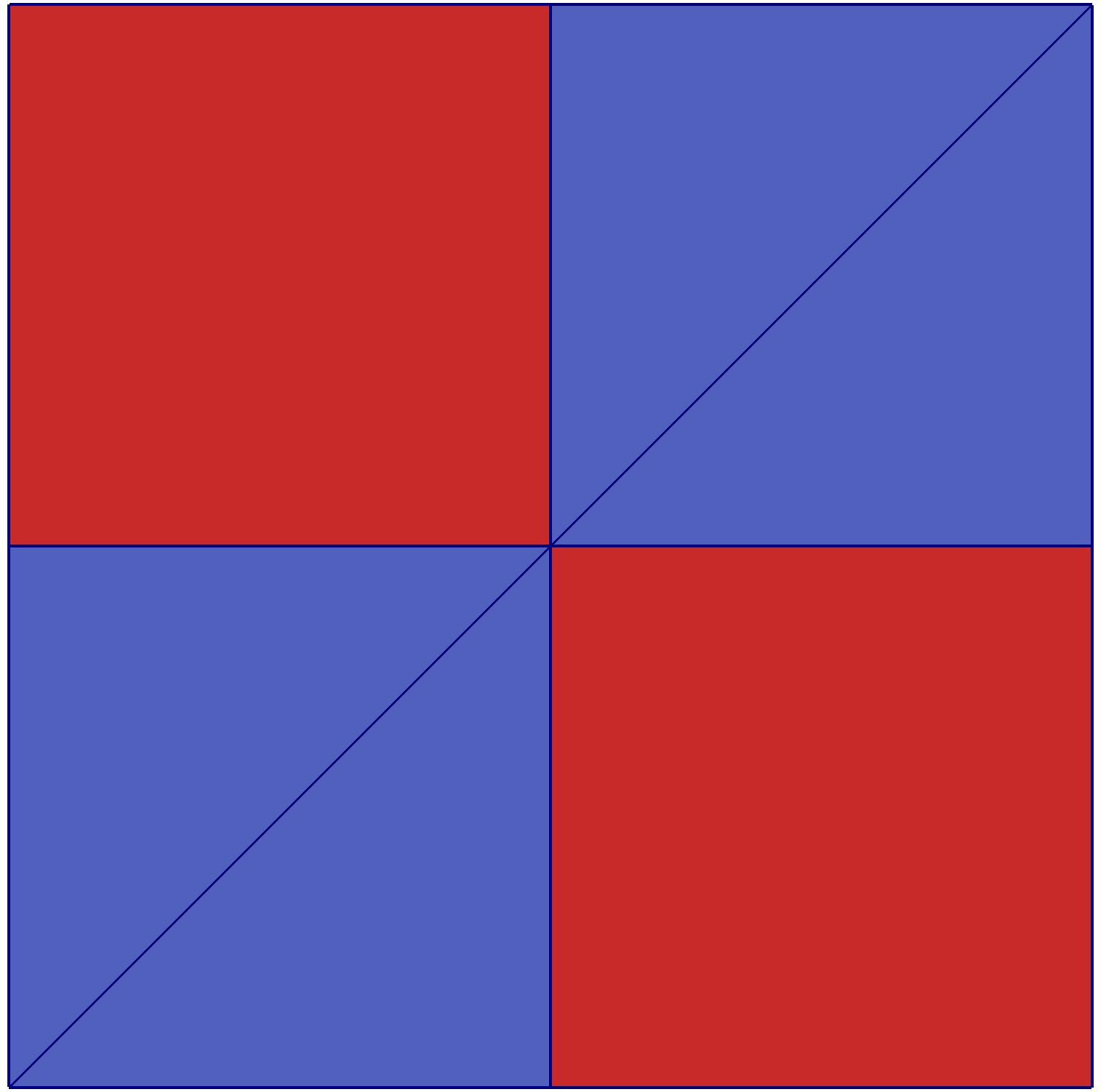}\\
 \includegraphics[width=0.32\textwidth]{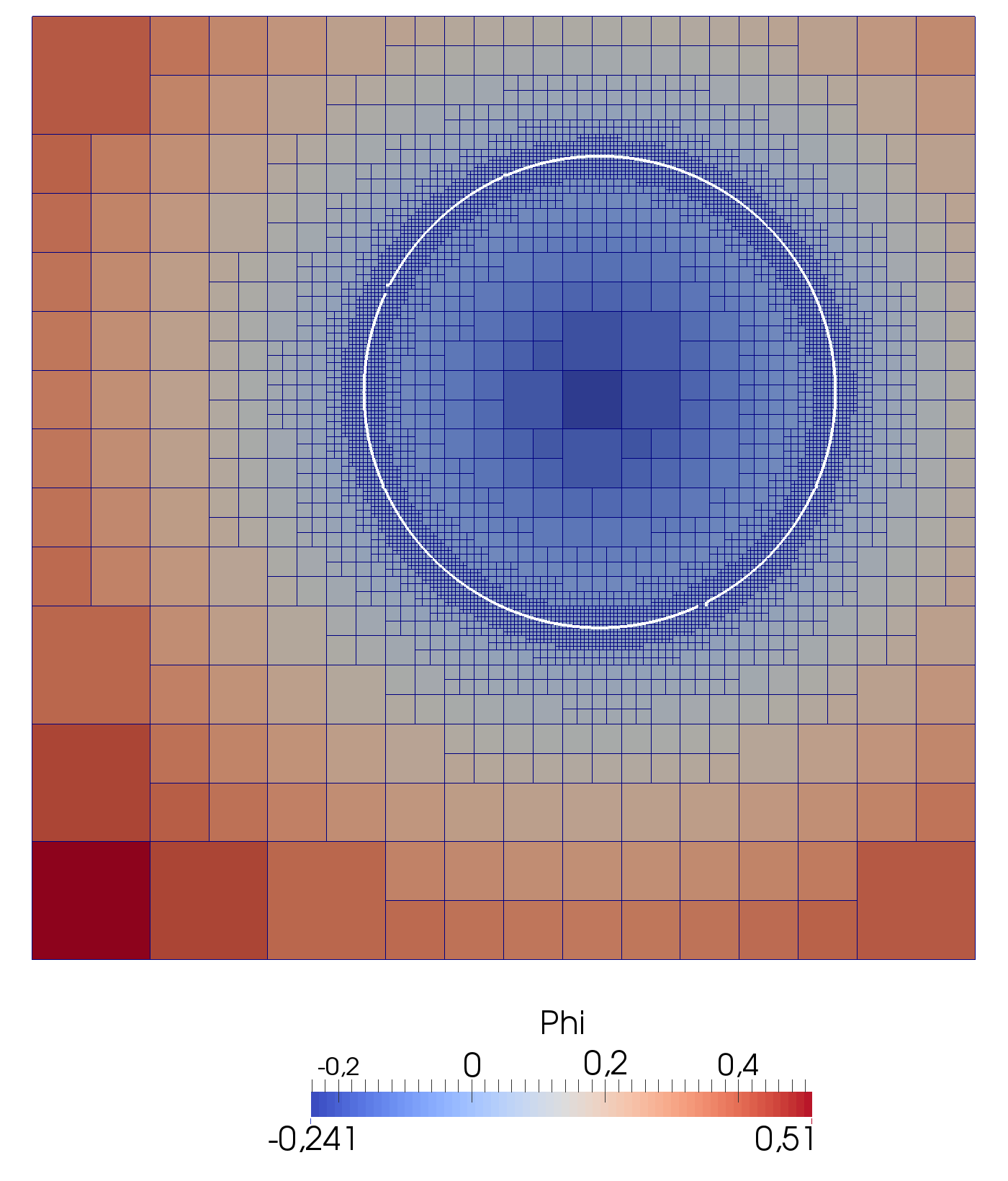}
 \includegraphics[width=0.32\textwidth]{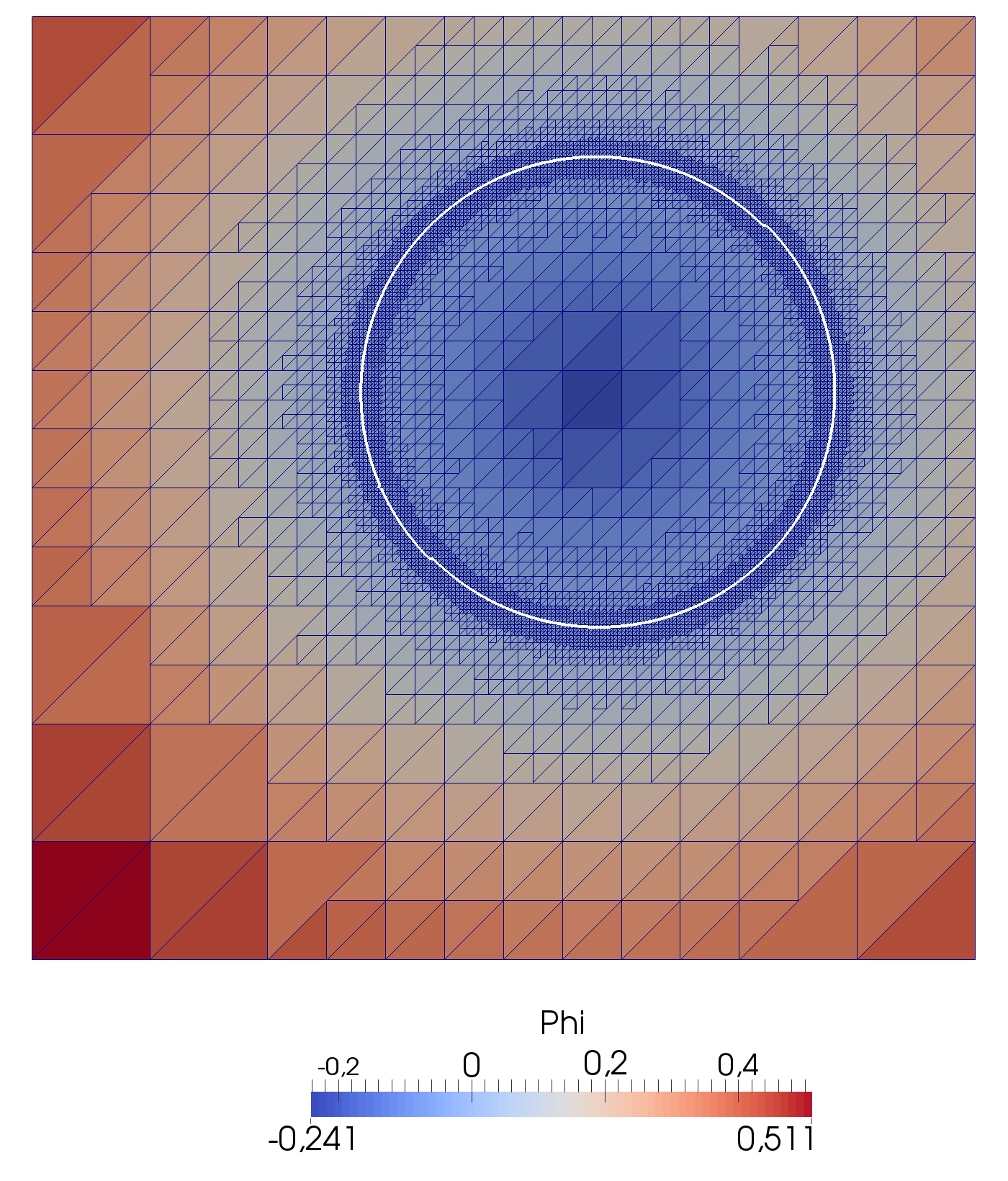}
 \includegraphics[width=0.32\textwidth]{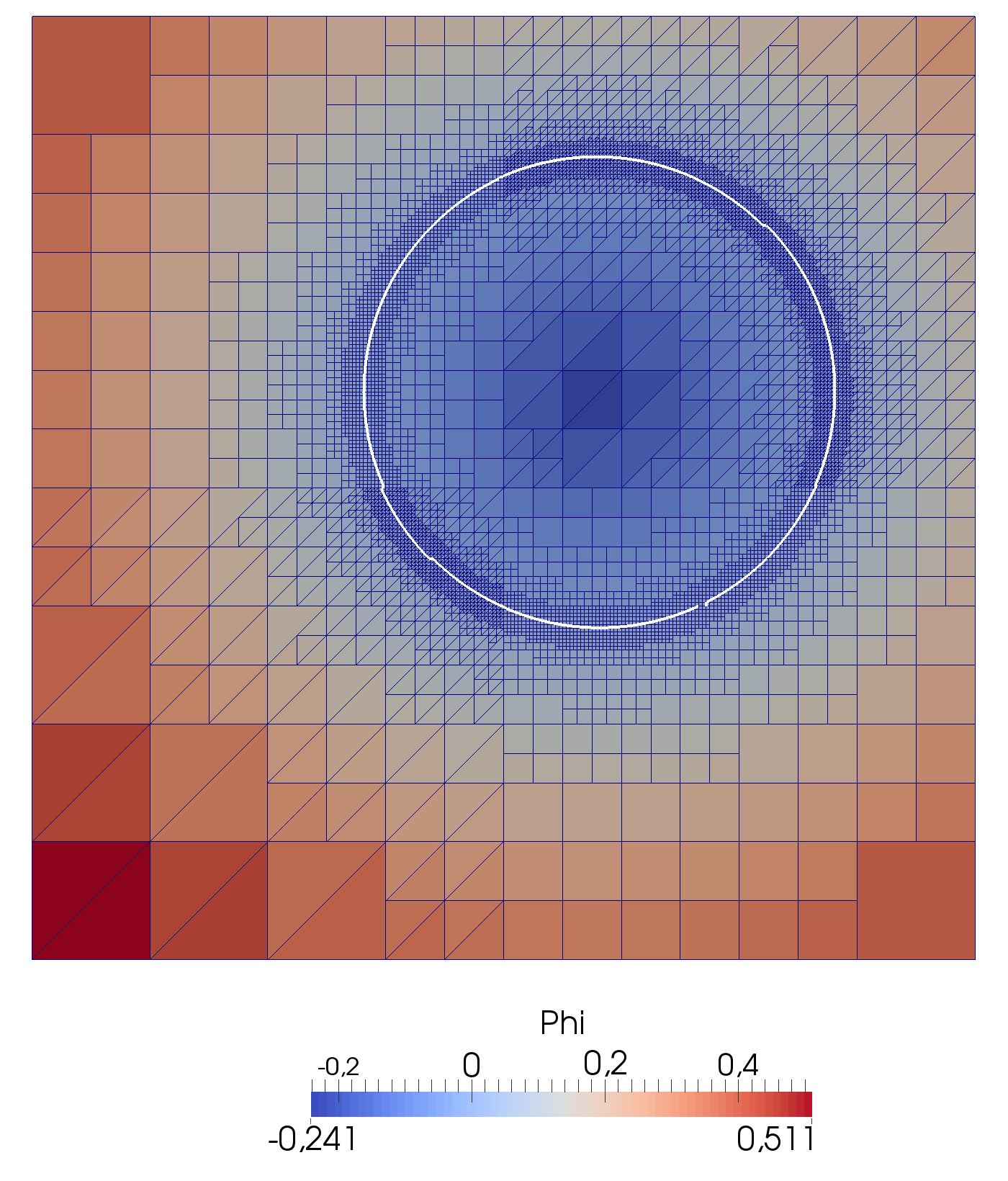}
\caption[2D unit square for advection (quad/triangle/hybrid).]{
Adapting three different unit square meshes at a circle with midpoint
$(0.6,0.6)$ and radius $0.25$, modeled as the zero level-set of a function
$\phi$. We show the coarse mesh in the top row and the adapted forest mesh in
the bottom row at time $t=0$. We model the unit square using one quadrilateral tree (left, red),
two triangle trees (middle, blue), and as a hybrid mesh of four triangular and two
quadrilateral trees (right). Starting from a uniform forest, we use six
refinement levels. The initial uniform forest has level $2$ in the
quadrilateral and triangle case, and level $1$ in the hybrid case. In the
bottom row the color represents the values $\phi_{E,t}$ of the approximated
level-set function.
}
\figlabel{fig:2dquadtrihybrid}
\end{figure}

\begin{figure}
\center
 \includegraphics[width=0.48\textwidth]{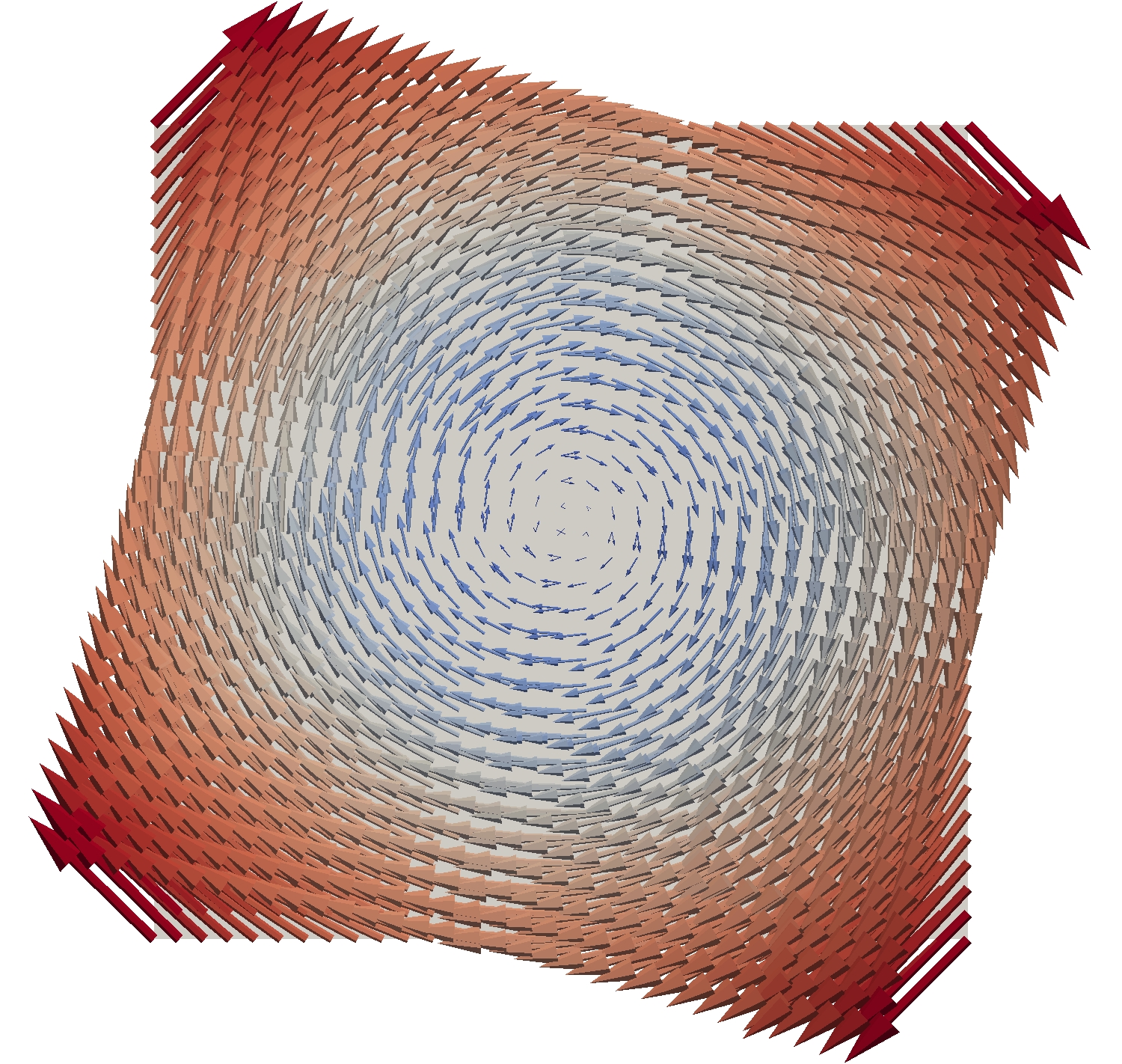}
 \includegraphics[width=0.48\textwidth]{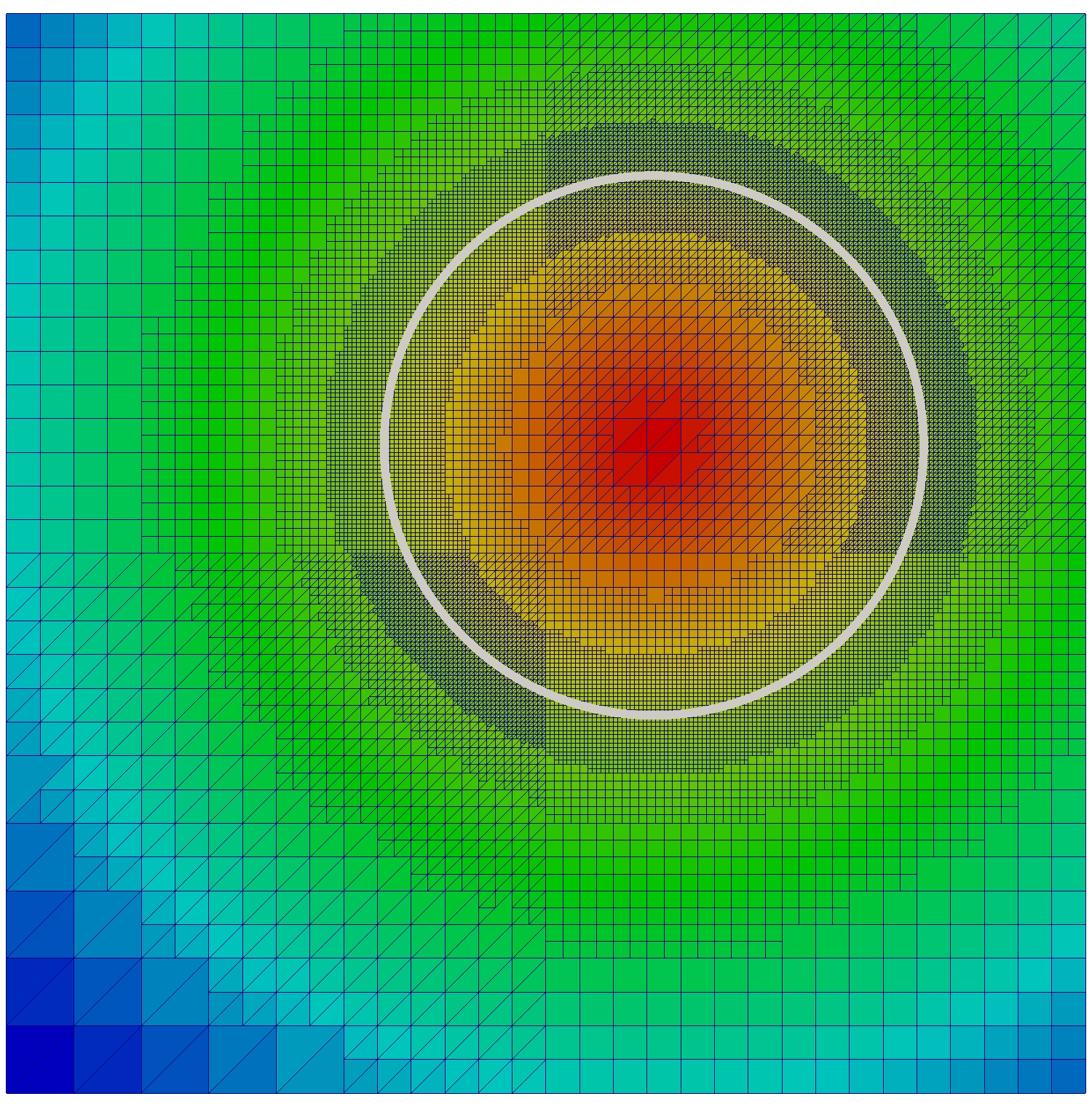}\\
 \includegraphics[width=0.48\textwidth]{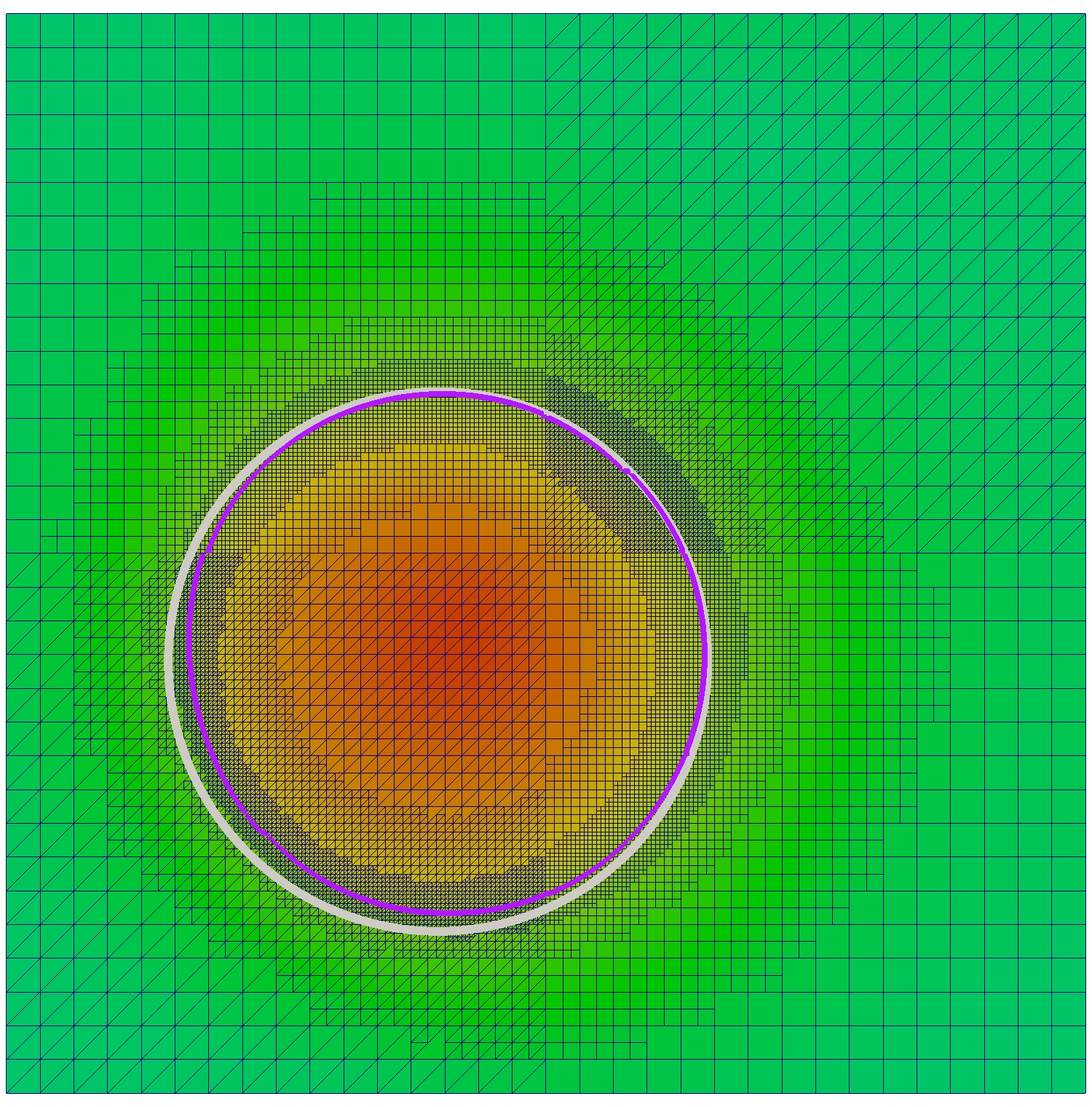}
 \includegraphics[width=0.48\textwidth]{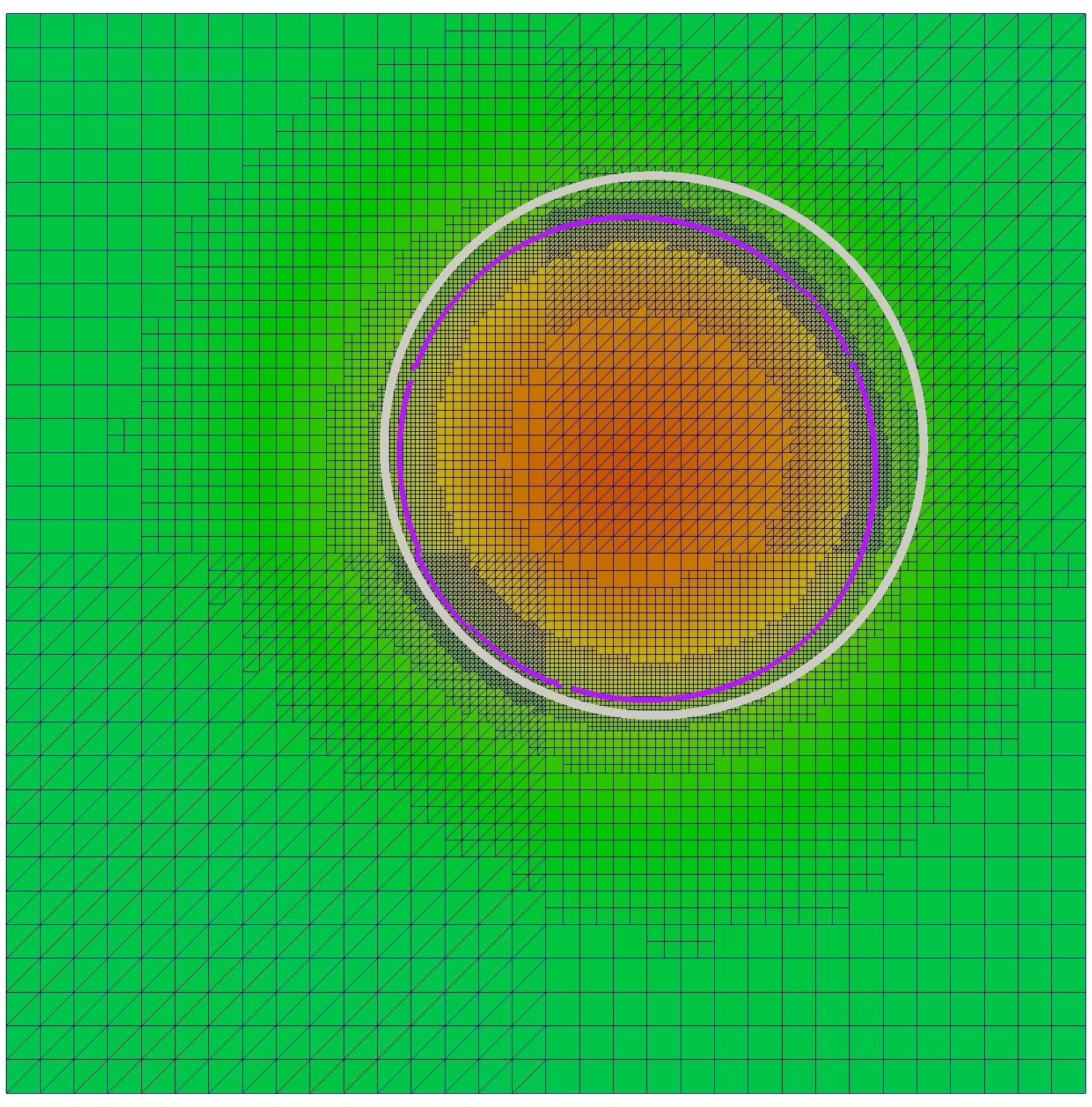}
\caption[Simulation of the 2D test case on a 2D hybrid mesh.]
{A simulation of the 2D test case on the 2D hybrid
mesh. Top left: The rotational flow $u$ from~\eqref{eq:2dflow}.
Top right, bottom left, bottom right: Times $t=0$, $t=0.5$, and $t=1$
of the simulation with initial refinement level $\ell = 4$, and with $r=3$
adaptive refinement levels.
The maximum refinement level is 7 and there are about 20,000 elements.
The number of elements in an equivalent uniform level 7 forest is 98,304. 
We depict the zero-level set of the analytical solution in white and the
zero-level set of the computed solution in purple.
}
\figlabel{fig:rotation2d}
\end{figure}
\begin{figure}
\center
 \includegraphics[width=0.47\textwidth]{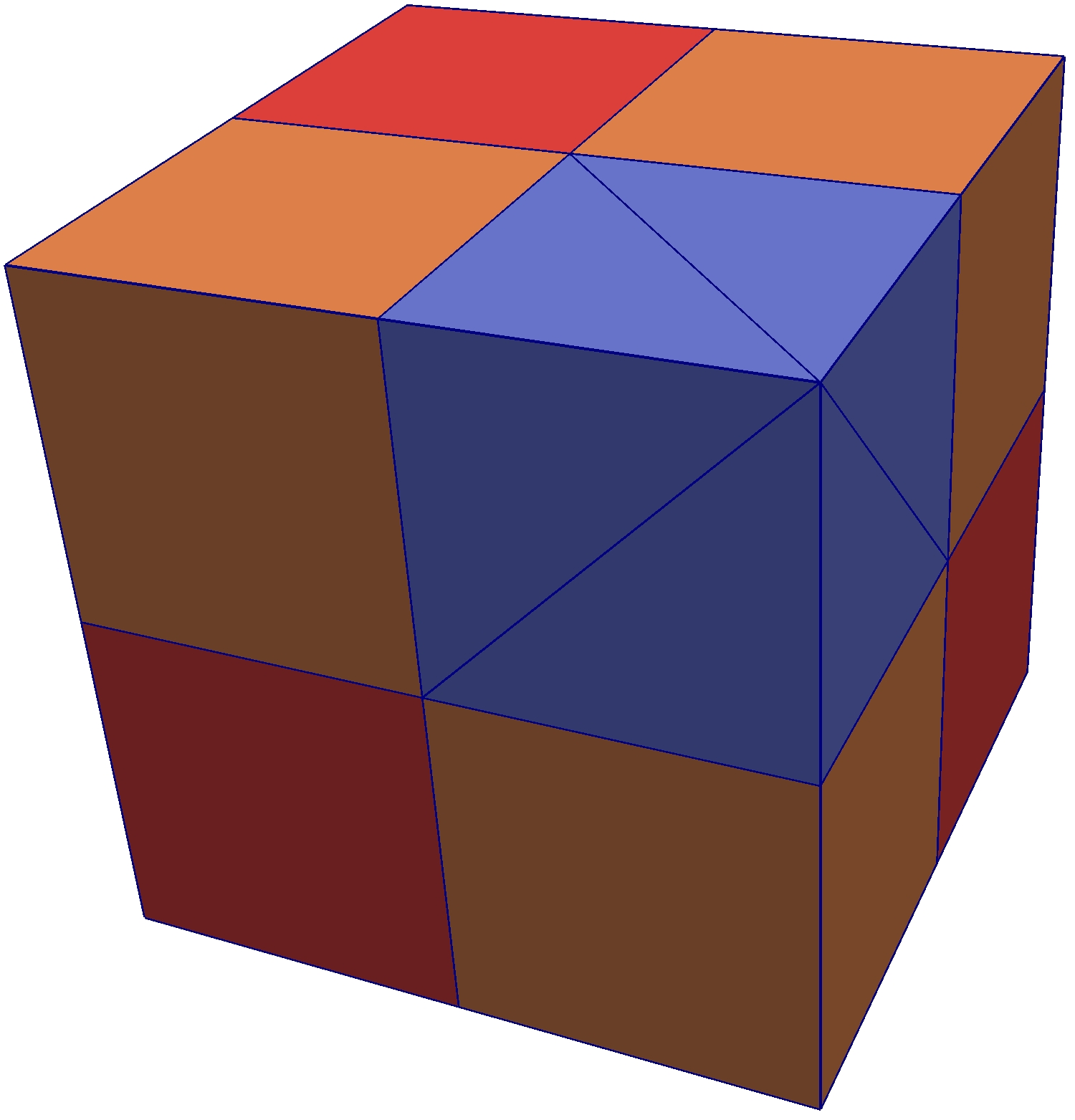}
 \includegraphics[width=0.47\textwidth]{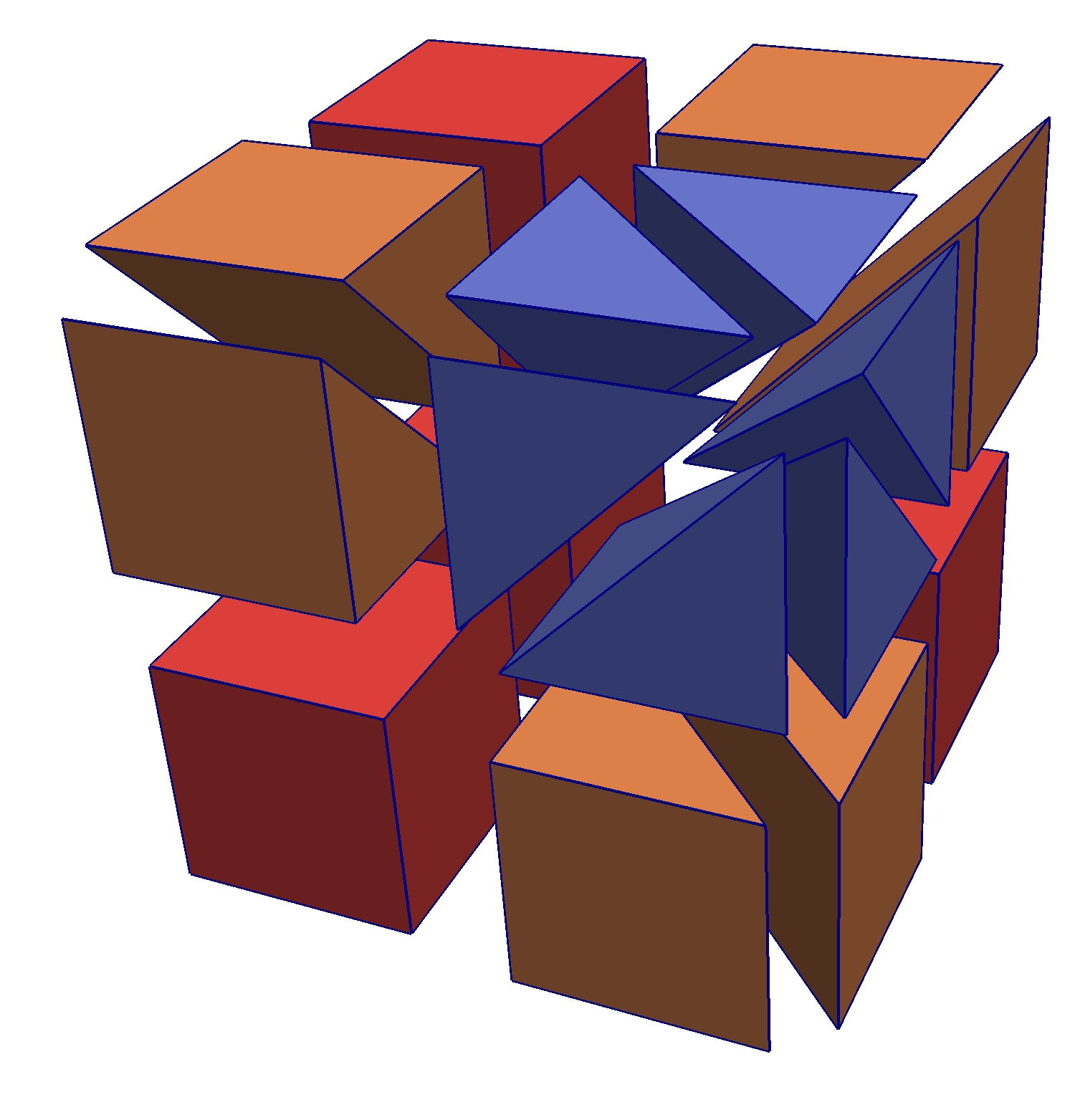}\\
 \includegraphics[width=0.52\textwidth]{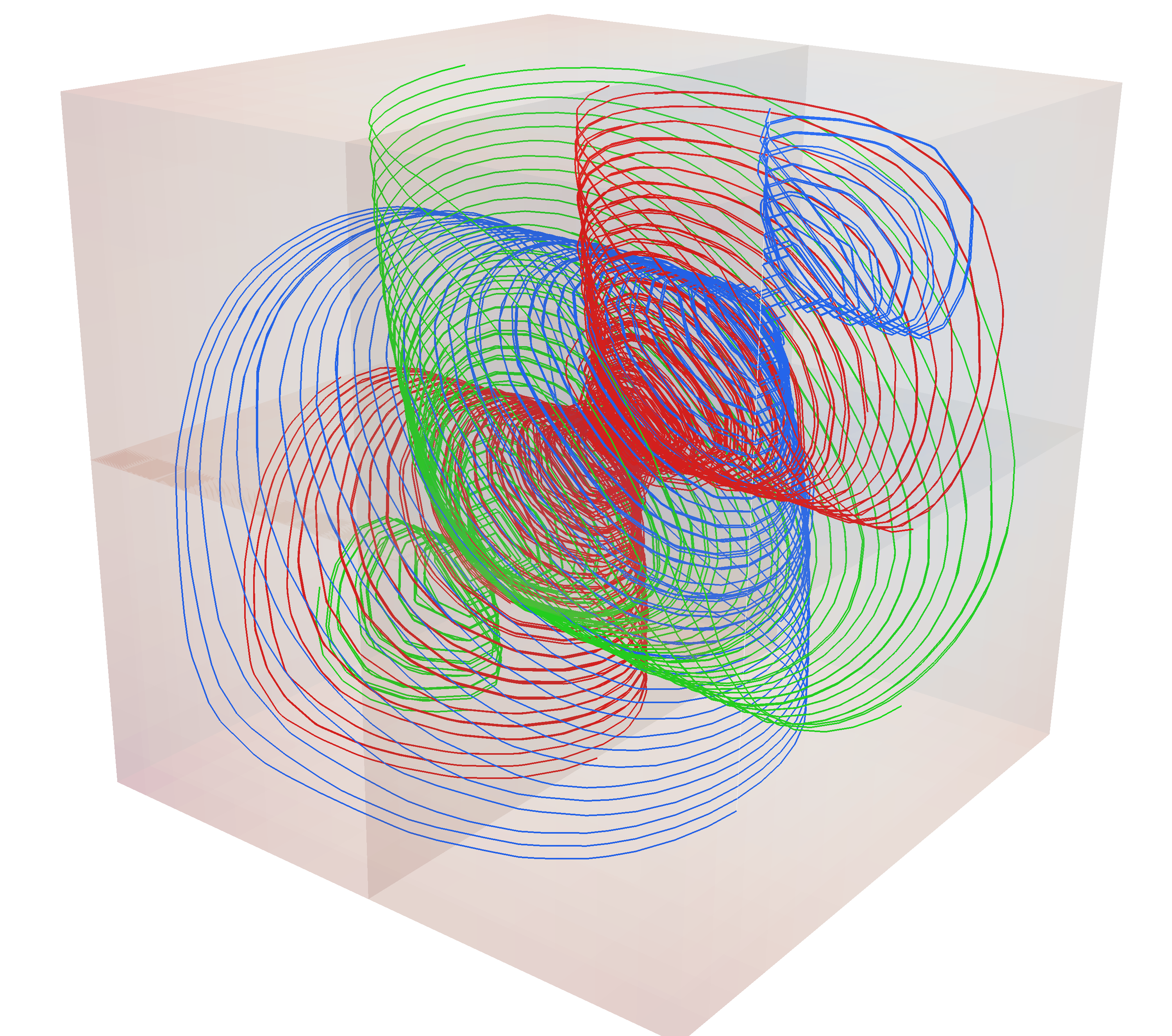}\hfill
 \includegraphics[width=0.47\textwidth]{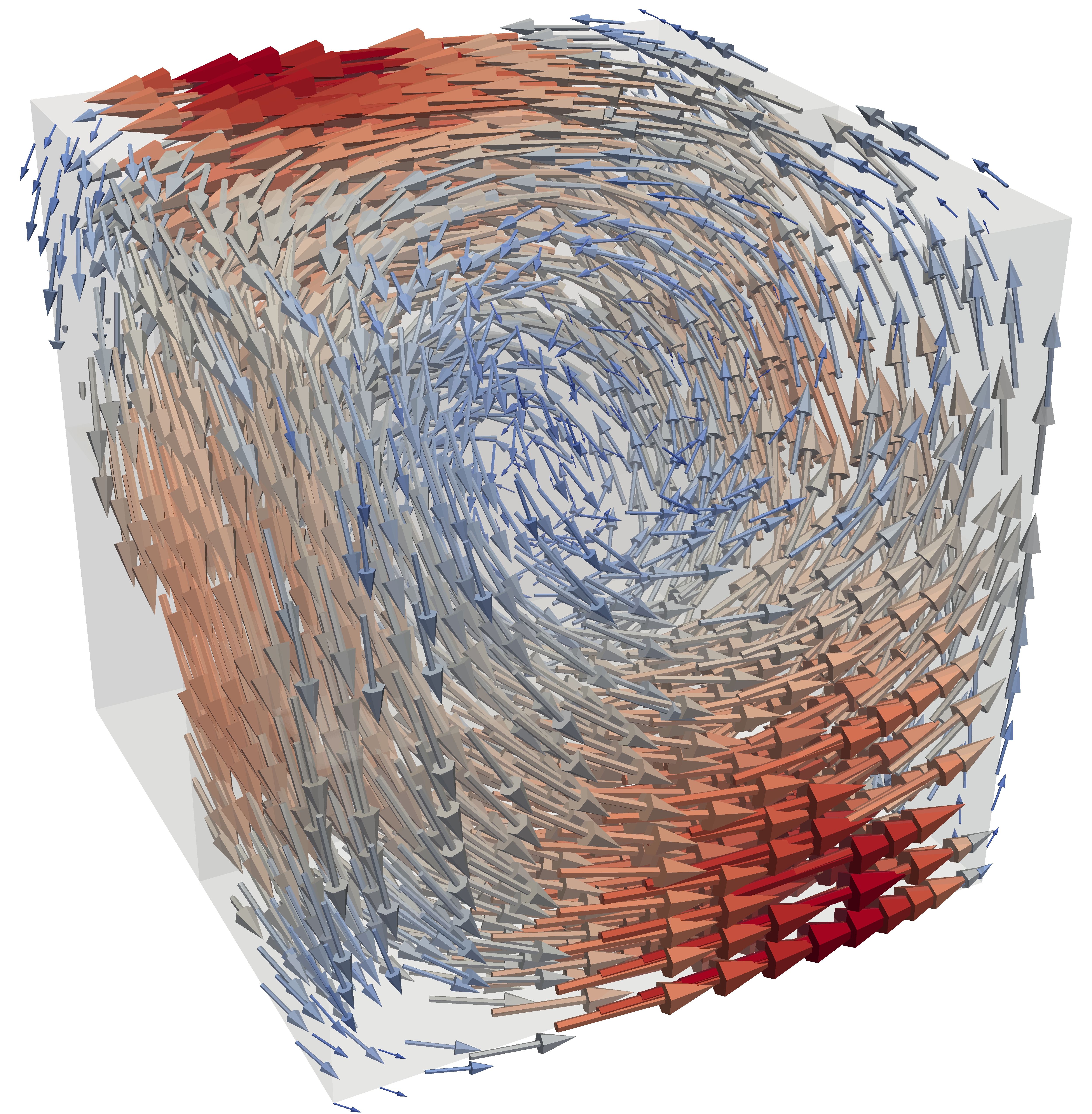}
\caption[Illustration of a 3D flow.]
{Illustrations of our 3D test case.
Top row: The 3D hybrid coarse mesh of the unit cube consisting of six
tetrahedra (blue), six prisms (orange), and four hexahedra (red). The
right-hand side picture shows an exploded view of the coarse mesh.
Bottom left: Streamlines of the flow $u^f$ from~\eqref{eq:cubeflow} with
$f(x)=\sin(\pi x)$; the colors serve the purpose of
distinguishing the streamlines. Bottom right: Some flow vectors $u^f(x)$; color
and size of the arrows indicate the magnitude $\|u^f(x)\|$ of the flow.
We observe that the flow rotates around a diagonal line through the cube
and that there is no outflow, as shown by Lemma~\ref{lem:ufdivfree}.}
\figlabel{fig:3Dflowincube}
\end{figure}

\begin{figure}
\center
\includegraphics[width=0.95\textwidth]{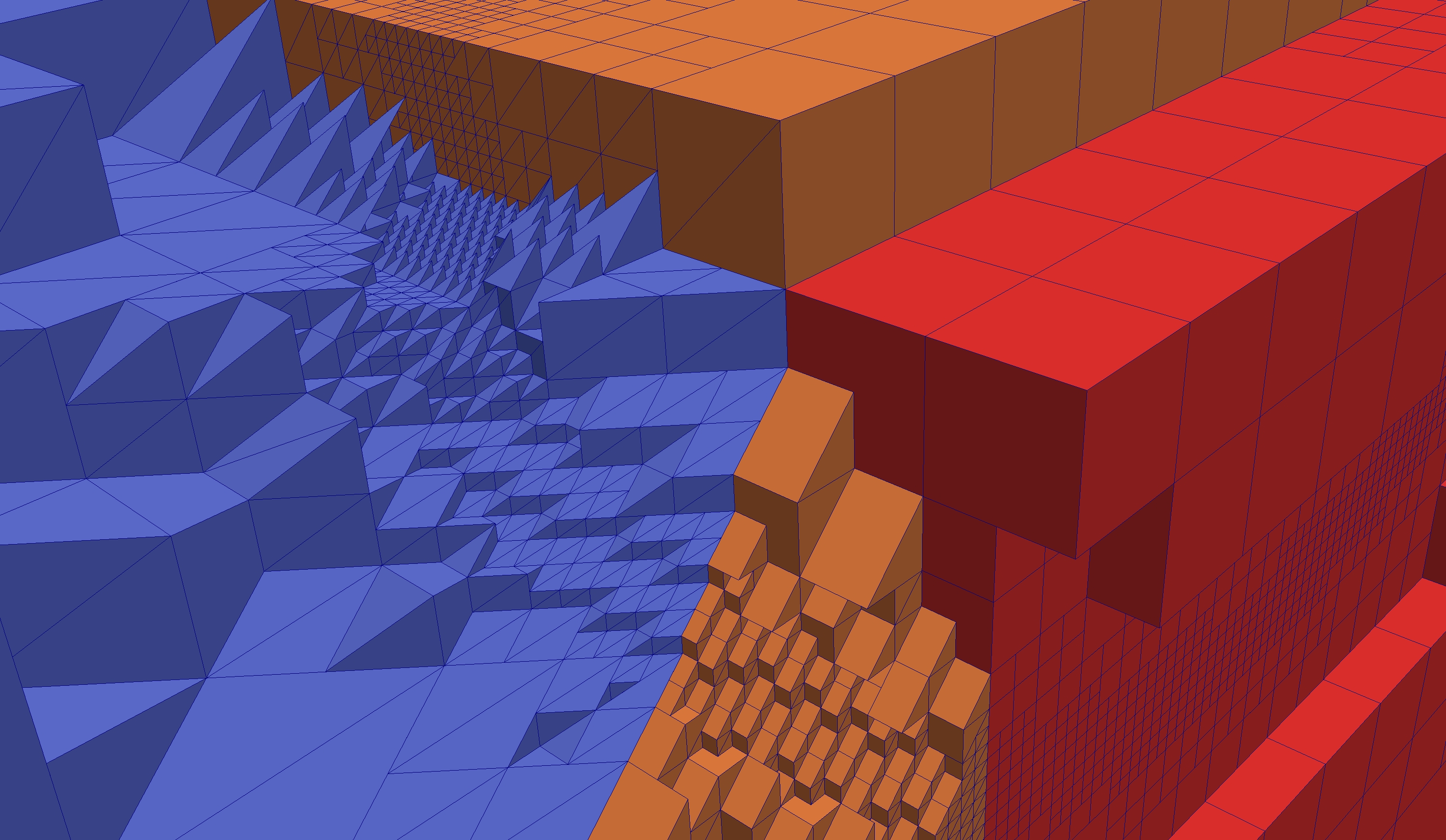}
\caption[A cutout view of a hybrid 3D mesh] 
{A view inside the hybrid 3D mesh
modelling the unit cube (cf.\ Figure~\ref{fig:3Dflowincube}).
The picture shows the region near the center of the unit cube where
tetrahedra (blue), prisms (orange), and hexahedra (red) are close to
each other. We see in detail how the prisms serve to transition from 
tetrahedra to hexahedra.
We use parameters $\ell=2$ and $r=5$ to generate this picture, thus the
finest elements are of refinement level 7.}
\figlabel{fig:hybrid-3d-cutout}
\end{figure}

\subsection{Convergence tests}
We perform uniform and adaptive runs for quadrilaterals, triangles, hexahedra,
tetrahedra, and hybrid meshes on $\Omega = [0,1]^d$ ($d$ being the dimension)
with small values of $\ell$ and $r$ in order to verify
the convergence of our method in terms of $\mathcal E_t^{\vol}$.
For each coarse mesh type we fix a CFL number $C$. Hence, each time we increase
the level by one, $\Delta t$ is divided by two, doubling the number of
time steps. We expect a first order convergence of the error at the final time.
In Table~\ref{tab:advecterrorvol1} we display the results for the 2D and 3D
test cases above and we indeed observe a first order convergence rate.
We additionally display the number of elements, which is constant over all
time steps in the uniform case and an average over all time steps for the
adaptive runs.
We see that with the adaptive runs we may achieve approximately the same error
as with uniform runs of the same maximal level while using significantly less
mesh elements.

\begin{table}\center
\scalebox{0.95}{
 \begin{tabular}{|c|c|r|r|}
  \hline
  \multicolumn{4}{|c|}{2D quadrilaterals, $C=0.5$}\\ \hline
  $\ell$ & $r$ & mesh size & $\mathcal E_1^{\vol}$ \\ \hline
  4&0&  256 & 95.7\%\\
  5&0& 1,024 & 60.2\%\\
  6&0& 4,096 & 33.7\%\\
  7&0&16,384 & 18.2\%\\\hline
  3&3& 3,613 & 34.0\%\\
  4&3& 7,475 & 21.9\%\\
\hline
 \end{tabular}
 \begin{tabular}{|c|c|r|r|}
  \hline
  \multicolumn{4}{|c|}{2D triangles, $C=0.1$}\\ \hline
  $\ell$ & $r$ &  mesh size & $\mathcal E_1^{\vol}$ \\ \hline
  4&0&  512 & 70.4\%\\
  5&0& 2,048 & 41.2\%\\
  6&0& 8,192 & 22.5\%\\
  7&0&32,768 & 11.4\%\\\hline
  3&3& 5,358 & 23.5\%\\
  4&3&11,542 & 15.0\%\\
\hline
 \end{tabular}
 \begin{tabular}{|c|c|r|r|}
  \hline
  \multicolumn{4}{|c|}{2D hybrid, $C=0.1$}\\ \hline
  $\ell$ & $r$ & mesh size & $\mathcal E_1^{\vol}$ \\ \hline
3 & 0 &  384  & 93.8\% \\
4 & 0 & 1,536 & 60.9\% \\
5 & 0 & 6,144 & 34.1\% \\
6 & 0 &24,576 & 18.7\% \\\hline
2 & 3 & 3,948 & 35.8\%\\
3 & 3 & 8,323 & 24.9\%\\\hline
\end{tabular}}\\[3ex]
\scalebox{0.95}{
 \begin{tabular}{|c|c|r|r|}
  \hline
  \multicolumn{4}{|c|}{3D hexahedra, $C=0.25$}\\ \hline
  $\ell$ & $r$ & mesh size & $\mathcal E_1^{\vol}$ \\ \hline
  4&0&    4,096 & 94.2\%\\
  5&0&   32,768 & 63.2\%\\
	6&0&  262,144 & 34.5\%\\
  7&0&2,097,152 & 18.1\%\\\hline
  3&3& 113,851 & 34,8\%\\
  4&3& 393,079 & 19.44\%\\
\hline
 \end{tabular}
 \begin{tabular}{|c|c|r|r|}
  \hline
  \multicolumn{4}{|c|}{3D tetrahedra, $C=0.1$}\\ \hline
  $\ell$ & $r$ & mesh size & $\mathcal E_1^{\vol}$ \\ \hline
  4&0&    24,576 & 70.3\%\\
	5&0&   196,608 & 41.0\%\\
  6&0& 1,572,864 & 21.8\%\\
  7&0&12,582,912 & 11.3\%\\\hline
  3&3&   652,615 & 22.0\%\\
  4&3& 2,305,350 & 12.1\%\\ 
\hline
 \end{tabular}
\begin{tabular}{|c|c|r|r|}
  \hline
  \multicolumn{4}{|c|}{3D hybrid, $C=0.1$}\\ \hline
  $\ell$ & $r$ & mesh size & $\mathcal E_1^{\vol}$ \\ \hline
3&0&    8,192 & 93.4\%\\
4&0&   65,536 & 63.7\%\\
5&0&  524,288 & 34.7\%\\
6&0&4,194,304 & 18.1\%\\\hline
2&3&  173,635 & 35.1\%\\
3&3&  559,860 & 19.8\%\\
\hline
 \end{tabular}
}
\caption[Verifying convergence of the advection solver]
{We verify convergence of the advection solver in terms
of volume loss of $\Omega^2$ at time $t=1$. We test six different coarse meshes: 2D
quadrilateral (top left), 2D triangles (top middle), 2D hybrid (top right), 3D
hexahedra (bottom left), 3D tetrahedra (bottom middle), and 3D hybrid (bottom
right); see also Figures~\ref{fig:3Dflowincube} and~\ref{fig:2dquadtrihybrid}. 
In each case the domain $\Omega$ is the unit cube
of the corresponding dimension. In 2D we use as flow $u$ a rotation around the 
center of the cube from equation~\eqref{eq:2dflow}, see also
the top left of Figure~\ref{fig:rotation2d}; in 3D we use 
the flow $u$ from equation~\eqref{eq:cubeflow} with $f(x)=\sin(\pi x)$;
see Figure~\ref{fig:3Dflowincube} for an illustration.
In the top part of each table, we list a uniform level $\ell$ test with
increasing values of $\ell$. We fix the CFL number $C$, which results in
doubling the number of time steps each time the level increases.
In the bottom part of each table, we show the results for adaptive refinement
to level $\ell+3$. For the refinement criterion~\eqref{eq:refcrit} we use a
band width of $b=4$.  In each case we verify a first order convergence rate.}
\figlabel{tab:advecterrorvol1}
\end{table}

\subsection{Large scale tests}

In this section We investigate the runtimes and scaling behavior of the solver
with the 3D tests with tetrahedra on the JUQUEEN supercomputer~\cite{Juqueen}.
For all of the following runs we use 16 MPI ranks (processes) per compute
node, thus 1 per compute core.

Firstly, we compare the amount of runtime spent in the AMR routines with the
actual time spent using the FV solver. To this end, we perform a
strong scaling tests with
16,384 processes, 32,768 processes, and 65,536 processes
and an average number of tetrahedra over all time steps of 8,344,140.
We display the results in Figure~\ref{fig:adv-strongscale-tet}, where we show
the runtimes of the different AMR routines and their respective proportion of
the overall runtime.
We exclude the runtime of the \texttt{Ripple-Balance} routine here, since
it is non-optimized and would distort the results.
We notice that the percentage of time spent in the AMR routines increases
with decreasing number of elements per process from around 35\% to 50\%,
which is expected for strong scaling with small numbers of elements per process.
Using a more involved numerical solver---for example higher order FV or DG---would
certainly decrease these percentages, since then the amount of solver time per
mesh element would increase while the time spent in AMR routines would remain
the same.

\begin{figure}
 \center
 \includegraphics[width=0.49\textwidth]{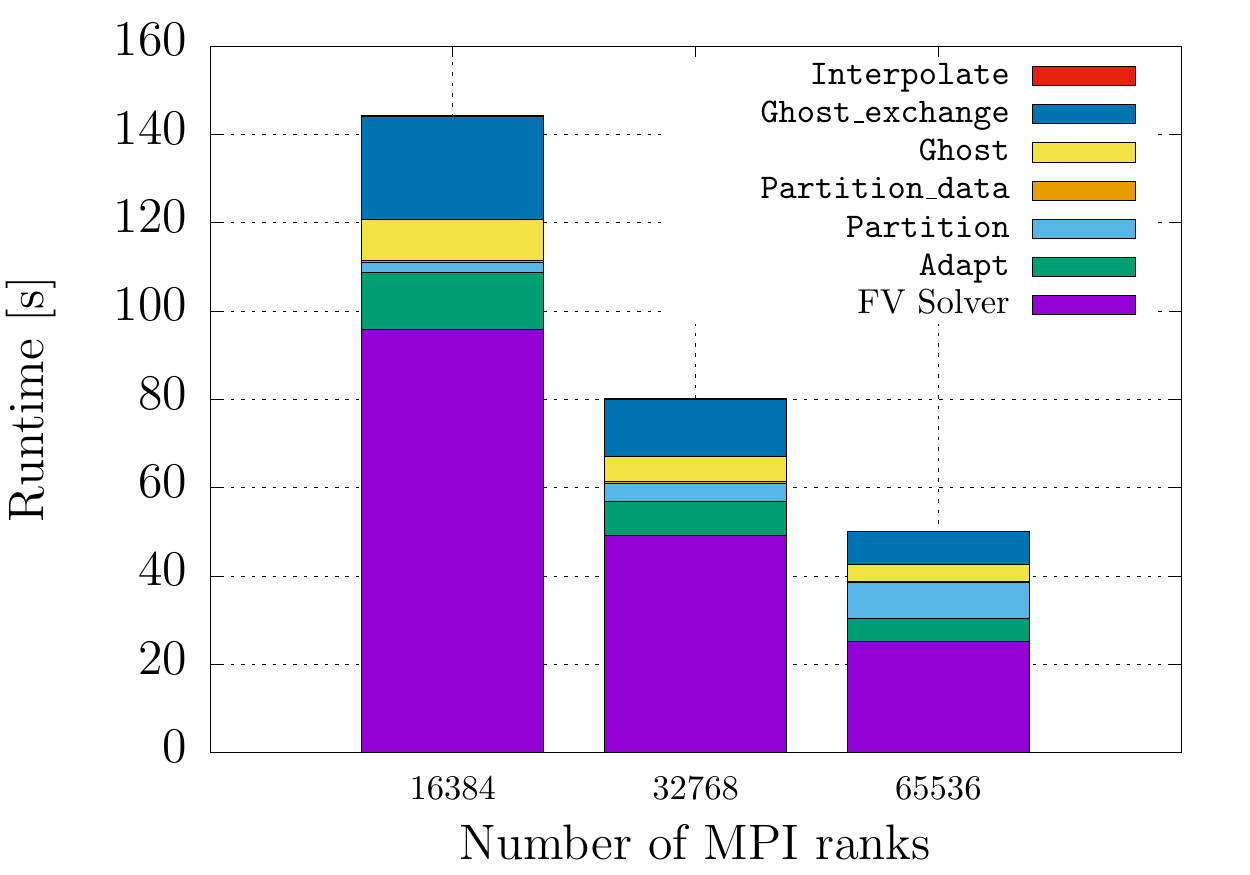}
 \includegraphics[width=0.49\textwidth]{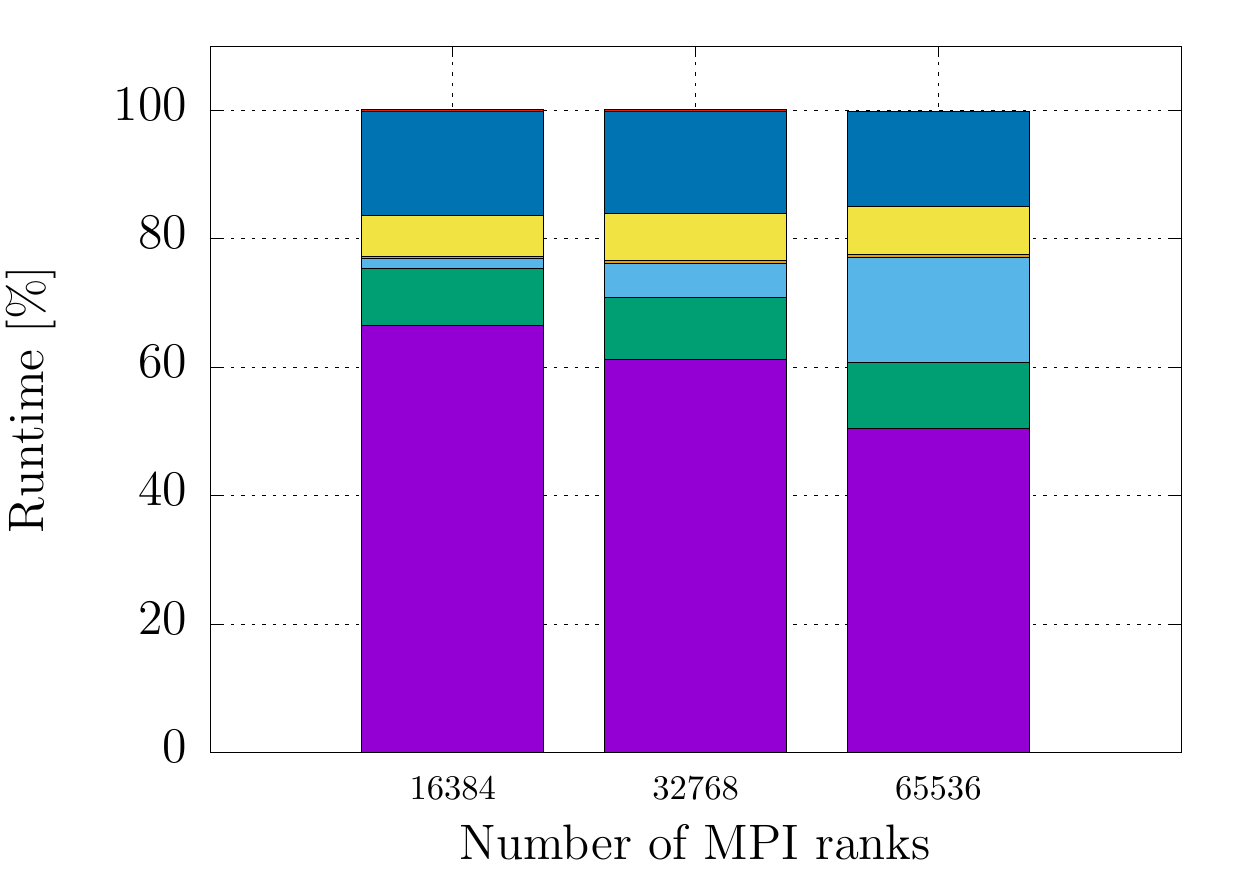}
  \caption[Strong scaling results for the advection solver]{
Strong scaling with tetrahedral elements. Using $\ell = 4$ and $r = 4$,
the maximum refinement level is $8$. The CFL number is $C=0.1$ and the
band-width parameter for refinement is $b=4$.  The average count of elements is
8,344,140. 
Left: Total runtime of the different AMR routines and the solver.
Right: Relative runtimes of the same methods. We observe that with decreasing
number of elements per process the relative runtime of AMR routines increases
from approximately 35\% to 50\%.
Since the CFL number is $0.1$ we only change the mesh in every 10-th time step.
Thus, \texttt{Adapt}, \texttt{Partition}, \texttt{Partition\_data},
\texttt{Ghost}, and \texttt{Interpolate} are only called in these time steps,
while \texttt{Ghost\_exchange} is called in every time step, explaining the
relatively large portion of runtime taken up by \texttt{Ghost\_exchange}.
}
\figlabel{fig:adv-strongscale-tet}
\end{figure}

Secondly, we perform strong scaling tests with 131,072 processes, 262,144
processes, and 458,752 processes, each case consisting of about $2.3\e9$ mesh
elements.
The equivalent uniform mesh would have more than $50\e9$ elements.
We use a short simulation time of $T=0.0005$ to keep the overall runtimes below
100 seconds.
We display the scaling results in Figure~\ref{fig:adv-strong-scale-large}---not
counting the runtime of the non-optimized \texttt{Ripple-Balance}---split up
into runtime of the AMR routines only and total runtime (AMR + solver).
Furthermore, we list the exact runtimes in
Table~\ref{tab:adv-strong-large-scale} and compute the parallel efficiency.
Compared to the base-line run with 131,072 processes we obtain an ideal strong
scaling efficiency on 458,752 processes if we do not take
\texttt{Ripple-Balance} into account.
Including \texttt{Ripple-Balance}, the efficiency is still above 90\%
for 262,144 processes, and 86.2\% (AMR), respectively 89.6\% (total), for
458,752 processes. We conclude that providing a state-of-the-art
\texttt{Balance} has the potential to increase the scaling behavior to the full
100\%.

\begin{figure}
 \includegraphics[width=0.68\textwidth]{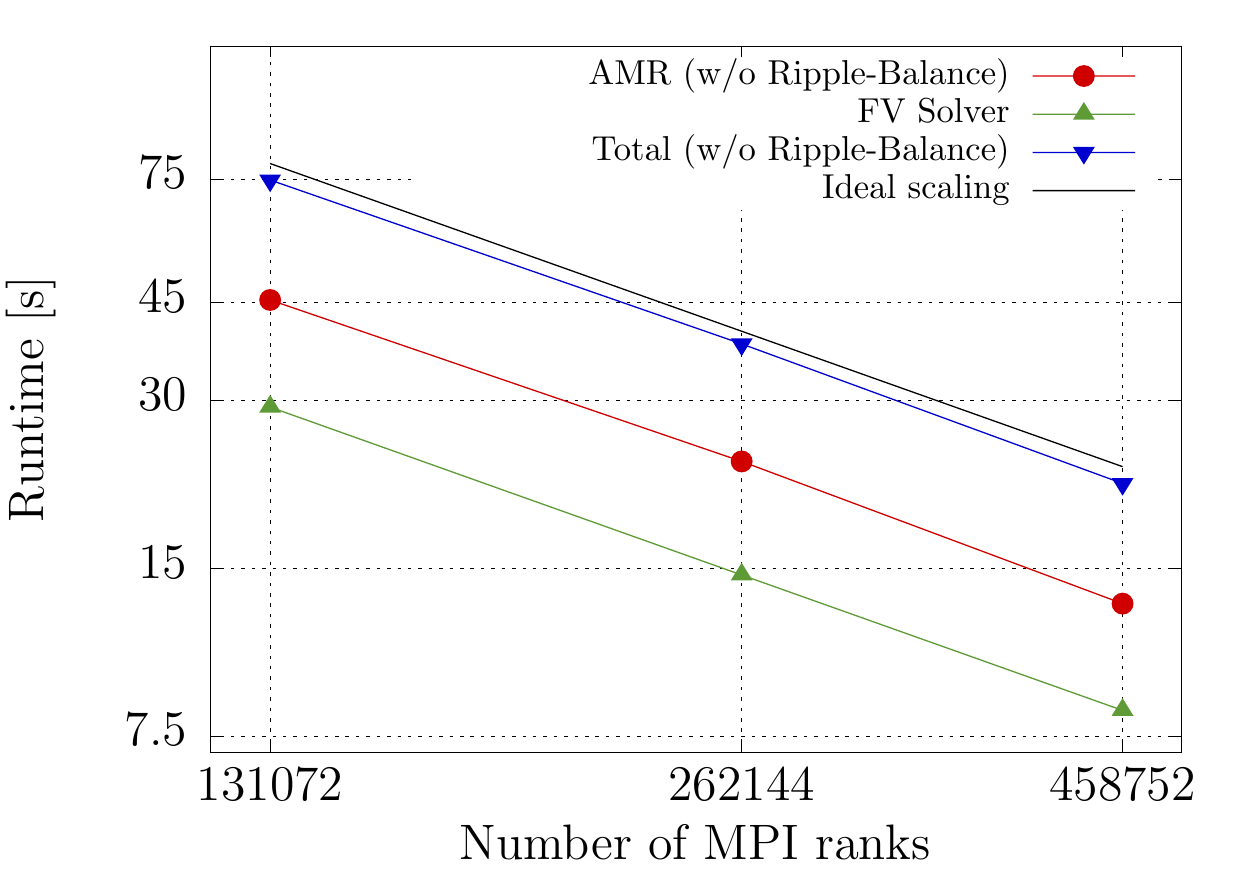}
\caption[Strong scaling results for the advection solver (full JUQUEEN)]
{ Strong scaling results with tetrahedra with 131,072 processes, 262,144
processes, and 458,752 processes. We use $\ell = 5$, $r = 6$, resulting in
2.3\e9 mesh elements. In order to decrease the overall runtime, we 
pick a simulation end time of $T=0.0005$ and adapt every $10$-th time step.
We show the runtime split between the FV solver (green triangles pointing up)
  and the AMR routines (red dots), as well as the total runtime (blue triangles
  pointing down). The AMR routines are
\texttt{Adapt}, \texttt{Partition}, \texttt{Ghost}, \texttt{Ghost\_exchange},
\texttt{Partition\_data}, and \texttt{Interpolate}.
The black line represents the ideal scaling behavior. We reach this ideal
  behavior on 458,752 MPI ranks with a parallel efficiency of 100\% compared
  to the baseline with 131,072 MPI ranks.  See also
  Table~\ref{tab:adv-strong-large-scale} where we explicitly list
the runtimes.}
\figlabel{fig:adv-strong-scale-large}
\end{figure}

\begin{table}
  \center
  \begin{tabular}{|r|R{17ex}|R{17ex}|r|r|}
    \hline
        &  \multicolumn{2}{c|}{Runtime (without Ripple-Balance)}&\multicolumn{2}{c|}{par.\ efficiency}\\
    $P$ & AMR & Total  & AMR & Total\\\hline
    131,072 & 45.5s & 74.7s & -- & --      \\
    262,144 & 23.3s & 37.9s & 97.6\% & 98.5\%  \\
    458,752 & 13.0s & 21.3s & 100.0\% & 100.0\% \\ \hline
  \end{tabular}\\[2ex]
  \begin{tabular}{|r|R{17ex}|R{17ex}|r|r|}
    \hline
        &  \multicolumn{2}{c|}{Runtime (including Ripple-Balance)}&\multicolumn{2}{c|}{par.\ efficiency}\\
    $P$ & AMR & Total  & AMR & Total\\\hline
    131,072 & 75.1s & 104.4s & -- & --      \\
    262,144 & 41.4s & 56.0s & 90.7\% & 93.2\%  \\
    458,752 & 24.9s & 33.3s & \phantom{1}86.2\% & \phantom{1}89.6\% \\ \hline
  \end{tabular}
\caption[Parallel efficiency of the advection solver]
{
In this table we present the runtimes and compute the parallel efficiency of
the AMR routines and the total runtime of the advection solver for the
tetrahedral strong scaling test on 131,072 up to 458,752 processes.
See also Figure~\ref{fig:adv-strong-scale-large}.
The top part of the table shows the runtimes without \texttt{Ripple-Balance},
we obtain a parallel efficiency of 100\% for the run on 458,752 processes
compared to the base line run with 131,072 processes.
In the bottom part, we include \texttt{Ripple-Balance}. The overall
parallel efficiency is then 86.2\% for the AMR routines and 89.6\% for 
the total runtime.
}
\figlabel{tab:adv-strong-large-scale}
\end{table}

\subsection{Comparison to uniform meshes}

It is sensible to ask whether we profit from using AMR at all. 
When we omit the adaptive refinement and use the initial level $\ell$ uniform mesh
in the solver, there is no overhead due to AMR algorithms and all of the
compute time is used to solve the actual numerical equation.
However, to reach the same accuracy, more mesh elements are needed, which in turn
results in a higher memory usage and possibly larger overall runtime.

In Table~\ref{tab:advect-adauni}  we compare adaptive and uniform runs on
tetrahedral meshes with 32,768 and 65,536 processes. In the uniform case we use
a refinement level of $\ell=8$ and in the adaptive tests we use an initial
level $\ell=4$ mesh which we adapt $r=4$ further levels, such that the finest
elements are of level $8$ as well.
First of all we notice that the adaptive meshes only use 8.3\% as many elements
as the uniform meshes ($8.3\e6$ and $100.6\e6$ elements) while resulting only
in a slightly larger computational error of $7.6\%$ volume loss compared to the
$5.8\%$ volume loss in the uniform case.
This large reduction in the number of elements points to a significant decrease
in memory usage.

Furthermore, we observe that the adaptive runs need less than half the
runtime as the uniform runs. These are total runtimes that include 
the relatively slow \texttt{Ripple-balance} routine and thus we can expect
that the overall gain would be even better with an optimized \texttt{Balance}
routine.
For applications with more highly localized physics, adaptivity may reduce the
number of elements in relation to uniform meshes by over three orders of
magnitude, which will lead to proportionally higher savings.

\begin{table}
  \center
  \begin{tabular}{|r|r|r|r|r|r|r|}
    \hline
    $P$ & $\ell$ & $r$ &   mesh size & $\mathcal E^{\vol}_1$ & time steps &Runtime \\\hline
     32,768 & 8 & 0   & 100,663,296 & 5.8\% & 5,196 & 469.9s \\
     32,768 & 4 & 4   &   8,336,500 & 7.6\% & 3,445 & 159.3s \\ \hline
     65,536 & 8 & 0   & 100,663,296 & 5.8\% & 5,196 & 240.3s \\
     65,536 & 4 & 4   &   8.339,130 & 7.6\% & 3,445 & 102.3s \\ \hline
  \end{tabular}
  \caption[Comparing uniform and adaptive meshes.]
{We compare adaptive and uniform runs of the same problem,
once with 32,768 and once with 65,536 processes. The finest refinement level
in all cases is $8$. The adaptive runs only use 8.3\% as many mesh elements
as the uniform runs and need less that half the runtime to obtain an only
slightly larger computational error. The runtimes in the last column are the
total runtimes of the solver and in particular include the non-optimized
\texttt{Ripple-balance} routine in the adaptive cases.}
\figlabel{tab:advect-adauni}
\end{table}

\subsection{A test with a larger coarse mesh}
We close this section with an application of the solver on a non-trivial
domain with a medium sized coarse mesh. 
We use the example of two-dimensional potential flow around a
disk~\cite{Batchelor00}. This flow is an analytical solution to the flow of an
incompressible fluid without viscosity around a disc with radius $R$ and
midpoint at the origin, with the flow being constant $1$ in $x$-direction and $0$
in $y$-direction far outside of the disc. In polar coordinates the flow field
$u$ is given by
\begin{subequations}
\label{eq:flowaround}
 \begin{align}
  u(r,\phi)_r &= \left(1-\frac{R^2}{r^2}\right)\cos(\phi),\\
  u(r,\phi)_\phi &= \left(-1-\frac{R^2}{r^2}\right)\sin(\phi).
 \end{align}
\end{subequations}
As radius we chose $R=0.15$, and we illustrate the flow in Figure~\ref{fig:flowaround}.

Since the flow is symmetric around the $x$-axis, we restrict our attention
to the region with $y$-coordinates greater or equal zero.
As domain $\Omega$ we choose the rectangle $[-0.5,1]\times[0,0.75]$
with the disk cut out.

We model this domain with a hybrid coarse mesh of 238 triangles and 351
quadri\-laterals as in Figure~\ref{fig:flowaroundcmesh}. The purpose of the
quadrilaterals is to properly resolve the flow close to the curved boundary. In
order to do so, we use a boundary layer of thin quadrilaterals that are
stretched in the direction of the flow, which we display in
Figure~\ref{fig:flowaroundcmesh} on the right.

\begin{figure}
\center
\includegraphics[width=0.8\textwidth]{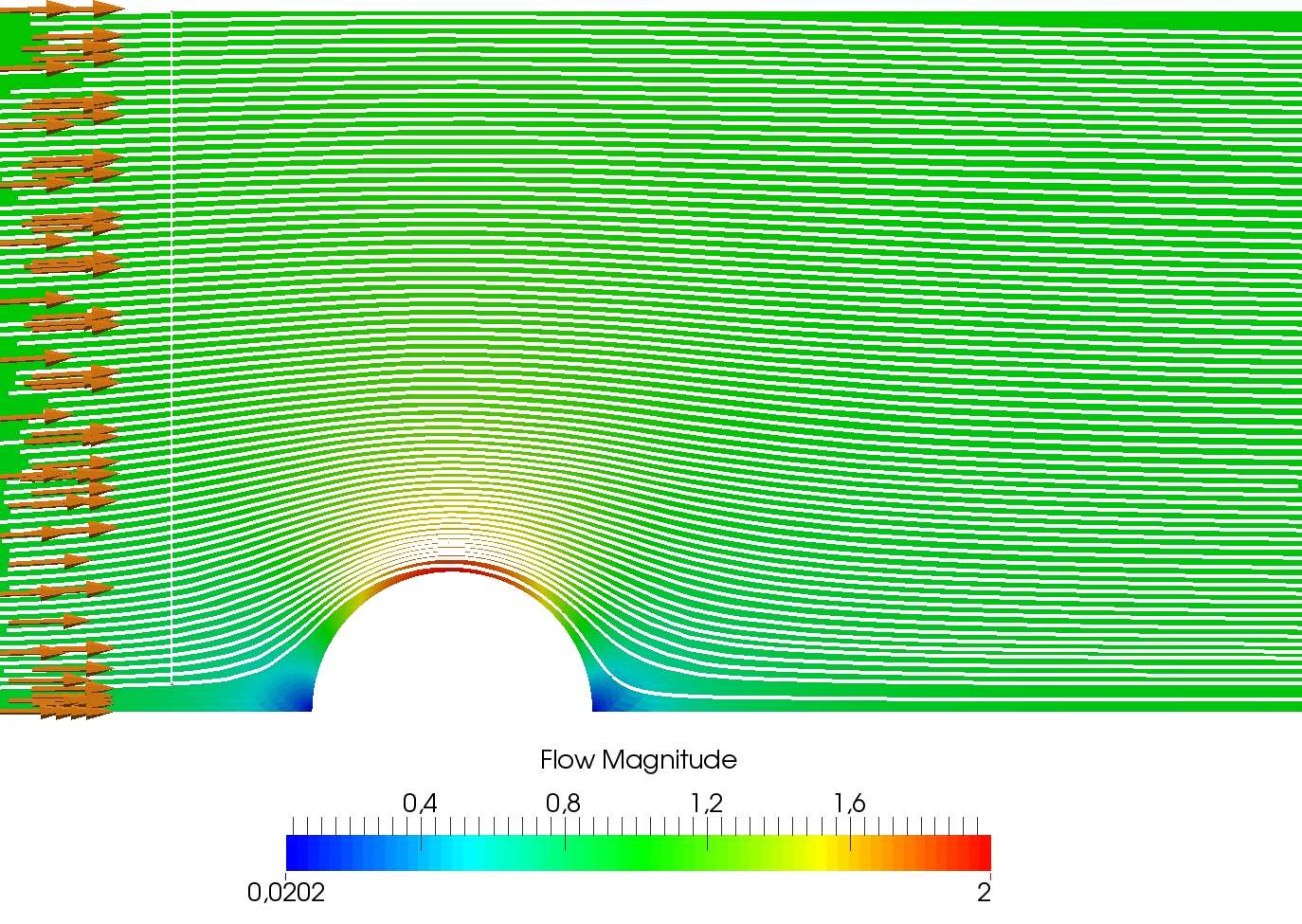}
\caption[2D flow around a disk.]
{Streamlines (white) of the flow \eqref{eq:flowaround} around a disk with
radius $R=0.15$. The arrows on the left-hand side indicate the inflow vector
(orange) while the background color indicates the magnitude of the flow velocity.%
}
\figlabel{fig:flowaround}
\end{figure}

\begin{figure}
\center
\includegraphics[width=0.55\textwidth]{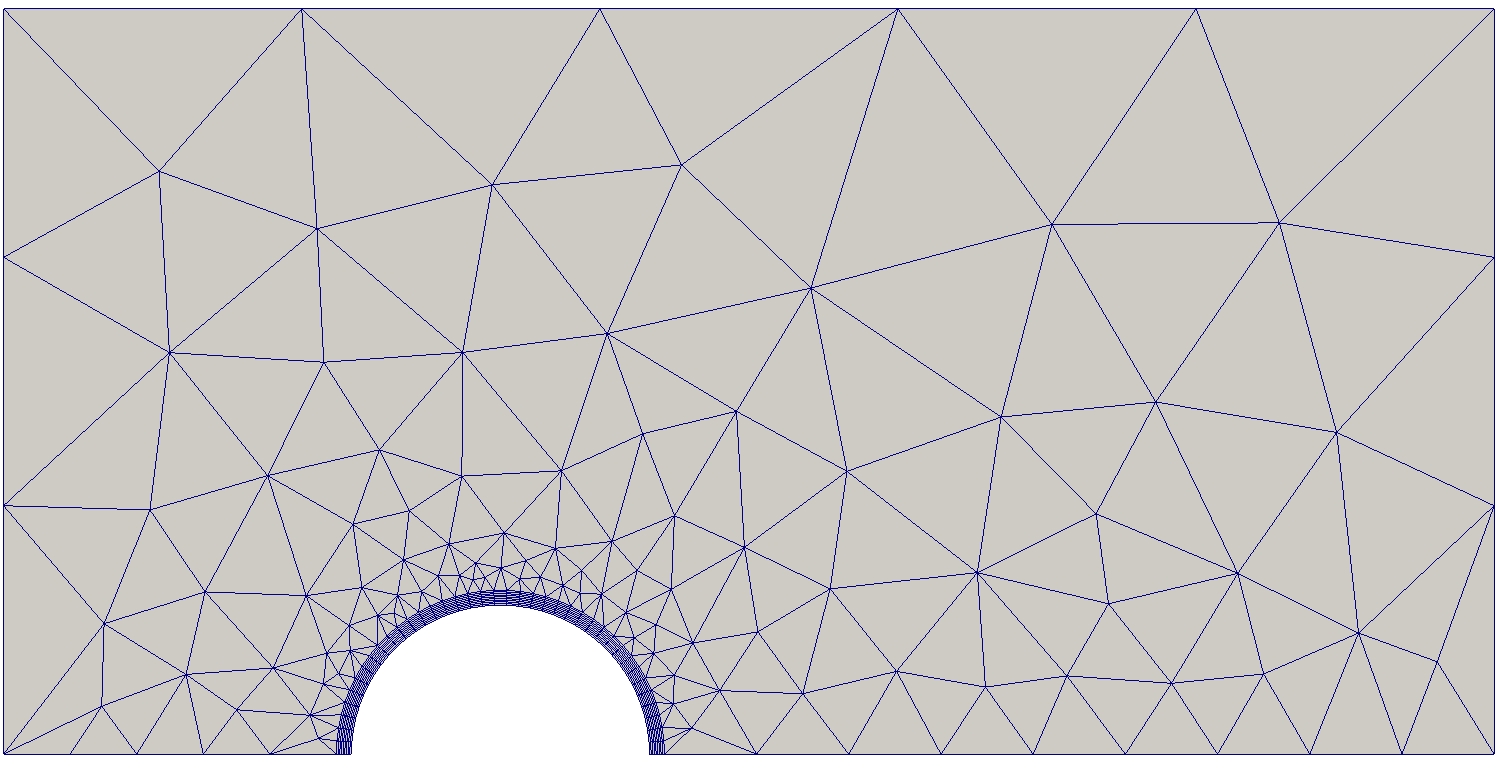}
\raisebox{0.247ex}{
\includegraphics[width=0.348\textwidth]{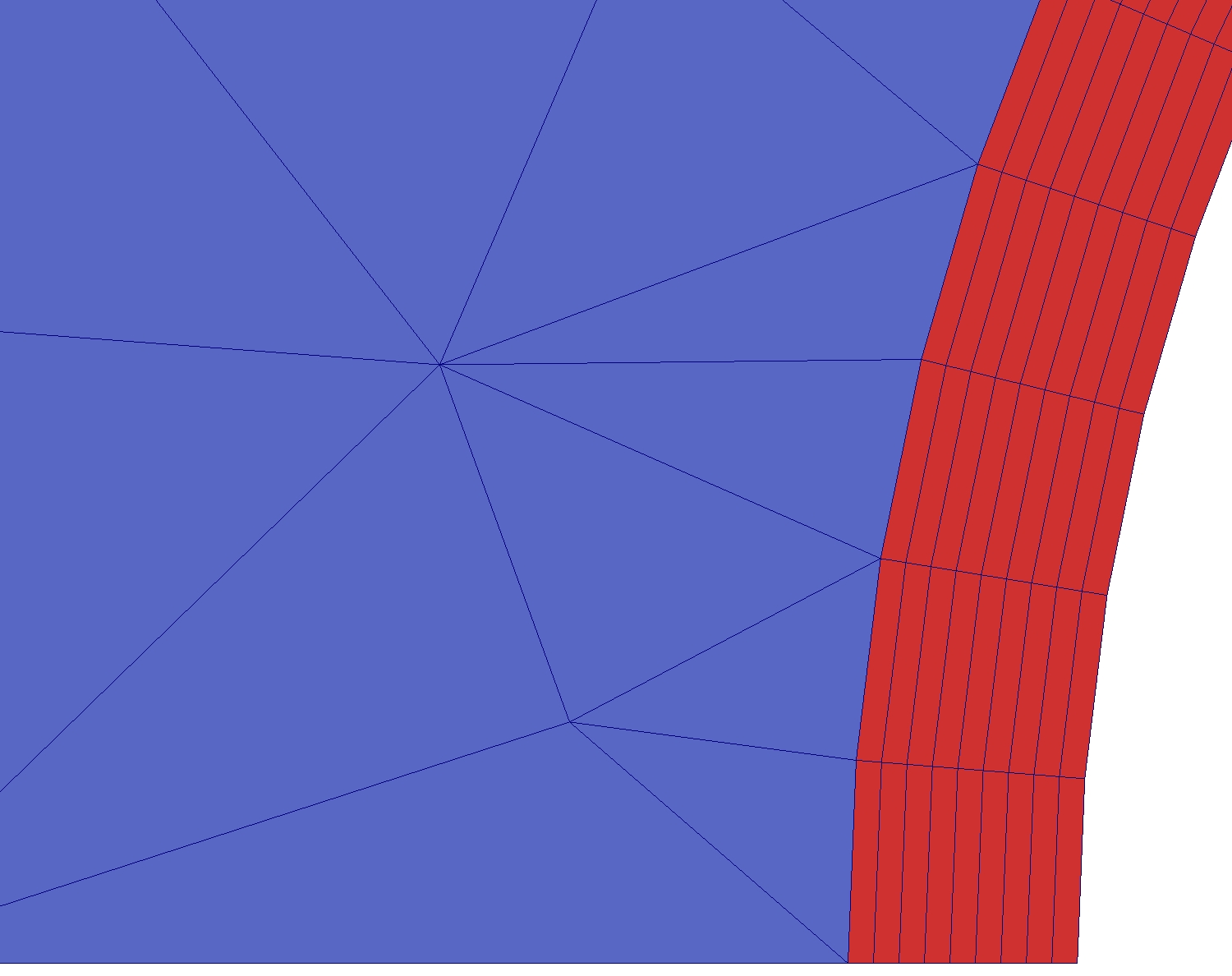}}
\caption[The coarse mesh that we use to model the domain with a disk cut out.]
{Left: We model the rectangle $[-0.5,1]\times[0,0.75]$ with the disk of
radius $0.1$ cut out as a hybrid triangle/quadrilateral mesh with 238 triangles
  and 351 quadrilaterals. Right (zoomed in): We use the quadrilaterals (red) to resolve the flow
close to the curved boundary and triangles (blue) to mesh the remaining 
  domain. This mesh was created with \texttt{Gmsh}~\cite{GeuzaineRemacle09}.
Throughout the mesh the sizes of the coarse elements differ by several orders
of magnitude, which motivates taking the volume of an element into account
  when refining according to~\eqref{eq:refcrit}.}
\figlabel{fig:flowaroundcmesh}
\end{figure}

\begin{remark}
We modify our refinement criterion. It is still sensible to refine in a band
around the zero level-set as in~\eqref{eq:refcrit}. However, the sizes of the
coarse mesh elements differ largely. In particular the quadrilateral elements
at the circle boundary are very small compared to the triangular elements
filling the rest of the domain. In the final refined mesh, all elements close
to the zero level-set should have approximately the same size. In order to
achieve this, we refine an element if it fulfills criterion~\eqref{eq:refcrit}
only if its volume is above a certain threshold, which is determined by the
volume of the smallest elements in the mesh; see
Figure~\ref{fig:flowaround-sim-ref} for an illustration.
We thus replace the coarsest and finest level restriction decribed in
Section~\ref{sec:refcrit} by a largest and smallest volume restriction.
\end{remark}

As initial level-set function $\phi_0$ we choose the signed distance function
to a circle with radius $0.1$ and midpoint $\begin{pmatrix} 0.2, &
0.11\end{pmatrix}^t$. In particular, this means that the zero level-set is 
advected close to the curved boundary.

In Figure~\ref{fig:flowaround-sim} we show six different time steps of the
computation with initial uniform level $\ell=1$ and $r=4$ adaptive refinement
levels. It is clearly visible how close the zero level-set is to the circular
hole.

\begin{figure}
\center
\includegraphics[width=0.95\textwidth]{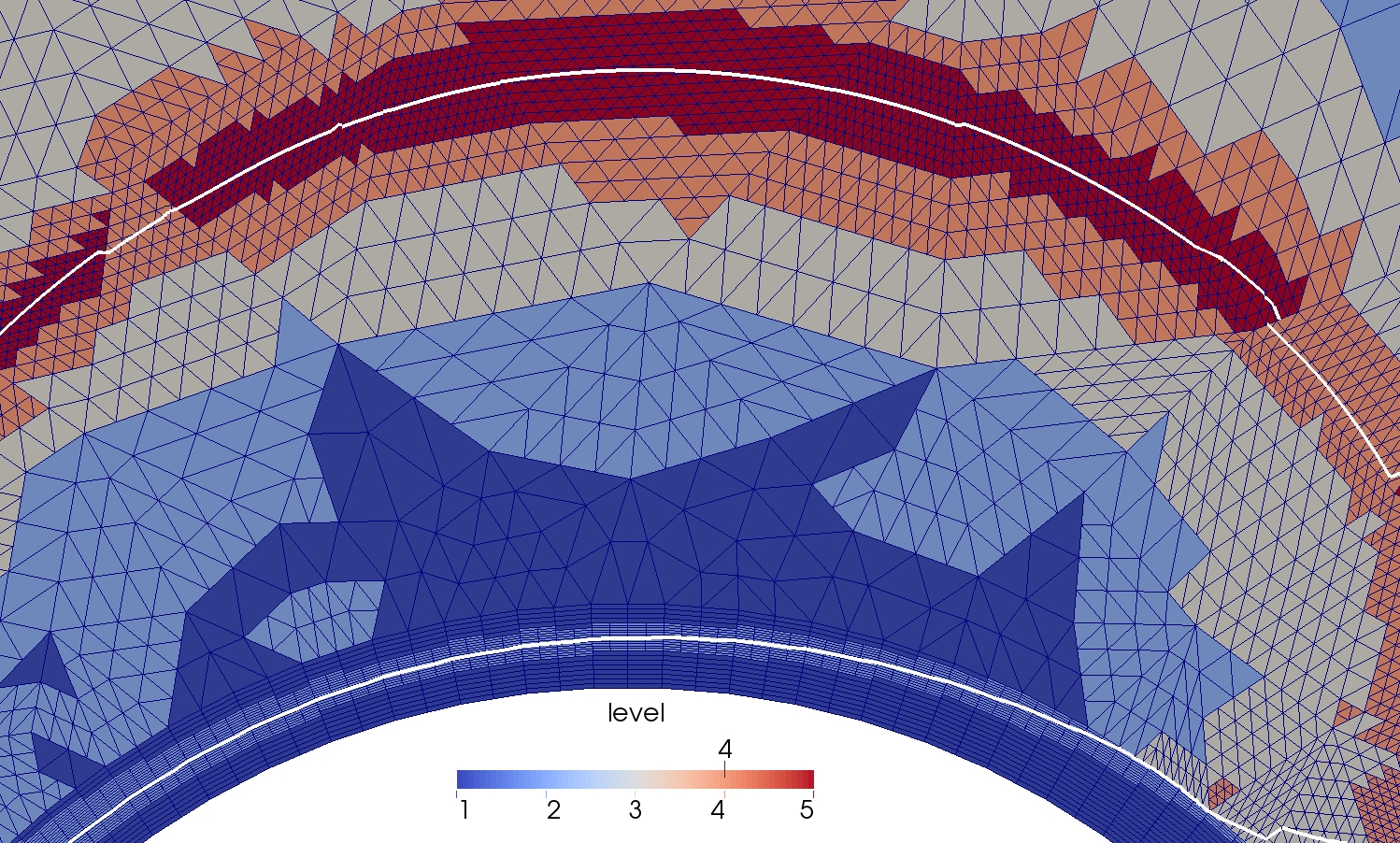}
  \caption[Refinement criterion with different coarse mesh element sizes.]
  {Zoomed in: The different refinement levels close to the zero level-set (white).
We refine an element if it is close to the zero level-set and its volume is
larger than a given lower bound. Thus, the finest elements close to the zero level-set all have
comparable volumes. The difference in refinement levels is observed 
when we compare the level 5 triangles (red) in the top part of the image
with the level 2 quadrilaterals (light blue) in the bottom part.}
  \figlabel{fig:flowaround-sim-ref}
\end{figure}

\begin{figure}
\center
\includegraphics[width=0.49\textwidth]{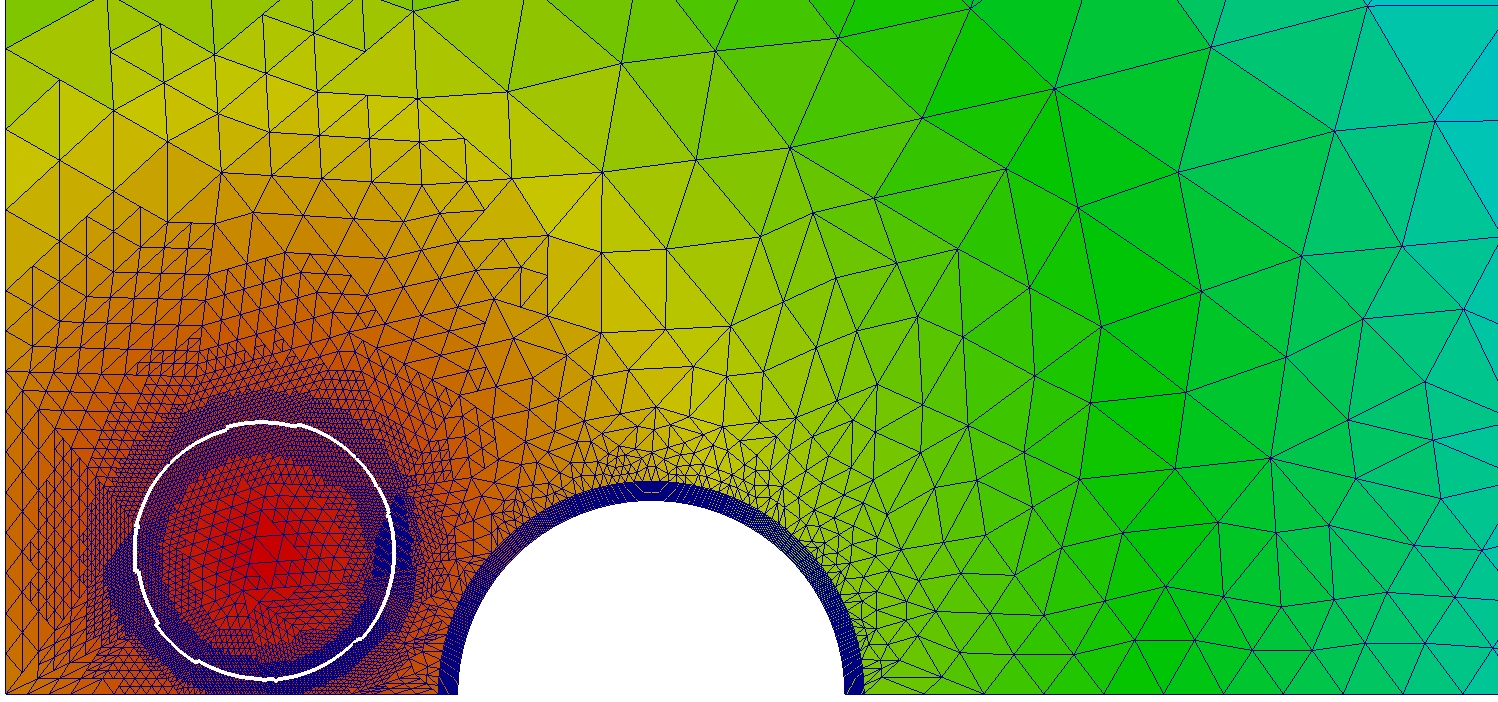}
  \includegraphics[width=0.49\textwidth]{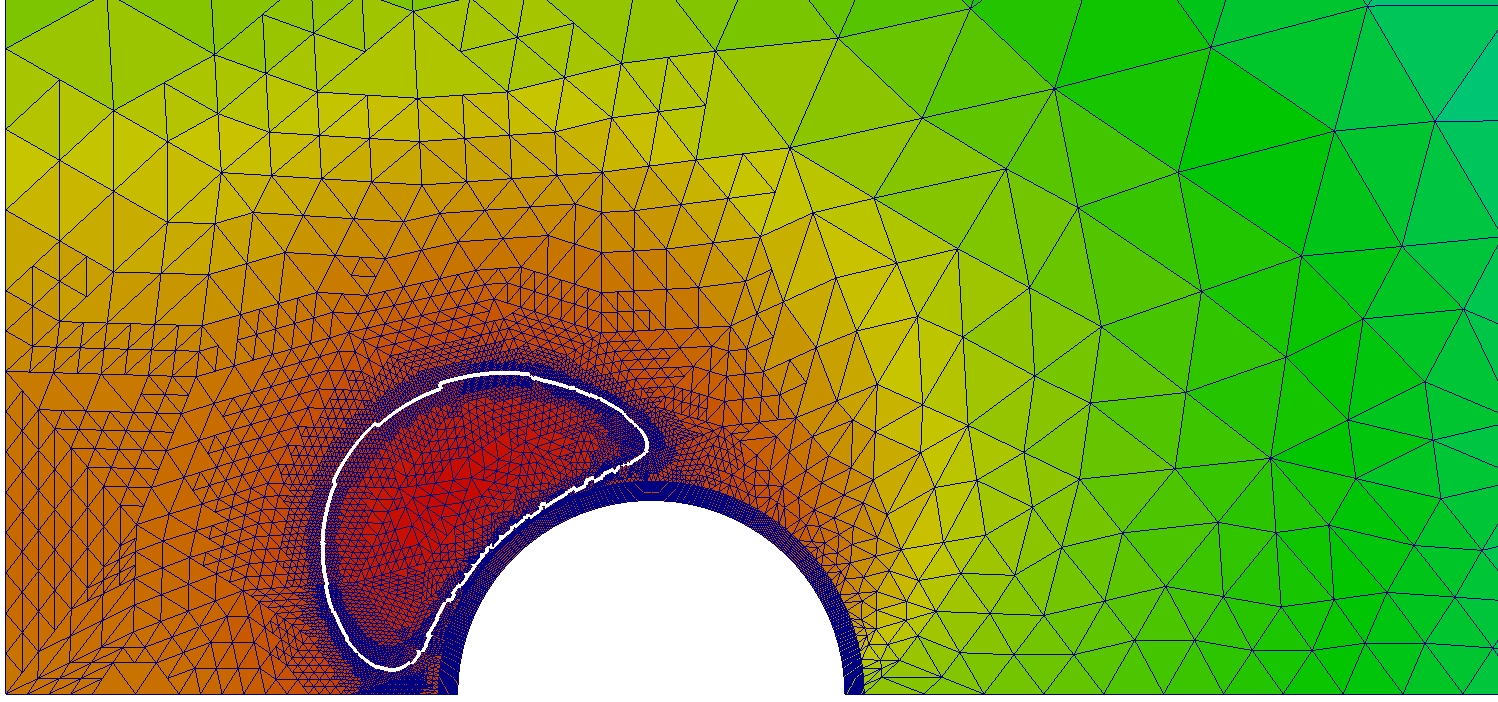}\\[1ex]
\includegraphics[width=0.49\textwidth]{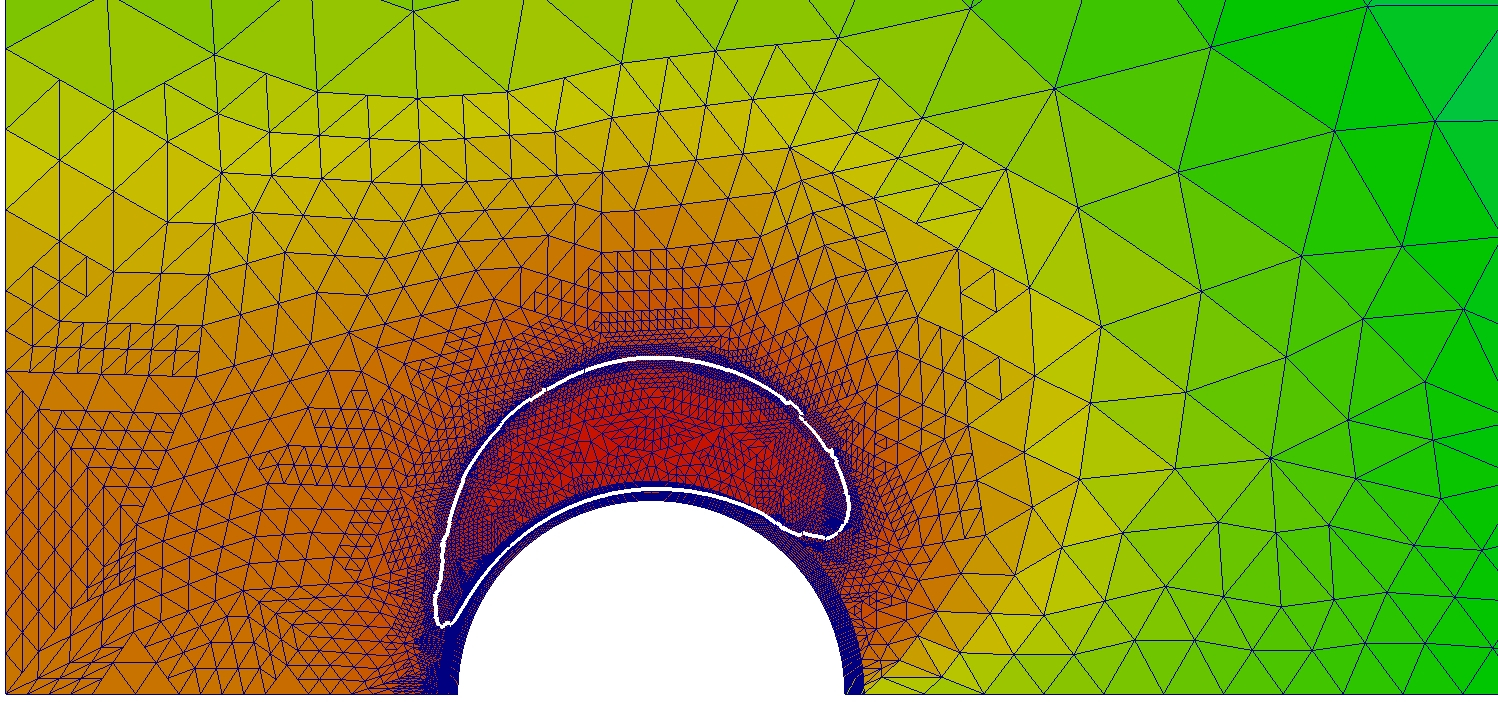}
  \includegraphics[width=0.49\textwidth]{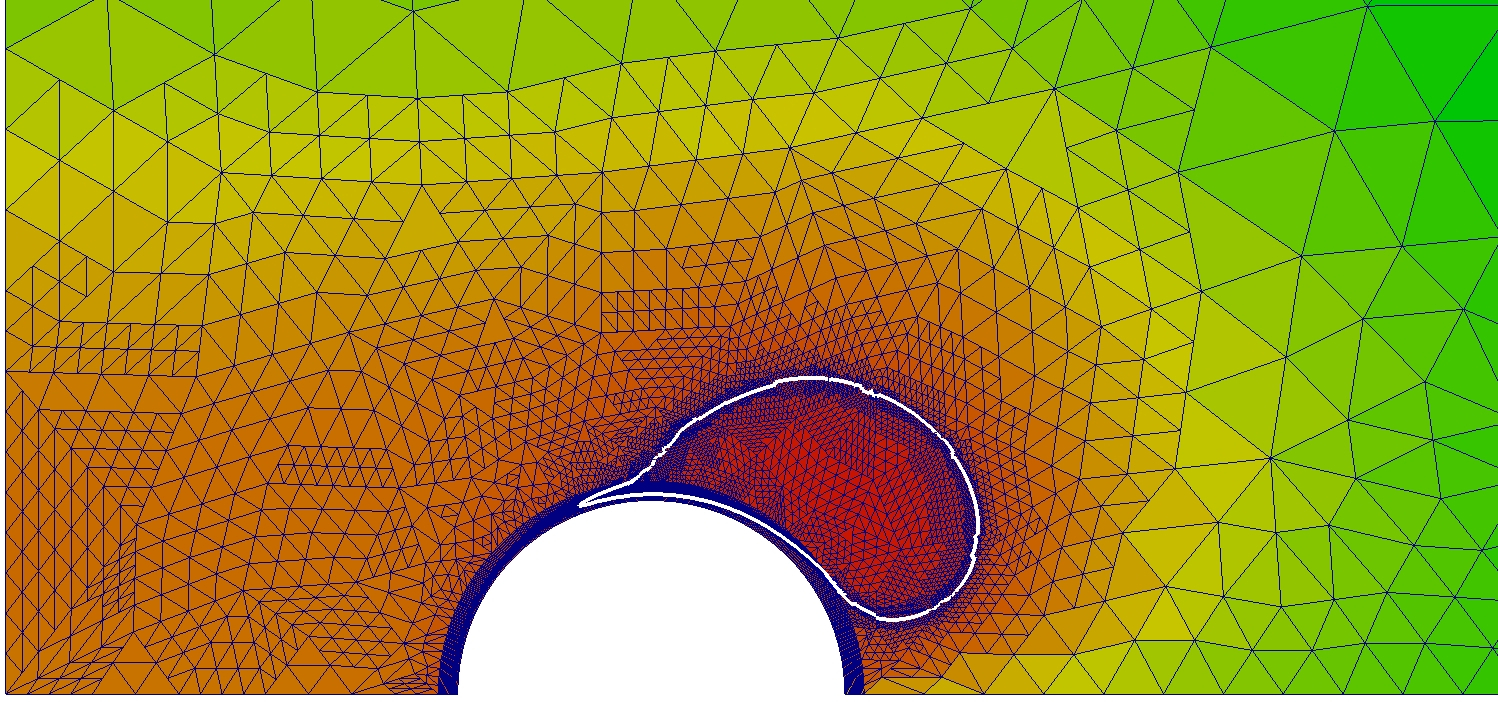}\\[1ex]
\includegraphics[width=0.49\textwidth]{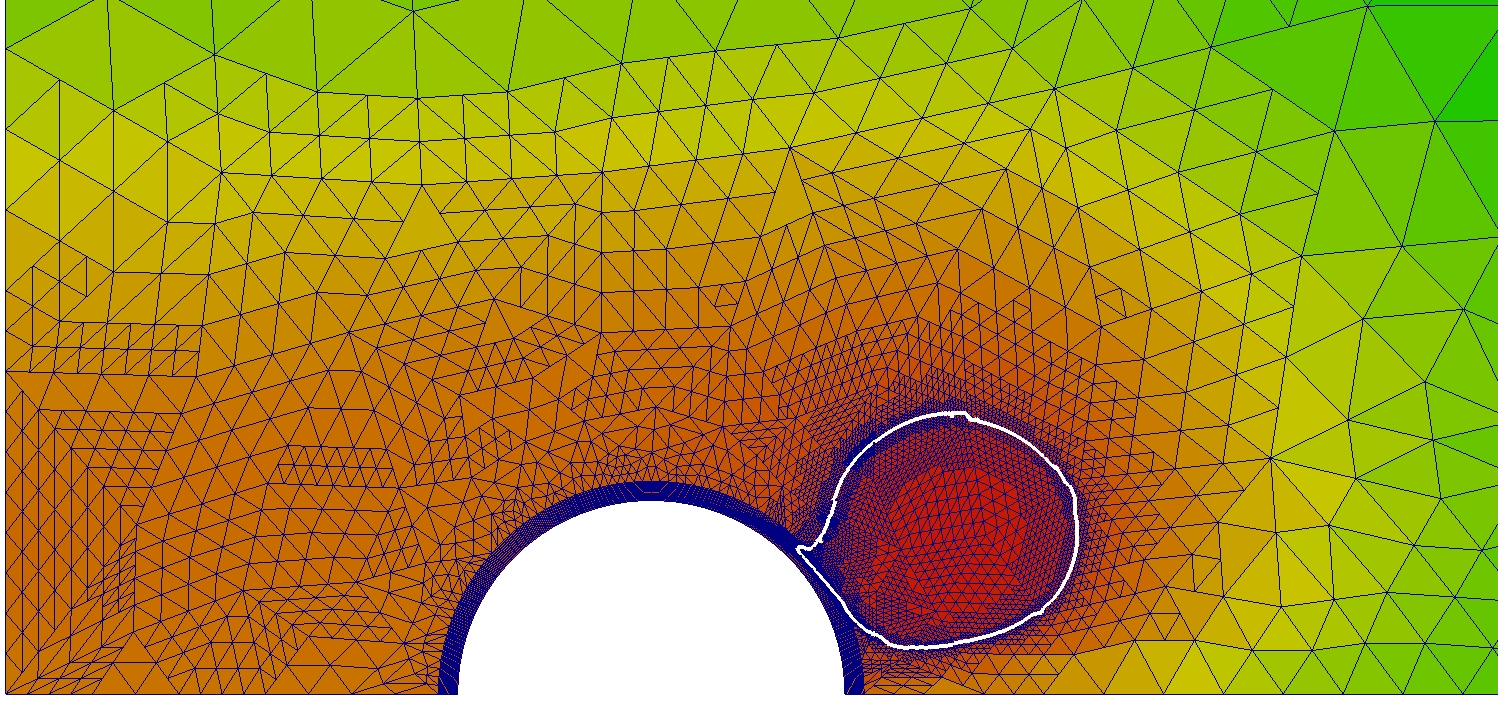}
  \includegraphics[width=0.49\textwidth]{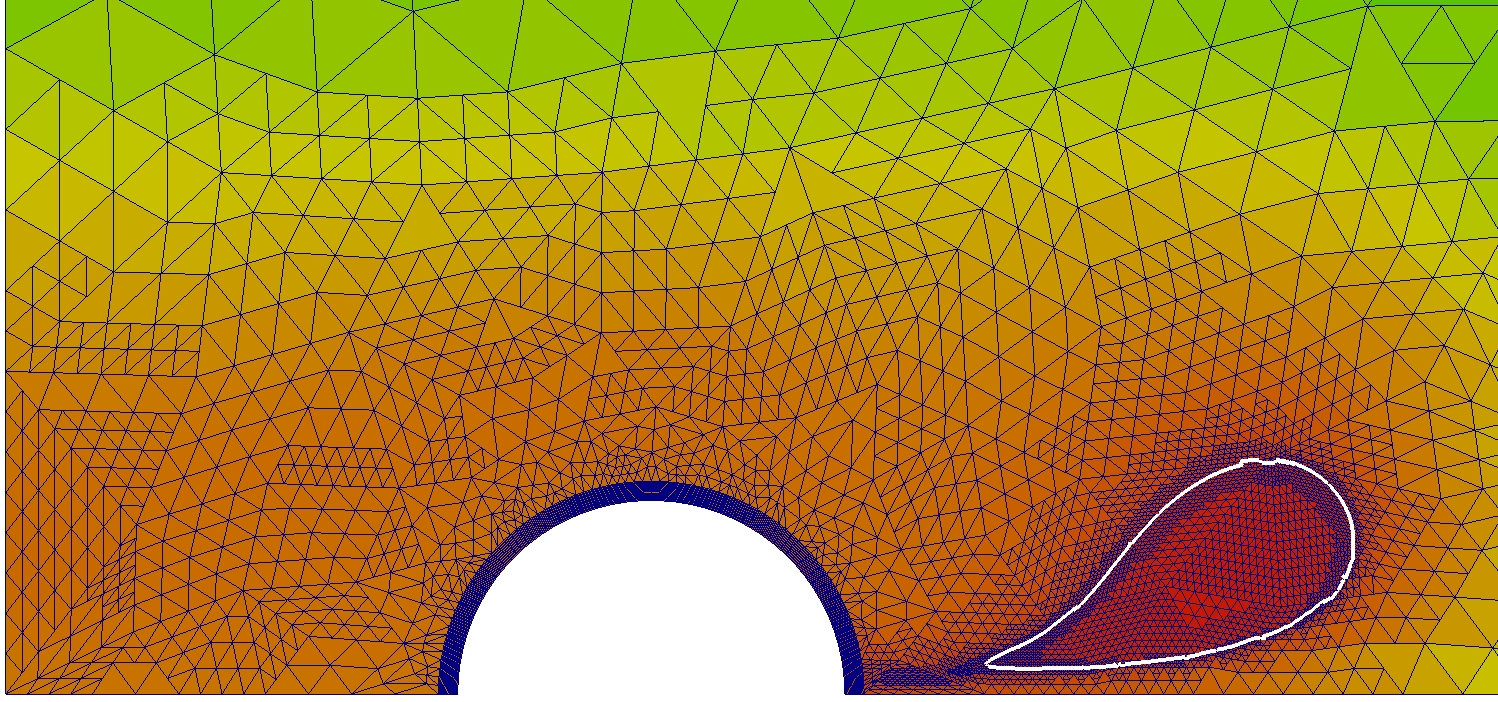}\\
\includegraphics[width=0.49\textwidth]{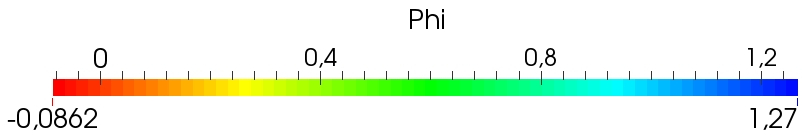}\\[3ex]
\caption[Solutions to the advection around a disk and the adapted mesh.]
{Six different time steps of the numerical solution
to the advection of a circular interface with the flow~\eqref{eq:flowaround}
  and maximum refinement level 5.
  The interface moves very close to the circular domain boundary. This movement
  is well resolved since we use flat quadrilaterals close to this boundary.
  In the more orange region to the left of the zero level-set we observe the 
  effect that $\phi$ looses its signed distance property, which we addressed in
  Remark~\ref{rem:reinit}. This results in elements keeping a finer level than necessary.}
\figlabel{fig:flowaround-sim}
\end{figure}
\chapter{Conclusion}
\label{ch:concl}
In this thesis we are concerned with the development of a new space-filling
curve (SFC) for parallel simplicial adaptive mesh refinement, as well as
scalable data structures and algorithms for parallel tree-based adaptive mesh
refinement (AMR) with general element types (such as for example triangles,
tetrahedra, quadrilaterals, hexahedra, and prisms) and support for hybrid
meshes.

Existing libraries for tree-based AMR focus on a single type of element and
SFC, such as for example \pforest~\cite{BursteddeWilcoxGhattas11} for
quadrilaterals/hexahedra with the Morton SFC. A major contribution of this
thesis is the development of an element-type independent approach by
separating the high-level AMR algorithms from the low-level element-type
implementation. By modifying and restructuring the existing meshing algorithms
we extend the AMR functionality to arbitrary element types.
In particular, to the best of our knowledge, we present the first discussion
and applications of tree-based AMR on tetrahedral and hybrid meshes.

While we applied our concept with little effort to mesh refinement and coarsening,
as well as partitioning, significant work was necessary to establish the
construction of a layer of ghost elements. 
Additionally, we devised a new communication reducing algorithm for coarse mesh
partitioning.

Besides a thorough theoretical investigation, we implemented the \tetcode AMR
library to demonstrate practicability and verified strong and weak scaling to
hundreds of thousands of parallel processes on current supercomputing systems.
To confirm that our algorithms provide all features required for a basic
numerical application, we implemented a finite volume solver for the advection
equation. We find that our implementation supports face-based numerical
applications in a modular and non-intrusive way.

\vspace{2ex}

A starting motivation for this thesis was to find an SFC for triangles and
tetrahedra that has similar properties to the Morton SFC for
quadrilaterals and hexahedra, such as fast computation via bitwise interleaving
and constant-time, level-independent low-level algorithms (constructing children,
face-neighbors, parents, etc.).
Though it was not certain that such an SFC exists, we were able to find a suitable
construction leading to the tetrahedral Morton (TM-) SFC.
In order to transfer the bitwise interleaving concept to simplices, we 
introduced the type of a simplex as additional information, and used the types
of an element's ancestors to interleave with the anchor node coordinates. 
Despite its similarities to the Morton SFC, segments of the TM-SFC can produce
more connected components when partioning the mesh in parallel, possibly
increasing the amount of parallel communication. 
We investigated this further and proved upper bounds for the possible counts and demonstrated that
in practical applications less than 7\% of the segments have more than three
connected components. Throughout this thesis we observed that the performance of
the TM-SFC is comparable to that of the Morton SFC.
Our discussions and examples show that the TM-SFC is an excellent choice for 
an SFC for triangles and tetrahedra and is a competitive addition to the 
multitude of existing SFCs.

In order to model complex domain geometries with arbitrarily large
numbers of trees, we investigated parallel partitioning of the coarse mesh.
We devised a new approach that reduces parallel communication while maintaining
the load-balance of the fine mesh. Because we duplicate a tree's data on all
processes that have local elements of this tree, repartitioning requires the
solution of an $n$-to-$m$ communication problem.
 We developed a sophisticated repartitioning algorithm 
to ensure that despite information being duplicated, no data is sent multiple
times. Additionally, we incorporated the communication of ghost trees
and demonstrated scaling on up to 917\e3 processes with tetrahedral and
hexahedral tree types and meshes of up to 371\e9 trees with faster
execution times than the fine mesh partitioning. 

For our development of element-type independent high-level AMR algorithms, the
\texttt{Ghost} algorithm to create a layer of ghost elements at each process
turned out to be the most challenging.
The most complex step in this algorithm is constructing
the face-neighbors of a given element for which we formulate a new element-type
independent algorithm.
With our approach we gain maximum flexibility with regards to the
implementations of the SFCs used in the neighboring trees. 
We optimized our \texttt{Ghost} routine by using recently developed recursive
search routines for tree-based AMR.
This algorithm has optimal strong and weak scaling behavior and we obtain
runtimes comparable to current state-of-the-art implementations with fixed
element-types~\cite{IsaacBursteddeWilcoxEtAl15}.

We concluded our investigation of element-type independent AMR by showing that
we can implement a scalable straightforward version of the 2:1 Balance routine based on
a ripple algorithm~\cite{TuOHallaronGhattas05,TuOHallaron04}, using the
available functionalities such as \texttt{Ghost}.

By implementing and investigating a finite volume solver we demonstrated that
\tetcode can easily be integrated in numerical solvers and maintains
excellent runtimes and scaling behavior in a practical application setting.
We underlined the flexibility in terms of element-types by performing
tests on different types of meshes including 2D and 3D hybrid meshes.
\vspace{2ex}

Concluding, we developed 
a new SFC for simplicial mesh refinement and 
a complete software library for element-type independent tree-based AMR that
supports different types of elements as well as hybrid meshes. 
We demonstrated excellent strong and weak scaling behavior on current
high-performance computing systems and showed that our software can be directly
used by numerical applications.

Our modular approach of separating the high-level and low-level algorithms
proved to be very succesful and
leads to a strong flexibility, allowing us to reuse the same high-level
algorithms for different element types and SFCs.
We demontrate this for example by testing \texttt{Ghost} and \texttt{Balance}
with hexehadral and with tetrahedral elements and
by applying the advection solver to various different meshes including 2D and
3D hybrid meshes.
With the definition of an API for the low-level algorithms, we enable users of
\tetcode to reuse existing implementations for an element-type, such as we did
throughout this thesis by using the \pforest implementation for the
quadrilateral and hexahedral Morton SFC.
This can greatly reduce the amount of work when extending existing applications
with the tree-based AMR functionality from \tetcode.
Another advantage is that additional element-types can be added to \tetcode
with little effort, such as shown in a Bachelor's thesis that includes the
support for prisms into \tetcode without any necessary changes to the
high-level algorithms~\cite{Knapp17}.

We are hopeful that the techniques presented in this thesis can lead to
tree-based AMR being used in areas that are currently dominated by unstructured
meshing, such as large-scale engineering applications with hybrid meshes; for
example~\cite{VinchurkarLongest08,MartineauStokesMundayEtAl06,KirbyBrazellYangEtAl17}.

\chapter{Outlook}
\label{ch:outlook}
Our results and algorithms provide a significant improvement in tree-based AMR
techniques, especially considering their generality and flexibility.
Nevertheless, there are various possibilities for further research and
improvement, of which we briefly discuss some examples here.

A first possible extension is to develop an SFC for pyramidal refinement
and establish the necessary low-level algorithms. This extends the
functionality to fully hybrid meshes consisting of tetrahedra, hexahedra,
prisms, and pyramids and thus allows for complex engineering applications
as for example described in~\cite{VinchurkarLongest08,KirbyBrazellYangEtAl17}. 
A similar possible extension are 4D elements such as hypercubes and 4D simplices.
These are for example used to model 4D space-time~\cite{FoteinosChrisochoides15}.

Also, if we consider the history of development of tree-based
quadrilateral/hexahedral AMR, with particular emphasis on the \pforest
library, we see several future research opportunities:

As we discuss in Sections~\ref{ch:balance} and~\ref{ch:app}, \texttt{Balance}
is the most expensive high-level algorithm in \tetcode,
which was also the case for the original implementation in \pforest.
Thus, developing a state-of-the art version along the lines
of~\cite{IsaacBursteddeGhattas12} would significantly reduce the runtime and
improve the overall scaling behavior of \texttt{Balance}.
A particular challenge in this development will be the generalization
of an element's insulation layer to arbitrary element types
and across trees of different types.

A second high-level algorithm to improve is \texttt{Iterate} to execute a user
provided callback on all mesh elements and element-to-element interfaces. In
this thesis we use a basic loop; however, a sophisticated implementation using
search algorithms that handles hanging nodes automatically is possible. Such a
universal mesh topology iterator for quadri\-late\-rals/hexa\-hedra is
described in~\cite{IsaacBursteddeWilcoxEtAl15}, and we believe these techniques
can be transferred to the general approach presented in this thesis.

Other points are the extension of the geometry handling to edge and vertex
neighbors and the development of a vertex numbering
scheme~\cite{IsaacBursteddeWilcoxEtAl15}. This can then be utilized by
application codes that are vertex based, such as for example finite element
solvers. A starting point could be the implementation of the \texttt{Nodes}
algorithm from~\cite{BursteddeWilcoxGhattas11}, which would require handling
of hanging edges and the computation of edge neighbors as a low-level
function and in the coarse mesh.

Furthermore, it would be interesting to use higher order geometry
representations for the coarse mesh. Each tree in the coarse mesh is equipped
with geometry interpolation points, which allows to use curved geometries for
each tree; see for
example~\cite{ZhangGuChenEtAl09,WilcoxStadlerBursteddeEtAl10}. This technique
can improve the overall geometric accuracy while decreasing the number of
coarse mesh trees. An implementation into \tetcode would be straightforward,
utilizing the existing data managment algorithms of the coarse mesh.

 \appendix
\chapter{The Low-Level API}
\label{ch:appendix}

In this appendix we list all low-level functions from the \tetcode library.  In
order to introduce a new element type, all of these functions have to be
implemented for this element type.
Currently, \tetcode supports the following element types.
\begin{itemize}
 \item 0D Vertices
 \item 1D Lines with Morton order
 \item 2D Quadrilaterals with Morton order using the \pforest library~\cite{Burstedde10a}
 \item 3D Hexahedra with Morton order using \pforest
 \item 2D Triangles with Tetrahedral-Morton order
 \item 3D Tetrahedra with Tetrahedral-Morton order
 \item 3D Prisms as a cross product of Lines and Triangles~\cite{Knapp17}
\end{itemize}
We list all low-level functions together with a brief description of their
effect in Table~\ref{tab:appendix}. For more details, we refer the reader
to the documentation of \tetcode~\cite{tetcode}.

\begin{longtable}{l|R{40ex}}
Function name & effect\\\hline
\texttt{t8\_element\_maxlevel}& Return the maximum possible refinement level\\
\texttt{t8\_element\_child\_eclass}& Return the element type of the $i$-th child\\
\texttt{t8\_element\_level}& Return an element's level\\
\texttt{t8\_element\_copy}& Copy an element\\
\texttt{t8\_element\_compare}& Compare the SFC indices of two elements\\\hline%
\texttt{t8\_element\_parent}& Compute the parent of an element\\
\texttt{t8\_element\_sibling}& Compute the $i$-th sibling (child of parent)\\
\texttt{t8\_element\_num\_corners}& Return the number of vertices of an element\\
\texttt{t8\_element\_num\_faces}& Return the number of faces of an element\\
\texttt{t8\_element\_max\_num\_faces}& Return the maximum number of faces of any descendant of an element\\\hline%
\texttt{t8\_element\_num\_children}& Return the number of children of an element\\
\texttt{t8\_element\_num\_face\_children}& Return the number of children at a given face\\
\texttt{t8\_element\_get\_face\_corner}& Given a face $f$ and an index $i$ return the element local index of the $i$-th corner of $f$\\
\texttt{t8\_element\_get\_corner\_face}& Given a corner $i$ and an index $j$ return
the $j$-th face sharing corner $i$\\
\texttt{t8\_element\_child}& Construct the $i$-th child of an element\\\hline%
\texttt{t8\_element\_children}& Construct all children of an element\\
\texttt{t8\_element\_child\_id}& Return the child-id of an element\\
\texttt{t8\_element\_ancestor\_id}& Return the child-id in the level $\ell$ ancestor\\
\texttt{t8\_element\_is\_family}& Given a collection of elements, query whether 
they form a family (that is, they form the output of \texttt{t8\_element\_children}
for some elemnet)\\
\texttt{t8\_element\_nca}& Construct the nearest common ancestor of two elements\\\hline%
\texttt{t8\_element\_face\_class}& Return the element type of the $i$-th face\\
\texttt{t8\_element\_children\_at\_face}& Construct all children that share the $i$-th face\\
\texttt{t8\_element\_face\_child\_face}&  Given a face of an element and 
the $i$-th child of that face, return the face number
   of the child of the element that matches the child face.\\
\texttt{t8\_element\_face\_parent\_face}& Given a face of an element, return the face number of the parent of the element that matches the element's face\\
\texttt{t8\_element\_tree\_face}& If a given face lies on the tree boundary, return
the face number of the respective tree face\\\hline%
\texttt{t8\_element\_transform\_face}& Transform the coordinates of a $(d-1)$-dimensional 
element across a $d$-dimensional tree-to-tree face connection\\
\texttt{t8\_element\_extrude\_face}& From a $(d-1)$-dimensional element
at a $d$-dimensional tree boundary, construct the respective $d$-dimensional
element within the tree.\\
\texttt{t8\_element\_boundary\_face}& From a $d$-dimensional element
and a face $f$ construct the $(d-1)$-dimensional element representing
the face\\
\texttt{t8\_element\_first\_descendant\_face}& Construct the
first descendant at a given face\\
\texttt{t8\_element\_last\_descendant\_face}& Construct the last 
descendant at a given face\\\hline%
\texttt{t8\_element\_is\_root\_boundary}& Query whether an element lies at its
root tree's boundary across a given face\\
\texttt{t8\_element\_face\_neighbor\_inside}& Construct the face-neighbor
of an element within the same tree\\
\texttt{t8\_element\_set\_linear\_id}& Construct an element given a SFC index\\
\texttt{t8\_element\_get\_linear\_id}& Return the SFC index of a given element\\
\texttt{t8\_element\_first\_descendant}& Construct the first descendant
of an element\\\hline%
\texttt{t8\_element\_last\_descendant}& Construct the last descendant of an
element\\
\texttt{t8\_element\_successor}& Construct the successor (w.r.t.\ the SFC index)
of a given element\\
\texttt{t8\_element\_anchor}& Return the (integer) coordinates of
the anchor node of an element\\
\texttt{t8\_element\_root\_len}& Return the (integer) length of the root 
tree of an element\\
\texttt{t8\_element\_vertex\_coords}& Return the (integer) coordinates\\
&of a given vertex of an element\\\hline%
\caption{All low-level functions provided by \tetcode with a brief description.}
\label{tab:appendix}
\end{longtable}
\bibliographystyle{plain}
\bibliography{./bib/ccgo_new,./bib/group,./bib/johannes}
\listoftables
\listoffigures
\end{document}